%% file: main.tex
\documentclass[letterpaper,twocolumn,10pt]{article}
\usepackage{usenix}

\usepackage{preamble}

\newcommand{\urlartifact}{\href{http://doi.org/10.5281/zenodo.17977528}{10.5281/zenodo.17977528}}

\begin{document}

\date{}

\title{\Large \bf The Adverse Effects of Omitting Records in Differential Privacy: \\ How Sampling and Suppression Degrade the Privacy--Utility Tradeoff \\ 
{\large\textit{(Long version of the paper accepted at USENIX Security '26)}}}

\author{
{\rm Àlex Miranda-Pascual}\\
Karlsruhe Institute of Technology\\Universitat Politècnica de Catalunya\\
alex.pascual@kit.edu
\and
{\rm Javier Parra-Arnau}\\
Universitat Politècnica de Catalunya\\
javier.parra@upc.edu
\and
{\rm Thorsten Strufe}\\
Karlsruhe Institute of Technology \\
thorsten.strufe@kit.edu
}

\maketitle

\begin{abstract}
    Sampling is renowned for its privacy amplification in differential privacy (DP), and is often assumed to improve the utility of a DP mechanism by allowing a noise reduction. In this paper, we further show that this last assumption is flawed: When measuring utility at equal privacy levels, sampling as preprocessing consistently yields penalties due to utility loss from omitting records over all canonical DP mechanisms---Laplace, Gaussian, exponential, and report noisy max---, as well as recent applications of sampling, such as clustering.

    Extending this analysis, we investigate suppression as a generalized method of choosing, or omitting, records. Developing a theoretical analysis of this technique, we derive privacy bounds for arbitrary suppression strategies under unbounded approximate DP. We find that our tested suppression strategy also fails to improve the privacy–utility tradeoff. Surprisingly, uniform sampling emerges as one of the best suppression methods---despite its still degrading effect. Our results call into question common preprocessing assumptions in DP practice.
\end{abstract}

\section{Introduction}\label{sec:introduction}

Differential privacy (DP)~\cite{dwork2006differential,dwork2014algorithmic} is firmly established as the state-of-the-art privacy framework, thanks to its strong privacy guarantees and mathematical formulation. 
A popular technique employed in the DP field is \textit{sampling}, widely regarded for its privacy amplification property~\cite{smith2009differential,balle2018privacy,steinke2022composition}: Applying a uniform sampling $\S$ to any $(\varepsilon,\delta)$-DP mechanism $\M$ yields that the composition $\M\circ\S$ is $(\varepsilon',\delta')$-DP with $\varepsilon'\leq\varepsilon$ and $\delta'\leq\delta$. Here, it is commonly assumed that these gains in privacy can be translated into higher utility by calibrating the privacy parameters in $\M\circ\S$ to achieve the same privacy guarantees, while reducing the perturbation applied by the DP mechanism~\cite{bun2023controlling,li2012sampling,fang2024privacy}.  
While the noise introduced by $\M$ can indeed be thus reduced, it remains an open question whether this noise reduction is in general sufficient to offset the utility loss caused by the loss of records through $\S$ itself. In fact, the contrary has now been demonstrated~\cite{raisa2024subsampling} for DP stochastic gradient descent (DP-SGD)~\cite{abadi2016deep}---a well-established and widespread mechanism using sampling.

The question is, beyond DP-SGD, whether the inherent utility loss caused by $\S$ can, in fact, be outweighed by the intended noise reduction. In the composition of $\M\circ\S$, distortion may derive from two principal sources: from the perturbation provided by the protection in $\M$ and from the utility loss caused by the omission of records in $\S$. 
By translating the privacy amplification into decreased protection demands, we reduce the perturbation of $\M$, thus establishing a tradeoff between these two error sources. In this paper, we investigate the utility effect of sampling in the composition $\M\circ\S$, and importantly, we expand upon these findings to broadly inquire into the possible privacy and utility effects more generally.

To begin, we investigate sampling as a preprocessing step, testing the assumption that it enhances utility through privacy gains. Our experiments compute and compare the utility guarantees of DP mechanisms with and without sampling under identical privacy parameters. We analyze the canonical DP mechanisms (e.g., Laplace, Gaussian, exponential, and report-noisy-max mechanisms~\cite{dwork2014algorithmic}) as well as clustering, which has previously been ``amplified'' by sampling~\cite{blocki2021differentiallyprivate}. Our findings reveal that the utility with sampling is worse than that without, indicating that any utility gained from sampling's privacy amplification does not compensate for the inherent utility loss caused by removing records through sampling in general.

This surprising revelation leads us to ask: \textit{What is the actual effect of deleting records as wanted?} Therefore, we introduce to DP the technique of \textit{suppression}~\cite{hundepool2012statistical}, a method from statistical disclosure control (SDC) that targets the deletion of vulnerable records. Suppression is effective in other privacy frameworks, such as $k$-anonymity~\cite{samarati2001protecting}, where it improves the privacy--utility tradeoff by selectively removing outliers~\cite{gramaglia2021glove,chaudhuri2006when}. Similar benefits could then be expected in DP, as outliers also complicate the DP privacy--utility tradeoff.

Therefore, we investigate whether DP mechanisms with suppression yield better privacy or utility than the same mechanism without. Given that outliers may vary in vulnerability across databases, the decision to delete records can significantly impact the privacy parameters. 
Therefore, we derive upper bounds on the privacy parameters of $\M\circ\S$ (where $\S$ now denotes suppression) in terms of those of $\M$. We impose no conditions on the suppression algorithm $\S$, thus covering all possible cases, including state-of-the-art sampling.

To assess utility, we replicate our sampling experiments for a family of suppression algorithms, obtaining the same findings. Among the tested mechanisms, our results indicate that DP mechanisms with this suppression do not outperform those without at fixed privacy levels, often yielding worse utility guarantees compared to sampling. Thus, despite the negative outcomes associated with sampling, it remains the superior method for record deletion.

In summary, this paper main contributions are as follows:
\begin{itemize}
    \item Our experimental study on uniform Poisson sampling over classic unbounded approximate DP mechanisms reveals that, for fixed privacy levels, the utility guarantees of the DP mechanism with sampling are worse than those of the mechanism without sampling.
    \item We introduce record suppression to DP and we prove how the privacy parameters of $\M$ are affected by preprocessing with any suppression algorithm $\S$, and when we obtain a privacy amplification. To the best of our knowledge, we are the first to provide such a general result for unbounded DP, and additionally, the first to provide a result that is also independent of the choice of~$\M$.
    \item We empirically show that even when factoring in the privacy amplification, our suppression in DP worsens the privacy--utility tradeoff analogously to sampling. 
    We show that, despite both techniques providing unfavorable outcomes, suppression rarely outperforms sampling.
\end{itemize}

Our findings offer new insights into the relationship between DP, sampling, and suppression. Above all, our findings highlight the need for careful consideration of data preprocessing strategies in privacy-preserving data analysis.

\section{Preliminaries and Background}\label{sec:preliminaries}

\subsection{Differential Privacy}

In this paper, we work with (pure) differential privacy ($\varepsilon$-DP)~\cite{dwork2006differential,dwork2014algorithmic} and its approximate counterpart, $(\varepsilon,\delta)$-DP~\cite{dwork2006our,dwork2014algorithmic}. Their definitions using our notation are as follows:

\begin{definition}[Differential privacy~\cite{dwork2014algorithmic}]\label{def:ApproximateDP}
    Let $\varepsilon,\delta\geq0$ and $\D$ be a class of databases drawn from a data universe $\X$. A randomized mechanism $\M$ with domain $\D$ is \textit{$(\varepsilon,\delta)$-DP} if for all neighboring $D,D'\in\D$ and all measurable $\meas\subseteq\Range(\M)$,
    \[
        \Prob\{\M(D)\in\meas\} \leq \e^{\varepsilon}\Prob\{\M(D')\in\meas\}+\delta.
    \]

    If $\delta=0$, we say that $\M$ is \textit{$\varepsilon$-DP}. 
\end{definition}

We will work almost exclusively with \textit{unbounded DP}, the original DP notion~\cite{dwork2006differential,dwork2014algorithmic}. Unbounded DP is obtained by selecting in \Cref{def:ApproximateDP} the \textit{unbounded neighborhood definition}~\cite{kifer2011no}. Two databases are said to be unbounded neighboring if one is obtained from the other by adding or deleting a record (or, their symmetric difference has size $1$: $|D\Delta D'|=1$). 
Note that the definition of DP allows for different variants by changing the neighborhood definition~\cite{kifer2011no,desfontaines2020sok}, such as \textit{bounded DP}~\cite{kifer2011no}. In this case, two databases are said to be \textit{bounded neighboring} if one is obtained from the other by substituting one record for another.

The \textit{privacy parameters}, $\varepsilon$ and $\delta$, quantify the privacy level of the mechanism $\M$, limiting the amount of information an attacker can extract about the input data. 
Intuitively, lower values of $\varepsilon$ and $\delta$ provide stronger privacy. Although $(\varepsilon,\delta)$-DP is defined for all $\varepsilon,\delta\geq0$,
only smaller values of $\varepsilon$ and $\delta$ provide reasonable or acceptable privacy levels (consensus would place these bounds at around $\varepsilon\leq2$ and $\delta<\frac{1}{|D|}$~\cite{dwork2014algorithmic}). 

Any $(\varepsilon,\delta)$-DP mechanism is also $(\varepsilon',\delta')$-DP for any $\varepsilon'\geq\varepsilon$ and $\delta'\geq\delta$. Since larger values provide weaker privacy, it makes sense to find the lowest possible values of $\varepsilon$ and $\delta$. 
In particular, given an $(\varepsilon,\delta)$-DP mechanism $\M$, we say that $\varepsilon$ and $\delta$ are \textit{tight} if there are no $\varepsilon''\leq\varepsilon$ and $\delta''\leq\delta$ (not both equal) such that $\M$ is $(\varepsilon'',\delta'')$-DP. 

\subsection{Sampling in DP}\label{sec:samplingDP}

\textit{Sampling} (also known as  \textit{subsampling}) is a non-perturbative masking method of SDC that consists in publishing a (random) subset of the original dataset of records~\cite{hundepool2012statistical}. 
Sampling is also well established for DP, with many theoretical works~\cite{smith2009differential,ullman2017cs7880,balle2018privacy,steinke2022composition,li2012sampling} considering it and studying its effect as a preprocessing algorithm, usually in search of improving the privacy guarantees of DP mechanisms.  Balle et~al.~\cite{balle2018privacy} term this search as the \textit{problem of privacy amplification}: Given a sampling algorithm $\S$ and a DP mechanism $\M$, the goal is to bound the privacy parameters of $\M\circ\S$ by those of $\M$. 
In particular, $\M\circ\S$ must also be DP, which requires that the sampling technique be well adapted to the neighborhood definition~\cite{balle2018privacy}.
The expected privacy enhancement is given by the ``privacy amplification by sampling'' principle~\cite{li2012sampling,balle2018privacy}, which holds that the privacy guarantees of a DP mechanism can be improved, with respect to the original database, when applied to a random subset of records. The rationale behind DP amplification by sampling is as follows: As records are dropped, the privacy should increase, and since there is uncertainty about which records are actually sampled, an attacker will be unable to tell which data has or has not been sampled~\cite{steinke2022composition}.

DP amplification by sampling was first introduced~\cite{smith2009differential} for \textit{Poisson sampling}, which samples each element $x$ in $D$ with a fixed probability $p\in[0,1]$. 
Li et~al.~\cite{li2012sampling} later provide the first tight bounds on the privacy parameters of $\M\circ\S$ (proven tight in~\cite{balle2018privacy}): If $\M$ is an unbounded $(\varepsilon,\delta)$-DP mechanism, then $\M\circ\S$ is unbounded $(\ln(1+p(\e^\varepsilon-1)),\delta p)$-DP 
(i.e., $\e^\varepsilon-1$ and $\delta$ are reduced by a factor of $p$). 
Since the privacy parameters of $\M\circ\S$ are smaller than those of $\M$, the privacy amplification by Poisson sampling is clear. Additionally, Poisson sampling allows, as well, for a slightly more general non-uniform definition, where each element $x$ in a database can be sampled with a different probability $p_x\in[0,1]$. The tight privacy parameters of $\M\circ\S$ are given by the same formula with $p=\max_{x\in\X} p_x$~\cite{steinke2022composition}.

Balle et~al.~\cite{balle2018privacy} and Steinke~\cite{steinke2022composition} provide independent theorems for unbounded and bounded DP. 
In particular, the theorems provide Poisson sampling for unbounded DP, and \textit{sampling without replacement} (SWOR)~\cite{beimel2013characterizing,wang2016learning,lin2013benefits} for bounded DP, which uniformly samples a subset of $D$ of size $m<|D|$. They prove that, for both sampling algorithms, $\M\circ\S$ is tightly $(\ln(1+p(\e^\varepsilon-1)),\delta p)$-DP for $p=\max_{x\in\X}\Prob\{x\in\S(D)\}$.
In particular, Balle et~al.~\cite{balle2018privacy} define their sampling algorithm $\S$ as any algorithm over $\D$ that returns subsets of the input database $D\in\D$.
However, to obtain bounds on the privacy parameters of $\M\circ\S$, their theory requires certain assumptions (referred to as ``$d_Y$-compatibility'' in the paper) to be satisfied, which are not achieved for all $\S$ under their definition.

Bun et~al.~\cite{bun2023controlling} introduce other sampling strategies to pure DP, such as cluster and stratified sampling. 
The proposed strategies allow more flexibility in sampling, but they all still involve some form of uniform selection. 
Further, the authors conclude that some sampling strategies cannot enjoy the privacy amplification property. 
In addition, they present a general theorem that provides a lower bound on the tight parameters of $\M\circ\S$ for most sampling strategies $\S$ in bounded pure DP.

\textbf{How sampling is applied in DP.}
Sampling has been widely used to design and improve DP mechanisms. The privacy amplification provided by the most popular variants, Poisson sampling and SWOR, can also be translated into less DP perturbation~\cite{bun2023controlling,li2012sampling,fang2024privacy}. Indeed, consider an $\varepsilon$-DP mechanism $\M_\varepsilon$ (e.g., the Laplace mechanism). By applying Poisson sampling, we can obtain that $\M_\varepsilon\circ\S$ is $\varepsilon'$-DP with $\varepsilon'<\varepsilon$, but we can also translate this privacy amplification by reducing the noise (of the Laplace mechanism) as follows: We choose $\varepsilon''$ so that $\varepsilon=\ln(1+p(\e^{\varepsilon''}-1))$ and replace $\M_\varepsilon$ with $\M_{\varepsilon''}$ (which adds less noise). The result is $\M_{\varepsilon''}\circ\S$ satisfying $\varepsilon$-DP for the initial (unmodified) privacy budget~$\varepsilon$. However, we note here that such noise reduction \textit{does not account for} the potential utility loss caused by sampling records (see \Cref{sec:samplingExperiments}). 

In particular, Poisson sampling has been used as a preprocessing step to, for example, ``amplify'' $k$-medians and $k$-means DP clustering mechanisms~\cite{blocki2021differentiallyprivate} and dependency-graph generation for the publication of synthetic high-dimensional databases~\cite{chen2015differentially}.
Furthermore, the literature has also used sampling in alternative settings to preprocessing, like within iterations of mechanisms
in machine learning and in the analysis of noisy stochastic gradient descent~\cite{abadi2016deep,koskela2023practical,mironov2019renyi,ponomareva2023how,jiang2024calibrating}. For example, DP stochastic gradient descent (DP-SGD)~\cite{abadi2016deep} uses multiple iterations of Poisson sampling over the Gaussian mechanism.
In the context of stochastic learning with Rényi DP, sampling has also been studied theoretically for the Laplace mechanism and others~\cite{zhu2019poisson,jiang2024calibrating}.

Recently, Räisä et al.~\cite{raisa2024subsampling} studied the effect of sampling on the variance of DP-SGD at fixed privacy levels and concluded that less sampling always leads to a better privacy--utility tradeoff. Even though their conclusion is limited to this concrete utility measure and to binary databases, 
this calls into question the actual benefit of sampling as a preprocessing step to DP, which we will analyze empirically in the next section. 

\section{The Effect of Uniform Sampling on Utility}\label{sec:samplingExperiments}

In this section, we study whether the privacy amplification provided by uniform Poisson sampling and translation into less perturbation can indeed provide benefits to the privacy--utility tradeoff of unbounded approximate DP mechanisms.

Thus, in these experiments, we will be comparing the utility levels of a DP mechanism without sampling to those of the same mechanism with sampling---considering the corresponding noise reduction to meet identical privacy guarantees. Our experimentation covers two types of mechanisms: (1)~Basic canonical DP mechanisms of DP (i.e., Laplace, Gaussian, exponential, and report noisy max), in the context of statistic computation (precisely, the mean and mode); and (2)~Clustering, covering a mechanism~\cite{gupta2010differentially} that has been previously ``amplified'' with sampling~\cite{blocki2021differentiallyprivate}. 
For completeness, also note that the Gaussian and Laplace mechanisms have also been studied with sampling~\cite{abadi2016deep,jiang2024calibrating,zhu2019poisson,ponomareva2023how}.

All mechanisms we test allow for rescaling the privacy parameters, essentially calibrating the noise added. Thus, given a DP mechanism and denoting its $(\varepsilon,\delta)$-DP instantiations as $\M_{\varepsilon,\delta}$, we will be comparing the utility values of $\M_{\varepsilon,\delta}$ and $\M_{\varepsilon'',\delta''}\circ\S$. To ensure that both mechanisms satisfy $(\varepsilon,\delta)$-DP for the same privacy parameters, it is enough to select $\varepsilon''=\ln(\frac{\e^{\varepsilon}-(1-p)}{p})$ and $\delta''=\frac{\delta}{p}$ where $p$ is the sampling rate of $\S$.  
We note that since $\varepsilon\leq\varepsilon''$ and $\delta\leq\delta''$, the utility loss of $\M_{\varepsilon'',\delta''}$ is, generally, equal or lower than that of $\M_{\varepsilon,\delta}$; yet, the question remains on how the utility losses of $\M_{\varepsilon,\delta}$ and $\M_{\varepsilon'',\delta''}\circ\S$ compare. Note that comparing the utility guarantees under the same privacy level allows for a fair comparison of the tradeoff.

\subsection{Experiment Setup}\label{sec:mechanismsandutilitymetrics}

\textbf{Mean computation.} We protect the mean in two ways: Using the Laplace and Gaussian mechanisms. 
We consider two independent $\frac{\varepsilon}{2}$-DP Laplace (or $(\frac{\varepsilon}{2},\frac{\delta}{2})$-DP Gaussian) mechanisms, one for the sum query $f_{\mathrm{sum}}$ that sums all the values in the database, and one for the counting query $f_{\mathrm{count}}$ that counts how many records are in the database. The noisy mean is then obtained by dividing the noisy sum by the noisy count, which is an $\varepsilon$-DP mechanism (or $(\varepsilon,\delta)$-DP) by sequential composition and post-processing~\cite{dwork2014algorithmic}. This variation, called NoisyAverage, is a well-known DP mechanism to compute the mean that reduces the overall noise that would be necessary to protect the mean query function directly with Laplace or Gaussian noise~\cite{li2016differential}.

\textbf{Mode computation.} We compute the mode in four ways using report noisy max (RNM) and the exponential mechanism~\cite{dwork2014algorithmic}. RNM is used to determine which of the $k$ count queries $f_k$ has the maximum value, and thus we can use it to return a perturbed mode with DP protection. RNM achieves $\varepsilon$-DP by adding Laplace noise: For all $D\in\D$, it is defined as $\M^{f,\varepsilon}_{\mathrm{RNM}}(D) = \argmax_{i\in[k]}\{f_i(D) + z_i\}$ with $z_i\sim\Lap(\frac{\Delta f_i}{\varepsilon})$ (i.i.d.). A variation of RNM is obtained by adding exponential noise from $\operatorname{Exp}(\frac{\varepsilon}{2\Delta f_i})$ instead, which still satisfies $\varepsilon$-DP. In our case, the query functions act over disjoint support (i.e., the elements of $\X$), and thus RNM can be viewed as a parallel composition~\cite{mcsherry2009privacy} of $|\X|$ Laplace mechanisms. 
This fact allows us to obtain an $(\varepsilon,\delta)$-DP variant using Gaussian mechanisms (note that this is not generally true for RNM with Gaussian noise~\cite{lebensold2024privacy}).
In addition, we also protect the mode using the exponential mechanism by defining the score function as the count query minus the maximum value in $\X$ (this last term ensures the score function is negative, avoiding computational inaccuracies caused by large floating-point numbers).

\textbf{Clustering mechanisms.} Recall that DP clustering has previously been amplified via sampling~\cite{blocki2021differentiallyprivate}. Our experiment covers their tested $k$-median algorithm~\cite{gupta2010differentially}, achieving DP through the exponential mechanism. We also test a different $k$-means algorithm (due to some ambiguity of the original method). We choose the well-known DP version of the $k$-means clustering (i.e., Lloyd's algorithm)
introduced by Blum et~al.~\cite{blum2005practical} for this purpose. 
This mechanism, also known as DPLloyd~\cite{su2017differentially}, achieves DP by computing the centroids in each iteration with the Laplace NoisyAverage mechanism over each cluster. We refer to Su et~al.~\cite{su2017differentially} for further details.

\textbf{Utility metrics.}
To keep our plots consistent, we ensure that for all utility metrics $u(\M,D)$, larger values indicate worse utility (increased errors), and values close to 0 indicate better utility preservation. We are thereby providing intuition on the amount of error or inaccuracy.
For the mean computation, we stick to common practice and take $u(\M,D)$ as the mean percent error (MPE) between the real mean $\frac{f_{\mathrm{sum}}(D)}{f_{\mathrm{count}}(D)}$ and the output noisy mean.
For the mode computation, we take $u(\M,D)$ as the probability of incorrectly returning the argument of the maximum of $D$. 
For the $k$-median clustering mechanisms, we take the average of the $L^2$ distances of each record to the closest median (the average \textit{cost}~\cite{gupta2010differentially}). Finally, for DPLloyd, we take the normalized intracluster variance (NICV), defined as the average of the squares of the $L^2$ distance of each record to the centroid of the assigned cluster. NICV is a common metric used to evaluate $k$-means clustering approaches including DPLloyd~\cite{su2017differentially}. 

\textbf{Databases.} For the computation of mean and mode, we consider three well-known popular numerical databases in the field of SDC. For each database, we select two columns to use in our evaluations, considering each column as its own one-dimensional database. \Cref{tab:databases} shows the selected databases and columns, where we prioritized different numerical ranges for variability and simpler-to-understand attributes for each database (such as ages).  
However, we do not compute the mode over the columns of the Census database because multiple elements reach the maximum count for each column (in particular, no element repeats in \texttt{FEDTAX}).

\begin{table}[t]
    \centering
    \rowcolors[]{1}{white}{lightgray}
    \small
    \begin{tabular}{c c c c }
         Database & Columns & \bettershortstack[c]{Value\\range} & \bettershortstack[c]{Sensitivity\\bounds} \\
         Adult~\cite{becker1996adult} & \texttt{age} 
         & 17--90 & 0--125 \\
         \cline{2-4}
         (Size: 32\,561) & \texttt{hours-per-week} & 1--99 & 0--100 \\
         Census~\cite{brand2002reference} & \texttt{FEDTAX} & 1--21\,260 & 0--31\,889 \\
         \cline{2-4}
         (Size: 1\,080) & \texttt{FICA} & 6--7\,932 & 0--11\,890 \\
         Irish~\cite{ayala-rivera2016cocoa} & \texttt{Age} & 15--84 & 0--125 \\
         \cline{2-4}
         (Size: 66\,666) & \texttt{Education} & 1--10 & 1--10 \\
    \end{tabular}
    \caption{Databases employed in the experimentation.}
    \label{tab:databases}
\end{table}

\begin{figure*}[ht]
    \centering
    \includegraphics[width=0.32\textwidth]{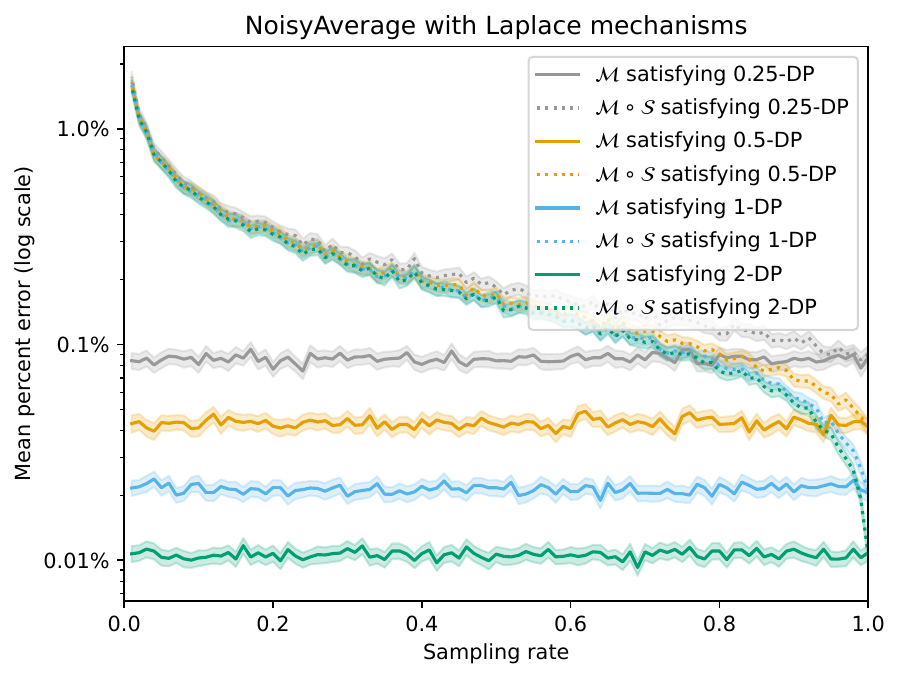}%
    \includegraphics[width=0.32\textwidth]{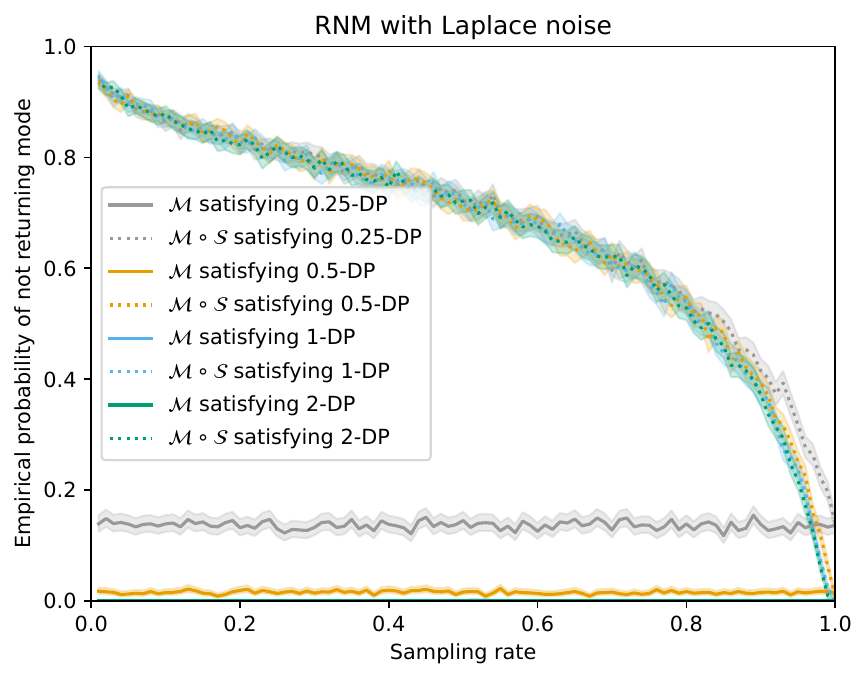}%
    \includegraphics[width=0.32\textwidth]{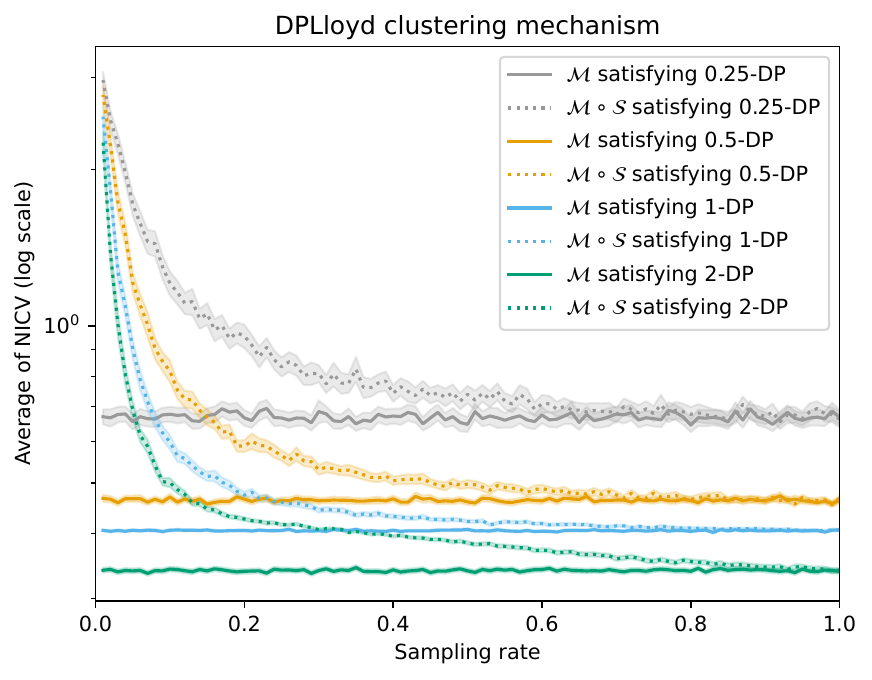}%
    \caption{Plots of the utility values of $\M$ and $\M\circ\S$ for the uniform Poisson sampling for the Adult database (and \texttt{age} column). The shaded areas correspond to a 95\% confidence interval (CI) for the mean of the utility metric (95\% Wilson CI for the mode).}
    \label{fig:uniformPoissonsampling}
\end{figure*}

In our computations, we will need bounds on the values of each column (e.g., to compute the sensitivity of the Laplace/Gaussian mechanism). Since DP is a property that does not depend on the choice of database, lower and upper bounds are usually chosen that do not necessarily match the range of values in the database. Following field practices~\cite{soria-comas2014enhancing}, we either select logical extremal bounds (e.g., $0$ to $125$ for ages) or $0$ to $\ceil{1.5\,\texttt{max\_{}value\_{}in\_{}database}}$ if no clear upper bound exists---note that this does not constitute a privacy violation, but the contrary, it is an estimation of the possible domain range meant to represent every database in $\D$~\cite{soria-comas2014enhancing}. The exact values chosen are shown in \Cref{tab:databases}.

We run DPLloyd on the Adult database~\cite{becker1996adult} under the same conditions as Su et~al.'s experiment~\cite{su2017differentially}: The clustering is performed over the six numeric columns of the database and for $k=5$ clusters. All values are (min-max) normalized to $[-1,1]$ as required by DPLloyd. 
The chosen $k$-median algorithm is not empirically evaluated in the original publications, but only theoretically~\cite{gupta2010differentially,blocki2021differentiallyprivate}. Therefore, following the mechanism requirements and due to large computational cost, we first generate a random two-column database over $\{1,\dots,100\}^2$. We sample 100 points using a Gaussian distribution with $\sigma=10$ (nearing to the closest integer) centered at four randomly selected accumulation points in $\{10,\dots,90\}^2$. The database is then normalized so the sensitivity is $1$ and we select $k=4$.

\subsection{Experiments and Results}\label{sec:samplingResults}

For every database and mechanism, we compute the utility metric values of the mechanism with and without sampling for various privacy parameters and sampling rates. We run the experiments for $\varepsilon\in\{0.25,0.5,1,2\}$. We use the optimal Gaussian mechanism~\cite{balle2018improving} that, unlike the classic version~\cite{dwork2014algorithmic}, is also defined for $\varepsilon\geq1$. In addition, parameter $\delta$ is set to $|D|^{-2}$ when working with the Gaussian mechanism (as suggested in the literature~\cite{dwork2014algorithmic}), the only mechanism that requires a non-zero $\delta$. We test every sampling algorithm $\S_p$ with sampling rate $p\in\{0.01,0.02,\dots,0.98,0.99\}$. Since all mechanisms and sampling algorithms are randomized, for each instantiation, we compute $u(\M_{\varepsilon,\delta},D)$ and $u(\M_{\varepsilon'',\delta''}\circ\S_p,D)$ 500 times (mean and DPLloyd), 2\,000 times (mode), or 20 times ($k$-median), and we always provide their means.

Our results are surprising as they show that the utility guarantees of $\M_{\varepsilon'',\delta''}\circ\S_p$ are worse than those of $\M_{\varepsilon,\delta}$ for all mechanisms, databases, privacy parameters $\varepsilon$ and $\delta$, and (almost all) sampling rates $p$ we tested. Räisä et~al.~\cite{raisa2024subsampling} make a similar observation for DP-SGD, but given the general interest in privacy amplification \cite{bun2023controlling,li2012sampling,fang2024privacy}, we consider it rather surprising that our results hold for all tested mechanisms. 

In \Cref{fig:uniformPoissonsampling}\footnote{
We provide all plots gallery in Appendix~\ref{sec:plotgallery}.}, we show that the utility values differ significantly for most sampling rates, with $\M\circ\S$ having worse utility than $\M$ (here, simplifying the notation). 
The difference becomes small near a sampling rate of 1, and we find a few rates where $\M\circ\S$ preserves utility better than $\M$.
We assume that this is due either to issues of floating-point precision, or rare beneficial random choices by the sampling algorithm.

Looking at the mechanisms independently, we note that the mean computation with sampling provides very small error, remaining less than $0.25\%$ MPE for sampling rates larger than $0.4$ and less than $2\%$ even for more abrasive rates near~$0$; yet, it always increases with respect to the mechanism without sampling. The mode computation expresses variations across the databases depending on the original distribution. For example, in the \texttt{hours-per-week} column in the Adult database, the mode represents more than half the results in the database, and thus both $\M$ and $\M\circ\S$ provide a perfect failure probability of $0$. For more varied data, like the \texttt{age} column in the Adult database, we see drastic utility losses with sampling, increasing the failure probability of RNM (Laplace) from under 18\% to over 60\% for most sampling rates. 
DPLloyd also exhibits a utility degradation under sampling, but contrary to the others, $u(\M\circ\S,D)$ remains quite close to $u(\M,D)$ until spiking at around $p=0.1$. The $k$-median clustering also shows a utility degradation similar to the previous plots; however, this does not contradict the theoretical utility evaluations performed on sampling~\cite{blocki2021differentiallyprivate} in which the sum of distances (rather than the average sum) is compared, resulting in bias with respect to the database and sample size.

In summary, our result shows that, for the tested mechanisms, it is preferable to apply mechanism $\M$ directly under the target privacy parameters than to rely on the privacy amplification of sampling to improve the privacy--utility tradeoff.

\section{Introducing Suppression to DP}\label{sec:introsuppression}

Our previous experimental results reveal that uniform Poisson sampling has detrimental effects on the utility that DP mechanisms yield: The utility gain from translating the privacy parameters is insufficient to counteract the utility loss caused by omitting records. 

Yet, records in databases are complex and diverse, and they have different degrees of vulnerability. Furthermore, experience shows that some records are harder to protect than others or have different costs for the protecting mechanism. 
Thus, any loss in utility that is caused by omitting records could be reduced when records are omitted strategically, for example, by targeting the hardest-to-protect or outlying records. 
This process is known as \textit{suppression}~\cite{hundepool2012statistical}, and it has been shown as a mechanism amplifier under syntactic privacy notions like $k$-anonymity~\cite{gramaglia2021glove,elemam2008protecting}. 
Nevertheless, to the best of our knowledge, no studies treat the effect of suppression for DP.

Thus, in the following, we introduce suppression to DP. Formally, suppression is a SDC non-perturbative masking technique like sampling~\cite{hundepool2012statistical}. 
In suppression, data values are deleted from the original dataset to eliminate easily identifiable features. Our suppression corresponds to \textit{whole-record suppression}~\cite{samarati2001protecting, gramaglia2021glove,elemam2008protecting,bayardo2005data,xu2014survey}, but we note that suppression can also refer to deleting specific data values of the records~\cite{hundepool2012statistical}.

We will formalize DP suppression as a generalization of sampling. While the state of the art on sampling in unbounded DP works exclusively with algorithms defined according to a uniform selection scheme~\cite{elemam2008protecting}, we define suppression completely general, covering any way of choosing or omitting records. 
Nevertheless, we acknowledge that our definition of suppression in DP does match the general definition of sampling by Balle et~al.~\cite{balle2018privacy}. We will provide bounds on the privacy parameters for all suppression algorithms and empirically evaluate a family of suppression algorithms.
In this way, we address an open question in the literature: What are the effects of sampling/suppression when defined more flexibly (e.g., in a non-uniform manner) on DP mechanisms?

\subsection{The Suppression Algorithm}

We now turn to investigating the effect of deleting some records on a DP mechanism. To keep our results well-defined, we assume that $\D$ is closed by subsets (or subdatabases), i.e., if $D\in\D$ and $C\subseteq D$, then $C\in\D$. This is the common assumption in the literature~\cite{dwork2014algorithmic}. Suppression is modeled by a \textit{suppression algorithm} $\S$ with domain $\D$ that, given any database $D\in\D$, deterministically or randomly outputs a subset of $D$ or, equivalently, suppresses a subset of $D$. Since databases $D\in\D$ are of finite size, we consider $\S(D)$ to be a discrete random variable that outputs subsets of $D$ (including the case where $\S$ is a deterministic function). Note that our definition does not impose any restrictions on the deletion process, nor does it establish relations between $\S(D_1)$ and $\S(D_2)$ for different $D_1,D_2\in\D$. In particular, completely different suppression techniques can be defined independently for each database in $\D$.
Thus, $\S$ is defined in a completely general way without any additional restrictions.

We introduce two types of suppression: database-dependent and database-independent. In \textit{database-independent} suppression, an element or subset of elements is deleted with the same probability regardless of the database to which it belongs, i.e., $\Prob\{C\subseteq\S(D_1)\}=\Prob\{C\subseteq\S(D_2)\}$ for all $D_1,D_2\in\D$ and for all $C\subseteq D_1,D_2$. Some examples of database-independent suppression consist of deleting records---deterministically or probabilistically---over or under a predefined threshold (e.g., deleting all individuals over the age of 100 or over the height of 2.10\,m). In particular, any same record is deleted with the same probability across all different databases (i.e., $\Prob\{x\in\S(D_1)\}=\Prob\{x\in\S(D_2)\}$ for all $x\in D_1,D_2$).
We refer to all other types of suppression as \textit{database-dependent}, where a same record may be deleted with different probabilities in two databases. This includes deleting records according to their mean or counts in the database, which varies across databases. 

We note that classic state-of-the-art sampling algorithms~\cite{li2012sampling,balle2018privacy,steinke2022composition} are all database-independent, with many~\cite{li2012sampling,balle2018privacy} also being \textit{uniform} (i.e., $\Prob\{x\in\S(D)\}$ is equal for all $D\in\D$ and $x\in D$).
Bun et~al.~\cite{bun2023controlling} introduce specific database-dependent sampling algorithms, but these contain some kind of uniform selection and are thus more limited in our suppression context.

\subsection{The Suppression Problem}\label{sec:suppressionproblem}

We extend the \textit{privacy amplification problem} of sampling~\cite{balle2018privacy} in order to define the \textit{suppression problem} as follows: Given an unbounded $(\varepsilon,\delta)$-DP mechanism $\M$ and a suppression algorithm $\S$, both with domain $\D$, what are the privacy and utility guarantees provided by $\M\circ\S$? In particular, with this problem, we are interested in understanding under which conditions $\M\circ\S$ has beneficial properties, and whether they improve over those of $\M$. 
Posing these questions, we are interested in understanding exactly when $\M\circ\S$ also satisfies approximate DP, and with which privacy parameters. In addition, given our experimental results (see \Cref{sec:samplingExperiments}), we highlight the importance of knowing the effect on the output utility of $\M\circ\S$ compared to that of $\M$.

The suppression problem already becomes relevant when discussing database-independent vs.\ -dependent suppression. 
Database-dependent suppression is better suited to the deletion of vulnerable or outlier records because these records usually depend on the database to which they belong. However, such records impose a cost on the privacy parameters. In essence, DP must protect or take into account any change between databases, which consumes privacy budget. For these reasons, differences between databases which, on the one hand, provide utility improvements in database-dependent suppression, may incur, on the other, costs for the privacy guarantees. In particular, $\S$ needs to respect the neighborhood relation to ensure low privacy parameters~\cite{balle2018privacy}. We illustrate this phenomenon in our privacy results in \Cref{sec:suppressiontheorem}.

In the following, we provide different answers to the suppression problem covering how suppression affects the privacy (\Cref{sec:suppressiontheorem}) and utility guarantees (\Cref{sec:experiments}), ultimately seeing that our suppression does not improve over sampling. 

\section{The Effect of Suppression on Privacy}\label{sec:suppressiontheorem}

In this section, we study how the privacy guarantees are affected by suppression algorithms. As previously mentioned, we will assume that $\M$ satisfies unbounded $(\varepsilon,\delta)$-DP, and study when $\M\circ\S$ satisfies unbounded $(\varepsilon^\S,\delta^\S)$-DP, deriving expressions for $\varepsilon^\S$ and $\delta^\S$. We will provide bounds independent of the choice of mechanism $\M$, and thus show the worst-case bounds with respect to $\M$.

Our goal is to show not only how specific targeted suppression, such as deleting outliers, affects the privacy parameters, but also the effect of any possible alternative suppression algorithm. 
This allows to identify which and how records can be deleted to improve the privacy guarantees, and thus we pave the way for data curators to easily learn what the privacy guarantees are after deleting exactly the records they want. To stay true to this, we want to impose the fewest conditions on $\S$ to provide the most general results possible. 

We note that $\S$ does not necessarily satisfy $(\varepsilon,\delta)$-DP and therefore the DP composition rules cannot be applied. Particularly for pure DP, since $\S$ only outputs subsets of the input database, it is possible that $\S(D)$ outputs a subset which cannot be output by $\S(D')$, violating the $\varepsilon$-DP definition. Therefore, the only algorithm $\S$ that satisfies pure DP is the mechanism that deletes all records, i.e., $\S(D)=\varnothing$ for all $D\in\D$, which satisfies $0$-DP. We note that $\M\circ\S$ can still satisfy DP even when $\S$ is non-DP.

In the following, we first tackle deterministic suppression (\Cref{sec:deterministicsuppression}), presenting some interesting results that show an initial understanding how deleting records affects the privacy parameters. Due to the limitations of this type of suppression in DP, we then study probabilistic suppression in a general way (\Cref{sec:mainsuppression}) and provide a specific probabilistic suppression strategy for outlier deletion (\Cref{sec:outlierscoresuppressiontheorem}). 

\textbf{All proofs of our results can be found in Appendix~\ref{sec:proofs}.}

\subsection{Deterministic Suppression}\label{sec:deterministicsuppression}

It is well known that $(\varepsilon,\delta)$-DP is a worst-case metric since any $(\varepsilon,\delta)$-DP mechanism $\M$ must satisfy $\Prob\{\M(D)\in\meas\}\leq\e^{\varepsilon}\Prob\{\M(D')\in\meas\}+\delta$ for all neighboring $D,D'\in\D\coloneqq\operatorname{Domain}(\M)$ (and all measurable $\meas\subseteq\Range(\M)$). One of the first results where we would intuitively expect privacy amplification is in reducing the number of inequalities that must be satisfied, which can easily be done by reducing the domain $\D$ of $\M$. In particular, by excluding the hardest-to-satisfy inequalities, we can lower the values of $\varepsilon$ and $\delta$, potentially obtaining that $\M$ over this reduced domain is $(\varepsilon',\delta')$-DP with $\varepsilon'<\varepsilon$ and $\delta'<\delta$.
We can obtain a domain reduction if we have a deterministic suppression algorithm that verifies $\S(\D)\subsetneq\D$, thus reducing the input of mechanism $\M$ from $\D$ to $\S(\D)$. 

\Cref{th:deterministicsuppression} shows how the privacy parameters of $\M$ are affected when preprocessed with a deterministic suppression algorithm $\S$. There are two factors that affect the privacy parameters of $\M\circ\S$: First, the privacy improvement we can gain by restricting the domain of $\M$ as we explained; and second, the effect on the privacy parameters caused by applying $\S$, which closely follows from the known preprocessing result on $c$-stable transformations~\cite{mcsherry2009privacy}.

\begin{theorem}[restate = THdeterminisitcsuppression, name = Effect of deterministic suppression]\label{th:deterministicsuppression}
    Let $\M$ be an $(\varepsilon,\delta)$-DP mechanism and $\S$ be a deterministic suppression algorithm, both with domain $\D$. Let $\intset$ be such that $\S(\D)\subseteq\intset\subseteq\D$, and suppose the restriction of $\M$ to domain $\intset$, $\M|_{\intset}$, is $(\varepsilon_\intset,\delta_\intset)$-DP. 
    Then, $\M\circ \S=\M|_{\intset}\circ\S$ with domain $\D$ is $(\varepsilon_\intset\,\Delta_\intset\S,\delta_\intset\sum^{\Delta_\intset\S-1}_{k=0}\e^{\varepsilon_\intset k})$-DP, where the \textit{sensitivity} of $\S$ is
    \[
        \Delta_\intset\S \coloneqq \sup_{\substack{D,D'\in\D\\\text{neighb.}}}
        \dist_\intset(\S(D),\S(D')),
    \]
    and $\dist_{\intset}(\S(D),\S(D'))$ is the minimum number of neighboring databases in $\intset$ needed to go from $\S(D)$ to $\S(D')$ (see~\cite{guerra-balboa2024composition}).
\end{theorem}

In the theorem, we show the effect of restricting the domain of $\M$ to intermediate subsets $\intset$ (such that $\S(\D)\subseteq\intset\subseteq\D$), since it is possible that the subset $\intset$ that provides the lowest privacy parameters is not $\S(\D)$.
Since the result holds for every choice of $\intset$, we can choose $\intset$ that minimizes the privacy parameters of $\M\circ\S$. We note that finding the minimum can be difficult since smaller $\intset$ intuitively yield smaller (or equal) values of $\varepsilon_\intset$ and $\delta_\intset$ but larger (or equal) values of $\Delta_\intset\S$. More formally, given $\intset\subseteq\intset'$, we have that $\varepsilon_\intset\leq\varepsilon_{\intset'}$ and $\delta_\intset\leq\delta_{\intset'}$ (if chosen tightly), but, at the same time, we have $\Delta_{\intset'}\S\leq\Delta_{\intset}\S$. In particular, we note that the smallest sensitivity is
\[
    \Delta_{\D}\S = \sup_{\substack{D,D'\in\D\\\text{neighb.}}} |\S(D)\Delta\S(D')| \leq \Delta_{\intset}\S.
\]

Moreover, there can exist $\intset\subsetneq\D$ such that $\Delta_\D\S=\Delta_\intset\S$ (e.g., $\intset=\bigcup_{D\in\D}\mathcal{P}(\S(D))$ where $\mathcal{P}$ denotes the power set). In particular, we also note that if $\delta=0$, then $\M\circ\S$ is $(\varepsilon\Delta_\D\S)$-DP, remaining in pure DP.

We now provide an applied example of \Cref{th:deterministicsuppression} and some of its consequences.

\begin{example}[Laplace mechanism with deterministic suppression]\label{ex:Laplace}
    Recall that the Laplace mechanism $\M_{f,b}$ of the query function $f\colon\D\to\R^k$ that adds noise drawn from the Laplace distribution $\Lap(b)$ with scale $b$ to each coordinate of $f(D)$ satisfies $\varepsilon$-DP with $\varepsilon=\frac{\Delta f}{b}$~\cite{dwork2014algorithmic}, where the sensitivity of $f$,
    \[
        \Delta f = \sup_{\substack{D,D'\in\D\\\text{neighb.}}} \norm{f(D)-f(D')}_1,
    \]
    depends on $\D$ and thus on the range of $f$. Considering a deterministic suppression $\S$ that reduces the domain (i.e., $\S(\D)\subsetneq\D$) and selecting $\intset\coloneqq\bigcup_{C\in\D}\mathcal{P}(\S(C))$, we obtain by \Cref{th:deterministicsuppression} that $\M_{f,b}\circ\S$ is $(\varepsilon_\intset\Delta_{\intset}\S)$-DP with
    \[
        \varepsilon_\intset=\frac{\Delta f|_{\intset}}{b}
        \quad\text{and}\quad \Delta f|_{\intset} \coloneqq \sup_{\substack{D,D'\in\intset\\\text{neighb.}}} \norm{f(D)-f(D')}_1 \leq \Delta f.
    \]
    
    That is, we obtain $\varepsilon_\intset<\varepsilon$ if and only if $\Delta f|_{\intset}\,\Delta_{\intset}\S<\Delta f$.
    This improvement can be leveraged to increase utility by raising the privacy parameter, which adds less noise (i.e., noise drawn from $\Lap(b')$ with $b'<b$) accordingly. Selecting $b'$ such that $\frac{\Delta f|_{\intset}}{b'}\Delta\S=\frac{\Delta f}{b}$ 
    holds, will ensure that the privacy parameters of both mechanisms remains constant and allow us to evaluate the effect of suppression (cf. \Cref{sec:experiments}).
\end{example}

Furthermore, bear in mind that $\varepsilon_{\intset}$ and $\delta_{\intset}$ depend on $\M$ and may change for different mechanisms. In particular,
there are mechanisms that cannot benefit from a domain reduction, like a Laplace or Gaussian mechanism for a counting query (since $\Delta f|_\intset=\Delta f=1$ for all $\intset$). Therefore, $\M\circ\S$ is always $(\varepsilon\Delta_\intset\S,\delta\sum^{\Delta_\intset\S-1}_{k=0}\e^{\varepsilon k})$-DP for all $(\varepsilon,\delta)$-DP mechanisms $\M$, which is an independent bound on the choice of $\M$. 

Moreover, \Cref{th:deterministicsuppression} provides a tight bound: For all privacy parameters and all $\S$, there exists a DP mechanism $\M$ such that $\M\circ\S$ is tightly $(\varepsilon_{\intset}\Delta_\intset\S,\delta_{\intset}\sum^{\Delta_\intset\S-1}_{k=0}\e^{\varepsilon_{\intset} k})$-DP if $\delta_{\intset}\sum^{\Delta_\intset\S-1}_{k=0}\e^{\varepsilon_{\intset} k}<1$ (see \Cref{prop:deterministicsuppressiontight}).
Therefore, \Cref{th:deterministicsuppression} provides a complete characterization for deterministic suppression, indicating that there are suppression algorithms $\S$ and mechanisms $\M$ such that $\M\circ\S$ provides weaker privacy than $\M$, since $\Delta_\intset\S$ can potentially be greater than~$1$.

When we have a suppression algorithm with sensitivity $\Delta_{\intset}\S=1$, we obtain that the privacy parameters given by \Cref{th:deterministicsuppression} remain constant (or decrease, if so by a domain reduction). By definition, all database-independent suppression algorithms $\S$ have sensitivity $\Delta_\intset\S=1$, such as fixing a subset $A$ of the universe of records $\X$ and defining $\S_A(D)=D\cap A$ for all $D\in\D$. This can be understood as removing the values outside $A$, the set of elements with ``good'' properties, or the records that are not outlying. For instance, we can use this to remove predefined extreme values, such as super-centenarians in an age database, or remote locations in a location database. In both of these examples, the suppressed records are defined independently of the choice of database, i.e., using public or common knowledge to designate people over a certain age as outliers or to define which map areas are remote.

Alternatively, database-dependent suppression can be useful for outlier deletion because it does not require any knowledge and allows suppression on a per-database basis. As a simple example, consider the database class $\D$ of databases containing people's ages (ranging from 0 to the maximum verified age) and other data. Applying a database-independent suppression that deletes all supercentenarians may make sense for many $D\in\D$, but the results can become skewed for specific databases in $\D$ such as a superagers database.

However, as covered in \Cref{sec:introsuppression}, DP must account for changes in-between databases, which increases the privacy parameters when applying a database-dependent suppression. 
In this case, this privacy degradation is represented by the sensitivity of $\S$, which is large or even infinite when defining a suppression strategy specifically to delete outliers or distant records. For example, deleting all records whose average distance to the other records in the database exceeds a certain threshold (\Cref{prop:determinsticsuppressionwithinfinitesensitivity}) or deleting the top $P\%$ of records that are furthest away from all other records in the database (for $P\leq 50$; \Cref{prop:determinsticsuppressionwithinfinitesensitivity}) are both suppression algorithms with $\Delta_{\D}\S=\infty$ (and thus $\Delta_{\intset}\S=\infty\geq \Delta_{\D}\S$ for all $\intset$). In these examples, adding or removing a record can have a large effect on the distance to the rest of the records in the database, and thus $\S(D)$ and $\S(D')$ can potentially be very different, which leads to $\Delta_{\D}\S=\infty$ and no DP guarantees.

In general, we find that many database-dependent suppression algorithms defined to suppress outliers require large or even infinite sensitivities, thus increasing the privacy parameters to unmanageable levels. Furthermore, while it is theoretically possible to construct a deterministic database-dependent suppression algorithm $\S$ with $\Delta_\D\S=1$ (see \Cref{prop:deterministicsuppressionDataDependentSensitivity1}), we have not found any that correspond to a meaningful way of deleting outliers.

\subsection{Probabilistic Suppression}\label{sec:mainsuppression}

In this section, we consider probabilistic suppression and extend the privacy evaluation of suppression to the whole spectrum, covering any suppression algorithm. Precisely, we provide results showing how the privacy parameters of $\M$ are affected when preprocessing with any probabilistic suppression and when suppression yields privacy amplification. 

Deterministic suppression can be viewed as a special case of probabilistic suppression, and thus our theorems generalize the results of the previous section. In particular, probabilistic suppression can still provide cases where $\M\circ\S$ is not DP, as we obtained in \Cref{sec:deterministicsuppression}. Therefore, we find it convenient to exclude such cases from our main theorem (\Cref{th:unboundedcase}) in order to provide a concise result, and we later explain how we can extend the theorem to the rest of the suppression algorithms, including those that do not achieve DP. 

Our proofs follow the steps of the existing theorem on Poisson sampling by Li et~al.~\cite{li2012sampling} but with our generalized suppression algorithm $\S$. We find that the newer sampling theorems~\cite{steinke2022composition,balle2018privacy} require additional conditions that do not directly generalize, or are inapplicable, to the more general suppression. Like these sampling results, our theorems work for any $(\varepsilon,\delta)$-DP mechanism $\M$ and the bounds are independent of the choice of $\M$. The bounds given are therefore worst-case with respect to $\M$, i.e., they represent bounds to the largest possible privacy parameters over all DP mechanisms.

The essential steps in the proof are bounding the privacy parameters of $\M\circ\S$ with that of $\M$ using the law of total probability, i.e., 
\[
    \Prob\{\M(\S(D))\in\meas\} = \quad\smashoperator[lr]{\sum_{C\in\supp(\S(D))}}\quad\Prob\{\M(C)\in\meas\}\Prob\{\S(D)=C\},   
\]
and finding a relation between the probability measures of $\S(D)$ and $\S(D')$ (for all neighboring $D,D'\in\D$). Here, finding a good relation is the challenging part. For the state-of-the-art uniform sampling~\cite{li2012sampling,balle2018privacy,steinke2022composition}, this relation is simply constant over the support of $\S(D)$, but the relation for general $\S$ can be hard to define and very complex in some exceptional cases. 

Therefore, as mentioned, we find it convenient to exclude these hard cases for now by assuming that $\S$ satisfies the \textit{support condition}: For all unbounded-neighboring $D,D'\in\D$ of the form $D'=D_{+y}\coloneqq D\uplus\{y\}$ and all $C\subseteq D$, there is a non-zero chance that $\S(D)$ outputs $C$ if and only if there is a non-zero chance that $\S(D')$ outputs $C$ or $C_{+y}$; or, formally, $C\in\supp(\S(D))$ if and only if $C\in\supp(\S(D'))$ or $C_{+y}\in\supp(\S(D'))$. 
We can think of this condition as basically ensuring that if an output $C$ is possible for $\S(D)$, then $C$ or $C_{+y}$ is also possible for $\S(D')$, which avoids dividing by $0$ in \Cref{th:unboundedcase}.

\begin{theorem}[restate = THgeneralunboundedtheorem, name = Suppression theorem]\label{th:unboundedcase}
    Let $\M$ be a mechanism that satisfies unbounded $(\varepsilon,\delta)$-DP and $\S$ be a suppression algorithm that satisfies the support condition, both with domain $\D$. Then, $\M\circ\S$ is unbounded $(\varepsilon^{\S},\delta^{\S})$-DP with 
    \[
        \varepsilon^{\S} = \sup_{\substack{D,D'\in\D\\\text{neighb.}}} \varepsilon^{\S}_{D,D'} \quad \text{and} \quad \delta^{\S} = \sup_{\substack{D,D'\in\D\\\text{neighb.}}} \delta^{\S}_{D,D'},
    \]
    where $\varepsilon^{\S}_{D,D'}$ and $\delta^{\S}_{D,D'}$ are as follows:
    If $D'=D_{+y}$, then
    \begin{align*}
        \e^{\varepsilon^{\S}_{D,D'}}=\quad\smashoperator[l]{\max_{\substack{C\in\supp(\S(D))}}}\frac{\Prob\{\S(D)=C\}}{\Prob\{\S(D')=C\}+\e^{-\varepsilon}\Prob\{\S(D')=C_{+y}\}}
    \end{align*}
    and
    \[
        \delta^{\S}_{D,D'}= \ \ \delta\ \ \smashoperator[l]{\sum_{\substack{C\in\supp(\S(D))}}}\frac{\Prob\{\S(D)=C\}\,\e^{-\varepsilon}\Prob\{\S(D')=C_{+y}\}}{\Prob\{\S(D')=C\}+\e^{-\varepsilon}\Prob\{\S(D')=C_{+y}\}},
    \]
    and, since the values are not symmetric with respect to $D,D'$, 
    \begin{align*}
        \e^{\varepsilon^{\S}_{D',D}}=\quad \smashoperator[lr]{\max_{\substack{C\in\supp(\S(D))}}}\quad\frac{\Prob\{\S(D')=C\}+\e^{\varepsilon}\Prob\{\S(D')=C_{+y}\}}{\Prob\{\S(D)=C\}}
    \end{align*}
    and
    \[
        \delta^{\S}_{D',D}=\ \ \delta\ \ \smashoperator[lr]{\sum_{C\in\supp(\S(D))}}\quad \Prob\{\S(D')=C_{+y}\} = \delta \Prob\{y\in\S(D')\}.
    \]
\end{theorem}

Our theorem always satisfies $\delta^\S\leq\delta$, but $\varepsilon^\S$ can be larger, equal, or smaller than $\varepsilon$. 
Therefore, our theorem does not always show a privacy amplification, but rather shows how the parameters change through a suppression algorithm, and gives us an intuition on how deleting records affects the privacy parameters. 
The bound we provide is not tight in general, but it is for some specific algorithms~$\S$, such as Poisson sampling. 

Even though \Cref{th:unboundedcase} is not generally tight, it provides an intuition on how the privacy budget is affected by suppression. When the values of $\frac{\Prob\{\S(D)=C\}}{\Prob\{\S(D_{+y})=C\}}$ and $\frac{\Prob\{\S(D)=C\}}{\Prob\{\S(D_{+y})=C_{+y}\}}$ (or their inverses) remain small, $\varepsilon^\S_{D,D'}$ (or $\varepsilon^\S_{D',D}$) remains small; but in other cases, $\max\{\varepsilon^\S_{D,D'},\varepsilon^\S_{D',D}\}$ takes on larger values, increasing the privacy parameters of $\M\circ\S$.
In particular, we see at play the fact that DP must always factor in changes between neighboring databases, as mentioned in \Cref{sec:introsuppression}.
For example, if we define $\S$ as probabilistically deleting those points that are furthest away from the mean, we must take into account the differences in distribution caused by adding any single record in any database.  

Our theorem provides an upper bound on the privacy guarantees of every suppression algorithm satisfying the support condition. The complexity of the equations in \Cref{th:unboundedcase} are just a consequence of the potential complexity of suppression algorithms. Therefore, \Cref{th:unboundedcase} is more useful in evaluating specific suppression strategies, as we will do in \Cref{sec:outlierscoresuppressiontheorem}. In addition, it also provides the tight bound for Poisson sampling~\cite{li2012sampling} and our tight bound for deterministic suppression with $\Delta_\D\S=1$ we provided in \Cref{th:deterministicsuppression}.

As mentioned earlier, we can adapt the proof of \Cref{th:unboundedcase} to obtain the result for all suppression/sampling algorithms $\S$. The full result is given in \Cref{th:generalunboundedtheorem}: Its idea is to assign each subset $C$ that does not satisfy the support condition with another term $C^*$ that does, so that the bound can be defined. However, there are multiple ways to assign $C^*$ to $C$, each giving a different bound of the privacy parameters.
In this case, the privacy parameters increase for assignments with large $|C\Delta C^*|$, and we lose the guarantee that $\delta^\S\leq\delta$ or that $\M\circ\S$ is pure DP if $\M$ is pure DP from \Cref{th:unboundedcase}.

In summary, \Cref{th:unboundedcase} (and the general \Cref{th:generalunboundedtheorem}) tells us that some probabilistic suppression can provide privacy amplifications just as sampling does, even without the improvement provided by domain reduction; while other suppression strategies can end up with larger privacy parameters, especially if $\S$ varies significantly between neighboring databases. However, even though there are specific examples that increase the values of the privacy parameters (e.g., the deterministic ones), we are unable to provide proof of the tightness of the results.
Finally, we note that these theorems are defined independently of the choice of $\M$, showing how privacy degrades in the worst case for any $(\varepsilon,\delta)$-DP mechanism. Nevertheless, certain mechanisms $\M$ could provide better bounds, such as the improvement provided by domain reduction.

In summary, selective suppression strategies can easily violate acceptable privacy bounds, resulting in privacy degradation instead of amplification; while suppression methods that delete records more uniformly across databases guarantee lower (better) privacy parameters.

\subsection{Distance-Based Probabilistic Suppression}\label{sec:outlierscoresuppressiontheorem}
In this section, we present a type of suppression strategy $\S$ to deal with the presence of outliers in databases. The privacy parameters of $\M\circ\S$ are obtained through \Cref{th:unboundedcase}. 

In this suppression strategy, every record is suppressed independently---as in Poisson sampling---but with a probability proportional to how different they are from the other records in the database. By carefully measuring these differences, we obtain that any DP mechanism $\M$ preprocessed by our $\S$ is also DP and can derive the precise expression of its privacy parameters (\Cref{th:outlierscoresuppression}). In addition, our probabilistic result avoids the large sensitivities of deterministic database-dependent suppression that we saw in \Cref{sec:deterministicsuppression}, recalling that the deterministic version of this result is not DP ($\Delta_\D\S=\infty$, as seen in \Cref{prop:determinsticsuppressionwithinfinitesensitivity}).

Formally, the difference between records is given through a (normalized) bounded distance  
$\dist\colon\X\times\X\to[0,1]$ defined over the data universe $\X$ (from where the databases are drawn). We choose parameters $m,M\in(0,1)$ with $m\leq M$ to control the extent to which the property of being an outlier is considered. Precisely, we take the \textit{$(m,M)$-transformation~$\tdist$ of $\dist$}, defined as $\tdist(x,y) = (m + (M-m)\dist(x,y)) \in [m,M]$
for all $x,y\in\X$. Then, for any non-empty database $D\in\D$, we define the \textit{outlier-score function $\out_D\colon D\to [m,M]$ over $D$ (with respect to $\tdist$)} such that
$\out_D(x) = \frac{1}{|D|}\sum_{\substack{y\in D}}\tdist(x,y)$
for all $x\in D$.

Our suppression algorithm will delete every record $x$ in $D$ independently with probability $\out_D(x)$, the average distance to all elements in $D$.
By definition, $\out_D(x)$ is guaranteed to be between $m$ and $M$, thus providing a lower and upper bound on the probability of a record being deleted. A large difference $M-m$ means that the suppression strategy discriminates strongly between inliers and outliers, while $M-m=0$ provides the uniform Poisson sampling, where each record is deleted independently with the same probability $m$. 

\Cref{th:outlierscoresuppression} shows the effect of this family of suppression algorithms, which we call \textit{outlier-score suppression}, on the privacy parameters of $\M\circ\S$.

\begin{theorem}[restate = THoutlierscoresuppression, name = Outlier-score suppression]\label{th:outlierscoresuppression}
    Let $\S$ be the outlier-score suppression algorithm that independently deletes each record $x\in D$ with probability $\out_D(x)$, i.e., $\S$ is defined so that
    \[
        \Prob\{\S(D)=C\} = \prod_{x\in C} (1-\out_D(x)) \prod_{x\in D\backslash C} \out_D(x) 
    \]
    for all $D\in\D$ and all $C\subseteq D$. Then, if $\M$ is $(\varepsilon,\delta)$-DP, we obtain that $\M\circ\S$ is $(\varepsilon^\S,\delta^\S)$-DP where $\delta^\S=\delta(1-m)$ and
    \[
        \varepsilon^{\S} = \max_{p\in[0,1]}\max\{l_1(p),l_2(p),l_3\} 
    \]
    up to an error\footnotemark of $2\cdot10^{-7}$, where 
    \begin{multline*}
        l_1(p) = \ln(\e^\varepsilon-(\e^\varepsilon-1)(pM+(1-p)m))\\
        +p\frac{M}{m}+(1-p)\frac{1-m}{1-(pM+(1-p)m)}-1
    \end{multline*}
    and
    \begin{multline*}
        l_2(p) = \ln\!\bigg(\e^\varepsilon-(\e^\varepsilon-1)\bigg(pM+(1-p)\frac{(M+m)-pM}{2-p}\bigg)\bigg)\\
        +p\frac{M}{m}+(1-p)\frac{1-\frac{(M+m)-pM}{2-p}}{1-M}-1
    \end{multline*}
    for all $p\in[0,1]$; and $l_3=-\ln(\e^{-\varepsilon}+(1-\e^{-\varepsilon})M)+1-\frac{1-M}{1-m}$.
\end{theorem}

\footnotetext{The proof of \Cref{th:outlierscoresuppression} is computer assisted and verified up to an error of $2\cdot 10^{-7}$ for every value of $m$ and $M$ in $\{0.01,0.02,\dots, 0.98,0.99\}$ (with $m\leq M$) and every value of $\varepsilon$ in $\{0,0.01,0.02,\dots, 1.98,1.99,2\}$, $\{2.1,2.2,\dots,9.9,10\}$, and $\{11,12,\dots,99,100\}$. Computational power is used to check that the bound we provide matches the empirically optimized value (see details in \Cref{remark:computation,remark:computationInverse}).  
We conjecture the expression extends to all $m$ and $M$, and to all $\varepsilon\leq 100$ due to the continuity of $\varepsilon^\S$.}

\begin{figure*}[t]
    \centering
    \includegraphics[width=0.25\textwidth]{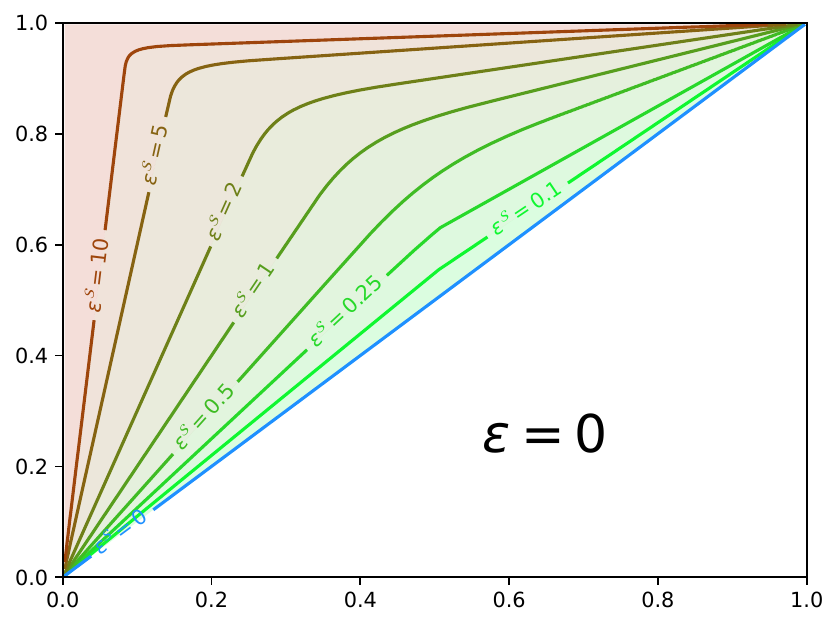}%
    \includegraphics[width=0.25\textwidth]{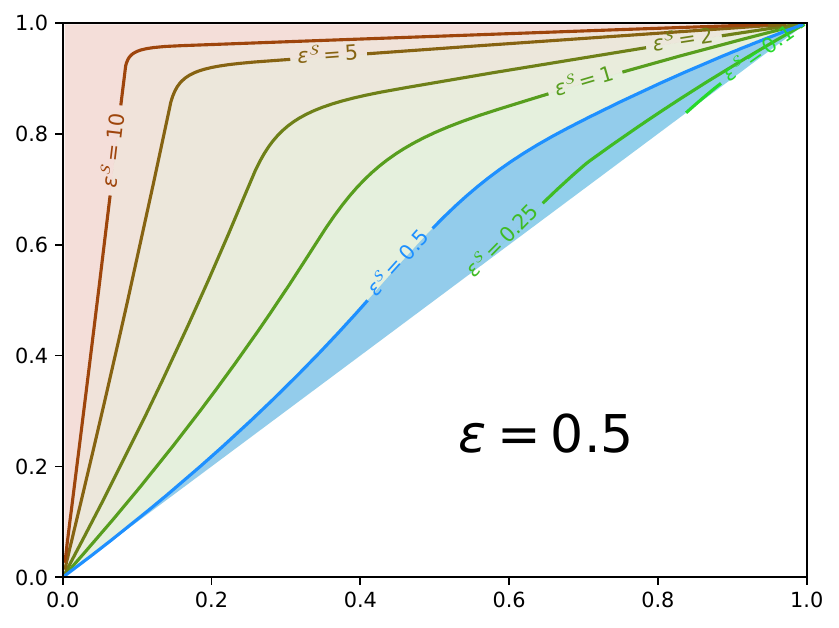}%
    \includegraphics[width=0.25\textwidth]{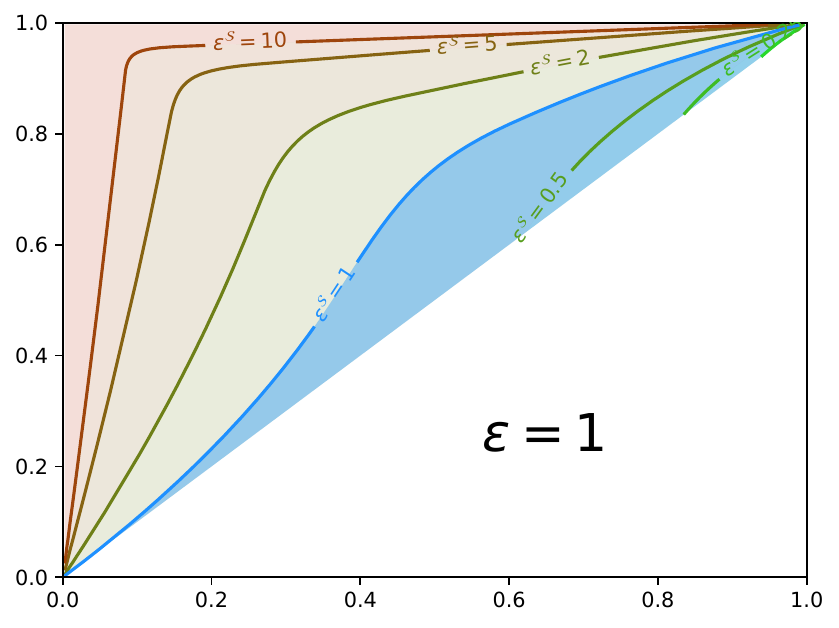}%
    \includegraphics[width=0.25\textwidth]{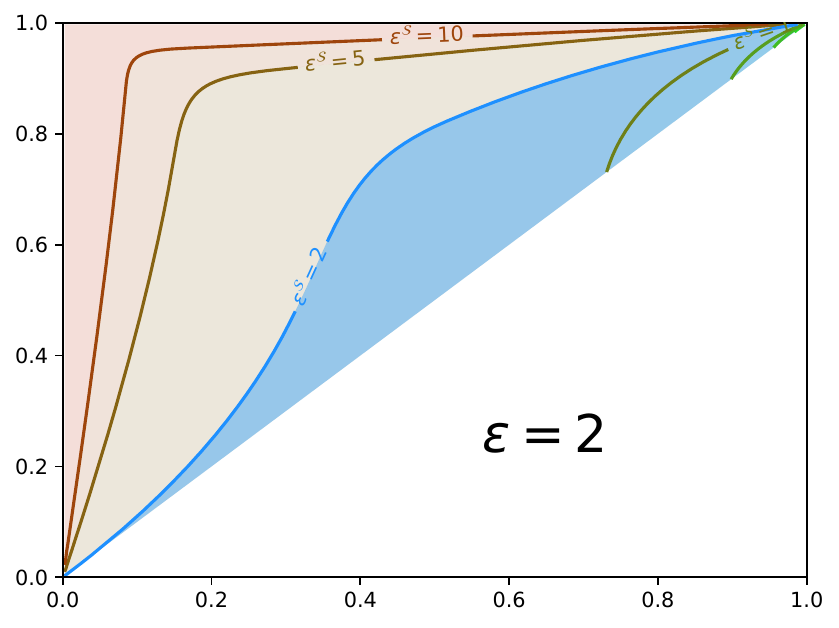}%
    \caption{Contour plots of $\varepsilon^\S\coloneqq\varepsilon^\S(\varepsilon,m,M)$ with respect to $m$ ($x$-axis) and $M$ ($y$-axis) for the values of $\varepsilon\in\{0,0.5,1,2\}$. Values increase as $M-m$ increases. The blue regions show the values of $m$ and $M$ with a privacy amplification (i.e., such that $\varepsilon^\S\leq\varepsilon$).}
    \label{fig:SuppressionTheoremEps}
\end{figure*}

Note that $\varepsilon^\S$ depends only on $\varepsilon$ and the constants $m$ and $M$, and $\delta^\S$ only on $\delta$ and $m$, and neither depend on the choice of mechanism $\M$ nor on the distance $\dist$. We plot the expression of $\varepsilon^\S$ with respect to $m$ and $M$ for some values of $\varepsilon$ in \Cref{fig:SuppressionTheoremEps} and provide an interactive plot that allows computing the exact $\varepsilon^\S$ from the three constants\footnote{\url{https://www.desmos.com/calculator/hydpubdqtm}}. We also provide a closed form of the precise values in \Cref{prop:expressionofp}, and we note that the maximum is usually obtained when $p=0$ or $p=1$, which quickly simplifies the expression of $\varepsilon^\S$.
In these cases, our result is tight with respect to \Cref{th:unboundedcase}, but we cannot show whether it is tight in general (see \Cref{remark:tightnessoutliersuppressiontheorem} for further discussion on tightness). We note that the complex expression of $\varepsilon^\S$ in \Cref{th:outlierscoresuppression} is a consequence of our efforts to provide the tightest bounds on the result.

We observe, similar to \Cref{th:unboundedcase}, that we always have $\delta^\S\leq\delta$, but $\varepsilon^\S$ can take values larger, equal, or smaller than~$\varepsilon$. Lower values are obtained near the diagonal, with the blue regions in \Cref{fig:SuppressionTheoremEps} representing the values such that $\varepsilon^\S\leq\varepsilon$, i.e., the values of $m$ and $M$ such that $\S$ provides a privacy amplification. In particular, \Cref{th:outlierscoresuppression} generalizes uniform Poisson sampling, which corresponds to the algorithms on the diagonal (i.e., $m=M$). As seen in the plots, $\varepsilon^\S$ increases as $m\to0$ or $M\to1$, with the limit values being $\infty$. In addition, the ratio $\frac{\varepsilon^{\S}}{\varepsilon}$ converges to $1$ when $\varepsilon\to\infty$. 

Overall, outlier-score suppression can be used to suppress outliers with higher probability. We note that our definition of outlyingness is similar to previous definitions of record vulnerability used in DP~\cite{meeus2024achilles}. However, our theorem shows that greatly differentiating outliers quickly increases the privacy parameters (see \Cref{fig:SuppressionTheoremEps}).  

\subsection{The Impact of Suppression on Privacy}

Suppression in DP can offer a privacy amplification, but not always: 
Since DP must protect any difference between neighboring databases, flexible suppression algorithms---such that $\S(D)$ and $\S(D')$ behave differently for neighboring databases $D,D'\in\D$---greatly increase the privacy parameters.

Nevertheless, our theorems in this section not only represent this phenomenon, but also show precisely how suppression affects the privacy parameters. In particular, this effect is depicted in our theorems with the sensitivity $\Delta_{\intset}\S$ in the deterministic case, and with the ratios $\frac{\Prob\{\S(D)=C\}}{\Prob\{\S(D_{+y})=C\}}$ and $\frac{\Prob\{\S(D)=C\}}{\Prob\{\S(D_{+y})=C_{+y}\}}$ (and their inverses) in the probabilistic case. Furthermore, the difference between $m$ and $M$ in \Cref{th:outlierscoresuppression} also exhibits the same effect: The privacy parameters increase if outliers are more distinguishable (i.e., if $M-m$ increases).

In this sense, DP causes a conflict between ensuring manageable privacy parameters in $\M\circ\S$ and having flexibility between $\S$ across databases. For example, in the deterministic case, we have seen that database-independent suppression $\S$ ensures $\Delta_{\D}\S=1$; hence, $\M\circ\S$ satisfies DP with at least the same privacy parameters than $\M$. On the other hand, the flexibility of database-deterministic suppression usually leads to very high sensitivities and privacy parameters. In some cases, these values are infinite (e.g., when records furthest from the database centroid are deleted), in which case $\M\circ\S$ cannot possibly be DP. 
Nevertheless, when $\Delta_\intset\S$ is not too large, there can be a privacy amplification if the mechanism and the suppression algorithm allow for a domain reduction. However, this domain reduction is highly dependent on the chosen mechanism, and there are always mechanisms for each suppression algorithm which do not benefit from it. Consequently, the guarantee here is not global. 

In contrast to deterministic suppression, our results show that in the probabilistic case, privacy amplifications can exist independently of $\M$. 
However, they are only possible if the probability of deleting records does not vary drastically across neighboring databases; otherwise, privacy parameters degrade.
We recall that suppression encompasses the state-of-the-art sampling, and it is especially the suppression algorithms close to uniform sampling that achieve a more significant privacy amplification---after all, these suppression algorithms ensure similar $\S(D)$ and $\S(D')$. In addition, the privacy parameters of suppression naturally cannot be less than those of uniform Poisson sampling (see \Cref{prop:final} and \Cref{cor:final}).  

In summary, due to the properties of DP, uniform suppression provides lower privacy parameters than targeted suppression. Our theorems confirm that privacy amplification is still achievable for other suppression strategies close to uniform suppression. However, the privacy parameters will easily start to increase when specifically protecting database outliers or vulnerable records if these records vary greatly between neighboring databases. This is especially true for deterministic suppression, which is more unforgiving than probabilistic suppression.

\section{The Effect of Suppression on Utility}\label{sec:experiments}

Having seen how suppression affects privacy, we are now interested in how it affects the mechanisms utility guarantees. In this section, we perform the same empirical evaluations we conducted for sampling in \Cref{sec:samplingExperiments}, that is, we compute the empirical evaluations of the utility guarantees of $\M\circ\S$ and compare them to those of $\M$. 

In this case, we believe that we cannot formulate a fair evaluation using deterministic suppression: Any reasonable suppression strategy that is database-independent leads to a very large (or infinite) privacy budget for $\M\circ\S$, and database-dependent suppression requires defining global statistics about the database, which could lead to an unfair comparison since it requires knowledge about the values in the database. Therefore, our evaluation centers on the outlier-score suppression algorithm of \Cref{sec:outlierscoresuppressiontheorem}. 

Note too that we only expect utility gains with a privacy amplification:
If $\S$ actually increases the privacy parameters, we would generally obtain a noise amplification---not reduction---when translating the privacy parameters down to those of $\M$, and $\M\circ\S$ would expectedly provide worse utility than $\M$ at fixed privacy levels under our utility metrics.

\subsection{Experiment Description and Setup}

Conducting experiments similar to those in \Cref{sec:samplingExperiments}, we will compare the utility guarantees of $\M\circ\S$ to those of $\M$ under the same privacy guarantees.
To ensure the same privacy level, we conduct the analogous transformation: By 
\Cref{th:outlierscoresuppression}, if $\M$ satisfies $(\varepsilon,\delta)$-DP, then $\M\circ\S$ satisfies $(\varepsilon^\S,\delta^\S)$-DP with $\varepsilon^\S=\varepsilon^\S(\varepsilon,m,M)$ and $\delta^\S=\delta^\S(\delta,m)$. 
So, to ensure that $\M\circ\S$ also satisfies $(\varepsilon,\delta)$-DP, we impose $\M$ to be $(\varepsilon'',\delta'')$-DP such that $\varepsilon^\S(\varepsilon'',m,M)=\varepsilon$ and $\delta^\S(\delta'',m)=\delta$.
We note that this process requires $\varepsilon^\S(\varepsilon,m,M)$ to have an inverse $\varepsilon''$ with respect to $\varepsilon$, which is not possible when $\varepsilon<\varepsilon^\S(0,m,M)$. This limitation is reflected by the unfilled areas in our plots.

We use the same mechanisms, utility metrics, databases, and privacy parameters $\varepsilon$ and $\delta$ as in \Cref{sec:samplingExperiments}. As the mechanisms are randomized, we also compute $u(\M,D)$ and $u(\M\circ\S,D)$ the same amount of times as before, and provide their means.
The only addition for this experiment is the distance function $\dist$ for $\S$, directly linked to how records are suppressed. In this case, we select distances that intuitively represent ways of deleting records and mimic potential choices made by data curators.

\textbf{Distance functions for $\S$.}
For the mean calculation, we select the absolute difference between values (i.e., the $L^k$ distance in $\R$ for any $k\in\N$). Thus, the suppression algorithm deletes with higher probability the values that are the furthest away from others in a weighted manner. Note that the mean minimizes this average distance, so records closer to it are less likely to be deleted.

For the mode calculation, we select $\dist$ as the discrete metric (i.e., $\dist(x,y)=1$ if $x\neq y$ and $\dist(x,x)=0$). This ensures that the values with higher counts in the database will be deleted with a lower probability than those with fewer.

For the clustering mechanisms, we choose $\dist$ as the $L^2$ distance between records. This is the same distance function used in the mechanisms for the assignment of clusters/medians and in their respective utility metrics.  

\subsection{Experimental Results}\label{sec:experimentsandresults}

To show whether $\M\circ\S$ can preserve utility better than $\M$, we plot the utility difference $u(\M,D)-u(\M\circ\S,D)$ (see \Cref{fig:Experiment2}).
Since higher values of $u(\M,D)$ are associated with worse utility, we obtain that $\M\circ\S$ provides better utility than $\M$ when $u(\M,D)-u(\M\circ\S,D)$, or the plot values, are positive. For each $\varepsilon\in\{0.25,0.5,1,2\}$, we provide plots of $u(\M,D)-u(\M\circ\S,D)$ with respect to $m$ and $M$ (similar to the privacy plots in \Cref{fig:SuppressionTheoremEps}).  
We compute the difference value at the points where $m$ and $M$ both are in $\{0.1,\dots,0.9\}$, with the precise value shown in the plots. The colors are then filled by triangulation, with the color grading being set to yellow for $0$, red for negative values ($\M\circ\S$ provides worse utility), and green for positive values ($\M\circ\S$ provides better utility). The blue line corresponds to the values of $m$ and $M$ such that $\varepsilon^\S(\varepsilon,m,M)=\varepsilon$ (as in \Cref{fig:SuppressionTheoremEps}).
 
All results across all databases and noise variations are qualitatively very similar, and we thus only plot representative examples in this section (\Cref{fig:Experiment2}). All other plots are included as a gallery in Appendix~\ref{sec:plotgallery}. 

\begin{figure*}[t]
    \centering
    \includegraphics[width=0.25\textwidth]{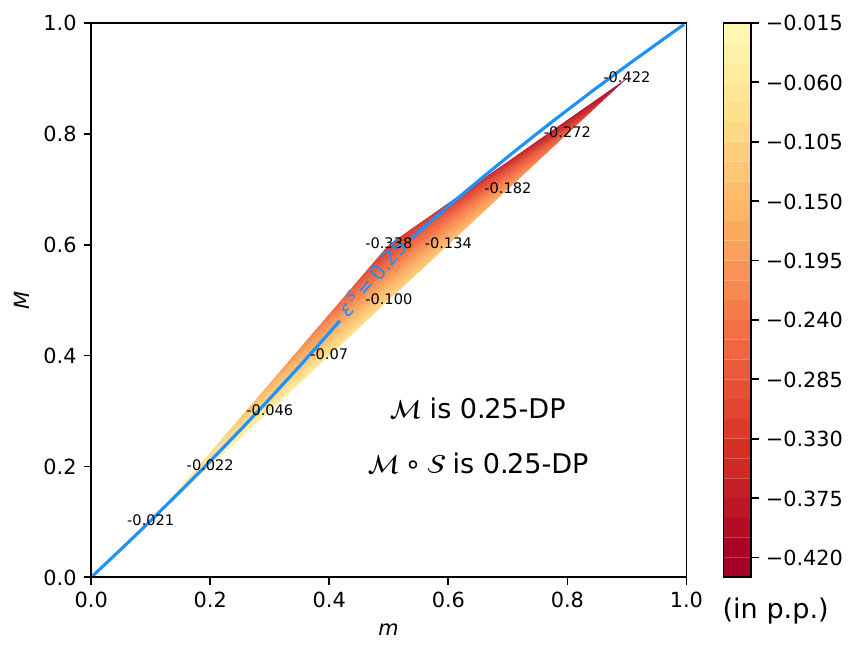}%
    \includegraphics[width=0.25\textwidth]{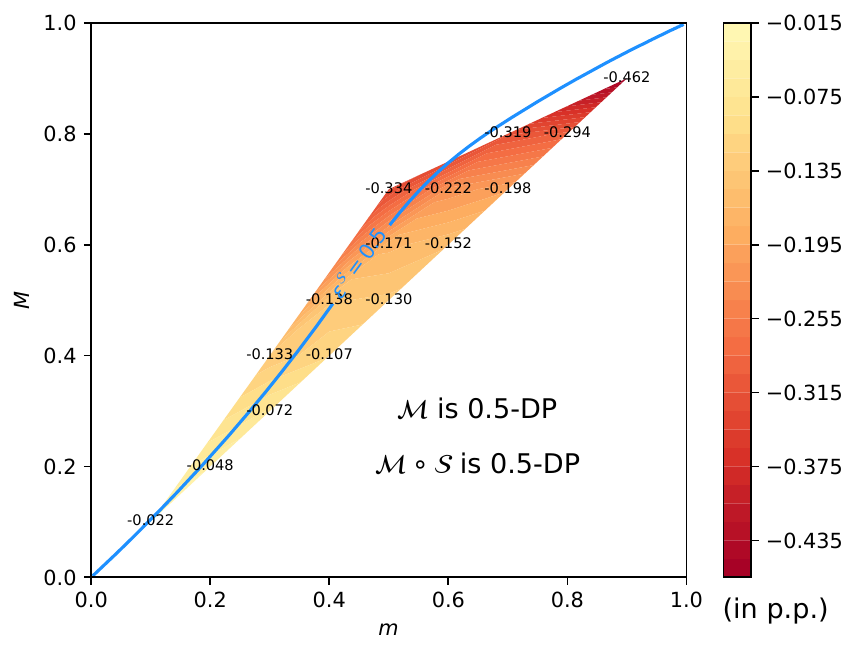}%
    \includegraphics[width=0.25\textwidth]{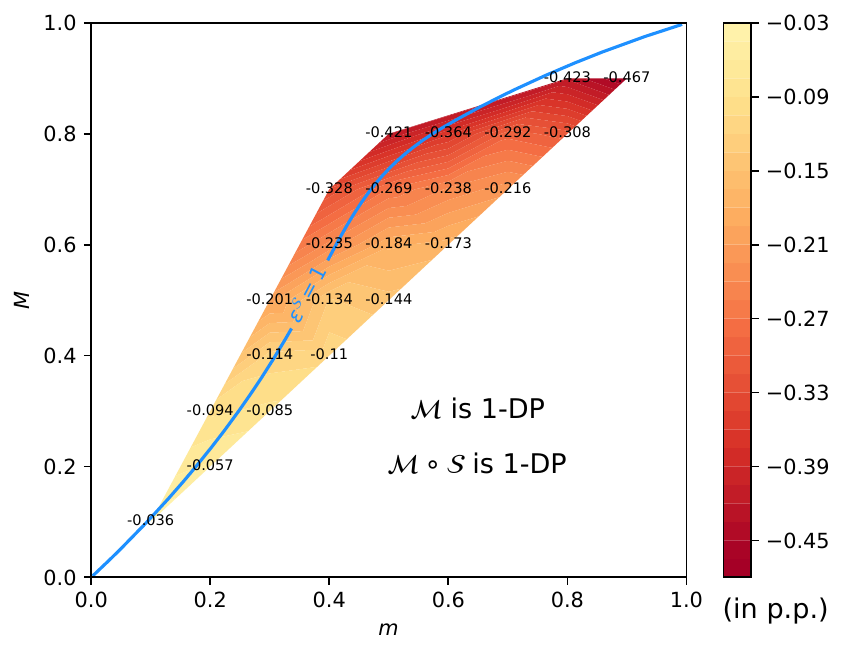}%
    \includegraphics[width=0.25\textwidth]{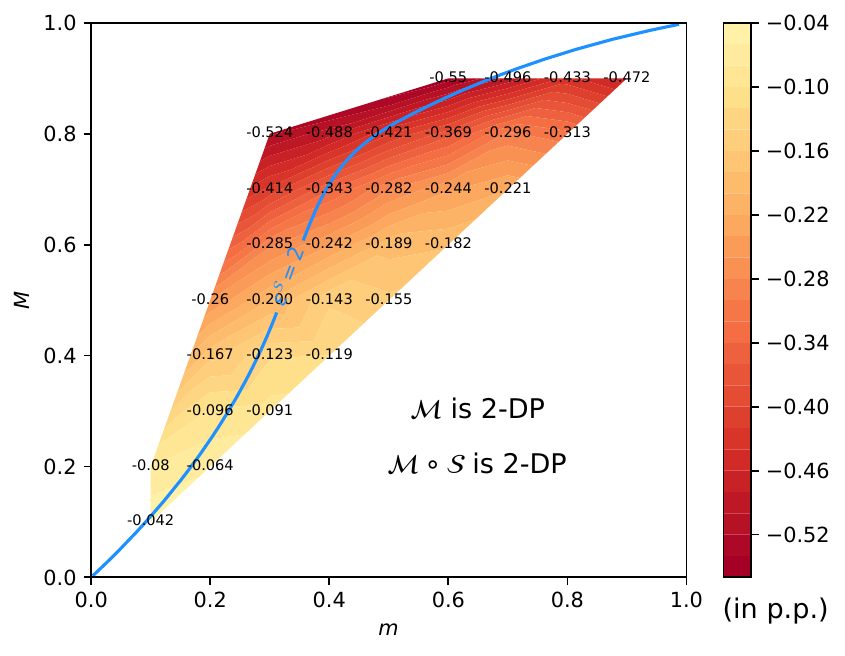}%

    \includegraphics[width=0.25\textwidth]{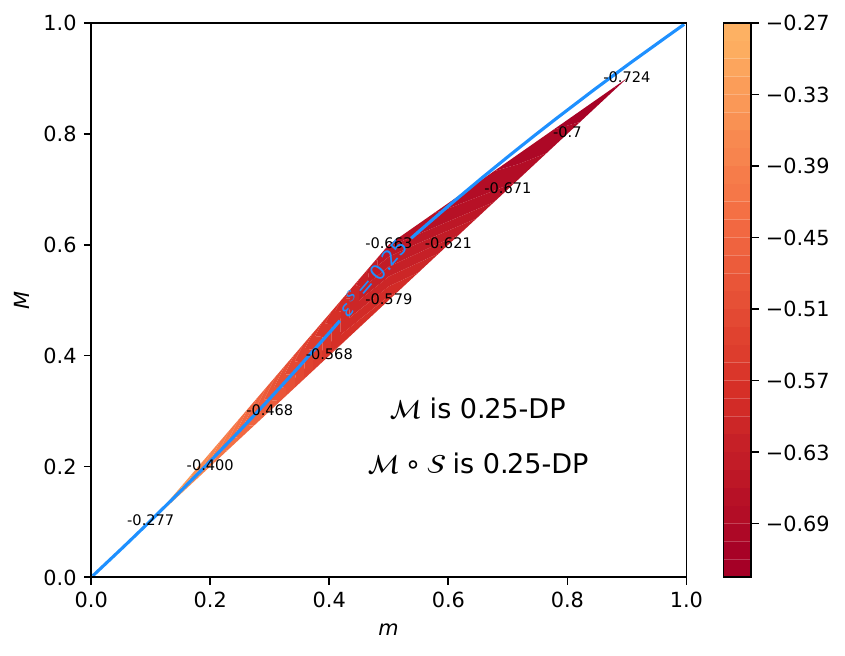}%
    \includegraphics[width=0.25\textwidth]{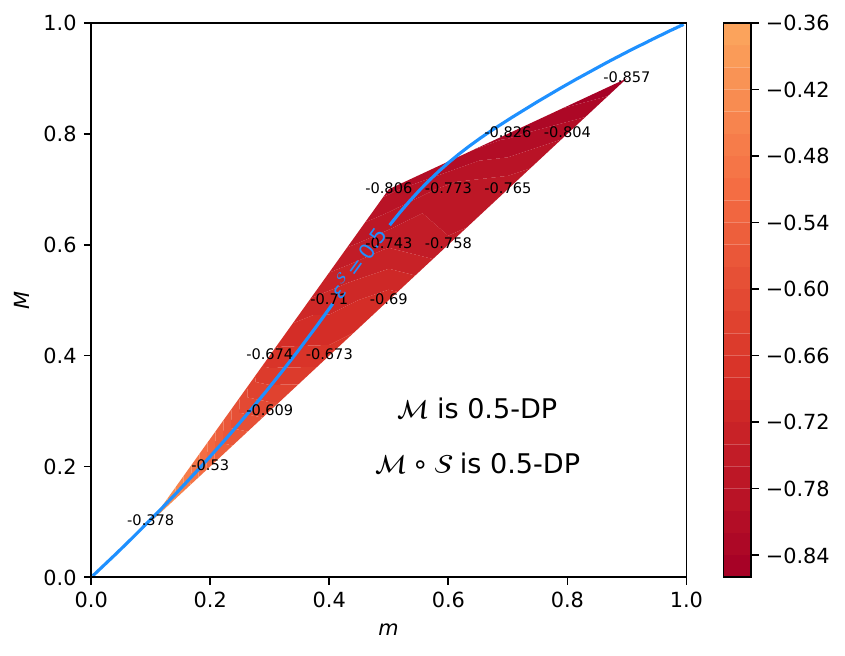}%
    \includegraphics[width=0.25\textwidth]{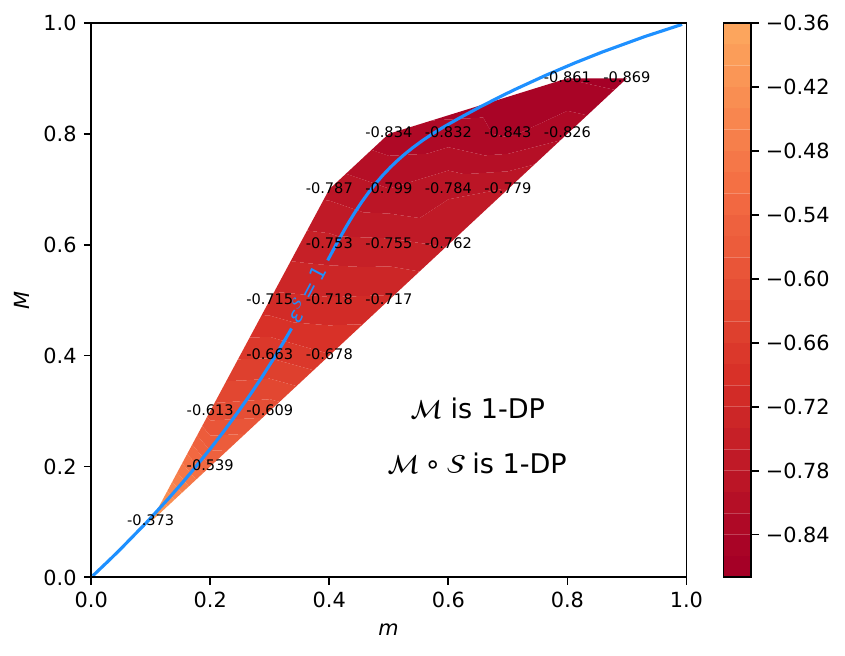}%
    \includegraphics[width=0.25\textwidth]{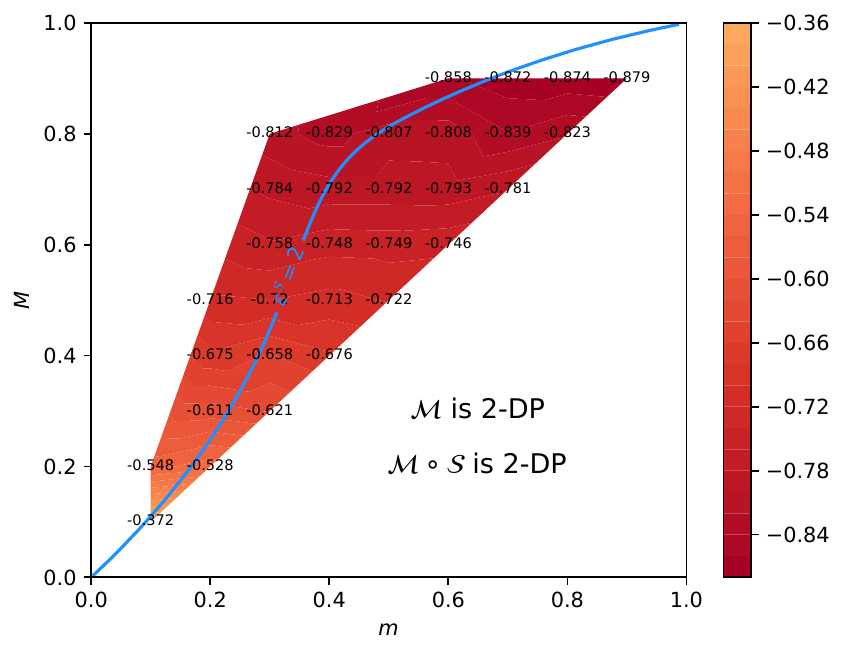}%

    \includegraphics[width=0.25\textwidth]{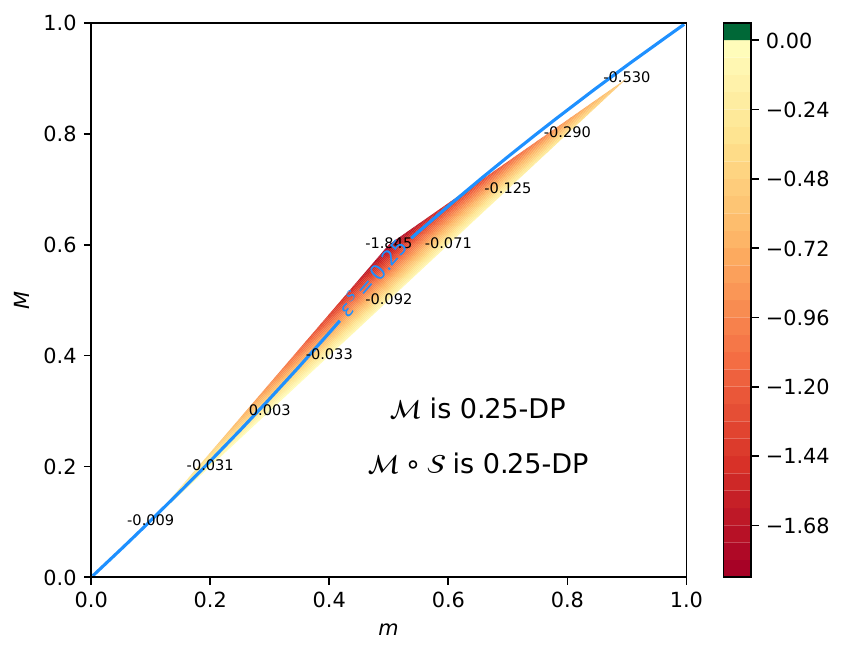}%
    \includegraphics[width=0.25\textwidth]{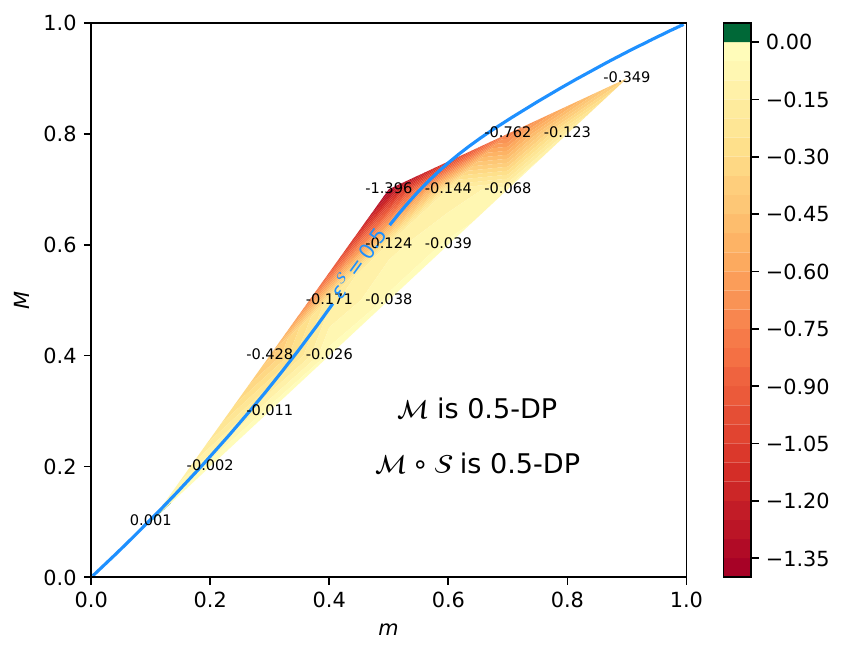}%
    \includegraphics[width=0.25\textwidth]{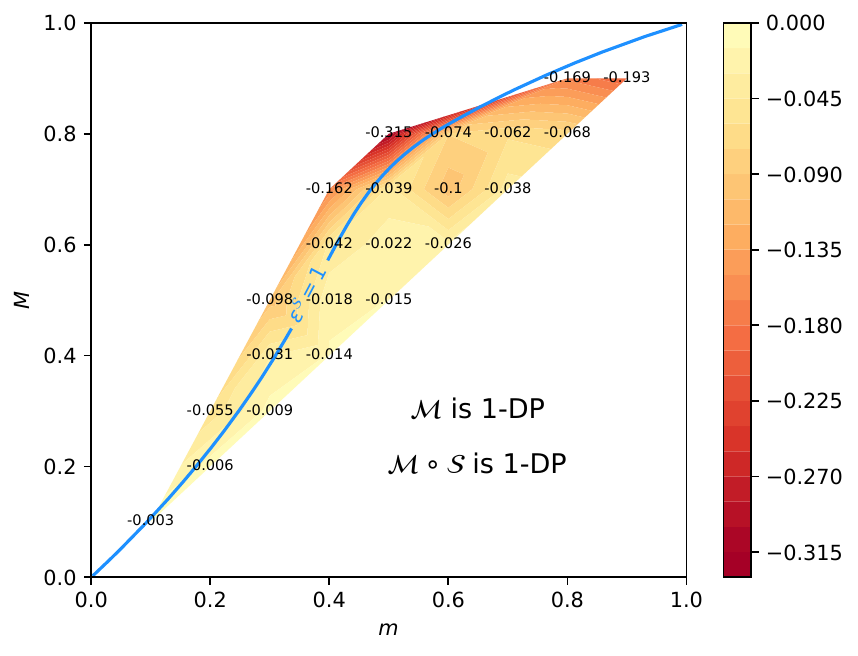}%
    \includegraphics[width=0.25\textwidth]{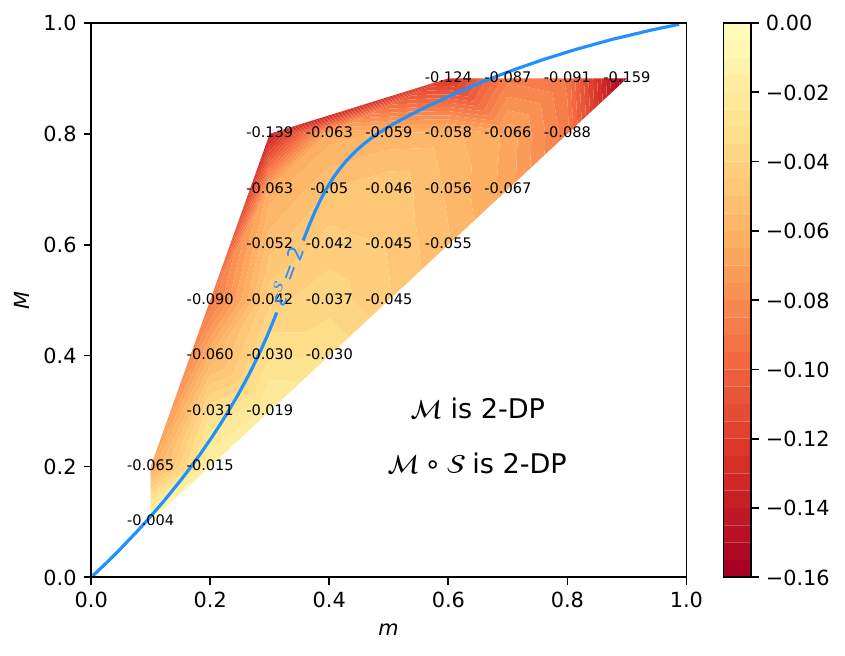}%
    \caption{Plots of the utility difference of $\M$ minus that of $\M\circ\S$ (i.e., $u(\M,D)-u(\M\circ\S,D)$) for the Adult database (and \texttt{age} column). The shown mechanisms are (top row) NoisyAverage with Laplace mechanisms; (middle row) RNM with Laplace noise; and (bottom row) DPLloyd clustering algorithm.}
    \label{fig:Experiment2}
\end{figure*}

\textbf{Results.} Our main observation is that there are almost no points where outlier-score suppression improves utility (cf., for instance, \Cref{fig:Experiment2}). This holds even when the privacy amplification is compensated with lower noise (i.e., the values between the diagonal and the blue line in the plots).

In general, the figures also show how quickly privacy degrades as more and more data is suppressed, with some plot values largely increasing around $m=M=0.9$. The difference also becomes smaller near $m=M=0.1$, but we do not plot for smaller values because numerical inaccuracies begin to cause distortions. 
In the great majority of cases, the values of $m$ and $M$ that show the least utility loss are those on the diagonal, corresponding to uniform Poisson sampling
(when comparing similar proportions of suppressed records). Our results also nearly always show that the utility difference under suppression worsens as $\varepsilon$ increases (for fixed $m$ and~$M$).

We now look at the three main experiments individually. \Cref{fig:Experiment2} (top row) shows the results for the NoisyAverage mechanism with Laplace noise, which rarely sees a utility gain. However, the actual MPE difference remains at insignificantly low levels all throughout, around less than $0.6\,\text{p.p.}$ (and sometimes even as low as $0.021\,\text{p.p.}$).

Similarly to the sampling counterpart, the utility of mode preservation decreases significantly when suppression is applied.
As shown in \Cref{fig:Experiment2} (middle row), there is up to an $87.9\,\textrm{p.p.}$ difference between the percentages of (in)correctly returning the mode of $\M$ and $\M\circ\S$ for RNM with Laplace noise on the \texttt{age} column (Adult database). 
These large differences are observed for values of $m$ and $M$ closer to $1$, which correspond to deleting a large portion of the database. 
The utility values, however, strongly depend on the record distribution in the database, similar to our observations in \Cref{sec:samplingExperiments}. 

The evaluation on clustering differs from the others, in the sense that the chosen utility metric does not assume the original database to be the ground truth or the base measurement. This means that the effect of suppression on the measured utility should theoretically be less damaging.
Nevertheless, our results on $k$-median and DPLloyd show the same phenomenon as the mean computation: $\M\circ\S$ provides worse utility than $\M$, with the difference remaining small all throughout (see \Cref{fig:Experiment2}, bottom row). DPLloyd shows some extremely small improvement over $\M$ for $m=M=0.1$ for $\varepsilon=0.25$ and $\varepsilon=0.5$ values, but these could be due to rounding errors caused for these low parameters.

\subsection{The Impact of Suppression on Utility}

Finally, we investigated whether the privacy amplification achieved by outlier-score suppression could result in less perturbation of the DP mechanism $\M$ such that the overall utility of $\M\circ\S$ is greater than $\M$ at fixed privacy levels. Our results show that this is not generally possible: The utility loss from outlier-score suppression carries much more weight in the utility measurement than the DP perturbation, and almost all reductions in the DP perturbation achieved through privacy amplification are too low to benefit the tradeoff. 
In particular, this effect appears even in cases where the utility loss is insignificantly small, such as in the NoisyAverage experiments.
Our results here thus largely follow the previous observations we made for sampling in \Cref{sec:samplingExperiments}. 
Moreover, we also note that most of our plots show a smaller utility difference on the diagonal. Hence, we see that utility is less affected by uniform Poisson sampling than by outlier-score suppression.

We believe that these results may be unexpected in some cases: Our method of assigning outliers depends on the chosen distances, which are deliberately selected to reduce privacy challenges and enhance utility. For instance, the distance chosen for the mode computation ensures that records with higher counts are less likely to be deleted than those with lower counts. However, in this case, the utility loss from deleting records is much greater than the utility gain from this selective deletion: Although the relative frequency of the mode in the database theoretically increases, the overall database size decreases, bringing the mode closer in count to the other records. This is why we obtain a significant loss of utility. Nevertheless, the fact that this type of suppression performs worse than uniform suppression, even when deleting the same proportion of records, offers an interesting insight into the privacy--utility tradeoff. We attribute this to the greater noise reduction in uniform suppression compared to non-uniform suppression due to the larger privacy amplification.

\section{Conclusion}\label{sec:conclusion}

Though sampling can provide orthogonal benefits like reducing time complexity, our privacy study shows that classic DP mechanisms without sampling consistently achieve a better privacy--utility tradeoff than the mechanisms with sampling, even when accounting for the noise reduction potentially gained by privacy amplification. 
These results motivated us to study the actual effect on the DP privacy--utility tradeoff when records are deleted as wanted. Here, we show that the positive suppression effects enjoyed in such privacy notions as $k$-anonymity do not transfer, mainly due to the DP definition itself. Since DP must account for any change between databases, more flexible strategic suppression algorithms that delete in a per-database setting come either at a weaker privacy amplification or, most often, at a privacy degradation.
 
We observe that
our theorems in \Cref{sec:suppressiontheorem} describe the privacy amplification---or reduction---effect that any suppression strategy can enjoy. This is a new and rigorous insight into the effects that any conceivable way of records omission may exert on the privacy guarantees of DP mechanisms.

Our evaluation on outlier-score suppression yields the same results as in sampling. For all databases and mechanisms tested, we found that utility is reduced compared to analyses without this preprocessing step at fixed privacy levels. However, the difference in error is quite insignificant in some cases.
We conclude that, in our case, the potential utility gain from modifying the privacy parameters is insufficient to overcome the substantial utility loss caused by omitting records. In addition, since our tested mechanisms include the canonical building blocks of most DP mechanisms, we expect this result to extend to most of them. Furthermore, our results show that, in the majority of cases, uniform Poisson sampling provides the least utility loss among our tested suppression algorithms.

As future work, we will extend our results and evaluations to cover other suppression strategies and mechanisms and evaluate the effect of sampling and suppression over other DP variants, like bounded DP or Rényi DP. An interesting evaluation would be to see whether suppression can improve utility when measuring utility according to some ground truth rather than the input database. 

In summary, our study provides new insights into sampling and suppression in DP, showing the particular need for balancing utility of these techniques against their demonstrated effects on the DP privacy--utility tradeoff. Overall, we show that sampling and outlier-score suppression both negatively impact the privacy--utility tradeoff, rendering the application of both techniques in DP questionable in this regard. 

\cleardoublepage
\appendix
\section*{Ethical Considerations}

The ethical implications of this work were thoroughly discussed. The authors of this paper declare the following:

\textbf{Basic Ethical Principle and Stakeholder Analysis:} This work was conducted following established ethical principles, guidelines, and best practices in data privacy. The focus is on differential privacy (DP), an important tool for protecting data with formal privacy guarantees. Our work deepens the understanding of sampling, a well-known technique in DP, and suppression, a technique used in privacy contexts but not in DP.
Given the importance and popularity of DP, sampling, and suppression, our work may impact different members of society as it reveals vulnerabilities of these tools. In particular, we identify two main stakeholder groups: (1)~individuals whose data were previously protected using a DP mechanism with sampling, and (2)~data curators and practitioners who have developed or wish to develop mechanisms with sampling.

While our study shows that sampling and our suppression worsen utility, they do not compromise the level of privacy or protection they or any mechanism using them provides. Therefore, the individuals' privacy remains unaffected, regardless of our results. Our work does not cause any direct or indirect harm to the first stakeholder group or their data.

Our study may impact the second group, as it reveals a potential weakness in their mechanisms. While our work raises questions about the benefits of sampling, it does not weaken or invalidate published mechanisms that use sampling. Rather, it reveals the possibility of refining the privacy--utility tradeoff, which may encourage data practitioners to conduct new experiments on previously published work. Similarly, our study emphasizes the need to further evaluate the effect of sampling when designing new sampling-based mechanisms.

\textbf{Bias and Fairness:} We acknowledge the potential biases that can accompany any empirical evaluation. To mitigate these biases, we experimented with multiple mechanisms, databases, and parameters, and performed multiple iterations to guarantee statistical accuracy. We paid careful attention to avoid undue generalization of our results. Our experimentation and results only cover the privacy--utility tradeoff of sampling and our suppression in unbounded DP.
We acknowledge that our conclusions may differ for other scenarios, suppression algorithms, or DP variants. We reiterate the need to evaluate the actual effect of sampling (or any preprocessing). 

\textbf{Decision:} Our work raises awareness of sampling and suppression, improving their understanding of these techniques and discussing how they do not improve the privacy--utility tradeoff of DP mechanisms in our settings. Our paper also opens the door to testing mechanisms that have previously used sampling, potentially leading to improvements. Although our results are negative, we hope that they can guide the future design of more effective privacy-preserving mechanisms. The potential improvements that follow from our work have led us to submit it for consideration at USENIX. 

\textbf{Subjects and Public Database Use:} Our work does not involve any direct human or animal participation. For our experiment, we used either synthetic databases or commonly used, publicly available standard databases derived from the US Census~\cite{becker1996adult,brand2002reference} and Irish Census Data~\cite{ayala-rivera2016cocoa}. These databases had all been previously and thoroughly anonymized to prevent the identification of any individual and are explicitly released for public research, and statistical and educational use. In our work, we use these databases only for statistical research and do not intend to deanonymize them or perform data linking that might risk identifying the participants or their data.

\textbf{Competing Interests:} The authors have no relevant financial or nonfinancial interests to disclose, and are not affiliated with any organization or entity with a financial or nonfinancial interest in the content or implications of this paper.

\textbf{Further ethical considerations:} The research activities conducted in this paper did not have the potential to negatively affect the authors or their health. This work has been conducted within the framework of the law.

\section*{Open Science}

To promote transparency and reproducibility, we make available an artifact that includes all the code from our work at: \urlartifact. This includes the code for all the experiments (i.e., the mean and mode computation, and the two clustering mechanisms) for both sampling and suppression, as discussed in \Cref{sec:samplingExperiments,sec:experiments}, the code for the numerical computation for \Cref{th:outlierscoresuppression} (see \Cref{remark:computation,remark:computationInverse}), and the code to generate all plots in the paper. All databases used are included in the artifact because they are from publicly accessible data. We provide all the necessary code to reproduce every result and plot of our paper. 

\section*{Acknowledgments}
Javier Parra-Arnau is a ``Ramón y Cajal'' fellow (ref.\ RYC2021-034256-I) funded by the MCIN\slash AEI\slash 10.13039\slash 501100011033 and the EU ``NextGenerationEU''/PRTR. This work was also supported by (i)~the ``DIstributed Smart Communications with Verifiable EneRgy-optimal Yields (DISCOVERY)'' project (PID2023-148716OB-C32), funded by the same two institutions above; (ii)~the Generalitat de Catalunya, under the AGAUR grant ``2021 SGR 01413''; (iii)~KASTEL Security Research Labs, Karlsruhe, funded by the Topic Engineering Secure Systems of the Helmholtz Association (HGF); and (iv)~Germany’s Excellence Strategy – EXC 2050/1 – Project ID 390696704 – Cluster of Excellence ``Centre for Tactile Internet with Human-in-the-Loop'' (CeTI). This paper has been edited by our textician, Daniel Shea. We would also like to thank Daniel Schadt, Patricia Guerra-Balboa, and Felix Morsbach for their useful help and comments.

\bibliographystyle{plain}
\bibliography{bibliography}

\newpage
\onecolumn

\section{Plots Gallery}\label{sec:plotgallery}

In this section, we present a complete gallery of all plots from our experiments. \Cref{sec:plots:PoissonSampling1,sec:plots:PoissonSampling2,sec:plots:PoissonSampling3} show the plots from the sampling experiment described in \Cref{sec:samplingExperiments}; and \Cref{sec:plots:SuppressionwithEpsDeltaChange1,sec:plots:SuppressionwithEpsDeltaChange2,sec:plots:SuppressionwithEpsDeltaChange3} show the plots from the suppression experiment with the privacy-parameter change described in \Cref{sec:experiments}. Furthermore, we provide additional plots showing the difference in utility values between the mechanism with and without suppression \textit{without the noise reduction} (i.e., in this case, $\M$ and $\M\circ\S$ are compared without changing the privacy parameters) in 
\Cref{sec:plots:SuppressionwithoutEpsDeltaChange1,sec:plots:SuppressionwithoutEpsDeltaChange2,sec:plots:SuppressionwithEpsDeltaChange3}.

Note that the plots of the mode computation for the \texttt{hours-per-week} column in the Adult database and for the \texttt{HighestEducationCompleted} column in the Irish database are correct. In these cases, the maximum corresponds to a highly representative value in the database. For instance, the maximum value of the \texttt{hours-per-week} column in the Adult database corresponds to over half of the records. Thus, both the mechanism with and without sampling/suppression output a perfect failure probability of 0, as explained in \Cref{sec:samplingExperiments}. Nevertheless, we include these plots for completeness.   

\input{plotsgallery/plotsPoisson}

\input{plotsgallery/plotsExperiment2}

\input{plotsgallery/plotsExperiment1}

\newpage

\section{Proofs, Additional Theorems, and Remarks}\label{sec:proofs}

\Cref{tab:notation} compiles the notation used in the proofs.

\begin{table}[H]
    \centering
    \begin{tabular}{|c|c|}
        \hline
        Symbol & Meaning \\
        \hline\hline
        $\X$ & Set of possible data records \\
        $\D$ & Class of finite databases drawn from $\X$ \\
        $D,D'$ & A pair of databases \\
        $D\sim D'$ & $D$ and $D'$ are neighboring databases \\
        $|D|$ & Size of $D$ (number of records) \\
        $x$ & Data record (element of $\X$)\\
        $\M$ & A randomized mechanism with domain $\D$ \\ $\Ran\coloneqq\Range(\M)$ & Set of possible outputs of all $\M(D)$ with $D\in\D$ \\
        $\meas$ & Measurable subset of $\Ran$ \\
        $\S$ & Suppression algorithm\\
        $C\subseteq D$ & Subdatabase of $D$ \\
        $D_{+y}$ & Database $D\uplus\{y\}$ \\ 
        $\supp(X)$ & Support of the random variable $X$ \\
        $[n]$ & Set $\{1,\dots,n\}$ \\
        \hline
    \end{tabular}
    \caption{\small Summary of the notation used in this paper.}
    \label{tab:notation}
\end{table}

\subsection{Proofs of \Cref{sec:deterministicsuppression}: \nameref*{sec:deterministicsuppression}}

\THdeterminisitcsuppression*
\begin{proof}
    Note that if $\Delta_{\intset}\S=\infty$, then we are done. We assume then that $\Delta_{\intset}\S$ is finite. 
    
    We fix $D,D'\in\D$. Suppose $k=\dist_{\intset}(\S(D),\S(D'))$, then we have that $k\leq\Delta_{\intset}\S<\infty$ and there exists a chain of $k+1$ databases $D_0,\dots,D_k\in\intset$ such that
    \[
        \S(D) = D_0 \sim D_1 \sim \dots \sim D_k = \S(D'),
    \]
    and applying the definition of DP over $\M|_{\intset}$, we obtain for all measurable $\meas\subseteq\mathcal{R}\coloneqq\Range(\M)$,
    \begin{gather*}
        \Prob\{\M|_{\intset}(\S(D))\in\meas\} = \Prob\{\M|_{\intset}(D_0)\in\meas\} \leq \e^{\varepsilon_{\intset}}\Prob\{\M|_{\intset}(D_1)\in\meas\} + \delta_{\intset} \leq \e^{\varepsilon_{\intset}}(\e^{\varepsilon_{\intset}}\Prob\{\M|_{\intset}(D_2)\in\meas\} + \delta_{\intset}) + \delta_{\intset} \leq \cdots \\\leq \e^{k\varepsilon_{\intset}}\Prob\{\M|_{\intset}(D_k)\in\meas\} + \delta_{\intset}\sum^{k-1}_{i=0}\e^{\varepsilon_{\intset} i} = \e^{k\varepsilon}\Prob\{\M|_{\intset}(\S(D'))\in\meas\} +  \delta_{\intset}\sum^{k-1}_{i=0}\e^{\varepsilon_{\intset} i} \leq\e^{\varepsilon\Delta_{\intset}\S}\Prob\{\M|_{\intset}(\S(D'))\in\meas\} +  \delta_{\intset}\sum^{\Delta_{\intset}\S-1}_{i=0}\e^{\varepsilon_{\intset} i}.
    \end{gather*}
    
    Since the last bound is independent of the choice of $D,D'\in\D$, we obtain the result. \qedhere
\end{proof}

\begin{proposition}\label{prop:deterministicsuppressiontight}
    The bound provided by \Cref{th:deterministicsuppression} is tight for all suppression algorithms $\S$ if $\delta_{\intset}\sum^{\Delta_\intset\S-1}_{i=0}\e^{\varepsilon_{\intset} i}<1$.
\end{proposition}
\begin{proof}
    The proof is a corollary of the existence of a mechanism $\M$ such that the group privacy bound is tight. Selecting such a mechanism, like the following one, suffices.

    First note that if $(\varepsilon_\intset,\delta_\intset)=(0,0)$ or $\Delta_\intset\S=0$, then the result holds since $\M\circ\S$ does not depend of its input, and thus is $(0,0)$-DP. We will now cover the rest of the cases. 

    We fix the suppression algorithm $\S$ and assume that $\Delta_\intset\S<\infty$. By definition, there exists a database $D\in\D$ such that
    \[
        \sup_{\substack{D'\in D:\\D'\sim D}} \dist_\intset(\S(D),\S(D'))=\Delta_\intset\S.
    \]
    
    We define the finite sequence $\hat{D}$ such that the first element is $\delta_\intset$ and it is followed, for all $j=1,\dots,\Delta_\intset\S$ in order, by $\ceil{\e^{\varepsilon_{\intset}\,j}}$ copies of $\frac{\e^{\varepsilon_{\intset}\,j}}{\ceil{\e^{\varepsilon_{\intset}\,j}}}\delta_\intset$. Then, for all $n\in[\Delta_\intset\S]$, we consider the sequence $D^{(n)}$ with indices in $\Z$ such that it contains $\hat{D}$ with the first element $\frac{\e^{\varepsilon_{\intset}\,n}}{\ceil{\e^{\varepsilon_{\intset}\,n}}}\delta_\intset$ at index $1$ and is preceded and succeeded by only zeros, i.e., informally,
    \[
        D^{(n)}=(\cdots,0,0,\text{``the elements of $\hat{D}$ in order''},0,0,\cdots).
    \]

    Note that the sum of the elements in $D^{(n)}$ is exactly $\delta_\intset\sum^{\Delta_\intset\S-1}_{i=0}\e^{\varepsilon_\intset i}$, and that every element $D^{(n)}_i$ in $D^{(n)}$ is between $0$ and $\delta_\intset$. Note too that any two elements $D_1,D_2$ in $D^{(n)}$ satisfies $D_1-\e^{\varepsilon_\intset}D_2 \leq \delta_\intset$. It is clear that if $D_1$ or $D_2$ equals 0, and otherwise, we have that
    \begin{align*}
        D_1-\e^{\varepsilon_\intset}D_2 
        = \delta_\intset\bigg(\frac{\e^{\varepsilon_\intset\,i_1}}{\ceil{\e^{\varepsilon_\intset\,i_1}}}-\frac{\e^{\varepsilon_\intset(i_2+1)}}{\ceil{\e^{\varepsilon_\intset\,i_2}}}\bigg)
        = \delta_\intset\bigg(\frac{\ceil{\e^{\varepsilon_\intset\,i_2}}\e^{\varepsilon_\intset\,i_1}-\ceil{\e^{\varepsilon_\intset\,i_1}}\e^{\varepsilon_\intset(i_2+1)}}{\ceil{\e^{\varepsilon_\intset\,i_1}}\ceil{\e^{\varepsilon_\intset\,i_2}}}\bigg) \\
        \leq \delta_\intset\bigg(\frac{(\e^{\varepsilon_\intset\,i_2}+1)\e^{\varepsilon_\intset\,i_1}-\e^{\varepsilon_\intset\,i_1}\e^{\varepsilon_\intset(i_2+1)}}{\ceil{\e^{\varepsilon_\intset\,i_1}}\ceil{\e^{\varepsilon_\intset\,i_2}}}\bigg)
        = \delta_\intset\bigg(\frac{\e^{\varepsilon_\intset\,i_2}\e^{\varepsilon_\intset\,i_1}(1-\e^{\varepsilon_\intset})+\e^{\varepsilon_\intset\,i_1}}{\ceil{\e^{\varepsilon_\intset\,i_1}}\ceil{\e^{\varepsilon_\intset\,i_2}}}\bigg)
        \leq \delta_\intset\bigg(1-\e^{\varepsilon_\intset} + \frac{1}{\ceil{\e^{\varepsilon_\intset\,i_2}}}\bigg),
    \end{align*}
    which is smaller than $\delta_\intset$ since
    \[  
        D_1-\e^{\varepsilon_\intset}D_2 \leq \delta_\intset \Longleftarrow 1-\e^{\varepsilon_\intset} + \frac{1}{\ceil{\e^{\varepsilon_\intset\,i_2}}}\leq 1 \Longleftrightarrow 0 \leq \e^{\varepsilon_\intset}\ceil{\e^{\varepsilon_\intset\,i_2}}.
    \]

    Since $\delta_\intset\sum^{\Delta_\intset\S-1}_{i=0}\e^{\varepsilon_\intset i}<1$, we have that
    \[
        \alpha \coloneqq \frac{(1-\sum^{\Delta_\intset\S-1}_{i=0})(1-\e^{-\varepsilon_\intset})}{1+\e^{\varepsilon_\intset}}\in(0,1]
    \]
    and, we consider the mechanism $\M|_\intset\colon\intset\to\Z$ such that 
    \[    
        \Prob\{\M|_\intset(C) = s\} = \alpha\e^{-\varepsilon_\intset|s-\Delta_\intset\S+d_C|} + D^{(d_C)}_{s}
    \]
    with $d_C\coloneqq\dist_{\intset}(\S(D),C)$. We denote $P_{d_C}(s)\coloneqq\Prob\{\M|_\intset(C) = s\}$ since this value only depends on $d_C$ and $s$. 
    We can see that this mechanism is well-defined since (i)~all probabilities sum to $1$: It is clear for $\varepsilon_\intset=0$ and, otherwise, 
    \begin{gather*}
        \sum_{\substack{s\in\Z}}P_{d_C}(s)=\sum_{\substack{s\in\Z\\s-\Delta_\intset\S+d_C< 0}}\alpha\e^{\varepsilon_\intset(s-\Delta_\intset\S+d_C)} + \alpha + \sum_{\substack{s\in\Z\\s-\Delta_\intset\S+d_C> 0}}\alpha\e^{-\varepsilon_\intset(s-\Delta_\intset\S+d_C)} + \sum_{\substack{s\in\Z}} D^{(d_C)}_s = \alpha\sum_{\substack{s\in\Z\\s> 0}}\e^{-\varepsilon_\intset s} + \alpha + \alpha\sum_{\substack{s\in\Z\\s> 0}}\e^{-\varepsilon_\intset s} +\delta_\intset\sum^{\Delta_\intset\S-1}_{i=0}\e^{\varepsilon_\intset i} \\
        = 2\alpha\frac{\e^{-\varepsilon_\intset}}{1-\e^{-\varepsilon_\intset}} + \alpha + \delta_\intset\sum^{\Delta_\intset\S-1}_{i=0}\e^{\varepsilon_\intset i}
        = \alpha\frac{1+\e^{-\varepsilon_\intset}}{1-\e^{-\varepsilon_\intset}} + \delta_\intset\sum^{\Delta_\intset\S-1}_{i=0}\e^{\varepsilon_\intset i} 
        = \bigg(1-\delta_\intset\sum^{\Delta_\intset\S-1}_{i=0}\e^{\varepsilon_\intset i}\bigg)+\delta_\intset\sum^{\Delta_\intset\S-1}_{i=0}\e^{\varepsilon_\intset i} = 1
    \end{gather*}
    and (ii)~all probabilities are between $0$ and $1$ (follows directly from case~(i) and $P_{d_C}(s)\geq0$). 
    
    We note that $\M|_\intset$ is constructed to satisfy 
    $(\varepsilon_\intset,\delta_\intset)$-DP tightly: Since the range of $\M|_\intset$ is discrete, seeing that $\M|_\intset$ is $(\varepsilon_\intset,\delta_\intset)$-DP is equivalent to seeing that  
    \begin{equation}
         P_{d_C}(s) \leq \e^{\varepsilon_\intset}P_{d_{C'}}(s) + \delta_\intset \label{eq:tightnessdeterministicexample}
    \end{equation}
    for all $s\in\Z$ and all unbounded-neighboring $C,C'\in\intset$. We fix $C$ and $C'$. Note that $\dist_\intset(C,C')=1$ and that $\abs{\dist_\D(\S(D),C)-\dist_\intset(\S(D),C')}=\abs{d_C-d_{C'}}\leq1$ by the triangular inequality. We note that if $d_C=d_{C'}$, then \Cref{eq:tightnessdeterministicexample} holds directly. Suppose, without loss of generality, that $d_{C'}=d_{C}+1$, and we see that the inequality \Cref{eq:tightnessdeterministicexample} is verified for all $s\in\Z$:
    \begin{itemize}
        \item If $s>d_C$:
        \begin{align*}
            P_{d_C}(s)
            &\leq\e^{\varepsilon_\intset}P_{d_C+1}(s)+\delta_\intset \\
            \Longleftrightarrow\alpha\e^{-\varepsilon_\intset(s-\Delta_\intset\S+d_C)} + D^{(d_C)}_{s}
            &\leq\e^{\varepsilon_\intset}(\alpha\e^{-\varepsilon_\intset(s-\Delta_\intset\S+d_C+1)} + D^{(d_C+1)}_{s})+\delta_\intset \\
            \Longleftrightarrow\alpha\e^{-\varepsilon_\intset(s-\Delta_\intset\S+d_C)} + D^{(d_C)}_{s}
            &\leq\alpha\e^{-\varepsilon_\intset(s-\Delta_\intset\S+d_C)} + \e^{\varepsilon_\intset}D^{(d_C+1)}_{s}+\delta_\intset \\
            \Longleftrightarrow D^{(d_C)}_{s} -\e^{\varepsilon_\intset}D^{(d_C+1)}_{s}
            &\leq \delta_\intset,
        \end{align*}
        which is satisfied as previously shown. The inverse inequality is also satisfied:
        \begin{align*}
            P_{d_C+1}(s)
            &\leq\e^{\varepsilon_\intset}P_{d_C}(s)+\delta_\intset \\
            \Longleftrightarrow\alpha\e^{-\varepsilon_\intset(s-\Delta_\intset\S+d_C+1)} + D^{(d_C+1)}_{s}
            &\leq\e^{\varepsilon_\intset}(\alpha\e^{-\varepsilon_\intset(s-\Delta_\intset\S+d_C)} + D^{(d_C)}_{s})+\delta_\intset \\
            \Longleftrightarrow\alpha\e^{-\varepsilon_\intset(s-\Delta_\intset\S+d_C+1)} + D^{(d_C+1)}_{s}
            &\leq\alpha\e^{-\varepsilon_\intset(s-\Delta_\intset\S+d_C-1)} + \e^{\varepsilon_\intset}D^{(d_C)}_{s}+\delta_\intset \\
            \Longleftrightarrow D^{(d_C+1)}_{s} -\e^{\varepsilon_\intset}D^{(d_C)}_{s}
            &\leq \delta_\intset + \alpha\e^{-\varepsilon_\intset(s-\Delta_\intset\S+d_C+1)}(\e^{2\varepsilon_\intset}-1).
        \end{align*}
        \item If $s\leq d_C$, the same inequalities can be constructed:
        \begin{align*}
            P_{d_C}(s)
            &\leq\e^{\varepsilon_\intset}P_{d_C+1}(s)+\delta_\intset \\
            \Longleftrightarrow\alpha\e^{\varepsilon_\intset(s-\Delta_\intset\S+d_C)} + D^{(d_C)}_{s}
            &\leq\e^{\varepsilon_\intset}(\alpha\e^{\varepsilon_\intset(s-\Delta_\intset\S+d_C+1)} + D^{(d_C+1)}_{s})+\delta_\intset \\
            \Longleftrightarrow\alpha\e^{\varepsilon_\intset(s-\Delta_\intset\S+d_C)} + D^{(d_C)}_{s}
            &\leq\alpha\e^{\varepsilon_\intset(s-\Delta_\intset\S+d_C+2)} + \e^{\varepsilon_\intset}D^{(d_C+1)}_{s}+\delta_\intset \\
            \Longleftrightarrow D^{(d_C)}_{s} -\e^{\varepsilon_\intset}D^{(d_C+1)}_{s}
            &\leq \delta_\intset + \alpha\e^{\varepsilon_\intset(s-\Delta_\intset\S+d_C)}(\e^{2\varepsilon_\intset}-1),
        \end{align*}
        and
        \begin{align*}
            P_{d_C+1}(s)
            &\leq\e^{\varepsilon_\intset}P_{d_C}(s)+\delta_\intset \\
            \Longleftrightarrow\alpha\e^{\varepsilon_\intset(s-\Delta_\intset\S+d_C+1)} + D^{(d_C+1)}_{s}
            &\leq\e^{\varepsilon_\intset}(\alpha\e^{\varepsilon_\intset(s-\Delta_\intset\S+d_C)} + D^{(d_C)}_{s})+\delta_\intset \\
            \Longleftrightarrow\alpha\e^{\varepsilon_\intset(s-\Delta_\intset\S+d_C+1)} + D^{(d_C+1)}_{s}
            &\leq\alpha\e^{\varepsilon_\intset(s-\Delta_\intset\S+d_C+1)} + \e^{\varepsilon_\intset}D^{(d_C)}_{s}+\delta_\intset \\
            \Longleftrightarrow D^{(d_C+1)}_{s} -\e^{\varepsilon_\intset}D^{(d_C)}_{s}
            &\leq \delta_\intset.
        \end{align*}
    \end{itemize}

    Having covered each case, we obtain that $\M|_\intset$ is $(\varepsilon_\intset,\delta_\intset)$-DP. We note that it is tightly DP since selecting $s=0$ and $C=\S(D)$, we have that
    \[
        P_{1}(0) = \alpha\e^{-\varepsilon_\intset(\Delta_\intset\S-1)} + D^{(1)}_{0} = \alpha\e^{-\varepsilon_\intset(\Delta_\intset\S-1)} + \delta_\intset
        =\e^{\varepsilon_\intset}(\alpha\e^{-\varepsilon_\intset\Delta_\intset\S} + 0)+\delta_\intset = \e^{\varepsilon_\intset}(\alpha\e^{-\varepsilon_\intset\Delta_\intset\S} + D^{(0)}_{0})+\delta_\intset
        =\e^{\varepsilon_\intset}P_{0}(0)+\delta_\intset.
    \]
    
    Now observe too that, for $\varepsilon_\intset\neq0$ and for all neighboring $C,C'\in\D$ with $d_{C'}=d_C+1$ and $\meas=\Z_{\leq0}$, 
    \begin{gather*}
        \Prob\{\M|_{\intset}(C)\in\meas\} = \sum_{s\leq0} \alpha\e^{-\varepsilon_\intset|s-\Delta_\intset\S+d_C|}+\sum_{s\leq0} D^{(d_C)}_s = \sum^\infty_{s=0} \alpha\e^{-\varepsilon_\intset(s+\Delta_\intset\S-d_C)}+\delta_\intset\sum^{d_C-1}_{i=0}\e^{\varepsilon_\intset i}
        = \alpha\frac{\e^{-\varepsilon(\Delta_\intset\S-d_C)}}{1-\e^{-\varepsilon_\intset}}+\delta_\intset\sum^{d_C-1}_{i=0}\e^{\varepsilon_\intset\,i}
    \end{gather*}
    and, thus,
    \[
        \Prob\{\M|_{\intset}(C')\in\meas\} = \alpha\frac{\e^{-\varepsilon_\intset(\Delta_\intset\S-d_C-1)}}{1-\e^{-\varepsilon_\intset}}+\delta_\intset\sum^{d_C}_{i=0}\e^{\varepsilon_\intset i}
        = \e^{\varepsilon_\intset}\bigg(\alpha\frac{\e^{-\varepsilon_\intset(\Delta_\intset\S-d_C)}}{1-\e^{-\varepsilon_\intset}}+\delta_\intset\sum^{d_C-1}_{i=0}\e^{\varepsilon_\intset i}\bigg) + \delta_\intset = \e^{\varepsilon_\intset}\Prob\{\M|_{\intset}(C)\in\meas\} + \delta_\intset. 
    \]

    Note that the last equality also holds for $\varepsilon_\intset=0$. Thus, finally selecting $C=\S(D)$ and $C'=\S(D')$ with $D'$ the neighboring database to $D$ such that $\dist_\intset(\S(D),\S(D'))=\Delta_\intset\S$, we obtain recursively that 
    \[
        \Prob\{\M(\S(D'))\in\meas\} = \e^{\varepsilon_\intset\,\Delta_\intset\S}\Prob\{\M(\S(D))\in\meas\} + \delta_\intset\sum^{\Delta_\intset\S-1}_{i=0}\e^{\varepsilon_\intset i},
    \]
    which proves tightness of \Cref{th:deterministicsuppression}.

    Finally, the case where $\Delta_\intset\S=\infty$ also follows from the previous proof. Since we can construct $D,D'$ at distance $k$ for all $k\in\N$, taking the limit leads to $\infty$-DP.
\end{proof}

\begin{proposition}\label{prop:determinsticsuppressionwithinfinitesensitivity}
    Let $\D=\D_{\X}$ be the class of all databases with elements drawn from $\X$ with $|\X|\geq2$. Let $\dist\colon \X\times\X\to[0,1]$ be any normalized distance over $\X$ 
    and we denote $\operatorname{avg}(x,D)$ the average distance of $x$ to all elements in $D$. 
    Consider the suppression strategy $\S_{K}$ such that $\S_{K}(D)=\{x\in D\mid \operatorname{avg}(x,D)\leq K\}$ for any $K\in(0,1)\cap\Q$. Then $\S_K$ has sensitivity $\Delta_{\D}\S_K=\infty$. 
\end{proposition}
\begin{proof}
    This proof consists of showing that there are neighboring databases, $D$ and $D'$, such that many elements are deleted in $\S(D)$ but not in $\S(D')$. If this difference scales with the size of the databases, then the sensitivity must be infinite because it captures the worst-case difference. We will now describe the database pairs $(D_n,D'_n)$ for which the difference becomes unbounded as $n\to\infty$. 

    For all $K\in(0,1)\cap\Q$, there exists $N\in\N, N\geq2$ such that $NK\in\N$. Since the distance is normalized, for every $n\in\N$, there exists $x'_n,y'_n\in\X$ such that $\dist(x'_n,y'_n)>1-\frac{1}{nN}$. 
    We consider an infinite sequence of databases $\{D_n\}_{n\in\N}$ such that $D_n$ contains $nNK$ copies of $x'_n$ and $nN(1-K)$ copies of $y'_n$ for all $n\in\N$ (in particular, $|D_n|=nN$). We consider their respective neighboring databases $\{D'_n\}_{n\in\N}$ that contain $nNK-1$ copies of $x'_n$ and $nN(1-K)$ copies of $y'_n$.

    In particular, we can see that 
    \[
        \operatorname{avg}(y'_n,D_n) = \frac{nNK\dist(x'_n,y'_n)+(nN(1-K))\dist(y'_n,y'_n)}{nN} = K\dist(x'_n,y'_n) \leq K
    \]
    and 
    \begin{gather*}
        \operatorname{avg}(y'_n,D'_n) = \frac{nNK}{nN-1}\dist(x'_n,y'_n) = \frac{(nN-1)K+K}{nN-1}\dist(x'_n,y'_n) > 
        K-\frac{K}{nN}+\frac{K}{nN-1}\bigg(1-\frac{1}{nN}\bigg) = K.
    \end{gather*}

    Thus, for all $n\in\N$, $S_K$ suppresses all copies of $y'_n$ from $D'_n$ but none from $D_n$. Consequently, denoting $Y_n=\{y'_n,\dots,y'_n\}$ as the database (multiset) with $nN(1-K)$ copies of $y'_n$, we have that $Y_n\subseteq\S_K(D_n)$ and $Y_n\cap\S_K(D'_n)=\varnothing$ for all $n\in\N$. Thus, $Y_n\subseteq\S_K(D_n)\Delta\S_K(D'_n)$ and 
    \[
        \Delta_{\D}\S_K = \sup_{D\sim D'}
        |\S_K(D)\Delta\S_K(D')| \geq \sup_{n\in\N} |\S_K(D_n)\Delta\S_K(D'_n)| \geq \sup_{n\in\N} |Y_n| = \sup_{n\in\N} nN(1-K) = \infty. \qedhere 
    \]
\end{proof}

\begin{proposition}\label{prop:determinsticsuppressionwithinfinitesensitivity2}
    Let $\D=\D_{\X}$ be the class of all databases with elements drawn from $\X$ with $|\X|\geq2$. Let $\dist\colon \X\times\X\to[0,1]$ be any distance over $\X$ 
    and we denote $\operatorname{avg}(x,D)$ the average distance of $x$ to all elements in $D$. 
    For any $p\in(0,\frac{1}{2}]$, consider the suppression strategy $\S_{p}$ that deletes the top $\floor{p|D|}$ records with the highest average distance to the records of its input $D$ (in the case of ties among the top $\floor{p|D|}$ elements, we delete all elements). Then $\S_p$ has sensitivity $\Delta_{\D}\S_p=\infty$.
\end{proposition}
\begin{proof}
    Like for \Cref{prop:determinsticsuppressionwithinfinitesensitivity}, the proof consists of showing that there are neighboring databases, $D$ and $D'$, such that many elements are deleted in $\S(D)$ but not in $\S(D')$. If this difference scales with the size of the databases, then the sensitivity must be infinite because it captures the worst-case difference. We will now describe the database pairs $(D_N,D'_N)$ for which the difference becomes unbounded as $N\to\infty$. 

    Let $x',y'\in\X$ such that $x'\neq y'$, i.e., $\dist(x',y')\neq0$.   
    
    For all $N\in\N$ such that $\floor{pN}\geq1$, we consider the database $D_N$ with $N$ elements, $\floor{pN}$ copies of $x'$ and $N-\floor{pN}$ copies of $y'$. We have that
    \[
        \operatorname{avg}(x',D_N)=\frac{N-\floor{pN}}{N}\dist(x',y') \qquad \text{and}\qquad \operatorname{avg}(y',D_N)=\frac{\floor{pN}}{N}\dist(x',y'). 
    \]
    
    If $p\leq\frac{1}{2}$, then $\floor{pN}\leq \floor{\frac{1}{2}N}\leq\frac{1}{2}N$ and thus $\operatorname{avg}(y',D_N)\leq\operatorname{avg}(x',D_N)$. In $D_N$, the $\floor{pN}$ records with the highest average distance are the $\floor{pN}$ copies of $x'$. Thus, $\S_p(D_N)=\{y',\dots,y'\}$, i.e., the $N-\floor{pN}$ copies of $y'$. 
    
    Now, if we consider the neighboring database $D'_N$ to $D_N$ defined so that it has one less copy of $x'$ (this is defined since $\floor{pN}\geq1$). In this case, we also obtain that $\operatorname{avg}(y',D_N)\leq\operatorname{avg}(x',D_N)$, but since there are fewer than $\floor{pN}$ copies of $x'$, every copy of $y'$ ties as the $\floor{pN}$th element with the highest average distance. Consequently, by definition of $\S_p$, we obtain that $\S_p(\D'_N)=\varnothing$.  
    
    Thus, 
    \[
        \Delta_{\D}\S_p = \sup_{D\sim D'}
        |\S_p(D)\Delta\S_p(D')| \geq \sup_{\substack{N\in\N:\\\floor{pN}\geq1}} |S_p(D_N)\Delta\S_p(D'_N)| \geq \sup_{\substack{N\in\N:\\\floor{pN}\geq1}} (N-\floor{pN}) = \infty. \qedhere 
    \]
\end{proof}

\begin{remark}\label{prop:deterministicsuppressionDataDependentSensitivity1}
    The suppression algorithm that, given $D$, outputs $D$ if $|D|=1$ and $\varnothing$ otherwise is data-dependent with sensitivity~$1$ by construction. This suppression algorithm has no applications and is only presented for illustrative purposes. 
\end{remark}

\subsection{Proofs of \Cref{sec:mainsuppression}: \nameref*{sec:mainsuppression}}

\THgeneralunboundedtheorem*
\begin{proof}
    First note that the expressions are well-defined since the support condition ensures that there is no division by $0$.
    
    We fix $D,D'\in\D$ unbounded-neighboring databases with $D'=D\uplus\{y\}$. We need to see that for all measurable $\meas\subseteq\Range(\M\circ\S)\subseteq\Range(\M)=\Ran$,
    \begin{gather}
        \Prob\{\M(\S(D))\in\meas\}\leq \e^{\varepsilon^\S_{D,D'}}\Prob\{\M(\S(D'))\in\meas\}+\delta^\S_{D,D'}, \label{eq:suppressionUnbounded1} \\
        \Prob\{\M(\S(D'))\in\meas\}\leq \e^{\varepsilon^\S_{D',D}}\Prob\{\M(\S(D))\in\meas\}+\delta^\S_{D',D}, \label{eq:suppressionUnbounded2}
    \end{gather}
    with $\varepsilon^\S_{D,D'}$, $\delta^\S_{D,D'}$, $\varepsilon^\S_{D',D}$ and $\delta^\S_{D',D}$ as defined in the statement.
    
   We prove first inequality~\ref{eq:suppressionUnbounded1}. To simplify notation, we denote $p_C \coloneqq \Prob\{\S(D')=C\}+\e^{-\varepsilon}\Prob\{\S(D')=C_{+y}\}$ for all $C\in\supp(\S(D))$, which is non-zero by the support condition. Since $\S(D)$ are discrete random variables, by the law of total probability and further manipulations, we have that 
    \begin{align*}
       \Prob\{\M(\S(D))\in\meas\} 
       ={}& \sum_{C\in\supp(\S(D))} \Prob\{\M(C)\in\meas\}\Prob\{\S(D)=C\} \\
       ={}& \sum_{C\in\supp(\S(D))} \Prob\{\M(C)\in\meas\}\frac{p_C}{p_C}\Prob\{\S(D)=C\} \\
       ={}& \sum_{C\in\supp(\S(D))} \Prob\{\M(C)\in\meas\}\Prob\{\S(D')=C\}\frac{\Prob\{\S(D)=C\}}{p_C} \\
       &+ \sum_{C\in\supp(\S(D))} \e^{-\varepsilon}\Prob\{\M(C)\in\meas\}\Prob\{\S(D')=C_{+y}\}\frac{\Prob\{\S(D)=C\}}{p_C}.
    \end{align*}

    Applying that $\M$ is $(\varepsilon,\delta)$-DP in the second sum, we obtain
    \begin{align*}
        \Prob\{\M(\S(D))\in\meas\} 
        ={}& \sum_{C\in\supp(\S(D))} \Prob\{\M(C)\in\meas\}\Prob\{\S(D')=C\}\frac{\Prob\{\S(D)=C\}}{p_C} \\
        &+ \sum_{C\in\supp(\S(D))} \e^{-\varepsilon}\Prob\{\M(C)\in\meas\}\Prob\{\S(D')=C_{+y}\}\frac{\Prob\{\S(D)=C\}}{p_C} \\
        \leq{}& \sum_{C\in\supp(\S(D))} \Prob\{\M(C)\in\meas\}\Prob\{\S(D')=C\}\frac{\Prob\{\S(D)=C\}}{p_C} \\
        &+ \sum_{C\in\supp(\S(D))} \e^{-\varepsilon}(\e^{\varepsilon}\Prob\{\M(C_{+y})\in\meas\}+\delta)\Prob\{\S(D')=C_{+y}\}\frac{\Prob\{\S(D)=C\}}{p_C} \\
        ={}& \sum_{C\in\supp(\S(D))} \Prob\{\M(C)\in\meas\}\Prob\{\S(D')=C\}\frac{\Prob\{\S(D)=C\}}{p_C} \\
        &+ \sum_{C\in\supp(\S(D))} \Prob\{\M(C_{+y})\in\meas\}\Prob\{\S(D')=C_{+y}\}\frac{\Prob\{\S(D)=C\}}{p_C} \\
        &+ \sum_{C\in\supp(\S(D))} \delta\e^{-\varepsilon}\Prob\{\S(D')=C_{+y}\}\frac{\Prob\{\S(D)=C\}}{p_C} \\
        \leq{}& \bigg(\max_{C\in\supp(\S(D))} \frac{\Prob\{\S(D)=C\}}{p_C}\bigg) \sum_{C\in\supp(\S(D))} \Prob\{\M(C)\in\meas\}\Prob\{\S(D')=C\} \\
        &+ \bigg(\max_{C\in\supp(\S(D))} \frac{\Prob\{\S(D)=C\}}{p_C}\bigg) \sum_{C\in\supp(\S(D))} \Prob\{\M(C_{+y})\in\meas\}\Prob\{\S(D')=C_{+y}\} \\
        &+ \sum_{C\in\supp(\S(D))} \delta\e^{-\varepsilon}\Prob\{\S(D')=C_{+y}\}\frac{\Prob\{\S(D)=C\}}{p_C}.
    \end{align*}

    Note that the last sum equals $\delta^\S_{D,D'}$ and $\e^{\varepsilon^\S_{D,D'}}=\max_{C\in\supp(\S(D))}\frac{\Prob\{\S(D)=C\}}{p_C}$. Therefore, we have that
    \begin{align*}
        \Prob\{\M(\S(D))\in\meas\} 
        \leq{}& \e^{\varepsilon^\S_{D,D'}}\sum_{C\in\supp(\S(D))} \Prob\{\M(C)\in\meas\}\Prob\{\S(D')=C\} \\
        &+ \e^{\varepsilon^\S_{D,D'}} \sum_{C\in\supp(\S(D))} \Prob\{\M(C_{+y})\in\meas\}\Prob\{\S(D')=C_{+y}\} + \delta^\S_{D,D'} \\
        ={}& \e^{\varepsilon^\S_{D,D'}}\sum_{\substack{C'\in\supp(\S(D')):\\y\notin C'}} \Prob\{\M(C')\in\meas\}\Prob\{\S(D')=C'\} \\
        &+ \e^{\varepsilon^\S_{D,D'}}\sum_{\substack{C'\in\supp(\S(D')):\\y\in C'}} \Prob\{\M(C')\in\meas\}\Prob\{\S(D')=C'\} + \delta^\S_{D,D'} \\
        ={}& \e^{\varepsilon^\S_{D,D'}}\sum_{C'\in\supp(\S(D'))} \Prob\{\M(C')\in\meas\}\Prob\{\S(D')=C'\} + \delta^\S_{D,D'} \\
        ={}& \e^{\varepsilon^\S_{D,D'}}\Prob\{\M(\S(D'))\in\meas\} + \delta^\S_{D,D'}, 
    \end{align*}
    proving inequality~\ref{eq:suppressionUnbounded1}.
    Now we see inequality~\ref{eq:suppressionUnbounded2}. Once again, by the law of total probability and further manipulations, we have that 
    \begin{align*}
       \Prob\{\M(\S(D'))\in\meas\} 
       ={}& \sum_{C'\in\supp(\S(D'))} \Prob\{\M(C')\in\meas\}\Prob\{\S(D')=C'\} \\
       ={}& \sum_{\substack{C'\in\supp(\S(D')):\\y\notin C'}} \Prob\{\M(C')\in\meas\}\Prob\{\S(D')=C'\}
       + \sum_{\substack{C'\in\supp(\S(D')):\\y\in C'}} \Prob\{\M(C')\in\meas\}\Prob\{\S(D')=C'\} \\
       ={}& \sum_{C\in\supp(\S(D))} \Prob\{\M(C)\in\meas\}\Prob\{\S(D')=C\}
       + \sum_{C\in\supp(\S(D))} \Prob\{\M(C_{+y})\in\meas\}\Prob\{\S(D')=C_{+y}\}.
     \end{align*}

    Now, since $\M$ is $(\varepsilon,\delta)$-DP and $C$ and $C_{+y}$ are unbounded-neighboring, we have that
    \begin{align*}
       \Prob\{\M(\S(D'))\in\meas\} 
       ={}& \sum_{C\in\supp(\S(D))} \Prob\{\M(C)\in\meas\}\Prob\{\S(D')=C\}
       + \sum_{C\in\supp(\S(D))} \Prob\{\M(C_{+y})\in\meas\}\Prob\{\S(D')=C_{+y}\} \\
       \leq{}& \sum_{C\in\supp(\S(D))} \Prob\{\M(C)\in\meas\}\Prob\{\S(D')=C\}
       + \sum_{C\in\supp(\S(D))} (\e^\varepsilon\Prob\{\M(C)\in\meas\}+\delta)\Prob\{\S(D')=C_{+y}\} \\
       ={}& \underbrace{\sum_{C\in\supp(\S(D))} \Prob\{\M(C)\in\meas\}(\Prob\{\S(D')=C\}+\e^\varepsilon\Prob\{\S(D')=C_{+y}\})}_{\eqqcolon a}
       + \underbrace{\delta\sum_{C\in\supp(\S(D))} \Prob\{\S(D')=C_{+y}\}.}_{\eqqcolon b}
    \end{align*}

    Let $q_C\coloneqq \Prob\{\S(D')=C\}+\e^\varepsilon\Prob\{\S(D')=C_{+y}\}$ to simplify the notation. Further manipulating the values $a$ and $b$, we obtain
    \begin{align*}
        a &= \sum_{C\in\supp(\S(D))} \Prob\{\M(C)\in\meas\}q_C \\
        &= \sum_{C\in\supp(\S(D))} \Prob\{\M(C)\in\meas\}\Prob\{\S(D)=C\}\frac{q_C}{\Prob\{\S(D)=C\}}\\
        &\leq \bigg(\max_{C\in\supp(\S(D))} \frac{q_C}{\Prob\{\S(D)=C\}}\bigg)\sum_{C\in\supp(\S(D))} \Prob\{\M(C)\in\meas\}\Prob\{\S(D)=C\} \\
        &= \e^{\varepsilon^\S_{D',D}}\Prob\{\M(\S(D))\in\meas\}, 
    \end{align*}
    and
    \begin{align*}
        b = \delta \sum_{C\in\supp(\S(D))} \Prob\{\S(D')=C_{+y}\} = \delta \Prob\{y\in\S(D')\} = \delta^\S_{D',D}.
    \end{align*}

    Thus, we obtain inequality~\ref{eq:suppressionUnbounded2}:
    \[
        \Prob\{\M(\S(D'))\in\meas\} \leq a+b \leq \e^{\varepsilon^\S_{D',D}}\Prob\{\M(\S(D))\in\meas\} + \delta^\S_{D',D}. \qedhere
    \]
\end{proof}

\begin{definition}[Assignation function]
    Let $\S$ be a suppression algorithm with domain $\D$. Let $D$ and $D'$ be two unbounded-neighboring databases in $\D$. Assume $D'=D_{+y}$. Then, an \textit{assignation $\mathcal{A}_{D,D'}$ from $D$ to $D'$} is any function 
    \[
        \mathcal{A}_{D,D'}\colon \{C\subseteq D\mid C\in\supp(\S(D))\} \longrightarrow \{C\subseteq D\mid C\in\supp(\S(D')) \text{ or } C_{+y}\in\supp(\S(D'))\},
    \]
    and an \textit{assignation $\mathcal{A}_{D',D}$ from $D'$ to $D$} is any function 
    \[
        \mathcal{A}_{D',D}\colon \{C\subseteq D\mid C\in\supp(\S(D')) \text{ or } C_{+y}\in\supp(\S(D'))\} \longrightarrow \{C\subseteq D\mid C\in\supp(\S(D))\}.
    \]

    We denote the fiber of $C$ under $\mathcal{A}_{D,D'}$ by $\mathcal{A}_{D,D'}^{-1}[C]$, which is defined as the set of all databases $C^*$ in the domain of $\mathcal{A}_{D,D'}$ such that $\mathcal{A}_{D,D'}(C^*)=C$ (analogously for $\mathcal{A}_{D',D}$). Naturally, if $C$ is not in the image of $\mathcal{A}_{D,D'}$, we have that $\mathcal{A}_{D,D'}^{-1}[C]=\varnothing$.  
\end{definition}

\begin{theorem}[Suppression theorem (general result)]\label{th:generalunboundedtheorem}
    Let $\M$ with domain $\D$ be a mechanism that satisfies unbounded $(\varepsilon,\delta)$-DP and $\S$ with domain $\D$ a suppression algorithm.
    
    Then, $\M\circ\S$ is unbounded $(\varepsilon^{\S},\delta^{\S})$-DP with 
    \[
        \varepsilon^{\S} = \sup_{\substack{D,D'\in\D\\\text{neighb.}}} \varepsilon^{\S}_{D,D'} \quad \text{and} \quad \delta^{\S} = \sup_{\substack{D,D'\in\D\\\text{neighb.}}} \delta^{\S}_{D,D'},
    \]
    where $\varepsilon^{\S}_{D,D'}$ and $\delta^{\S}_{D,D'}$ depend on an assignation $\mathcal{A}_{D,D'}$ and are defined as follows:
    If $D'=D_{+y}$, then
    \begin{align*}
        \e^{\varepsilon^{\S}_{D,D'}}=\max_{\substack{C\subseteq D:\\C\in\supp(\S(D'))\\\text{or }C_{+y}\in\supp(\S(D'))}}\frac{1}{\Prob\{\S(D')=C\}+\e^{-\varepsilon}\Prob\{\S(D')=C_{+y}\}}\sum_{\substack{C^*\in\mathcal{A}^{-1}_{D,D'}[C]}}\e^{\varepsilon|C\Delta C^*|}\Prob\{\S(D)=C^*\}
    \end{align*}
    and
    \begin{multline*}
        \delta^{\S}_{D,D'}=\delta\sum_{\substack{C\subseteq D:\\C\in\supp(\S(D'))\\\text{or }C_{+y}\in\supp(\S(D'))}}\bigg(\frac{\e^{-\varepsilon}\Prob\{\S(D')=C_{+y}\}}{\Prob\{\S(D')=C\}+\e^{-\varepsilon}\Prob\{\S(D')=C_{+y}\}}\sum_{\substack{C^*\in\mathcal{A}^{-1}_{D,D'}[C]}}\e^{\varepsilon|C\Delta C^*|}\Prob\{\S(D)=C^*\}\\+\sum_{\substack{C^*\in\mathcal{A}^{-1}_{D,D'}[C]}}\Prob\{\S(D)=C^*\}\sum^{|C\Delta C^*|-1}_{k=0}\e^{\varepsilon k}\bigg),
    \end{multline*}
    and, since the values are not symmetric with respect to $D$ and $D'$, 
    \begin{align*}
        \e^{\varepsilon^{\S}_{D',D}}=\max_{\substack{C\in\supp(\S(D))}}\frac{1}{\Prob\{\S(D)=C\}}\sum_{\substack{C^*\in\mathcal{A}^{-1}_{D',D}[C]}}\e^{\varepsilon|C\Delta C^*|}\big(\Prob\{\S(D')=C^*\}+\e^{\varepsilon}\Prob\{\S(D')=C^*_{+y}\}\big)
    \end{align*}
    and
    \[
        \delta^{\S}_{D',D}=\delta\sum_{\substack{C\in\supp(\S(D))}}\sum_{\substack{C^*\in\mathcal{A}^{-1}_{D',D}[C]}}\bigg(\Prob\{\S(D')=C^*\}\sum^{|C\Delta C^*|-1}_{k=0}\e^{\varepsilon k}+\Prob\{\S(D')=C^*_{+y}\}\sum^{|C\Delta C^*|}_{k=0}\e^{\varepsilon k}\bigg).
    \]
\end{theorem}
\begin{proof}
    First note that the expressions are well-defined since the arguments of the maximums and sums ensure that there is no division by $0$.
    
    We fix $D,D'\in\D$ unbounded-neighboring databases with $D'=D\uplus\{y\}$. We fix two assignation $\mathcal{A}_{D,D'}$ and $\mathcal{A}_{D',D}$. We need to see that for all measurable $\meas\subseteq\Range(\M\circ\S)\subseteq\Range(\M)=\Ran$,
    \begin{gather}
        \Prob\{\M(\S(D))\in\meas\}\leq \e^{\varepsilon^\S_{D,D'}}\Prob\{\M(\S(D'))\in\meas\}+\delta^\S_{D,D'}, \label{eq:suppressionUnbounded1General} \\
        \Prob\{\M(\S(D'))\in\meas\}\leq \e^{\varepsilon^\S_{D',D}}\Prob\{\M(\S(D))\in\meas\}+\delta^\S_{D',D}, \label{eq:suppressionUnbounded2General}
    \end{gather}
    with $\varepsilon^\S_{D,D'}$, $\delta^\S_{D,D'}$, $\varepsilon^\S_{D',D}$ and $\delta^\S_{D',D}$ as defined in the statement. We note that by the group privacy theorem~\cite{dwork2014algorithmic} of approximate DP, 
    \[
        \Prob\{\M(C^*)\in\meas\} \leq \e^{\varepsilon|C^*\Delta \mathcal{A}_{D,D'}(C^*)|} \Prob\{\M(\mathcal{A}_{D,D'}(C^*))\in\meas\} + \delta\Sigma_\varepsilon(|C^*\Delta \mathcal{A}_{D,D'}(C^*)|), 
    \]
    where $\Sigma_\varepsilon(l) = \sum^{l-1}_{k=0}\e^{\varepsilon k}$.
    
    We prove first inequality~\ref{eq:suppressionUnbounded1General}. Since $\S(D)$ are discrete random variables, by the law of total probability and the group privacy property, we have
    \begin{align*}
       \Prob\{\M(\S(D))\in\meas\} 
       ={}& \sum_{C^*\in\supp(\S(D))} \Prob\{\M(C^*)\in\meas\}\Prob\{\S(D)=C^*\} \\
       \leq{}& \sum_{C^*\in\supp(\S(D))} (\e^{\varepsilon|C^*\Delta \mathcal{A}_{D,D'}(C^*)|}\Prob\{\M(\mathcal{A}_{D,D'}(C^*))\in\meas\}+\delta\Sigma_\varepsilon(|C^*\Delta \mathcal{A}_{D,D'}(C^*)|))\Prob\{\S(D)=C^*\},
    \end{align*}
    and since 
    \[
        \{(C^*,\mathcal{A}_{D,D'}(C^*))\mid C^*\in\supp(\S(D))\} = \{(C^*,C)\mid C\in\Imag(\mathcal{A}_{D,D'}) \text{ and } C^*\in\mathcal{A}_{D,D'}^{-1}[C]\}, 
    \]
    we obtain
    \begin{align*}
       \MoveEqLeft[7] \sum_{C^*\in\supp(\S(D))} (\e^{\varepsilon|C^*\Delta \mathcal{A}_{D,D'}(C^*)|}\Prob\{\M(\mathcal{A}_{D,D'}(C^*))\in\meas\}+\delta\Sigma_\varepsilon(|C^*\Delta \mathcal{A}_{D,D'}(C^*)|))\Prob\{\S(D)=C^*\} \\
       ={}& \sum_{\substack{C\in\Imag(\mathcal{A}_{D,D'})}} \sum_{\substack{C^*\in\mathcal{A}_{D,D'}^{-1}[C]}} (\e^{\varepsilon|C^*\Delta C|} \Prob\{\M(C)\in\meas\}+\delta\Sigma_\varepsilon(|C^*\Delta C|))\Prob\{\S(D)=C^*\} \\
       ={}& \sum_{\substack{C\in\operatorname{Im}(\mathcal{A}_{D,D'})}} \bigg(\Prob\{\M(C)\in\meas\}\underbrace{\sum_{\substack{C^*\in\mathcal{A}_{D,D'}^{-1}[C]}} \e^{\varepsilon|C\Delta C^*|} \Prob\{\S(D)=C^*\}}_{\eqqcolon E_C} +\delta\underbrace{\sum_{\substack{C^*\in\mathcal{A}_{D,D'}^{-1}[C]}}\Sigma_\varepsilon(|C\Delta C^*|)\Prob\{\S(D)=C^*\}}_{\coloneqq D_C}\bigg).
    \end{align*}

    To simplify notation, we denote $p_C \coloneqq \Prob\{\S(D')=C\}+\e^{-\varepsilon}\Prob\{\S(D')=C_{+y}\}$, which is non-zero for $C\in\Imag(\mathcal{A}_{D,D'})$. Thus,
    \begin{align*}
       \Prob\{\M(\S(D))\in\meas\} 
       \leq{}& \sum_{\substack{C\in\operatorname{Im}(\mathcal{A}_{D,D'})}} (\Prob\{\M(C)\in\meas\}\,E_C + \delta D_C) \\
       ={}& \sum_{\substack{C\in\operatorname{Im}(\mathcal{A}_{D,D'})}} \bigg(\Prob\{\M(C)\in\meas\}\frac{p_C}{p_C}E_C + \delta D_C\bigg) \\
       ={}& \sum_{\substack{C\in\operatorname{Im}(\mathcal{A}_{D,D'})}} \Prob\{\M(C)\in\meas\}\Prob\{\S(D')=C\}\frac{E_C}{p_C} \\
       &+ \sum_{\substack{C\in\operatorname{Im}(\mathcal{A}_{D,D'})}} \e^{-\varepsilon}\Prob\{\M(C)\in\meas\}\Prob\{\S(D')=C_{+y}\}\frac{E_C}{p_C} + \delta\sum_{\substack{C\in\operatorname{Im}(\mathcal{A}_{D,D'})}} D_C.
    \end{align*}

    Applying that $\M$ is $(\varepsilon,\delta)$-DP in the second sum, we obtain
    \begin{align*}
        \Prob\{\M(\S(D))\in\meas\} 
        \leq{}& \sum_{C\in\operatorname{Im}(\mathcal{A}_{D,D'})} \Prob\{\M(C)\in\meas\}\Prob\{\S(D')=C\}\frac{E_C}{p_C} \\
        &+ \sum_{\substack{C\in\operatorname{Im}(\mathcal{A}_{D,D'})}} \e^{-\varepsilon}(\e^{\varepsilon}\Prob\{\M(C_{+y})\in\meas\}+\delta)\Prob\{\S(D')=C_{+y}\}\frac{E_C}{p_C} + \delta\sum_{\substack{C\in\operatorname{Im}(\mathcal{A}_{D,D'})}} D_C \\
        ={}& \sum_{\substack{C\in\operatorname{Im}(\mathcal{A}_{D,D'})}} \Prob\{\M(C)\in\meas\}\Prob\{\S(D')=C\}\frac{E_C}{p_C} \\
        &+ \sum_{\substack{C\in\operatorname{Im}(\mathcal{A}_{D,D'})}} \Prob\{\M(C_{+y})\in\meas\}\Prob\{\S(D')=C_{+y}\}\frac{E_C}{p_C} \\
        &+ \sum_{\substack{C\in\operatorname{Im}(\mathcal{A}_{D,D'})}} \delta\bigg(\e^{-\varepsilon}\Prob\{\S(D')=C_{+y}\}\frac{E_C}{p_C} + D_C\bigg).
    \end{align*}

    Note that the last sum equals $\delta^\S_{D,D'}$ and $\e^{\varepsilon^\S_{D,D'}}=\max_{C\in\Imag(\mathcal{A}_{D,D'})}\frac{E_C}{p_C}$. Therefore, we have that
    \begin{align*}
        \Prob\{\M(\S(D))\in\meas\} 
        \leq{}& \e^{\varepsilon^\S_{D,D'}}\sum_{C\in\operatorname{Im}(\mathcal{A}_{D,D'})} \Prob\{\M(C)\in\meas\}\Prob\{\S(D')=C\} \\
        &+ \e^{\varepsilon^\S_{D,D'}} \sum_{C\in\operatorname{Im}(\mathcal{A}_{D,D'})} \Prob\{\M(C_{+y})\in\meas\}\Prob\{\S(D')=C_{+y}\} + \delta^\S_{D,D'} \\
        \leq{}& \e^{\varepsilon^\S_{D,D'}}\sum_{\substack{C'\in\supp(\S(D')):\\y\notin C'}} \Prob\{\M(C')\in\meas\}\Prob\{\S(D')=C'\} \\
        &+ \e^{\varepsilon^\S_{D,D'}}\sum_{\substack{C'\in\supp(\S(D')):\\y\in C'}} \Prob\{\M(C')\in\meas\}\Prob\{\S(D')=C'\} + \delta^\S_{D,D'} \\
        ={}& \e^{\varepsilon^\S_{D,D'}}\sum_{C'\in\supp(\S(D'))} \Prob\{\M(C')\in\meas\}\Prob\{\S(D')=C'\} + \delta^\S_{D,D'} \\
        ={}& \e^{\varepsilon^\S_{D,D'}}\Prob\{\M(\S(D'))\in\meas\} + \delta^\S_{D,D'}, 
    \end{align*}
    proving inequality~\ref{eq:suppressionUnbounded1General}.
    Now we see inequality~\ref{eq:suppressionUnbounded2General}. Once again, by the law of total probability and the group privacy property
    \begin{align*}
       \Prob\{\M(\S(D'))\in\meas\} 
       ={}& \sum_{C^*\in\supp(\S(D'))} \Prob\{\M(C^*)\in\meas\}\Prob\{\S(D')=C^*\} \\
       \leq{}& \sum_{C^*\in\supp(\S(D'))} (\e^{\varepsilon|C^*\Delta\mathcal{A}_{D',D}(C^*)|}\Prob\{\M(\mathcal{A}_{D',D}(C^*))\in\meas\}+\delta\Sigma_\varepsilon(|C^*\Delta \mathcal{A}_{D',D}(C^*)|))\Prob\{\S(D')=C^*\},
    \end{align*}
    and since
    \[
        \{(C^*,\mathcal{A}_{D',D}(C^*))\mid C^*\in \supp(\S(D')) \text{ and } y\notin C^*\} = \{(C^*,C')\mid C'\in\Imag(\mathcal{A}_{D',D}) \text{ and } C^*\in\mathcal{A}^{-1}_{D',D}[C']\}
    \]
    and 
    \[
        \{(C^*,\mathcal{A}_{D',D}(C^*))\mid C^*\in \supp(\S(D')) \text{ and } y\in C^*\} = \{(C^*_{+y},C')\mid C'\in\Imag(\mathcal{A}_{D',D}) \text{ and } C^*\in\mathcal{A}^{-1}_{D',D}[C']\},
    \]
    we obtain
    \begin{align*}
       \MoveEqLeft[7]\sum_{C^*\in\supp(\S(D'))} (\e^{\varepsilon|C^*\Delta\mathcal{A}_{D',D}(C^*)|}\Prob\{\M(\mathcal{A}_{D',D}(C^*))\in\meas\}+\delta\Sigma_\varepsilon(|C^*\Delta \mathcal{A}_{D',D}(C^*)|))\Prob\{\S(D')=C^*\} \\
       ={}& \sum_{C'\in\Imag(\mathcal{A}_{D',D})}\sum_{C^*\in\mathcal{A}_{D',D}^{-1}[C']} (\e^{\varepsilon|C^*\Delta C'|}\Prob\{\M(C')\in\meas\}+\delta\Sigma_\varepsilon(|C^*\Delta C'|))\Prob\{\S(D')=C^*\} \\
       &+ \sum_{C'\in\Imag(\mathcal{A}_{D',D})}\sum_{C^*\in\mathcal{A}_{D',D}^{-1}[C']} (\e^{\varepsilon|C^*_{+y}\Delta C'|}\Prob\{\M(C')\in\meas\}+\delta\Sigma_\varepsilon(|C^*_{+y}\Delta C'|))\Prob\{\S(D')=C^*_{+y}\} \\
       ={}& \sum_{C'\in\Imag(\mathcal{A}_{D',D})}\bigg(\Prob\{\M(C')\in\meas\} \sum_{C^*\in\mathcal{A}_{D',D}^{-1}[C']} (\e^{\varepsilon|C'\Delta C^*|}\Prob\{\S(D')=C^*\}+\e^{\varepsilon|C'\Delta C^*_{+y}|}\Prob\{\S(D')=C^*_{+y}\}) \\
       &\qquad\qquad + \delta\sum_{C^*\in\mathcal{A}_{D',D}^{-1}[C']}(\Sigma_\varepsilon(|C'\Delta C^*|)\Prob\{\S(D')=C^*\}+\Sigma_\varepsilon(|C'\Delta C^*_{+y}|)\Prob\{\S(D')=C^*_{+y}\})\bigg) \\
       ={}& \sum_{C'\in\Imag(\mathcal{A}_{D',D})}\bigg(\Prob\{\M(C')\in\meas\} \sum_{C^*\in\mathcal{A}_{D',D}^{-1}[C']} \e^{\varepsilon|C'\Delta C^*|}(\Prob\{\S(D')=C^*\}+\e^{\varepsilon}\Prob\{\S(D')=C^*_{+y}\}) \\
       &\qquad\qquad + \delta\sum_{C^*\in\mathcal{A}_{D',D}^{-1}[C']}(\Sigma_\varepsilon(|C'\Delta C^*|)\Prob\{\S(D')=C^*\}+\Sigma_\varepsilon(|C'\Delta C^*|+1)\Prob\{\S(D')=C^*_{+y}\})\bigg).
    \end{align*}

    Now, we denote $E'_C$ as the value of the first inner sum in the previous equation and $D'_C$ as the value of the second inner sum. Since $\Prob\{\S(D)=C\}\neq0$ for all $C\in\Imag(\mathcal{A}_{D',D})$, we obtain
    \begin{align*}
        \Prob\{\M(\S(D'))\in\meas\} &\leq \sum_{C'\in\Imag(\mathcal{A}_{D',D})}(\Prob\{\M(C')\in\meas\}\,E'_C + \delta D'_C) \\
        &= \sum_{C'\in\Imag(\mathcal{A}_{D',D})}\Prob\{\M(C')\in\meas\}\Prob\{\S(D)=C'\}\frac{E'_C}{\Prob\{\S(D)=C'\}} + \delta\sum_{C'\in\Imag(\mathcal{A}_{D',D})} D'_C\\
        &\leq \e^{\varepsilon^\S_{D',D}}\sum_{C'\in\Imag(\mathcal{A}_{D',D})} \Prob\{\M(C')\in\meas\}\Prob\{\S(D)=C'\} + \delta^\S_{D',D}\\
        &\leq \e^{\varepsilon^\S_{D',D}}\sum_{C\in\supp(\S(D))} \Prob\{\M(C)\in\meas\}\Prob\{\S(D)=C\} + \delta^\S_{D',D}\\
        &= \e^{\varepsilon^\S_{D',D}}\Prob\{\M(\S(D))\in\meas\} + \delta^\S_{D',D}, 
    \end{align*}
    obtaining, thus, inequality~\ref{eq:suppressionUnbounded2General}.
\end{proof}

\begin{remark}[On \Cref{th:generalunboundedtheorem}]
    We first note that since the result holds for all assignations $\mathcal{A}_{D,D'}$, and thus we can optimize its choice to obtain the smallest $\varepsilon^\S_{D,D'}$ and $\delta^\S_{D,D'}$.  
    
    Additionally, \Cref{th:generalunboundedtheorem} indeed generalizes \Cref{th:unboundedcase}: If the support condition is verified, then 
    \[
        \{C\subseteq D\mid C\in\supp(\S(D))\} = \{C\subseteq D\mid C\in\supp(\S(D')) \text{ or } C_{+y}\in\supp(\S(D'))\}
    \]
    for all neighboring databases $D$ and $D_{+y}$. Thus, we can select $\mathcal{A}_{D,D'}=\id$ for all pairs of neighboring databases $D,D'\in\D$, which provides the precise bounds of \Cref{th:unboundedcase}. We expect that this assignation minimizes $\varepsilon^\S_{D,D'}$ and $\delta^\S_{D,D'}$ for the majority of algorithms $\S$. Similarly, we believe that selecting $\mathcal{A}_{D,D'}(C)=C$ for all $C\in\{C\subseteq D\mid C\in\supp(\S(D))\}\cap\{C\subseteq D\mid C\in\supp(\S(D')) \text{ or } C_{+y}\in\supp(\S(D'))\}$ is a good choice for the assignation.    
\end{remark}

\begin{proposition}\label{prop:final}
    Let $\S$ with domain $\D$ be a suppression algorithm and let $\varepsilon^\S_{\mathrm{tight}}$ and $\delta^\S_{\mathrm{tight}}$ be the tight privacy parameters of $\M\circ\S$ over all $(\varepsilon,\delta)$-DP mechanisms $\M$ with domain $\D$ ($\delta\leq 1$), i.e., for all $(\varepsilon,\delta)$-DP mechanism $\M$, we have that $\M\circ\S$ is $(\varepsilon^\S_{\mathrm{tight}},\delta^\S_{\mathrm{tight}})$-DP and there does not exist $\varepsilon''\leq\varepsilon^\S_{\mathrm{tight}}$ and $\delta''\leq\delta^\S_{\mathrm{tight}}$ (not both equal) such that $\M\circ\S$ is $(\varepsilon'',\delta'')$-DP for all $(\varepsilon,\delta)$-DP mechanisms $\M$. Then, 
    \[
        \varepsilon^\S_{\mathrm{tight}} \geq \ln\!\bigg(1+(\e^{\varepsilon}-1)\sup_{D'\in\D}\sup_{y\in D'}\Prob\{y\in\S(D')\}\bigg) \qquad\text{and}\qquad \delta^\S_{\mathrm{tight}}\geq \delta\sup_{D'\in\D}\sup_{y\in D'}\Prob\{y\in\S(D')\},
    \]
    where the right terms of the inequalities are the privacy parameters of the uniform Poisson sampling with sampling rate $p=\sup_{D'\in\D}\sup_{y\in D'}\Prob\{y\in\S(D')\}$.
\end{proposition}
\begin{proof}
    It is sufficient to see that there is an $(\varepsilon,\delta)$-DP mechanism $\M$ such that $\M\circ\S$ is tightly $(\varepsilon''',\delta''')$-DP with
    \[
        \varepsilon''' =  \ln\!\bigg(1+(\e^{\varepsilon}-1)\sup_{D'\in\D}\sup_{y\in D'}\Prob\{y\notin\S(D')\}\bigg) \qquad\text{and}\qquad  \delta''' =  \delta\sup_{D'\in\D}\sup_{y\in D'}\Prob\{y\notin\S(D')\}.
    \]
    
    Consider the mechanism $\M\colon\D\to\{0,1\}$ such that 
    \begin{gather*}
        \Prob\{\M(C)=0\}=\begin{cases}
            \frac{1-\delta}{1+\e^{\varepsilon}} & \text{if $y\notin C$,} \\
            \frac{\e^{\varepsilon}+\delta}{1+\e^{\varepsilon}} & \text{if $y\in C$,} \\
        \end{cases} \qquad\text{and}\qquad 
        \Prob\{\M(C)=1\}=\begin{cases}
            \frac{\e^{\varepsilon}+\delta}{1+\e^{\varepsilon}} & \text{if $y\notin C$,} \\
            \frac{1-\delta}{1+\e^{\varepsilon}} & \text{if $y\in C$.} \\
        \end{cases}
    \end{gather*}

    This mechanism is well-defined and satisfies $(\varepsilon,\delta)$-DP tightly: Indeed, for two neighboring databases $C,C'\in\D$ and $b\in\{0,1\}$, 
    we have that $\Prob\{\M(C)=b\}$ and $\Prob\{\M(C')=b\}$ are either equal (if $y\in C,C'$ or $y\notin C,C'$) or verify these inequalities (if $y\in C$ and $y\notin C'$, or vice versa): 
    \begin{gather*}
        \Prob\{\M(C)=b\} = \frac{1-\delta}{1+\e^{\varepsilon}} \leq \e^{\varepsilon}\frac{\e^{\varepsilon}+\delta}{1+\e^{\varepsilon}} + \delta = \e^\varepsilon\Prob\{\M(C')=b\} + \delta\\
        \Prob\{\M(C)=b\} = \frac{\e^{\varepsilon}+\delta}{1+\e^{\varepsilon}} = \e^{\varepsilon}\frac{1-\delta}{1+\e^{\varepsilon}} + \delta = \e^\varepsilon\Prob\{\M(C')=b\} + \delta.
    \end{gather*}

    Let $D,D'$ be two neighboring databases such that $D'=D_{+y}$. Then, we have the following equalities
    \begin{align*}
        \Prob\{\M(\S(D))=0\} &= \Prob\{\M(\S(D))=0\mid y\notin\S(D)\}\Prob\{y\notin\S(D)\} +  \Prob\{\M(\S(D))=0\mid y\in\S(D)\}\Prob\{y\in\S(D)\} \\
        &= \frac{1-\delta}{1+\e^{\varepsilon}}\cdot 1 + 0, \\
        \Prob\{\M(\S(D))=1\} &= \Prob\{\M(\S(D))=1\mid y\notin\S(D)\}\Prob\{y\notin\S(D)\} +  \Prob\{\M(\S(D))=1\mid y\in\S(D)\}\Prob\{y\in\S(D)\} \\
        &= \frac{\e^{\varepsilon}+\delta}{1+\e^{\varepsilon}}\cdot 1 + 0;
    \end{align*}
    and
    \begin{align*}
        \Prob\{\M(\S(D'))=0\} &= \Prob\{\M(\S(D'))=0\mid y\notin\S(D')\}\Prob\{y\notin\S(D')\} +  \Prob\{\M(\S(D'))=0\mid y\in\S(D')\}\Prob\{y\in\S(D')\} \\
        &= \frac{1-\delta}{1+\e^{\varepsilon}}(1-\Prob\{y\in\S(D')\}) + \frac{\e^{\varepsilon}+\delta}{1+\e^{\varepsilon}}\Prob\{y\in\S(D')\}, \\
        \Prob\{\M(\S(D'))=1\} &= \Prob\{\M(\S(D'))=1\mid y\notin\S(D')\}\Prob\{y\notin\S(D')\} +  \Prob\{\M(\S(D'))=1\mid y\in\S(D')\}\Prob\{y\in\S(D')\} \\
        &= \frac{\e^{\varepsilon}+\delta}{1+\e^{\varepsilon}}(1-\Prob\{y\in\S(D')\}) + \frac{1-\delta}{1+\e^{\varepsilon}}\Prob\{y\in\S(D')\}. \\
    \end{align*}

    We denote $\alpha\coloneqq\frac{\e^{\varepsilon}+\delta}{\e^{\varepsilon}+1}$ (then $1-\alpha=\frac{1-\delta}{\e^{\varepsilon}+1}$) and $q\coloneqq\Prob\{y\in\S(D')\}$ to simplify notation. Thus, we have that
    \[
        A\coloneqq\frac{\Prob\{\M(\S(D))=0\}-\delta q}{\Prob\{\M(\S(D'))=0\}} = \frac{(1-\alpha)-\delta q}{(1-\alpha)(1-q)+\alpha q} = \frac{(1-\alpha)-\delta q}{(1-\alpha)+(2\alpha-1)q},
    \]
    \[
        B\coloneqq\frac{\Prob\{\M(\S(D'))=0\}-\delta q}{\Prob\{\M(\S(D))=0\}} = \frac{(1-\alpha)(1-q)+\alpha q-\delta q}{1-\alpha} = (1-q) + \frac{\alpha-\delta}{1-\alpha}q = (1-q) + \e^{\varepsilon}q,
    \]
    \[
        C\coloneqq\frac{\Prob\{\M(\S(D))=1\}-\delta q}{\Prob\{\M(\S(D'))=1\}} = \frac{\alpha-\delta q}{\alpha(1-q)+(1-\alpha)q} = \frac{\alpha-\delta q}{\alpha - (2\alpha-1)q},
    \]
    and
    \[
        D\coloneqq\frac{\Prob\{\M(\S(D'))=1\}-\delta q}{\Prob\{\M(\S(D))=1\}} = \frac{\alpha(1-q)+(1-\alpha)q-\delta q}{\alpha} = (1-q) + \frac{1-\alpha-\delta}{\alpha}q = (1-q) + \frac{1-\delta(\e^{\varepsilon}+2)}{\e^{\varepsilon}+\delta}q.
    \]

    Now, we easily have that $A\leq 1$ since $2\alpha-1=\frac{\e^{\varepsilon}-(1-2\delta)}{\e^{\varepsilon}+1}\geq 0$, and that $D\leq 1$ since $\frac{1-\delta(\e^{\varepsilon}+2)}{\e^{\varepsilon}+\delta}\leq 1$. In addition, we have that $C\leq B$:
    \begin{align*}
        C\leq B &\Longleftrightarrow \frac{\alpha-\delta q}{\alpha - (2\alpha-1)q} \leq (1-q) + \e^{\varepsilon}q \\
        &\Longleftrightarrow \alpha-\delta q \leq (1-q)\alpha - (1-q)(2\alpha-1)q + \e^{\varepsilon}q\alpha - \e^{\varepsilon}(2\alpha-1)q^2 \\
        &\Longleftrightarrow \alpha-\delta q \leq \alpha + q(-\alpha-2\alpha+1+\e^{\varepsilon}\alpha) + q^2(2\alpha-1)(1-\e^{\varepsilon}) \\
        &\Longleftrightarrow 0 \leq q(1-3\alpha+\e^{\varepsilon}\alpha+\delta) - q^2(2\alpha-1)(\e^{\varepsilon}-1),
    \end{align*}

    Now, the inequality holds for $q=0$ directly. For $q\neq0$ and $\varepsilon=0$, we have that the inequality also holds: In this case $\alpha=\frac{1+\delta}{2}$ and
    \[
        C\leq B \Longleftrightarrow 0\leq q(1-3\alpha+\alpha+\delta) + 0 = q\bigg(1-2\frac{1+\delta}{2}+\delta\bigg) = 0.    
    \]

    For $q\neq 0$ and $\varepsilon\neq0$, we have that 
    \[
        C\leq B \Longleftrightarrow 0 \leq q(1-3\alpha+\e^{\varepsilon}\alpha+\delta) - q^2(2\alpha-1)(\e^{\varepsilon}-1)
        \Longleftrightarrow q \leq \frac{1-3\alpha+\e^{\varepsilon}\alpha+\delta}{(2\alpha-1)(\e^{\varepsilon}-1)},
    \]
    where the last inequality holds since 
    \begin{gather*}
        \frac{1-3\alpha+\e^{\varepsilon}\alpha+\delta}{(2\alpha-1)(\e^{\varepsilon}-1)} 
        = \frac{1}{2\alpha-1}\bigg(\alpha+\frac{1-2\alpha+\delta}{\e^{\varepsilon}-1}\bigg) = \frac{1}{2\alpha-1}\bigg(\frac{\e^{\varepsilon}+\delta}{\e^{\varepsilon}+1}+\frac{-\frac{\e^{\varepsilon}-(1-2\delta)}{\e^{\varepsilon}+1}+\delta}{\e^{\varepsilon}-1}\bigg) \\
        = \frac{1}{2\alpha-1}\bigg(\frac{\e^{\varepsilon}+\delta}{\e^{\varepsilon}+1}+\frac{-\e^{\varepsilon}+(1-2\delta)+\delta(\e^{\varepsilon}+1)}{(\e^{\varepsilon}-1)(\e^{\varepsilon}+1)}\bigg)
        = \frac{1}{2\alpha-1}\bigg(\frac{\e^{\varepsilon}+\delta}{\e^{\varepsilon}+1}+\frac{-\e^{\varepsilon}+(1-\delta)+\delta\e^{\varepsilon}}{(\e^{\varepsilon}-1)(\e^{\varepsilon}+1)}\bigg) \\
        = \frac{1}{2\alpha-1}\bigg(\frac{\e^{\varepsilon}+\delta}{\e^{\varepsilon}+1}-\frac{(\e^{\varepsilon}-1)(1-\delta)}{(\e^{\varepsilon}-1)(\e^{\varepsilon}+1)}\bigg) = \frac{1}{2\alpha-1}
    \end{gather*}
    is larger than 1. Thus, we have that 
    \[
        \Prob\{\M(\S(D))\in\meas\} \leq \e^{\varepsilon_{D,D'}}\Prob\{\M(\S(D'))\in\meas\} + \delta_{D,D'} \qquad\text{and}\qquad \Prob\{\M(\S(D'))\in\meas\} \leq \e^{\varepsilon_{D',D}}\Prob\{\M(\S(D))\in\meas\} + \delta_{D',D}
    \]
    with $\varepsilon_{D,D'}=\varepsilon_{D',D}=\ln(\e^{\varepsilon}q+(1-q))$ and $\delta_{D,D'}=\delta_{D',D}=\delta q$ for all subsets $\meas\subseteq\{0,1\}$, such that an inequality is tight for at least one subset (in this case, the right inequality for $\meas=\{0\}$). Taking the supremum over all neighboring $D,D'\in\D$, we obtain that $\M\circ\S$ is tightly $(\varepsilon''',\delta''')$-DP with 
    \[
        \varepsilon'''=\ln\!\bigg(1+(\e^{\varepsilon}-1)\sup_{D'\in\D}\sup_{y\in D'}\Prob\{y\in\S(D')\}\bigg) \qquad\text{and}\qquad \delta'''= \delta\sup_{D'\in\D}\sup_{y\in D'}\Prob\{y\in\S(D')\},
    \]
    that by construction verify $\varepsilon'''\leq\varepsilon^\S_{\mathrm{tight}}$ and $\delta'''\leq\delta^\S_{\mathrm{tight}}$, following the statement notation. This concludes the proof. 
\end{proof}

\subsection{Proofs of \Cref{sec:outlierscoresuppressiontheorem}: \nameref*{sec:outlierscoresuppressiontheorem}}
\begin{proposition}[Inhereted properties of {$(m,M)$}-transformations]\label{prop:(mM)-transformeddistance}
    The $(m,M)$-transformation $\tdist$ of a normalized distance $\dist$ inherits the following properties from $\dist$:
    \begin{enumerate}
        \item $\tdist(x,y)\in[m,M]$ for all $x,y\in\X$, 
        \item $\tdist(x,y)=m$ if and only if $x=y$,
        \item (Symmetry) $\tdist(x,y)=\tdist(y,x)$, and
        \item (Transformed triangular inequality) $\tdist(x,y)\leq \tdist(x,z)+\tdist(z,y)-m$ for all $x,y,z\in\X$. 
    \end{enumerate}
\end{proposition}
\begin{proof}
    Recall that $\tdist(x,y)=m+(M-m)\dist(x,y)$ for all $x,y\in\X$. Since $\dist(x,y)\in[0,1]$ for all $x,y\in\X$ and equals 0 if and only if $x=y$, the first two properties follow. Symmetry also follows directly from the symmetry of $\dist$. The transformed triangular inequality comes from the triangular inequality of $\dist$: For all $x,y,z\in\X$,
    \begin{multline*}
        \tdist(x,z) = m + (M-m)\dist(x,z) \leq m + (M-m)(\dist(x,y)+\dist(y,z)) \\ 
        = (m + (M-m)\dist(x,y)) + (m + (M-m)\dist(y,z)) - m = \tdist(x,y) + \tdist(y,z) - m. \qedhere
    \end{multline*}
\end{proof}

\THoutlierscoresuppression*
\begin{proof}
    First note that the random variables $\S(D)$ are well-defined since $\out_D(x)\in[m,M]\subset[0,1]$ for all $x\in D$, which follow a distribution with support $\{C\mid C\subseteq D\}$, the subsets of $D$. In particular, $\S$ verifies the support condition needed to apply \Cref{th:unboundedcase}.
    
    We fix two unbounded-neighboring databases, $D$ and $D'$, with $D'=D_{+y}$. For all $C\subseteq D$, we have that 
    \[
        \Prob\{\S(D)=C\} = \prod_{x\in C} (1-\out_{D}(x)) \prod_{x\in D\backslash C}\out_{D}(x),
    \]
    \begin{align*}
        \Prob\{\S(D')=C\} &= \prod_{x\in C} (1-\out_{D'}(x)) \prod_{x\in D'\backslash C}\out_{D'}(x)
        = \out_{D'}(y) \prod_{x\in C} (1-\out_{D'}(x)) \prod_{x\in D\backslash C}\out_{D'}(x)
    \end{align*}
    and 
    \begin{align*}
        \Prob\{\S(D')=C_{+y}\} &= \prod_{x\in C_{+y}} (1-\out_{D'}(x)) \prod_{x\in D'\backslash C_{+y}}\out_{D'}(x)
        = (1-\out_{D'}(y)) \prod_{x\in C} (1-\out_{D'}(x)) \prod_{x\in D\backslash C}\out_{D'}(x).
    \end{align*}
    
    Then, from \Cref{th:unboundedcase}, we obtain that $\M\circ\S$ satisfies
    \[
        \Prob\{(\M\circ\S)(D)\in\meas\} \leq \e^{\varepsilon^\S_{D,D'}} \Prob\{(\M\circ\S)(D')\in\meas\} + \delta^\S_{D,D'}
    \] 
    for all measurable $\meas\subseteq\Ran\coloneqq\Range(\M)=\Range(\M\circ\S)$ with 
    \begin{align*}
        \e^{\varepsilon^{\S}_{D,D'}} &= \max_{\substack{C\subseteq D}}\frac{\Prob\{\S(D)=C\}}{\Prob\{\S(D')=C\}+\e^{-\varepsilon}\Prob\{\S(D')=C_{+y}\}} \\
        &= \max_{\substack{C\subseteq D}}\bigg(\frac{1}{\out_{D'}(y)+\e^{-\varepsilon}(1-\out_{D'}(y))} \prod_{x\in C} \frac{1-\out_{D}(x)}{1-\out_{D'}(x)} \prod_{x\in D\backslash C} \frac{\out_{D}(x)}{\out_{D'}(x)}\bigg),
    \end{align*}
    and
    \begin{align*}
        \delta^{\S}_{D,D'}&=\delta\sum_{\substack{C\in\supp(\S(D))}}\frac{\Prob\{\S(D)=C\}\,\e^{-\varepsilon}\Prob\{\S(D')=C_{+y}\}}{\Prob\{\S(D')=C\}+\e^{-\varepsilon}\Prob\{\S(D')=C_{+y}\}} \\
        &=\delta\sum_{\substack{C\in\supp(\S(D))}}\frac{\Prob\{\S(D)=C\}\,\e^{-\varepsilon}(1-\out_{D'}(y))}{\out_{D'}(y)+\e^{-\varepsilon}(1-\out_{D'}(y))} \\
        &=\delta\frac{\e^{-\varepsilon}(1-\out_{D'}(y))}{\out_{D'}(y)+\e^{-\varepsilon}(1-\out_{D'}(y))}\sum_{\substack{C\in\supp(\S(D))}}\Prob\{\S(D)=C\} \\
        &=\delta\frac{\e^{-\varepsilon}(1-\out_{D'}(y))}{\out_{D'}(y)+\e^{-\varepsilon}(1-\out_{D'}(y))};
    \end{align*}
    and satisfies
    \[
        \Prob\{(\M\circ\S)(D')\in\meas\} \leq \e^{\varepsilon^\S_{D',D}} \Prob\{(\M\circ\S)(D)\in\meas\} + \delta^\S_{D',D}
    \]
    for all measurable $\meas\subseteq\Ran$ with 
    \begin{align*}
        \e^{\varepsilon^{\S}_{D',D}}&=\max_{\substack{C\in\supp(\S(D))}}\frac{\Prob\{\S(D')=C\}+\e^{\varepsilon}\Prob\{\S(D')=C_{+y}\}}{\Prob\{\S(D)=C\}} \\
        &= \max_{\substack{C\subseteq D}}\bigg((\out_{D'}(y)+\e^{\varepsilon}(1-\out_{D'}(y))) \prod_{x\in C} \frac{1-\out_{D'}(x)}{1-\out_{D}(x)} \prod_{x\in D\backslash C} \frac{\out_{D'}(x)}{\out_{D}(x)}\bigg),
    \end{align*}
    and
    \begin{align*}
        \delta^{\S}_{D',D} &= \delta\sum_{C\in\supp(\S(D))} \Prob\{\S(D')=C_{+y}\}
        = \delta(1-\out_{D'}(y))\underbrace{\sum_{C\subseteq D}\bigg(\prod_{x\in C} (1-\out_{D'}(x)) \prod_{x\in D\backslash C}\out_{D'}(x)\bigg)}_{=1} = \delta(1-\out_{D'}(y)).
    \end{align*}

    Note that we obtain
    \[
        \delta^{\S}_{D,D'} = \delta\frac{\e^{-\varepsilon}(1-\out_{D'}(y))}{\out_{D'}(y)+\e^{-\varepsilon}(1-\out_{D'}(y))} \leq \delta(1-\out_{D'}(y)) = \delta^{\S}_{D',D}.
    \]

    Thus, 
    \[
        \delta^{\S} = \sup_{\substack{D,D'\in\D\\D\sim D'}} \delta^{\S}_{D,D'} = \sup_{\substack{D,D'\in\D\\D\sim D'}} \delta(1-\out_{D'}(y)) = \delta(1-m).
    \]

    Now we look for a bound for $\varepsilon^{\S}$. We first study separately the case where $D=\varnothing$ and $D'=\{y\}$. In this case, $\out_{D'}(y)=\dist(y,y)=m$ and we have that 
    \[
        \e^{\varepsilon^{\S}_{\varnothing,\{y\}}} = \frac{1}{m+\e^{-\varepsilon}(1-m)} \leq m+\e^{\varepsilon}(1-m)
        = \e^{\varepsilon^{\S}_{\{y\},\varnothing}}.
    \]

    Hereafter, we consider $1\leq|D|<|D'|$. We are now going to see that the previous expressions for $\e^{\varepsilon^{\S}_{D,D'}}$ and $\e^{\varepsilon^{\S}_{D',D}}$ can be written in terms of $N\coloneqq|D|$, $\tdist(x,y)$ and $\out_{D}(x)$ for $x\in D$. First, we note that
    \begin{gather*}
        \out_{D'}(y) = \frac{1}{|D'|}\bigg(\sum_{x\in D'} \tdist(x,y)\bigg) = \frac{1}{|D|+1}\bigg(\tdist(y,y) + \sum_{x\in D} \tdist(x,y)\bigg) = \frac{1}{N+1}\bigg(m + \sum_{x\in D} \tdist(x,y)\bigg),
    \end{gather*}
    and, for all $x\in D\subseteq D'$,
    \begin{align*}
        \out_{D'}(x) = \frac{1}{|D'|}\sum_{z\in D'}\tdist(x,z)
        = \frac{1}{N+1}\bigg(\sum_{\substack{z\in D}}\tdist(x,z) + \tdist(x,y)\bigg)
        = \frac{1}{N+1}(N\out_D(x) + \tdist(x,y)).
    \end{align*}

    In addition,  
    \begin{align*}
        \frac{\out_{D'}(x)}{\out_D(x)} 
        = \frac{\frac{1}{N+1}(N\out_D(x) + \tdist(x,y))}{\out_D(x)} 
        = \frac{N+\frac{\tdist(x,y)}{\out_D(x)}}{N+1}
    \end{align*}
    and
    \begin{gather*}
        \frac{1-\out_{D'}(x)}{1-\out_{D}(x)} 
        = \frac{1-\frac{1}{N+1}(N\out_D(x) + \tdist(x,y))}{1-\out_{D}(x)}
        = \frac{N+1-(N\out_D(x) + \tdist(x,y))}{(N+1)(1-\out_{D}(x))} \\
        = \frac{N(1-\out_D(x))+1- \tdist(x,y)}{(N+1)(1-\out_{D}(x))}
        = \frac{N+\frac{1-\tdist(x,y)}{1-\out_{D}(x)}}{N+1}.
    \end{gather*}

    Therefore, 
    \begin{align*}
        \e^{\varepsilon^{\S}_{D,D'}} &= \max_{\substack{C\subseteq D}}\bigg(\frac{1}{\out_{D'}(y)+\e^{-\varepsilon}(1-\out_{D'}(y))} \prod_{x\in C} \frac{1-\out_{D}(x)}{1-\out_{D'}(x)} \prod_{x\in D\backslash C} \frac{\out_{D}(x)}{\out_{D'}(x)}\bigg) \\
        &= \max_{\substack{C\subseteq D}}\bigg(\frac{1}{\e^{-\varepsilon}+(1-\e^{-\varepsilon})\out_{D'}(y)}\prod_{x\in C} \frac{N+1}{N+\frac{1-\tdist(x,y)}{1-\out_{D}(x)}} \prod_{x\in D\backslash C} \frac{N+1}{N+ \frac{\tdist(x,y)}{\out_D(x)}}\bigg) \\
        &= \max_{\substack{C\subseteq D}}\bigg(\bigg(\e^{-\varepsilon}+(1-\e^{-\varepsilon})\frac{m+\sum_{x\in D} \tdist(x,y)}{N+1}\bigg)^{-1}\prod_{x\in C} \frac{N+1}{N+\frac{1-\tdist(x,y)}{1-\out_{D}(x)}} \prod_{x\in D\backslash C} \frac{N+1}{N+ \frac{\tdist(x,y)}{\out_D(x)}}\bigg),
    \end{align*}
    and 
    \begin{align*}
        \e^{\varepsilon^{\S}_{D',D}} &= \max_{\substack{C\subseteq D}}\bigg((\out_{D'}(y)+\e^{\varepsilon}(1-\out_{D'}(y)))\prod_{x\in C} \frac{1-\out_{D'}(x)}{1-\out_{D}(x)} \prod_{x\in D\backslash C} \frac{\out_{D'}(x)}{\out_{D}(x)}\bigg) \\
        &= \max_{\substack{C\subseteq D}}\bigg((\e^{\varepsilon}-(\e^{\varepsilon}-1)\out_{D'}(y)) \prod_{x\in C}\frac{N+\frac{1-\tdist(x,y)}{1-\out_{D}(x)}}{N+1} \prod_{x\in D\backslash C} \frac{N+ \frac{\tdist(x,y)}{\out_D(x)}}{N+1}\bigg) \\
        &= \max_{\substack{C\subseteq D}}\bigg(\bigg(\e^{\varepsilon}-(\e^{\varepsilon}-1)\frac{m+\sum_{x\in D} \tdist(x,y)}{N+1}\bigg) \prod_{x\in C}\frac{N+\frac{1-\tdist(x,y)}{1-\out_{D}(x)}}{N+1} \prod_{x\in D\backslash C} \frac{N+ \frac{\tdist(x,y)}{\out_D(x)}}{N+1}\bigg).
    \end{align*}

    The terms inside the maximum in the previous values only depend on $N$, $a_{D,D'}\coloneqq\{\tdist(x,y)\}_{x\in D}$, and $z_{D}\coloneqq\{\out_D(x)\}_{x\in D}$. We use the following notation for the terms inside the previous maximums:
    \[
        \e^{\varepsilon^{\S}_{D,D'}} = \max_{C\subseteq D} g_{C,N}(a_{D,D'};z_{D}) \quad \text{and} \quad \e^{\varepsilon^{\S}_{D',D}} = \max_{C\subseteq D} f_{C,N}(a_{D,D'};z_{D}).
    \]

    Now we provide a bound that does not depend on $D$ and $D'$ (i.e., independent of $N$, $a_{D,D'}$ and $z_{D}$). By \Cref{lemma:boundsfordandout}, 
    \[
        \sum_{x'\in D}\out_{D}(x') \leq (2N-2)\out_D(x) - (N-2)m
    \]
    and 
    \[
        m\leq \out_D(x) \leq \frac{\min\{m+(N-1)M,(N-2)(\tdist(x,y)-m)+\sum_{x'\in D}\tdist(x',y)\}}{N} 
    \]
    for all $x\in D$. Fixing an order of the elements of $D$ such that $a_{D,D'}=(\tdist(x_1,y),\dots,\tdist(x_N,y))$ and $z_{D}=(\out_D(x_1),\dots,\out_D(x_N))$, we obtain that any element $z_{D}$ must be an element of $P_{a_{D,D'}}\subseteq[m,M]^N$, where $P_{a_{D,D'}}$ is the polytope defined by the previous inequalities, i.e., the polytope of \Cref{prop:PolytopeIsConvex}. Therefore, 
    \[
        \sup_{\substack{D,D'\in\D\backslash\{\varnothing\}:\\D'=D_{+y}}} \e^{\varepsilon^{\S}_{D,D'}} \leq \sup_{N\in\N}\max_{J\subseteq[N]}\max_{a\in[m,M]^N}\max_{z\in [m,M]^N} g_{C_J,N}(a;z),    
    \]
    and
    \[
        \sup_{\substack{D,D'\in\D\backslash\{\varnothing\}:\\D'=D_{+y}}} \e^{\varepsilon^{\S}_{D',D}} \leq \sup_{N\in\N}\max_{J\subseteq[N]}\max_{a\in[m,M]^N}\max_{z\in P_a} f_{C_J,N}(a;z).
    \]
    where $C_J=\{x_i\mid i\in J\}$. These suprema are computed in \Cref{th:boundinverse} and \Cref{th:bounddependingonp}, respectively, and depend only on our fixed variables $\varepsilon$, $m$, and $M$ (we recall that these proofs are computationally verified for some values of $\varepsilon$, $m$, and $M$; see \Cref{remark:computation,remark:computationInverse}). In this case, 
    \[
        \sup_{\substack{D,D'\in\D\backslash\{\varnothing\}:\\D'=D_{+y}}} \e^{\varepsilon^{\S}_{D,D'}} \leq \max\!\big\{(\e^{-\varepsilon}+(1-\e^{-\varepsilon})M)^{-1}\e^{1-\frac{1-M}{1-m}},(\e^{-\varepsilon}+(1-\e^{-\varepsilon})m)^{-1}\e^{1-\frac{m}{M}}\big\},    
    \]    
    and 
    \[
        \sup_{\substack{D,D'\in\D\backslash\{\varnothing\}:\\D'=D_{+y}}} \e^{\varepsilon^{\S}_{D,D'}} \leq \max_{p\in[0,1]}\max\{L_1(p),L_2(p)\}
    \]
    with 
    \[
        L_1(p) = (\e^\varepsilon-(\e^\varepsilon-1)(pM+(1-p)m))\e^{p\frac{M}{m}+(1-p)\frac{1-m}{1-(pM+(1-p)m)}-1}
    \]
    and
    \[
        L_2(p) = \bigg(\e^\varepsilon-(\e^\varepsilon-1)\bigg(pM+(1-p)\frac{(M+m)-pM}{2-p}\bigg)\bigg)\e^{p\frac{M}{m}+(1-p)\frac{1-\frac{(M+m)-pM}{2-p}}{1-M}-1}.
    \]

    Consequently, 
    \begin{gather*}
        \varepsilon^\S = \sup_{\substack{D,D'\in\D:\\D\sim D'}} \varepsilon^\S_{D,D'} = \max\!\bigg\{\sup_{\substack{y\in\X}}\varepsilon^\S_{\varnothing,\{y\}},\sup_{\substack{D,D'\in\D\backslash\{\varnothing\}:\\D'=D_{+y}}}\varepsilon^\S_{D,D'},\sup_{\substack{D,D'\in\D\backslash\{\varnothing\}:\\D'=D_{+y}}}\varepsilon^\S_{D',D}\bigg\} \\
        = \max_{p\in[0,1]}\max\!\bigg\{\ln(L_1(p)),\ln(L_2(p)),-\ln(\e^{-\varepsilon}+(1-\e^{-\varepsilon})M)+1-\frac{1-M}{1-m}\bigg\}.
    \end{gather*}

    We note that term $\ln((\e^{-\varepsilon}+(1-\e^{-\varepsilon})m)^{-1}\e^{1-\frac{m}{M}})$ is bounded always by the other three and is therefore superfluous in the previous expression (see \Cref{remark:computationInverse}). This value is the same as the one in the statement, concluding the proof.\qedhere
\end{proof}

\begin{proposition}\label{prop:expressionofp}
    Let $\varepsilon\geq0$, and let $m,M\in(0,1)$ such that $m\leq M$. 
    
    Let $L_1\colon[0,1]\to\R^+$ and $L_2\colon[0,1]\to\R^+$ be two functions such that
    \[
        L_1(p) = (\e^\varepsilon-(\e^\varepsilon-1)(pM+(1-p)m))\e^{p\frac{M}{m}+(1-p)\frac{1-m}{1-(pM+(1-p)m)}-1}
    \]
    and
    \[
        L_2(p) = \bigg(\e^\varepsilon-(\e^\varepsilon-1)\bigg(pM+(1-p)\frac{(M+m)-pM}{2-p}\bigg)\bigg)\e^{p\frac{M}{m}+(1-p)\frac{1-\frac{(M+m)-pM}{2-p}}{1-M}-1}
    \]
    for all $p\in[0,1]$, and let $l_1=\ln\circ L_1$ and $l_2=\ln\circ L_2$, i.e., 
    \[
        l_1(p) = \ln(\e^\varepsilon-(\e^\varepsilon-1)(pM+(1-p)m))+p\frac{M}{m}+(1-p)\frac{1-m}{1-(pM+(1-p)m)}-1
    \]
    and
    \[
        l_2(p) = \ln\!\bigg(\e^\varepsilon-(\e^\varepsilon-1)\bigg(pM+(1-p)\frac{(M+m)-pM}{2-p}\bigg)\bigg)+p\frac{M}{m}+(1-p)\frac{1-\frac{(M+m)-pM}{2-p}}{1-M}-1
    \]
    for all $p\in[0,1]$.
    
    Then for $\varepsilon\neq0$ and $m\neq M$, $L_1$ and $l_1$ achieve their maximum over $[0,1]$ when $p_1=\min\{1,\max\{V_1,0\}\}$ with
    \[
        V_1 = \begin{cases}
            -\frac{1}{3a_1}\bigg(b_1 + \sqrt[3]{\frac{D_{1,1}+\sqrt{D^2_{1,1}-4D_{1,0}^3}}{2}}+\sqrt[3]{\frac{D_{1,1}-\sqrt{D^2_{1,1}-4D_{1,0}^3}}{2}}\bigg) &\text{if $D^2_{1,1}-4D^3_{1,0}>0$} \\
            -\frac{1}{3a_1}\Big(b_1 + 2\sqrt{D_{1,0}}\cos\!\big(\frac{1}{3}\arccos\!\big(\frac{D_{1,1}}{2R_1}\big)\big)\Big) &\text{if $D^2_{1,1}-4D_{1,0}^3\leq0$} \\
        \end{cases}
    \]
    with
    \begin{gather*}
        a_1 = (\e^{\varepsilon}-1)\frac{M}{m}(M-m)^2, \\
        b_1 = -\frac{M-m}{m}((m^2-4Mm+2M)(\e^{\varepsilon}-1)+\e^{\varepsilon}M), \\
        c_1 = \frac{1-m}{m}((\e^{\varepsilon}-1)(2m^2-4Mm-m)+(3\e^{\varepsilon}-1)M), \\
        d_1 = -(1-m)\bigg((\e^{\varepsilon}-1)(m-2)+\frac{\e^{\varepsilon}}{m}\bigg),\\
        D_{1,0} = b_1^2 - 3a_1c_1, \\
        D_{1,1} = 2b_1^3-9a_1b_1c_1+27a_1^2d_1, \\
        R_1 = \sqrt{D_{1,0}^3};
    \end{gather*}
    and $L_2$ and $l_2$ achieve their maximum over $[0,1]$ when $p_2=\min\{1,\max\{V_2,0\}\}$ with
    \[
        V_2 = -\frac{1}{3a_2}\bigg(b_2 + 2\sqrt{D_{2,0}}\cos\!\bigg(\frac{1}{3}\arccos\!\bigg(\frac{D_{2,1}}{2R_2}\bigg)\bigg)\bigg)
    \]
    with
    \begin{gather*}
        a_2 = \frac{\e^{\varepsilon}-(\e^{\varepsilon}-1)m}{m}, \\
        b_2 = -\frac{6\e^{\varepsilon}-(\e^{\varepsilon}-1)(M+5m)}{m}, \\
        c_2 = \frac{1}{(1-M)m}(m((\e^{\varepsilon}-1)(m+9M-9)-\e^{\varepsilon})+4M((\e^{\varepsilon}-1)M-4\e^{\varepsilon}+1)+12\e^{\varepsilon}), \\
        d_2 = -(2\e^{\varepsilon}-(\e^{\varepsilon}-1)(M+m))\bigg(\frac{4-4M-m}{(1-M)m}\bigg)+2(\e^{\varepsilon}-1),\\
        D_{2,0} = b_2^2 - 3a_2c_2, \\
        D_{2,1} = 2b_2^3-9a_2b_2c_2+27a_2^2d_2, \\
        R_2 = \sqrt{D_{2,0}^3}.
    \end{gather*}

    For $\varepsilon=0$ and $m\neq M$, $L_1$ and $l_1$ achieve their maximum over $[m,M]$ when $p_1=\min\{1,\max\{V_1,0\}\}$ with
    \[
        V_1 = \frac{1-m}{M-m}-\frac{\sqrt{Mm(1-m)(1-M)}}{M(M-m)},
    \]
    and $L_2$ and $l_2$ achieve their maximum over $[0,1]$ when $p_2=\min\{1,\max\{V_2,0\}\}$
    with 
    \[
        V_2 = 2 - \frac{\sqrt{m(1-M)}}{1-M}. 
    \]

    For $m=M$, $L_1$ and $L_2$ are constantly $(\e^{\varepsilon}-(\e^{\varepsilon}-1)m)=(\e^{\varepsilon}-(\e^{\varepsilon}-1)M)$. 
\end{proposition}
\begin{proof}
    We first note that $L_1$, $L_2$, $l_1$ and $l_2$ are well-defined functions over $[0,1]$ (in particular, $L_1(p),L_2(p)\geq 1$ for all $p\in[0,1]$). We note too that the point that achieves the maximum for $L_1$ is the same point that achieves the maximum for $l_1$ and vice versa, since $\ln$ and $\exp$ are strictly increasing functions (and analogously for $L_2$ and $l_2$). We thus only center on finding the maximum for $L_1$ and $L_2$. 

    With the exception of the degenerate cases ($\varepsilon=0$ and $m=M$), finding the values $p_1,p_2\in[0,1]$ that maximize $L_1$ and $L_2$ consists of solving a cubic polynomial. Thus, the expression shown in the statement is just a direct application of the solution of a generic cubic polynomial. 
    
    We now describe the optimal bound more precisely. Suppose we are in the non-degenerate cases, i.e., $\varepsilon\neq0$ and $m<M$. We briefly consider $L_1$ and $L_2$ as functions over the real line. From their expression, we can see that $L_1$ has an asymptote when $p=\frac{1-m}{M-m}>1$ (i.e., $1-(pM+(1-p)m)=0$) and $L_2$ has an asymptote when $p=2$. Functions $L_1$ and $L_2$ are defined and are continuous over the real line except on their respective asymptote.
    
    The derivatives of $L_1$ and $L_2$ with respect to $p$ are, respectively,
    \[
        \frac{\partial L_1}{\partial p}(p) = -(M-m)\frac{a_1p^3+b_1p^2+c_1p+d_1}{(1-(pM+(1-p)m))^2}\e^{p\frac{M}{m}+(1-p)\frac{1-m}{1-(pM+(1-p)m)}-1}   
    \]
    and
    \[
        \frac{\partial L_2}{\partial p}(p) = -(M-m)\frac{a_2p^3+b_2p^2+c_2p+d_2}{(2-p)^3}\e^{p\frac{M}{m}+(1-p)\frac{1-\frac{(M+m)-pM}{2-p}}{1-M}-1}   
    \]
    where $a_1$, $b_1$, $c_1$, $d_1$, $a_2$, $b_2$, $c_2$ and $d_2$ are the constants of the statement (depending only on $m$, $M$ and $\varepsilon$). Note that we select the sign such that $a_1,a_2>0$ (note too that $a_1=a_2=0$ only in the degenerate cases). The critical points of $L_j$ for $j\in\{1,2\}$ are therefore the roots of $a_jp^3+b_jp^2+c_jp+d_j=0$ (except if a root matches the respective asymptote).
    There are two possibilities regarding the roots of a cubic polynomial with real coefficients: that they are all real or that there are a real and two imaginary roots (one conjugate of the other). 
    According to the general cubic formula, the three roots of the polynomial are
    \[
        -\frac{1}{3a_j}\Bigg(b_j + \xi^k\sqrt[3]{\frac{D_{j,1}+\sqrt{D^2_{j,1}-4D_{j,0}^3}}{2}}+\xi^{-k}\sqrt[3]{\frac{D_{j,1}-\sqrt{D^2_{j,1}-4D_{j,0}^3}}{2}}\Bigg)
    \]
    for $k\in\{0,1,2\}$, where $D_{j,0}$ and $D_{j,1}$ are the expressions of the statement and $\xi=\frac{1+\im\sqrt{3}}{2}=\e^{\frac{2\pi}{3}\im}$ is a primitive cubic root of the unity. We note the following, if $D^2_{j,1}-4D_{j,0}^3>0$, then both
    \[
        \sqrt[3]{\frac{D_{j,1}+\sqrt{D^2_{j,1}-4D_{j,0}^3}}{2}}\quad\text{and}\quad\sqrt[3]{\frac{D_{j,1}-\sqrt{D^2_{j,1}-4D_{j,0}^3}}{2}}
    \]
    are different real numbers and 
    \[
        \xi^k\sqrt[3]{\frac{D_{j,1}+\sqrt{D^2_{j,1}-4D_{j,0}^3}}{2}}+\xi^{-k}\sqrt[3]{\frac{D_{j,1}-\sqrt{D^2_{j,1}-4D_{j,0}^3}}{2}}
    \]
    is real for $k=0$ and imaginary for $k=1$ and $k=2$. Therefore, $a_jp^3+b_jp^2+c_jp+d_j=0$ has one real root,
    \[
        r^+_{j} = -\frac{1}{3a_j}\Bigg(b_j + \sqrt[3]{\frac{D_{j,1}+\sqrt{D^2_{j,1}-4D_{j,0}^3}}{2}}+\sqrt[3]{\frac{D_{j,1}-\sqrt{D^2_{j,1}-4D_{j,0}^3}}{2}}\Bigg)   
    \]
    and two imaginary roots. On the other hand, if $D^2_{j,1}-4D_{j,0}^3<0$, then 
    \[
        \sqrt[3]{\frac{D_{j,1}+\sqrt{D^2_{j,1}-4D_{j,0}^3}}{2}}\quad\text{and}\quad\sqrt[3]{\frac{D_{j,1}-\sqrt{D^2_{j,1}-4D_{j,0}^3}}{2}}
    \]
    are imaginary numbers, one the conjugate of the other. Indeed, we consider the polar form\footnote{The polar form of a non-zero complex number of the form $a+b\im$ ($a,b\in\R$) is defined as $R\e^{\im\theta}$ with $R=\sqrt{a^2+b^2}\neq0$ and $\theta=\operatorname{sgn}(b)\arccos(\frac{a}{R})$.} 
    of $\frac{D_{j,1}+\sqrt{D^2_{j,1}-4D_{j,0}^3}}{2}=\frac{D_{j,1}}{2}+\im\frac{\sqrt{-(D^2_{j,1}-4D_{j,0}^3)}}{2}$, which is $R_j\e^{\im\theta_j}$ with
    \[
        R_j = \sqrt{\bigg(\frac{D_{j,1}}{2}\bigg)^2+\frac{-(D_{j,1}^2-4D_{j,0}^3)}{2^2}} = \sqrt{\frac{D_{j,1}^2-(D_{j,1}^2-4D_{j,0}^3)}{4}}=\sqrt{D_{j,0}^3}
        \qquad \text{and} \qquad
        \theta_j 
        = \arccos\!\bigg(\frac{D_{j,1}}{2R_j}\bigg).
    \]
    
    Then, by the properties of the cubic root,  
    \[
        \sqrt[3]{\frac{D_{j,1}+\sqrt{D^2_{j,1}-4D_{j,0}^3}}{2}} = \sqrt[3]{R_j\e^{\im\theta_j}} = \sqrt[3]{R_j}\e^{\im\frac{\theta_j}{3}} = \sqrt{D_{j,0}}\e^{\im\frac{\theta_j}{3}}
    \]
    and 
    \[
        \sqrt[3]{\frac{D_{j,1}-\sqrt{D^2_{j,1}-4D_{j,0}^3}}{2}} = \sqrt[3]{R_j\e^{-\im\theta_j}} = \sqrt[3]{R_j}\e^{-\im\frac{\theta_j}{3}} = \sqrt{D_{j,0}}\e^{-\im\frac{\theta_j}{3}}.
    \]

    Note that $\sqrt{D_{j,0}}$ and $\sqrt{D_{j,0}^3}$ are well-defined positive numbers since $D^2_{j,1}-4D_{j,0}^3<0$ ensures that $D_{j,0}$ is positive. Therefore, 
    \[
        \xi^k\sqrt{D_{j,0}}\e^{\im\frac{\theta_j}{3}}+\xi^{-k}\sqrt{D_{j,0}}\e^{-\im\frac{\theta_j}{3}} = \sqrt{D_{j,0}}\big(\e^{\im\frac{\theta_j+2\pi k}{3}}+\e^{-\im\frac{\theta_j+2\pi k}{3}}\big) = 2\sqrt{D_{j,0}}\cos\!\bigg(\frac{\theta_j+2\pi k}{3}\bigg)
    \]
    is real for all $k\in\{0,1,2\}$, since it is the sum of an imaginary number and its conjugate. In conclusion, the three real roots are
    \[
        r^-_{j,k}\coloneqq-\frac{1}{3a_j}\Bigg(b_j + 2\sqrt{D_{j,0}}\cos\!\bigg(\frac{1}{3}\arccos\!\bigg(\frac{D_{j,1}}{2R_j}\bigg)+\frac{2\pi k}{3}\bigg)\Bigg)
    \]
    for $k\in\{0,1,2\}$. Additionally, by the definition of cosine and arccosine, we have that 
    \[
        \cos\!\bigg(\frac{1}{3}\arccos(x)+\frac{2\pi k}{3}\bigg) \in \begin{cases}
            [\frac{1}{2},1] & \text{if $k=0$,}\\
            [-\frac{1}{2},\frac{1}{2}] & \text{if $k=2$,}\\
            [-1,-\frac{1}{2}] & \text{if $k=1$.}
        \end{cases} 
    \]
    for all $x\in[-1,1]$, the domain of $\arccos$. Then, since $a_1,a_2>0$, 
    \begin{equation}
        r^-_{j,0} \leq -\frac{b_j+\sqrt{D_{j,0}}}{3a_j} \leq r^-_{j,2} \leq -\frac{b_j-\sqrt{D_{j,0}}}{3a_j} \leq r^-_{j,1}, \label{eq:three_real_roots_inequality}
    \end{equation}
    where the intermediate values correspond to when the cosine equals $\pm\frac{1}{2}$.

    Finally, in the case where $D^2_{j,1}-4D_{j,0}^3=0$, we have that 
    \[
        \sqrt[3]{\frac{D_{j,1}+\sqrt{D^2_{j,1}-4D_{j,0}^3}}{2}}=\sqrt[3]{\frac{D_{j,1}-\sqrt{D^2_{j,1}-4D_{j,0}^3}}{2}}=\sqrt[3]{\frac{D_{j,1}}{2}}
    \]
    and, therefore, 
    \[
        \xi^k\sqrt[3]{\frac{D_{j,1}}{2}}+\xi^{-k}\sqrt[3]{\frac{D_{j,1}}{2}} = (\xi^k + \xi^{-k})\sqrt[3]{\frac{D_{j,1}}{2}}
    \]
    is real for $k\in\{0,1,2\}$ (for the cases $k\in\{1,2\}$, it corresponds to the sum of a complex number and its conjugate). In this case, there is a root with multiplicity higher than 1, and the roots correspond also to the expressions $r^-_{j,0}$, $r^-_{j,1}$ and $r^-_{j,2}$ previously given for case $D^2_{j,1}-4D_{j,0}^3<0$.
    
    Now that we have an expression of the real roots of the polynomials $a_1p^3+b_1p^2+c_1p+d_1=0$ and $a_2p^3+b_2p^2+c_2p+d_2=0$, and thus of the critical points of $L_1$ and $L_2$, we find the global maximum in the interval $p\in[0,1]$. 

    We study functions $L_1$ and $L_2$ separately. 

    Function $L_1$ is continuous over $\R\backslash\{\frac{1-m}{M-m}\}$, has a unique zero at $p=\frac{1-m}{M-m}+\frac{1}{(\e^{\varepsilon}-1)(M-m)}>\frac{1-m}{M-m}>1$ and is positive at $p\in(-\infty,\frac{1-m}{M-m})\cup(\frac{1-m}{M-m},\frac{1-m}{M-m}+\frac{1}{(\e^{\varepsilon}-1)(M-m)})$ and negative at $p\in(\frac{1-m}{M-m}+\frac{1}{(\e^{\varepsilon}-1)(M-m)},\infty)$. The limits of function $L_1$ to $\pm\infty$ and at the asymptote are
    \begin{gather*}
        \lim_{p\to-\infty} L_1(p)=0, \quad \lim_{p\to (\frac{1-m}{M-m})^-} L_1(p)=0, \quad
        \lim_{p\to (\frac{1-m}{M-m})^+} L_1(p)=\infty \quad\text{and}\quad
        \lim_{p\to+\infty} L_1(p)=-\infty.
    \end{gather*}

    Consequently, $L_1$ has at least one local maxima in $(-\infty,\frac{1-m}{M-m})$ by an extension of Rolle's theorem to infinite intervals. This maxima is at $p=r^+_1$ when $a_1p^3 + b_1p^2 + c_1p + d_1 = 0$ had one real root ($D^2_{1,1}-4D^3_{j,0}>0$).  We now will see that in the case we have three real roots (i.e., $D^2_{j,1}-4D_{j,0}^3\leq0$), we have that
    \begin{equation}\label{eq:inequalityforarccos}
        -\frac{b_1-\sqrt{D_{1,0}}}{3a_1} > \frac{1-m}{M-m},
    \end{equation}
    which implies that $r^-_{1,1}\geq \frac{1-m}{M-m}$ by \Cref{eq:three_real_roots_inequality}. Indeed, substituting the expression $D_{1,0}=b_1^2-3a_1c_1$, we obtain that \Cref{eq:inequalityforarccos} is equivalent to
    \begin{gather*}
        \sqrt{b_1^2 - 3a_1c_1} > 3a_1\frac{1-m}{M-m}+b_1,
    \end{gather*}
    and therefore it suffices to see that
    \[
        b^2_1-3a_1c_1 > \bigg(3a_1\frac{1-m}{M-m}+b_1\bigg)^2
    \]
    holds. So, this inequality corresponds to 
    \begin{align*}
        b^2_1-3a_1c_1 &> \bigg(3a_1\frac{1-m}{M-m}+b_1\bigg)^2 \\
        \Longleftrightarrow b^2_1-3a_1c_1 &> \bigg(3a_1\frac{1-m}{M-m}\bigg)^2 + b^2_1 + 6a_1b_1\frac{1-m}{M-m} \\
        \Longleftrightarrow 0 &> 3a_1\bigg(\frac{1-m}{M-m}\bigg)^2 + 2b_1\frac{1-m}{M-m}+c_1,
    \end{align*}
    and substituting the expressions of $a_1$, $b_1$ and $c_1$, it corresponds to 
    \begin{align*}
        0>{}&  3(\e^{\varepsilon}-1)M(1-m) - 2((m^2-4Mm+2M)(\e^{\varepsilon}-1)+\e^{\varepsilon}M) 
        + (\e^{\varepsilon}-1)(2m^2-4Mm-m)+(3\e^{\varepsilon}-1)M \\
        \Longleftrightarrow 0>{}&  3(\e^{\varepsilon}-1)(M-Mm) - (2m^2-8Mm+4M)(\e^{\varepsilon}-1)-2\e^{\varepsilon}M
        + (\e^{\varepsilon}-1)(2m^2-4Mm-m)+(\e^{\varepsilon}-1)M+2\e^{\varepsilon}M, \\
        \Longleftrightarrow 0>{}&  (\e^{\varepsilon}-1)(3M-3Mm -2m^2+8Mm-4M+2m^2-4Mm-m+M), \\
        \Longleftrightarrow 0>{}&  -(\e^{\varepsilon}-1)m(1-M),
    \end{align*}
    which clearly is satisfied. Thus, $1<r^-_{1,1}$. Computing the second derivative of $L_1$ at $p=r^-_{1,1}$, we obtain that 
    \[
        \frac{\partial^2L_1}{\partial p^2} (r^-_{1,1}) = -(M-m)\frac{(r^-_{1,0}-r^-_{1,1})(r^-_{1,2}-r^-_{1,1})}{(1-(r^-_{1,1}M+(1-r^-_{1,1})m))^2}\e^{r^-_{1,1}\frac{M}{m}+(1-r^-_{1,1})\frac{1-m}{1-(r^-_{1,1}M+(1-r^-_{1,1})m)}-1},
    \]
    which is non-zero if and only if $r^-_{1,1}\neq r^-_{1,2}$. Thus, either $1<r^-_{1,1}=r^-_{1,2}$ or $r^-_{1,1}$ is a maximum or minimum of $L_1$. This last case implies that $r^-_{1,2}$ must either be a minimum or maximum (respectively) and that $r^-_{1,2}>\frac{1-m}{M-m}>1$ since otherwise we would have a contradiction with the limits at $\pm\infty$ and at the asymptote. Consequently, $r^-_{1,0}$ must be the critical point at $(-\infty,\frac{1-m}{M-m})$.
       
    In conclusion, if $D^2_{1,1}-4D^3_{1,0}>0$, then $L_1$ has one local maxima, $r^+_{1}$, which must be in interval $(-\infty,\frac{1-m}{M-m})$. On the other hand, when $D^2_{1,1}-4D^3_{1,0}\leq 0$, then $L_1$ has one local maxima, $r^-_{1,0}$, which must be in the interval $(-\infty,\frac{1-m}{M-m})$, and two other critical points in $(\frac{1-m}{M-m},\infty)$. Since in both cases, $r^+_{1}$ and $r^-_{1,0}$ are the only critical points and local maxima in the interval $(-\infty,\frac{1-m}{M-m})\supseteq[0,1]$, we conclude that
    \[
        \argmax_{p\in[0,1]} L_1(p) = \min\{1,\max\{V_1,0\}\},  
    \]
    with $V_1$ the value in the statement. 

    Function $L_2$ is continuous over $\R\backslash\{2\}$, has a unique zero at $p=1+\frac{\e^{\varepsilon}-(\e^{\varepsilon}-1)M}{\e^{\varepsilon}-(\e^{\varepsilon}-1)m}\in(1,2)$ and is positive at $p\in(-\infty,1+\frac{\e^{\varepsilon}-(\e^{\varepsilon}-1)M}{\e^{\varepsilon}-(\e^{\varepsilon}-1)m})\cup(2,\infty)$ and negative at $p\in(1+\frac{\e^{\varepsilon}-(\e^{\varepsilon}-1)M}{\e^{\varepsilon}-(\e^{\varepsilon}-1)m},2)$. The limits of function $L_2$ to $\pm\infty$ and at the asymptote are
    \[
        \lim_{p\to-\infty} L_2(p)=0, \quad \lim_{p\to 2^-} L_2(p)=0, \quad
        \lim_{p\to 2^+} L_2(p)=\infty \quad\text{and}\quad
        \lim_{p\to+\infty} L_2(p)=\infty.
    \]

    Consequently, by extensions of Rolle's theorem to infinity, we can affirm that $L_2$ has three critical points: a local maximum at $(-\infty,1+\frac{\e^{\varepsilon}-(\e^{\varepsilon}-1)M}{\e^{\varepsilon}-(\e^{\varepsilon}-1)m})$, a local minimum at $(1+\frac{\e^{\varepsilon}-(\e^{\varepsilon}-1)M}{\e^{\varepsilon}-(\e^{\varepsilon}-1)m},2)$ and another local minimum at $(2,\infty)$. These critical points correspond respectively to $r^-_{2,0}$, $r^-_{2,2}$ and $r^-_{2,1}$. Therefore, since $r^-_{2,0}$ is the only critical point that can be in $[0,1]$, 
    \[
        \argmax_{p\in[0,1]} L_2(p) = \min\{1,\max\{r^-_{2,0},0\}\}.  
    \]

    Finally, we tackle the degenerate cases. For $\varepsilon=0$, $L_1$ is continuous and positive over $\R\backslash\{\frac{1-m}{M-m}\}$. The limits of function $L_1$ to $\pm\infty$ and at the asymptote are
    \begin{gather*}
        \lim_{p\to-\infty} L_1(p)=0, \quad \lim_{p\to (\frac{1-m}{M-m})^-} L_1(p)=0, \quad
        \lim_{p\to (\frac{1-m}{M-m})^+} L_1(p)=\infty \quad\text{and}\quad
        \lim_{p\to+\infty} L_1(p)=\infty.
    \end{gather*}

    Function $L_1$ has two critical points since the derivative corresponds to a polynomial of degree $2$ multiplied by some positive terms. Thus, the critical points correspond to the roots of such polynomial:
    \[
        p_1 = \frac{1-m}{M-m}-\frac{\sqrt{Mm(1-m)(1-M)}}{M(M-m)} \quad\text{and}\quad p_2 = \frac{1-m}{M-m}+\frac{\sqrt{Mm(1-m)(1-M)}}{M(M-m)}, 
    \]
    one at each side of the asymptote. Consequently, we can affirm that $p_1$ is a local maximum and $p_2>1$ is a local minimum. Thus,  
    \[
        \argmax_{p\in[0,1]} L_1(p) = \min\{1,\max\{p_1,0\}\}.
    \]

    For $\varepsilon=0$, function $L_2$ behaves similarly to $L_1$. It is continuous and positive over $\R\backslash\{2\}$. The limits of function $L_2$ to $\pm\infty$ and at the asymptote are
    \begin{gather*}
        \lim_{p\to-\infty} L_2(p)=0, \quad \lim_{p\to 2^-} L_2(p)=0, \quad
        \lim_{p\to 2^+} L_2(p)=\infty \quad\text{and}\quad
        \lim_{p\to+\infty} L_2(p)=\infty.
    \end{gather*}

    Function $L_2$ has two critical points since the derivative corresponds to a polynomial of degree $2$ multiplied by some positive terms. Thus, the critical points correspond to the roots of such polynomial:
    \[
        p_1 = 2-\frac{\sqrt{m(1-M)}}{1-M} \quad\text{and}\quad p_2 = 2+\frac{\sqrt{m(1-M)}}{1-M}, 
    \]
    one at each side of the asymptote. Consequently, we can affirm that $p_1$ is a local maximum and $p_2>1$ is a local minimum. Thus,  
    \[
        \argmax_{p\in[0,1]} L_2(p) = \min\{1,\max\{p_1,0\}\}.
    \]

    Finally, for the case $m=M$, we can easily see that $L_1$ and $L_2$ are constantly $(\e^{\varepsilon}-(\e^{\varepsilon}-1)m)=(\e^{\varepsilon}-(\e^{\varepsilon}-1)M)$. 
\end{proof}

\begin{remark}[Tightness of the privacy parameters in \Cref{th:outlierscoresuppression}]\label{remark:tightnessoutliersuppressiontheorem}
    The bounds we provide are tight in the colored areas in \Cref{fig:SuppressionTheoremSimplifiedAreas} with respect to \Cref{th:unboundedcase} (see \Cref{prop:tightnessoutlierscore}), but we are not able to show if they are tight in general. The pair of neighboring databases $D$ and $D_{+y}$ that achieves these ``tight'' bounds corresponds to having all elements of $D$ be the closest possible records (i.e., all copies of the same element $x$) and $y$ be the furthest record from $x$ (red area); and having all elements of $D$ be the furthest away from each other and $y$ be the closest to all of them (green and purple areas; possible under certain metrics). We believe that the values that achieve the tight bounds in the uncolored region would correspond to a combination of these cases, i.e., a portion of records are as close as possible, while the rest are as far away as possible. In this case, the portion would be determined by the value $p$ that achieves the maximum. 

    \begin{figure}[H]
        \centering
        \includegraphics[width=0.5\columnwidth]{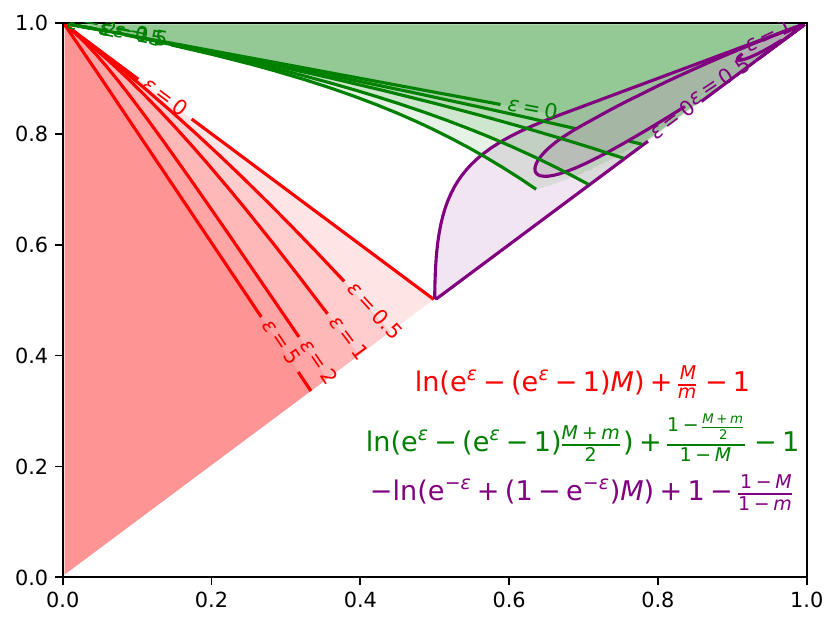}
        \caption{The colored areas show the values of $m$ ($x$-axis) and $M$ ($y$-axis) such that $\varepsilon^\S=\varepsilon^\S(\varepsilon,m,M)$ simplifies, which depends on the given $\varepsilon$. The expression of $\varepsilon^\S$ simplifies for the values $m$ and $M$ within a colored region to the expression of the same color. These colored areas are also where the privacy parameters are tight with respect to \Cref{th:unboundedcase}.}
        \label{fig:SuppressionTheoremSimplifiedAreas}
    \end{figure}
\end{remark}

\begin{proposition}\label{prop:tightnessoutlierscore}
    The privacy parameters given in \Cref{th:outlierscoresuppression} are tight to those of \Cref{th:unboundedcase} (i.e., are equal) when the maximum is achieved at $p=0$ or $p=1$ (including when the maximum is $l_3$). Note that this condition represents the values of $m$, $M$, and $\varepsilon$ colored in \Cref{fig:SuppressionTheoremSimplifiedAreas}. Furthermore, when the maximum is achieved at $p=1$, tightness is achieved independently of the choice of $\dist$ (for $|\X|\geq 2$). 
\end{proposition}
\begin{proof}
    We first see that the result holds true when the maximum is achieved at $p=1$ and that it is independent of the choice of $\dist$. Note that $l_1(1)=l_2(1)$. Since the privacy parameters of \Cref{th:unboundedcase} are always upper bounded by the privacy parameters of \Cref{th:outlierscoresuppression} (this follows from the proof of \Cref{th:outlierscoresuppression}), we will prove that for all $m,M\in(0,1)$ with $m\leq M$, and all distances $\dist\colon\X\times\X\to[0,1]$, there exist a sequence $((D_N,D'_N))_{N\in\N}$ and a pair $(\overline{D},\overline{D}')$ of neighboring databases such that 
    \[
        \max\!\big\{\varepsilon^\S_{D_N,D'_N},\varepsilon^\S_{D'_N,D_N}\big\} \xrightarrow[]{N\to\infty} \max\{l_1(1),l_3\} \qquad\text{and}\qquad \max\!\big\{\delta^\S_{\overline{D},\overline{D}'},\delta^\S_{\overline{D},\overline{D}'}\big\} = \delta(1-m)
    \]
    for $\varepsilon^\S_{D,D'}$, $\varepsilon^\S_{D',D}$, $\delta^\S_{D,D'}$, and $\delta^\S_{D',D}$ denotes the expressions given in \Cref{th:unboundedcase}. This fact ensures that the privacy parameters of \Cref{th:outlierscoresuppression} are tight with respect to \Cref{th:unboundedcase} if $l_1(1)=l_2(1)$ or $l_3$ achieves the maximum in the expression of $\varepsilon^\S$ in \Cref{th:outlierscoresuppression}, or equivalently, if the maximum is achieved when $p=1$. 

    Following the steps at the start of the proof of \Cref{th:outlierscoresuppression}, we have that
    \[
        \Prob\{(\M\circ\S)(D)\in\meas\} \leq \e^{\varepsilon^\S_{D,D'}} \Prob\{(\M\circ\S)(D')\in\meas\} + \delta^\S_{D,D'}
    \] 
    for all $D\in\D$ and $D'=D_{+y}\in\D$, and all measurable $\meas\subseteq\Ran$ with  
    \begin{align*}
        \e^{\varepsilon^{\S}_{D,D'}} = \max_{\substack{C\subseteq D}}\bigg(\bigg(\e^{-\varepsilon}+(1-\e^{-\varepsilon})\frac{m+\sum_{x\in D} \tdist(x,y)}{N+1}\bigg)^{-1}\prod_{x\in C} \frac{N+1}{N+\frac{1-\tdist(x,y)}{1-\out_{D}(x)}} \prod_{x\in D\backslash C} \frac{N+1}{N+ \frac{\tdist(x,y)}{\out_D(x)}}\bigg),
    \end{align*}
    and
    \begin{align*}
        \delta^{\S}_{D,D'}=\delta\frac{\e^{-\varepsilon}(1-\out_{D'}(y))}{\out_{D'}(y)+\e^{-\varepsilon}(1-\out_{D'}(y))};
    \end{align*}
    and satisfies
    \[
        \Prob\{(\M\circ\S)(D')\in\meas\} \leq \e^{\varepsilon^\S_{D',D}} \Prob\{(\M\circ\S)(D)\in\meas\} + \delta^\S_{D',D}
    \]
    for all $D\in\D$ and $D'=D_{+y}\in\D$, and all measurable $\meas\subseteq\Ran$ with  
    \begin{align*}
        \e^{\varepsilon^{\S}_{D',D}} = \max_{\substack{C\subseteq D}}\bigg(\bigg(\e^{\varepsilon}-(\e^{\varepsilon}-1)\frac{m+\sum_{x\in D} \tdist(x,y)}{N+1}\bigg) \prod_{x\in C}\frac{N+\frac{1-\tdist(x,y)}{1-\out_{D}(x)}}{N+1} \prod_{x\in D\backslash C} \frac{N+\frac{\tdist(x,y)}{\out_D(x)}}{N+1}\bigg)
    \end{align*}
    and
    \begin{align*}
        \delta^{\S}_{D',D} = \delta(1-\out_{D'}(y)).
    \end{align*} 

    We first prove that the pair $(\overline{D},\overline{D}')$ exists, which is direct: We select $\overline{D}$ as the database of only copies of $x\in\X$ and $D'=D_{+x}$ (which also only has copies of $x$). Then $\out_{\overline{D}'}(x)=m$ and   
    \[
        \max\!\big\{\delta^\S_{\overline{D},\overline{D}'},\delta^\S_{\overline{D},\overline{D}'}\big\} = \max\!\bigg\{\delta\frac{\e^{-\varepsilon}(1-m)}{m+\e^{-\varepsilon}(1-m)},\delta(1-m)\bigg\} = \delta(1-m)
    \]

    We now prove that the sequence $((D_N,D'_N))_{N\in\N}$ of neighboring databases exists for $|\X|\geq2$. Since $\dist\colon\X\times\X\to[0,1]$ is normalized, there exists $x,y\in\X$ such that $\dist(x,y)=1$ or $x\in\X$ and a sequence $(y_\alpha)_{\alpha}$ of elements in $\X$ such that $d_\alpha\coloneqq\dist(x,y_\alpha)\to1$ when $\alpha\to\infty$. We consider only the second case since the first is included in the second.

    We consider the database $D_{N,\alpha}$ with only $N\geq1$ copies of $x$ and its neighboring database $D'_{N,\alpha}\coloneqq D_{N,\alpha+\{y_\alpha\}}$. It is clear that $\tdist(x,y_\alpha)=m+(M-m)d_\alpha$ and
    \begin{gather*}
        \out_{D_{N,\alpha}}(x) = \frac{1}{N}\sum_{x'\in D_{N,\alpha}} \tdist(x',x) = m.
    \end{gather*}

    Substituting the values in the previous expressions of $\e^{\varepsilon^\S_{D,D'}}$ and $\e^{\varepsilon^\S_{D',D}}$, we obtain that 
    \begin{gather*}
        \exp(\varepsilon^{\S}_{D_{N,\alpha},D'_{N,\alpha}}) = \max_{\substack{C\subseteq D_{N,\alpha}}}\bigg(\bigg(\e^{-\varepsilon}+(1-\e^{-\varepsilon})\bigg(m + (M-m)\frac{Nd_\alpha}{N+1}\bigg)\bigg)^{-1}\prod_{x\in C} \frac{N+1}{N+\frac{1-(m + (M-m)d_\alpha)}{1-m}} \prod_{x\in D_{N,\alpha}\backslash C} \frac{N+1}{N+ \frac{m + (M-m)d_\alpha}{m}}\bigg),\\
        \exp(\varepsilon^{\S}_{D_{N,\alpha},D'_{N,\alpha}}) = \max_{\substack{C\subseteq D_{N,\alpha}}}\bigg(\bigg(\e^{\varepsilon}-(\e^{\varepsilon}-1)\bigg(m+(M-m)\frac{Nd_\alpha}{N+1}\bigg)\bigg) \prod_{x\in C}\frac{N+\frac{1-(m+(M-m)d_\alpha)}{1-m}}{N+1} \prod_{x\in D_{N,\alpha}\backslash C} \frac{N+ \frac{m+(M-m)d_\alpha}{m}}{N+1}\bigg).
    \end{gather*} 

    Furthermore, using that $\frac{1-(m+(M-m)d_\alpha)}{1-m}\leq 1\leq\frac{m+(M-m)d_\alpha}{m}$, allows us to solve the maximums: 
    \begin{gather*}
        \exp(\varepsilon^{\S}_{D_{N,\alpha},D'_{N,\alpha}}) = \bigg(\e^{-\varepsilon}+(1-\e^{-\varepsilon})\bigg(m + (M-m)\frac{Nd_\alpha}{N+1}\bigg)\bigg)^{-1}\bigg(\frac{N+1}{N+\frac{1-(m + (M-m)d_\alpha)}{1-m}}\bigg)^N,\\
        \exp(\varepsilon^{\S}_{D_{N,\alpha},D'_{N,\alpha}}) = \bigg(\e^{\varepsilon}-(\e^{\varepsilon}-1)\bigg(m+(M-m)\frac{Nd_\alpha}{N+1}\bigg)\bigg) \bigg(\frac{N+ \frac{m+(M-m)d_\alpha}{m}}{N+1}\bigg)^N.
    \end{gather*} 

    Now, we compute the limits when $N\to\infty$ and $\alpha\to\infty$. We will use the well-known limit $\lim_{n\to\infty}(1+\frac{a}{n+1})^n=\e^a$ for all $a\geq0$:
    \begin{align*}
        \lim_{N\to\infty}\lim_{\alpha\to\infty} \exp(\varepsilon^{\S}_{D_{N,\alpha},D'_{N,\alpha}}) &= \lim_{N\to\infty}\lim_{\alpha\to\infty} \bigg(\e^{-\varepsilon}+(1-\e^{-\varepsilon})\bigg(m + (M-m)\frac{Nd_\alpha}{N+1}\bigg)\bigg)^{-1}\bigg(\frac{N+1}{N+\frac{1-(m + (M-m)d_\alpha)}{1-m}}\bigg)^N \\
        &= \lim_{N\to\infty}\bigg(\e^{-\varepsilon}+(1-\e^{-\varepsilon})\bigg(m + (M-m)\frac{N}{N+1}\bigg)\bigg)^{-1}\bigg(\frac{N+1}{N+\frac{1-M}{1-m}}\bigg)^N \\
        &= (\e^{-\varepsilon}+(1-\e^{-\varepsilon})M)^{-1}\lim_{N\to\infty}\bigg(\frac{N+\frac{1-M}{1-m}}{N+1}\bigg)^{-N} \\
        &= (\e^{-\varepsilon}+(1-\e^{-\varepsilon})M)^{-1}\lim_{N\to\infty}\bigg(1+\frac{\frac{1-M}{1-m}-1}{N+1}\bigg)^{-N} \\
        &= (\e^{-\varepsilon}+(1-\e^{-\varepsilon})M)^{-1}\lim_{N\to\infty}\bigg(1+\frac{\frac{1-M}{1-m}-1}{N+1}\bigg)^{-N} \\
        &= (\e^{-\varepsilon}+(1-\e^{-\varepsilon})M)^{-1}\e^{-(\frac{1-M}{1-m}-1)} = \exp(l_3),
    \end{align*}
    and
    \begin{align*}
        \lim_{N\to\infty}\lim_{\alpha\to\infty} \exp(\varepsilon^{\S}_{D_{N,\alpha},D'_{N,\alpha}}) &= \lim_{N\to\infty}\lim_{\alpha\to\infty} \bigg(\e^{\varepsilon}-(\e^{\varepsilon}-1)\bigg(m+(M-m)\frac{Nd_\alpha}{N+1}\bigg)\bigg) \bigg(\frac{N+ \frac{m+(M-m)d_\alpha}{m}}{N+1}\bigg)^N \\
        &=\lim_{N\to\infty} \bigg(\e^{\varepsilon}-(\e^{\varepsilon}-1)\bigg(m+(M-m)\frac{N}{N+1}\bigg)\bigg) \bigg(\frac{N+ \frac{M}{m}}{N+1}\bigg)^N \\
        &= (\e^{\varepsilon}-(\e^{\varepsilon}-1)M) \lim_{N\to\infty}\bigg(1+\frac{ \frac{M}{m}-1}{N+1}\bigg)^N \\
        &= (\e^{\varepsilon}-(\e^{\varepsilon}-1)M) \e^{\frac{M}{m}-1} = \exp(l_1(1)).
    \end{align*}

    Thus,
    \[
        \lim_{N\to\infty}\lim_{\alpha\to\infty}\max\{\varepsilon^{\S}_{D_{N,\alpha},D'_{N,\alpha}},\varepsilon^{\S}_{D'_{N,\alpha},D_{N,\alpha}}\} = \max\{l_3,l_1(1)\}
    \]
    which concludes the first part of the proof. 

    Now we see that there exists a $\dist\colon\X\times\X\to[0,1]$ such that the statement holds true when the maximum is $l_2(0)$. We note that this covers the case where the maximum is achieved at $p=0$ since $l_1(0)=\ln(\e^{\varepsilon}-(\e^\varepsilon-1)m)\leq l_2(0)$, and excludes the case when $l_3$ is the maximum that we covered before. For this, we impose $\X$ to be infinite, let $x_0\in\X$, and we consider the distance 
    \[
        \dist(x,x') = \begin{cases}
            0 &\text{if $x=x'$,} \\
            \frac{1}{2} & \text{if either $x=x_0$ or $x'=x_0$,} \\
            1 & \text{otherwise.}
        \end{cases}
    \]

    As in the previous case, it is sufficient to see that for all $m,M\in(0,1)$ with $m\leq M$, there exist a sequence $((D_N,D'_N))_{N\in\N}$ and a pair $(\overline{D},\overline{D}')$ of neighboring databases such that 
    \[
        \max\!\big\{\varepsilon^\S_{D_N,D'_N},\varepsilon^\S_{D'_N,D_N}\big\} \xrightarrow[]{N\to\infty} l_2(0) \qquad\text{and}\qquad \max\!\big\{\delta^\S_{\overline{D},\overline{D}'},\delta^\S_{\overline{D},\overline{D}'}\big\} = \delta(1-m). 
    \]

    In particular, we can select the same pair $(\overline{D},\overline{D}')$ as in the previous case. Now, we prove that a sequence $((D_N,D'_N))_{N\in\N}$ of neighboring databases exists. We define $D_N=\{x_1,\dots,x_N\}$ where each element in $D_N$ is distinct from all the others, and none of them are $x_0$ (this exists because $\X$ is infinite), and we consider $D'_N=D_{N+x_0}$.

    It is clear that, for all $i\in[n]$, $\tdist(x_i,x_0)=m+(M-m)\frac{1}{2}=\frac{M+m}{2}$ and
    \begin{gather*}
        \out_{D_{N}}(x_i) = \frac{1}{N}\sum_{x'\in D_{N}} \tdist(x',x_i) = \frac{m+(N-1)M}{N}.
    \end{gather*}

    Substituting the values in the previous expressions of $\e^{\varepsilon^\S_{D,D'}}$ and $\e^{\varepsilon^\S_{D',D}}$, we obtain that 
    \begin{gather*}
        \exp(\varepsilon^{\S}_{D_{N},D'_{N}}) = \max_{\substack{C\subseteq D_{N}}}\left(\bigg(\e^{-\varepsilon}+(1-\e^{-\varepsilon})\frac{m+N\frac{M+m}{2}}{N+1}\bigg)^{-1}\prod_{x\in C} \frac{N+1}{N+\frac{1-\frac{M+m}{2}}{1-\frac{m+(N-1)M}{N}}} \prod_{x\in D_{N}\backslash C} \frac{N+1}{N+ \frac{(M-m)\frac{1}{2}}{\frac{m+(N-1)M}{N}}}\right),\\
        \exp(\varepsilon^{\S}_{D_{N},D'_{N}}) = \max_{\substack{C\subseteq D_{N}}}\left(\bigg(\e^{\varepsilon}-(\e^{\varepsilon}-1)\frac{m+N\frac{M+m}{2}}{N+1}\bigg) \prod_{x\in C}\frac{N+\frac{1-\frac{M+m}{2}}{1-\frac{m+(N-1)M}{N}}}{N+1} \prod_{x\in D_{N}\backslash C} \frac{N+ \frac{\frac{M+m}{2}}{\frac{m+(N-1)M}{N}}}{N+1}\right).
    \end{gather*} 

    Now, for $N\geq 2$, we have that  using that $\frac{\frac{M+m}{2}}{\frac{m+(N-1)M}{N}}\leq 1\leq\frac{1-\frac{M+m}{2}}{1-\frac{m+(N-1)M}{N}}$, allows us to solve the maximums: 
    \begin{gather*}
        \exp(\varepsilon^{\S}_{D_{N},D'_{N}}) = \bigg(\e^{-\varepsilon}+(1-\e^{-\varepsilon})\frac{m+N\frac{M+m}{2}}{N+1}\bigg)^{-1}\left(\frac{N+1}{N+\frac{1-\frac{M+m}{2}}{1-\frac{m+(N-1)M}{N}}}\right)^N,\\
        \exp(\varepsilon^{\S}_{D_{N},D'_{N}}) = \bigg(\e^{\varepsilon}-(\e^{\varepsilon}-1)\frac{m+N\frac{M+m}{2}}{N+1}\bigg) \left(\frac{N+\frac{1-\frac{M+m}{2}}{1-\frac{m+(N-1)M}{N}}}{N+1}\right)^N.
    \end{gather*}

    Now, we compute the limit when $N\to\infty$. We will use the well-known limit $\lim_{n\to\infty}(1+\frac{\alpha}{n+1})^n=\e^\alpha$ for all $\alpha\geq0$:
    \begin{align*}
        \lim_{N\to\infty} \exp(\varepsilon^{\S}_{D_{N},D'_{N}}) &= \lim_{N\to\infty}\bigg(\e^{-\varepsilon}+(1-\e^{-\varepsilon})\frac{m+N\frac{M+m}{2}}{N+1}\bigg)^{-1}\left(\frac{N+1}{N+\frac{1-\frac{M+m}{2}}{1-\frac{m+(N-1)M}{N}}}\right)^N \\
        &= \bigg(\e^{-\varepsilon}+(1-\e^{-\varepsilon})\frac{M+m}{2}\bigg)^{-1}\lim_{N\to\infty}\left(\frac{N+1}{N+\frac{1-\frac{M+m}{2}}{1-\frac{m+(N-1)M}{N}}}\right)^N \\
        &= \bigg(\e^{-\varepsilon}+(1-\e^{-\varepsilon})\frac{M+m}{2}\bigg)^{-1}\lim_{N\to\infty}\left(1+\frac{\frac{1-\frac{M+m}{2}}{1-\frac{m+(N-1)M}{N}}-1}{N+1}\right)^{-N} \\
        &= \bigg(\e^{-\varepsilon}+(1-\e^{-\varepsilon})\frac{M+m}{2}\bigg)^{-1}\e^{1-\frac{1-\frac{M+m}{2}}{1-M}} \eqqcolon \exp(l_4), 
    \end{align*}
    and
    \begin{align*}
        \lim_{N\to\infty}\exp(\varepsilon^{\S}_{D_{N},D'_{N}}) &= \lim_{N\to\infty}\left(\e^{\varepsilon}-(\e^{\varepsilon}-1)\frac{m+N\frac{M+m}{2}}{N+1}\right) \left(\frac{N+\frac{1-\frac{M+m}{2}}{1-\frac{m+(N-1)M}{N}}}{N+1}\right)^N \\
        &=\bigg(\e^{\varepsilon}-(\e^{\varepsilon}-1)\frac{M+m}{2}\bigg) \lim_{N\to\infty}\left(\frac{N+\frac{1-\frac{M+m}{2}}{1-\frac{m+(N-1)M}{N}}}{N+1}\right)^N \\
        &=\bigg(\e^{\varepsilon}-(\e^{\varepsilon}-1)\frac{M+m}{2}\bigg) \lim_{N\to\infty}\left(1+\frac{\frac{1-\frac{M+m}{2}}{1-\frac{m+(N-1)M}{N}}-1}{N+1}\right)^N \\
        &=\bigg(\e^{\varepsilon}-(\e^{\varepsilon}-1)\frac{M+m}{2}\bigg) \e^{\frac{1-\frac{M+m}{2}}{1-M}-1} = \exp(l_2(0)) \\
    \end{align*}

    Thus,
    \[
        \lim_{N\to\infty}\max\{\varepsilon^{\S}_{D_{N},D'_{N}},\varepsilon^{\S}_{D'_{N},D_{N}}\} = \max\{l_4,l_2(0)\}.
    \]

    Finally, to complete the proof, we see that $l_4\leq l_2(0)$. This follows easily directly from the fact that 
    \[
        -\ln\!\bigg(\e^{-\varepsilon}+(1-\e^{-\varepsilon})\frac{M+m}{2}\bigg)\leq \ln\!\bigg(\e^{\varepsilon}-(\e^{\varepsilon}-1)\frac{M+m}{2}\bigg) \quad\text{and}\quad 1-\frac{1-\frac{M+m}{2}}{1-M} \leq 0 \leq \frac{1-\frac{M+m}{2}}{1-M}-1,
    \]
    where the first inequality is equivalent to the inequality $\frac{M+m}{2}\leq 1$.
\end{proof}

\begin{corollary}\label{cor:final}
    The outlier-score suppression algorithm with parameters $m$ and $M$ cannot provide a greater privacy amplification than the uniform Poisson sampling with sampling rate $1-m$. 
\end{corollary}
\begin{proof}
    This is a direct corollary of \Cref{prop:final} with the observation that
    \[
        \sup_{D'\in\D}\sup_{y\in D'}\Prob\{y\in\S(D')\} = 1-m. \qedhere
    \]
\end{proof}

\subsubsection{Additional Theorems for the Proof of the Bound of \Cref{th:outlierscoresuppression}}

This section includes the additional results for obtaining the bound given in \Cref{th:outlierscoresuppression}. The proofs are lengthy because we aim to provide the tightest inequalities possible.

We provide a summary of the results below. We need to maximize the expressions
\[
    \e^{\varepsilon^{\S}_{D,D'}} = \max_{\substack{C\subseteq D}}\bigg(\bigg(\e^{-\varepsilon}+(1-\e^{-\varepsilon})\frac{m+\sum_{x\in D} \tdist(x,y)}{N+1}\bigg)^{-1}\prod_{x\in C} \frac{N+1}{N+\frac{1-\tdist(x,y)}{1-\out_{D}(x)}} \prod_{x\in D\backslash C} \frac{N+1}{N+ \frac{\tdist(x,y)}{\out_D(x)}}\bigg),
\]
and 
\[
    \e^{\varepsilon^{\S}_{D',D}}= \max_{\substack{C\subseteq D}}\bigg(\bigg(\e^{\varepsilon}-(\e^{\varepsilon}-1)\frac{m+\sum_{x\in D} \tdist(x,y)}{N+1}\bigg) \prod_{x\in C}\frac{N+\frac{1-\tdist(x,y)}{1-\out_{D}(x)}}{N+1} \prod_{x\in D\backslash C} \frac{N+ \frac{\tdist(x,y)}{\out_D(x)}}{N+1}\bigg)
\]
\Cref{lemma:convex-large} verify that the expression of $\exp(\varepsilon^{\S}_{D',D})$ is convex with respect to the variables $z_{D}$. 
\Cref{lemma:boundsfordandout} provides inequalities bounding $z_{D}$ in terms of $a_{D,D'}$, and \Cref{prop:PolytopeIsConvex} proves that the domain of $z_{D}$ is a convex polytope. These results allow us to conclude that the maximum with respect to $z_D$ is achieved at one of the vertices of the polytope, reducing the complexity of the maximum (at the beginning of \Cref{th:bounddependingonp}). \Cref{prop:formerproblem} and \Cref{th:bounddependingonp} provide an upper bound $\max_{p\in[0,1]}\max\{L_1(p),L_2(p)\}$ for $\max_{D\sim D'}\exp(\varepsilon^{\S}_{D',D})$ (the numerical computation is limited to \Cref{th:bounddependingonp}; see more details in \Cref{remark:computation}).  \Cref{th:boundinverse} provides an upper bound $\max\{L_3,L_4\}$ for $\max_{D\sim D'}\exp(\varepsilon^{\S}_{D,D'})$. Calculating this bound is more straightforward because we can use a looser inequality, since $\exp(\varepsilon^{\S}_{D,D'})$ is usually bounded by $\max_{p\in[0,1]}\max\{L_1(p),L_2(p)\}$. This theorem also requires numerical computation (see \Cref{remark:computationInverse}).

\begin{lemma}\label{lemma:Convex+Affine=Convex}
    Let $f\colon\R^n\to\R$ be a convex function and $g\colon\R^m\to\R^n$ be an affine function, i.e., of the form $g(x_1,\dots,x_m) = A(x_1,\dots,x_m)^\intercal + B$ with $A$ an $n\times m$ matrix in $\R$ and $B\in\R^n$. Then $f\circ g\colon\R^m\to\R$ is convex. 
\end{lemma}
\begin{proof}
    The result follows directly. Let $\lambda\in(0,1)$ and $x,y\in\R^m$. By definition, since $g$ is affine, we have that $g(\lambda x+(1-\lambda)y) = \lambda g(x)+(1-\lambda)g(y)$, and since $f$ is convex, we have that $f(g(\lambda x+(1-\lambda)y)) \leq \lambda f(g(x))+(1-\lambda)f(g(y))$. In conclusion, $f\circ g$ is convex.
\end{proof}

\begin{remark}[Positive semi-definite matrices]\label{remark:positivesemi-definite}
    We recall some properties of positive semi-definite matrices. We say an $n\times n$ matrix $A$ is \textit{positive semi-definite} if for all $T=(t_1,\dots,t_n)\in\R^n$, $T^\intercal A T\geq 0$. 
    
    We recall the following properties of positive semi-definite matrices for our proofs: If $A$ is symmetric, then $A$ is positive semi-definite if and only if its eigenvalues are non-negative. Consequently, any non-negative diagonal matrix is positive semi-definite since its eigenvalues are the elements in the diagonal, and the square $n\times n$ matrix of ones (i.e., such that all entries are $1$) is also positive semi-definite since its eigenvalues are $n$ with multiplicity $1$ and $0$ with multiplicity $n-1$. In addition, the sum of two positive semi-definite matrices is also positive semi-definite, and therefore,
    \[
        A = \begin{pmatrix}
            a_1 & 1 & \cdots & \cdots & 1 \\
            1 & a_2 & \ddots &  & \vdots \\
            \vdots & \ddots & \ddots & \ddots & \vdots \\
            \vdots & & \ddots & a_{n-1} & 1 \\
            1 & \cdots & \cdots & 1 & a_n
        \end{pmatrix}
    \]
    with $a_1,\dots,a_n\geq 1$ is positive semi-definite. 

    Finally, a twice-differentiable multivariate function $f$ is convex if and only if its Hessian $H=(\frac{\partial^2}{\partial x_i\partial x_j} f)_{i,j}$ is positive semi-definite for all values in the domain of $f$. 
\end{remark}

\begin{lemma}\label{lemma:convex-large}
    Let $N\in\N$, $J\in\{0,\dots,N\}$ and $m,M\in(0,1)$ such that $m\leq M$. 
    For all $i\in[N]$, let $a_i\in[m,M]$. Consider the $N$-variate function $f\colon[m,M]^N\to\R^+$
    defined such that
    \[
        f(\mathbf{z})
        = \prod^J_{i=1}\bigg(N+\frac{a_i}{z_i}\bigg)\prod^N_{i=J+1}\bigg(N+\frac{1-a_i}{1-z_i}\bigg)
    \]
    for all $\mathbf{z}=(z_1,\dots,z_N)\in[m,M]^N$. 
    Then, $f$ is convex. 

    Additionally, the function obtained by recursively substituting $z_i$ by an affine function depending on the other free variables is convex.
\end{lemma}
\begin{proof}
    We are going to prove that $f$ is convex. We need to see that the Hessian, $H(\mathbf{z})=(\frac{\partial^2}{\partial z_i\partial z_j} f)_{i,j}(\mathbf{z})$ is positive semi-definite for all $\mathbf{z}\in[m,M]^N$, i.e., for all $T=(t_1,\dots,t_N)\in\R^N$, $T^\intercal H(\mathbf{z}) T\geq 0$ (see \Cref{remark:positivesemi-definite}). 

    To simplify notation, we omit $\mathbf{z}$ for the rest of the proof, and we consider
    \[
        A_i=\begin{cases}
            a_i &\text{if $i\leq J$,} \\
            1-a_i & \text{if $i>J$,}
        \end{cases} 
        \quad \text{and} \quad 
        Z_i=\begin{cases}
            z_i &\text{if $i\leq J$,} \\
            1-z_i & \text{if $i>J$.}
        \end{cases} 
    \]
    
    This allows us to write $f=f(\mathbf{z})=f(z_1,\dots,z_N) = \prod^N_{i=1}(N+\frac{A_i}{Z_i})$. Note too that $\frac{A_i}{Z_i}\geq 0$ and that the derivative of $N+\frac{A_i}{Z_i}$ with respect to $z_i$ is $-\alpha_i\frac{A_i}{(Z_i)^2}$ with $\alpha_i=1$ if $i\leq J$ and $\alpha_i=-1$ if $i>J$ (more precisely, $\alpha_i\coloneqq\frac{\partial Z_i}{\partial z_i}$). The components for the Hessian matrix $H$ are as follows: for all $i\in [N]$, 
    \begin{align*}
        \frac{\partial^2}{(\partial z_i)^2} f &= 2(\alpha_i)^2\frac{A_i}{(Z_i)^3}\prod^N_{\substack{k=1\\k\neq i}}\bigg(N+\frac{A_k}{Z_k}\bigg) \\
        &= 2\bigg(\frac{\alpha_i}{Z_i}\bigg)^2\frac{A_i}{Z_i}\bigg(\frac{1}{N+\frac{A_i}{Z_i}}\bigg)\prod^N_{k=1}\bigg(N+\frac{A_k}{Z_k}\bigg) \\
        &= 2\bigg(\frac{\alpha_i}{Z_i}\frac{\frac{A_i}{Z_i}}{N+\frac{A_i}{Z_i}}\bigg)^2\frac{N+\frac{A_i}{Z_i}}{\frac{A_i}{Z_i}} f
    \end{align*}
    and, for all $i,j\in[N]$, $i\neq j$,
    \begin{align*}
        \frac{\partial^2}{\partial z_i\partial z_j} f &=  \alpha_i\frac{A_i}{(Z_i)^2}\alpha_j\frac{A_j}{(Z_j)^2}\prod^N_{\substack{k=1\\k\neq i,j}}\bigg(N+\frac{A_k}{Z_k}\bigg) \\ 
        &= \frac{\alpha_i}{Z_i}\frac{A_i}{Z_i}\frac{\alpha_j}{Z_j}\frac{A_j}{Z_j}\bigg(\frac{1}{N+\frac{A_i}{Z_i}}\bigg)\bigg(\frac{1}{N+\frac{A_j}{Z_j}}\bigg)\prod^N_{k=1}\bigg(N+\frac{A_k}{Z_k}\bigg) \\
        &= \bigg(\frac{\alpha_i}{Z_i}\frac{\frac{A_i}{Z_i}}{N+\frac{A_i}{Z_i}}\bigg)\bigg(\frac{\alpha_j}{Z_j}\frac{\frac{A_j}{Z_j}}{N+\frac{A_j}{Z_j}}\bigg) f
    \end{align*}

    Taking $B_k=\frac{\alpha_k}{Z_k}\frac{\frac{A_k}{Z_k}}{N+\frac{A_k}{Z_k}}$ and $C_k=\frac{N+\frac{A_k}{Z_k}}{\frac{A_k}{Z_k}}\geq 1$ for all $k\in[N]$, we consider the diagonal matrix $B=\operatorname{diag}(B_1,\dots,B_n)=B^\intercal$ and observe that
    \[
        H = B^\intercal
        \begin{pmatrix}
            2C_1 & 1 & \cdots & \cdots & 1 \\
            1 & 2C_2 & \ddots &  & \vdots \\
            \vdots & \ddots & \ddots & \ddots & \vdots \\
            \vdots & & \ddots & 2C_{N-1} & 1 \\
            1 & \cdots & \cdots & 1 & 2C_N
        \end{pmatrix}
        B f.
    \]

    Let $C$ denote the central matrix, which we know is positive semi-definite by \Cref{remark:positivesemi-definite}. Now we prove that $H$ is positive semi-definite. Let $T=(t_1,\dots,t_n)\in\R^N$. We obtain that
    \[
        T^\intercal H T = T^\intercal (B^\intercal C B f) T = (BT)^\intercal C (BT) f
    \]
    and since $C$ is positive semi-definite, $(BT)^\intercal C (BT)\geq 0$ and thus $T^\intercal H T\geq 0$ since $f\geq 0$. Since this holds for all $\mathbf{z}\in[m,M]^N$, we obtain that $f$ is convex. 
    
    The last part of the statement follows directly from \Cref{lemma:Convex+Affine=Convex}, since the substituted equations are affine functions. \qedhere
\end{proof}

\begin{lemma}\label{lemma:Polytopeconvexinequality}
    Let $N\in\N$ with $N>2$ and $m,M\in(0,1)$ such that $m\leq M$. For all $i\in[N]$, let $a_i\in[m,M]$ and
    \[
        M_i = \frac{1}{N}\min\!\bigg\{m+(N-1)M,(N-2)(a_i-m)+\sum^N_{j=1} a_j\bigg\}.
    \]

    For all non-empty subsets $K\subseteq [N]$, let 
    \[
        B^{(K)} = \frac{(N-2)m+\sum_{l\in[N]\backslash K} M_l}{2N-2-|K|}.
    \]

    Then, $m\leq B^{(K)}\leq M_k$ for all $k\in K$.
\end{lemma}
\begin{proof}
    We will need two inequalities regarding the addition and subtraction of minimums. First, we have that $\min\{a,b\}+\min\{c,d\}\leq\min\{a+c,b+d\}$ for all $a,b,c,d\in\R$:
    \begin{enumerate}[i)]
        \item If $a\leq b$ and $c\leq d$, then $a+c\leq a+c$ is satisfied.
        \item If $a\leq b$ and $c\geq d$, then $a+d\leq a+c$ and $a+d\leq b+d$, and therefore $a+d\leq \min\{a+c,b+d\}$ is satisfied. Analogously, when $a\geq b$ and $c\leq d$, the inequality $b+c\leq \min\{a+c,b+d\}$ satisfies.
        \item If $a\geq b$ and $c\geq d$, then $b+d\leq b+d$ is also satisfied.  
    \end{enumerate}

    Secondly, $\min\{a,b\}-\min\{a,c\}\geq\min\{0,b-c\}$ for all $a,b,c\in\R$:
    \begin{enumerate}[i)]
        \item If $a\leq b,c$, then $a-a=0\geq\min\{0,b-c\}$ is satisfied.
        \item If $b\leq a\leq c$, then $b-c$ is negative and $b-a\geq \min\{0,b-c\}=b-c$ is satisfied since $a\leq c$. 
        \item If $c\leq a\leq b$, then $b-c$ and $a-c$ are positive and $a-c\geq \min\{0,b-c\}=0$ satisfies. 
        \item If $b,c\leq a$, then $b-c\geq \min\{0,b-c\}$ is also satisfied.
    \end{enumerate}
    
    We fix a non-empty subset $K$ of $[N]$ and $k\in K$. We first see that $m\leq B^{(K)}$, which follows directly from $m\leq M_k$:
    \[
        B^{(K)} = \frac{(N-2)m+\sum_{l\in[N]\backslash K} M_l}{2N-2-|K|} \geq \frac{(N-2)m+\sum_{l\in[N]\backslash K} m}{2N-2-|K|} = \frac{2N-2-|K|}{2N-2-|K|}m = m.
    \]
    
    We now prove an additional inequality before seeing $B^{(K)}\leq M_k$. Using the first inequality, we obtain 
    \begin{align*}
        B^{(K)} &= \frac{(N-2)m+\sum_{l\in [N]\backslash K}M_l}{2N-2-|K|} \\
        &= \frac{(N-2)m+\frac{1}{N}\sum_{l\in [N]\backslash K}\min\!\big\{m+(N-1)M,(N-2)(a_l-m)+\sum^N_{j=1} a_j\big\}}{2N-2-|K|} \\
        &\leq \frac{(N-2)m+\frac{1}{N}\min\!\big\{\sum_{l\in [N]\backslash K}(m+(N-1)M),\sum_{l\in [N]\backslash K}\big((N-2)(a_l-m)+\sum^N_{j=1} a_j\big)\big\}}{2N-2-|K|} \\
        &= \frac{(N-2)Nm+\min\!\big\{(N-|K|)(m+(N-1)M),(N-2)\sum_{l\in [N]\backslash K}(a_l-m)+(N-|K|)\sum^N_{j=1} a_j\big\}}{(2N-2-|K|)N}.
    \end{align*}
    
    We call this last term $C^{(K)}$. We now see that $B^{(K)}\leq M_k$, which is equivalent to seeing $\frac{(2N-2-|K|)N}{N-2}(M_k-B^{(K)})\geq0$ since $\frac{(2N-2-|K|)N}{N-2}\geq0$. Therefore, 
    \begin{align*}
        \MoveEqLeft[8] \frac{(2N-2-|K|)N}{N-2}(M_k-B^{(K)}) \\ 
        \geq{}& \frac{(2N-2-|K|)N}{N-2}(M_k-C^{(K)}) \\
        ={}& \frac{(2N-2-|K|)}{N-2}\min\!\bigg\{m+(N-1)M,(N-2)(a_k-m)+\sum^N_{j=1} a_j\bigg\} \\
        {}&- Nm - \frac{1}{N-2}\min\!\bigg\{(N-|K|)(m+(N-1)M),(N-2)\sum_{l\in [N]\backslash K}(a_l-m)+(N-|K|)\sum^N_{j=1} a_j\bigg\} \\
        ={}& \frac{N-2}{N-2}\min\!\bigg\{m+(N-1)M,(N-2)(a_k-m)+\sum^N_{j=1} a_j\bigg\} \\
        {}&+\frac{N-|K|}{N-2}\min\!\bigg\{m+(N-1)M,(N-2)(a_k-m)+\sum^N_{j=1} a_j\bigg\} \\
        {}&- Nm - \frac{1}{N-2}\min\!\bigg\{(N-|K|)(m+(N-1)M),(N-2)\sum_{l\in [N]\backslash K}(a_l-m)+(N-|K|)\sum^N_{j=1} a_j\bigg\} \\
        ={}& \min\!\bigg\{m+(N-1)M,(N-2)(a_k-m)+\sum^N_{j=1} a_j\bigg\} - Nm \\
        {}&+\frac{1}{N-2}\min\!\bigg\{(N-|K|)(m+(N-1)M),(N-2)(N-|K|)(a_k-m)+(N-|K|)\sum^N_{j=1} a_j\bigg\} \\
        {}& - \frac{1}{N-2}\min\!\bigg\{(N-|K|)(m+(N-1)M),(N-2)\sum_{l\in [N]\backslash K}(a_l-m)+(N-|K|)\sum^N_{j=1} a_j\bigg\}. 
    \end{align*}

    Now we apply the inequality $\min\{a,b\}-\min\{a,c\}\geq\min\{0,b-c\}$ we proven before, and obtain that
    \begin{align}
        \frac{(2N-2-|K|)N}{N-2}(M_k-B^{(K)}) \geq{}& \min\!\bigg\{m+(N-1)M,(N-2)(a_k-m)+\sum^N_{j=1} a_j\bigg\} - Nm \nonumber\\
        {}&+\frac{1}{N-2}\min\!\bigg\{0,(N-2)\bigg((N-|K|)(a_k-m)-\sum_{l\in [N]\backslash K}(a_l-m)\bigg)\bigg\} \nonumber\\
        ={}& \min\!\bigg\{(N-1)(M-m),(N-2)a_k-2(N-1)m+\sum^N_{j=1} a_j\bigg\} \nonumber\\
        {}&+\min\!\bigg\{0,\sum_{l\in [N]\backslash K}(a_k-a_l)\bigg\} \nonumber\\
        ={}& \min\!\bigg\{(N-1)(M-m),\sum^N_{\substack{j=1\\j\neq k}} (a_j+a_k-2m)\bigg\}+\min\!\bigg\{0,\sum_{l\in [N]\backslash K}(a_k-a_l)\bigg\}.\nonumber
    \end{align}

    Now, by the properties of the absolute value, we have that 
    \[
        \min\!\bigg\{0,\sum_{l\in [N]\backslash K}(a_k-a_l)\bigg\} \geq -\abs*{\sum_{l\in [N]\backslash K}(a_k-a_l)} \geq -\sum_{l\in [N]\backslash K}\abs{a_k-a_l}
    \]
    and therefore, 
    \begin{align*}
        \MoveEqLeft[6] \frac{(2N-2-|K|)N}{N-2}(M_k-B^{(K)}) 
        \\ 
        &{}= \min\!\bigg\{(N-1)(M-m),\sum^N_{\substack{j=1\\j\neq k}} (a_j+a_k-2m)\bigg\}+\min\!\bigg\{0,\sum_{l\in [N]\backslash K}(a_k-a_l)\bigg\} \\
        &{}\geq \min\!\bigg\{(N-1)(M-m),\sum^N_{\substack{j=1\\j\neq k}} (a_j+a_k-2m)\bigg\}-\sum_{l\in [N]\backslash K}\abs{a_k-a_l} \\
        &{}= \min\!\bigg\{(N-1)(M-m)-\sum_{l\in [N]\backslash K}\abs{a_k-a_l} ,\sum^N_{\substack{j=1\\j\neq k}} (a_j+a_k-2m)-\sum_{l\in [N]\backslash K}\abs{a_k-a_l} \bigg\}\\
        &{}= \min\!\bigg\{(|K|-1)(M-m)+\sum_{l\in [N]\backslash K}(M-m-\abs{a_k-a_l}),
        \sum_{\substack{j\in K\\j\neq k}} (a_j+a_k-2m)+ \sum_{l\in [N]\backslash K}(a_j+a_k-2m-\abs{a_k-a_l})\bigg\}.
    \end{align*}

    Finally, note that $|a_k-a_l|\leq M-m$, and that 
    \[
        \sum_{l\in [N]\backslash K}(a_j+a_k-2m-\abs{a_k-a_l}) = \sum_{l\in [N]\backslash K}((a_j-m)+(a_k-m)-\abs{(a_k-m)-(a_l-m)})\leq0
    \]
    since $a+b-\abs{a-b}\geq0$ for all $a,b\geq0$. Therefore, both terms of the minimum are positive. Thus, we obtain that $B^{(K)}\leq M_k$ for all $k\in K$. \qedhere
\end{proof}

\begin{lemma}\label{lemma:boundsfordandout}
    Let $N\in\N$ and $m,M\in(0,1)$ such that $m\leq M$. Let $\tdist$ be an $(m,M)$-transformation of a distance over $\X$. Then, for all $x_1,\dots,x_N,y\in\X$ and $l\in[N]$, we have 
    \[
        \sum^N_{i=1}\out_{D}(x_i) \leq (2N-2)\out_D(x_l) - (N-2)m
    \]
    and 
    \[
        m\leq \out_D(x_l) \leq \frac{1}{N}\min\!\bigg\{m+(N-1)M,(N-2)(\tdist(x_l,y)-m)+\sum^N_{j=1} \tdist(x_j,y)\bigg\}   
    \]
    for $D=\{x_1,\dots,x_N\}$.
\end{lemma}
\begin{proof}
    We consider the first inequality. We see it is true for $N=1$ and $N=2$ separately. 
    For $N=1$, we have that $\out_D(x_1)=m$, and for $N=2$, we have that
    \[
        \out_D(x_1) =\frac{m+\tdist(x_1,x_2)}{2} = \out_D(x_2).
    \]
    
    Thus, the inequalities follow directly (in these cases, equality is satisfied).
    
    Suppose now that $N>2$ and we fix $l\in[N]$. Then, we obtain
    \begin{align*}
        \sum^N_{\substack{i=1\\i\neq l}} \out_D(x_i)
        &= \sum^N_{\substack{i=1\\i\neq l}} \frac{1}{N}\sum^N_{\substack{j=1}} \tdist(x_j,x_i) \\
        &= \sum^N_{\substack{i=1\\i\neq l}} \frac{1}{N}\bigg(m+\tdist(x_l,x_i)+\sum^N_{\substack{j=1\\j\neq i,l}} \tdist(x_j,x_i)\bigg) \\
        &= \frac{N-1}{N}m+\frac{1}{N}\sum^N_{\substack{i=1\\i\neq l}} \tdist(x_l,x_i)+\frac{1}{N}\sum^N_{\substack{i=1\\i\neq l}}\sum^N_{\substack{j=1\\j\neq i,l}} \tdist(x_j,x_i) \\
        &= \frac{N-1}{N}m+\frac{1}{N}\sum^N_{\substack{i=1}} \tdist(x_l,x_i)-\frac{1}{N}m+\frac{1}{N}\sum^N_{\substack{i=1\\i\neq l}}\sum^N_{\substack{j=1\\j\neq i,l}} \tdist(x_j,x_i) \\
        &= \out_D(x_l)+\frac{1}{N}\sum^N_{\substack{i=1\\i\neq l}}\sum^N_{\substack{j=1\\j\neq i,l}} \tdist(x_j,x_i)+\frac{N-2}{N}m.
    \end{align*}

    Now, applying the transformed triangular inequality (\Cref{prop:(mM)-transformeddistance}) to each term of the sum we obtain that
    \begin{align*}
        \sum^N_{\substack{i=1\\i\neq l}} \out_D(x_i)
        &\leq \out_D(x_l)+\frac{1}{N}\sum^N_{\substack{i=1\\i\neq l}}\sum^N_{\substack{j=1\\j\neq i,l}} (\tdist(x_j,x_l)+\tdist(x_l,x_i)-m) +\frac{N-2}{N}m \\
        &= \out_D(x_l)+\frac{1}{N}\sum^N_{\substack{i=1\\i\neq l}}\bigg(\sum^N_{\substack{j=1\\j\neq i,l}}\tdist(x_j,x_l)+(N-2)\tdist(x_l,x_i)-(N-2)m\bigg) +\frac{N-2}{N}m\\
        &= \out_D(x_l)+\frac{1}{N}\sum^N_{\substack{i=1\\i\neq l}}\bigg(\sum^N_{\substack{j=1}}\tdist(x_j,x_l)+(N-3)\tdist(x_l,x_i)-(N-1)m\bigg) +\frac{N-2}{N}m\\
        &= \out_D(x_l)+\frac{N-1}{N}\sum^N_{\substack{j=1}}\tdist(x_j,x_l)+\frac{N-3}{N}\bigg(\sum^N_{\substack{i=1}}\tdist(x_l,x_i)-m\bigg)-\frac{(N-1)^2}{N}m+\frac{N-2}{N}m \\
        &= \out_D(x_l)+(N-1)\out_D(x_l)+(N-3)\bigg(\out_D(x_l)-\frac{m}{N}\bigg)-\frac{(N-1)^2}{N}m+\frac{N-2}{N}m \\
        &= (2N-3)\out_D(x_l)-\frac{(N-3)+(N-1)^2-(N-2)}{N}m \\
        &= (2N-3)\out_D(x_l)-(N-2)m.
    \end{align*}

    Finally, adding $\out_D(x_l)$ at both sides of the inequality gives us the first inequality:
    \[
        \sum^N_{i=1}\out_{D}(x_i) \leq (2N-2)\out_D(x_l) - (N-2)m.
    \]

    Now we prove the second inequality. By the transformed triangular inequality (\Cref{prop:(mM)-transformeddistance}) and the fact that $m\leq \tdist(x,x')\leq M$ for all $x,x'\in\X$, we obtain that $m\leq\out_D(x_l)$ and
    \begin{align*}
        \out_D(x_l) &= \frac{1}{N}\bigg(m+\sum^N_{\substack{j=1\\j\neq l}} \tdist(x_j,x_i)\bigg) \\
        &\leq \frac{1}{N}\bigg(m+\min\!\bigg\{\sum^N_{\substack{j=1\\j\neq l}}M,\sum^N_{\substack{j=1\\j\neq l}} (\tdist(x_j,y)+\tdist(y,x_l)-m)\bigg\}\bigg) \\
        &= \frac{1}{N}\bigg(m+\min\!\bigg\{(N-1)M,(N-1)(\tdist(y,x_l)-m)+\sum^N_{\substack{j=1\\j\neq l}} \tdist(x_j,y)\bigg\}\bigg) \\
        &= \frac{1}{N}\min\!\bigg\{m+(N-1)M,(N-2)(\tdist(y,x_l)-m)+\sum^N_{\substack{j=1}} \tdist(x_j,y)\bigg\}. \qedhere
    \end{align*}
\end{proof}

\begin{proposition}\label{prop:PolytopeIsConvex}
    Let $N\in\N$ and $m,M\in(0,1)$ such that $m<M$. For all $i\in[N]$, let $a_i\in[m,M]$ and
    \[
        M_i = \frac{1}{N}\min\!\bigg\{m+(N-1)M,(N-2)(a_i-m)+\sum^N_{j=1} a_j\bigg\}.
    \]
    
    Let $P$ be the intersection of the closed $r$-dimensional hypercube $C=[m,M_1]\times\cdots\times[m,M_N]$ and the closed half-spaces $H_i\coloneqq\{z\in\R^N \mid \sum^N_{j=1} z_j \leq (2N-2)z_i - (N-2)m\}$ for all $i\in[N]$.

    Then, $P$ is a closed convex polytope in $\R^N$. For $N>2$, we have that $P$ has at most $2^N$ vertices, which are 
    \[
        z^{(K)}=(z^{(K)}_1,\dots,z^{(K)}_N) \text{ such that } z^{(K)}_k = \begin{cases}
            B^{(K)} & \text{if $k\in K$,} \\
            M_k & \text{if $k\in [N]\backslash K$,}
        \end{cases}
    \]
    for every subset $K\subseteq[N]$, where
    \[
        B^{(K)} = \frac{(N-2)m+\sum_{l\in[N]\backslash K} M_l}{2N-2-|K|}.
    \]

    For the cases $N=1$ and $N=2$, we obtain that $P$ is respectively $\{m\}$ and $\{(t,t)\in\R^2\mid m\leq t\leq M_1\}$. 
\end{proposition}
\begin{proof}
    First, observe that $D$ is non-empty since it contains $(m,m,\dots,m)\in\R^N$. 
    Since $P$ is the intersection of $N+1$ closed convex sets, it is closed and convex. Also, since $C=[m,M_1]\times\cdots\times[m,M_N]$ is a polytope and the boundary of each half-space is a hyperplane, every side of $D$ is flat, and $P$ is a polytope.
    We see the cases $N=1$ and $N=2$ separately. When $N=1$, we have that $M_1=m$, $C=\{m\}$ and $H_1=\{z\in\R\mid z\leq m\}$ and thus $P=\{m\}$.
    When $N=2$, then $M_1=M_2$, $C=[m,M_1]^2$ and $H_1\cap H_2=\{z\in\R^2\mid z_1=z_2\}$, so $P=\{(t,t)\in\R^2\mid m\leq t\leq M_1\}$ (its vertices are $(m,m)$ and $(M_1,M_1)$). We consider $N>2$ for the rest of the proof.

    We now study the vertices of $D$. We denote every hyperplane that defines the half-space $H_i$ as $\partial H_i\coloneqq\{z\in\R^N \mid \sum^N_{j=1} z_j \leq (2N-2)z_i - (N-2)m\}$ for $i\in[N]$. We also denote the ``lower'' and ``upper'' hyperplanes that define the hypercube as $L_i\coloneqq\{z\in C\mid z_i=m\}$ and $U_i\coloneqq\{z\in C \mid z_i=M_i\}$ for  $i\in[N]$. By construction, all the faces of the polytope $P$ are contained in one or more of the elements of $\H=\{\partial H_i\}_{i\in[N]}\cup\{L_i\}_{i\in[N]}\cup\{U_i\}_{i\in[N]}$. 
    
    We know that the vertices of $P$ are the points of $P$ among 
    \begin{enumerate}[(i)]
        \item the vertices of $[m,M_1]\times\cdots\times[m,M_N]$,
          \item the intersection vertices between the hypercube facets and the hyperplanes,
        \item the intersection vertices between a subset of hyperplanes. 
    \end{enumerate}

    Formally, by the definition of vertex, a point $\mathbf{p}\in\R^r$ is a vertex of $P$ if $\mathbf{p}\in P$ and there exists $\H'\subseteq\H$ such that $\{\mathbf{p}\} = \bigcap_{H\in\H'} H$.

    We can see quickly that the only point verifying (iii) is $(m,m,\dots,m)\in\R^N$. That is, since every hyperplane $\partial H_i$ contains this point, the intersection of any combination of these hyperplanes must contain it.

    Regarding the candidates in (ii), we can also see that the intersection of any facet $L_i$ with any hyperplane is $(m,m,\dots,m)$. The intersection of the hyperplanes with $U_i$ is more complex.

    We first note that $\partial H_i$ and $U_i$ for any $i\in[N]$ do not intersect in $P$. Indeed, if $z\in \partial H_i\cap U_i$, then $z_i=M_i$ and 
    \[
        (2N-2)M_i - (N-2)m = M_i+\sum^N_{\substack{j=1\\j\neq i}} z_j.
    \]

    Consider two cases:
    \begin{enumerate}[(a)]
        \item Case $M_i=\frac{m+(N-1)M}{N}$ (i.e., the first term of the minimum of the expression of $M_i$ is achieved), including when $a_i=m$. Then we apply $z_j \leq M_j \leq \frac{m+(N-1)M}{N}=M_i$ and obtain a contradiction:
        \[
            (2N-2)M_i - (N-2)m = M_i + \sum^N_{\substack{j=1\\j\neq i}} z_j \leq NM_i \Longleftrightarrow (N-2)(M_i-m)\leq 0.
        \]
        \item Case $M_i=\frac{1}{N}((N-2)(a_i-m)+\sum^N_{k=1} a_k)$ (i.e., the second term of the minimum is achieved), excluding when $a_i=m$. Then, by applying $z_j\leq M_j\leq \frac{1}{N}((N-2)(a_j-m)+\sum^N_{k=1} a_k)$, we obtain that
        \[
            \sum^N_{\substack{j=1}} z_j \leq \frac{N-2}{N}\sum^N_{\substack{j=1}}(a_j-m)+\frac{1}{N}\sum^N_{\substack{j=1}}\sum^N_{k=1} a_k = \frac{2N-2}{N}\sum^N_{k=1}a_k - (N-2)m,
        \]
        which leads to a contradiction:
        \begin{gather*}
            \frac{2N-2}{N}\bigg((N-2)(a_i-m)+\sum^N_{k=1} a_k\bigg) - (N-2)m \leq \frac{2N-2}{N}\sum^N_{k=1}a_i - (N-2)m 
            \Longleftrightarrow \frac{2(N-1)(N-2)}{N}(a_i-m)\leq 0. 
        \end{gather*}
    \end{enumerate}

    Consequently, $\bigcap_{k\in K} \partial H_k \cap \bigcap_{l\in L} U_l$ is non-empty if and only if $K$ and $L$ are disjoint subsets of $[N]$. Since it is necessary to have at least $N$ hyperplanes to ensure that the intersection is a single point, we deduce that $\bigcap_{k\in K} \partial H_k \cap \bigcap_{l\in L} U_l$ is a point if and only if $K$ and $L$ form a partition of $[N]$ (in particular, $|K|+|L|=N$). 

    We consider a fixed partition $\{K,([N]\backslash K)\}$ of $[N]$ and the intersection $\bigcap_{k\in K} \partial H_k \cap \bigcap_{l\in [N]\backslash K} U_l$. We obtain that $x_l=M_l$ for all $l\in [N]\backslash K$ and $x_k=x_{k'}$ for all $k,k'\in K$ (since $(2N-2)z_{k}-(N-2)m=\sum^N_{j=1}z_j=(2N-2)z_{k'}-(N-2)m$). Substituting these values in the equation defining $\partial H_k$ for any $k\in K$ when $K\neq\varnothing$, we obtain
    \[
         |K|z_{k} + \sum_{l\in [N]\backslash K}M_l = (2N-2)z_{k}-(N-2)m \Longleftrightarrow z_k = \frac{(N-2)m+\sum_{l\in [N]\backslash K}M_l}{2N-2-|K|}.
    \]

    This value corresponds to $B^{(K)}$. Thus, the point we obtain is 
    \[
        z^{(K)}=(z^{(K)}_1,\dots,z^{(K)}_N) \text{ such that } z^{(K)}_k = \begin{cases}
            B^{(K)} & \text{if $k\in K$,} \\
            M_k & \text{if $k\in [N]\backslash K$.}
        \end{cases}
    \]

    In addition, by \Cref{lemma:Polytopeconvexinequality}, we that $m\leq z^{(K)}_k\leq M_k$ for all $k\in[N]$. We can also easily see that $z^{(K)}\in H_l$ for all $l\in[N]$. In conclusion, $z^{(K)}\in P$ for all $K\subseteq[N]$ and the vertices of $P$ are $\{z^{(K)}\mid K\subseteq[N]\}$ (note that $(m,\dots,m)=z^{([N])}$ and $(M_1,\dots,M_N)=z^{(\varnothing)}$), and therefore $P$ has at most $2^N$ vertices (we obtain less only when $z^{(K)}=z^{(K')}$ for some $K\neq K'$).
\end{proof}

\begin{proposition}\label{prop:formerproblem}
    Let $\varepsilon\geq0$. Let $N>2$ and $J$ be integers such that $0\leq J<N$. Let $m,M\in(0,1)$ such that $m\leq M$ and let $c\in[m,M]$. 
    Let $f\colon[m,M]^{N-J}\to\R$ such that, for all $(a_1,\dots,a_{N-J})\in[m,M]^{N-J}$,
    \[
        f(a_1,\dots,a_{N-J}) = 
        \bigg(\e^\varepsilon-(\e^\varepsilon-1)\frac{m+Jc+S}{N+1}\bigg)\prod^{N-J}_{i=1}\bigg(N+\frac{1-a_i}{1-\frac{\min\{m+(N-1)M,(N-2)(a_i-m)+Jc+S\}}{N}}\bigg)
    \]
    where $S=\sum^{N-J}_{i=1}a_i$. Then,
    \[
        f(a_1,\dots,a_{N-J}) \leq 
        \max\!\bigg\{\max_{t\in[m,U_{N-J-1}]} f_{\mathrm{diag}}(t),\max_{k\in\{0,1,\dots,N-J\}}\max_{t\in[m,U_{k}]} f_k(t)\bigg\} 
    \]
    where
    \[
        f_{\mathrm{diag}}(t) = \bigg(\e^\varepsilon-(\e^\varepsilon-1)\frac{m+Jc+(N-J)t}{N+1}\bigg)\bigg(N+\frac{1-t}{1-\frac{(2N-J-2)t-(N-2)m+Jc}{N}}\bigg)^{N-J}    
    \]
    for all $t\in[m,U_{N-J-1}]$; and for all $k\in\{0,1,\dots,N-J\}$,
    \[
        f_k(t) = \bigg(\e^\varepsilon-(\e^\varepsilon-1)\frac{m+S_k(t)}{N+1}\bigg)\bigg(N+\frac{1-t}{1-\frac{(N-2)(t-m)+S_k(t)}{N}}\bigg)
        \bigg(N+\frac{1-m}{1-\frac{S_k(t)}{N}}\bigg)^{N-J-k-1}\bigg(N+\frac{1-B_k(t)}{1-\frac{m+(N-1)M}{N}}\bigg)^{k}
    \] 
    with $S_k(t)=Jc+t+(N-J-k-1)m+kB_k(t)$ and
    \[
        B_k(t) = \frac{(N-1)(M+m)-(Jc+(N-J-k-1)m+t)}{N+k-2}
    \]
    for all $t\in[m,U_k]$; and
    \[
        U_k = \frac{(N-1)(M+m)-(Jc+(N-J-k-1)m)}{N+k-1}.
    \]
\end{proposition}
\begin{proof}
    We search for the maximum of $f$ over its domain $[m,M]^{N-J}$, denoting its elements as $a\coloneqq(a_1,\dots,a_{N-J})$. We consider the subset of $Q\subseteq[m,M]
    ^{N-J}$ such that the inequalities
    \[
        (N-2)a_i-(N-1)m+Jc+S\leq (N-1)M \Longleftrightarrow a_i \leq \frac{(N-1)(M+m)-Jc-S}{N-2}
    \]
    hold for all $i\in[N-J]$ (i.e., the minimum of $\min\{m+(N-1)M,(N-2)(a_i-m)+Jc+S\}$ is achieved in the right term for all $i\in[N-J]$) and we will first see that 
    \[
        \max_{a\in [m,M]^{N-J}} f(a) = \max_{a\in Q} f(a),
    \]
    i.e., that for every $a\in[m,M]^{N-J}\backslash Q$ there exists $a'\in Q$ such that $f(a)\leq f(a')$. 

    We consider the following change of variables: $b=\phi(a)$ defined as
    \[
        \begin{pmatrix}
            b_{1} \\
            b_{2} \\
            \vdots \\
            b_{N-J}
        \end{pmatrix}=
        \phi\!\begin{pmatrix}
            a_{1} \\
            a_{2} \\
            \vdots \\
            a_{N-J}
        \end{pmatrix} 
        =
        \begin{pmatrix}
            1 & \frac{1}{N-1} & \dots  & \frac{1}{N-1} \\
            \frac{1}{N-1} & 1 & \ddots & \vdots \\
            \vdots & \ddots & \ddots & \frac{1}{N-1} \\
            \frac{1}{N-1} &\dots & \frac{1}{N-1} & 1 
        \end{pmatrix}\begin{pmatrix}
            a_{1} \\
            a_{2} \\
            \vdots \\
            a_{N-J}
        \end{pmatrix}
         + 
        \begin{pmatrix}
            - m + \frac{Jc}{N-1} \\
            - m + \frac{Jc}{N-1} \\
            \vdots \\
            - m + \frac{Jc}{N-1}
        \end{pmatrix},
    \]
    whose inverse transformation is
    \[
        \begin{pmatrix}
            a_{1} \\
            a_{2} \\
            \vdots \\
            a_{N-J}
        \end{pmatrix} =
        \frac{N-1}{(N-2)(2N-J-2)}\begin{pmatrix}
            2N-J-3 & -1 & \dots  & -1 \\
            -1 & 2N-J-3 & \ddots & \vdots \\
            \vdots & \ddots & \ddots & -1 \\
            -1 &\dots & -1 & 2N-J-3
        \end{pmatrix}
        \left(
        \begin{pmatrix}
            b_{1} \\
            b_{2} \\
            \vdots \\
            b_{N-J}
        \end{pmatrix} + 
        \begin{pmatrix}
            m - \frac{Jc}{N-1} \\
            m - \frac{Jc}{N-1} \\
            \vdots \\
            m - \frac{Jc}{N-1}
        \end{pmatrix}\right),
    \]

    Under this transformation we have that $Q=\{a\in[m,M]^{N-J}\mid b=\phi(a)\leq M\}$. Let $g(b) \coloneqq f(\phi^{-1}(b))$ for all $b\in\phi([m,M]^{N-J})$, i.e.,  
    \begin{align*}
        g(b) &= \bigg(\e^{\varepsilon}-(\e^{\varepsilon}-1)\frac{m+Jc\frac{N-2}{2N-J-2}+\frac{N-1}{2N-J-2}\sum^{N-J}_{i=1}b_i+\frac{(N-1)(N-J)}{2N-J-2}m}{N+1}\bigg)\\
        &\qquad\qquad\cdot \prod^{N-J}_{i=1}\bigg(N+\frac{1-\frac{N-1}{(N-2)(2N-J-2)}((2N-J-3)b_i-\sum^{N-J}_{j=1,\,j\neq i}b_j+(N-2)(m-\frac{Jc}{N-1}))}{1-\frac{m+(N-1)\min\{M,b_i\}}{N}}\bigg) \\
        &=\bigg(\e^{\varepsilon}-(\e^{\varepsilon}-1)\frac{m+Jc\frac{N-2}{2N-J-2}+\frac{N-1}{2N-J-2}\sum^{N-J}_{i=1}b_i+\frac{(N-1)(N-J)}{2N-J-2}m}{N+1}\bigg)\\
        &\qquad\qquad \cdot \prod^{N-J}_{i=1}\bigg(N+\frac{1-\big(\frac{N-1}{N-2}b_i-\frac{N-1}{(N-2)(2N-J-2)}\sum^{N-J}_{j=1}b_j+\frac{N-1}{2N-J-2}(m-\frac{Jc}{N-1})\big)}{1-\frac{m+(N-1)\min\{M,b_i\}}{N}}\bigg).
    \end{align*}

    Now, we assume that $a\in[m,M]^{N-J}\backslash Q$, and thus there is at least one index $i\in[N-J]$ such that $b_i>M$. We assume without loss of generality that $b_1,b_2,\dots,b_k\geq M$ and $b_{k+1},\dots,b_{N-J}<M$. We consider the function 
    \begin{align*}
        \psi\colon [M,\max\{b_1,\dots,b_k\}] & \longrightarrow \R \\*
        x & \longmapsto \ln\!\bigg(N+\frac{1-\big(\frac{N-1}{N-2}x-\frac{N-1}{(N-2)(2N-J-2)}\sum^{N-J}_{j=1}b_j+\frac{N-1}{2N-J-2}(m-\frac{Jc}{N-1})\big)}{1-\frac{m+(N-1)M}{N}}\bigg),
    \end{align*}
    which is a concave function (its second derivative verifies $\psi''(x)=\frac{(N-1)^2}{(N-2)^2(1-\frac{m+(N-1)M}{N})^2\psi(x)^2}$ for all $x$ in its domain). Therefore, by Jensen's inequality,
    \[
        \frac{1}{k}\sum^k_{i=1} \psi(b_i) \leq \psi\bigg(\frac{1}{k}\sum^k_{i=1} b_i\bigg)
    \]
    since $b_1,\dots,b_k\in[M,\max\{b_1,\dots,b_k\}]$, and applying the exponential and the $k$th power at both sides of the inequality, we obtain
    \begin{multline*}
        \prod^{k}_{i=1}\bigg(N+\frac{1-\big(\frac{N-1}{N-2}b_i-\frac{N-1}{(N-2)(2N-J-2)}\sum^{N-J}_{j=1}b_j+\frac{N-1}{2N-J-2}(m-\frac{Jc}{N-1})\big)}{1-\frac{m+(N-1)M}{N}}\bigg) \\
        \leq \bigg(N+\frac{1-\big(\frac{N-1}{N-2}\frac{1}{k}\sum^k_{i=1} b_i-\frac{N-1}{(N-2)(2N-J-2)}\sum^{N-J}_{j=1}b_j+\frac{N-1}{2N-J-2}(m-\frac{Jc}{N-1})\big)}{1-\frac{m+(N-1)M}{N}}\bigg)^k
    \end{multline*}
    and it follows directly that
    \[
        g(b_1,\dots,b_k,b_{k+1},\dots,b_{N-J}) \leq g\bigg(\frac{1}{k}\sum^k_{i=1}b_i,\dots,\frac{1}{k}\sum^k_{i=1}b_i,b_{k+1},\dots,b_{N-J}\bigg).
    \]

    Therefore, it is enough to find the maximum for 
    \begin{align*}
        g_2(b')={}&g_2(v,b_{k+1},b_{k+2},\dots,b_{N-J}) \\
        \coloneqq{}& g(v,\dots,v,b_{k+1},b_{k+2},\dots,b_{N-J}) \\
        ={}& \bigg(\e^{\varepsilon}-(\e^{\varepsilon}-1)\frac{m+Jc\frac{N-2}{2N-J-2}+\frac{N-1}{2N-J-2}(kv+\sum^{N-J}_{i=k+1}b_i)+\frac{(N-1)(N-J)}{2N-J-2}m}{N+1}\bigg)\\
        &\qquad\cdot \bigg(N+\frac{1-\frac{N-1}{(N-2)(2N-J-2)}((2N-J-k-2)v-\sum^{N-J}_{j=k+1}b_j+(N-2)(m-\frac{Jc}{N-1}))}{1-\frac{m+(N-1)M}{N}}\bigg)^k \\
        &\qquad\cdot \prod^{N-J}_{i=k+1}\bigg(N+\frac{1-\frac{N-1}{(N-2)(2N-J-2)}((2N-J-3)b_i-(kv+\sum^{N-J}_{j=k+1,\,j\neq i}b_j)+(N-2)(m-\frac{Jc}{N-1}))}{1-\frac{m+(N-1)b_i}{N}}\bigg)
    \end{align*}
    
    Now, we consider $h$ defined as the second and third term of $g_2$, i.e.,
    \begin{align*}
        h(b')&=h(v,b_{k+1},b_{k+2},\dots,b_{N-J}) \\
        &= \bigg(N+\frac{1-\frac{N-1}{(N-2)(2N-J-2)}((2N-J-k-2)v-\sum^{N-J}_{j=k+1}b_j+(N-2)(m-\frac{Jc}{N-1}))}{1-\frac{m+(N-1)M}{N}}\bigg)^k \\
        &\qquad\cdot \prod^{N-J}_{i=k+1}\bigg(N+\frac{1-\frac{N-1}{(N-2)(2N-J-2)}((2N-J-3)b_i-(kv+\sum^{N-J}_{j=k+1,\,j\neq i}b_j)+(N-2)(m-\frac{Jc}{N-1}))}{1-\frac{m+(N-1)b_i}{N}}\bigg)
    \end{align*}
    for $b'\in[m,M]^{N-J-k+1}$, and we compute its derivative with respect to $v$:
    \begin{align*}
        \frac{\partial}{\partial v}h(b') ={}& k\frac{\frac{-\frac{N-1}{(N-2)(2N-J-2)}(2N-J-k-2)}{1-\frac{m+(N-1)M}{N}}}{N+\frac{1-\frac{N-1}{(N-2)(2N-J-2)}((2N-J-k-2)v-\sum^{N-J}_{j=k+1}b_j+(N-2)(m-\frac{Jc}{N-1}))}{1-\frac{m+(N-1)M}{N}}}h(b') \\
        {}& + \sum^{N-J}_{i=k+1} \frac{\frac{\frac{N-1}{(N-2)(2N-J-2)}k}{1-\frac{m+(N-1)b_i}{N}}}{N+\frac{1-\frac{N-1}{(N-2)(2N-J-2)}((2N-J-3)b_i-(kv+\sum^{N-J}_{j=k+1,\,j\neq i}b_j)+(N-2)(m-\frac{Jc}{N-1}))}{1-\frac{m+(N-1)b_i}{N}}}h(b') \\
        ={}& \frac{-k\frac{N-1}{(N-2)(2N-J-2)}(2N-J-k-2)h(b')}{(1-\frac{m+(N-1)M}{N})N+1-\frac{N-1}{(N-2)(2N-J-2)}((2N-J-k-2)v-\sum^{N-J}_{j=k+1}b_j+(N-2)(m-\frac{Jc}{N-1}))}\\
        {}& + \sum^{N-J}_{i=k+1} \frac{k\frac{N-1}{(N-2)(2N-J-2)}h(b')}{(1-\frac{m+(N-1)b_i}{N})N+1-\frac{N-1}{(N-2)(2N-J-2)}((2N-J-3)b_i-(kv+\sum^{N-J}_{j=k+1,\,j\neq i}b_j)+(N-2)(m-\frac{Jc}{N-1}))}.
    \end{align*}

    Let
    \[
        A(v,b_{k+1},\dots,b_{N-J}) = N-(m+(N-1)M)+1\\
        -\frac{N-1}{(N-2)(2N-J-2)}\bigg((2N-J-k-2)v-\sum^{N-J}_{j=k+1}b_j+(N-2)\bigg(m-\frac{Jc}{N-1}\bigg)\bigg)
    \]
    and, for $i=\{k+1,k+2,\dots,N-J\}$,
    \[
        B_i(v,b_{k+1},\dots,b_{N-J}) = N-(m+(N-1)M)+1\\
        -\frac{N-1}{(N-2)(2N-J-2)}\bigg((2N-J-3)b_i-kv-\sum^{N-J}_{\substack{j=k+1\\j\neq i}}b_j+(N-2)\bigg(m-\frac{Jc}{N-1}\bigg)\bigg). 
    \]
    
    Note that $A$ and $B_i$ are positive and $A(v,b_{k+1},\dots,b_{N-J})\leq B_i(v,b_{k+1},\dots,b_{N-J})$ since $b_{k+1},\dots,b_{N-J}\leq M < v$.
    Consequently, 
    \begin{align*}
        \frac{\partial}{\partial v}h(b') &= -\frac{k(N-1)h(b')}{(N-2)(2N-J-2)}\bigg(\frac{2N-J-k-2}{A(b')} - \sum_{i=k+1}^{N-J}\frac{1}{B_i(b')}\bigg) \\
        &\leq -\frac{k(N-1)h(b')}{(N-2)(2N-J-2)}\bigg(\frac{2N-J-k-2}{A(b')} - \sum_{i=k+1}^{N-J}\frac{1}{A(b')}\bigg) \\
        &= -\frac{k(N-1)h(b')}{(N-2)(2N-J-2)A(b')}(2N-J-k-2-(N-J-k)) \\
        &= -\frac{k(N-1)h(b')}{(2N-J-2)A(b')} \leq 0,
    \end{align*}
    and for all $v>M$. Thus, since $h$ is decreasing with respect to $v$ for $v>M$, i.e., 
    \[
        h(v,b_{k+1},\dots,b_{N-J}) \leq  h(M,b_{k+1},\dots,b_{N-J})
    \]
    and, since $\e^{\varepsilon}-(\e^{\varepsilon}-1)\frac{m+Jc\frac{N-2}{2N-J-2}+\frac{N-1}{2N-J-2}(kv+\sum^{N-J}_{i=k+1}b_i)+\frac{(N-1)(N-J)}{2N-J-2}m}{N+1}$ is also decreasing with respect to $v$ for $v\geq M$,
    \[
        g_2(v,b_{k+1},\dots,b_{N-J}) \leq g_2(M,b_{k+1},\dots,b_{N-J})
    \]
    and 
    \begin{align*}
        \MoveEqLeft[4] f(a_1,\dots,a_{N-J}) = f(\phi^{-1}(b_1,\dots,b_{N-J})) \leq f(\phi^{-1}(M,\dots,M,b_{k+1},\dots,b_{N-J}))
    \end{align*}
    with $\phi^{-1}(M,\dots,M,b_{k+1},\dots,b_{N-J})\in Q$ (trivially, the first coordinate of $\phi(\phi^{-1}(M,\dots,M,b_{k+1},\dots,b_{N-J}))$ is smaller or equal than $M$).

    In summary, we saw that 
    \[
        \max_{a\in [m,M]^{N-J}} f(a) = \max_{a\in Q} f(a).
    \]

    Now we study its maximum over $Q$. Note that $Q$ is a closed convex polytope that is symmetric with respect to $a_1,\dots,a_{N-J}$. The edges of $Q$ are
    \begin{align*}
        (m,\overset{(N-J-1)}{\cdots},m,t) &\qquad\text{for $t\in(m,U_0)$}, \\
        (m,\overset{(N-J-2)}{\cdots},m,B_1(t),t) &\qquad\text{for $t\in(m,U_1)$}, \\
        (m,\overset{(N-J-3)}{\cdots},m,B_2(t),B_2(t),t) &\qquad\text{for $t\in(m,U_2)$}, \\
        \vdots \\
        (m,B_{N-J-2}(t),\overset{(N-J-2)}{\cdots},B_{N-J-2}(t),t) &\qquad\text{for $t\in(m,U_{N-J-2})$}, \\
        (B_{N-J-1}(t),\overset{(N-J-1)}{\cdots},B_{N-J-1}(t),t) &\qquad\text{for $t\in(m,U_{N-J-1})$},
    \end{align*}
    as well as all edges obtained by permuting the coordinates (the total number of edges is $(N-J)2^{N-J-1}$), where 
    \[
        B_k(t) = \frac{(N-1)(M+m)-(Jc+(N-J-k-1)m+t)}{N+k-2}
    \]
    and
    \[
        U_k = \frac{(N-1)(M+m)-(Jc+(N-J-k-1)m)}{N+k-1}.
    \]

    Observe that, for all $a\in Q$, we have that
    \[
        f(a) = \bigg(\e^\varepsilon-(\e^\varepsilon-1)\frac{m+Jc+S}{N+1}\bigg)\prod^{N-J}_{i=1}\bigg(N+\frac{1-a_i}{1-\frac{(N-2)(a_i-m)+Jc+S}{N}}\bigg). 
    \]

    Now, we will see that the maximum over $Q$ is achieved at one of the edges of $Q$ or at its diagonal, i.e., when $a_1=\cdots=a_{N-J}$. To simplify the calculation, we will restrict the function to the hyperplanes such that $Jc+a_1+\cdots+a_{N-J}$ is constant. We consider now $I_s\coloneqq[m,\min\{M,\frac{(N-1)(M+m)-s}{N-2}\}]$ and the function 
    \begin{align*}
        \varphi_s\colon I_s & \longrightarrow \R
         \\
        x & \longmapsto \ln\!\bigg(N+\frac{1-x}{1-\frac{(N-2)(x-m)+s}{N}}\bigg)
    \end{align*}
    for a fixed $s\in[Nm,NM]$. We study now the convexity and concavity with respect to the constant $s$. The second derivative of $\varphi_s$ verifies 
    \[
        \varphi''_s(x) = -\frac{(s - (N - 2) m - 2) (2 (N - 2) (N - 1)(x-1) + (2N - 3)(s - (N - 2)m - 2))}{((N - 2)(x-m) + s - N)^{2} ((N - 1) x + s - (N - 2) m - N - 1)^{2}}
    \]
    for all $x\in I_s$, and we study its sign over $I_s$. We first see that the following inequality holds for all $x\in I_s$:
    \begin{gather*}
        2 (N - 2) (N - 1)(x-1) + (2N - 3)(s - (N - 2)m - 2)
        < 0 \\
        \Longleftrightarrow x < 1 - \frac{(2N-3)(s-(N-2)m-2)}{2(N-2)(N-1)}. 
    \end{gather*}

    Since $x\leq M<1$, the inequality clearly holds if $s - (N - 2) m - 2\leq 0$. On the other hand, if $s - (N - 2) m - 2\geq 0$, then 
    \begin{align*}
        x &\leq \frac{(N-1)(M+m)-s}{N-2} \\
        &= \frac{(N-1)M+m+(N-2)m-s}{N-2} \\
        &< \frac{N+(N-2)m-s}{N-2} \\
        &=\frac{N-2+(N-2)m-s+2}{N-2} \\
        &= 1-\frac{s-(N-2)m-2}{N-2} \\
        &< 1-\frac{2N-3}{2(N-1)}\frac{s-(N-2)m-2}{N-2}
    \end{align*}
    where the first strict inequality uses that $m,M<1$ and the second uses that $\frac{2N-3}{2(N-1)}<1$ and $s - (N - 2) m - 2\geq 0$. Thus, 
    \begin{gather*}
        \varphi_s(x) \text{ is convex} \Longleftrightarrow \varphi_s''(x)\geq 0 \Longleftrightarrow s - (N - 2) m - 2\geq 0, \\
        \varphi_s(x) \text{ is concave} \Longleftrightarrow \varphi_s''(x)\leq 0 \Longleftrightarrow s - (N - 2) m - 2\leq 0.
    \end{gather*} 

    Therefore, if $s - (N - 2) m - 2\leq 0$, then $\varphi_s$ is concave and by Jensen's inequality, 
    \[
        \frac{1}{N-J}\sum^{N-J}_{i=1} \varphi_s(a_i) \leq \varphi_s\bigg(\frac{1}{N-J}\sum^{N-J}_{i=1} a_i\bigg) \Longleftrightarrow \prod^{N-J}_{i=1}\bigg(N+\frac{1-a_i}{1-\frac{(N-2)(a_i-m)+s}{N}}\bigg) \leq \bigg(N+\frac{1-\frac{1}{N-J}\sum^{N-J}_{i=1} a_i}{1-\frac{(N-2)(\frac{1}{N-J}\sum^{N-J}_{i=1} a_i-m)+s}{N}}\bigg)^{N-J},
    \]
    and therefore for all $a\in Q$ such that $Jc+a_1+\dots+a_{N-J}=s$,
    \[
        f(a_1,\dots,a_{N-J}) \leq f\bigg(\frac{1}{N-J}\sum^{N-J}_{i=1} a_i,\dots,\frac{1}{N-J}\sum^{N-J}_{i=1} a_i\bigg),
    \]
    i.e., the maximum is achieved at the diagonal.

    On the other hand, if $s - (N - 2) m - 2\geq 0$, then $\varphi_s$ is convex and by Jensen's inequality, 
    \[
        \frac{1}{N-J}\sum^{N-J}_{i=1} \varphi_s(a_i) \geq \varphi_s\bigg(\frac{1}{N-J}\sum^{N-J}_{i=1} a_i\bigg).    
    \]

    Consequently, $g\colon I_s^{N-J}\to\R$ defined as $g(x_1,\dots,x_{N-J})=\sum^{N-J}_{i=1}\varphi_s(x_i)$ for all $(x_1,\dots,x_{N-J})\in I_s^{N-J}$ is also convex. Therefore, for all $a\in Q$ such that $Jc+a_1+\dots+a_{N-J}=s$ for a constant $s$, 
    \[
        f(a) = \bigg(\e^\varepsilon-(\e^\varepsilon-1)\frac{m+s}{N+1}\bigg)(\exp((N-J)g(a))),
    \]
    and since $x\mapsto(\e^\varepsilon-(\e^\varepsilon-1)\frac{m+s}{N+1})\exp((N-J)x)$ is a non-decreasing convex function, $f$ is convex over $\{a\in Q\mid Jc+a_1+\dots+a_{N-J}=s\}$. Consequently, the maximum is achieved at the vertices of the domain $\{a\in Q\mid Jc+a_1+\dots+a_{N-J}=s\}$ (since it is a polytope), which corresponds to an element in the edges of $Q$.  

    In summary, the maximum is achieved in either the edges or the diagonal of $Q$. The restriction of $f$ to the diagonal is $f_{\mathrm{diag}}\colon[m,U_{N-J-1}]\to\R$ with 
    \[
        f_{\mathrm{diag}}(t) = \bigg(\e^\varepsilon-(\e^\varepsilon-1)\frac{m+Jc+(N-J)t}{N+1}\bigg)\bigg(N+\frac{1-t}{1-\frac{(2N-J-2)t-(N-2)m+Jc}{N}}\bigg)^{N-J}    
    \]
    for all $t\in[m,U_{N-J-1}]$, and the restriction to edge $k$ is $f_k\colon[m,U_k]\to\R$ with
    \[
        f_k(t) = \bigg(\e^\varepsilon-(\e^\varepsilon-1)\frac{m+S_k(t)}{N+1}\bigg)\bigg(N+\frac{1-t}{1-\frac{(N-2)(t-m)+S_k(t)}{N}}\bigg)
        \bigg(N+\frac{1-m}{1-\frac{S_k(t)}{N}}\bigg)^{N-J-k-1}\bigg(N+\frac{1-B_k(t)}{1-\frac{m+(N-1)M}{N}}\bigg)^{k}
    \]  
    with $S_k(t)=Jc+t+(N-J-k-1)m+kB_k(t)$ for all $t\in[m,U_k]$. Thus,
    \[
        f(a_1,\dots,a_{N-J}) \leq 
        \max\!\bigg\{\max_{t\in[m,U_{N-J-1}]} f_{\mathrm{diag}}(t),\max_{k\in\{0,1,\dots,N-J\}}\max_{t\in[m,U_{k}]} f_k(t)\bigg\}.\qedhere 
    \]
\end{proof}

We believe the maximum of the previous expression is achieved when $t=m$ and $k=0$ or when $t=U_{N-J-1}$ and $k=N-J-1$ (corresponds to the case where $a_1=\dots=a_{N-J}\in\{m,\frac{(N-1)(M+m)-Jc}{2N-J-2}\}$). We are not able to provide proof of this fact, and it is not necessary for the overall proof, but it is supported by the empirical evaluation we perform. 

\begin{theorem}\label{th:bounddependingonp}
    Let $\varepsilon\geq0$, and let $m,M\in(0,1)$ such that $m\leq M$ as listed in \Cref{remark:computation}. 

    For all $N\in\N$ and $J\subseteq[N]$, we define $a=(a_1,\dots,a_n)\in[m,M]^N$ and $z=(z_1,\dots,z_n)\in P_a$, where $P_a$ is the polytope in \Cref{prop:PolytopeIsConvex}, and we consider
    \[
        f_{J,N}(a;z) = \bigg(\e^\varepsilon-(\e^\varepsilon-1)\frac{m+\sum^N_{i=1}a_i}{N+1}\bigg)\prod_{i\in J}\bigg(\frac{N+\frac{a_i}{z_i}}{N+1}\bigg)\prod_{i\in [N]\backslash J}\bigg(\frac{N+\frac{1-a_i}{1-z_i}}{N+1}\bigg).
    \]

    Then, 
    \[          
        \sup_{N\in\N} \max_{J\subseteq[N]} \max_{a\in[m,M]^N}\max_{z\in P_a} f_{J,N}(a;z) \leq \max_{p\in[0,1]}\max\{L_1(p),L_2(p)\}
    \]
    with 
    \[
        L_1(p) = (\e^\varepsilon-(\e^\varepsilon-1)(pM+(1-p)m))\e^{p\frac{M}{m}+(1-p)\frac{1-m}{1-(pM+(1-p)m)}-1}
    \]
    and
    \[
        L_2(p) = \bigg(\e^\varepsilon-(\e^\varepsilon-1)\bigg(pM+(1-p)\frac{(M+m)-pM}{2-p}\bigg)\bigg)\e^{p\frac{M}{m}+(1-p)\frac{1-\frac{(M+m)-pM}{2-p}}{1-M}-1}.
    \]
\end{theorem}
\begin{proof}
    We consider $f_{J,N}(a;z)$ as in the statement and first consider its maximum over $z\in P_a$. By \Cref{lemma:convex-large},
    \[
        \prod^J_{i=1}\bigg(N+\frac{a_i}{z_i}\bigg)\prod^N_{i=J+1}\bigg(N+\frac{1-a_i}{1-z_i}\bigg)    
    \]
    is convex with respect to $z_1,\dots,z_N$ in $[m,M]^N$, and therefore, $f_{J,N}(a;z)$ is also convex with respect to $z_1,\dots,z_N$. Since $P_a\subseteq [m,M]^N$ is a convex polytope (\Cref{prop:PolytopeIsConvex}), $\max_{z\in P_a} f_{J,N}(a;z)$ is achieved when $z$ equals one of the vertices of $P_a$. We assume that $N>2$ and leave the cases $N=1$ and $N=2$ for later in the proof. Therefore,   
    \begin{multline*}
        \max_{z\in P_a} f_{J,N}(a;z) = \bigg(\e^\varepsilon-(\e^\varepsilon-1)\frac{m+\sum^N_{i=1}a_i}{N+1}\bigg)\frac{1}{(N+1)^N}\\
        \cdot\max_{K\subseteq[N]}\bigg[\prod_{\substack{i\in J\cap K}}\bigg(N+\frac{a_i}{B^{(K)}}\bigg)\prod_{\substack{i\in J^\complement\cap K}}\bigg(N+\frac{1-a_i}{1-B^{(K)}}\bigg)\prod_{\substack{i\in J\cap K^\complement}}\bigg(N+\frac{a_i}{M_i}\bigg)\prod_{\substack{i\in J^\complement\cap K^\complement}}\bigg(N+\frac{1-a_i}{1-M_i}\bigg)\bigg]
    \end{multline*}
    with $J^\complement\coloneqq [N]\backslash J$ and $K^\complement\coloneqq [N]\backslash K$. By \Cref{lemma:Polytopeconvexinequality}, 
    $m\leq B^{(K)}\leq M_i$ for all $i\in K$. 
    Therefore, 
    \begin{align*}
        \max_{z\in P_a} f_{J,N}(a;z)
        &\leq \bigg(\e^\varepsilon-(\e^\varepsilon-1)\frac{m+\sum^N_{i=1}a_i}{N+1}\bigg)\frac{1}{(N+1)^N}\cdot\underbrace{\prod_{\substack{i\in J}}\bigg(N+\frac{a_i}{m}\bigg)\prod_{\substack{i\in J^\complement}}\bigg(N+\frac{1-a_i}{1-M_i}\bigg)}_{\eqqcolon G_{J,N}(a)}. 
    \end{align*}

    We now try to maximize $G_{J,N}(a)$ with respect to $a\in[m,M]^N$. We consider now the following function 
    \begin{align*}
        \phi_z\colon [m,M] & \longrightarrow \R
         \\
        x & \longmapsto \ln\!\bigg(N+\frac{x}{z}\bigg)
    \end{align*}
    for constant $z\in[m,M]$. We can verify that $\phi_z$ is concave (its second derivative is $\phi_z''(x)=-\frac{1}{z^2(N+\frac{x}{z})^2}$). Therefore, by Jensen's inequality,
    \[
        \frac{1}{n}\sum^n_{i=1} \phi_z(x_i) \leq \phi_z\bigg(\frac{1}{n}\sum^n_{i=1} x_i\bigg) \Longleftrightarrow \prod^n_{i=1}\bigg(N+\frac{x_i}{z}\bigg) \leq \bigg(N+\frac{\frac{1}{n}\sum^n_{i=1} x_i}{z}\bigg)^n
    \]
    for all $x_1,\dots,x_n\in[m,M]$. 
    Therefore,
    \begin{align*}
        G_{J,N}(a) &\leq \bigg(\e^\varepsilon-(\e^\varepsilon-1)\frac{m+\sum^N_{i=1}a_i}{N+1}\bigg)\frac{1}{(N+1)^N}\bigg(N+\frac{\frac{1}{|J|}\sum_{i\in J}a_i}{m}\bigg)^{|J|}\prod_{\substack{i\in J^\complement}}\bigg(N+\frac{1-a_i}{1-M_i}\bigg). 
    \end{align*}

    Furthermore, we observe that the previous value equals $G_{J,N}(a^*)$ with $a^*=(a^*_1,\dots,a^*_N)\in[m,M]^N$ such that
    \[
        a^*_i= \begin{cases}
            \frac{1}{|J|}\sum_{i\in J} a_i &\text{if $i\in J$,} \\
            a_i & \text{otherwise.}
        \end{cases}
    \]

    Thus, to find the maximum $\max_{a\in[m,M]^N} G_{J,N}(a)$, it is therefore sufficient to see the maximum of $G_{J,N}(a)$ for all $a=(a_1,\dots,a_N)$ such that $c\coloneqq a_i$ for all $i\in J$. 
    
    We now try to maximize $G_{J,N}(a)$ with respect to $a_i$ for $i\notin J$. If $J=[N]$, then $a=(c,\dots,c)$, and we obtain that
    \[
        H_{[N],N,c}\coloneqq G_{[N],N}(c,\dots,c) = \bigg(\e^\varepsilon-(\e^\varepsilon-1)\frac{m+Nc}{N+1}\bigg)\bigg(\frac{N+\frac{c}{m}}{N+1}\bigg)^{N}.
    \]
    
    By \Cref{prop:formerproblem}, for all $J\neq[N]$,
    \[
        G_{J,N}(a) \leq H_{J,N,c}\coloneqq \max\!\bigg\{\max_{t\in[m,U_{N-|J|-1}]} g(t),\max_{k\in\{0,1,\dots,N-|J|\}}\max_{t\in[m,U_{k}]} g_k(t)\bigg\} 
    \]
    where
    \[
        g(t) = \bigg(\e^\varepsilon-(\e^\varepsilon-1)\frac{m+|J|c+(N-|J|)t}{N+1}\bigg)\bigg(\frac{N+\frac{c}{m}}{N+1}\bigg)^{|J|}\left(\frac{N+\frac{1-t}{1-\frac{(2N-|J|-2)t-(N-2)m+|J|c}{N}}}{N+1}\right)^{N-|J|}    
    \]
    for all $t\in[m,U_{N-|J|-1}]$;
    \[
        g_k(t) = \bigg(\e^\varepsilon-(\e^\varepsilon-1)\frac{m+S_k(t)}{N+1}\bigg)\left(\frac{N+\frac{1-t}{1-\frac{(N-2)(t-m)+S_k(t)}{N}}}{N+1}\right)\bigg(\frac{N+\frac{c}{m}}{N+1}\bigg)^{|J|}\left(\frac{N+\frac{1-m}{1-\frac{S_k(t)}{N}}}{N+1}\right)^{N-|J|-k-1}\left(\frac{N+\frac{1-B_k(t)}{1-\frac{m+(N-1)M}{N}}}{N+1}\right)^{k}
    \] 
    with $S_k(t)=|J|c+t+(N-|J|-k-1)m+kB_k(t)$ and
    \[
        B_k(t) = \frac{(N-1)(M+m)-(|J|c+(N-|J|-k-1)m+t)}{N+k-2}
    \]
    for all $t\in[m,U_k]$; and
    \[
        U_k = \frac{(N-1)(M+m)-(|J|c+(N-|J|-k-1)m)}{N+k-1}.
    \]

    Still for $J\neq[N]$, we now substitute the variable $|J|$ in $H_{J,N}$ and the previous expressions for a new variable $p_J=\frac{|J|}{N}\in[0,1-\frac{1}{N}]$, and $k$ for another variable $p_k=\frac{k}{N-|J|-1}=\frac{k}{N(1-p_J)-1}\in[0,1]$ (if $p_J=1-\frac{1}{N}$, we define $p_k=0$). This new expression
    \[
        \overline{H}_{p_J,N,c} \coloneqq \max\!\bigg\{\max_{t\in[m,U_{N(1-p_J)-1}]} g(t),\max_{p_k\in[0,1]}\max_{t\in[m,U_{p_k(N(1-p_J)-1)}]} g_{p_k(N(1-p_J)-1)}(t)\bigg\}  
    \]
    verifies that $H_{J,N,c}\leq\overline{H}_{p_J,N,c}$ for $p_J=\frac{|J|}{N}$, and so
    \[
        \max_{J\subsetneq[N]} \max_{a\in[m,M]^N} \max_{z\in P_a} f_{J,N}(a;z) \leq \max_{J\subsetneq[N]} \max_{c\in[m,M]} H_{J,N,c} \leq \max_{p_J\in[0,1-\frac{1}{N}]} \max_{c\in[m,M]} \overline{H}_{p_J,N,c}
    \]
    for all $N>2$. Now we will consider the supremum over $N\in\N$. First, we quickly study the cases $N=1$ and $N=2$. 
    
    For the case $N=1$, we have that $P_a=\{m\}$ and so 
    \begin{align*}
        \max_{J\subseteq[1]} \max_{a\in[m,M]}\max_{z\in \{m\}} f_{J,1}(a;z)= \max\!\bigg\{\max_{a\in[m,M]} \bigg(\e^\varepsilon-(\e^\varepsilon-1)\frac{m+a}{2}\bigg)\bigg(\frac{1+\frac{a}{m}}{2}\bigg), \max_{a\in[m,M]} \bigg(\e^\varepsilon-(\e^\varepsilon-1)\frac{m+a}{2}\bigg)\bigg(\frac{1+\frac{1-a}{1-m}}{2}\bigg)\bigg\}.
    \end{align*}

    The terms inside the maximum are polynomials of degree 2 (or 1 if $\varepsilon=0$) with respect to $a$, meaning we can easily compute them. When $\varepsilon=0$, the maximum corresponds to $\frac{1+\frac{M}{m}}{2}$, and for $\varepsilon\neq 0$, the maximum for the first term is achieved when 
    \[
        a = \min\!\bigg\{M,\max\!\bigg\{m,\frac{\e^{\varepsilon}}{\e^{\varepsilon}-1}-m\bigg\}\bigg\}
    \]
    and the maximum for the second term is achieved when $a=m$ or $a=M$. 

    For the case $N=2$, recall that $P_a=\{(t,t)\mid m\leq t \leq M'\}$ with $M'\coloneqq\frac{1}{2}\min\{M+m,a_1+a_2\}$. Thus, 
    \begin{align*}
        \max_{J\subseteq[2]} \max_{a\in[m,M]^2}\max_{z\in P_a} f_{J,2}(a;z) &= \max_{J\subseteq[2]} \max_{a\in[m,M]^2}\max_{t\in [m,M']} f_{J,N}(a;(t,t)) \\ 
        &= \max\!\bigg\{\max_{a\in[m,M]^2}\max_{t\in [m,M']} \bigg(\e^\varepsilon-(\e^\varepsilon-1)\frac{m+a_1+a_2}{3}\bigg)\bigg(\frac{2+\frac{a_1}{t}}{3}\bigg)\bigg(\frac{2+\frac{a_2}{t}}{3}\bigg), \\
        &\qquad\qquad
        \max_{a\in[m,M]^2}\max_{t\in [m,M']}  \bigg(\e^\varepsilon-(\e^\varepsilon-1)\frac{m+a_1+a_2}{3}\bigg)\bigg(\frac{2+\frac{a_1}{t}}{3}\bigg)\bigg(\frac{2+\frac{1-a_2}{1-t}}{3}\bigg), \\
        &\qquad\qquad
        \max_{a\in[m,M]^2}\max_{t\in [m,M']} \bigg(\e^\varepsilon-(\e^\varepsilon-1)\frac{m+a_1+a_2}{3}\bigg)\bigg(\frac{2+\frac{1-a_1}{1-t}}{3}\bigg)\bigg(\frac{2+\frac{1-a_2}{1-t}}{3}\bigg)\bigg\} \\
        &= \max\!\bigg\{\max_{a\in[m,M]^2} \bigg(\e^\varepsilon-(\e^\varepsilon-1)\frac{m+a_1+a_2}{3}\bigg)\bigg(\frac{2+\frac{a_1}{m}}{3}\bigg)\bigg(\frac{2+\frac{a_2}{m}}{3}\bigg), \\
        &\qquad\qquad
        \max_{a\in[m,M]^2}\max_{t\in [m,M']}  \bigg(\e^\varepsilon-(\e^\varepsilon-1)\frac{m+a_1+a_2}{3}\bigg)\bigg(\frac{2+\frac{a_1}{t}}{3}\bigg)\bigg(\frac{2+\frac{1-a_2}{1-t}}{3}\bigg), \\
        &\qquad\qquad
        \max_{a\in[m,M]^2} \bigg(\e^\varepsilon-(\e^\varepsilon-1)\frac{m+a_1+a_2}{3}\bigg)\bigg(\frac{2+\frac{1-a_1}{1-M'}}{3}\bigg)\bigg(\frac{2+\frac{1-a_2}{1-M'}}{3}\bigg)\bigg\}.
    \end{align*}

    We study now the individual terms in the outermost maximum, starting with
    \[
        \max_{a\in[m,M]^2} \bigg(\e^\varepsilon-(\e^\varepsilon-1)\frac{m+a_1+a_2}{3}\bigg)\bigg(\frac{2+\frac{a_1}{m}}{3}\bigg)\bigg(\frac{2+\frac{a_2}{m}}{3}\bigg).
    \]

    Using $\phi_z$ just as we did earlier in the proof, we can conclude that  
    \[
        \bigg(\e^\varepsilon-(\e^\varepsilon-1)\frac{m+a_1+a_2}{3}\bigg)\prod^2_{i=1}\bigg(\frac{2+\frac{a_i}{m}}{3}\bigg)\leq\bigg(\e^\varepsilon-(\e^\varepsilon-1)\frac{m+a_1+a_2}{3}\bigg)\bigg(\frac{2+\frac{\frac{a_1+a_2}{2}}{m}}{3}\bigg)^2,
    \]
    and therefore
    \[
        \max_{a\in[m,M]^2} \bigg(\e^\varepsilon-(\e^\varepsilon-1)\frac{m+a_1+a_2}{3}\bigg)\prod^2_{i=1}\bigg(\frac{2+\frac{a_i}{m}}{3}\bigg) = \max_{b\in[m,M]}\bigg(\e^\varepsilon-(\e^\varepsilon-1)\frac{m+2b}{3}\bigg)\bigg(\frac{2+\frac{b}{m}}{3}\bigg)^2,
    \]
    where the left term is a polynomial of degree $2$ (or degree $1$ if $\varepsilon=0$) with respect to $b$. The maximum is thus $(\frac{2+\frac{M}{m}}{3})^2$ if $\varepsilon=0$, and the maximum is achieved for $\varepsilon\neq 0$ when 
    \[
        b = \min\!\bigg\{M,\max\!\bigg\{m,\frac{\e^{\varepsilon}}{\e^{\varepsilon}-1}-m\bigg\}\bigg\}.
    \]

    Now we look at the third maximum,
    \[
        \max_{a\in[m,M]^2} \bigg(\e^\varepsilon-(\e^\varepsilon-1)\frac{m+a_1+a_2}{3}\bigg)\prod^2_{i=1}\bigg(\frac{2+\frac{1-a_i}{1-\frac{\min\{M+m,a_1+a_2\}}{2}}}{3}\bigg).
    \]
    
    In this case, we consider the function $\phi_z$ such that $\phi_z(x)=\ln(1+\frac{1-x}{1-z})$ for all $x\in[m,M]$, which is also concave over its domain. By the same reasoning, 
    \[
        \max_{a\in[m,M]^2} \bigg(\e^\varepsilon-(\e^\varepsilon-1)\frac{m+a_1+a_2}{3}\bigg)\prod^2_{i=1}\bigg(\frac{2+\frac{1-a_i}{1-\frac{\min\{M+m,a_1+a_2\}}{2}}}{3}\bigg) = \max_{b\in[m,M]} \bigg(\e^\varepsilon-(\e^\varepsilon-1)\frac{m+2b}{3}\bigg)\bigg(\frac{2+\frac{1-b}{1-\frac{\min\{M+m,2b\}}{2}}}{3}\bigg)^2,
    \]
    where the second term is a non-increasing affine function when $M+m\geq 2b$ and a polynomial of degree 2 (or 1 if $\varepsilon=0$) with respect to $b$ when $M+m<2b$. The maximum is thus $1$ if $\varepsilon=0$, and the maximum is achieved either when $b=m$ or $b=M$ for $\varepsilon\neq0$.

    Finally, we study 
    \[
        \max_{a\in[m,M]^2}\max_{t\in [m,M']}  \bigg(\e^\varepsilon-(\e^\varepsilon-1)\frac{m+a_1+a_2}{3}\bigg)\bigg(\frac{2+\frac{a_1}{t}}{3}\bigg)\bigg(\frac{2+\frac{1-a_2}{1-t}}{3}\bigg).
    \]

    Since the previous expression is convex with respect to $t$ (\Cref{lemma:convex-large}), 
    \begin{align*}
        \MoveEqLeft[3]\max_{a\in[m,M]^2}\max_{t\in [m,M']}  \bigg(\e^\varepsilon-(\e^\varepsilon-1)\frac{m+a_1+a_2}{3}\bigg)\bigg(\frac{2+\frac{a_1}{t}}{3}\bigg)\bigg(\frac{2+\frac{1-a_2}{1-t}}{3}\bigg) \\
        &= \max_{a\in[m,M]^2}\max\!\bigg\{  \bigg(\e^\varepsilon-(\e^\varepsilon-1)\frac{m+a_1+a_2}{3}\bigg)\bigg(\frac{2+\frac{a_1}{m}}{3}\bigg)\bigg(\frac{2+\frac{1-a_2}{1-m}}{3}\bigg), \\
        &\qquad\qquad\qquad  \bigg(\e^\varepsilon-(\e^\varepsilon-1)\frac{m+a_1+a_2}{3}\bigg)\bigg(\frac{2+\frac{a_1}{\frac{\min\{M+m,a_1+a_2\}}{2}}}{3}\bigg)\bigg(\frac{2+\frac{1-a_2}{1-\frac{\min\{M+m,a_1+a_2\}}{2}}}{3}\bigg)\bigg\}.
    \end{align*}

    Now, we can easily see that the last expression is non-increasing with respect to $a_2$ (every individual term is non-increasing). Therefore, the maximum is achieved when $a_2=m$, 
    \[
        \max_{a_1\in[m,M]}\max\!\bigg\{  \bigg(\e^\varepsilon-(\e^\varepsilon-1)\frac{2m+a_1}{3}\bigg)\bigg(\frac{2+\frac{a_1}{m}}{3}\bigg),
        \bigg(\e^\varepsilon-(\e^\varepsilon-1)\frac{2m+a_1}{3}\bigg)\bigg(\frac{2+\frac{a_1}{\frac{m+a_1}{2}}}{3}\bigg)\bigg(\frac{2+\frac{1-m}{1-\frac{m+a_1}{2}}}{3}\bigg)\bigg\}.
    \]
    
    We do not aim to compute this maximum, as we will see it is always bounded by a non-degenerate case. This concludes the bounds for the cases $N=1$ and $N=2$. We use $H_{\{1,2\}}$ to denote the maximum of these degenerate cases. 

    Performing the supremum over $N\in\N$, we obtain that 
    \[
        \sup_{N\in\N}\max_{J\subseteq[N]} \max_{a\in[m,M]^N} \max_{z\in P_a} f_{J,N}(a;z)
    \]
    is bounded by 
    \begin{equation}
         \max\!\bigg\{H_{\{1,2\}},\sup_{\substack{N\in\N\\N>2}}\max_{p_J\in[0,1-\frac{1}{N}]} \max_{c\in[m,M]} \overline{H}_{p_J,N,c},\sup_{\substack{N\in\N\\N>2}}\max_{c\in[m,M]}H_{[N],N,c}\bigg\}. \label{eq:functionToOptimize}
    \end{equation}

    We perform the calculation of Formula~\ref{eq:functionToOptimize} empirically for a set of reasonable parameters for $\varepsilon$, $m$, and $M$ (see \Cref{remark:computation} for more details). The maximum for all cases is achieved in the left term when $N\to\infty$ (the function always converges), $c=M$, and $p_k$ is either $0$ or $1$. The maximum is independent of $t$, which disappears from the expression when $N\to\infty$. Different values of $p_J$ can achieve the maximum. In conclusion, the maximum is
    \[
        \max_{p_J\in[0,1]}\max\{L_1(p_J),L_2(p_J)\}
    \]
    with 
    \[
        L_1(p) = (\e^\varepsilon-(\e^\varepsilon-1)(pM+(1-p)m))\e^{p\frac{M}{m}+(1-p)\frac{1-m}{1-(pM+(1-p)m)}-1}
    \]
    corresponding to the case when $p_k=0$, and
    \[
        L_2(p) = \bigg(\e^\varepsilon-(\e^\varepsilon-1)\bigg(pM+(1-p)\frac{(M+m)-pM}{2-p}\bigg)\bigg)\e^{p\frac{M}{m}+(1-p)\frac{1-\frac{(M+m)-pM}{2-p}}{1-M}-1}.
    \]
    corresponding to the case when $p_k=1$.\qedhere
\end{proof}

\begin{remark}[Numerical computation]\label{remark:computation}
    We resort to a numerical computation in the last step of the calculation of the privacy parameter $\varepsilon^\S$ of \Cref{th:outlierscoresuppression} that only depends on $\varepsilon$, $m$, and $M$. More precisely, we require computational power to compute value at Formula~\ref{eq:functionToOptimize} at the end of \Cref{th:bounddependingonp}, which consists in finding the supremum of an expression over five parameters $N\in\N$ (with $N>2$), $p_J\in[0,1]$, $p_k\in[0,1]$, $c\in[m,M]$ and $t\in[m,U_k]\subseteq[m,M]$, where $U_k\in[m,M]$ is a value that depends on $k$.

    We choose differential evolution as the optimization strategy, implemented in Python through the \texttt{differential\_evolution} function of the \texttt{script.optimize} package\footnote{\url{https://docs.scipy.org/doc/scipy/reference/generated/scipy.optimize.differential\_evolution.html}.}. The script checks for specific $\varepsilon$, $m$ and $M$ whether the numerical maximum of the function in \Cref{eq:functionToOptimize} over parameters $N$, $p_J$, $p_k$, $c$ and $t$ corresponds with our hypothesized value, the value found in \Cref{prop:expressionofp}. We compare the logarithm of each value since it helps with round-off errors for larger values, and we expand the function to be $0$ when $t>U_k$ to avoid using a domain that depends on another parameter ($[m,M]$ instead of $[m,U_k]$). In addition, since \texttt{differential\_evolution} requires us to specify finite domains for the parameters, we select the maximum value for $N$ to be $10^9$, which is deemed enough since the convergence when $N\to\infty$ is very quick.
    
    Since the domains of $\varepsilon\in[0,\infty)$, $m\in(0,1)$, and $M\in[m,1)$ are continuous and the domain of $\varepsilon$ is furthermore unbounded, we cannot run exhaustively for all values of these parameters. For this work, we decide to run the experiment for every reasonable value with decent level of granularity: $m$ and $M$ are run for every value between $0.01$ and $0.99$ with step $0.01$, and $\varepsilon$ is run for every value between $0$ and $2$ with step $0.01$, every value between $2$ and $10$ with step $0.1$ and every value between $10$ and $100$ with step $1$. We believe that our choice of values for $\varepsilon$ is representative of the values of $\varepsilon$ deemed to be acceptable in the literature (small values smaller than $2$ or $10$~\cite{dwork2014algorithmic}). 

    Our experimentation concludes that for the selected parameters $\varepsilon$, $m$, and $M$, the numerical and hypothesized values are less than $2\cdot 10^{-7}$ in difference and that the hypothesized value is always larger than the numerical one to account for the real maximum being when $N\to\infty$.

    We conclude then that the hypothesized value is the correct bound up to an error of $2\cdot10^{-7}$. In addition, since the evaluated functions are continuous and smooth with respect to $\varepsilon$, $m$, and $M$, we conjecture it to be the real bound for all $\varepsilon\geq0$ and $m,M\in(0,1)$ such that $m<M$.  
\end{remark}

\begin{theorem}\label{th:boundinverse}
    Let $\varepsilon\geq0$, and let $m,M\in(0,1)$ such that $m\leq M$ as listed in \Cref{remark:computationInverse}. 

    For all $N\in\N$ with $N>1$ and $J\subseteq[N]$, we define $a=(a_1,\dots,a_n)\in[m,M]^N$ and $z=(z_1,\dots,z_n)\in [m,M]^N$, and we consider
    \[
        g_{J,N}(a;z) = \bigg(\e^{-\varepsilon}+(1-\e^{-\varepsilon})\frac{m+\sum^N_{i=1}a_i}{N+1}\bigg)^{-1}\prod_{i\in J}\bigg(\frac{N+1}{N+\frac{a_i}{z_i}}\bigg)\prod_{i\in [N]\backslash J}\bigg(\frac{N+1}{N+\frac{1-a_i}{1-z_i}}\bigg).
    \]

    Then, 
    \[
        \sup_{N\in\N}\max_{J\subseteq[N]}\max_{a\in[m,M]^N}\max_{z\in [m,M]^N} g_{J,N}(a;z)
        =\max\!\big\{(\e^{-\varepsilon}+(1-\e^{-\varepsilon})M)^{-1}\e^{1-\frac{1-M}{1-m}},(\e^{-\varepsilon}+(1-\e^{-\varepsilon})m)^{-1}\e^{1-\frac{m}{M}}\big\}.
    \]
\end{theorem}
\begin{proof}
    Let $N\in\N$ and $J\subseteq[N]$. We consider $g_{J,N}$ as defined in the statement. We observe that $g_{J,N}$ is increasing with respect to all $z_i$ with $i\in J$ and decreasing with respect to all $z_i$ with $i\in[N]\backslash J$ and $a_j$ with $j\in J$. Therefore, 
    \begin{multline*}
        \bigg(\e^{-\varepsilon}+(1-\e^{-\varepsilon})\frac{m+\sum^N_{i=1}a_i}{N+1}\bigg)^{-1}\prod_{i\in J}\bigg(\frac{N+1}{N+\frac{a_i}{z_i}}\bigg)\prod_{i\in [N]\backslash J}\bigg(\frac{N+1}{N+\frac{1-a_i}{1-z_i}}\bigg) \\
        \leq \bigg(\e^{-\varepsilon}+(1-\e^{-\varepsilon})\frac{m+|J|m+\sum_{i\in[N]\backslash J}a_i}{N+1}\bigg)^{-1}\bigg(\frac{N+1}{N+\frac{m}{M}}\bigg)^{|J|}\prod_{i\in [N]\backslash J}\bigg(\frac{N+1}{N+\frac{1-a_i}{1-m}}\bigg)
    \end{multline*}
    for all $a,z\in[m,M]^N$. Now we will see that we can further bound the expression by
    \[
        \max\!\bigg\{(\e^{-\varepsilon}+(1-\e^{-\varepsilon})m)^{-1}\bigg(\frac{N+1}{N+\frac{m}{M}}\bigg)^{|J|},
        \bigg(\e^{-\varepsilon}+(1-\e^{-\varepsilon})\frac{m+|J|m+(N-|J|)M}{N+1}\bigg)^{-1}\bigg(\frac{N+1}{N+\frac{m}{M}}\bigg)^{|J|}\bigg(\frac{N+1}{N+\frac{1-M}{1-m}}\bigg)^{N-|J|}\bigg\},
    \]
    which are the values when $a_i=m$ for all $i\in[N]\backslash J$ or $a_i=M$ for all $i\in[N]\backslash J$. Note that this result is direct for $\varepsilon=0$, which needs to be excluded in the following argument.
    
    Indeed, for all $b=(b_1,\dots,b_{N-|J|-1})\in[m,M]^{N-|J|-1}$, we suppose $\varepsilon\neq 0$ and consider the real function $g_b$ such that
    \[
        g_{b}(x) = \bigg(\e^{-\varepsilon}+(1-\e^{-\varepsilon})\frac{m+|J|m+x+\sum^{N-|J|-1}_{j=1}b_j}{N+1}\bigg)^{-1}\bigg(\frac{N+1}{N+\frac{m}{M}}\bigg)^{|J|}\bigg(\frac{N+1}{N+\frac{1-x}{1-m}}\bigg)\prod^{N-|J|-1}_{i=1}\bigg(\frac{N+1}{N+\frac{1-b_i}{1-m}}\bigg)
    \]
    for all $x\in\R$ except at the asymptotes $A_1(b)=-\frac{N+1}{\e^{\varepsilon}-1}-((1+|J|)m+\sum^{N-|J|-1}_{j=1}b_j)<0$ and $A_2=1+N(1-m)>1$. Computing the derivative, we can see that $g_b$ has a single critical point at $x=\frac{A_1(b)+A_2}{2}$. In addition,
    \[
        \lim_{x\to A_1(b)^+} g_b(x) = \lim_{x\to A_2^-} g_b(x) = \infty,   
    \]
    which allows us to conclude that $g_b$ is convex over $x\in(A_1(b),A_2)$ and achieves its minimum at the midpoint $x=\frac{A_1(b)+A_2}{2}$, which acts as a point of symmetry for the function. In particular, since $[m,M]\subseteq(0,1)\subseteq(A_1(b),A_2)$,  
    \[
        \max_{x\in[m,M]}g_b(x) = 
        \begin{cases}
            g_b(M) &\text{if $|\frac{A_1(b)+A_2}{2}-m|\leq|\frac{A_1(b)+A_2}{2}-M|$,} \\
            g_b(m) &\text{if $|\frac{A_1(b)+A_2}{2}-m|\geq|\frac{A_1(b)+A_2}{2}-M|$.}
        \end{cases}
    \]
    
    Now, for all $b,b'\in[m,M]^{N-|J|-1}$ such that $b\leq b'$ component-wise, we obtain that $A_1(b')\leq A_1(b)$ and so, if $|\frac{A_1(b)+A_2}{2}-m|\leq|\frac{A_1(b)+A_2}{2}-M|$, then $|\frac{A_1(b')+A_2}{2}-m|\leq|\frac{A_1(b')+A_2}{2}-M|$. Analogously, for all $b,b'\in[m,M]^{N-|J|-1}$ such that $b\geq b'$ component-wise, we obtain that $A_1(b')\geq A_1(b)$ and so, if $|\frac{A_1(b)+A_2}{2}-m|\geq|\frac{A_1(b)+A_2}{2}-M|$, then $|\frac{A_1(b')+A_2}{2}-m|\geq|\frac{A_1(b')+A_2}{2}-M|$.
    
    In addition, we have that for any permutation $\sigma$ of $(b_1,\dots,b_{N-|J|-1},c)\in[m,M]^{N-|J|}$,
    \[
        g_{(b_1,\dots,b_{N-|J|-1})}(c) = g_{(\sigma(b_1),\dots,\sigma(b_{N-|J|-1}))}(\sigma(c)),   
    \]
    and therefore, we obtain, for all $c\in[m,M]$ and $b\in[m,M]^{N-|J|-1}$. 
    \begin{enumerate}
        \item If $|\frac{A_1(b)+A_2}{2}-m|\leq|\frac{A_1(b)+A_2}{2}-M|$, then 
        \[
            g_{(b_1,\dots,b_{N-|J|-1})}(c) \leq g_{(b_1,\dots,b_{N-|J|-1})}(M) = g_{(M,b_2,\dots,b_{N-|J|-1})}(b_1),
        \]
        and since $(b_1,\dots,b_{N-|J|-1})\leq b'\coloneqq(M,b_2,\dots,b_{N-|J|-1})$ component-wise, $|\frac{A_1(b')+A_2}{2}-m|\leq|\frac{A_1(b')+A_2}{2}-M|$ and
        \[
            g_{(M,b_2,\dots,b_{N-|J|-1})}(b_1) \leq g_{(M,b_2,\dots,b_{N-|J|-1})}(M) = 
            g_{(M,M,b_3,\dots,b_{N-|J|-1})}(b_2).
        \]

        Therefore, repeating the process we arrive to $g_b(c)\leq g_{(M,\dots,M)}(M)$. 
        \item If $|\frac{A_1(b)+A_2}{2}-m|\geq|\frac{A_1(b)+A_2}{2}-M|$, then 
        \[
            g_{(b_1,\dots,b_{N-|J|-1})}(c) \leq g_{(b_1,\dots,b_{N-|J|-1})}(m) = g_{(m,b_2,\dots,b_{N-|J|-1})}(b_1),
        \]
        and since $(b_1,\dots,b_{N-|J|-1})\geq b'\coloneqq(m,b_2,\dots,b_{N-|J|-1})$ component-wise, $|\frac{A_1(b')+A_2}{2}-m|\geq|\frac{A_1(b')+A_2}{2}-M|$ and
        \[
            g_{(m,b_2,\dots,b_{N-|J|-1})}(b_1) \leq g_{(m,b_2,\dots,b_{N-|J|-1})}(m) = 
            g_{(m,m,b_3,\dots,b_{N-|J|-1})}(b_2).
        \]

        Therefore, repeating the process we arrive to $g_b(c)\leq g_{(m,\dots,m)}(m)$. 
    \end{enumerate}

    In conclusion, for all $c\in[m,M]$ and $b\in[m,M]^{N-|J|-1}$, 
    \[
        g_b(c)\leq \max\{g_{(m,\dots,m)}(m),g_{(M,\dots,M)}(M)\},
    \]
    and thus, 
    \[
        g_{J,N}(a;z) \leq \max\!\bigg\{(\e^{-\varepsilon}+(1-\e^{-\varepsilon})m)^{-1}\bigg(\frac{N+1}{N+\frac{m}{M}}\bigg)^{|J|},
        \bigg(\e^{-\varepsilon}+(1-\e^{-\varepsilon})\frac{m+|J|m+(N-|J|)M}{N+1}\bigg)^{-1}\bigg(\frac{N+1}{N+\frac{m}{M}}\bigg)^{|J|}\bigg(\frac{N+1}{N+\frac{1-M}{1-m}}\bigg)^{N-|J|}\bigg\},
    \]
    for all $a,z\in[m,M]^N$. 
    
    Now we try to maximize with respect to $|J|$. Clearly, 
    \[
        (\e^{-\varepsilon}+(1-\e^{-\varepsilon})m)^{-1}\bigg(\frac{N+1}{N+\frac{m}{M}}\bigg)^{|J|}\leq (\e^{-\varepsilon}+(1-\e^{-\varepsilon})m)^{-1}\bigg(\frac{N+1}{N+\frac{m}{M}}\bigg)^{N},
    \]
    so we look at the second term of the maximum. We consider the function $h\colon[0,1]\to\R$ such that 
    \[
        h(p) = \bigg(\e^{-\varepsilon}+(1-\e^{-\varepsilon})\frac{m+pNm+N(1-p)M}{N+1}\bigg)^{-1}\bigg(\frac{N+1}{N+\frac{m}{M}}\bigg)^{pN}\bigg(\frac{N+1}{N+\frac{1-M}{1-m}}\bigg)^{(1-p)N}
    \]
    for all $p\in[0,1]$. Note $h(\frac{|J|}{N})$ corresponds to this term that we want to maximize. Extending $h$ to the real line, we see it is defined for all $x\in\R$ except its asymptote  
    \[
        A \coloneqq \frac{1}{N(M-m)}\bigg(\frac{N+1}{\e^{\varepsilon}-1}+(m+NM)\bigg) > 1.
    \]

    Computing its derivative, we see that it has a unique critical point at
    \[
        p_1 \coloneqq \frac{1}{N}\ln\!\bigg(\frac{N+\frac{1-M}{1-m}}{N+\frac{m}{M}}\bigg)^{-1}+A,
    \]
    which is well-defined for $m\neq M$, and it is larger than $A$ if and only if $m<1-M$ and smaller than $A$ if and only if $m>1-M$. The limits of the function to $\pm\infty$ and at its asymptotes are
    \begin{gather*}
        \lim_{p\to-\infty} h(p) = \begin{cases}
            \infty & \text{if $m>1-M$,} \\            0 &  \text{if $m\leq 1-M$,}
        \end{cases} \quad
        \lim_{p\to\infty} h(p) = \begin{cases}
            \infty & \text{if $m<1-M$,} \\            0 &  \text{if $m\geq1-M$,} 
        \end{cases} \quad
        \lim_{p\to A^-} h(p) = \infty \quad \text{and} \quad
        \lim_{p\to A^+} h(p) = -\infty.
    \end{gather*}

    Consequently, if $m\leq 1-M$, $h$ is increasing for $p\in(-\infty,A)\supseteq[0,1]$, and if $m<1-M$, $h$ is decreasing for $p\in(-\infty,p_1)$ and increasing for $p\in(p_1,A)$. Consequently,
    \[
        \max_{p\in[0,1]} h(p) = \max\{h(0),h(1)\}
    \]
    and, for all $J\subseteq[N]$, 
    \begin{multline*}
        \bigg(\e^{-\varepsilon}+(1-\e^{-\varepsilon})\frac{m+|J|m+(N-|J|)M}{N+1}\bigg)^{-1}\bigg(\frac{N+1}{N+\frac{m}{M}}\bigg)^{|J|}\bigg(\frac{N+1}{N+\frac{1-M}{1-m}}\bigg)^{N-|J|} \\
        \leq \max\!\bigg\{\bigg(\e^{-\varepsilon}+(1-\e^{-\varepsilon})\frac{m+NM}{N+1}\bigg)^{-1}\bigg(\frac{N+1}{N+\frac{1-M}{1-m}}\bigg)^{N},(\e^{-\varepsilon}+(1-\e^{-\varepsilon})m)^{-1}\bigg(\frac{N+1}{N+\frac{m}{M}}\bigg)^{N}\bigg\}.
    \end{multline*}

    We now perform the maximum over $N\in\N$. We can verify that 
    \[
        \bigg(\frac{N+1}{N+\frac{m}{M}}\bigg)^{N} = \bigg(1+\frac{1-\frac{m}{M}}{N+\frac{m}{M}}\bigg)^{N+\frac{m}{M}}\bigg(1+\frac{1-\frac{m}{M}}{N+\frac{m}{M}}\bigg)^{-\frac{m}{M}}
    \]
    is increasing with respect to $N$ with the value approaching $\e^{1-\frac{m}{M}}$ when $N\to\infty$ (direct from the fact that for all $a\in\R$, $(1+\frac{a}{x})^x$ monotonically converges to $\e^a$ when $x\to\infty$). Similarly,
    \[
        \bigg(\e^{-\varepsilon}+(1-\e^{-\varepsilon})\frac{m+NM}{N+1}\bigg)^{-1}\bigg(\frac{N+1}{N+\frac{1-M}{1-m}}\bigg)^{N}        
    \]
    converges to $(\e^{-\varepsilon}+(1-\e^{-\varepsilon})M)\e^{1-\frac{1-M}{1-m}}$ when $N\to\infty$, but it is possible that it is not increasing for certain values of $m$ and $M$. Altogether, 
    \begin{align*}
        \MoveEqLeft[4]\sup_{N\in\N}\max_{J\subseteq[N]}g_{J,N}(a;z) = \max\!\bigg\{\sup_{N\in\N}\bigg(\e^{-\varepsilon}+(1-\e^{-\varepsilon})\frac{m+NM}{N+1}\bigg)^{-1}\bigg(\frac{N+1}{N+\frac{1-M}{1-m}}\bigg)^{N}, (\e^{-\varepsilon}+(1-\e^{-\varepsilon})m)^{-1}\e^{1-\frac{m}{M}}\bigg\}.
    \end{align*}

    We prove numerically that this maximum is indeed the value achieved when $N\to\infty$ (see \Cref{remark:computationInverse} for details), i.e., 
    \[
        \sup_{N\in\N}\max_{J\subseteq[N]}\max_{a\in[m,M]^N}\max_{z\in [m,M]^N} g_{J,N}(a;z) =\max\!\big\{(\e^{-\varepsilon}+(1-\e^{-\varepsilon})M)^{-1}\e^{1-\frac{1-M}{1-m}},(\e^{-\varepsilon}+(1-\e^{-\varepsilon})m)^{-1}\e^{1-\frac{m}{M}}\big\}.
    \]

    Note that the equality holds since
    \[
        \lim_{N\to\infty}g_{[N],N}((M,\dots,M);(m,\dots,m)) = (\e^{-\varepsilon}+(1-\e^{-\varepsilon})M)^{-1}\e^{1-\frac{1-M}{1-m}}    
    \]
    and 
    \[
        \lim_{N\to\infty}g_{\varnothing,N}((m,\dots,m);(M,\dots,M)) = (\e^{-\varepsilon}+(1-\e^{-\varepsilon})m)^{-1}\e^{1-\frac{m}{M}}.    
    \]

    Neither of the terms is superfluous since they can both be the largest value for different choices of $\varepsilon$, $m$, and $M$.
\end{proof}

\begin{remark}[Numerical computation]\label{remark:computationInverse}
    We resort to a numerical computation to check that
    \begin{multline*}
        \max\!\bigg\{\sup_{N\in\N}\bigg(\e^{-\varepsilon}+(1-\e^{-\varepsilon})\frac{m+NM}{N+1}\bigg)^{-1}\bigg(\frac{N+1}{N+\frac{1-M}{1-m}}\bigg)^{N},(\e^{-\varepsilon}+(1-\e^{-\varepsilon})m)^{-1}\e^{1-\frac{m}{M}}\bigg\} \\
        =\max\!\big\{(\e^{-\varepsilon}+(1-\e^{-\varepsilon})M)^{-1}\e^{1-\frac{1-M}{1-m}},(\e^{-\varepsilon}+(1-\e^{-\varepsilon})m)^{-1}\e^{1-\frac{m}{M}}\big\},
    \end{multline*}
    which we use in \Cref{th:boundinverse}. The script is implemented in Python and checks that both values are the same up to an error of $2\cdot10^{-7}$. The maximum over $N$ is found using the \texttt{find\_local\_maximum} checking up to $N=10^9$. We deem this value to be enough since the term converges quickly when $N\to\infty$. In addition, we also verify that 
    \[
        -\ln(\e^{-\varepsilon}+(1-\e^{-\varepsilon})m) + 1-\frac{m}{M} \leq \max\{l_1(p),l_2(p),l_3\}
    \]
    as used at the end of \Cref{th:outlierscoresuppression}.
    
    Like \Cref{remark:computation}, since the domains of $\varepsilon\in[0,\infty)$, $m\in(0,1)$, and $M\in[m,1)$ are continuous and the domain of $\varepsilon$ is furthermore unbounded, we cannot run exhaustively for all values of these parameters. For this work, we decide to run the experiment for every reasonable value with decent level of granularity: $m$ and $M$ are run for every value between $0.01$ and $0.99$ with step $0.01$, and $\varepsilon$ is run for every value between $0$ and $2$ with step $0.01$, every value between $2$ and $10$ with step $0.1$ and every value between $10$ and $100$ with step $1$. We believe that our choice of values for $\varepsilon$ is representative of the values of $\varepsilon$ deemed to be acceptable in the literature (small values smaller than $2$ or $10$~\cite{dwork2014algorithmic}). 

    Our experimentation concludes that for the selected parameters $\varepsilon$, $m$, and $M$, the numerical and hypothesized values are less than $2\cdot10^{-7}$ in difference and that the hypothesized value is always larger than the numerical one to account for the real maximum being when $N\to\infty$.
\end{remark}

\end{document}

%% file: plotsgallery/plotsPoisson.tex
\subsection{Plots of the Uniform Poisson Sampling for the Mean Computation}\label{sec:plots:PoissonSampling1}

\begin{figure}[H]
    \centering
    \includegraphics[width=0.4\textwidth]{PaperPlots/Adult/age/age_uniform_Poisson_sampling_laplace_MPE+SD.pdf}%
    \includegraphics[width=0.4\textwidth]{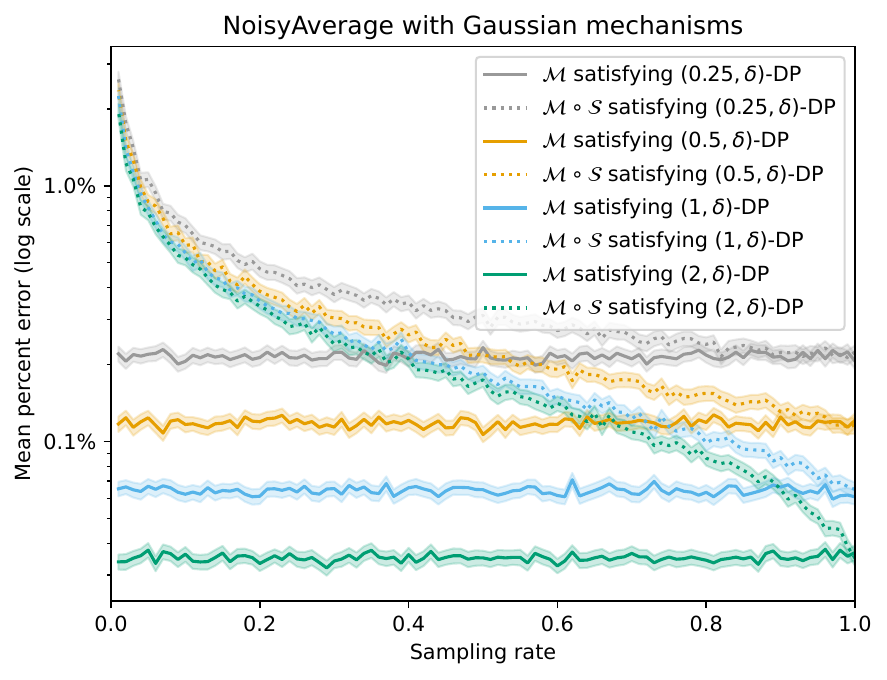}%
    \caption{Plots of the utility values of $\M$ and $\M\circ\S$ for the uniform Poisson sampling over the \texttt{age} column in the Adult database. The shaded areas correspond to a 95\% confidence interval for the mean of the utility metric.}
    \label{fig:uniformPoissonsampling-Adult-age}
\end{figure}

\begin{figure}[H]
    \centering
    \includegraphics[width=0.4\textwidth]{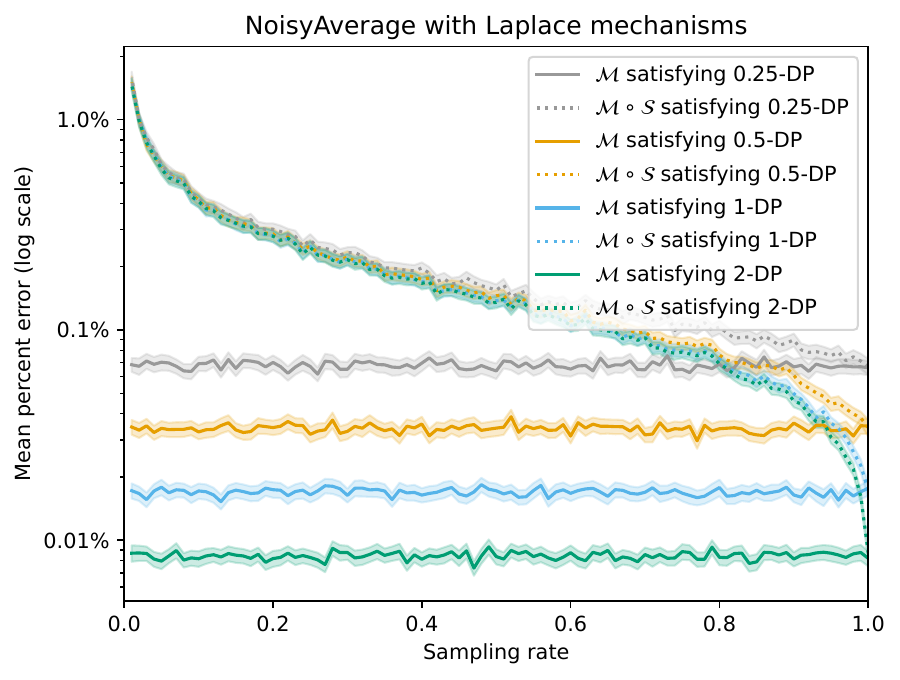}%
    \includegraphics[width=0.4\textwidth]{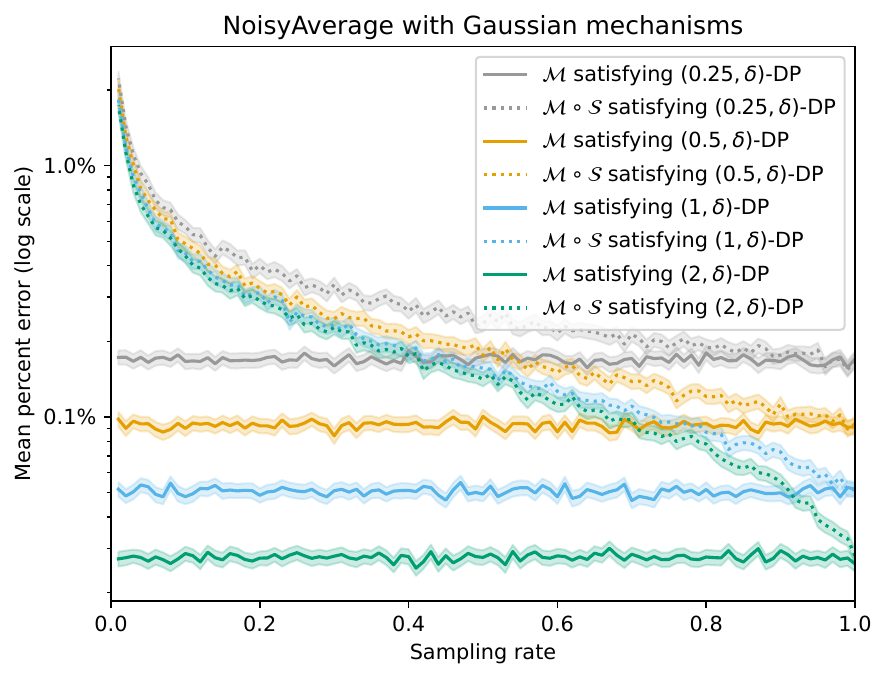}%
    \caption{Plots of the utility values of $\M$ and $\M\circ\S$ for the uniform Poisson sampling over the \texttt{hours-per-week} column in the Adult database. The shaded areas correspond to a 95\% confidence interval for the mean of the utility metric.}
    \label{fig:uniformPoissonsampling-Adult-hours-per-week}
\end{figure}

\begin{figure}[H]
    \centering
    \includegraphics[width=0.4\textwidth]{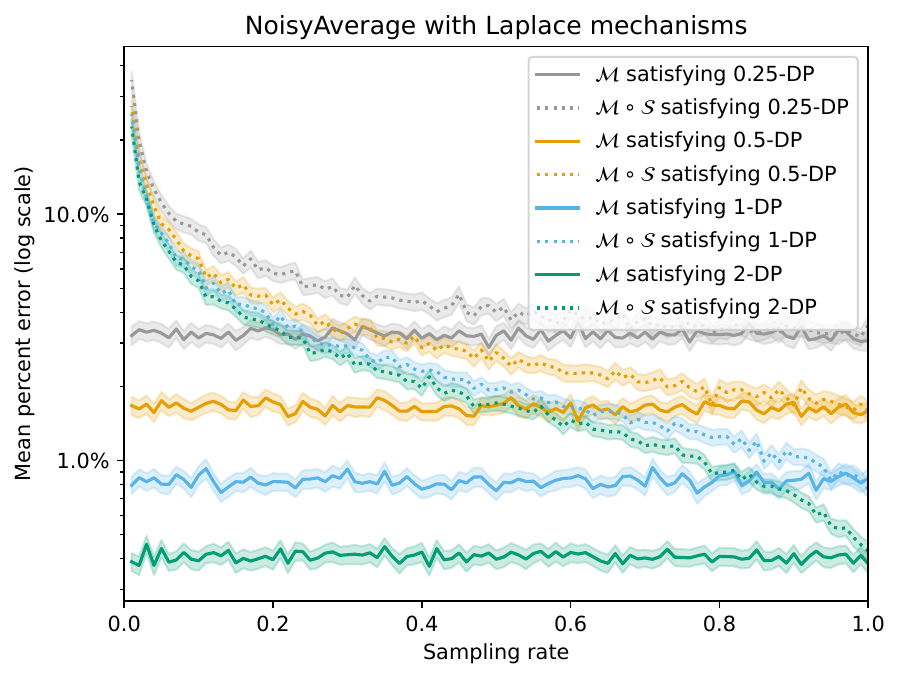}%
    \includegraphics[width=0.4\textwidth]{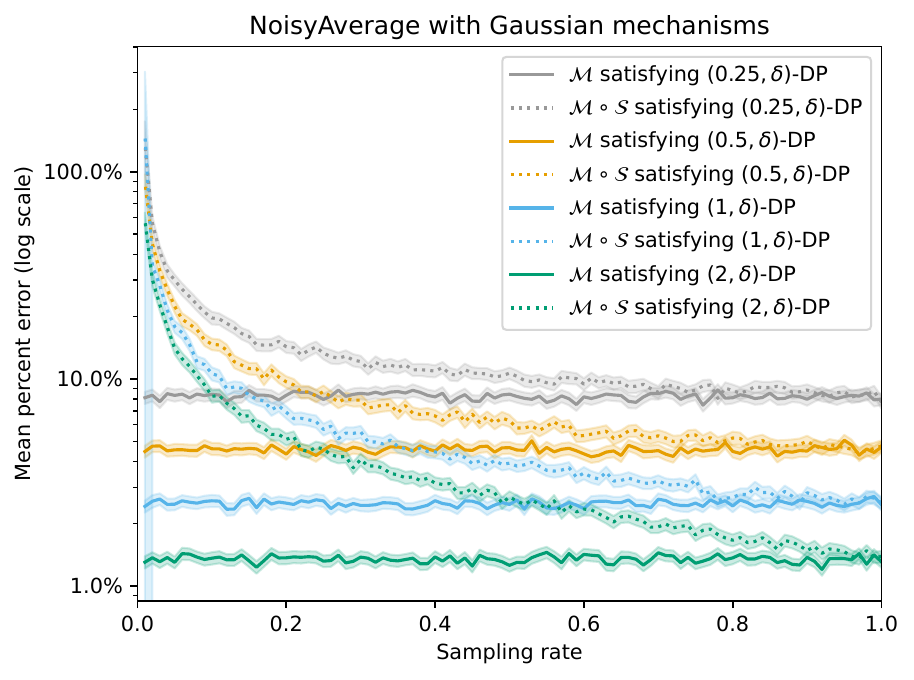}%
    \caption{Plots of the utility values of $\M$ and $\M\circ\S$ for the uniform Poisson sampling over the \texttt{FEDTAX} column in the Census database. The shaded areas correspond to a 95\% confidence interval for the mean of the utility metric.}
    \label{fig:uniformPoissonsampling-Census-FEDTAX}
\end{figure}

\begin{figure}[H]
    \centering
    \includegraphics[width=0.4\textwidth]{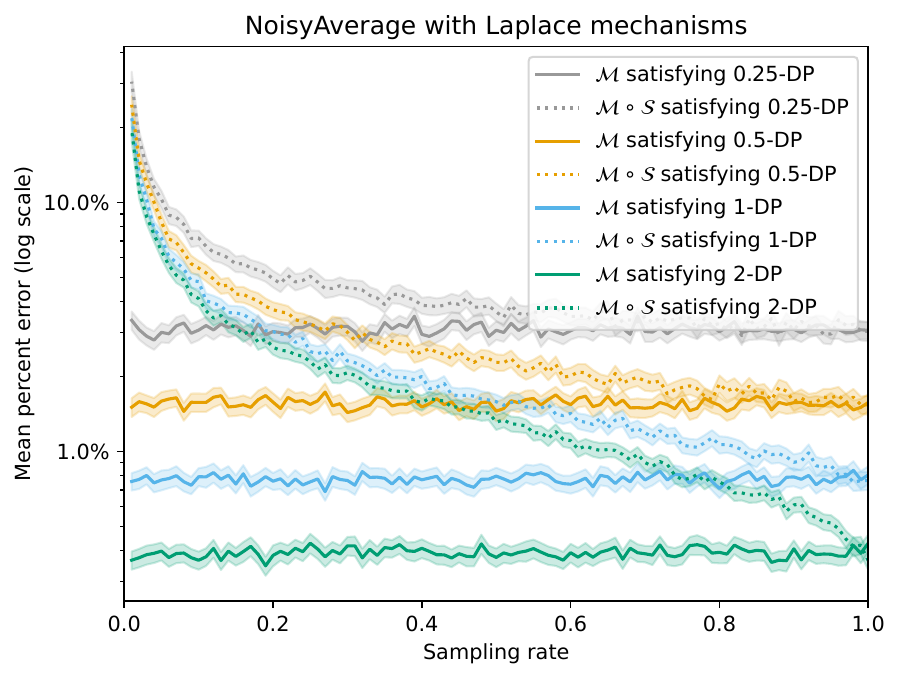}%
    \includegraphics[width=0.4\textwidth]{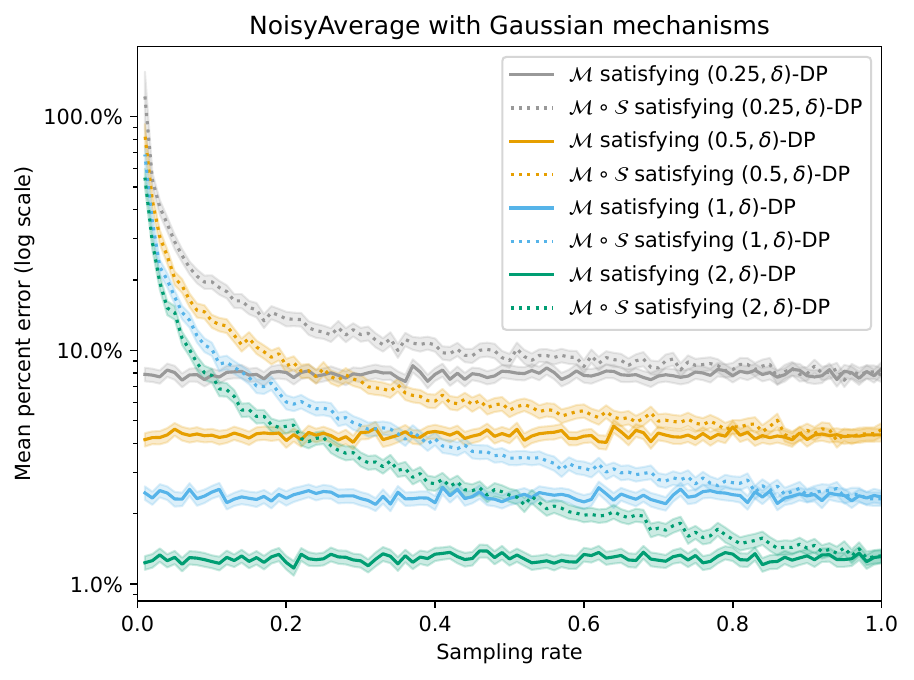}%
    \caption{Plots of the utility values of $\M$ and $\M\circ\S$ for the uniform Poisson sampling over the \texttt{FICA} column in the Census database. The shaded areas correspond to a 95\% confidence interval for the mean of the utility metric.}
    \label{fig:uniformPoissonsampling-Census-FICA}
\end{figure}

\begin{figure}[H]
    \centering
    \includegraphics[width=0.4\textwidth]{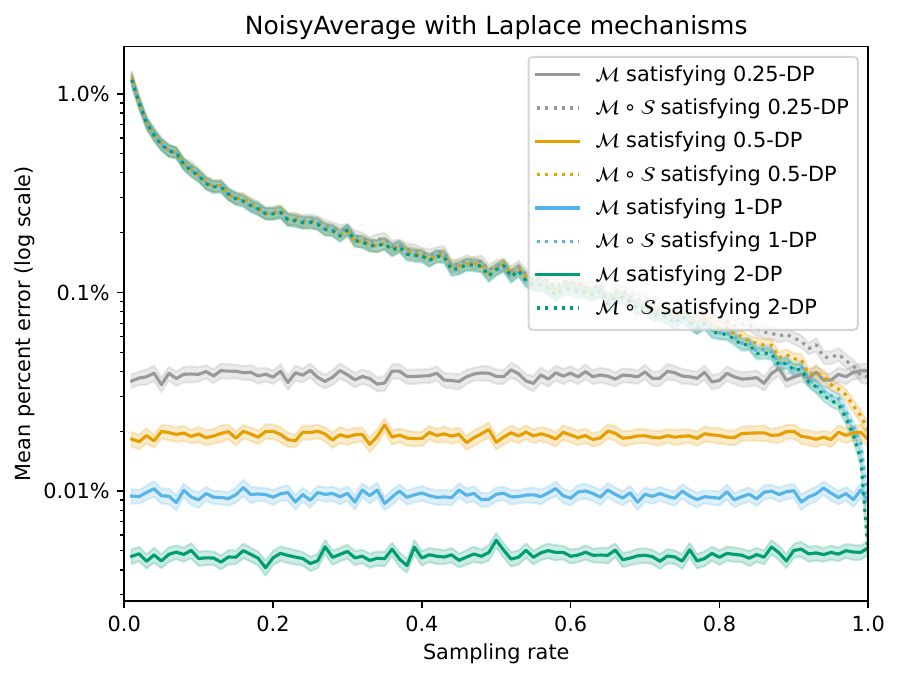}%
    \includegraphics[width=0.4\textwidth]{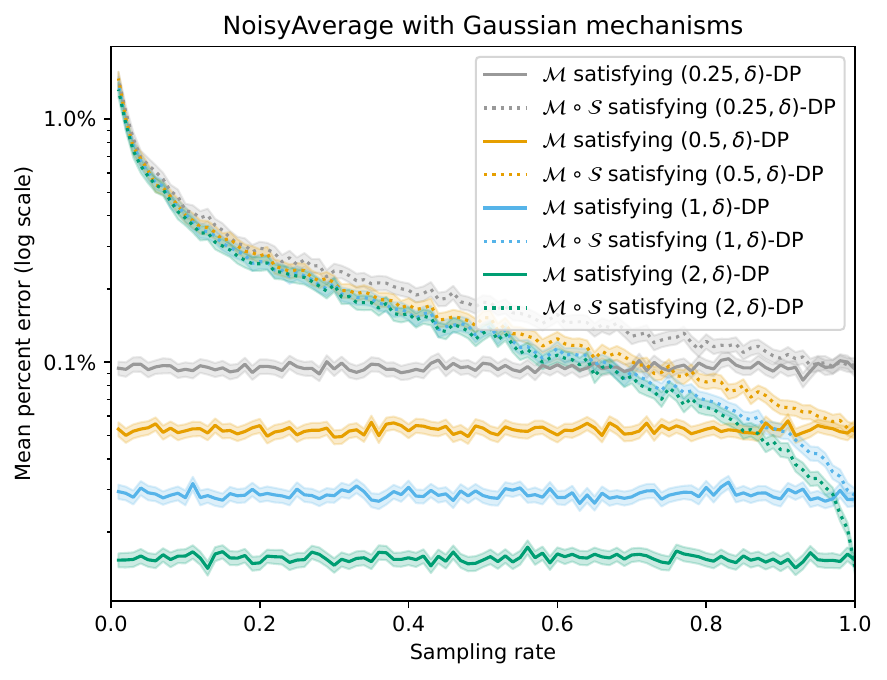}%
    \caption{Plots of the utility values of $\M$ and $\M\circ\S$ for the uniform Poisson sampling over the \texttt{Age} column in the Irish database. The shaded areas correspond to a 95\% confidence interval for the mean of the utility metric.}
    \label{fig:uniformPoissonsampling-Irishn-Age}
\end{figure}

\begin{figure}[H]
    \centering
    \includegraphics[width=0.4\textwidth]{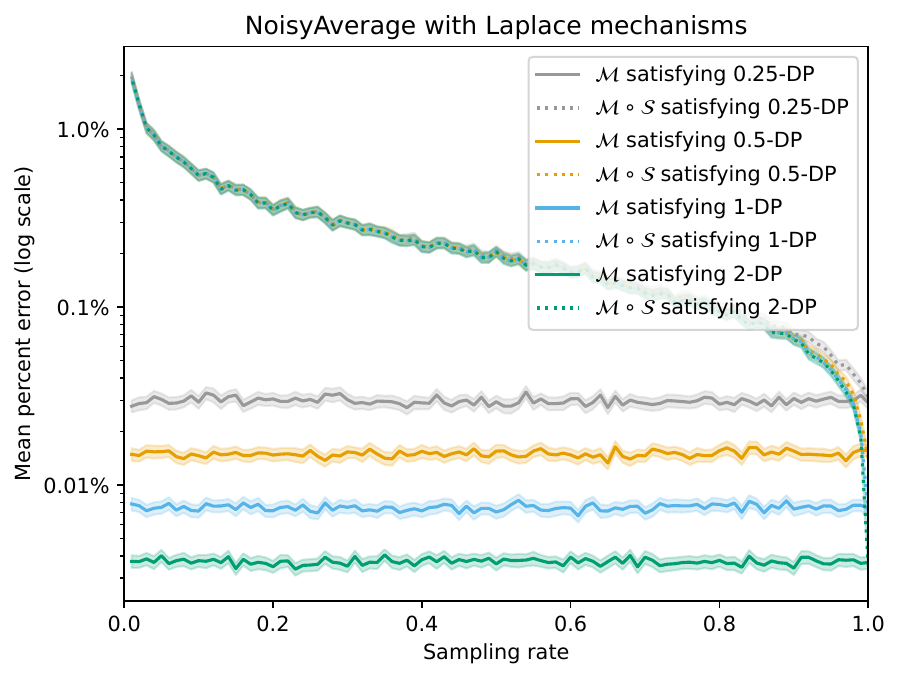}%
    \includegraphics[width=0.4\textwidth]{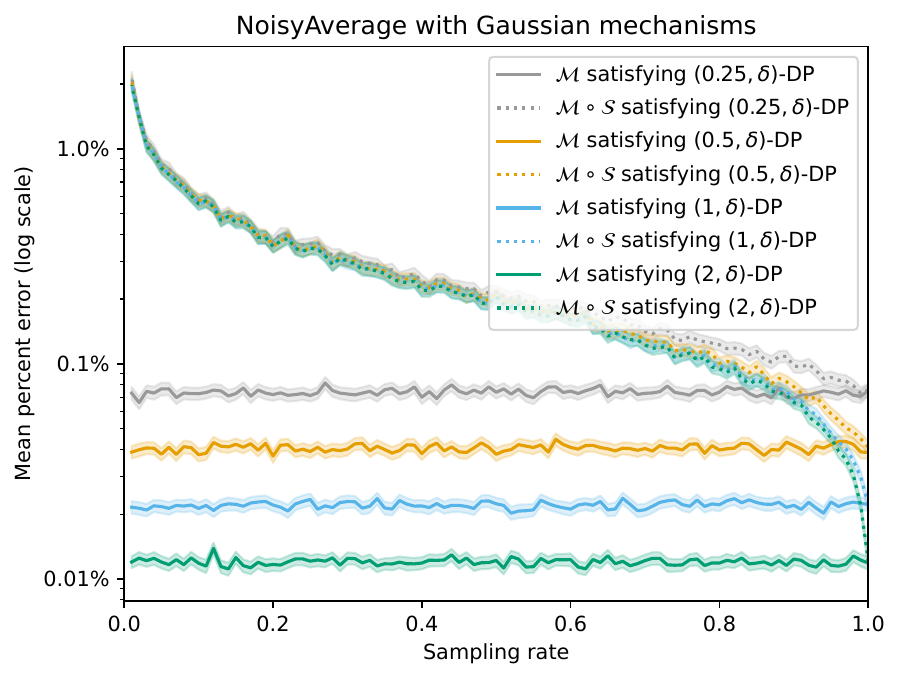}%
    \caption{Plots of the utility values of $\M$ and $\M\circ\S$ for the uniform Poisson sampling over the \texttt{HighestEducationCompleted} column in the Irish database. The shaded areas correspond to a 95\% confidence interval for the mean of the utility metric.}
    \label{fig:uniformPoissonsampling-Irishn-HighestEducationCompleted}
\end{figure}

\subsection{Plots of the Uniform Poisson Sampling for the Mode Computation}\label{sec:plots:PoissonSampling2}
\begin{figure}[H]   
    \centering
    \includegraphics[width=0.4\textwidth]{PaperPlots/RNM-Adult/age/age_uniform_Poisson_sampling_laplace_EmpProb+SD.pdf}%
    \includegraphics[width=0.4\textwidth]{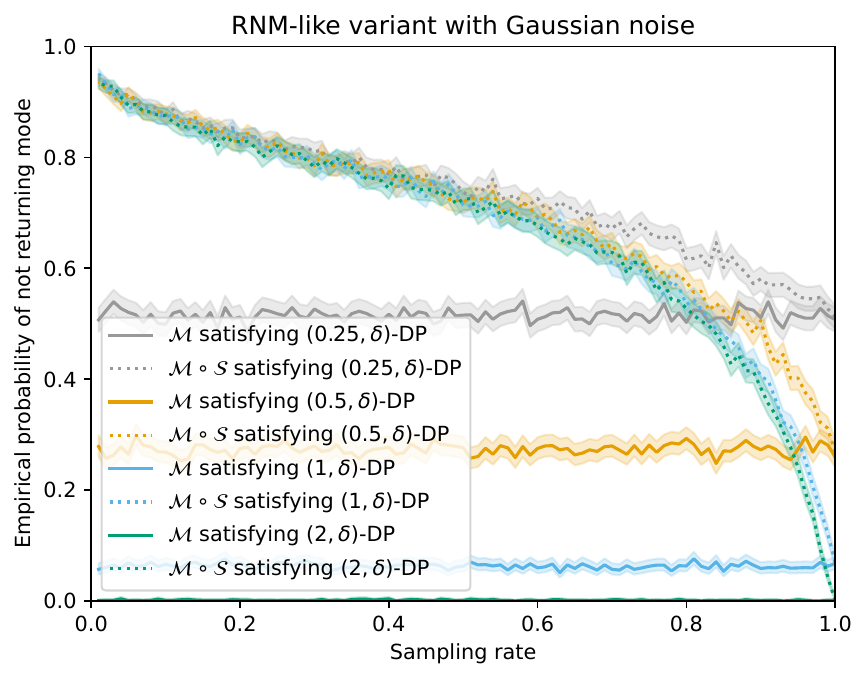}%
    
    \includegraphics[width=0.4\textwidth]{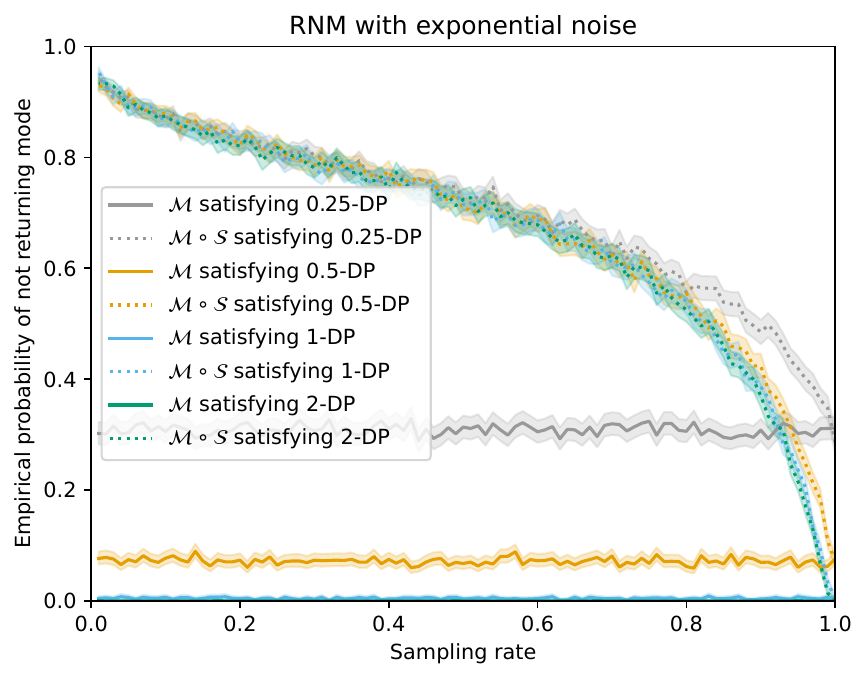}%
    \includegraphics[width=0.4\textwidth]{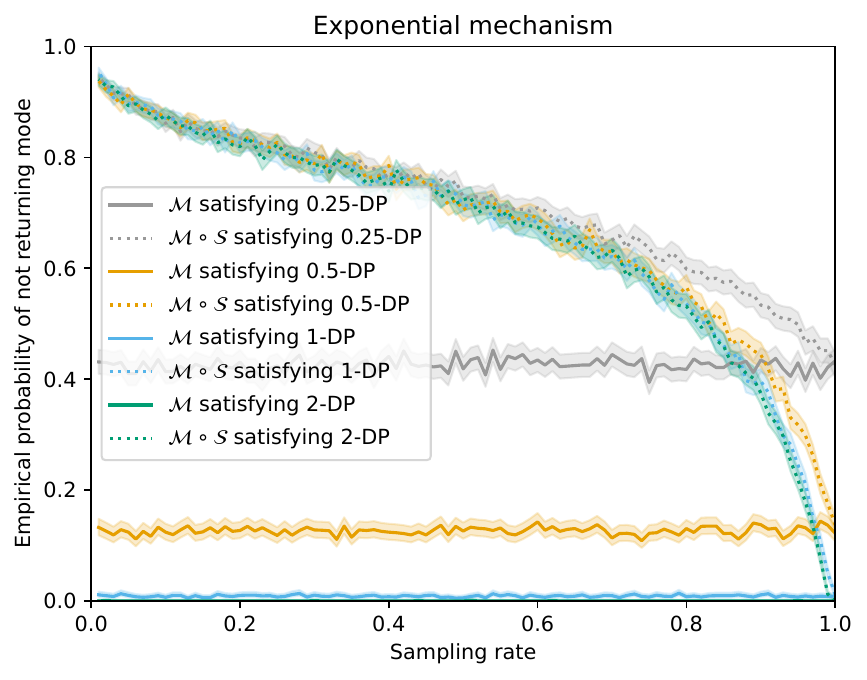}%
    \caption{Plots of the utility values of $\M$ and $\M\circ\S$ for the uniform Poisson sampling over the \texttt{age} column in the Adult database. The shaded areas correspond to a 95\% Wilson confidence interval for the mean of the utility metric.}
    \label{fig:uniformPoissonsampling-RNM-Adult-age}
\end{figure}

\begin{figure}[H]
    \centering
    \includegraphics[width=0.4\textwidth]{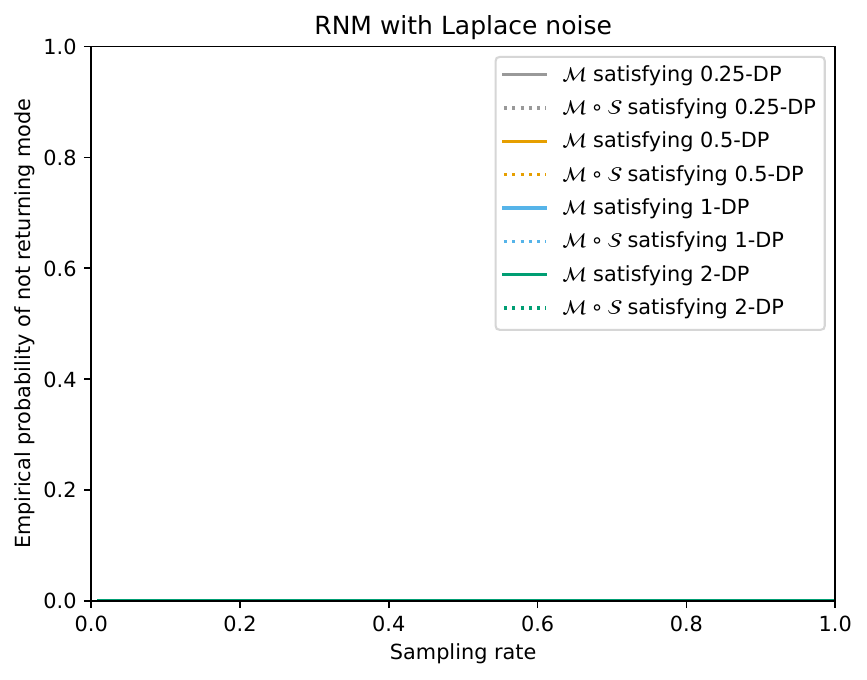}%
    \includegraphics[width=0.4\textwidth]{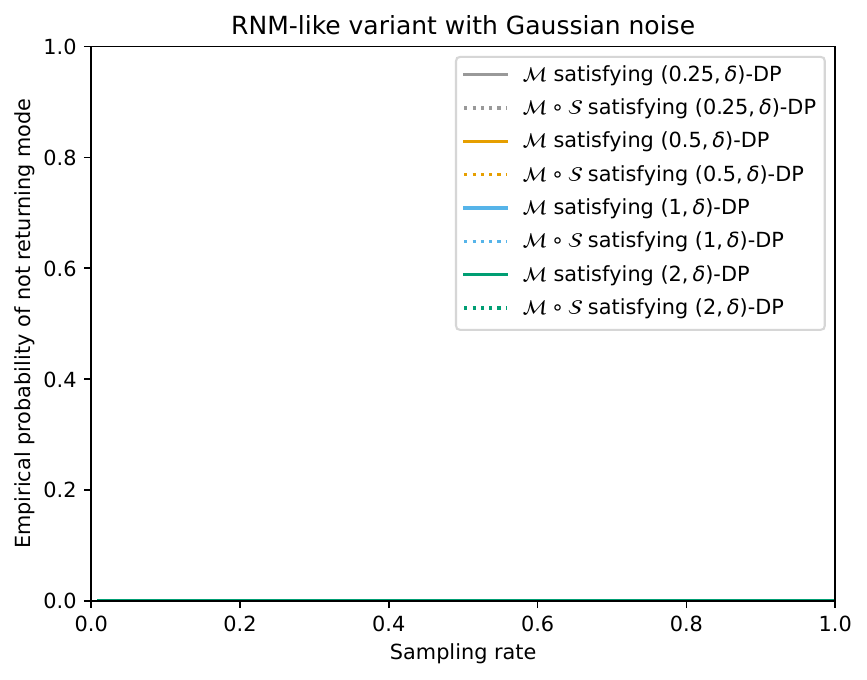}%
    
    \includegraphics[width=0.4\textwidth]{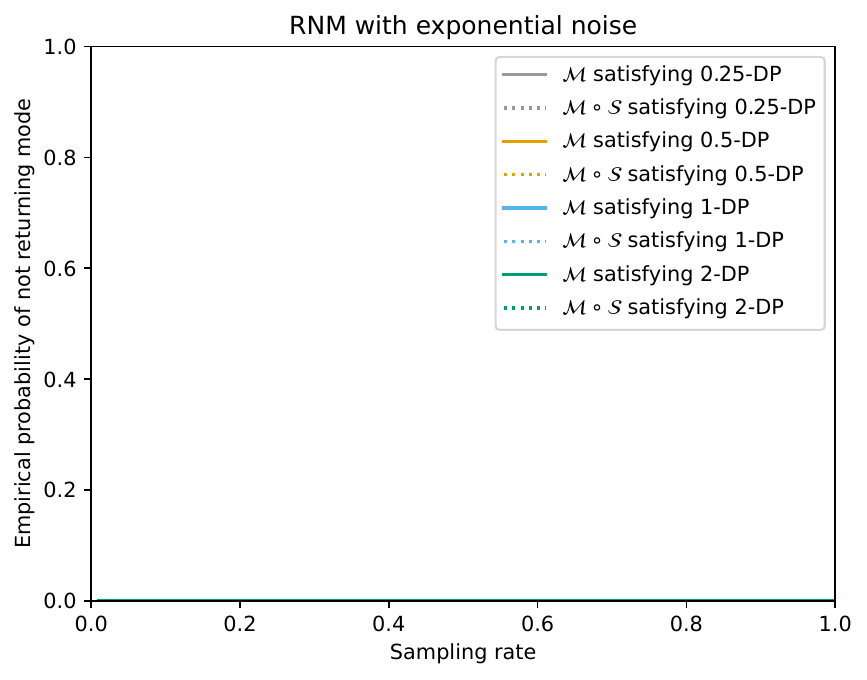}%
    \includegraphics[width=0.4\textwidth]{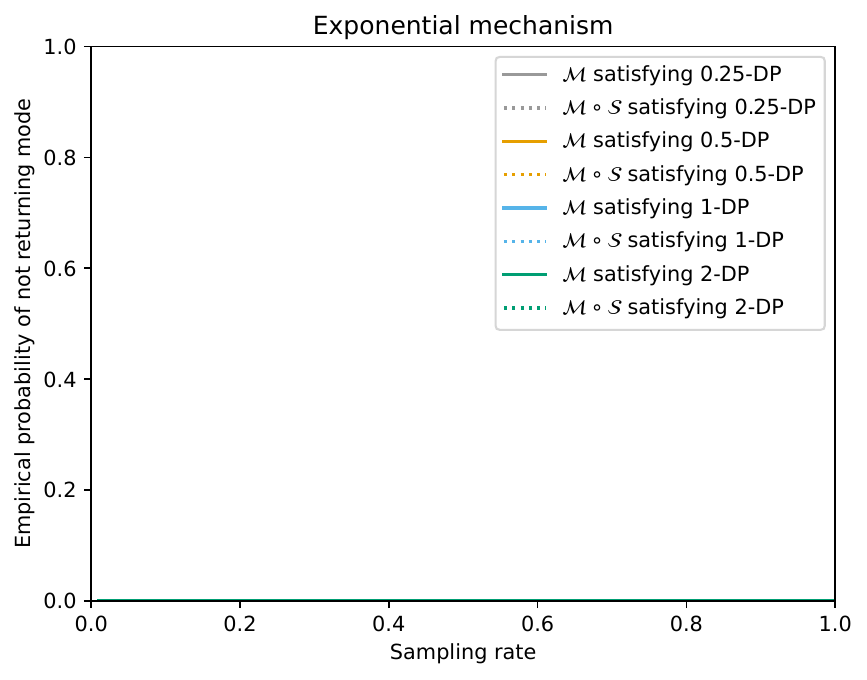}%
    \caption{Plots of the utility values of $\M$ and $\M\circ\S$ for the uniform Poisson sampling over the \texttt{hours-per-week} column in the Adult database. The shaded areas correspond to a 95\% Wilson confidence interval for the mean of the utility metric.}
    \label{fig:uniformPoissonsampling-RNM-Adult-hours-per-week}
\end{figure}

\begin{figure}[H]
    \centering
    \includegraphics[width=0.4\textwidth]{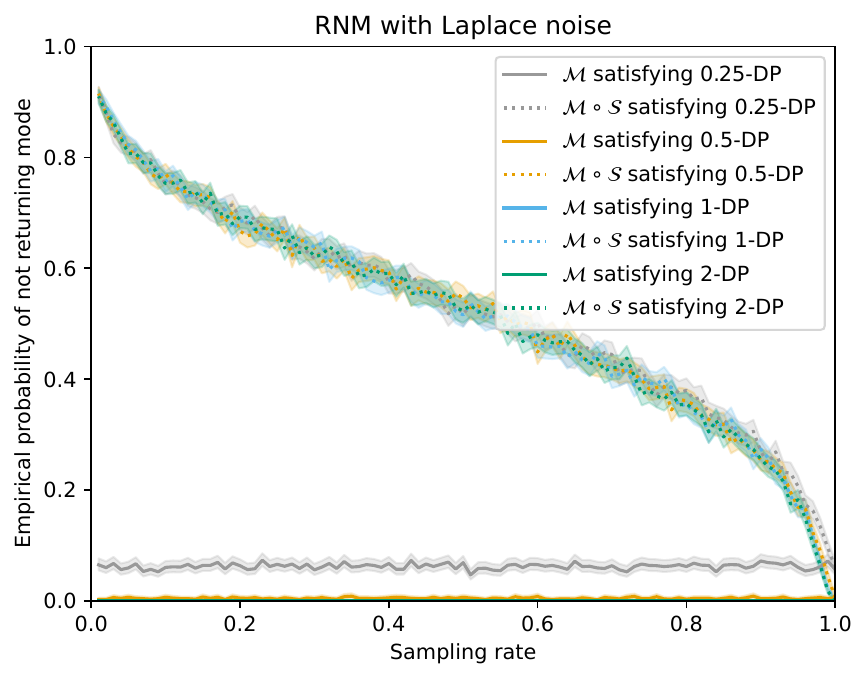}%
    \includegraphics[width=0.4\textwidth]{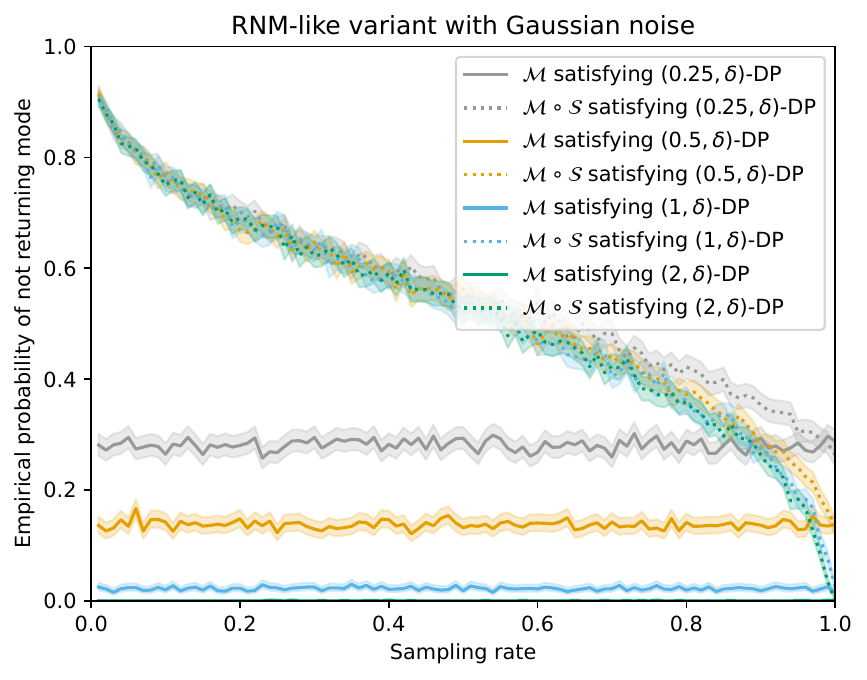}%
    
    \includegraphics[width=0.4\textwidth]{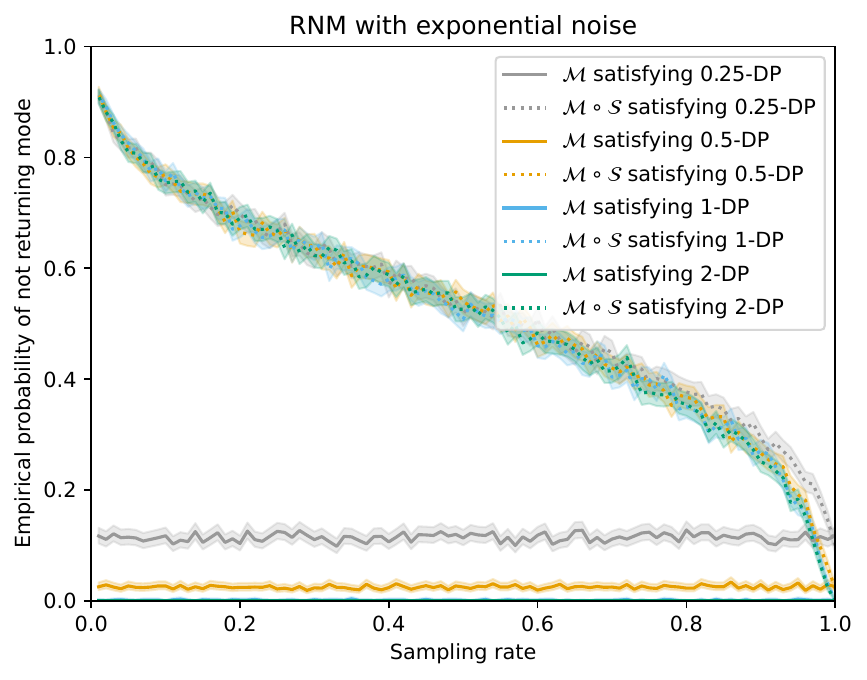}%
    \includegraphics[width=0.4\textwidth]{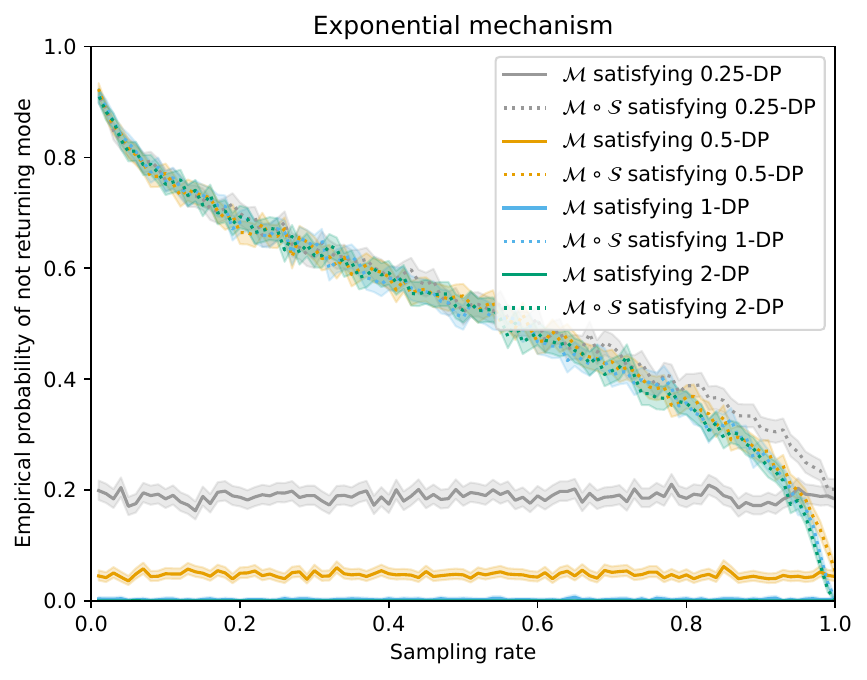}%
    \caption{Plots of the utility values of $\M$ and $\M\circ\S$ for the uniform Poisson sampling over the \texttt{Age} column in the Irish database. The shaded areas correspond to a 95\% Wilson confidence interval for the mean of the utility metric.}
    \label{fig:uniformPoissonsampling-RNM-Irishn-Age}
\end{figure}

\begin{figure}[H]
    \centering
    \includegraphics[width=0.4\textwidth]{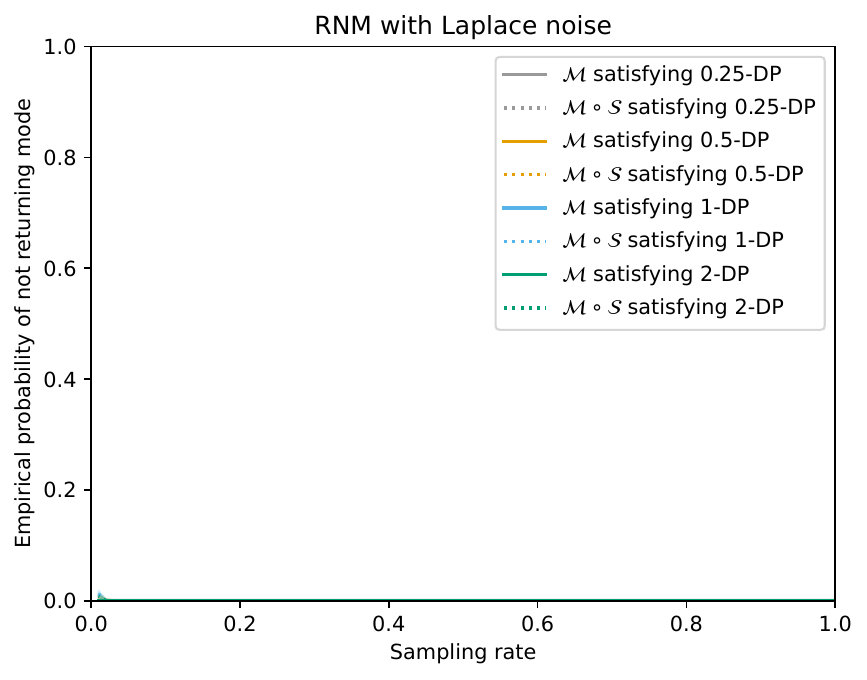}%
    \includegraphics[width=0.4\textwidth]{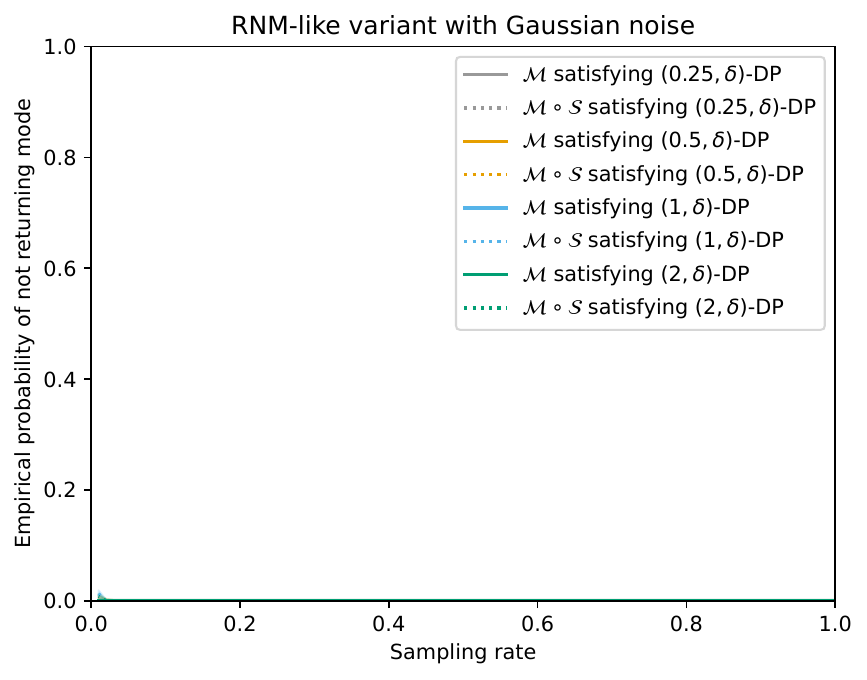}%
    
    \includegraphics[width=0.4\textwidth]{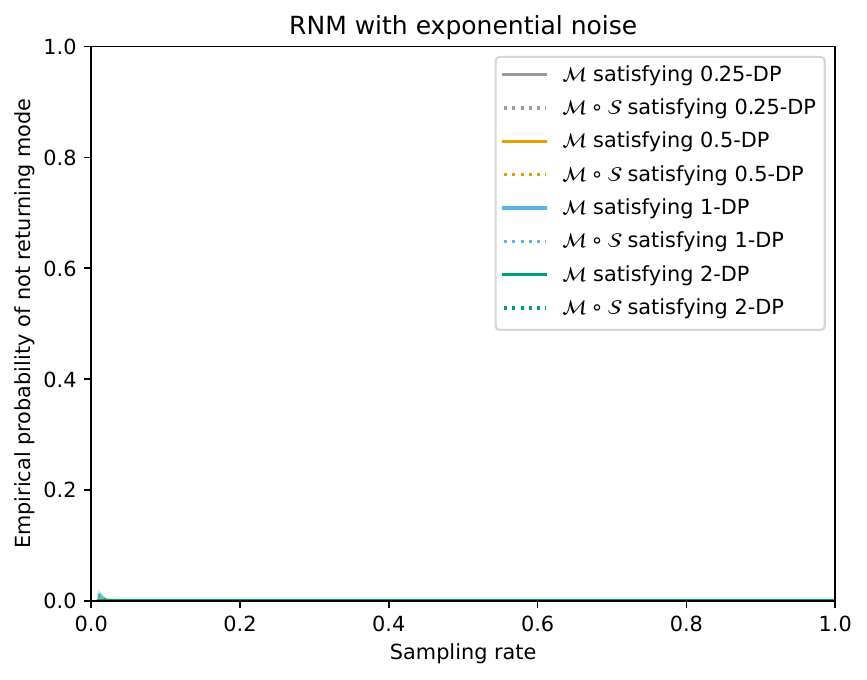}%
    \includegraphics[width=0.4\textwidth]{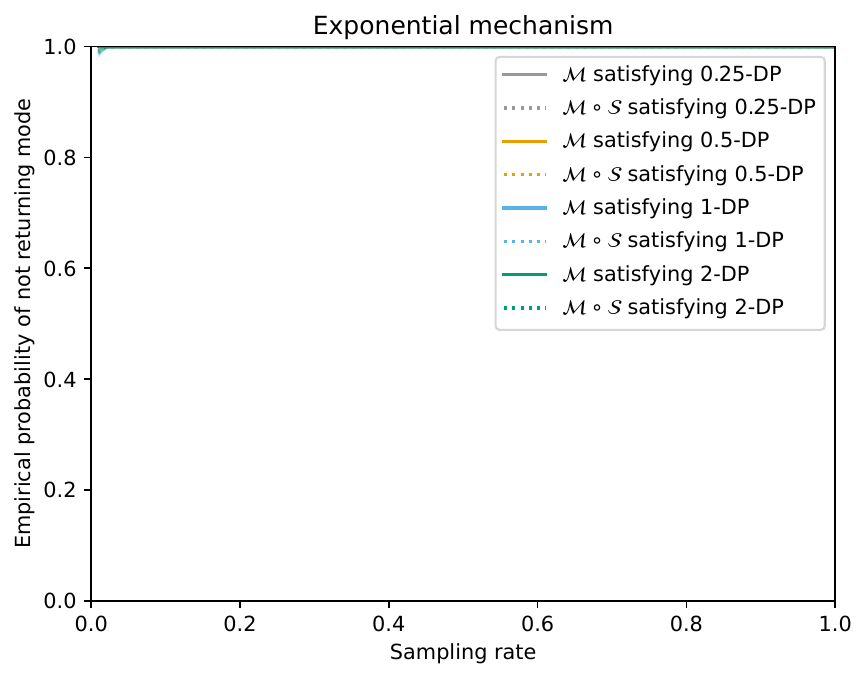}%
    \caption{Plots of the utility values of $\M$ and $\M\circ\S$ for the uniform Poisson sampling over the \texttt{HighestEducationCompleted} column in the Irish database. The shaded areas correspond to a 95\% Wilson confidence interval for the mean of the utility metric.}
    \label{fig:uniformPoissonsampling-RNM-Irishn-HighestEducationCompleted}
\end{figure}

\subsection{Plots of the Uniform Poisson Sampling for the Clustering Mechanisms}\label{sec:plots:PoissonSampling3}
\begin{figure}[H]
    \centering
    \includegraphics[width=0.4\textwidth]{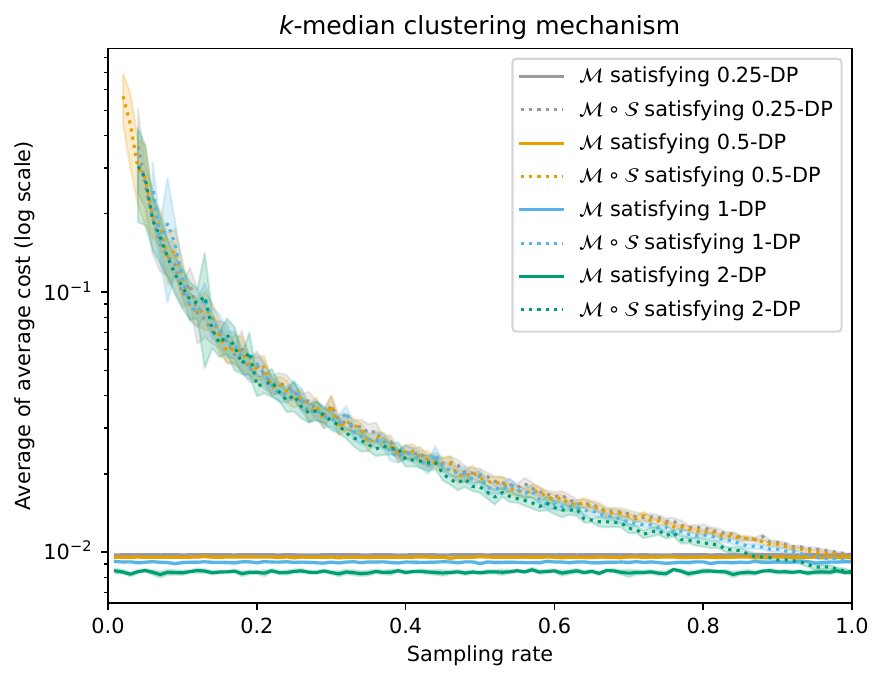}%
    \includegraphics[width=0.4\textwidth]{PaperPlots/Clustering/age_fnlwgt_education-num_capital-gain_capital-loss_hours-per-week_uniform_Poisson_sampling_Average+SD.pdf}%
    \caption{Plots of the utility values of $\M$ and $\M\circ\S$ for the uniform Poisson sampling for (left) $k$-median over our synthetic database and (right) DPLloyd over the six numerical columns in the Adult database. The shaded areas correspond to a 95\% confidence interval for the mean of the utility metric.}
    \label{fig:uniformPoissonsampling-Clustering}
\end{figure}

%% file: plotsgallery/plotsExperiment2.tex
\subsection{Plots for the Mean Computation}\label{sec:plots:SuppressionwithEpsDeltaChange1}

\begin{figure}[H]
    \includegraphics[width=0.25\textwidth]{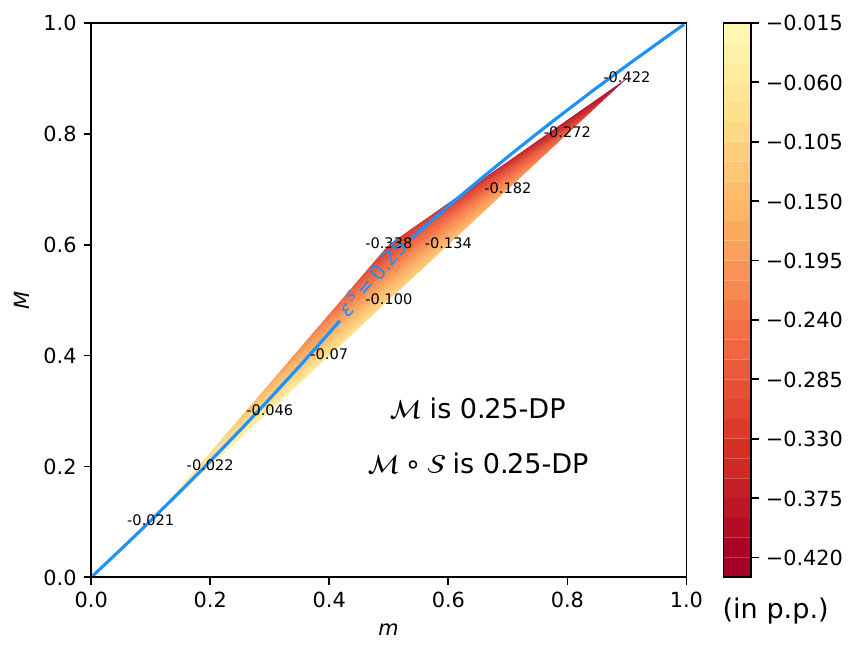}%
    \includegraphics[width=0.25\textwidth]{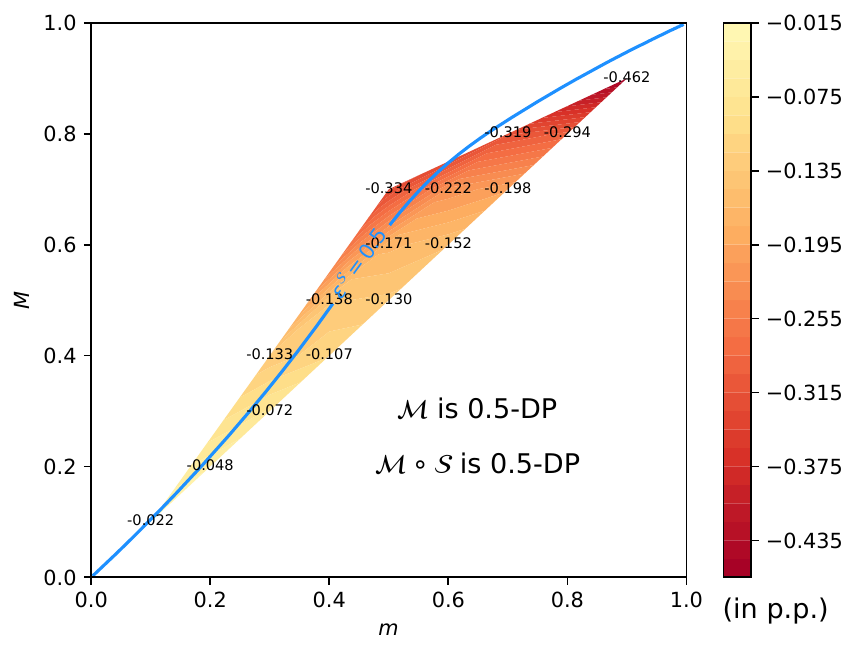}%
    \includegraphics[width=0.25\textwidth]{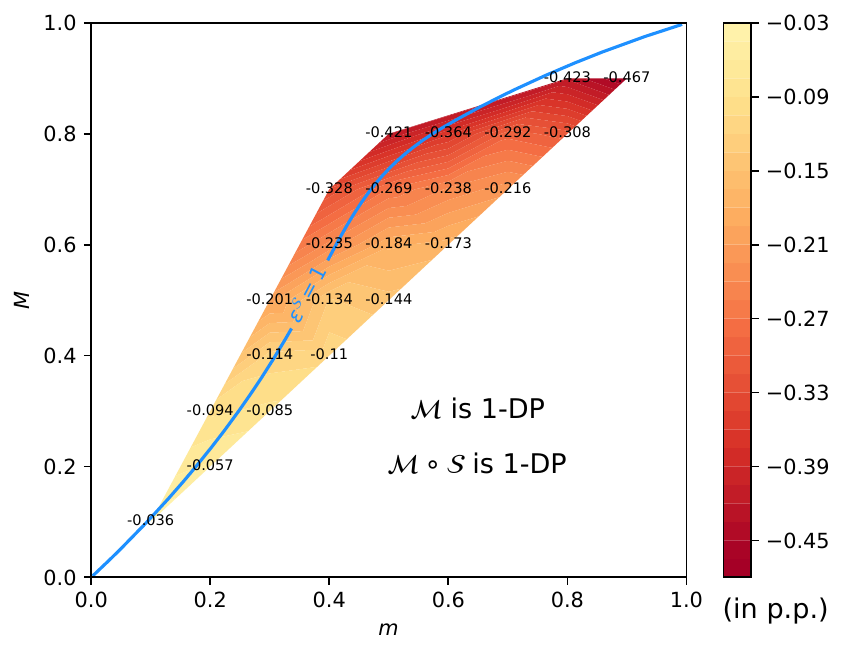}%
    \includegraphics[width=0.25\textwidth]{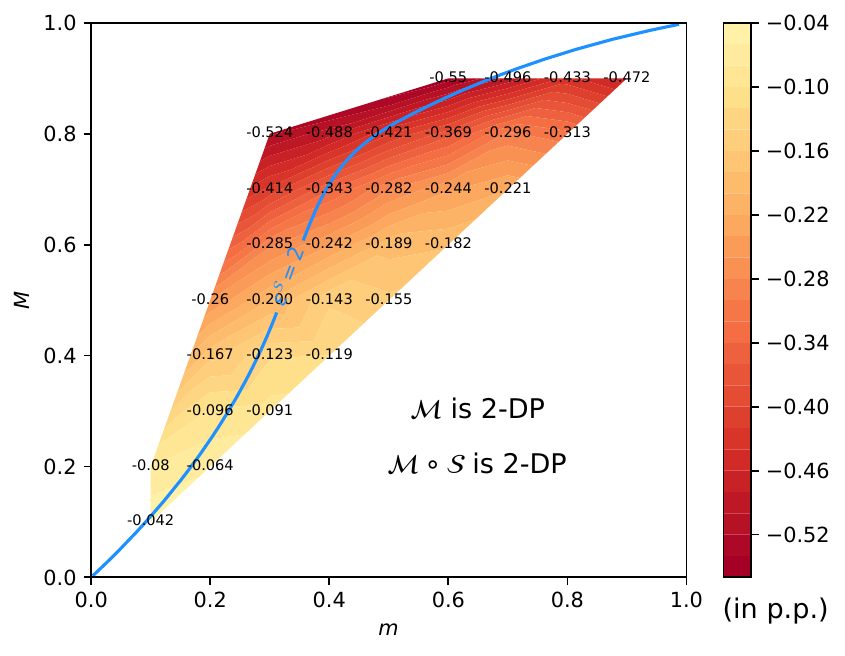}%
    \caption{The mean percent error (MPE) of $\M$ minus that of $\M\circ\S$ at the same privacy levels. Results shown for the NoisyAverage with Laplace mechanisms over the \texttt{age} column in the Adult database.}
    \label{fig:Experiment2-NoisyAverage-Laplace-Adult-age}
\end{figure}

\begin{figure}[H]
    \includegraphics[width=0.25\textwidth]{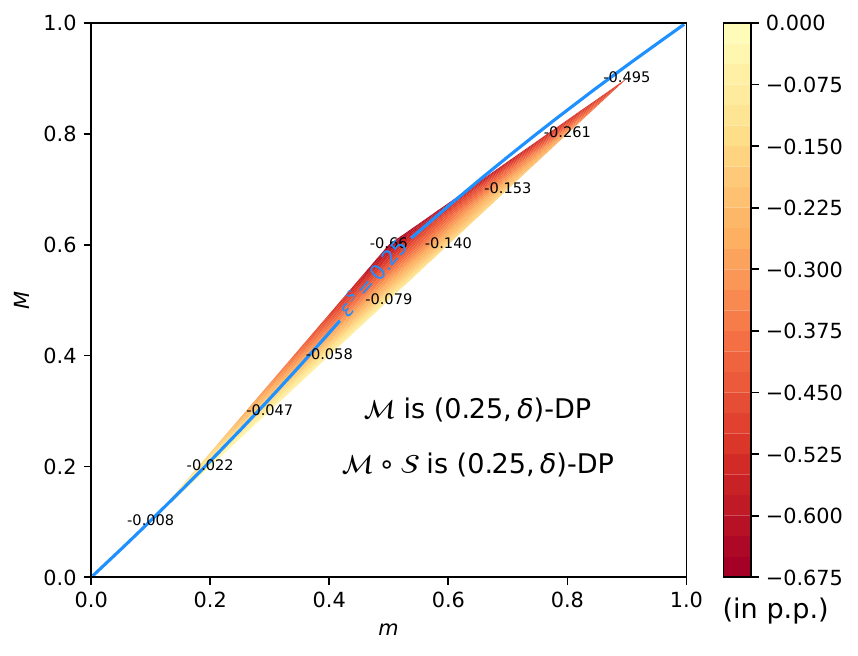}%
    \includegraphics[width=0.25\textwidth]{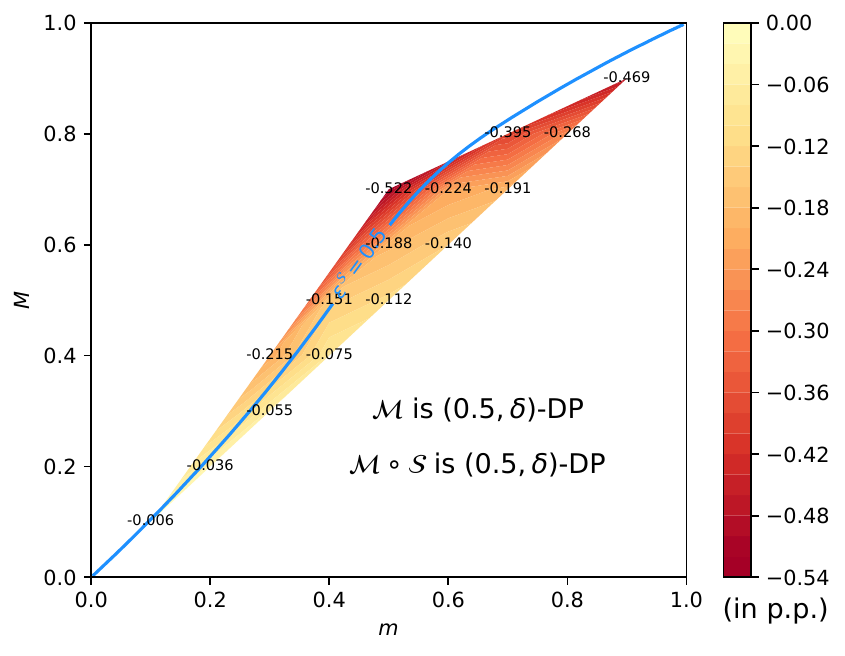}%
    \includegraphics[width=0.25\textwidth]{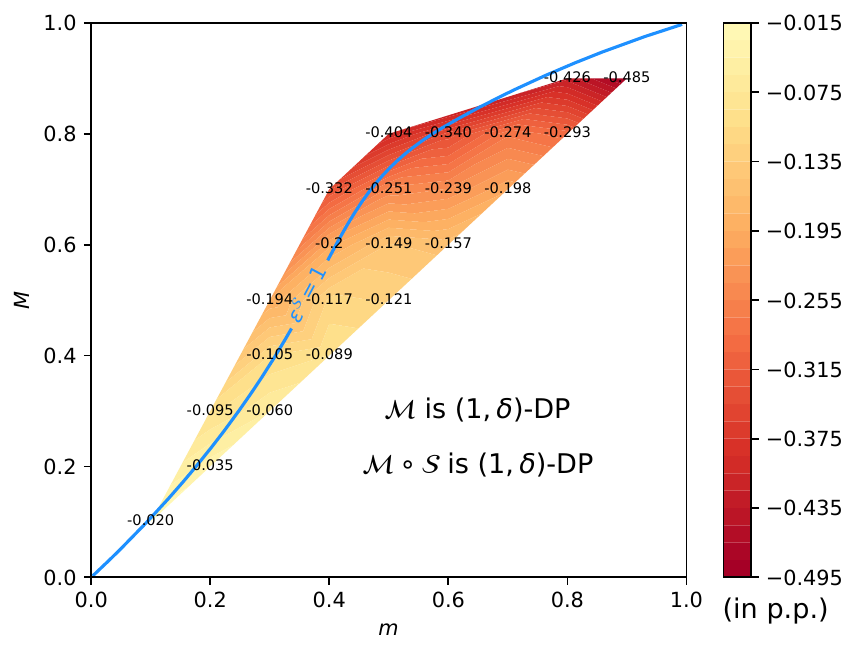}%
    \includegraphics[width=0.25\textwidth]{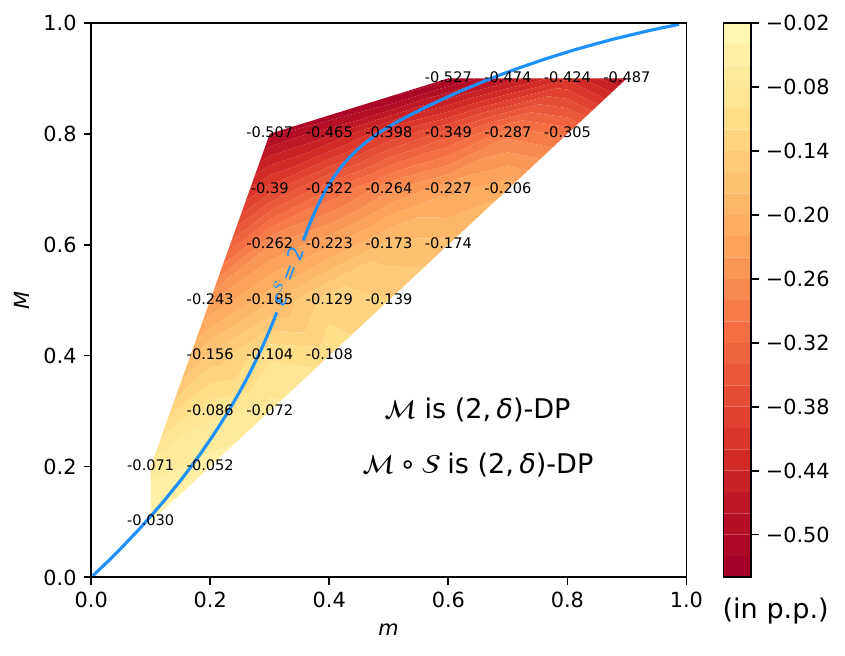}%
    \caption{The mean percent error (MPE) of $\M$ minus that of $\M\circ\S$ at the same privacy levels. Results shown for the NoisyAverage with Gaussian mechanisms over the \texttt{age} column in the Adult database.}
    \label{fig:Experiment2-NoisyAverage-Gaussian-Adult-age}
\end{figure}

\begin{figure}[H]
    \includegraphics[width=0.25\textwidth]{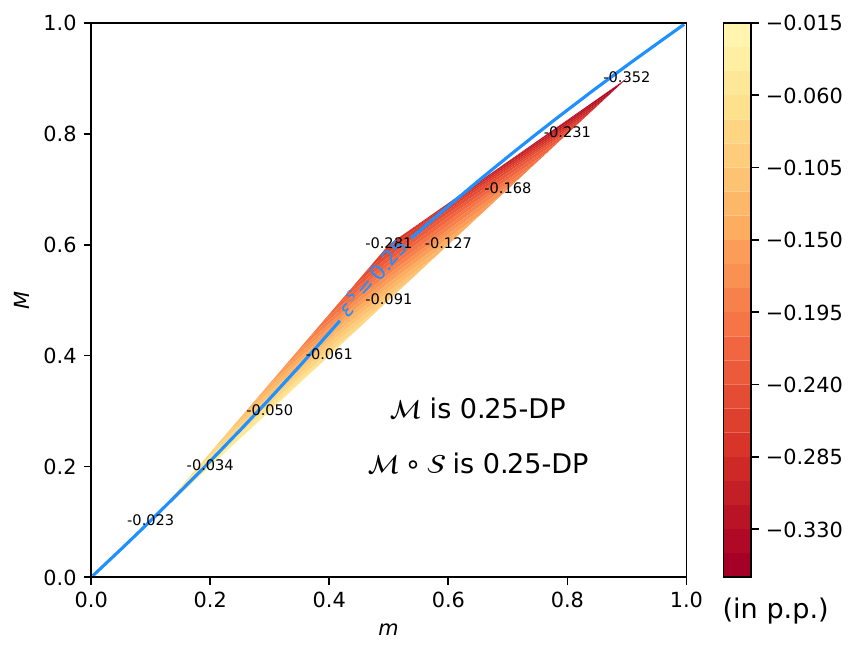}%
    \includegraphics[width=0.25\textwidth]{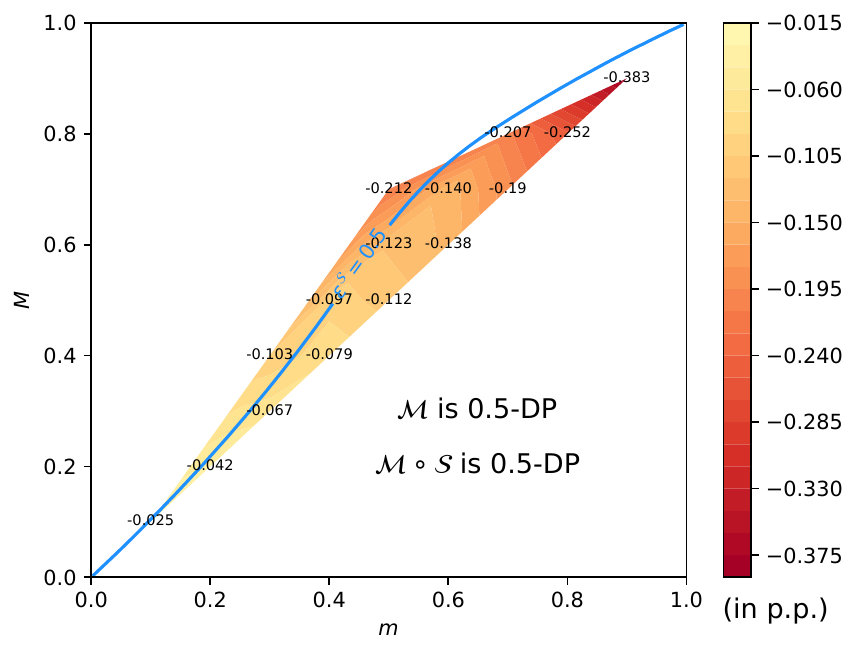}%
    \includegraphics[width=0.25\textwidth]{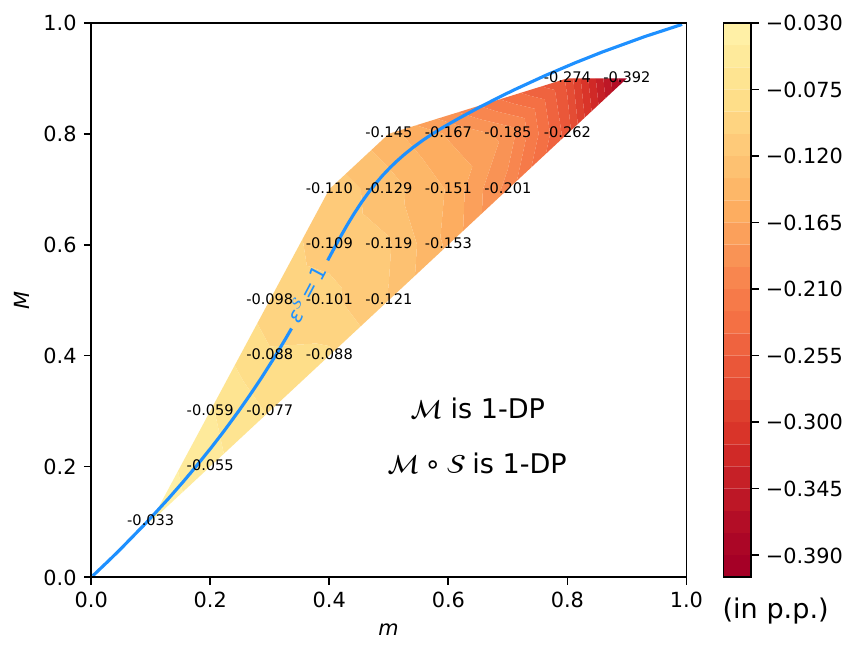}%
    \includegraphics[width=0.25\textwidth]{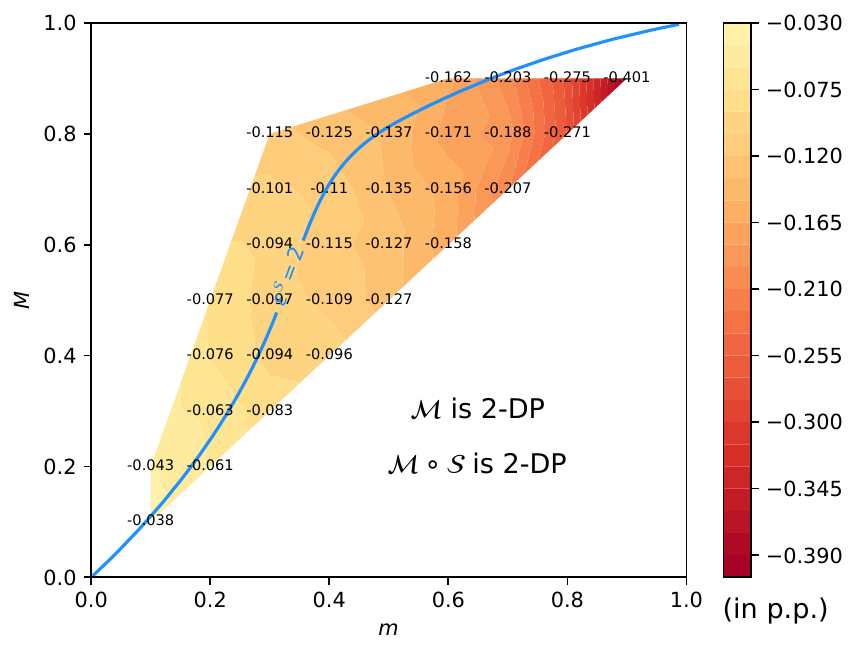}%
    \caption{The mean percent error (MPE) of $\M$ minus that of $\M\circ\S$ at the same privacy levels. Results shown for the NoisyAverage with Laplace mechanisms over the \texttt{hours-per-week} column in the Adult database.}
    \label{fig:Experiment2-NoisyAverage-Laplace-Adult-hours-per-week}
\end{figure}

\begin{figure}[H]
    \includegraphics[width=0.25\textwidth]{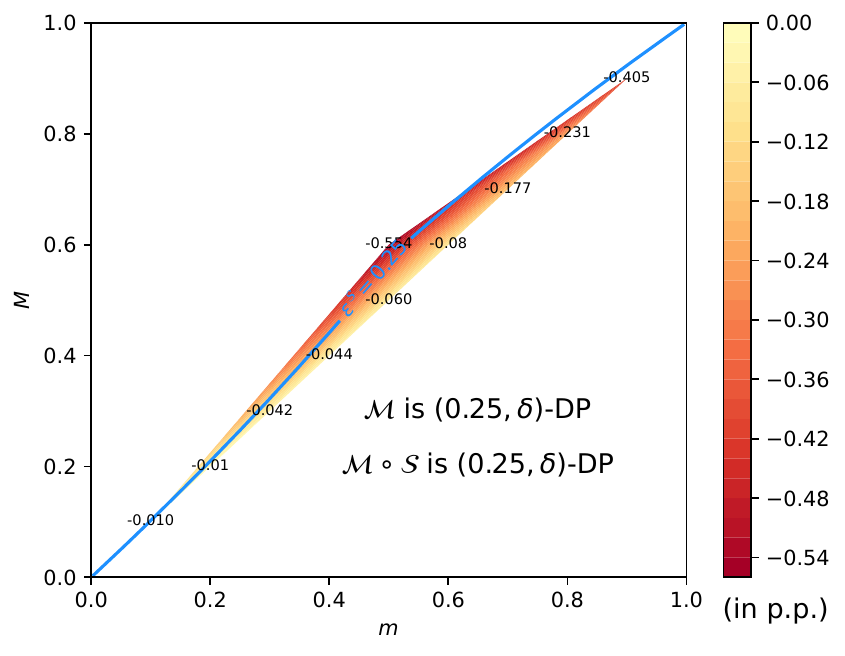}%
    \includegraphics[width=0.25\textwidth]{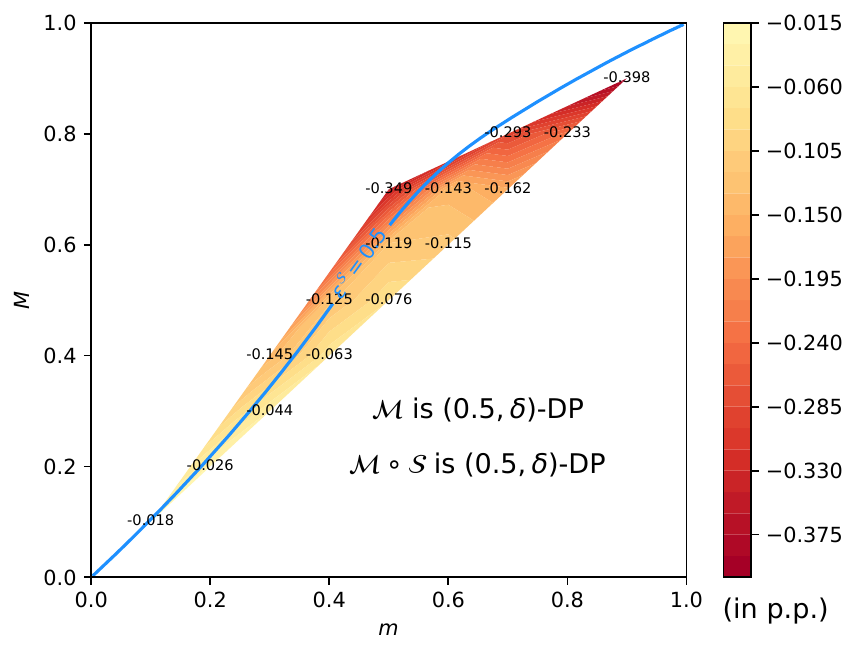}%
    \includegraphics[width=0.25\textwidth]{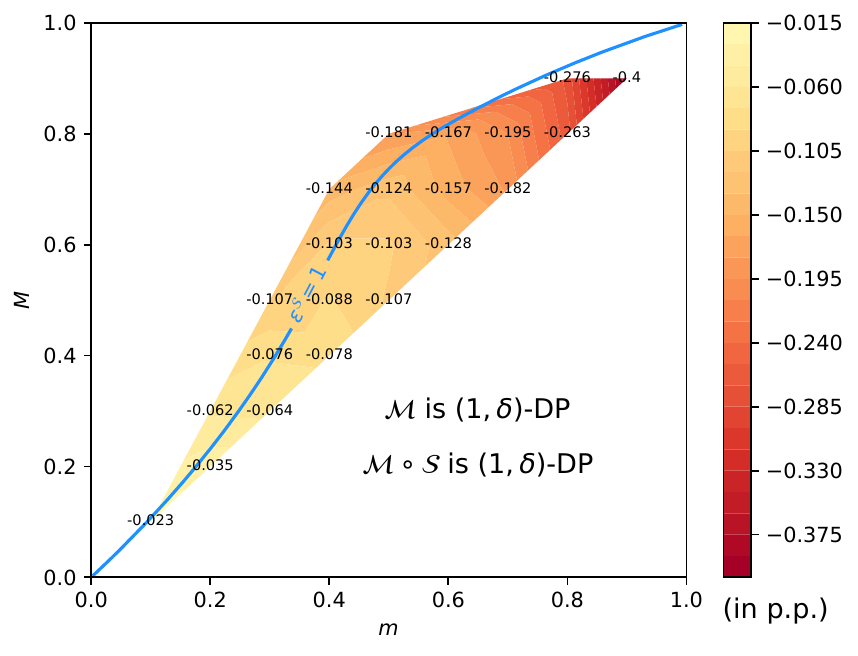}%
    \includegraphics[width=0.25\textwidth]{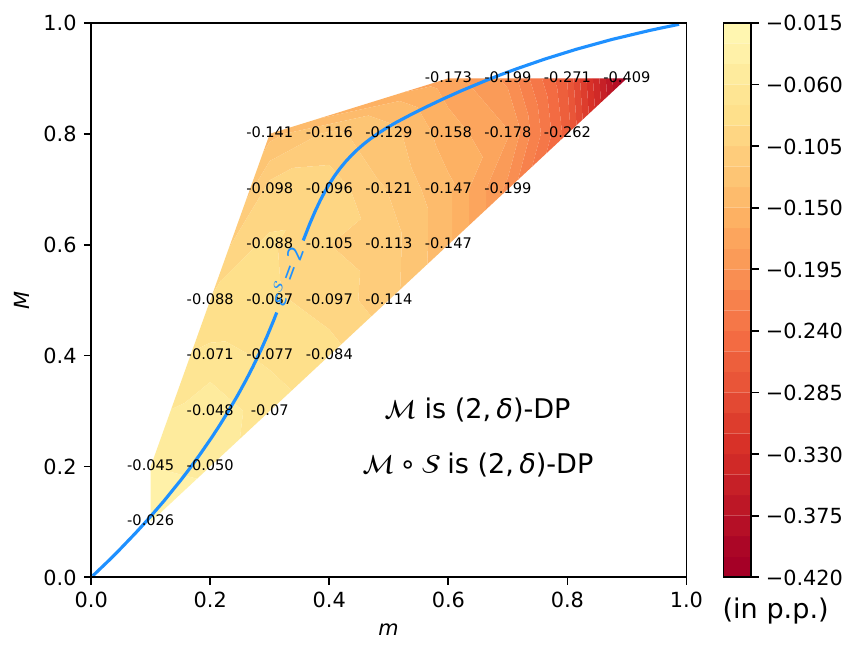}%
    \caption{The mean percent error (MPE) of $\M$ minus that of $\M\circ\S$ at the same privacy levels. Results shown for the NoisyAverage with Gaussian mechanisms over the \texttt{hours-per-week} column in the Adult database.}
    \label{fig:Experiment2-NoisyAverage-Gaussian-Adult-hours-per-week}
\end{figure}

\begin{figure}[H]
    \includegraphics[width=0.25\textwidth]{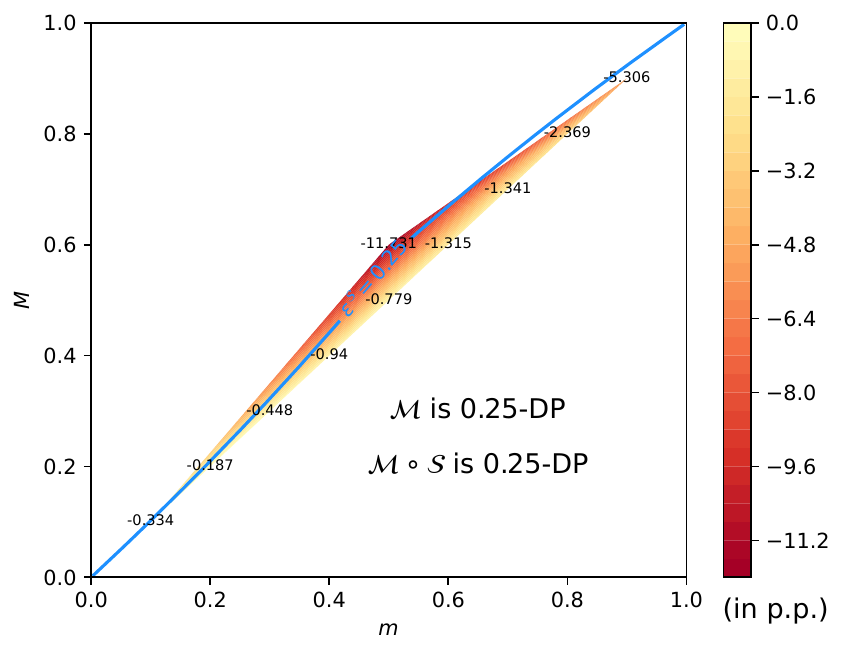}%
    \includegraphics[width=0.25\textwidth]{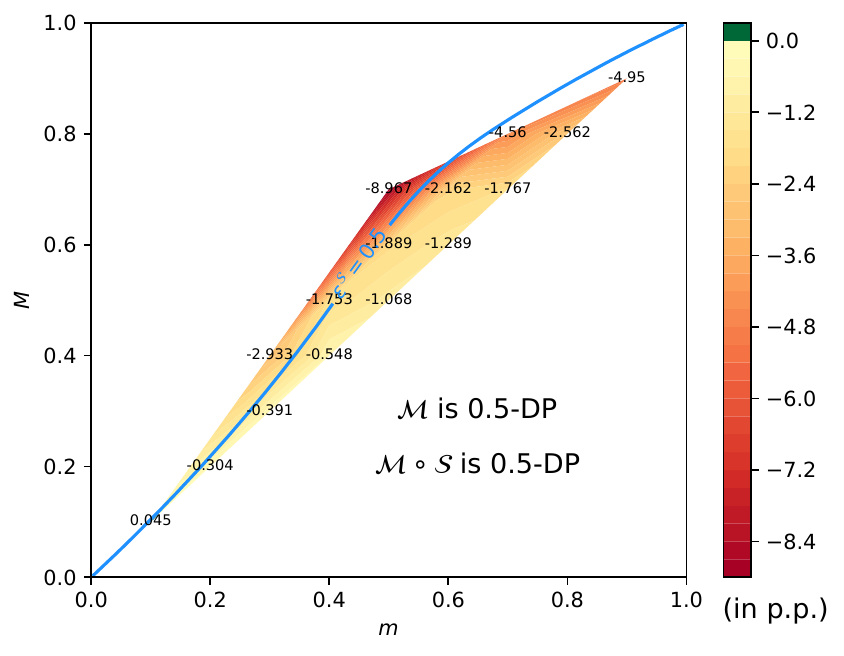}%
    \includegraphics[width=0.25\textwidth]{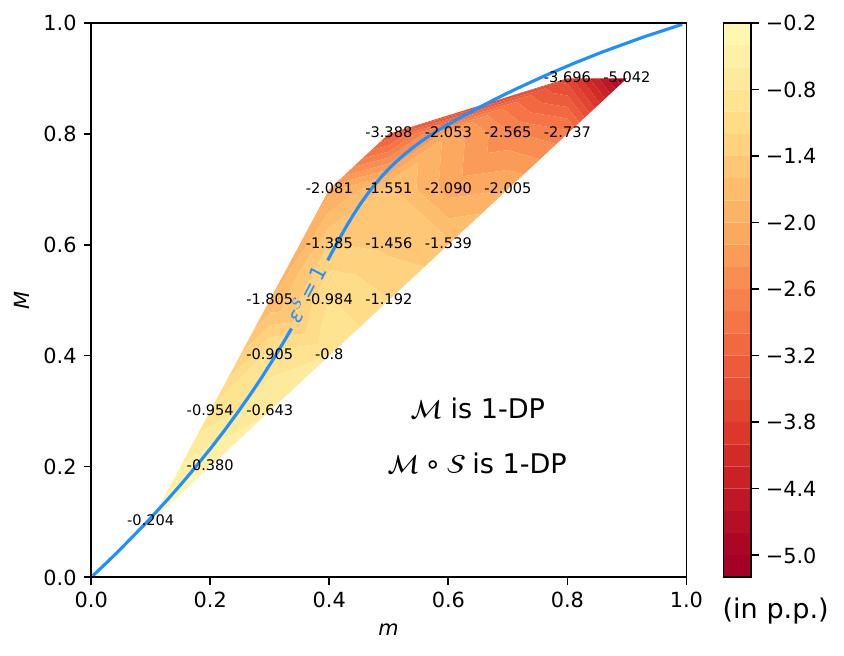}%
    \includegraphics[width=0.25\textwidth]{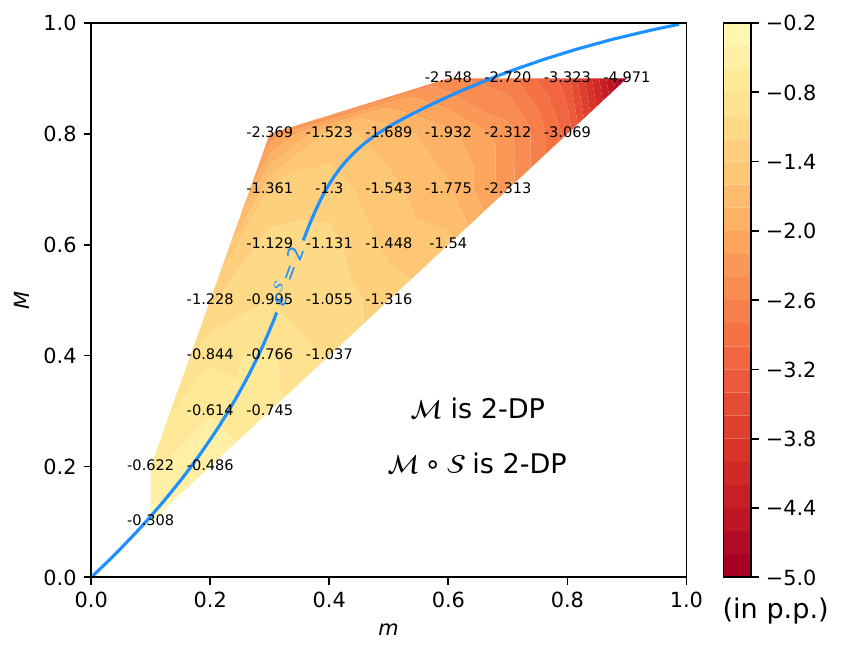}%
    \caption{The mean percent error (MPE) of $\M$ minus that of $\M\circ\S$ at the same privacy levels. Results shown for the NoisyAverage with Laplace mechanisms over the \texttt{FEDTAX} column in the Census database.}
    \label{fig:Experiment2-NoisyAverage-Laplace-Census-FEDTAX}
\end{figure}

\begin{figure}[H]
    \includegraphics[width=0.25\textwidth]{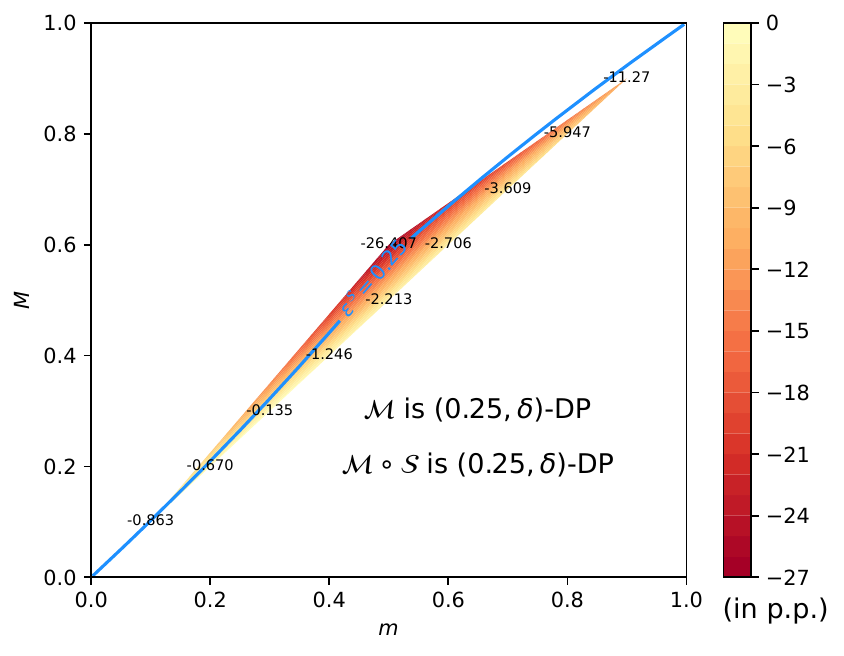}%
    \includegraphics[width=0.25\textwidth]{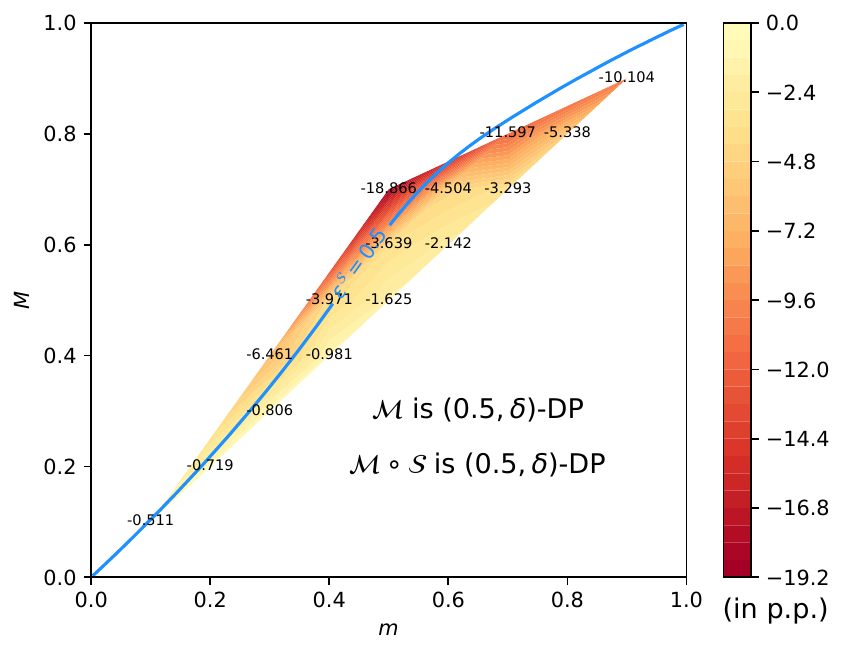}%
    \includegraphics[width=0.25\textwidth]{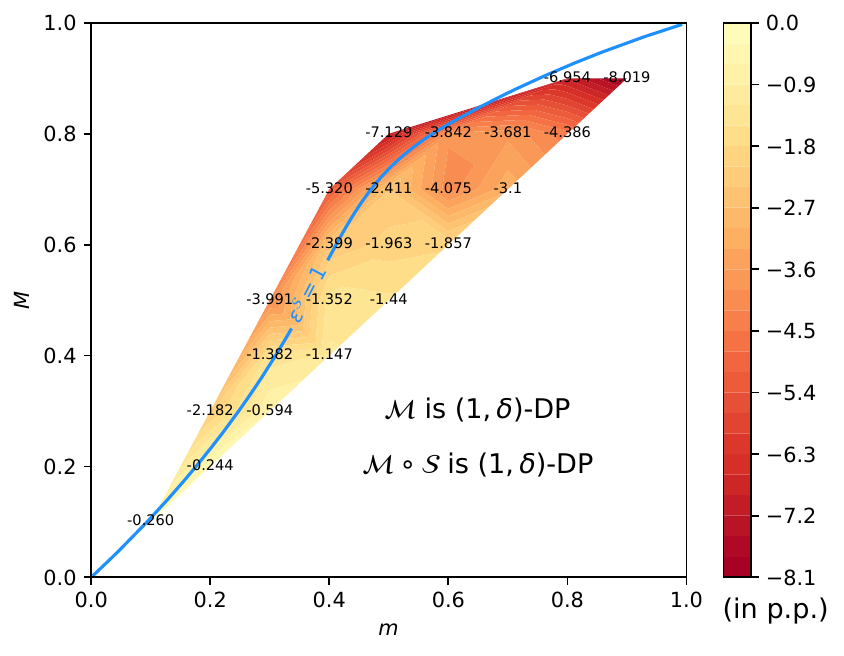}%
    \includegraphics[width=0.25\textwidth]{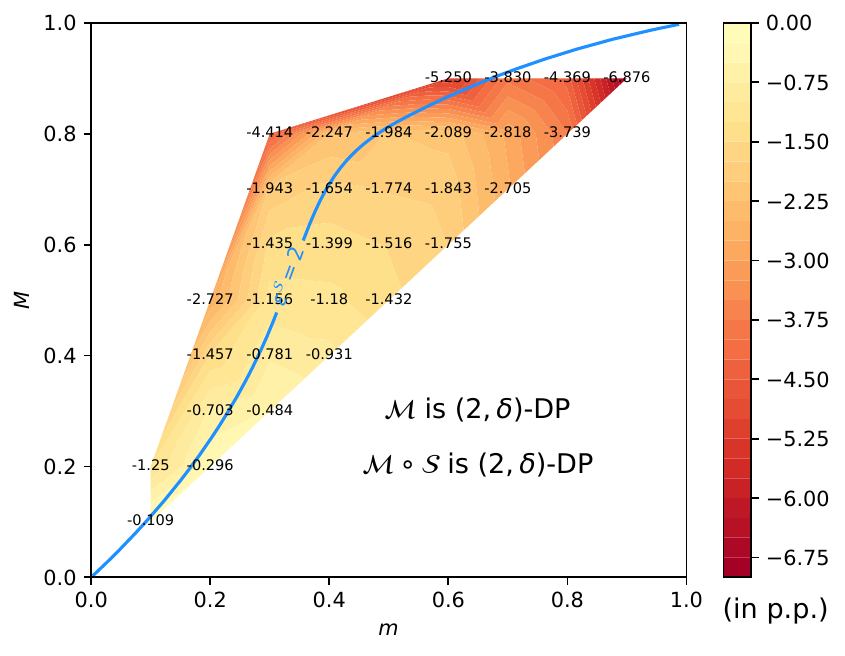}%
    \caption{The mean percent error (MPE) of $\M$ minus that of $\M\circ\S$ at the same privacy levels. Results shown for the NoisyAverage with Gaussian mechanisms over the \texttt{FEDTAX} column in the Census database.}
    \label{fig:Experiment2-NoisyAverage-Gaussian-Census-FEDTAX}
\end{figure}

\begin{figure}[H]
    \includegraphics[width=0.25\textwidth]{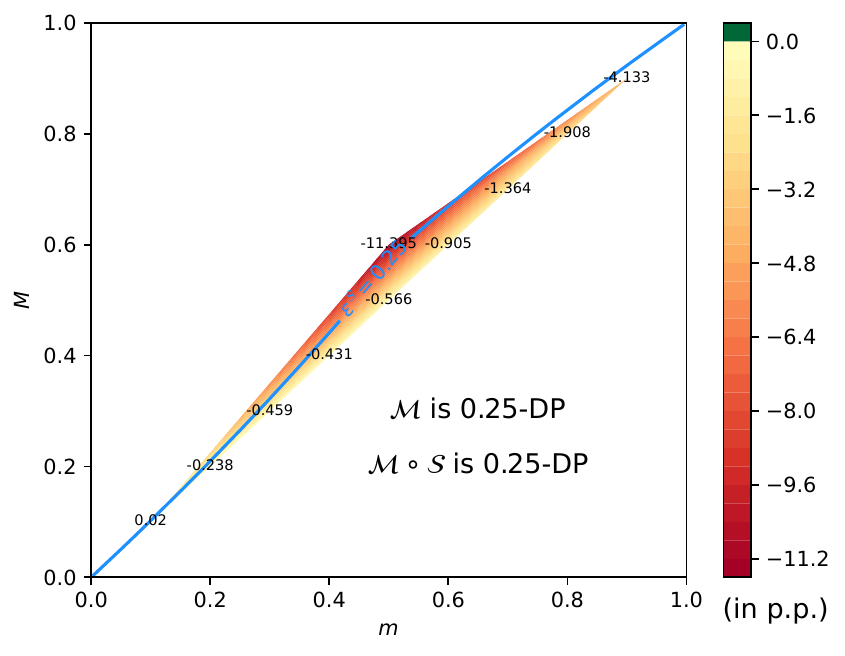}%
    \includegraphics[width=0.25\textwidth]{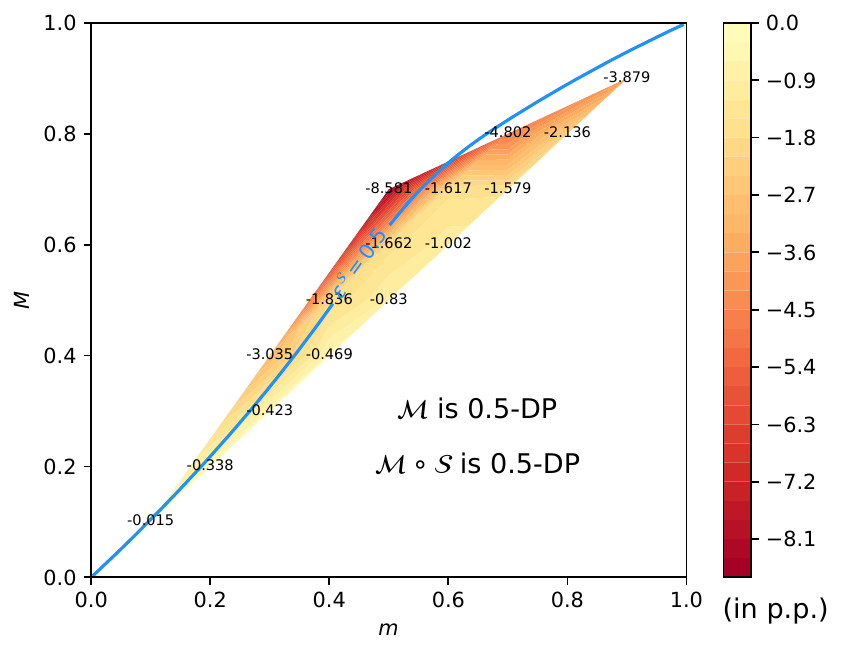}%
    \includegraphics[width=0.25\textwidth]{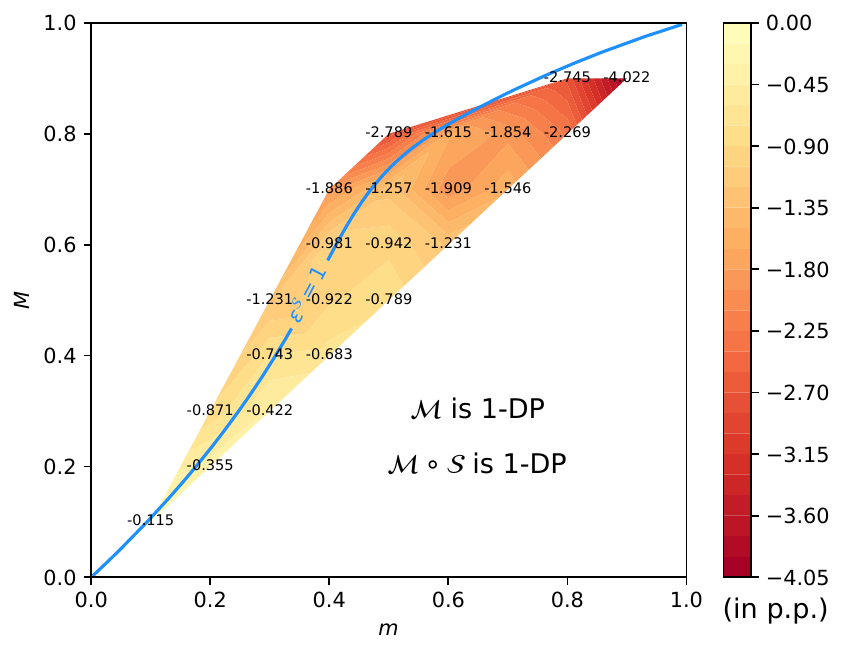}%
    \includegraphics[width=0.25\textwidth]{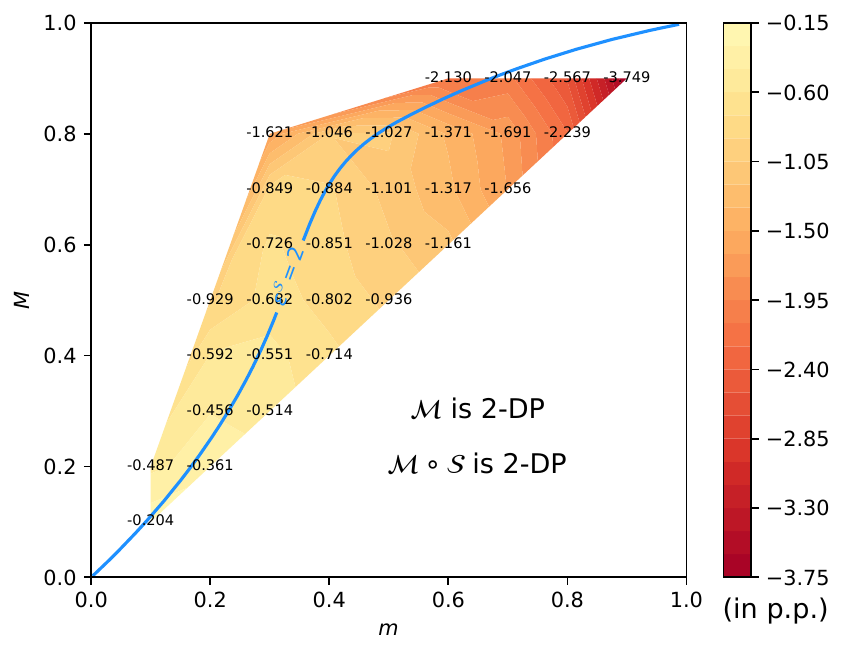}%
    \caption{The mean percent error (MPE) of $\M$ minus that of $\M\circ\S$ at the same privacy levels. Results shown for the NoisyAverage with Laplace mechanisms over the \texttt{FICA} column in the Census database.}
    \label{fig:Experiment2-NoisyAverage-Laplace-Census-FICA}
\end{figure}

\begin{figure}[H]
    \includegraphics[width=0.25\textwidth]{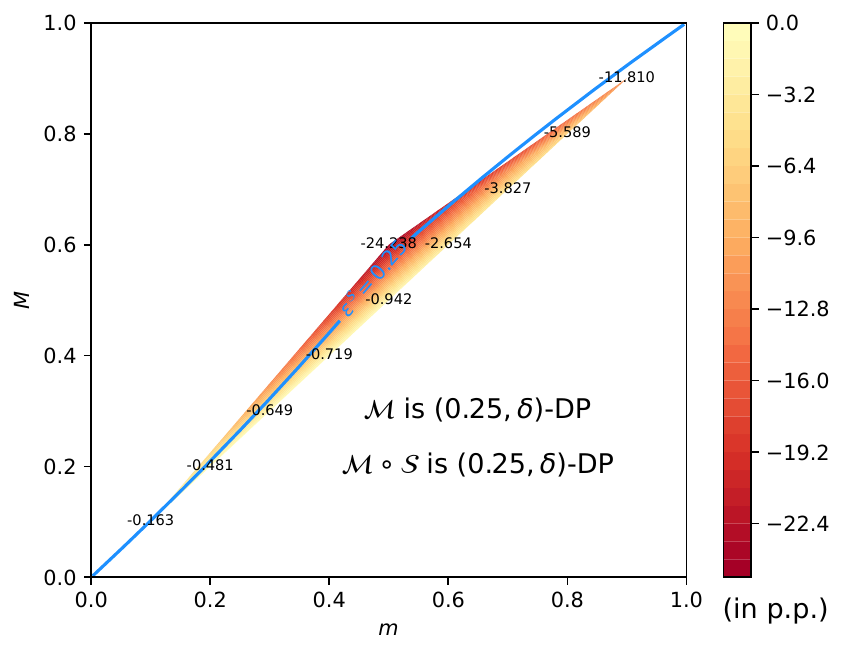}%
    \includegraphics[width=0.25\textwidth]{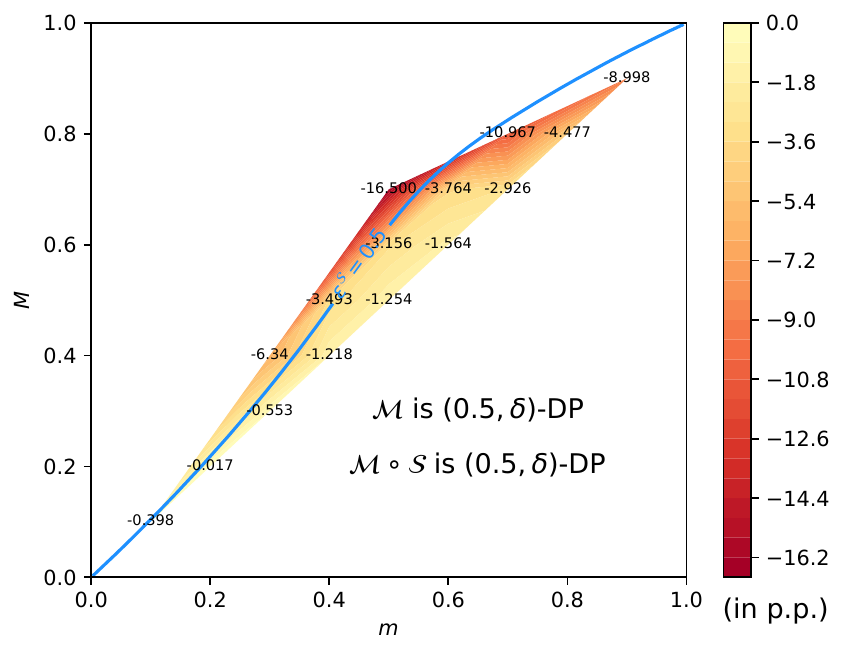}%
    \includegraphics[width=0.25\textwidth]{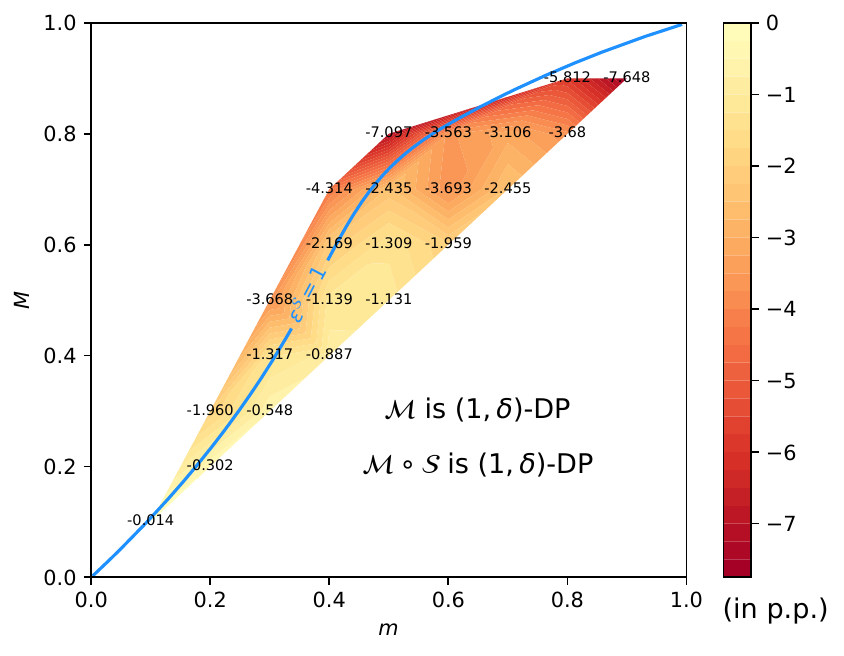}%
    \includegraphics[width=0.25\textwidth]{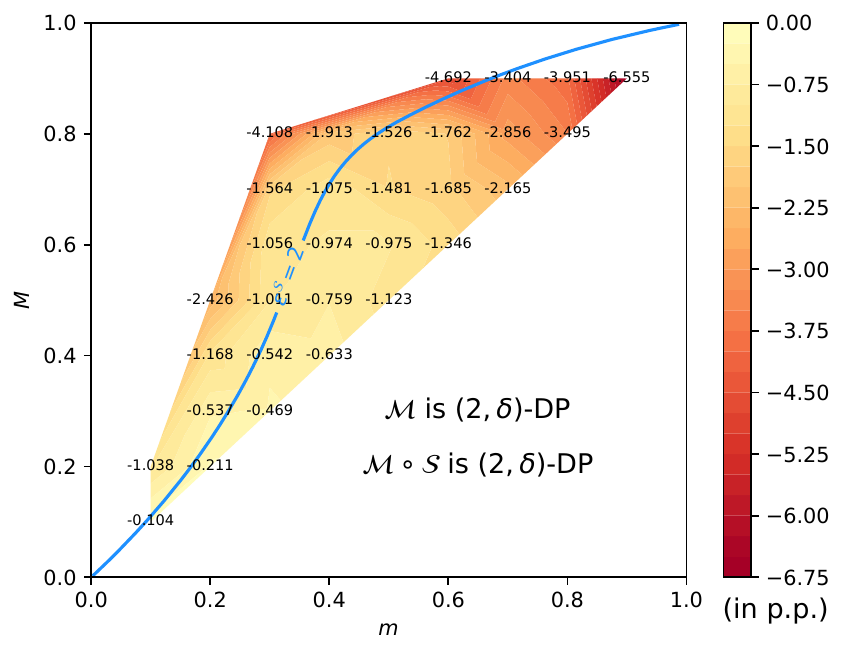}%
    \caption{The mean percent error (MPE) of $\M$ minus that of $\M\circ\S$ at the same privacy levels. Results shown for the NoisyAverage with Gaussian mechanisms over the \texttt{FICA} column in the Census database.}
    \label{fig:Experiment2-NoisyAverage-Gaussian-Census-FICA}
\end{figure}

\begin{figure}[H]
    \includegraphics[width=0.25\textwidth]{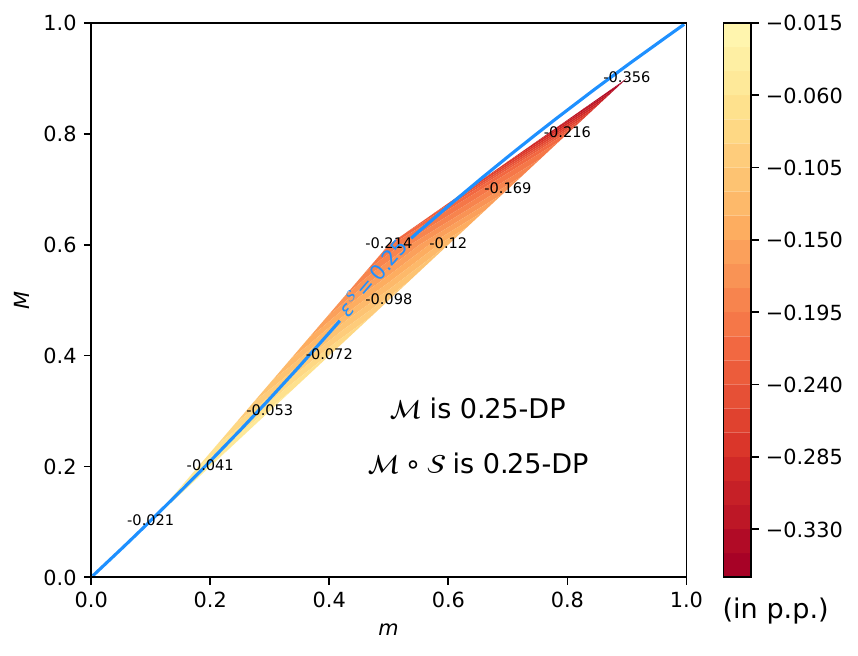}%
    \includegraphics[width=0.25\textwidth]{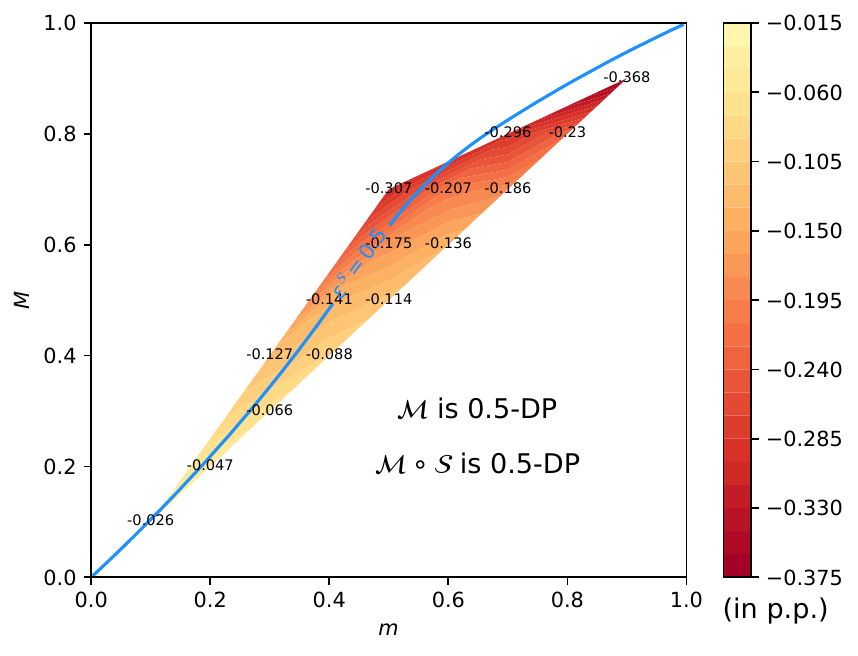}%
    \includegraphics[width=0.25\textwidth]{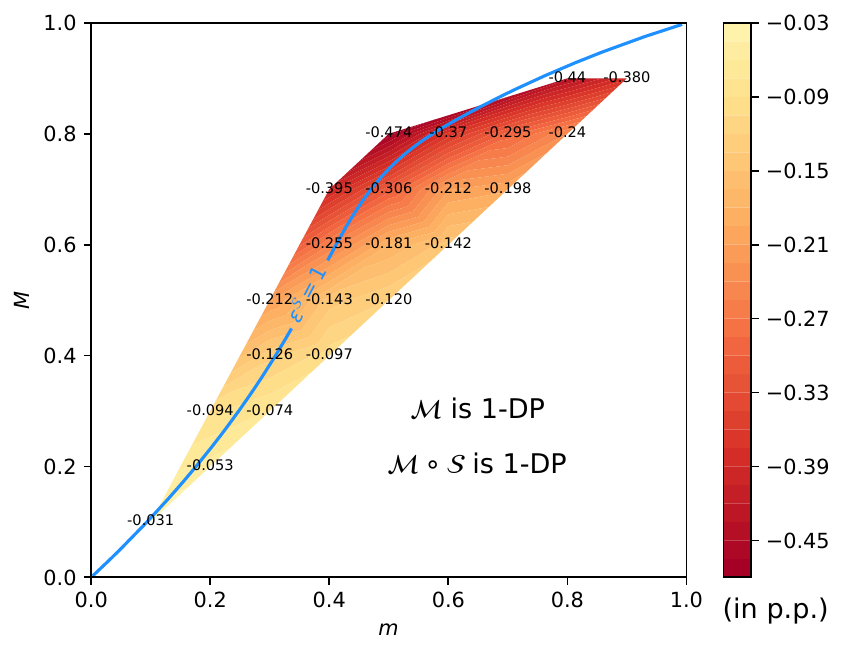}%
    \includegraphics[width=0.25\textwidth]{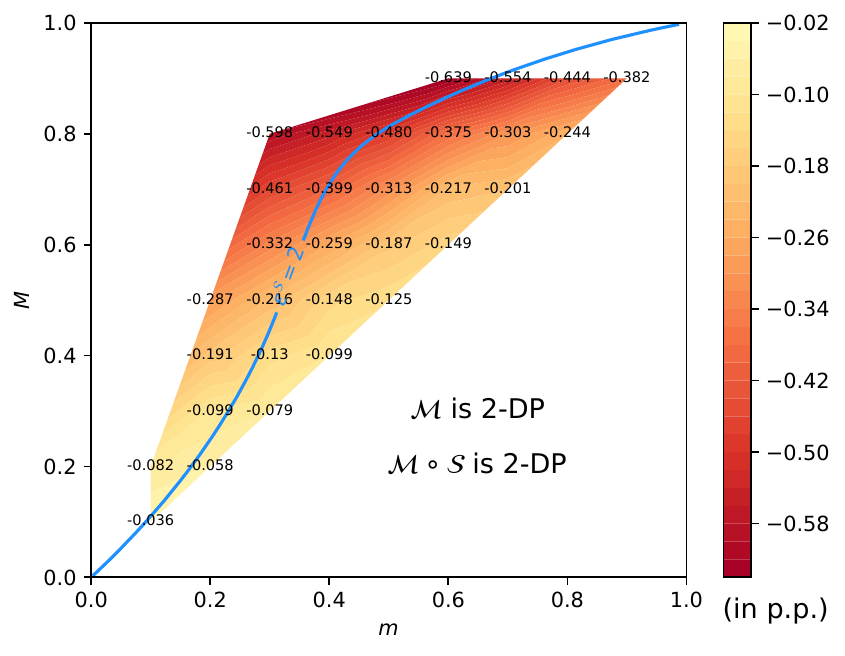}%
    \caption{The mean percent error (MPE) of $\M$ minus that of $\M\circ\S$ at the same privacy levels. Results shown for the NoisyAverage with Laplace mechanisms over the \texttt{Age} column in the Irish database.}
    \label{fig:Experiment2-NoisyAverage-Laplace-Irishn-Age}
\end{figure}

\begin{figure}[H]
    \includegraphics[width=0.25\textwidth]{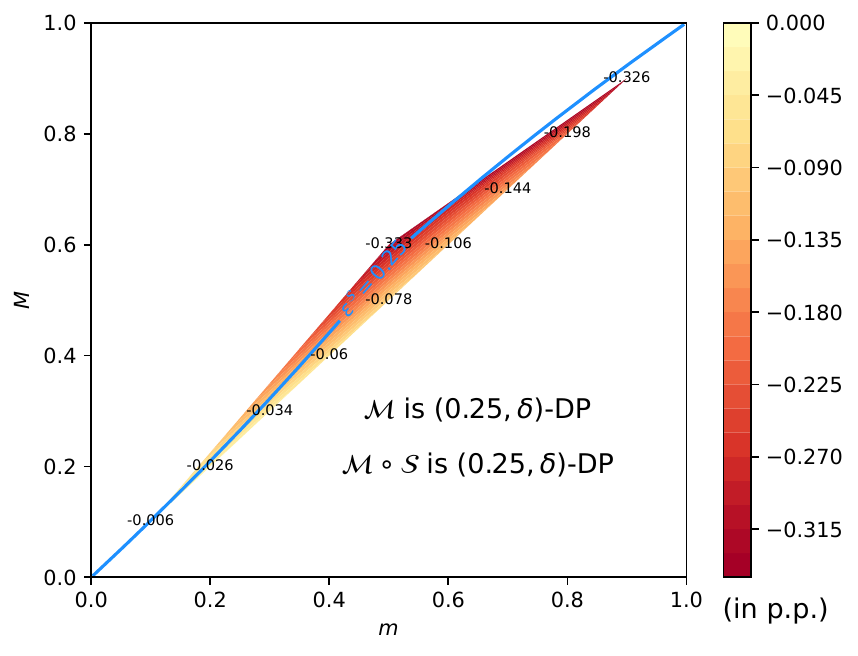}%
    \includegraphics[width=0.25\textwidth]{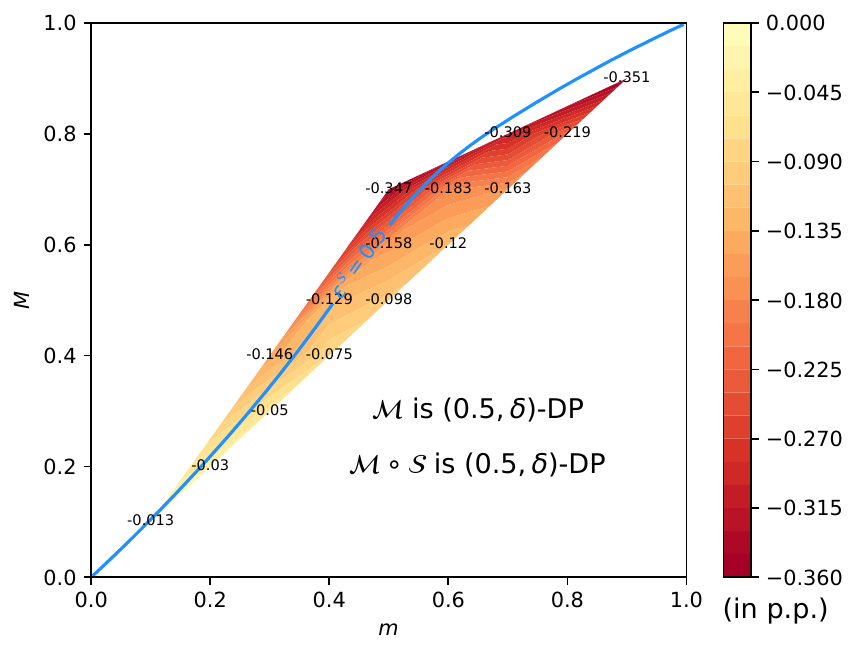}%
    \includegraphics[width=0.25\textwidth]{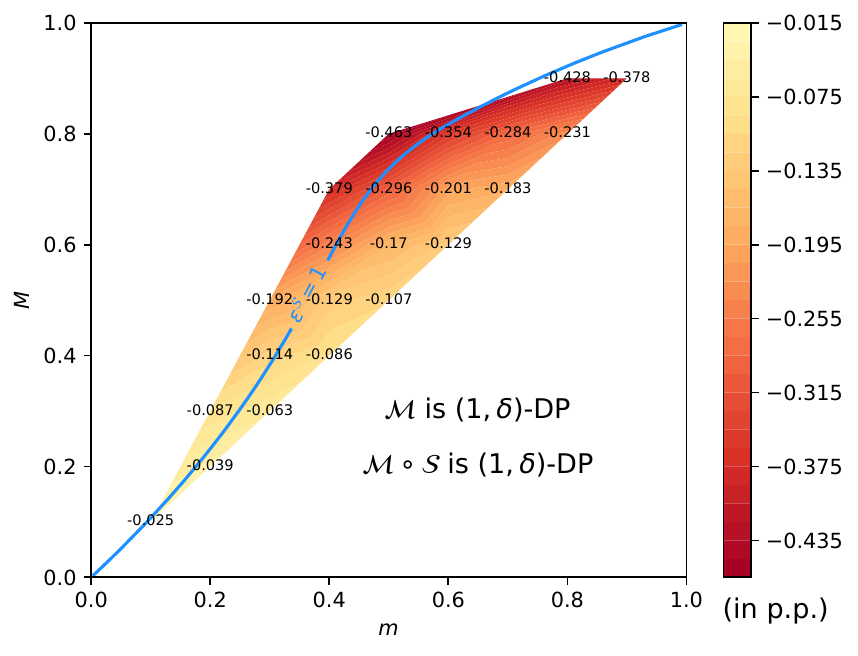}%
    \includegraphics[width=0.25\textwidth]{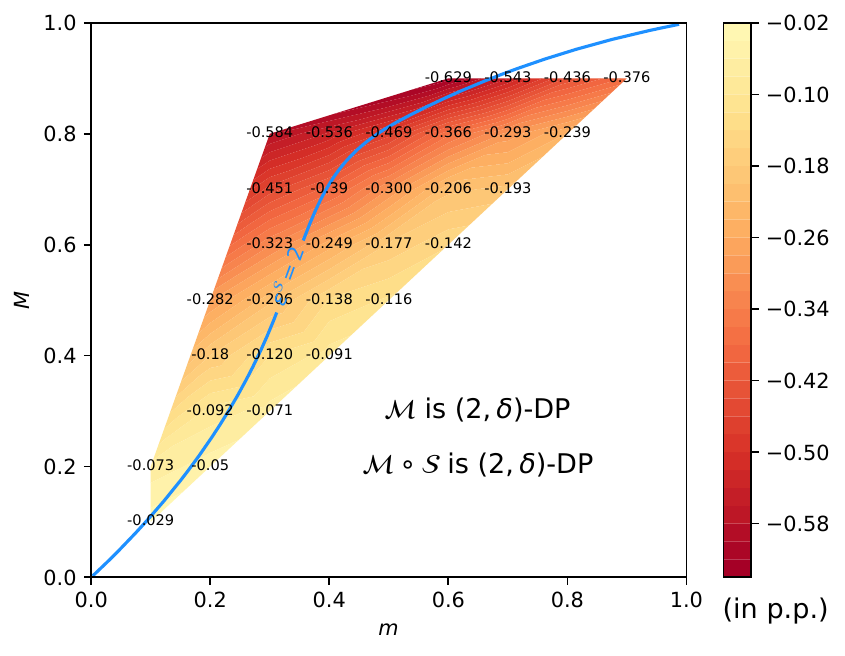}%
    \caption{The mean percent error (MPE) of $\M$ minus that of $\M\circ\S$ at the same privacy levels. Results shown for the NoisyAverage with Gaussian mechanisms over the \texttt{Age} column in the Irish database.}
    \label{fig:Experiment2-NoisyAverage-Gaussian-Irishn-Age}
\end{figure}

\begin{figure}[H]
    \includegraphics[width=0.25\textwidth]{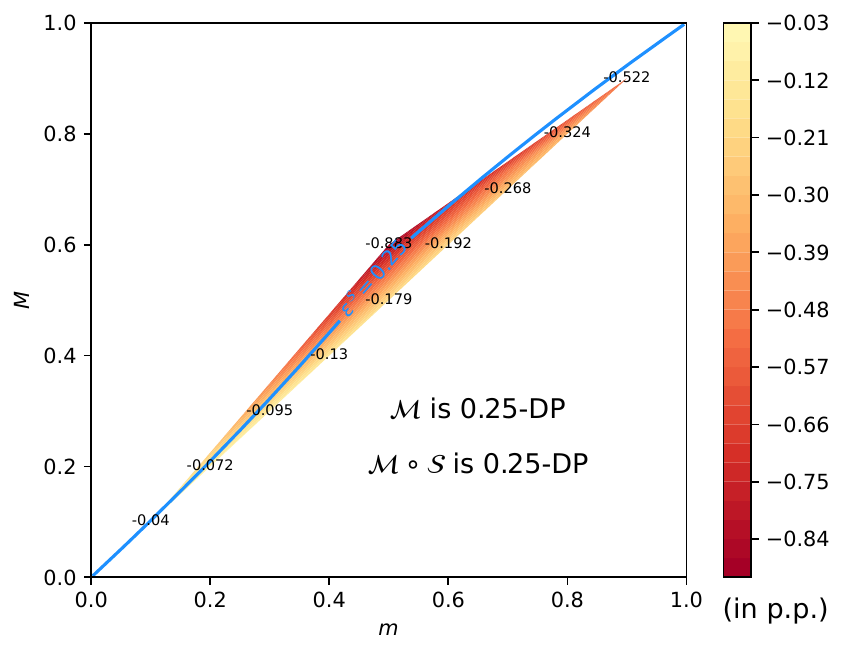}%
    \includegraphics[width=0.25\textwidth]{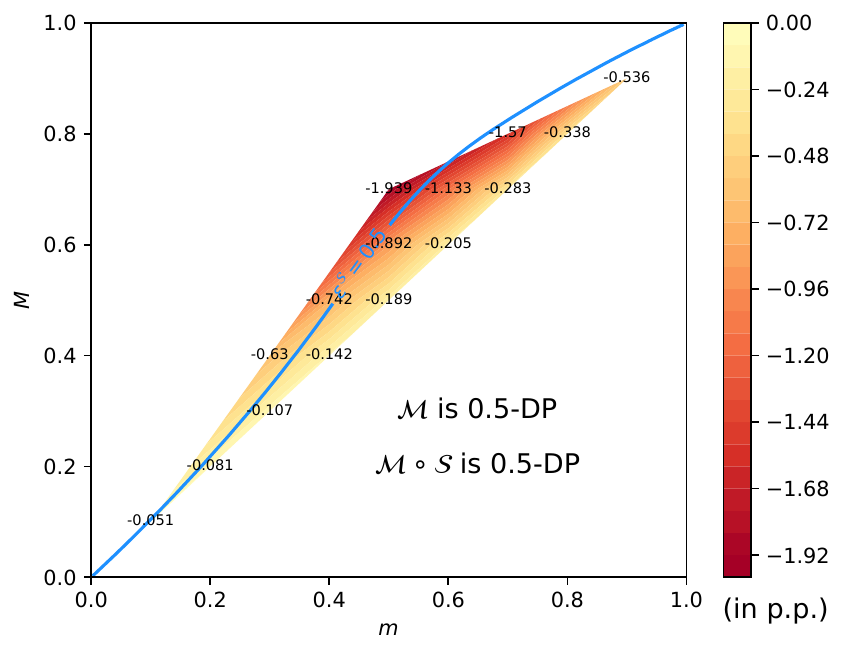}%
    \includegraphics[width=0.25\textwidth]{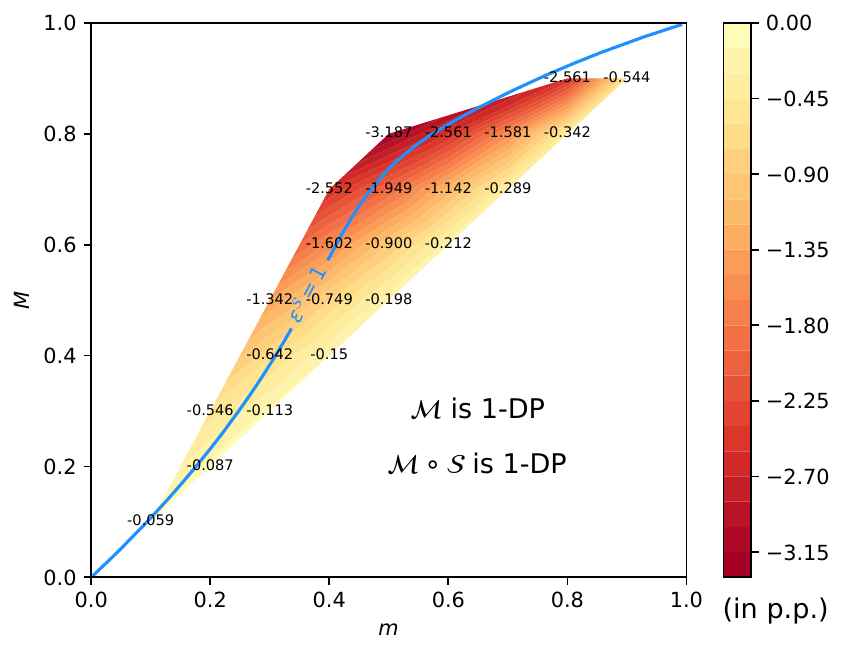}%
    \includegraphics[width=0.25\textwidth]{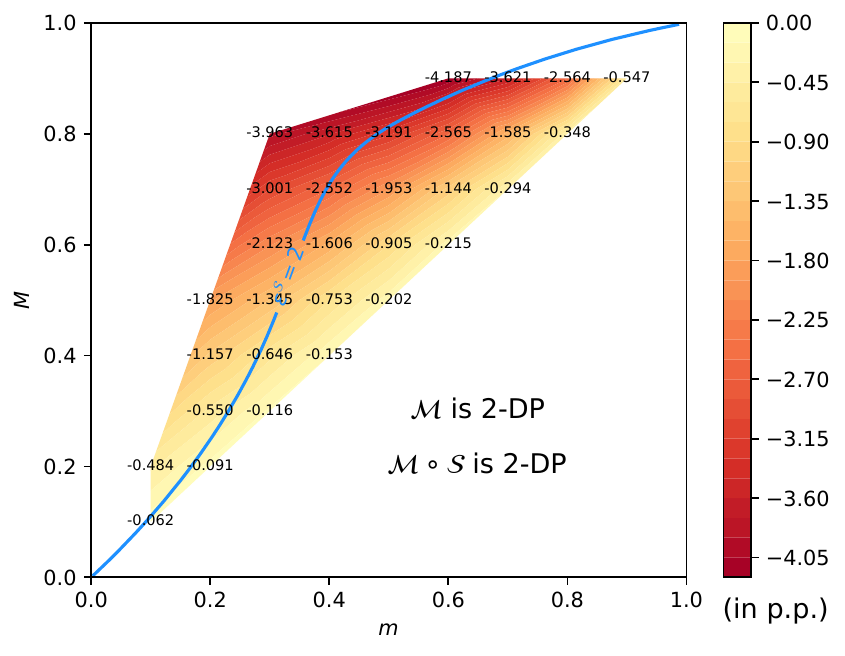}%
    \caption{The mean percent error (MPE) of $\M$ minus that of $\M\circ\S$ at the same privacy levels. Results shown for the NoisyAverage with Laplace mechanisms over the \texttt{HighestEducationCompleted} column in the Irish database.}
    \label{fig:Experiment2-NoisyAverage-Laplace-Irishn-HighestEducationCompleted}
\end{figure}

\begin{figure}[H]
    \includegraphics[width=0.25\textwidth]{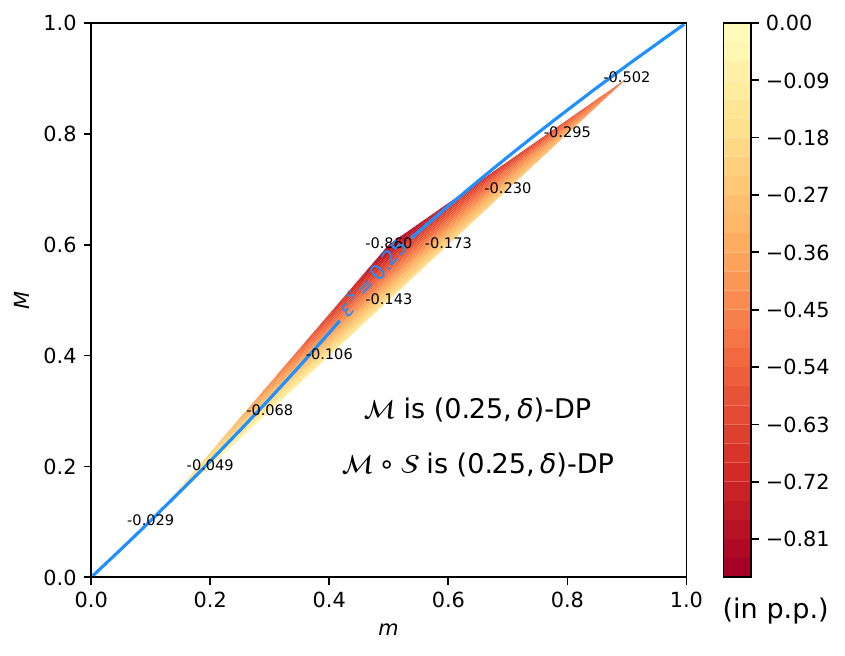}%
    \includegraphics[width=0.25\textwidth]{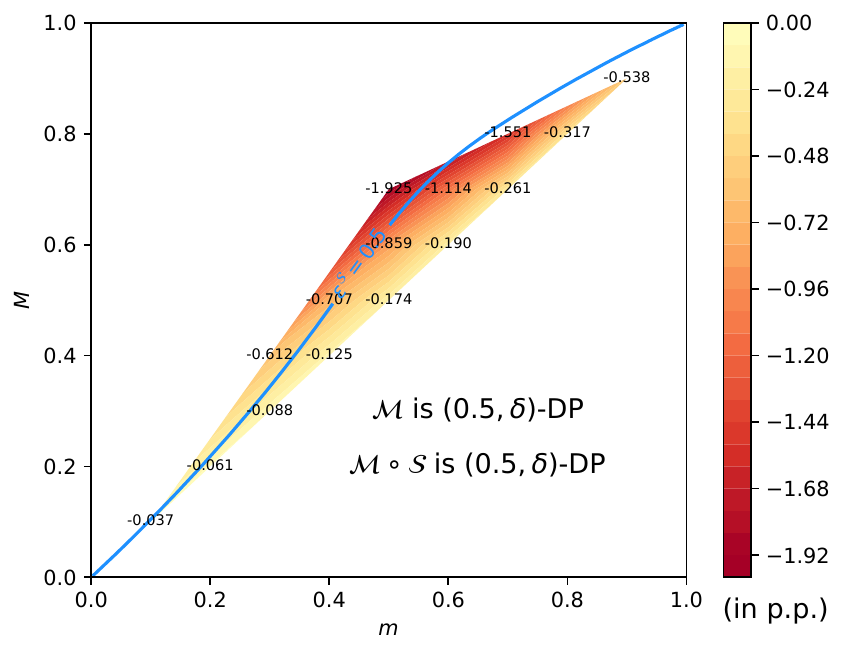}%
    \includegraphics[width=0.25\textwidth]{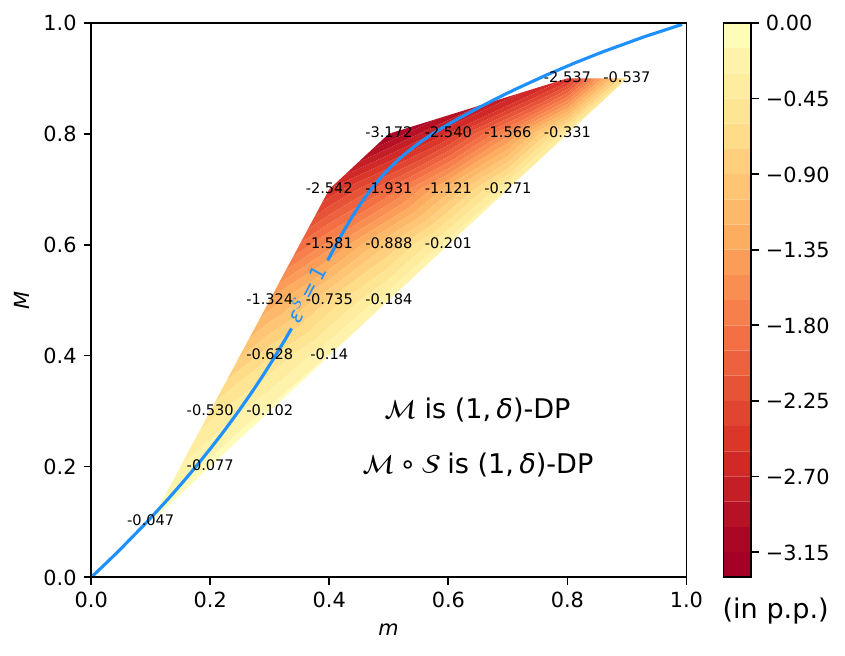}%
    \includegraphics[width=0.25\textwidth]{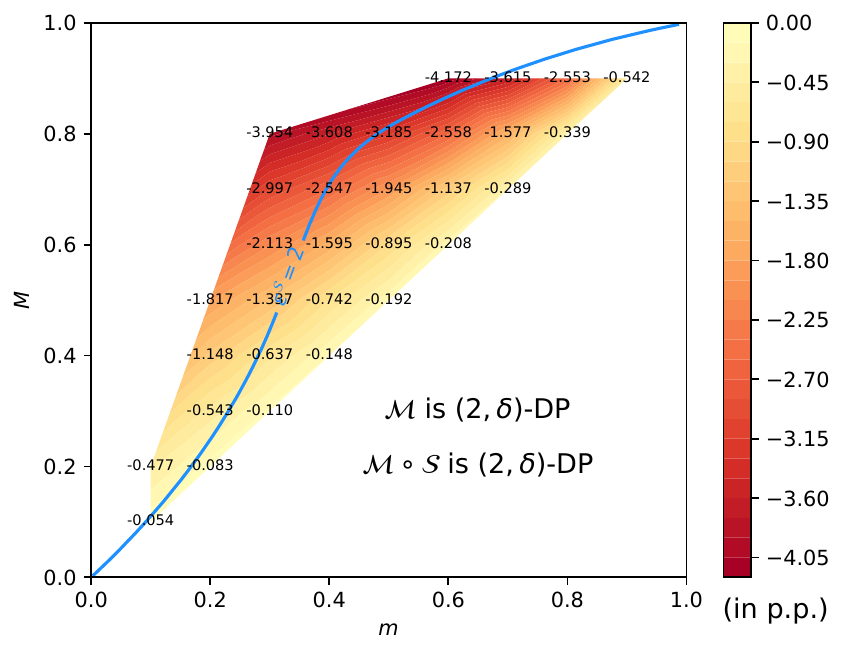}%
    \caption{The mean percent error (MPE) of $\M$ minus that of $\M\circ\S$ at the same privacy levels. Results shown for the NoisyAverage with Gaussian mechanisms over the \texttt{HighestEducationCompleted} column in the Irish database.}
    \label{fig:Experiment2-NoisyAverage-Gaussian-Irishn-HighestEducationCompleted}
\end{figure}

\subsection{Plots for the Mode Computation}\label{sec:plots:SuppressionwithEpsDeltaChange2}

\begin{figure}[H]
    \includegraphics[width=0.25\textwidth]{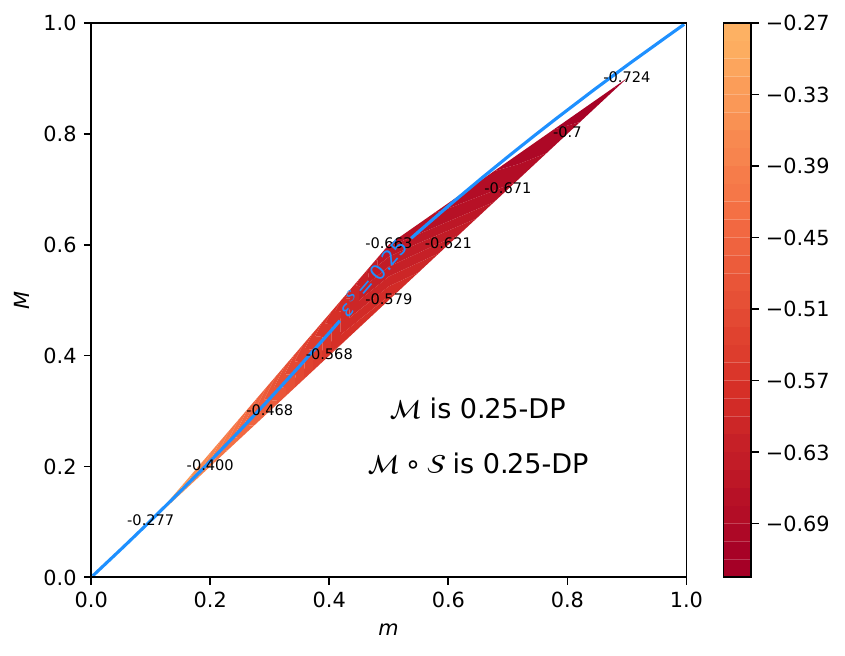}%
    \includegraphics[width=0.25\textwidth]{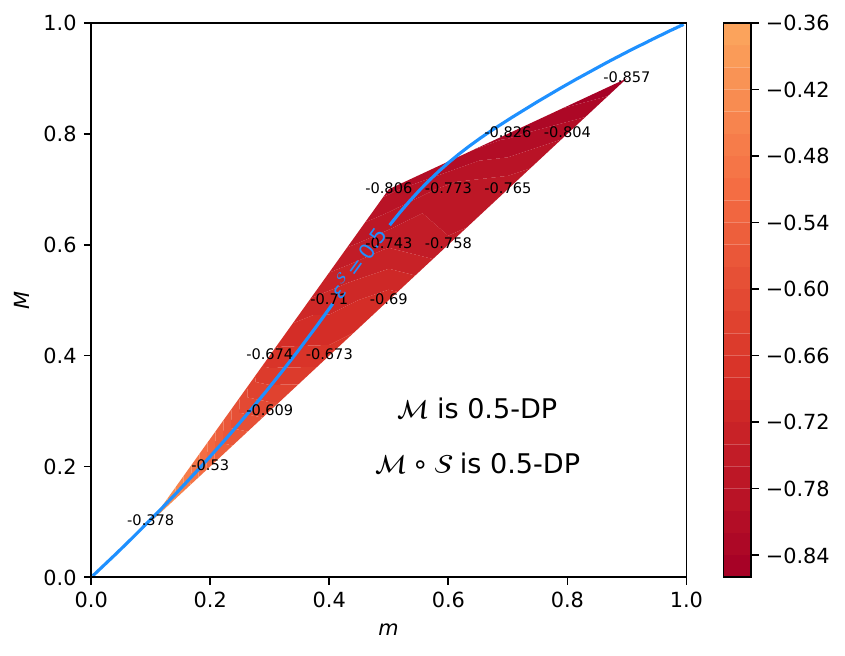}%
    \includegraphics[width=0.25\textwidth]{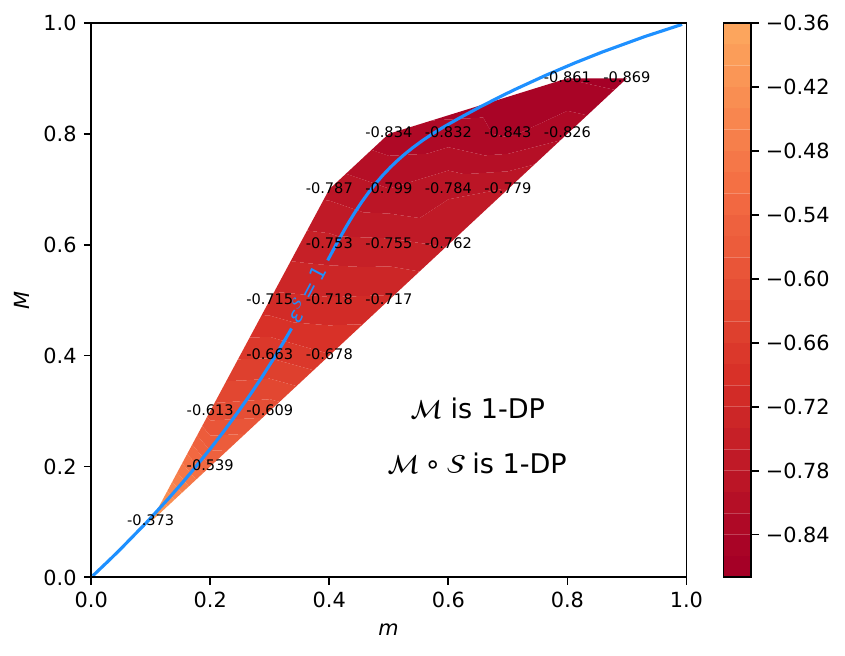}%
    \includegraphics[width=0.25\textwidth]{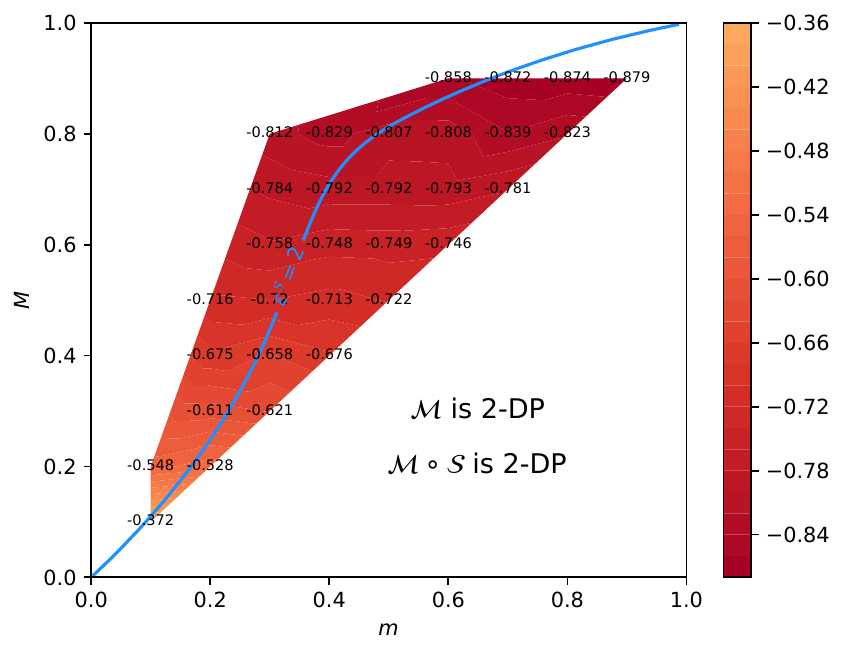}%
    \caption{The probability of outputting an incorrect mode of $\M$ minus that of $\M\circ\S$ at the same privacy levels. Results shown for the RNM with Laplace noise over the \texttt{age} column in the Adult database.}
    \label{fig:Experiment2-RNM-Laplace-Adult-age}
\end{figure}

\begin{figure}[H]
    \includegraphics[width=0.25\textwidth]{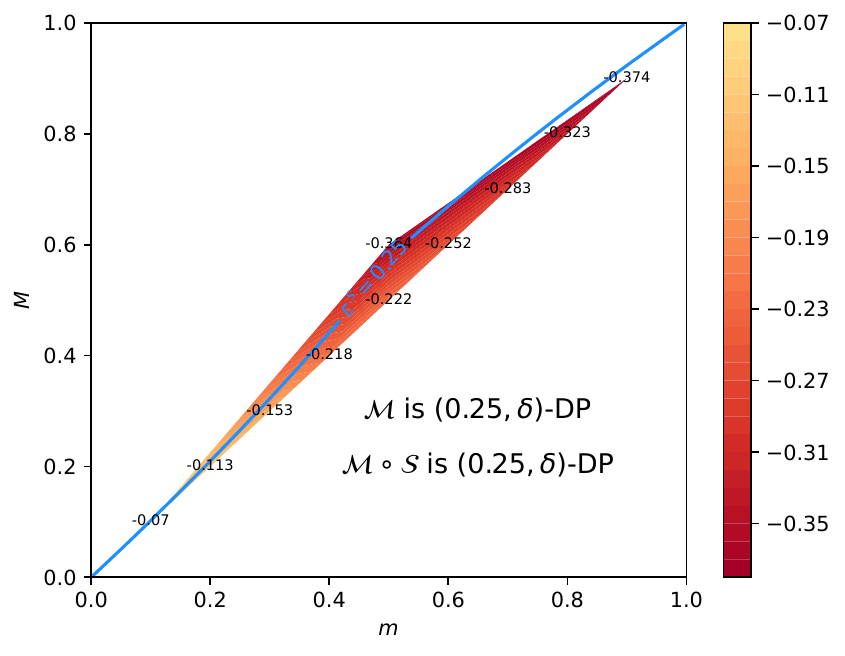}%
    \includegraphics[width=0.25\textwidth]{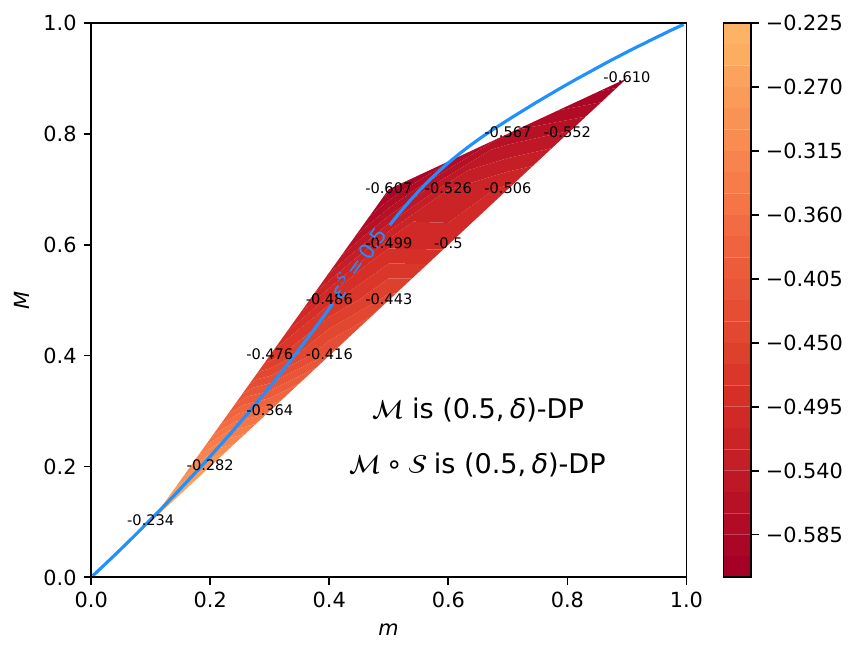}%
    \includegraphics[width=0.25\textwidth]{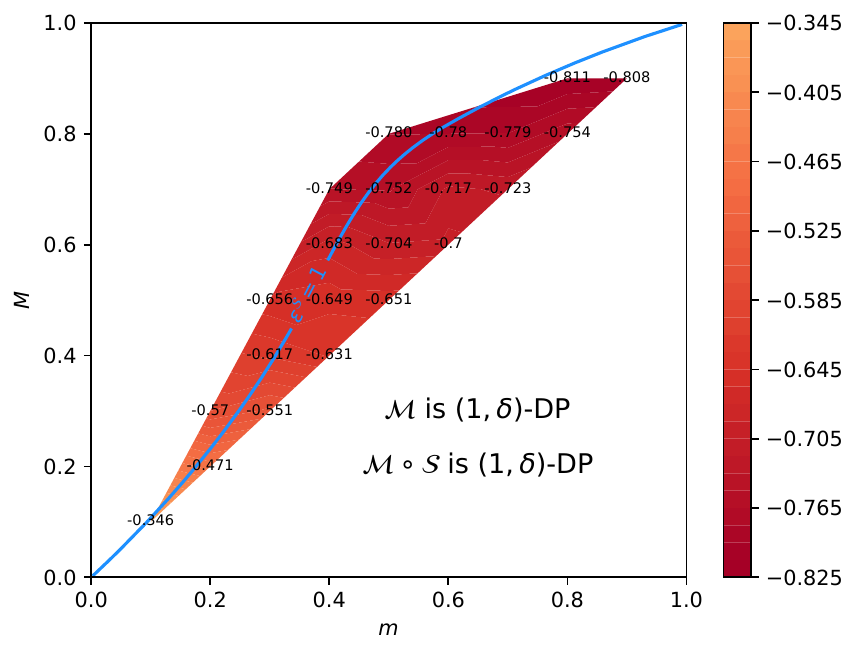}%
    \includegraphics[width=0.25\textwidth]{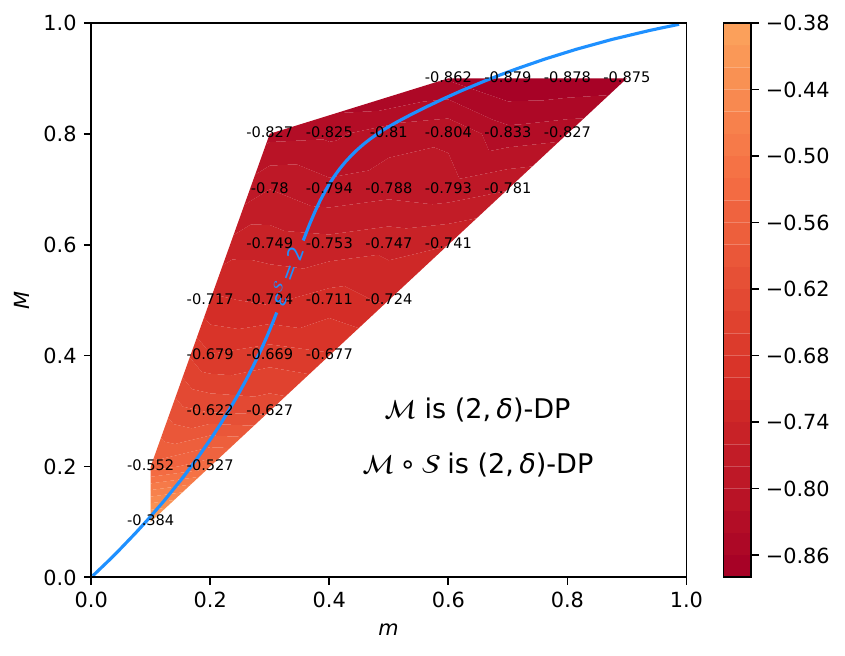}%
    \caption{The probability of outputting an incorrect mode of $\M$ minus that of $\M\circ\S$ at the same privacy levels. Results shown for the RNM-like variant with Gaussian noise over the \texttt{age} column in the Adult database.}
    \label{fig:Experiment2-RNM-Gaussian-Adult-age}
\end{figure}

\begin{figure}[H]
    \includegraphics[width=0.25\textwidth]{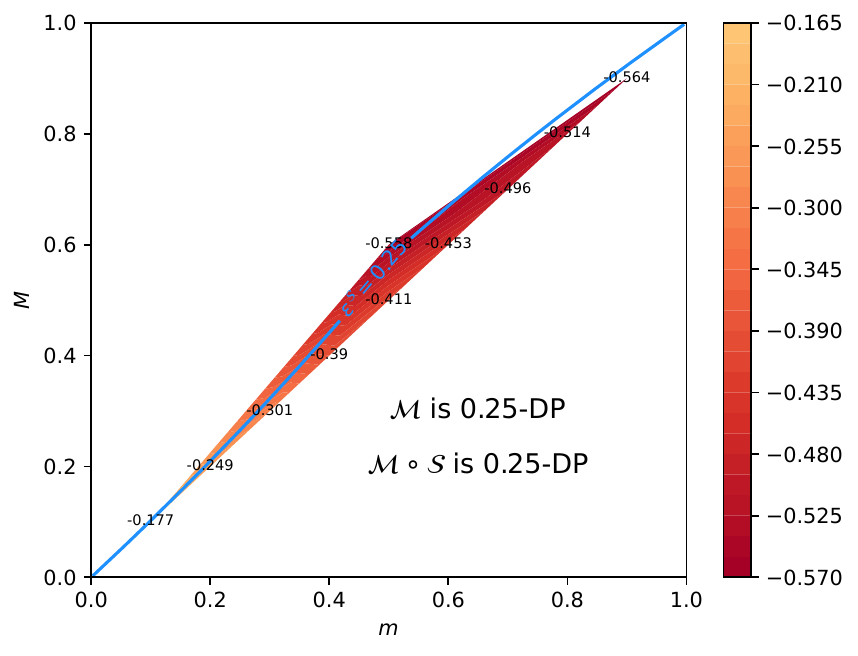}%
    \includegraphics[width=0.25\textwidth]{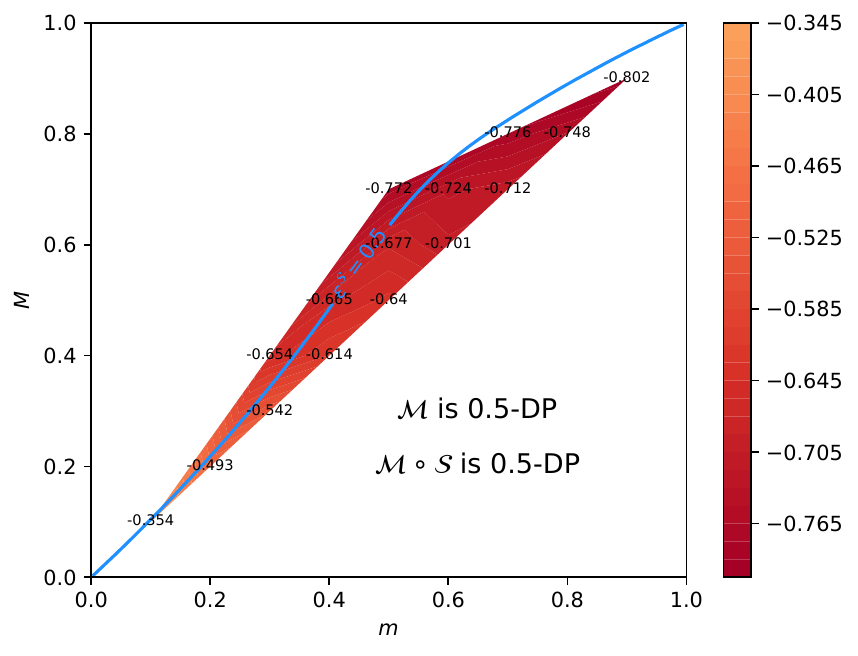}%
    \includegraphics[width=0.25\textwidth]{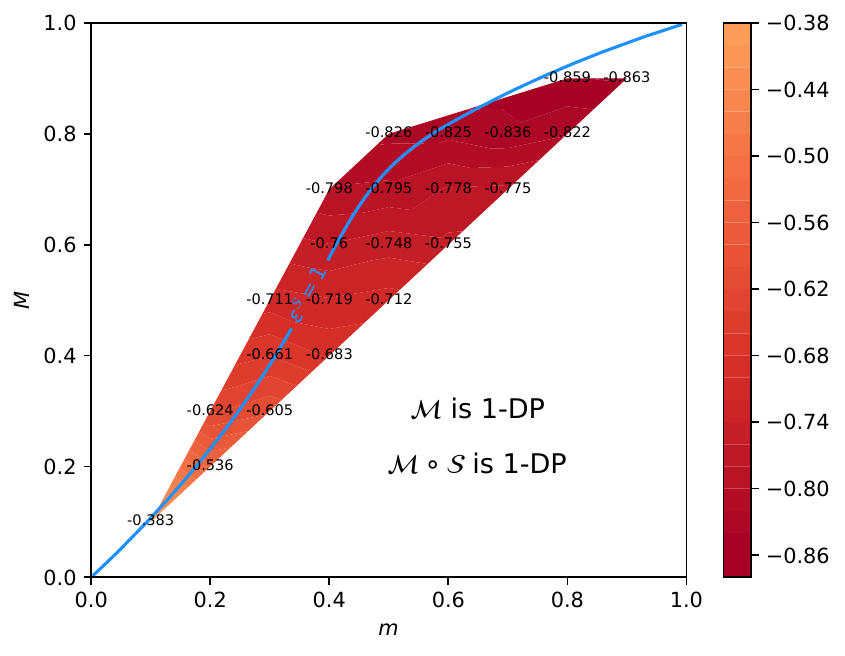}%
    \includegraphics[width=0.25\textwidth]{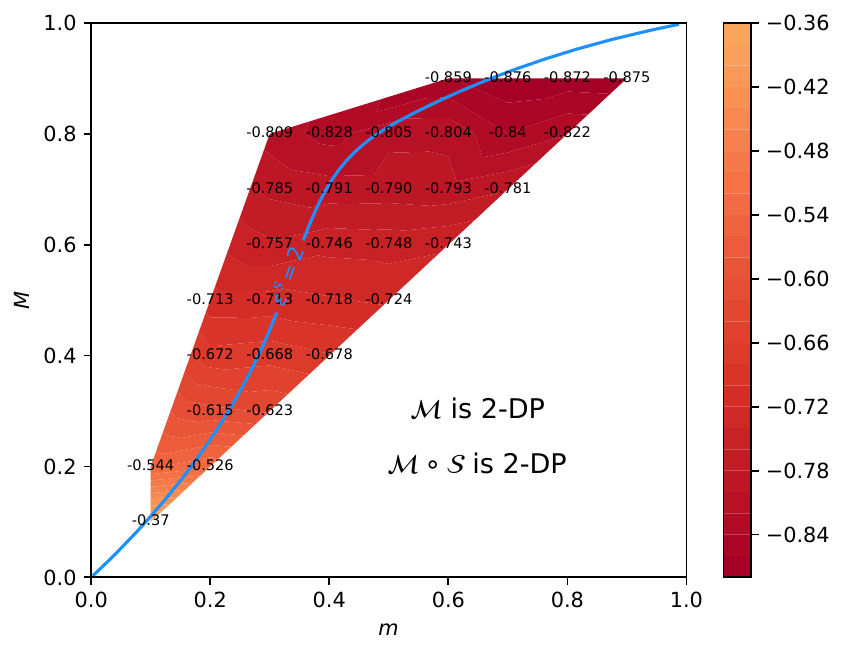}%
    \caption{The probability of outputting an incorrect mode of $\M$ minus that of $\M\circ\S$ at the same privacy levels. Results shown for the RNM with exponential noise over the \texttt{age} column in the Adult database.}
    \label{fig:Experiment2-RNM-Exponential-Adult-age}
\end{figure}

\begin{figure}[H]
    \includegraphics[width=0.25\textwidth]{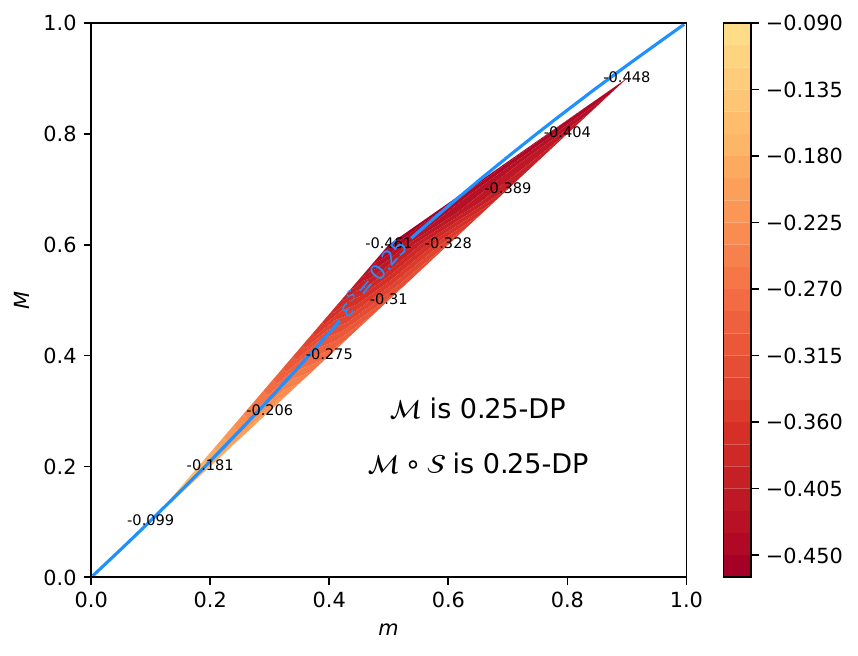}%
    \includegraphics[width=0.25\textwidth]{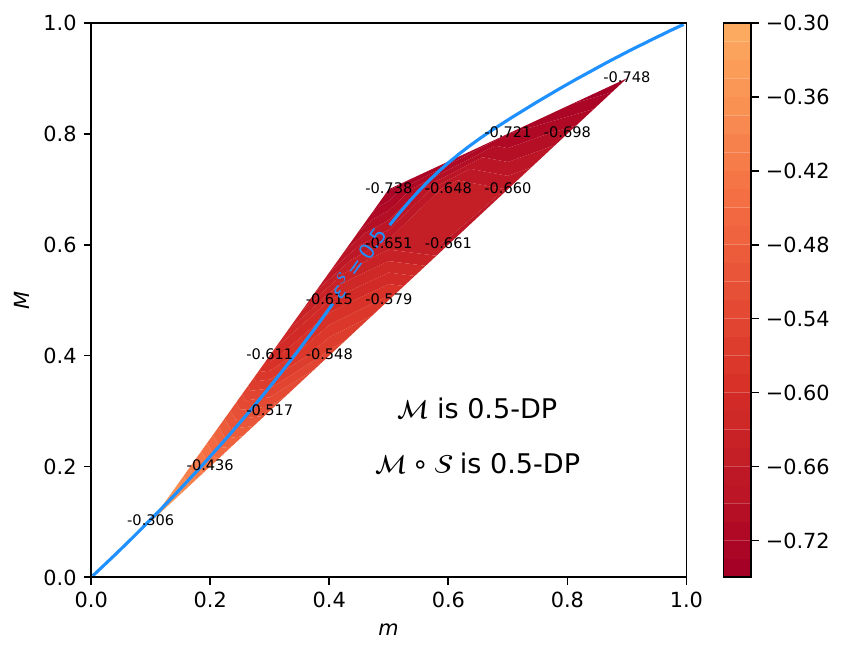}%
    \includegraphics[width=0.25\textwidth]{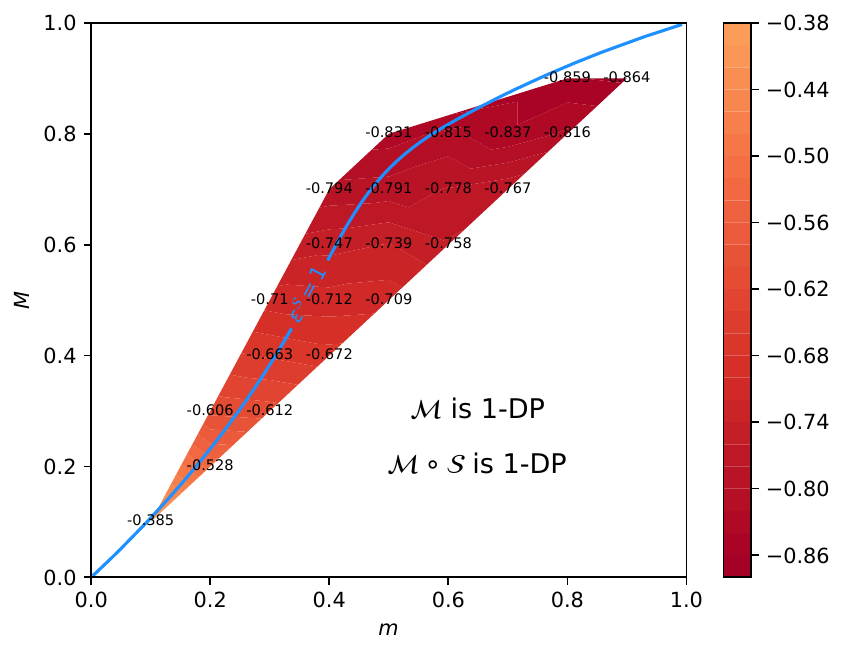}%
    \includegraphics[width=0.25\textwidth]{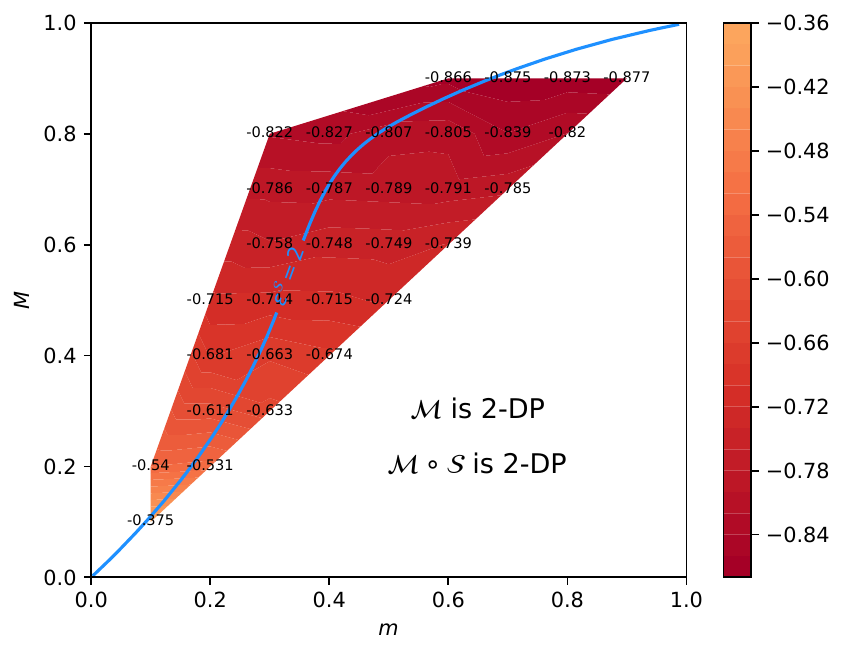}%
    \caption{The probability of outputting an incorrect mode of $\M$ minus that of $\M\circ\S$ at the same privacy levels. Results shown for the exponential mechanism over the \texttt{age} column in the Adult database.}
    \label{fig:Experiment2-RNM-ExponentialMechanism-Adult-age}
\end{figure}

\begin{figure}[H]
    \includegraphics[width=0.25\textwidth]{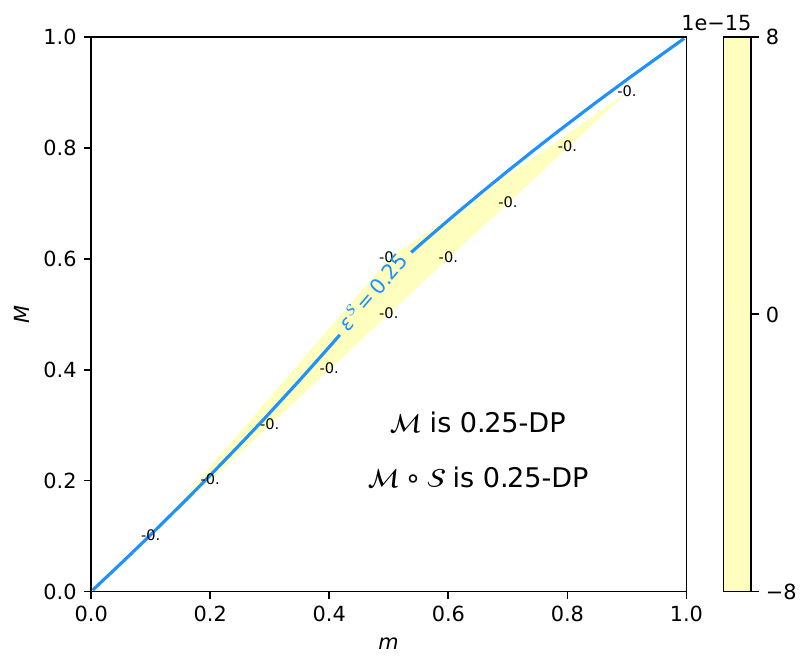}%
    \includegraphics[width=0.25\textwidth]{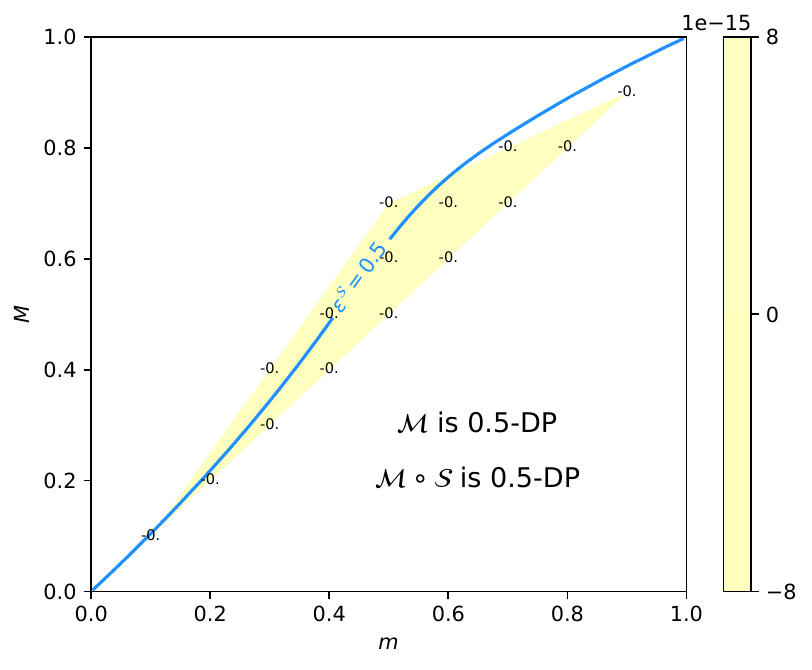}%
    \includegraphics[width=0.25\textwidth]{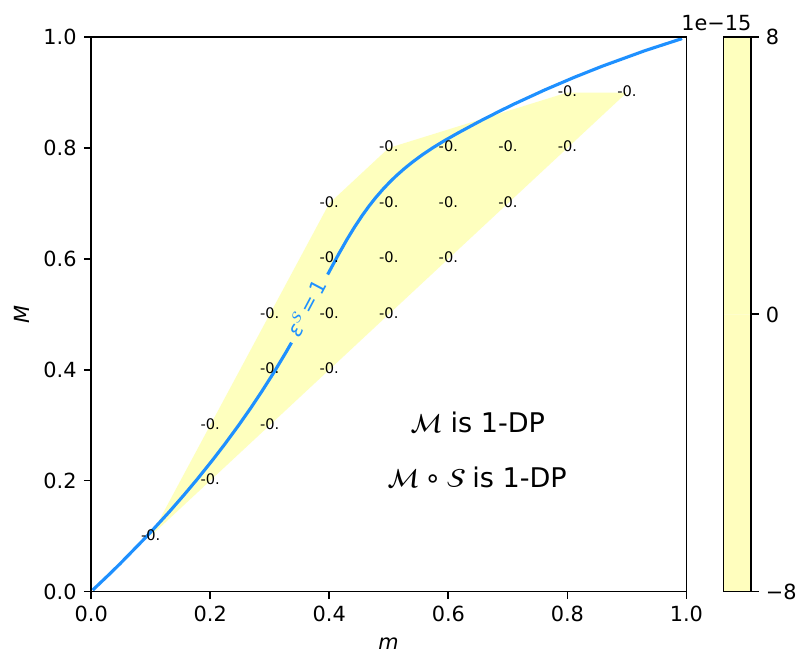}%
    \includegraphics[width=0.25\textwidth]{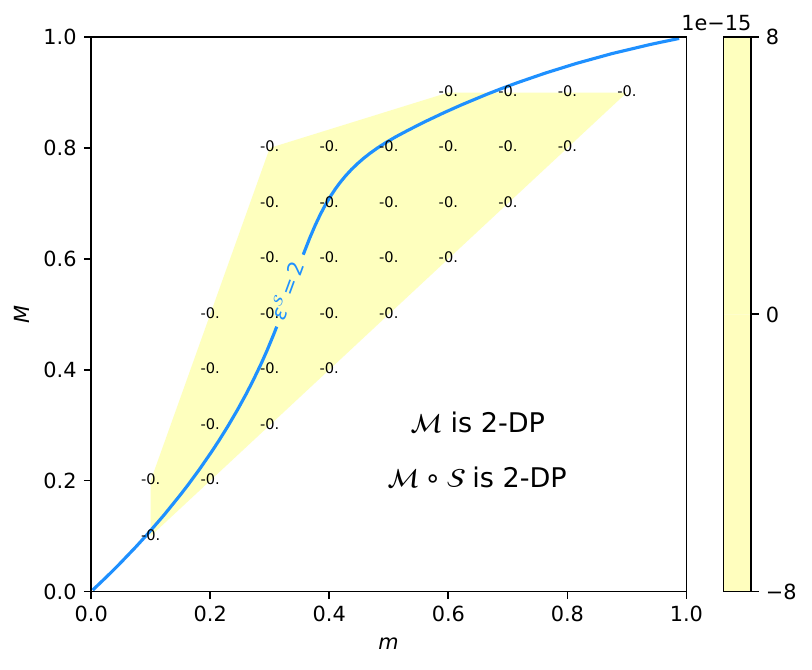}%
    \caption{The probability of outputting an incorrect mode of $\M$ minus that of $\M\circ\S$ at the same privacy levels. Results shown for the RNM with Laplace noise over the \texttt{hours-per-week} column in the Adult database.}
    \label{fig:Experiment2-RNM-Laplace-Adult-hours-per-week}
\end{figure}

\begin{figure}[H]
    \includegraphics[width=0.25\textwidth]{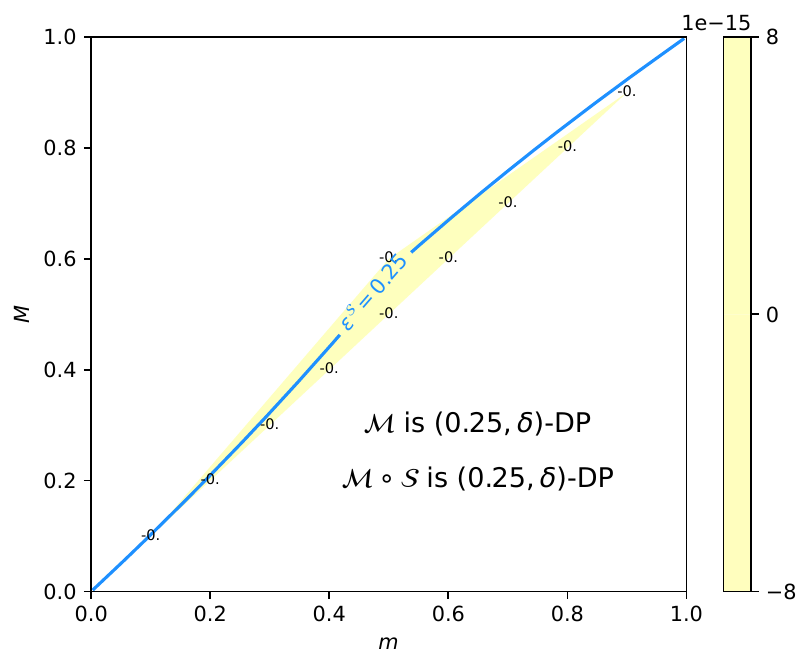}%
    \includegraphics[width=0.25\textwidth]{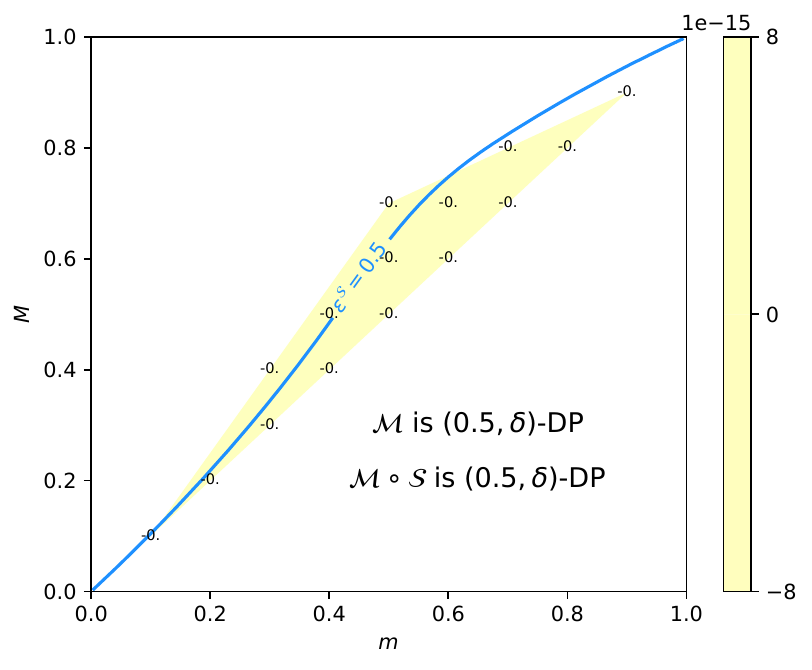}%
    \includegraphics[width=0.25\textwidth]{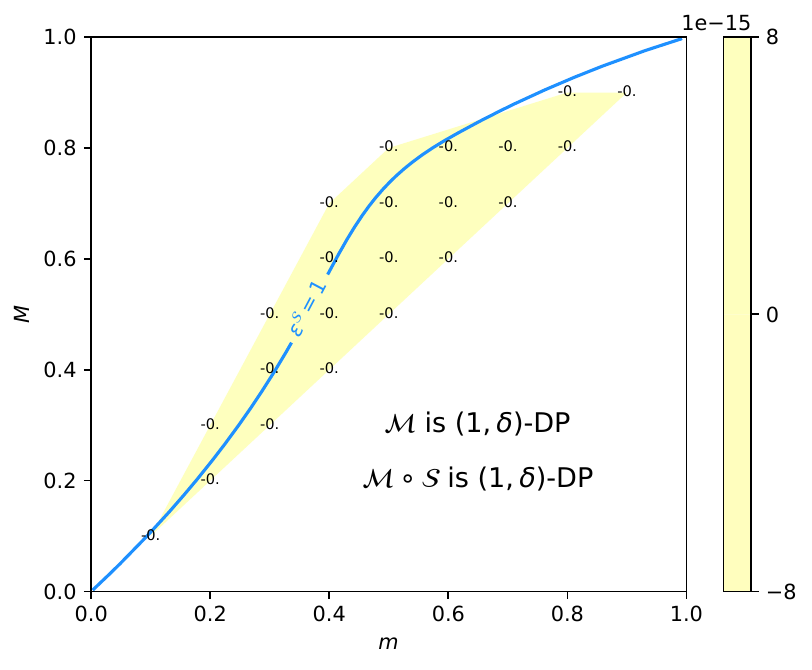}%
    \includegraphics[width=0.25\textwidth]{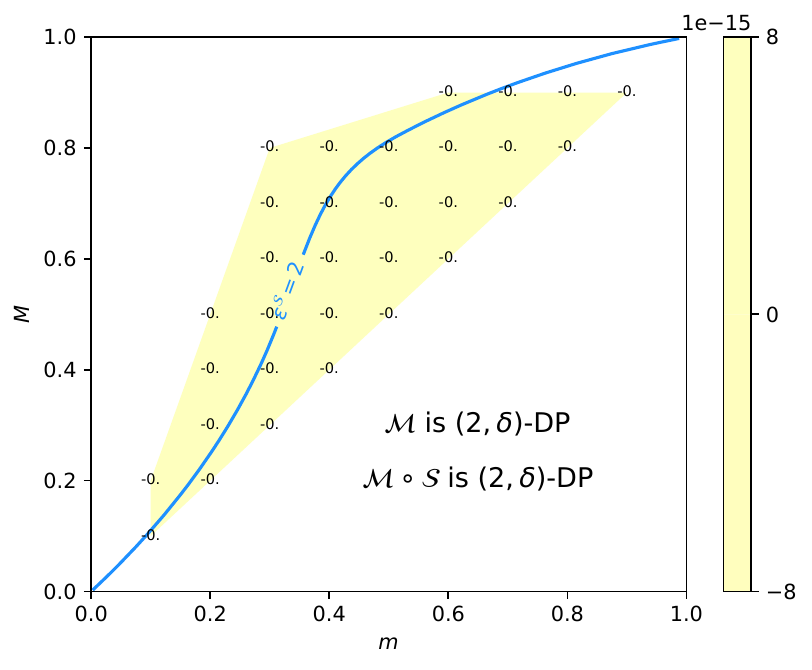}%
    \caption{The probability of outputting an incorrect mode of $\M$ minus that of $\M\circ\S$ at the same privacy levels. Results shown for the RNM-like variant with Gaussian noise over the \texttt{hours-per-week} column in the Adult database.}
    \label{fig:Experiment2-RNM-Gaussian-Adult-hours-per-week}
\end{figure}

\begin{figure}[H]
    \includegraphics[width=0.25\textwidth]{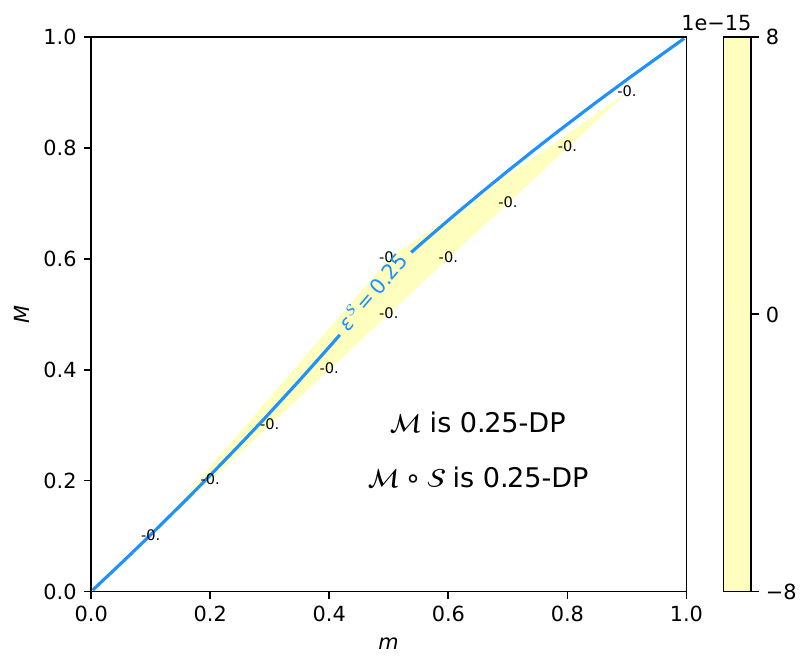}%
    \includegraphics[width=0.25\textwidth]{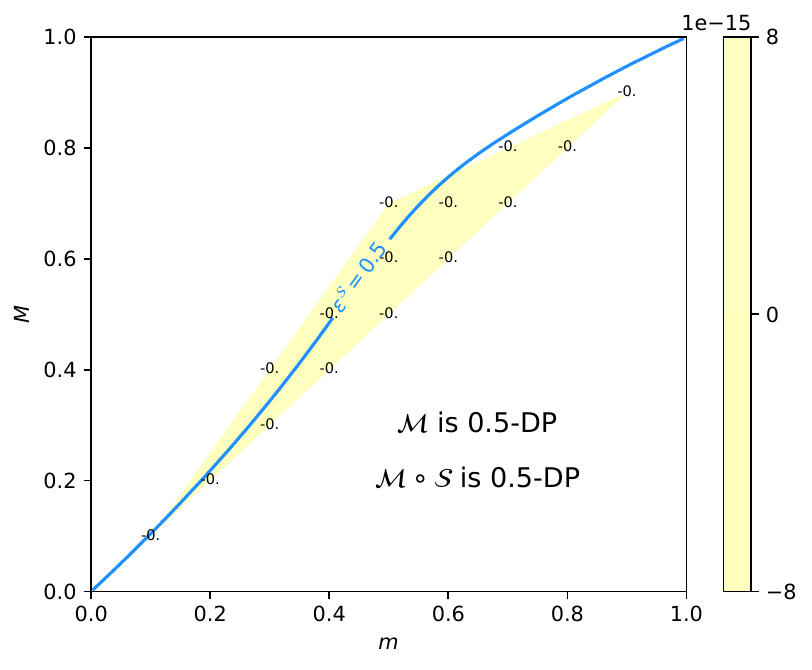}%
    \includegraphics[width=0.25\textwidth]{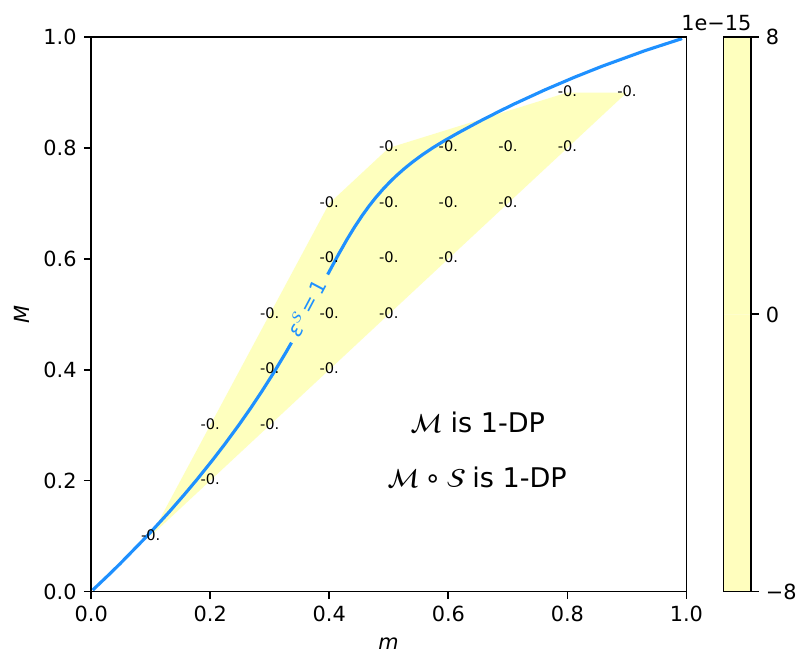}%
    \includegraphics[width=0.25\textwidth]{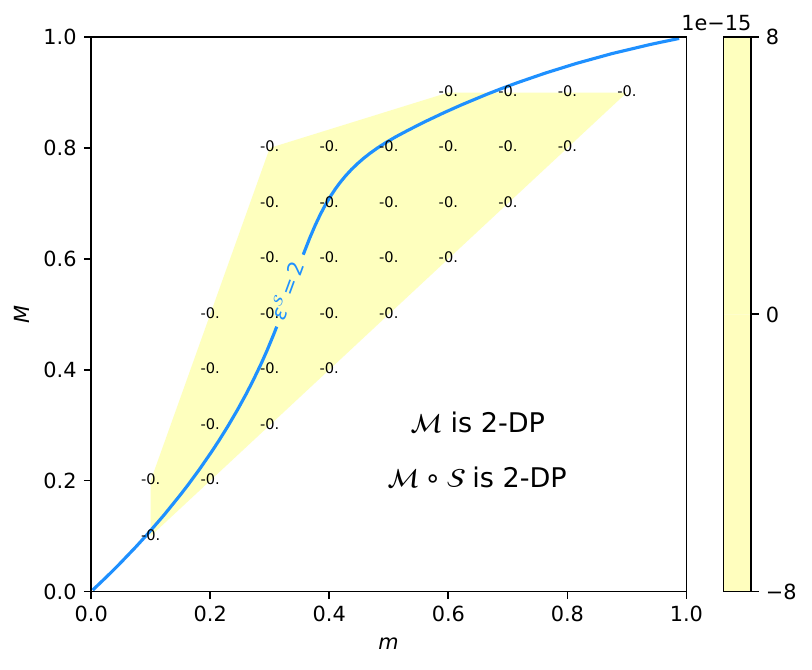}%
    \caption{The probability of outputting an incorrect mode of $\M$ minus that of $\M\circ\S$ at the same privacy levels. Results shown for the RNM with exponential noise over the \texttt{hours-per-week} column in the Adult database.}
    \label{fig:Experiment2-RNM-Exponential-Adult-hours-per-week}
\end{figure}

\begin{figure}[H]
    \includegraphics[width=0.25\textwidth]{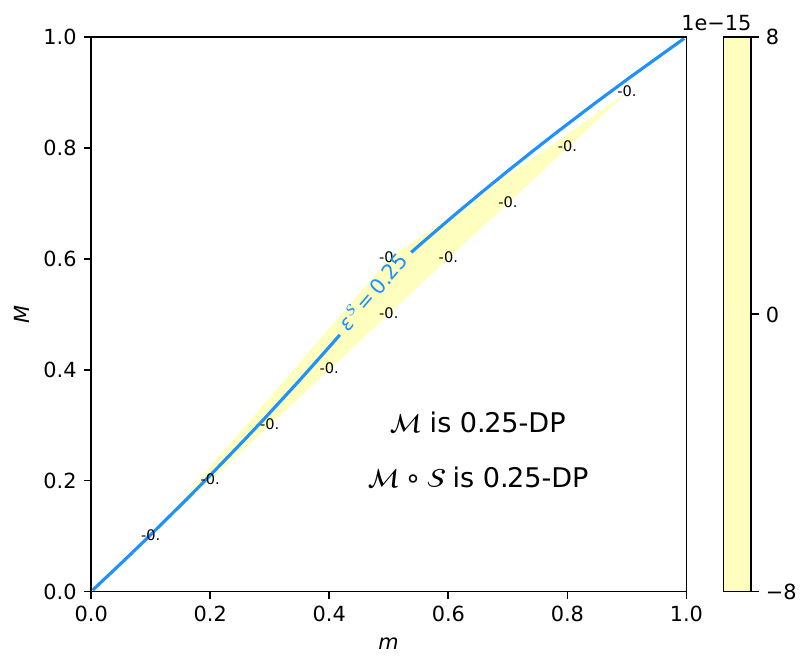}%
    \includegraphics[width=0.25\textwidth]{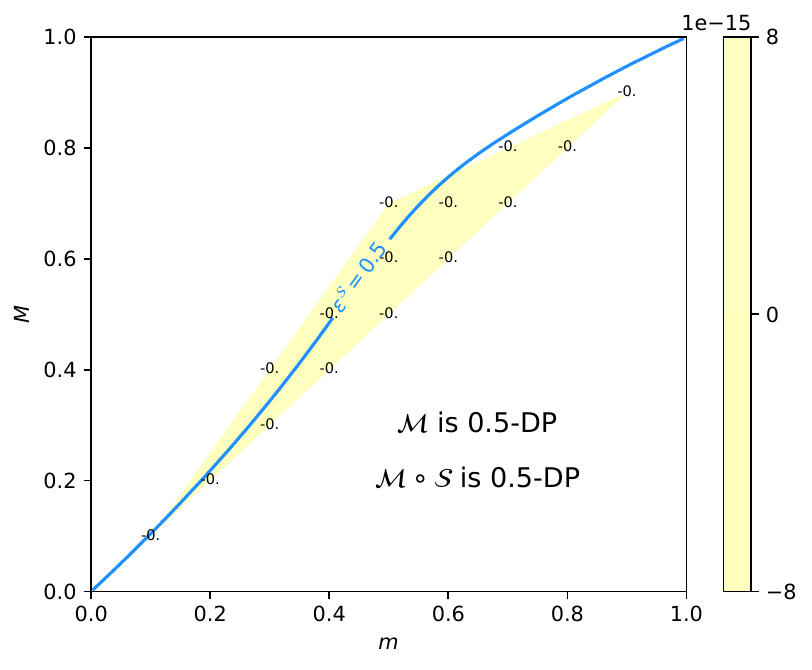}%
    \includegraphics[width=0.25\textwidth]{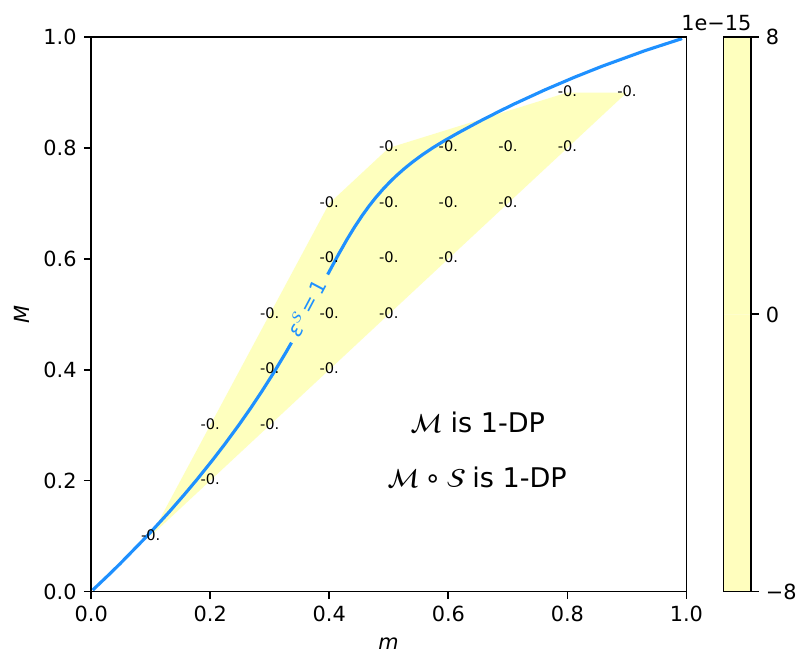}%
    \includegraphics[width=0.25\textwidth]{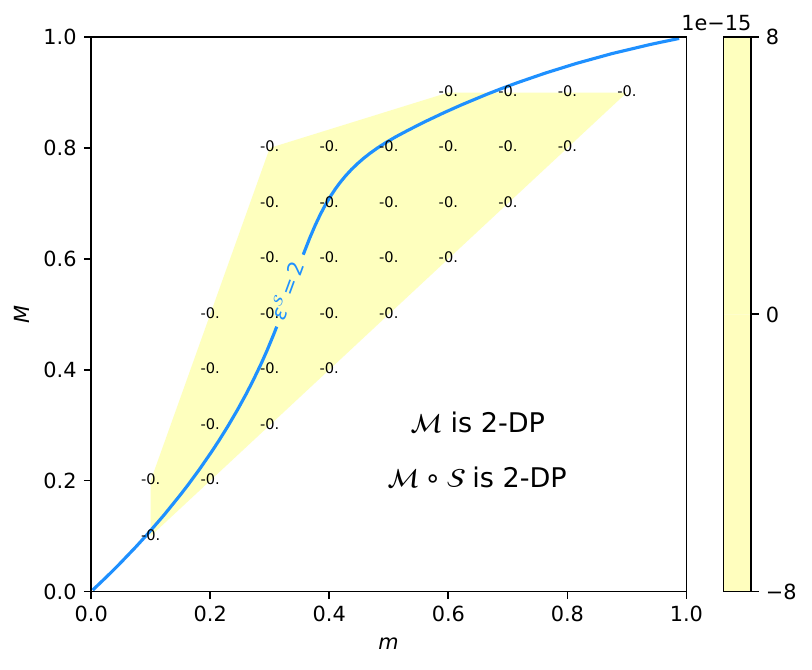}%
    \caption{The probability of outputting an incorrect mode of $\M$ minus that of $\M\circ\S$ at the same privacy levels. Results shown for the exponential mechanism over the \texttt{hours-per-week} column in the Adult database.}
    \label{fig:Experiment2-RNM-ExponentialMechanism-Adult-hours-per-week}
\end{figure}

\begin{figure}[H]
    \includegraphics[width=0.25\textwidth]{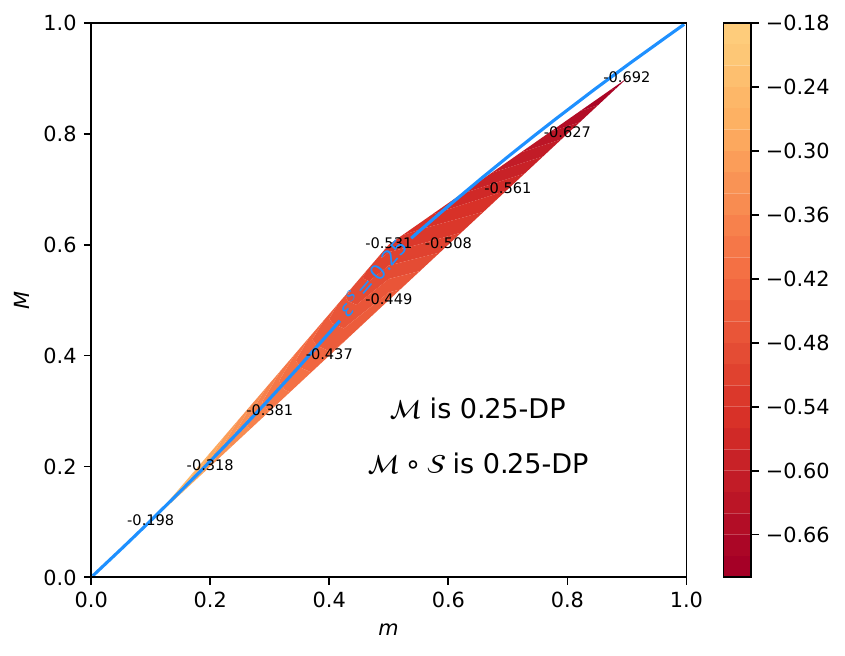}%
    \includegraphics[width=0.25\textwidth]{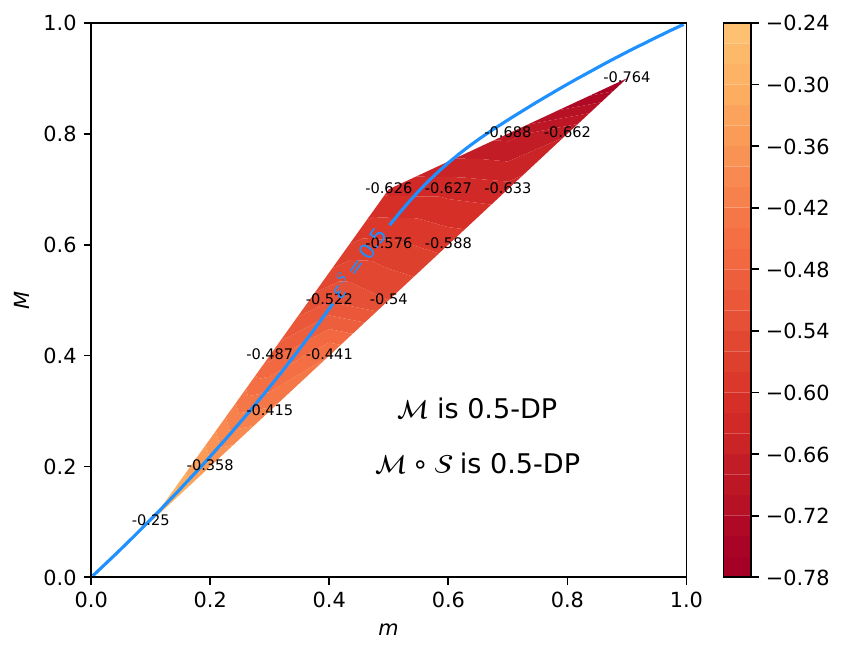}%
    \includegraphics[width=0.25\textwidth]{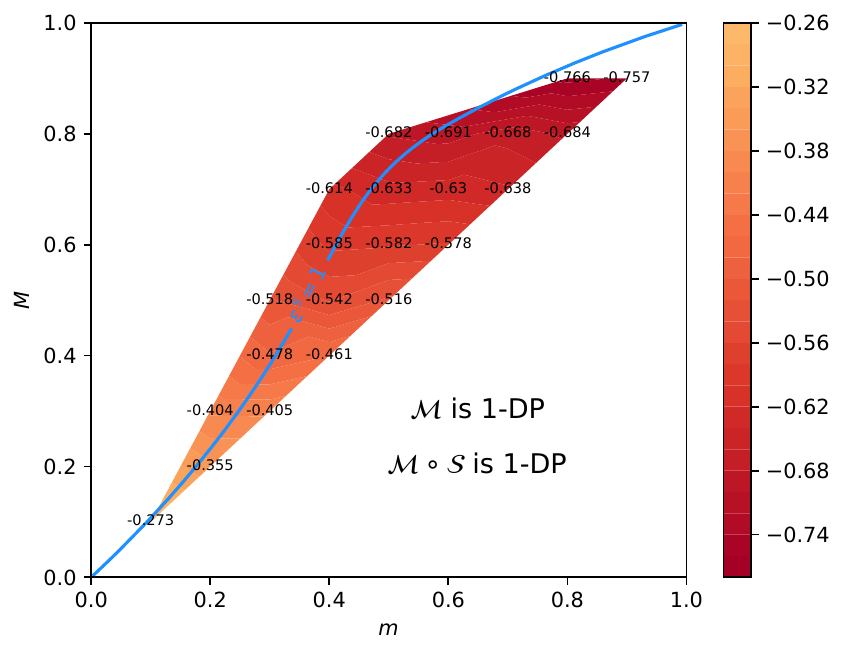}%
    \includegraphics[width=0.25\textwidth]{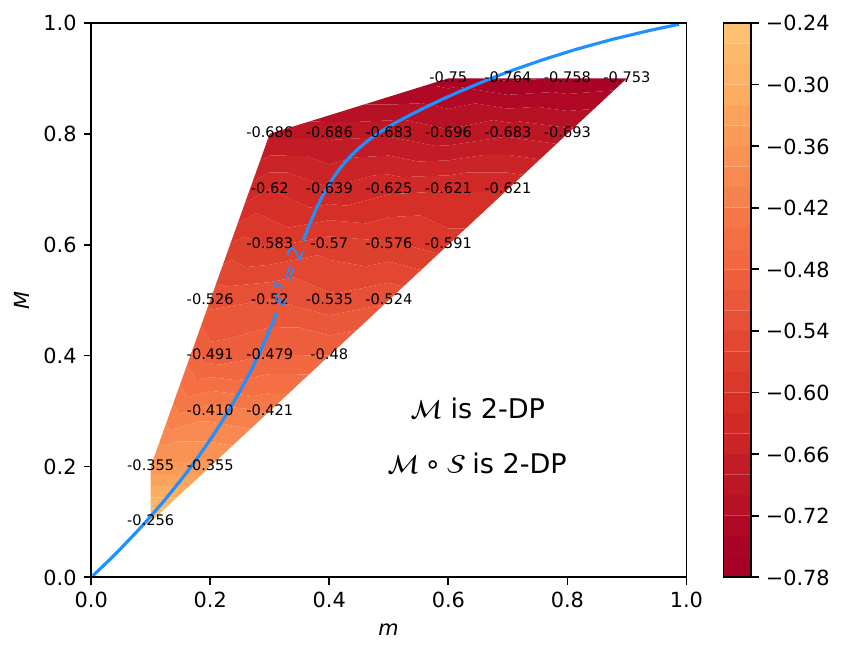}%
    \caption{The probability of outputting an incorrect mode of $\M$ minus that of $\M\circ\S$ at the same privacy levels. Results shown for the RNM with Laplace noise over the \texttt{Age} column in the Irish database.}
    \label{fig:Experiment2-RNM-Laplace-Irishn-Age}
\end{figure}

\begin{figure}[H]
    \includegraphics[width=0.25\textwidth]{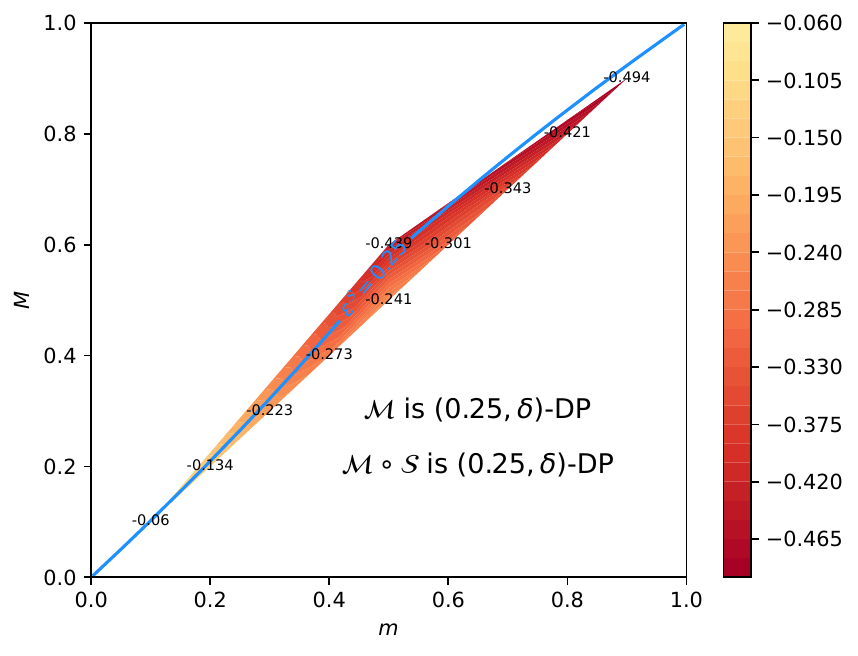}%
    \includegraphics[width=0.25\textwidth]{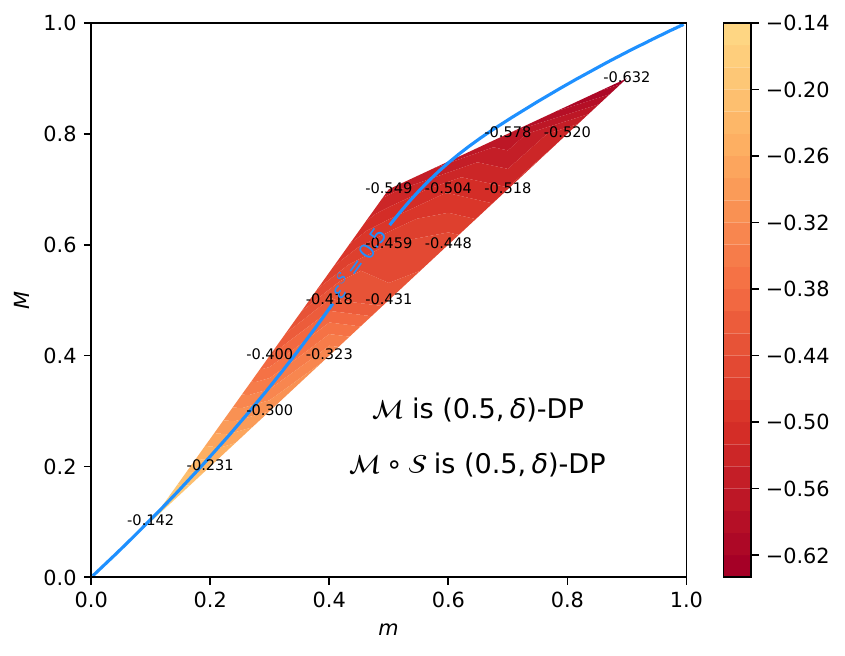}%
    \includegraphics[width=0.25\textwidth]{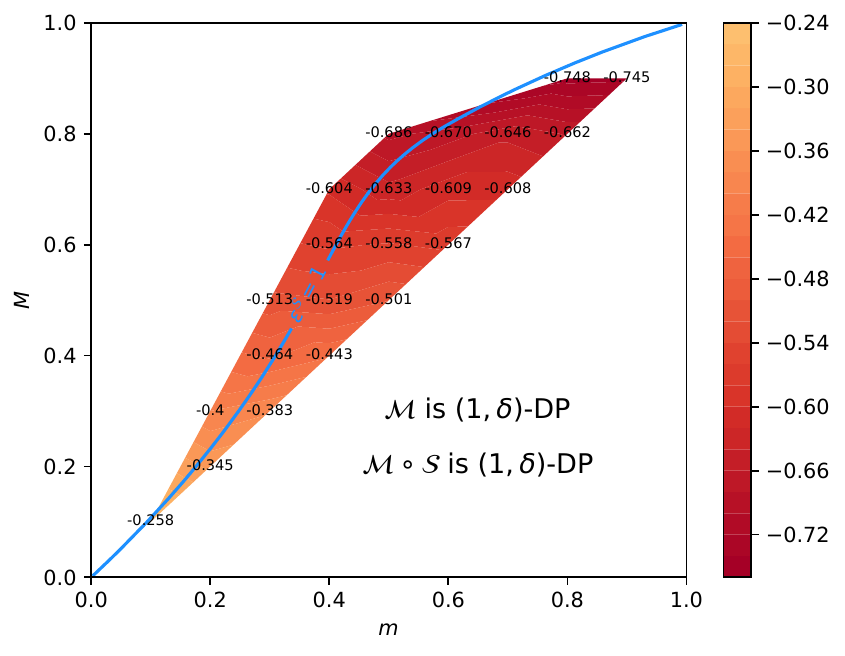}%
    \includegraphics[width=0.25\textwidth]{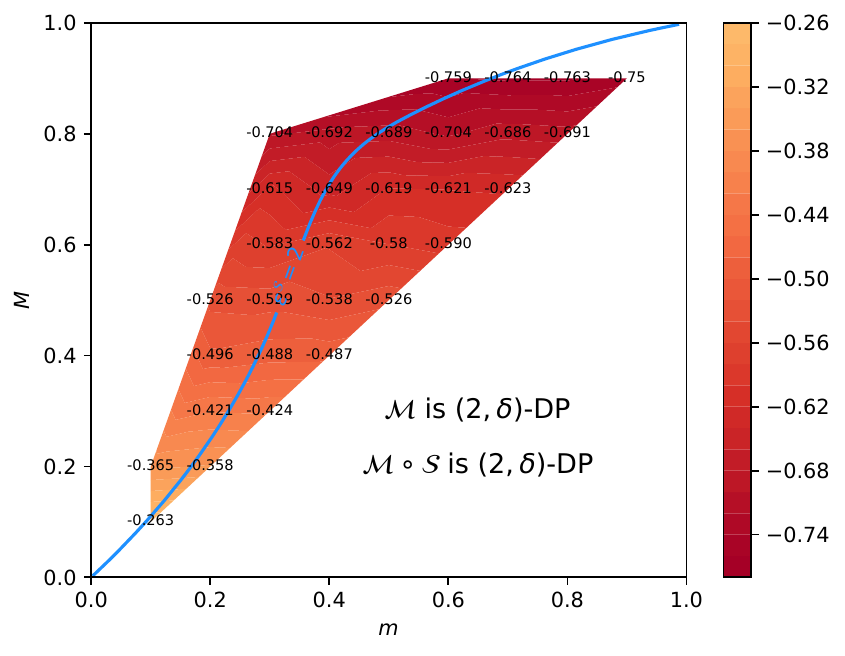}%
    \caption{The probability of outputting an incorrect mode of $\M$ minus that of $\M\circ\S$ at the same privacy levels. Results shown for the RNM-like variant with Gaussian noise over the \texttt{Age} column in the Irish database.}
    \label{fig:Experiment2-RNM-Gaussian-Irishn-Age}
\end{figure}

\begin{figure}[H]
    \includegraphics[width=0.25\textwidth]{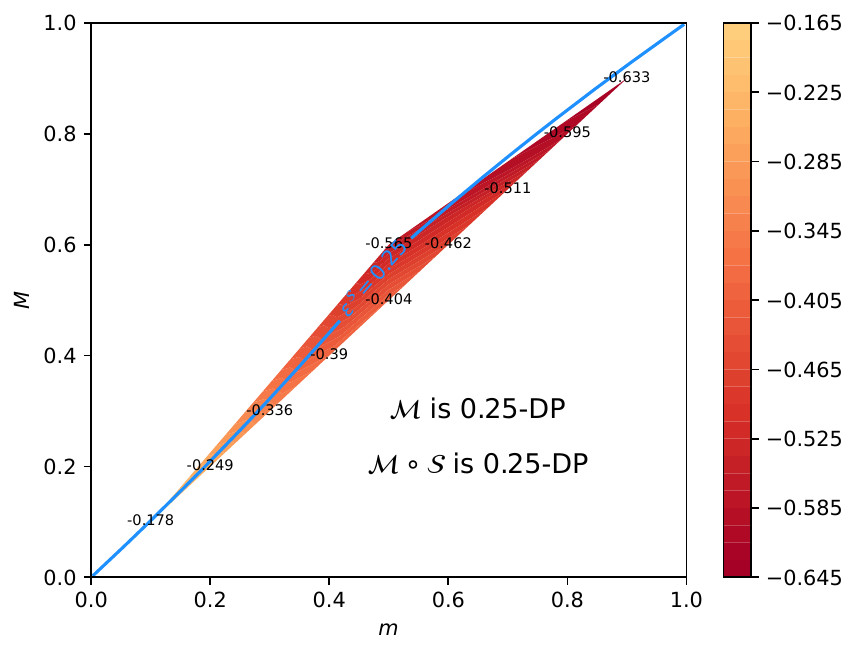}%
    \includegraphics[width=0.25\textwidth]{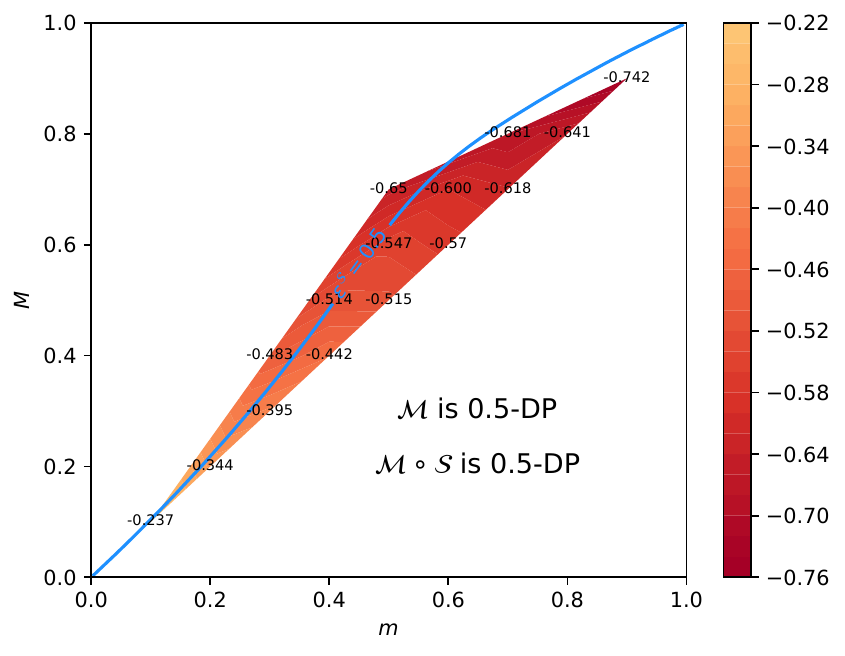}%
    \includegraphics[width=0.25\textwidth]{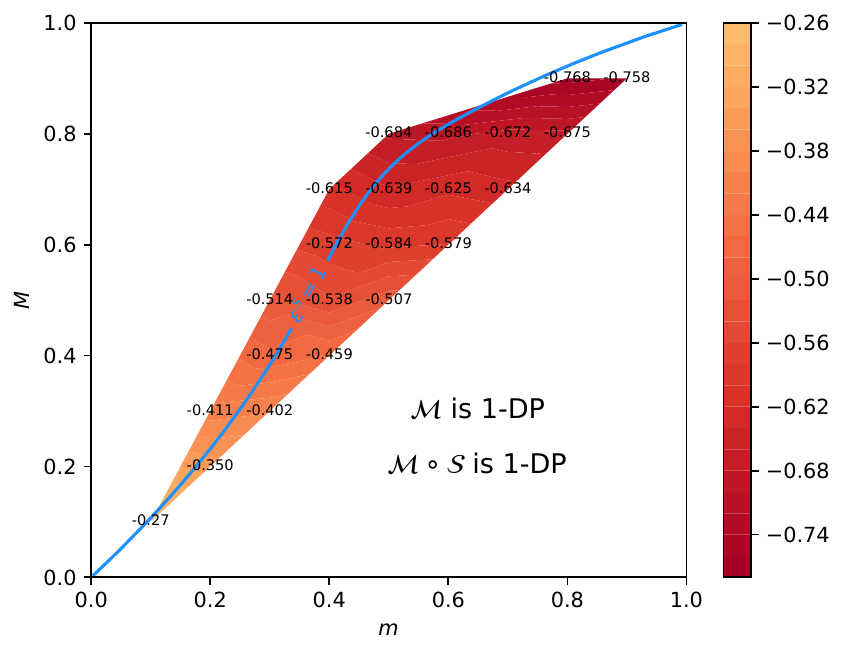}%
    \includegraphics[width=0.25\textwidth]{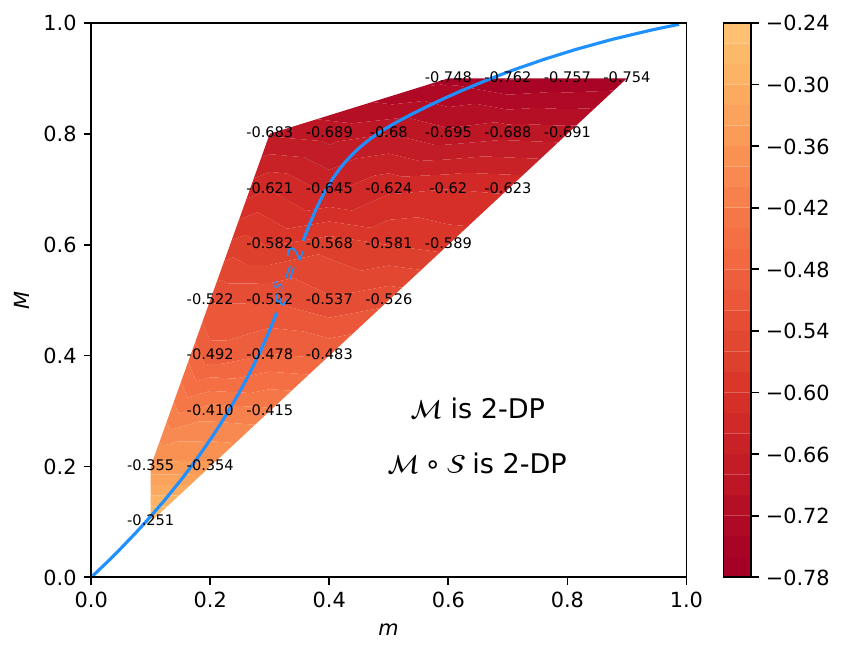}%
    \caption{The probability of outputting an incorrect mode of $\M$ minus that of $\M\circ\S$ at the same privacy levels. Results shown for the RNM with exponential noise over the \texttt{Age} column in the Irish database.}
    \label{fig:Experiment2-RNM-Exponential-Irishn-Age}
\end{figure}

\begin{figure}[H]
    \includegraphics[width=0.25\textwidth]{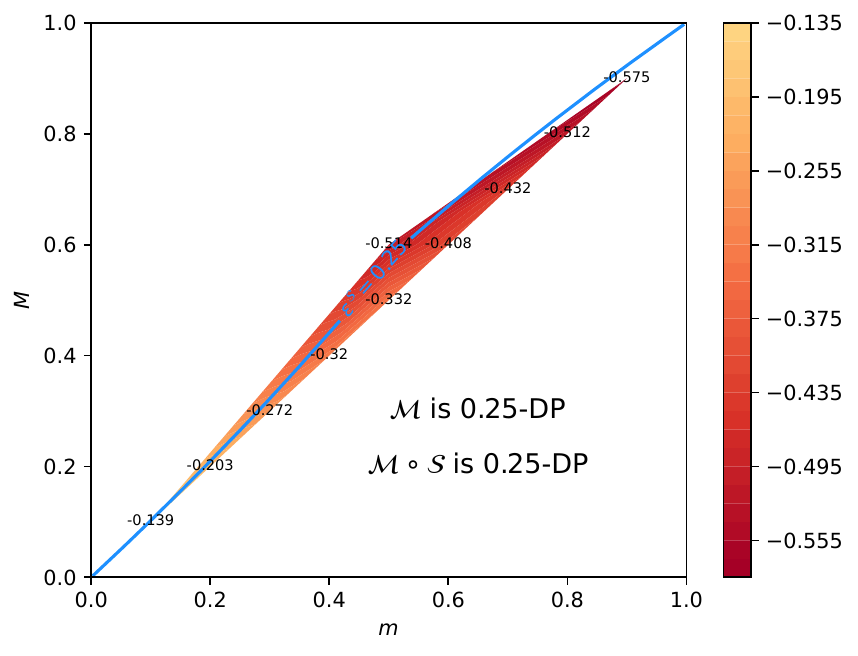}%
    \includegraphics[width=0.25\textwidth]{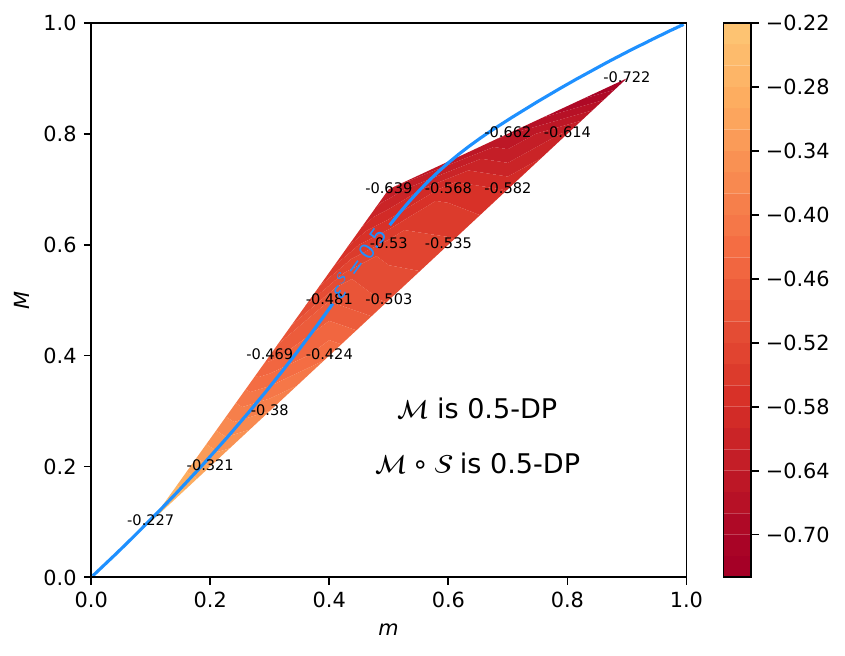}%
    \includegraphics[width=0.25\textwidth]{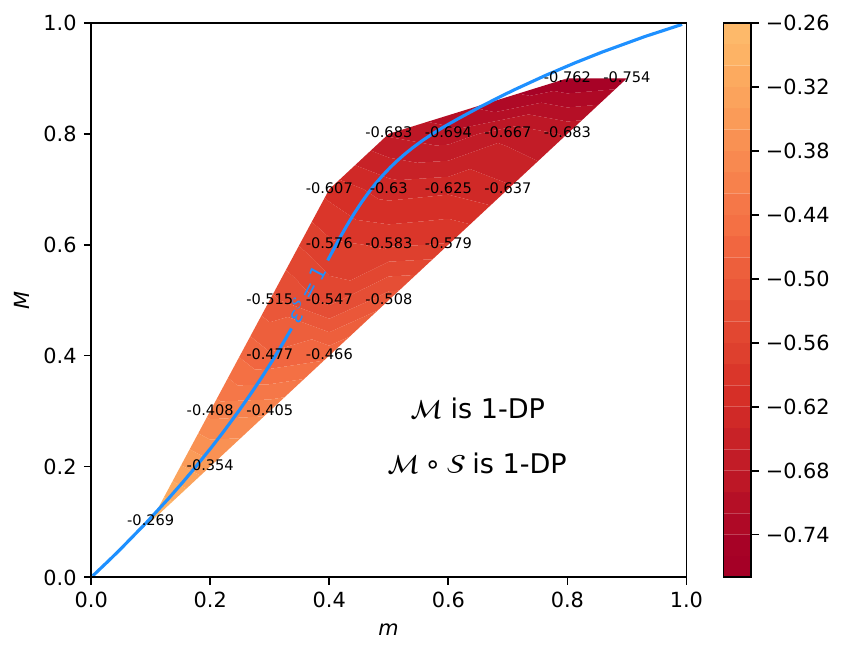}%
    \includegraphics[width=0.25\textwidth]{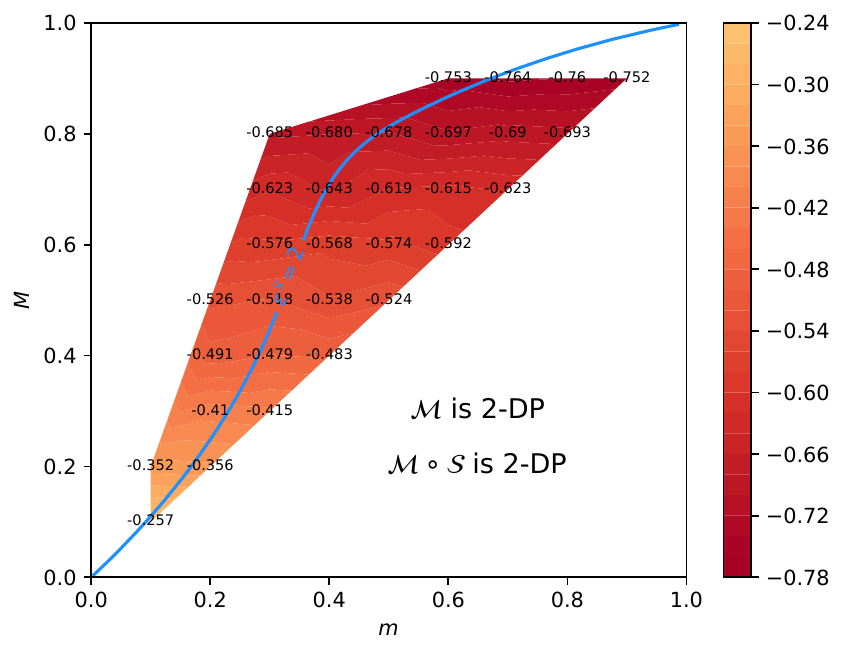}%
    \caption{The probability of outputting an incorrect mode of $\M$ minus that of $\M\circ\S$ at the same privacy levels. Results shown for the exponential mechanism over the \texttt{Age} column in the Irish database.}
    \label{fig:Experiment2-RNM-ExponentialMechanism-Irishn-Age}
\end{figure}

\begin{figure}[H]
    \includegraphics[width=0.25\textwidth]{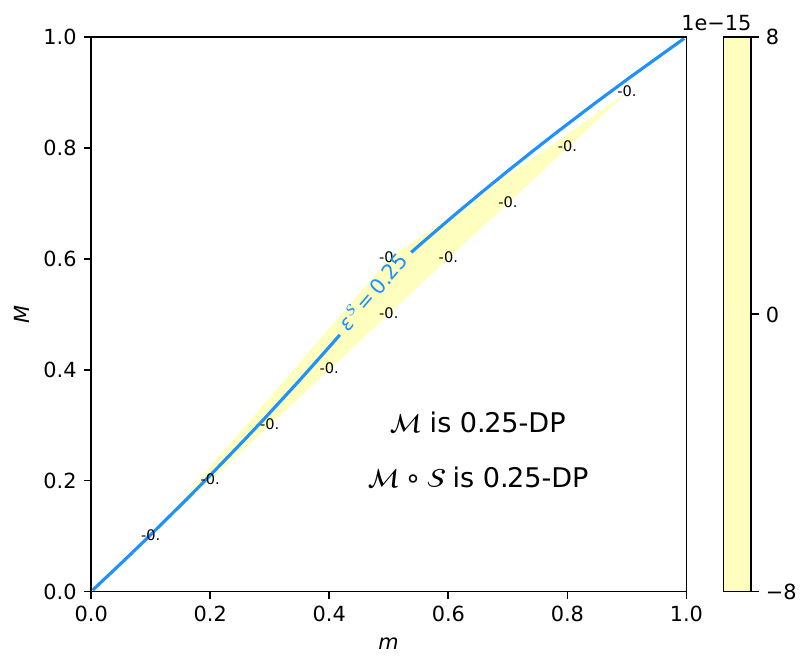}%
    \includegraphics[width=0.25\textwidth]{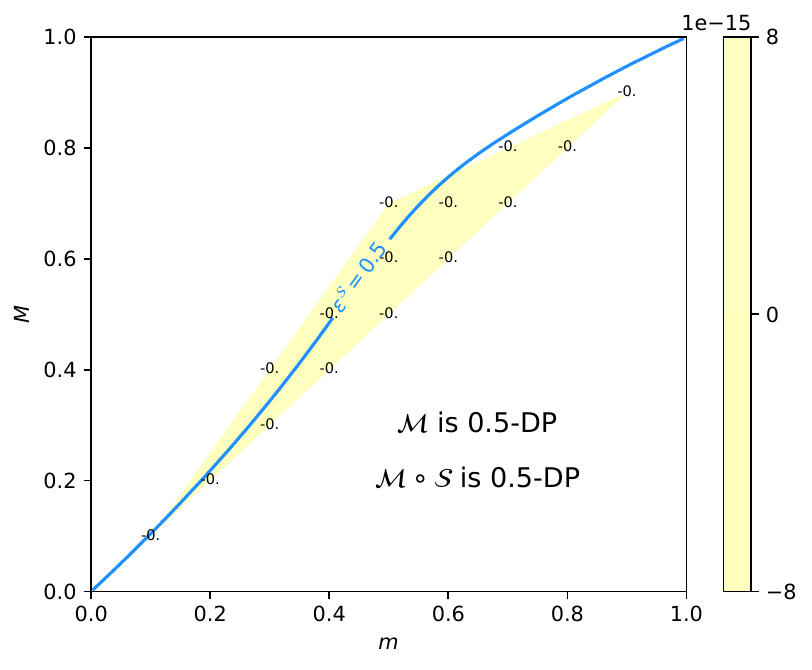}%
    \includegraphics[width=0.25\textwidth]{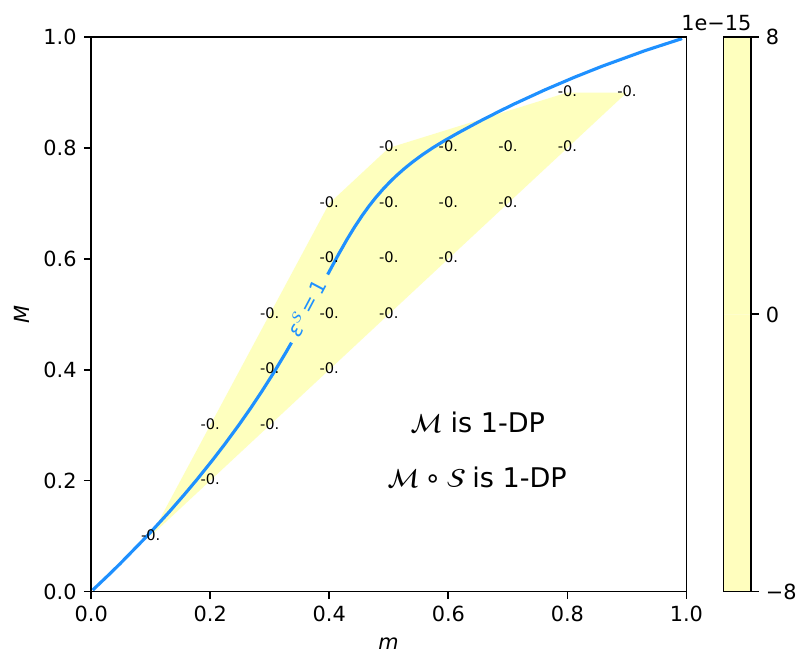}%
    \includegraphics[width=0.25\textwidth]{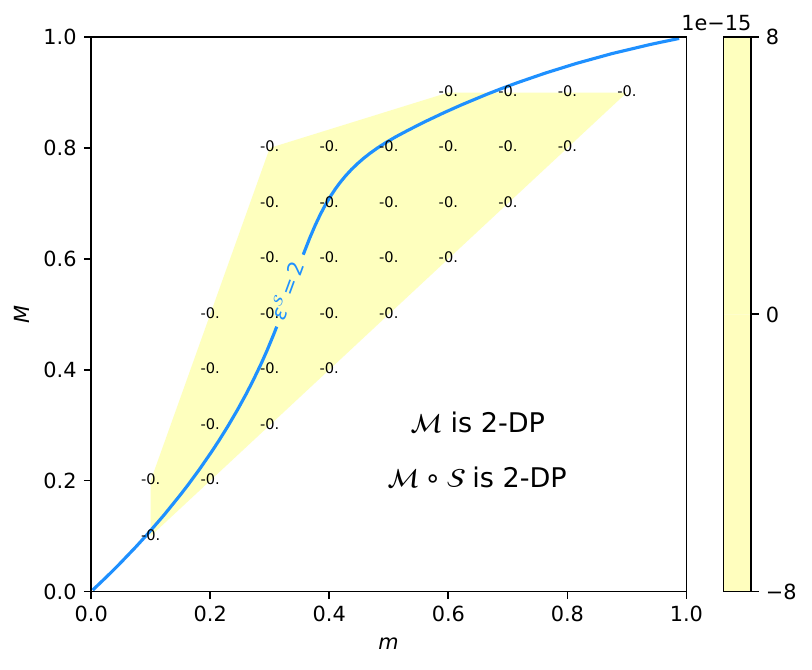}%
    \caption{The probability of outputting an incorrect mode of $\M$ minus that of $\M\circ\S$ at the same privacy levels. Results shown for the RNM with Laplace noise over the \texttt{HighestEducationCompleted} column in the Irish database.}
    \label{fig:Experiment2-RNM-Laplace-Irishn-HighestEducationCompleted}
\end{figure}

\begin{figure}[H]
    \includegraphics[width=0.25\textwidth]{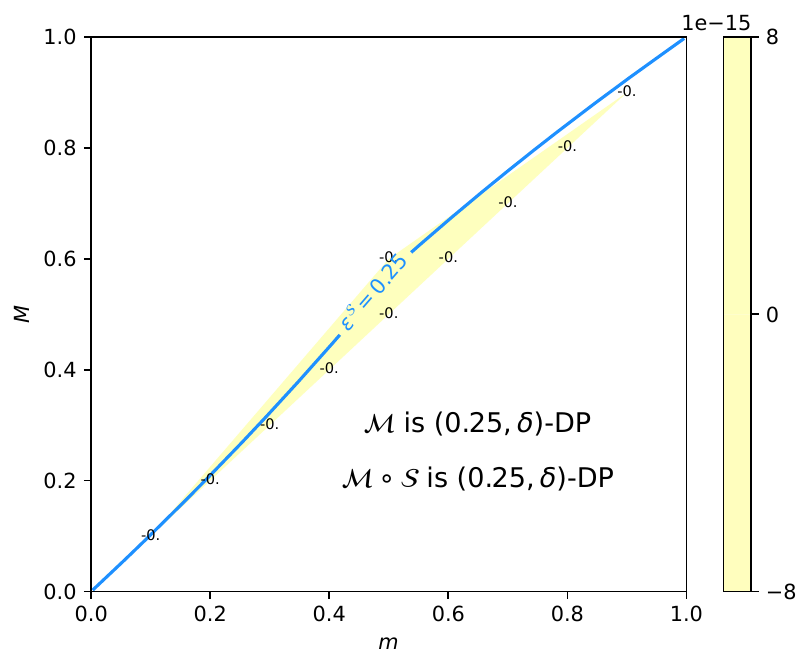}%
    \includegraphics[width=0.25\textwidth]{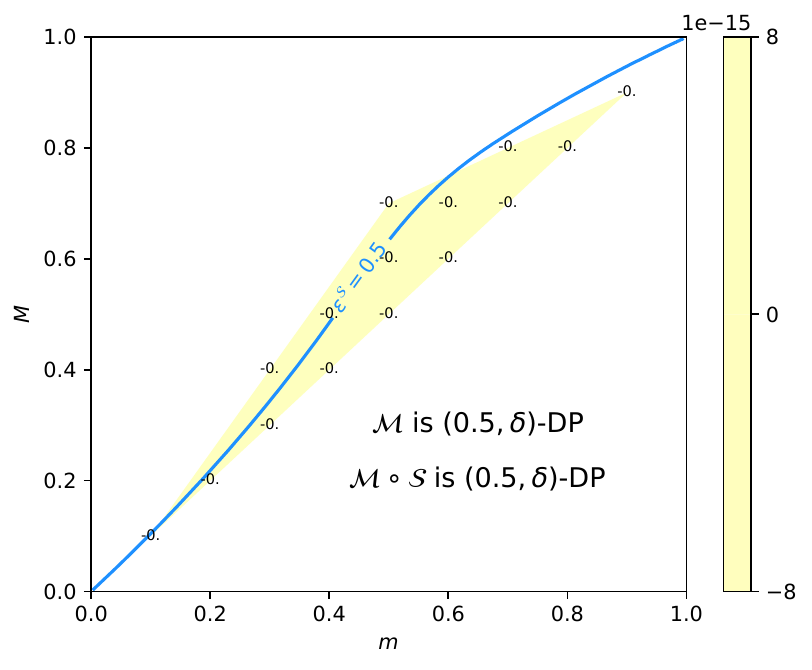}%
    \includegraphics[width=0.25\textwidth]{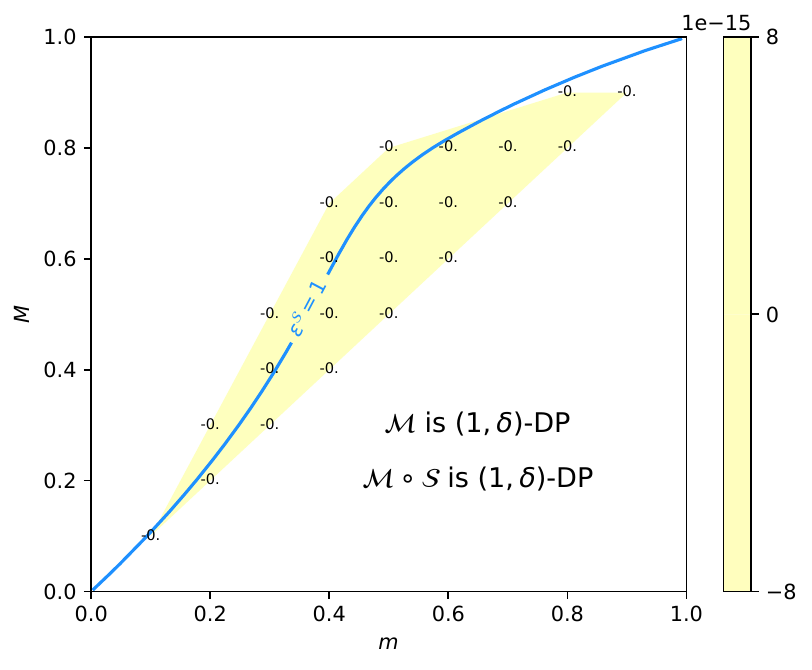}%
    \includegraphics[width=0.25\textwidth]{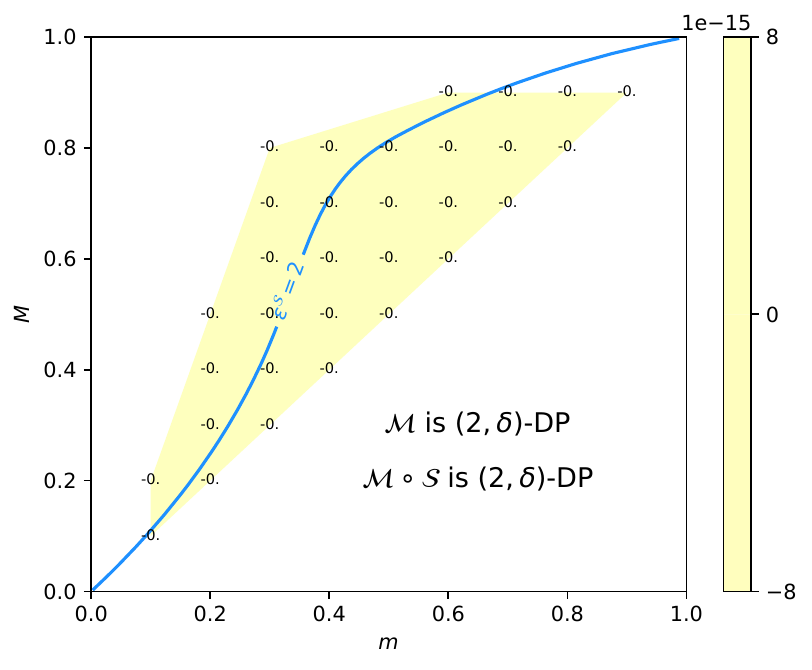}%
    \caption{The probability of outputting an incorrect mode of $\M$ minus that of $\M\circ\S$ at the same privacy levels. Results shown for the RNM-like variant with Gaussian noise over the \texttt{HighestEducationCompleted} column in the Irish database.}
    \label{fig:Experiment2-RNM-Gaussian-Irishn-HighestEducationCompleted}
\end{figure}

\begin{figure}[H]
    \includegraphics[width=0.25\textwidth]{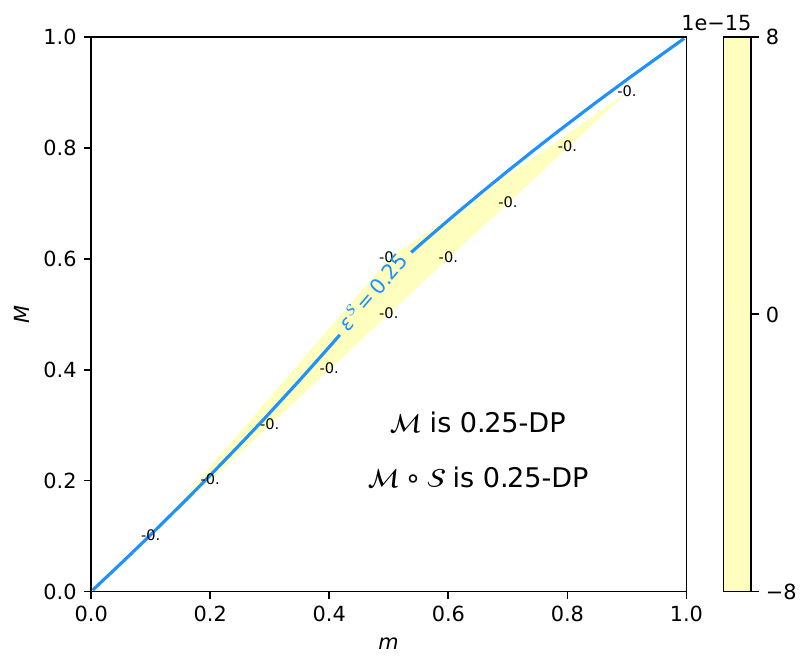}%
    \includegraphics[width=0.25\textwidth]{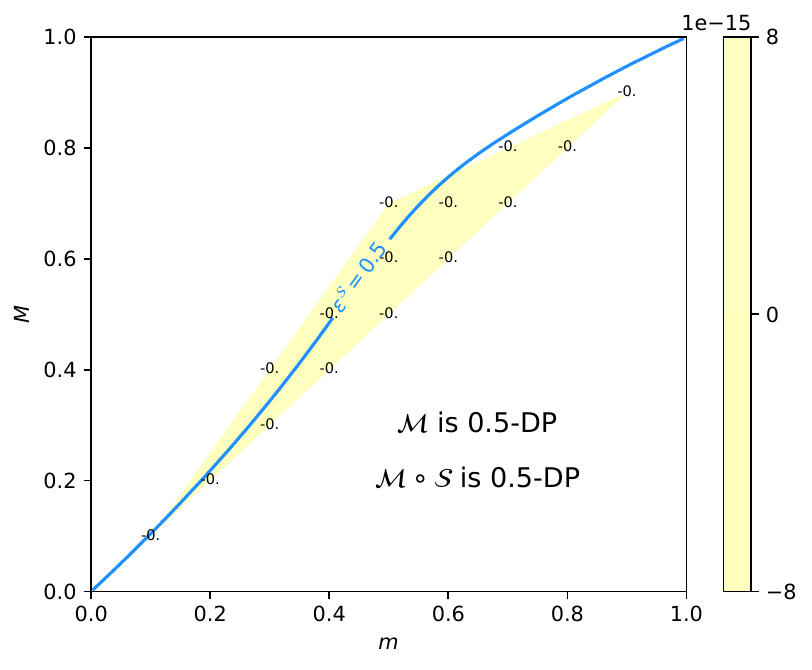}%
    \includegraphics[width=0.25\textwidth]{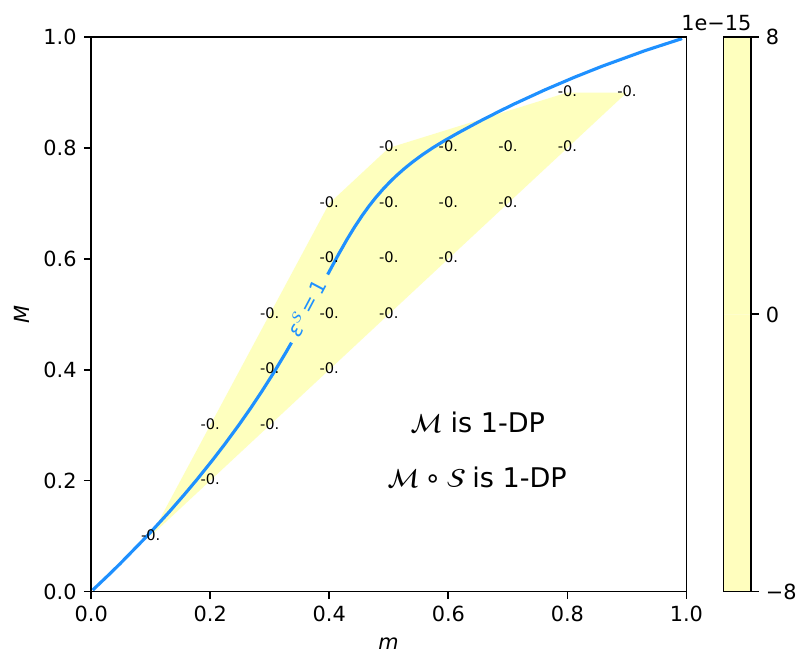}%
    \includegraphics[width=0.25\textwidth]{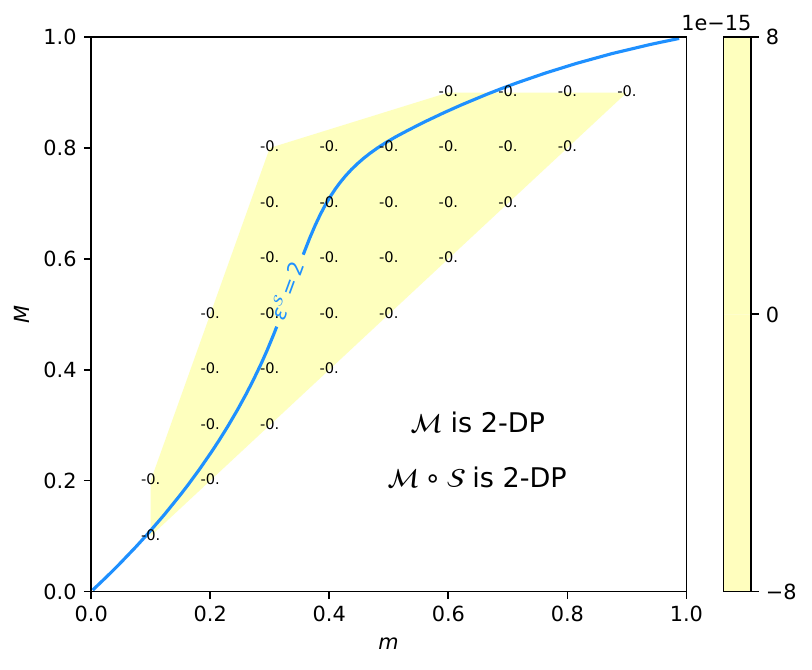}%
    \caption{The probability of outputting an incorrect mode of $\M$ minus that of $\M\circ\S$ at the same privacy levels. Results shown for the RNM with exponential noise over the \texttt{HighestEducationCompleted} column in the Irish database.}
    \label{fig:Experiment2-RNM-Exponential-Irishn-HighestEducationCompleted}
\end{figure}

\begin{figure}[H]
    \includegraphics[width=0.25\textwidth]{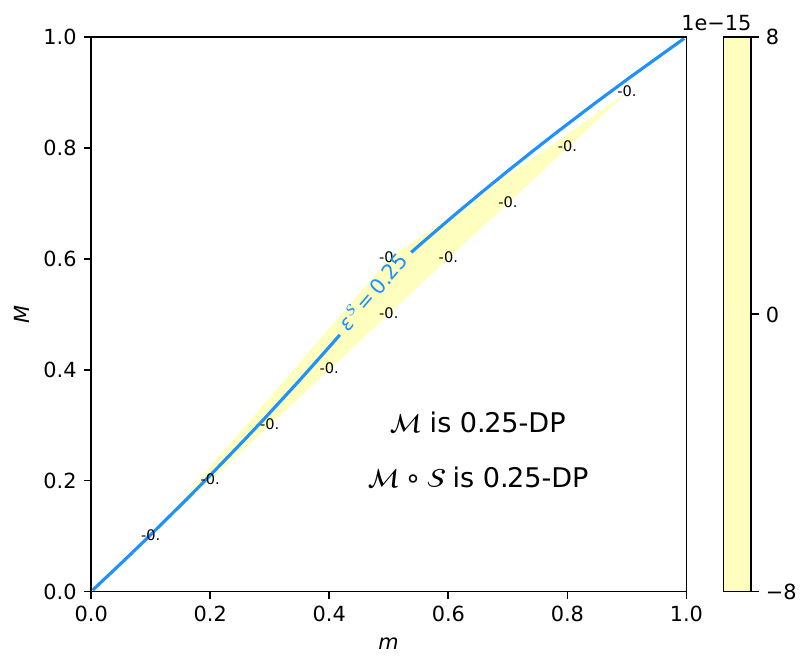}%
    \includegraphics[width=0.25\textwidth]{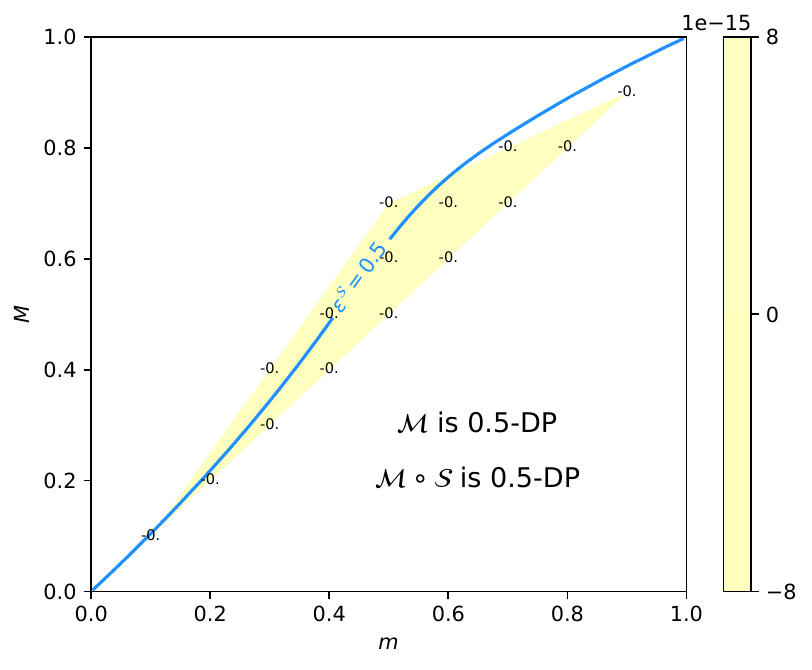}%
    \includegraphics[width=0.25\textwidth]{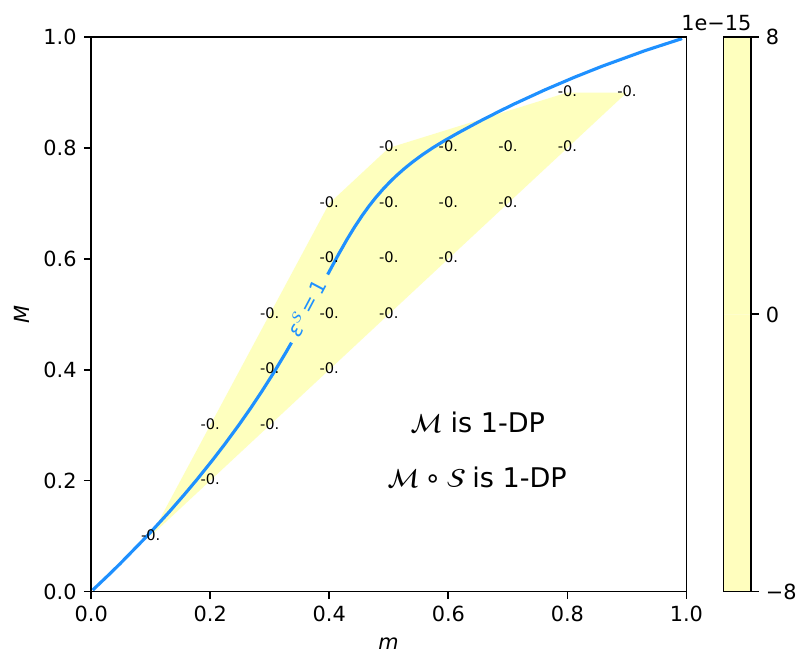}%
    \includegraphics[width=0.25\textwidth]{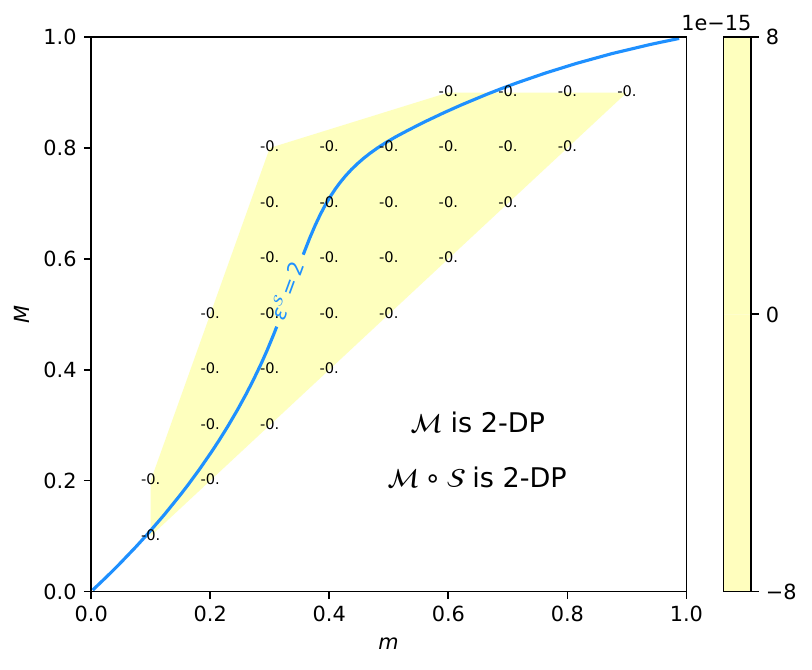}%
    \caption{The probability of outputting an incorrect mode of $\M$ minus that of $\M\circ\S$ at the same privacy levels. Results shown for the exponential mechanism over the \texttt{HighestEducationCompleted} column in the Irish database.}
    \label{fig:Experiment2-RNM-ExponentialMechanism-Irishn-HighestEducationCompleted}
\end{figure}

\subsection{Plots for the Clustering Mechanisms}\label{sec:plots:SuppressionwithEpsDeltaChange3}

\begin{figure}[H]
    \includegraphics[width=0.25\textwidth]{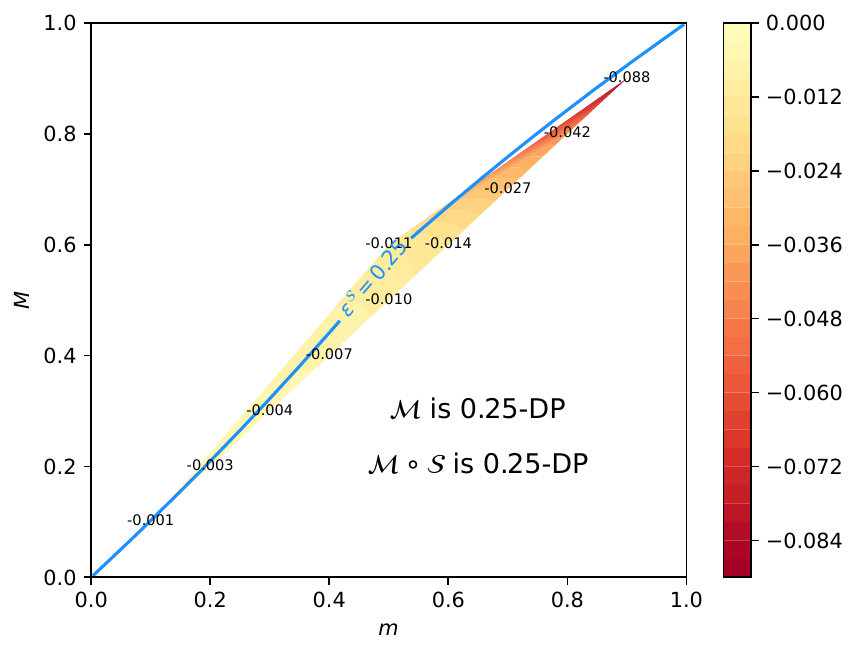}%
    \includegraphics[width=0.25\textwidth]{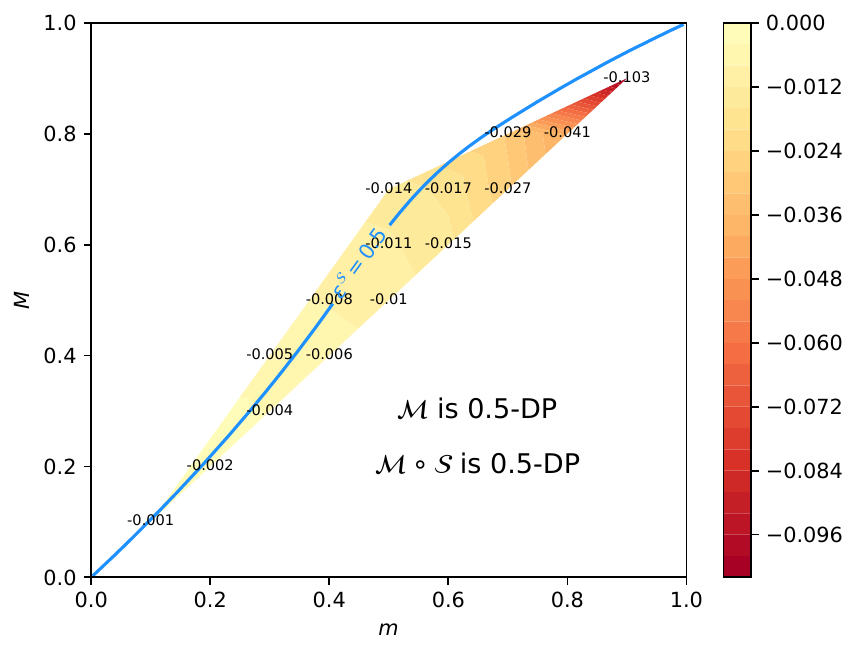}%
    \includegraphics[width=0.25\textwidth]{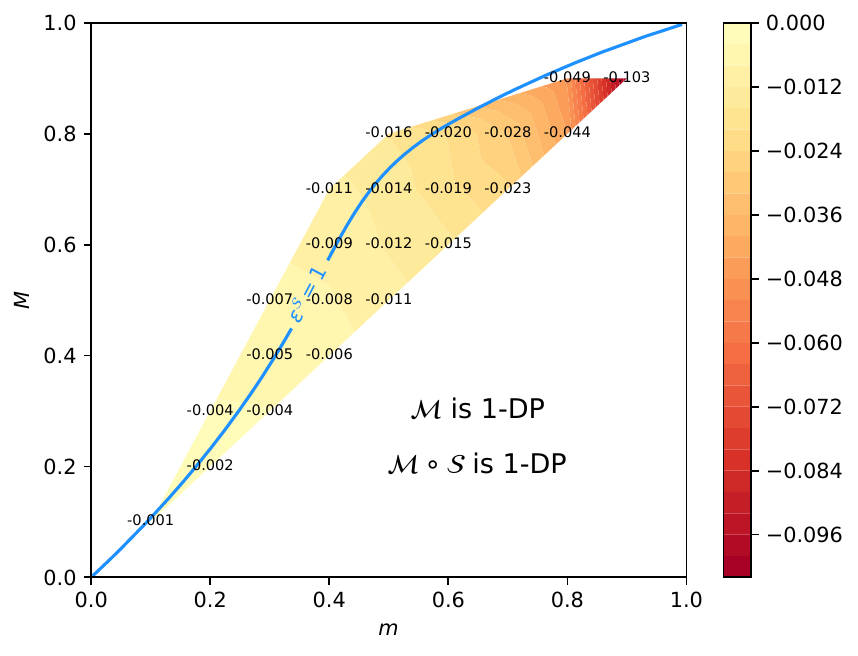}%
    \includegraphics[width=0.25\textwidth]{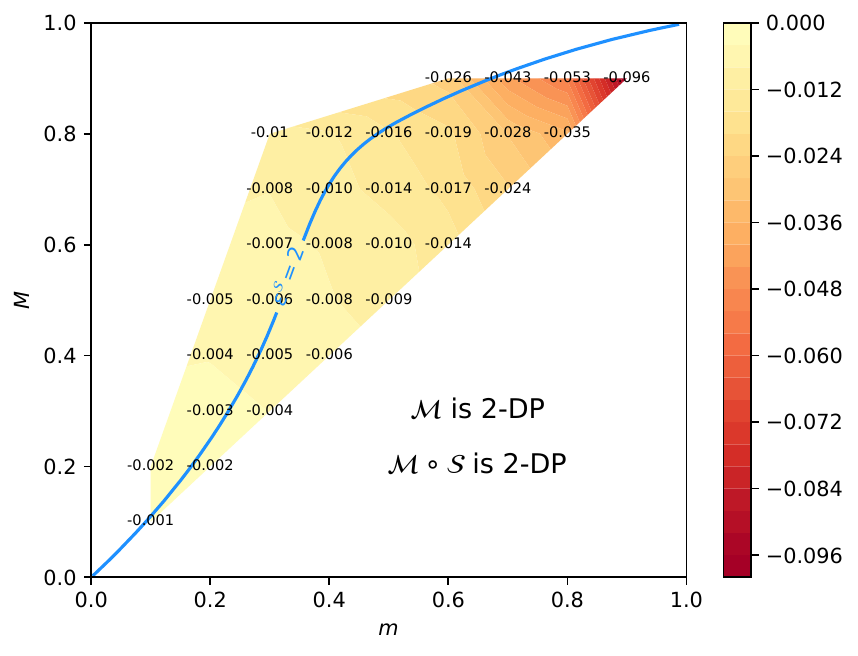}%
    \caption{The average cost of $\M$ minus that of $\M\circ\S$ at the same privacy levels. Results shown for $k$-median over our synthetic database.}
    \label{fig:Experiment2-Clusteringkmedian}
\end{figure}

\begin{figure}[H]
    \includegraphics[width=0.25\textwidth]{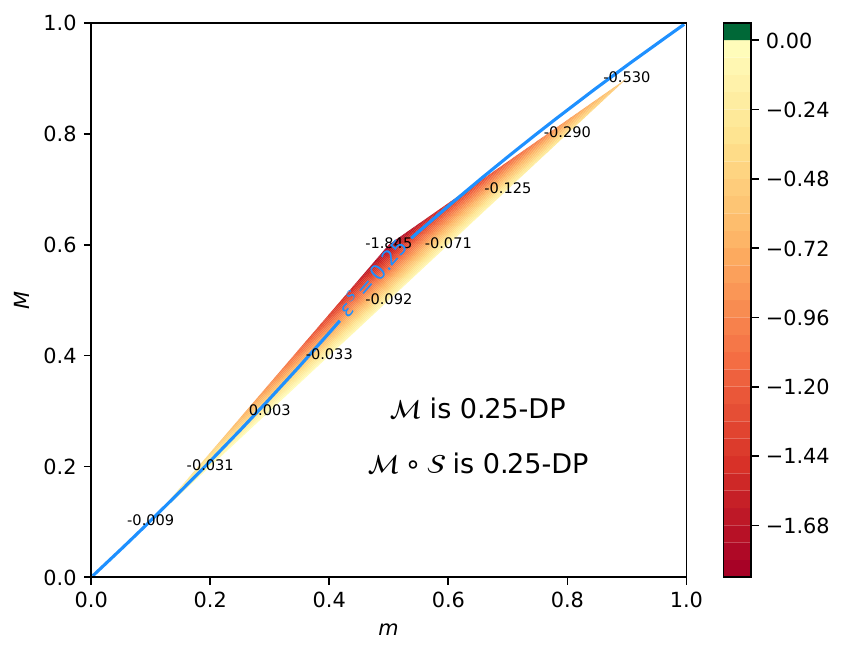}%
    \includegraphics[width=0.25\textwidth]{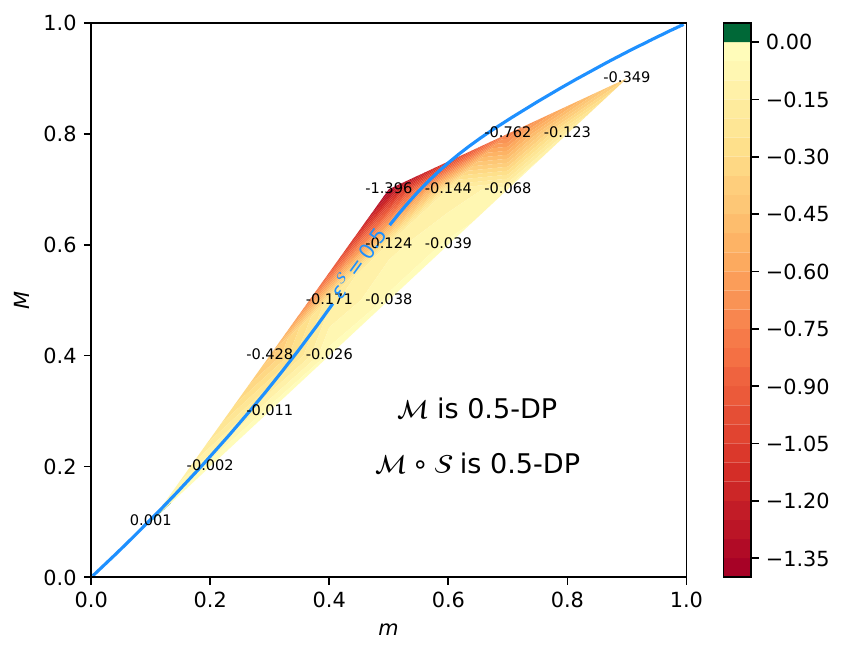}%
    \includegraphics[width=0.25\textwidth]{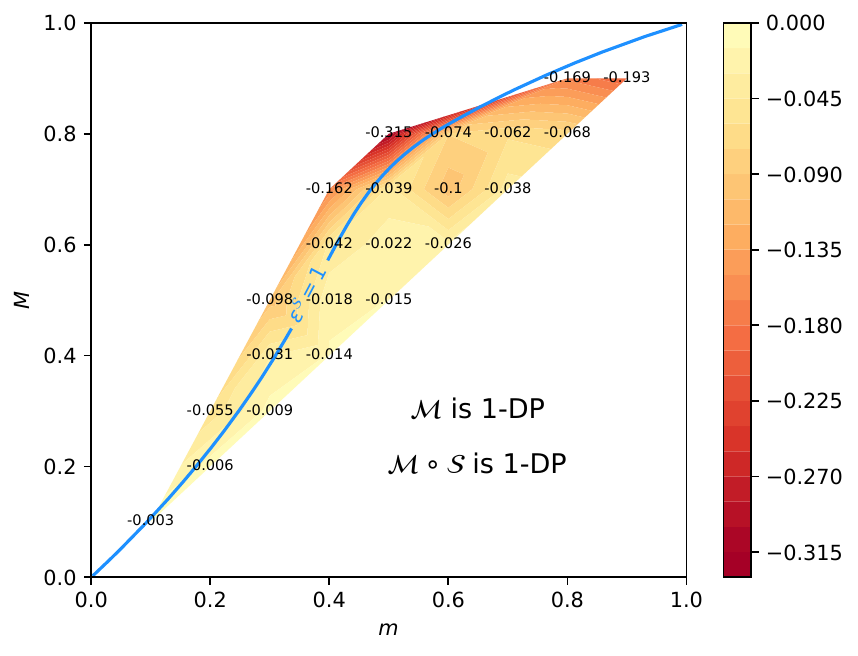}%
    \includegraphics[width=0.25\textwidth]{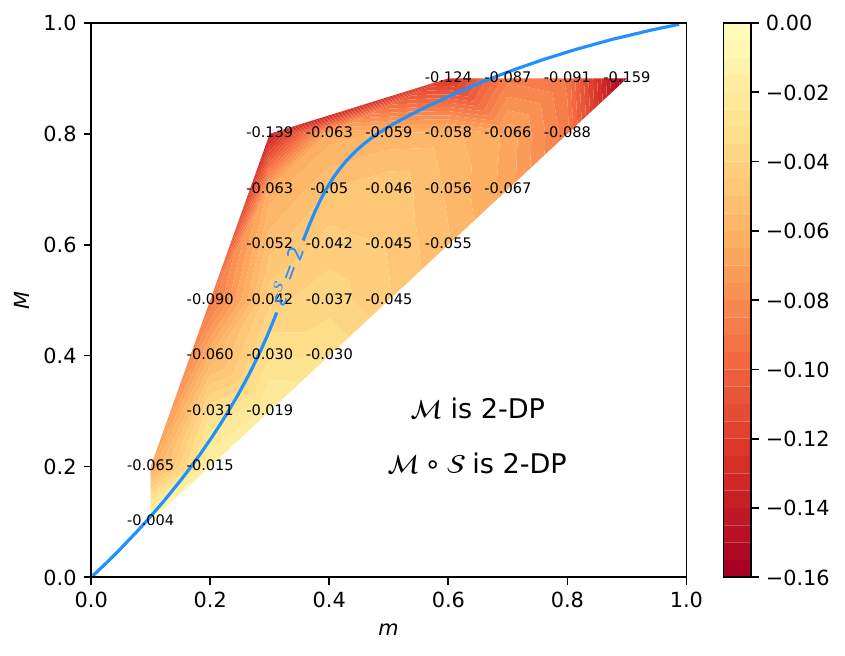}%
    \caption{The normalized intracluster variance of $\M$ minus that of $\M\circ\S$ at the same privacy levels. Results shown for DPLloyd over the six numerical columns of the Adult database.}
    \label{fig:Experiment2-ClusteringDPLloyd}
\end{figure}

%% file: plotsgallery/plotsExperiment1.tex
\subsection{Plots of the Utility Difference between the Mechanisms \textit{without} the Noise Reduction for the Mean Computation}\label{sec:plots:SuppressionwithoutEpsDeltaChange1}

\begin{figure}[H]
    \includegraphics[width=0.25\textwidth]{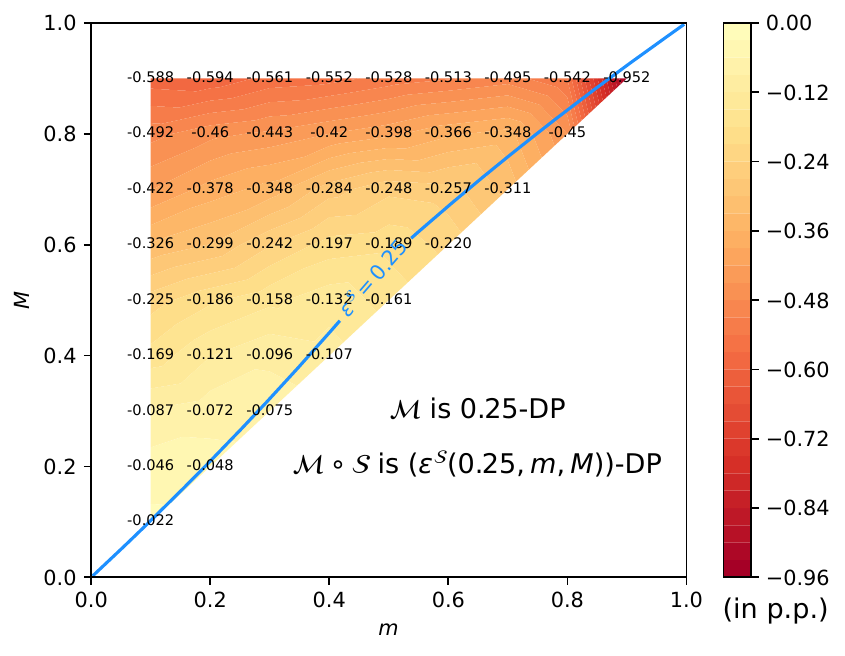}%
    \includegraphics[width=0.25\textwidth]{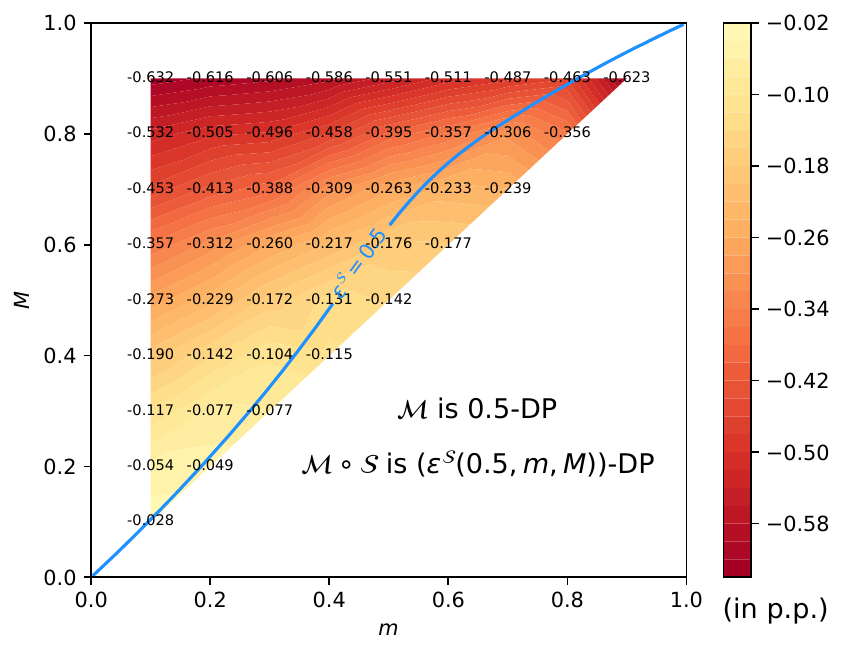}%
    \includegraphics[width=0.25\textwidth]{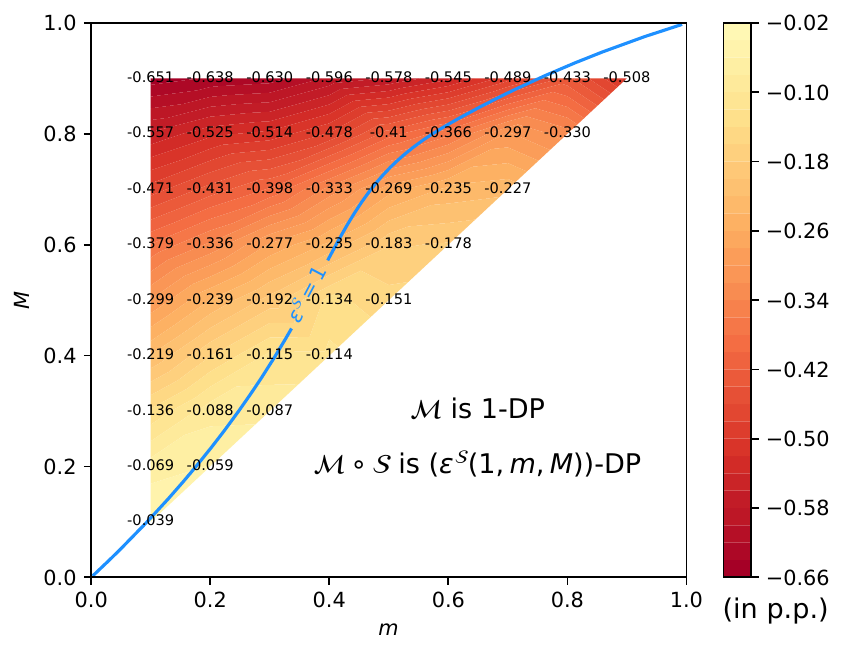}%
    \includegraphics[width=0.25\textwidth]{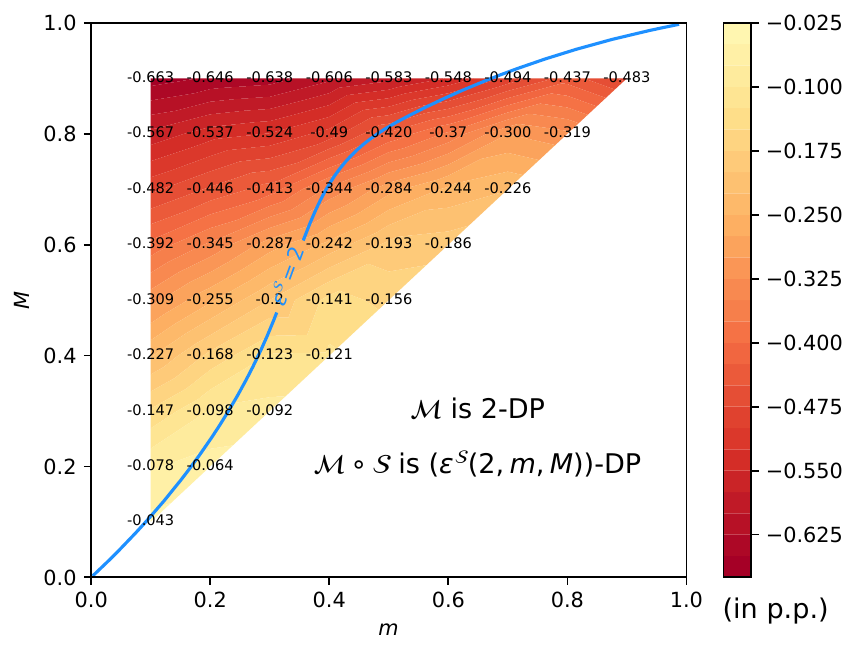}%
    \caption{The mean percent error (MPE) of $\M$ minus that of $\M\circ\S$ without the noise reduction. Results shown for the NoisyAverage with Laplace mechanisms over the \texttt{age} column in the Adult database.}
    \label{fig:Experiment1-NoisyAverage-Laplace-Adult-age}
\end{figure}

\begin{figure}[H]
    \includegraphics[width=0.25\textwidth]{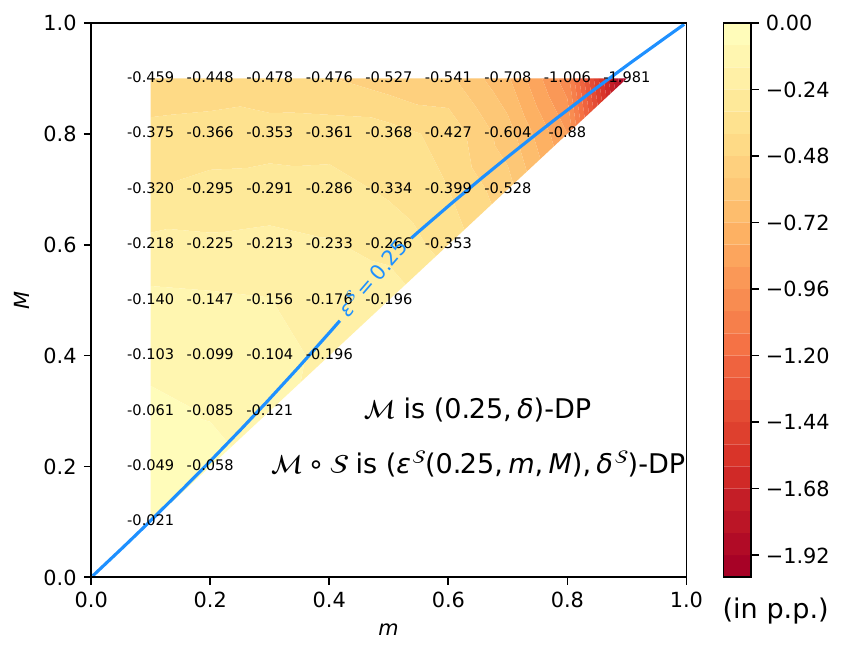}%
    \includegraphics[width=0.25\textwidth]{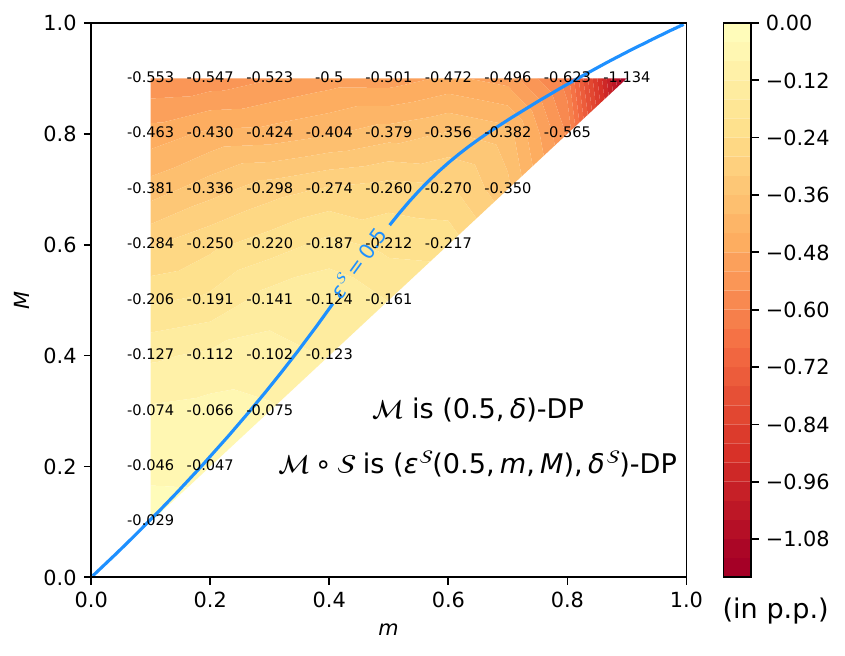}%
    \includegraphics[width=0.25\textwidth]{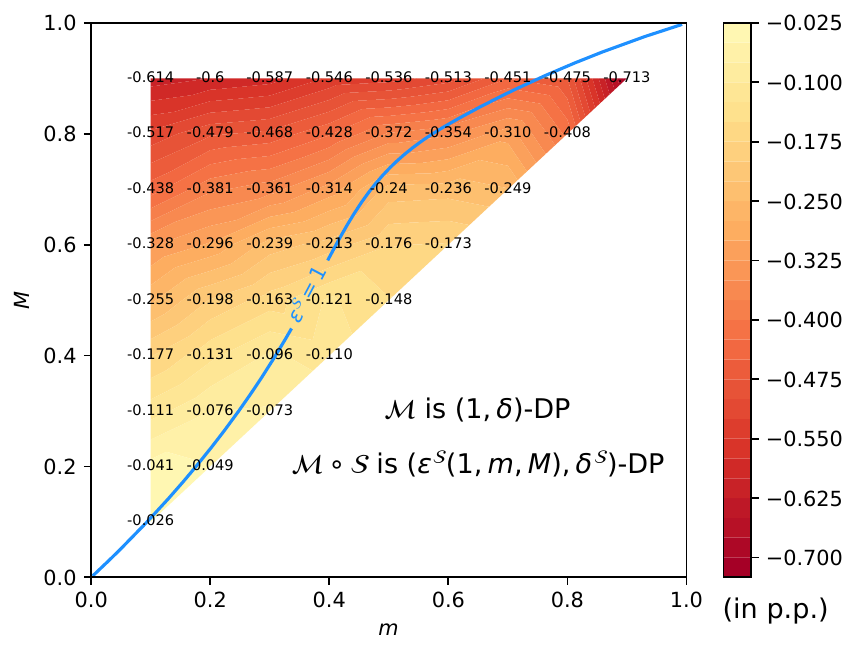}%
    \includegraphics[width=0.25\textwidth]{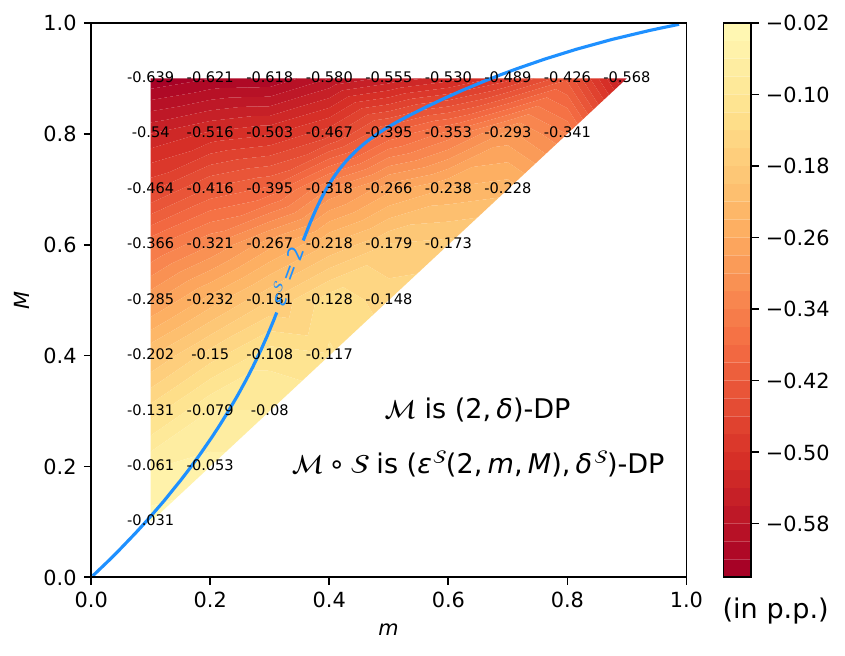}%
    \caption{The mean percent error (MPE) of $\M$ minus that of $\M\circ\S$ without the noise reduction. Results shown for the NoisyAverage with Gaussian mechanisms over the \texttt{age} column in the Adult database.}
    \label{fig:Experiment1-NoisyAverage-Gaussian-Adult-age}
\end{figure}

\begin{figure}[H]
    \includegraphics[width=0.25\textwidth]{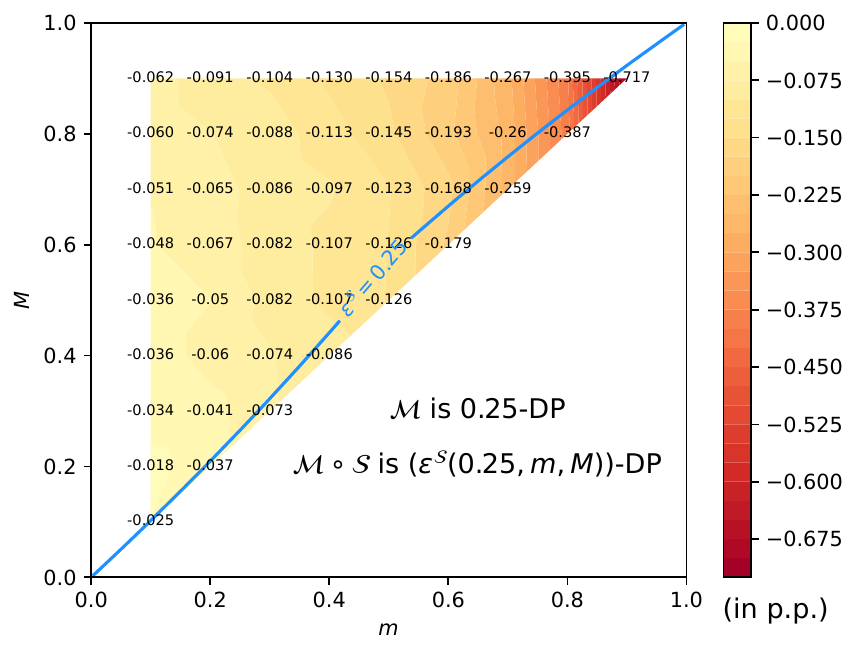}%
    \includegraphics[width=0.25\textwidth]{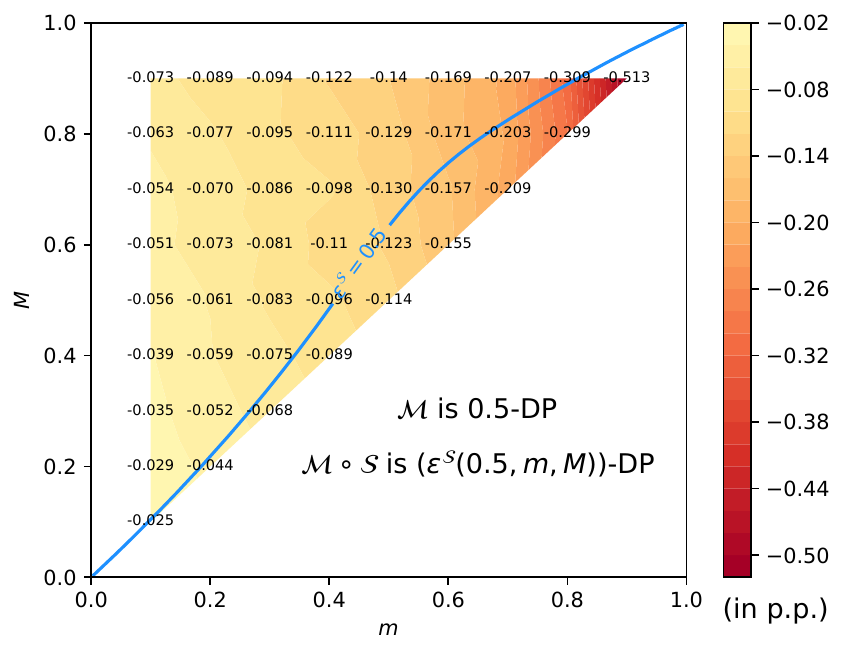}%
    \includegraphics[width=0.25\textwidth]{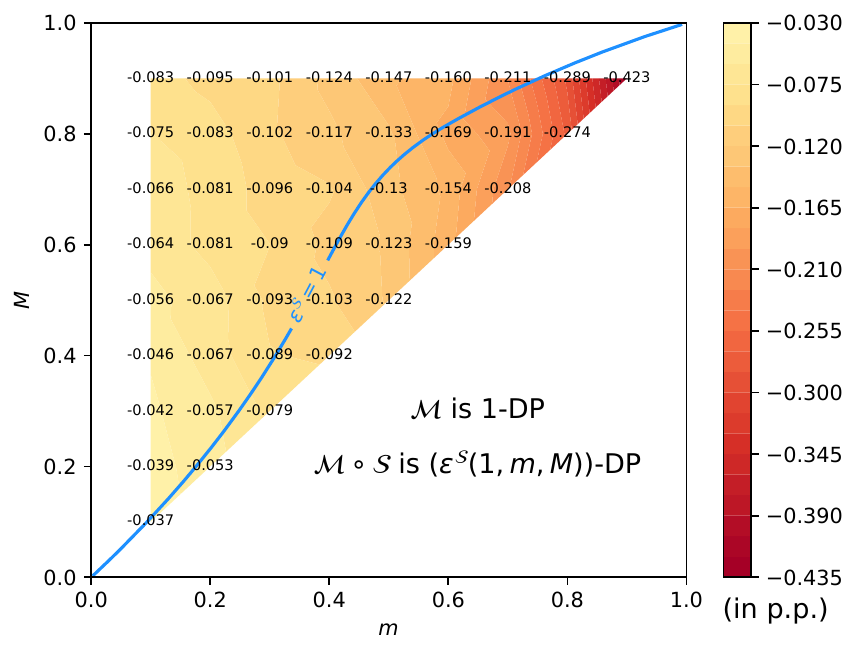}%
    \includegraphics[width=0.25\textwidth]{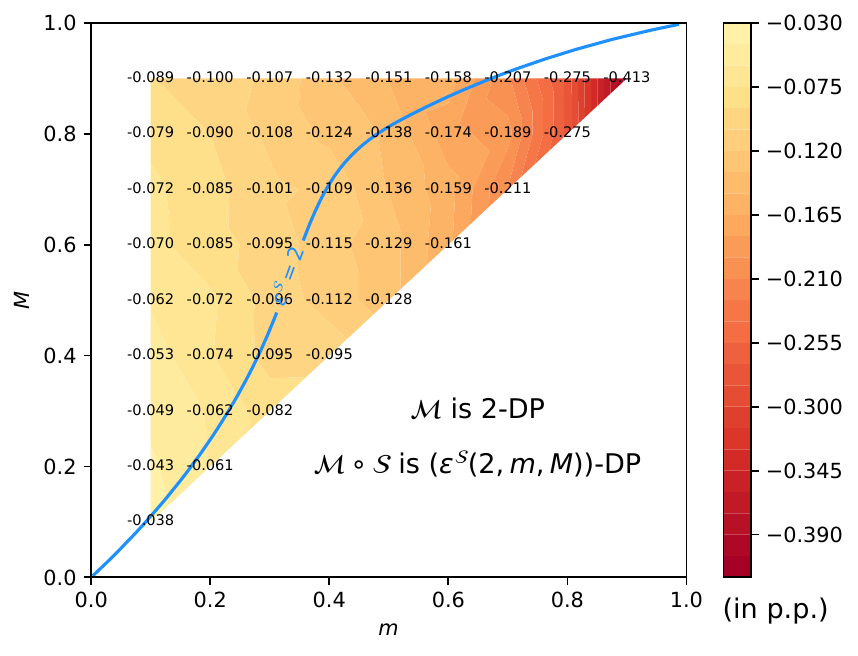}%
    \caption{The mean percent error (MPE) of $\M$ minus that of $\M\circ\S$ without the noise reduction. Results shown for the NoisyAverage with Laplace mechanisms over the \texttt{hours-per-week} column in the Adult database.}
    \label{fig:Experiment1-NoisyAverage-Laplace-Adult-hours-per-week}
\end{figure}

\begin{figure}[H]
    \includegraphics[width=0.25\textwidth]{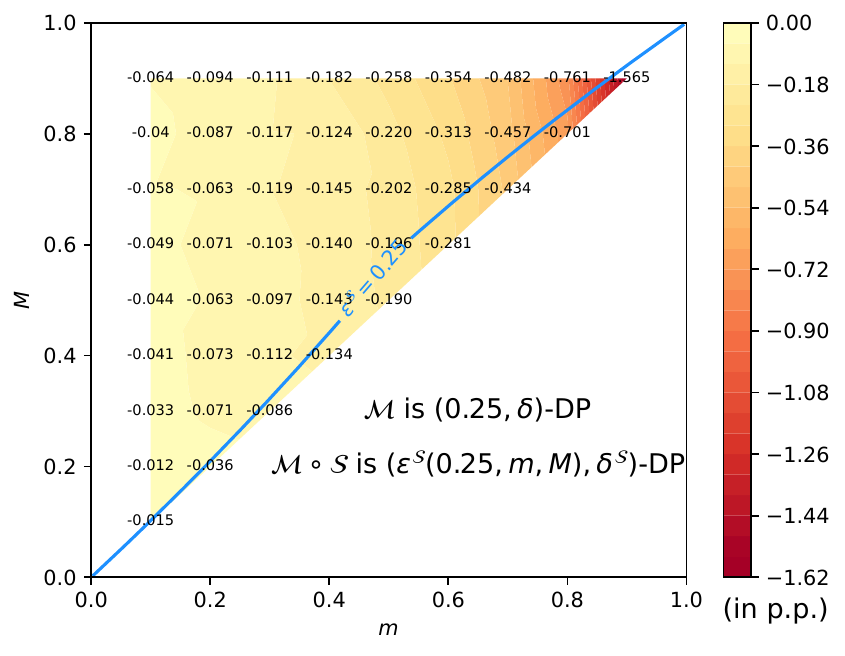}%
    \includegraphics[width=0.25\textwidth]{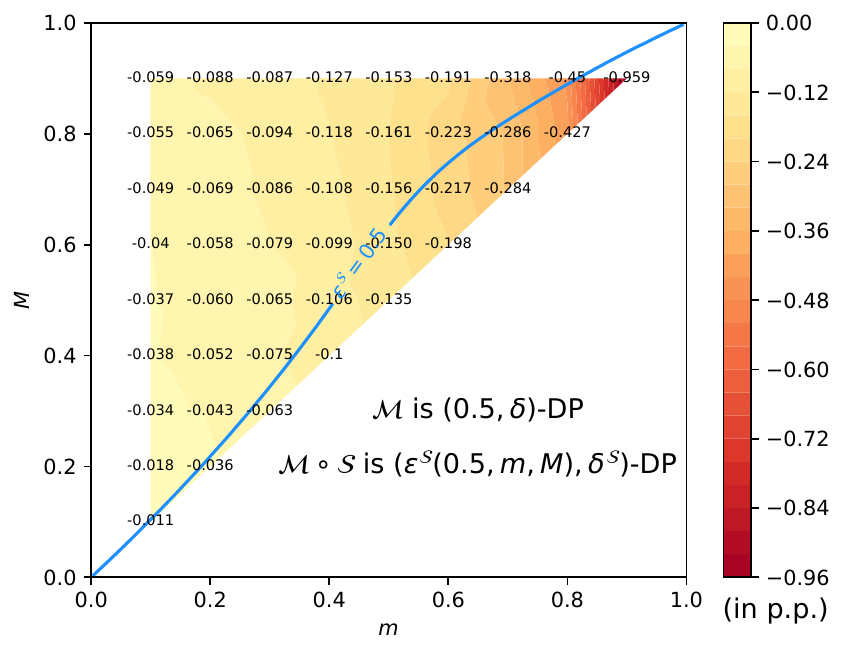}%
    \includegraphics[width=0.25\textwidth]{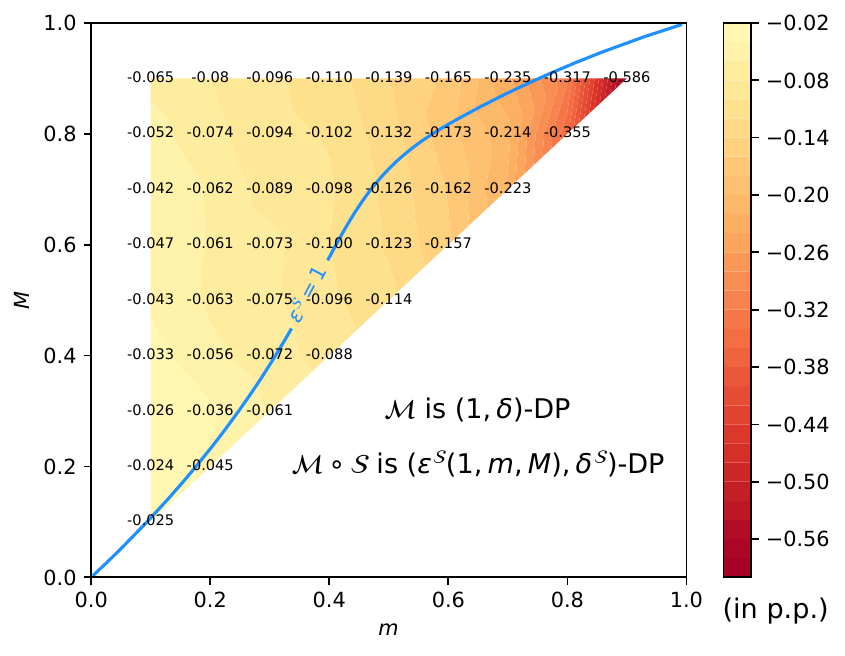}%
    \includegraphics[width=0.25\textwidth]{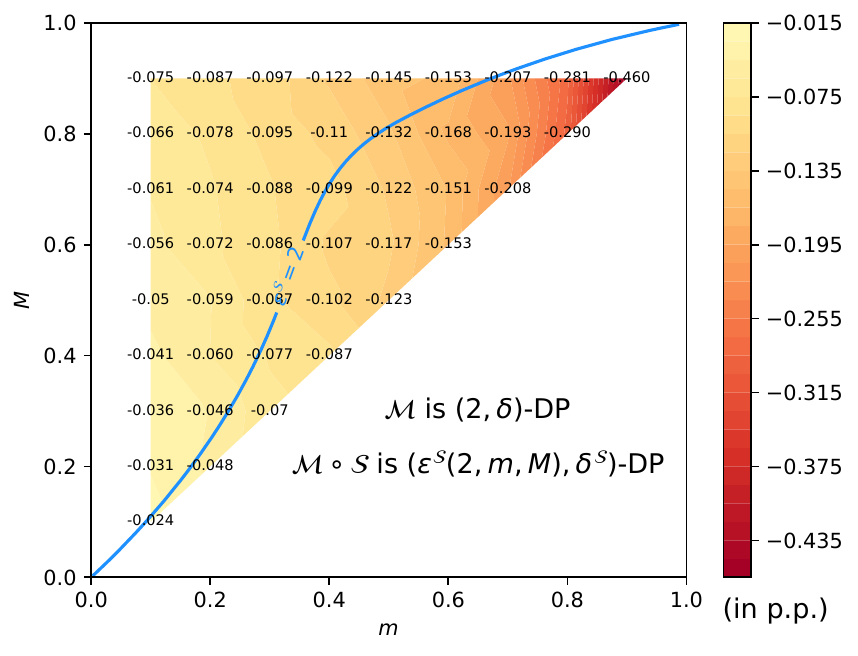}%
    \caption{The mean percent error (MPE) of $\M$ minus that of $\M\circ\S$ without the noise reduction. Results shown for the NoisyAverage with Gaussian mechanisms over the \texttt{hours-per-week} column in the Adult database.}
    \label{fig:Experiment1-NoisyAverage-Gaussian-Adult-hours-per-week}
\end{figure}

\begin{figure}[H]
    \includegraphics[width=0.25\textwidth]{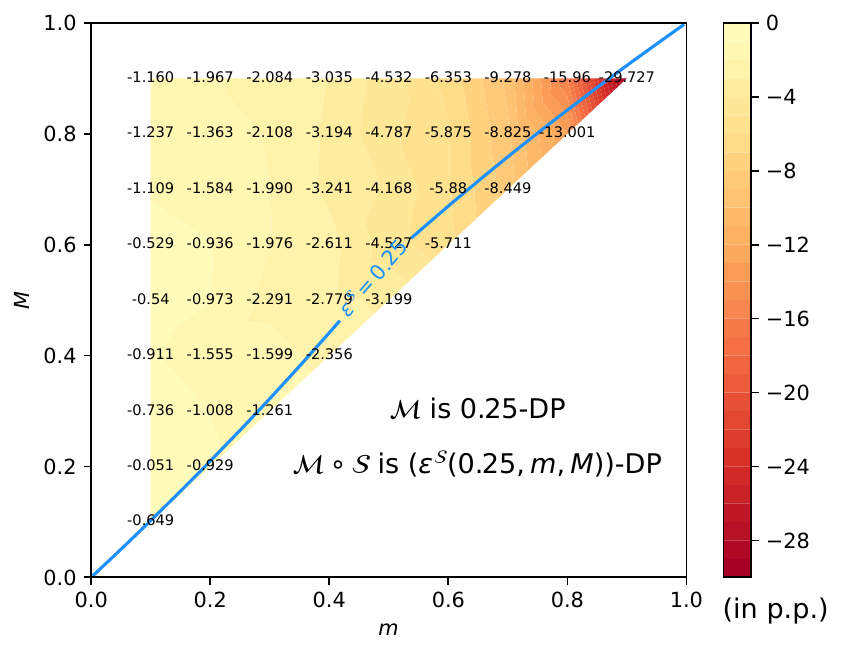}%
    \includegraphics[width=0.25\textwidth]{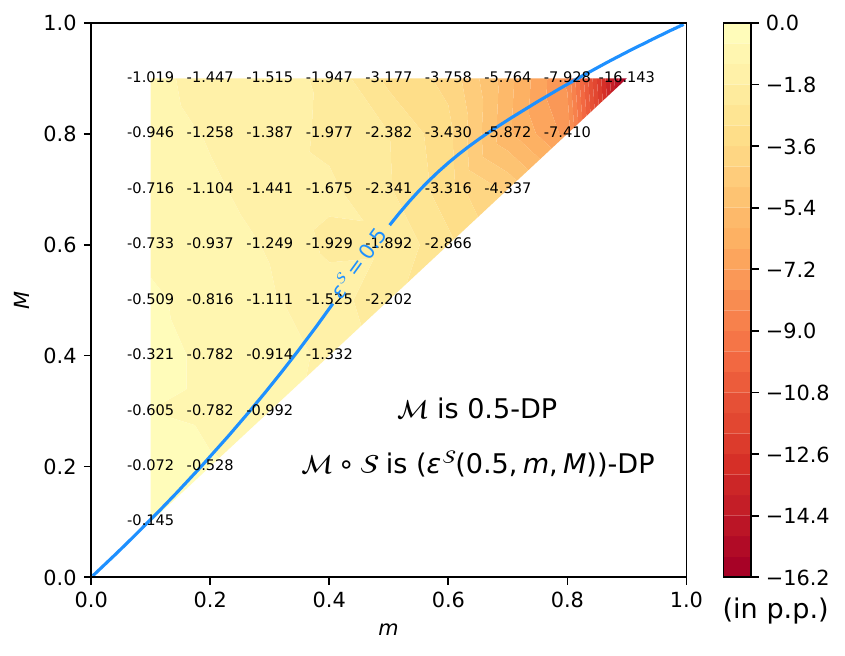}%
    \includegraphics[width=0.25\textwidth]{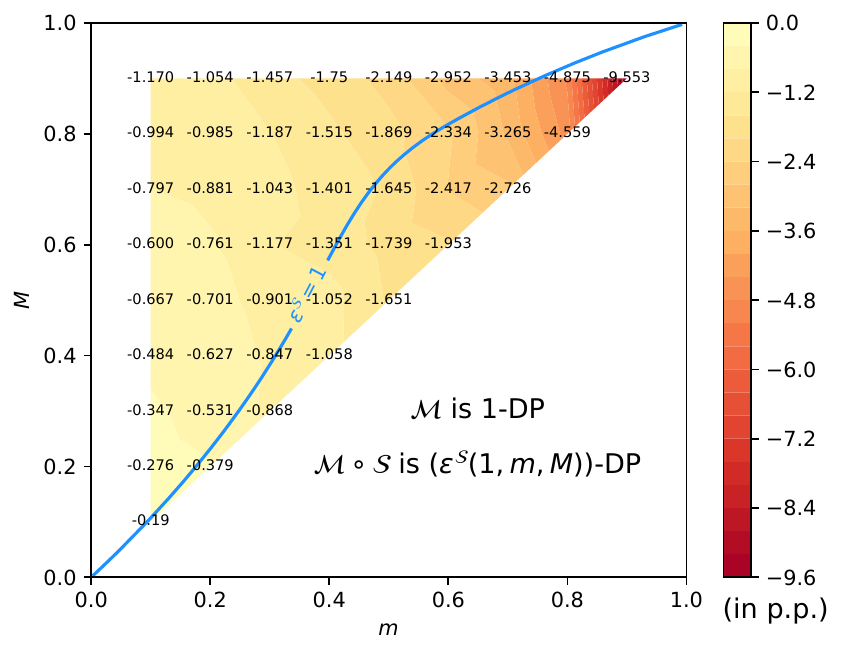}%
    \includegraphics[width=0.25\textwidth]{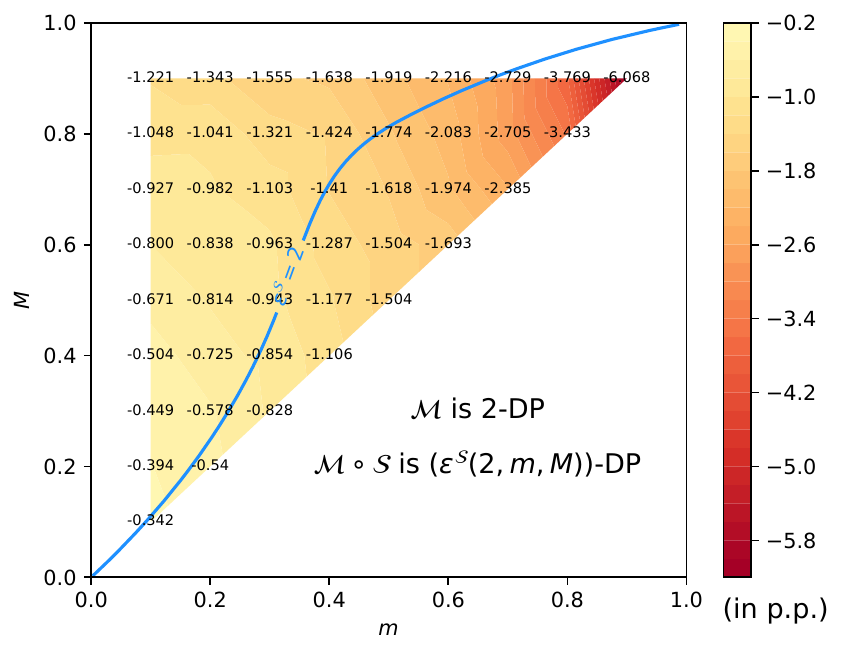}%
    \caption{The mean percent error (MPE) of $\M$ minus that of $\M\circ\S$ without the noise reduction. Results shown for the NoisyAverage with Laplace mechanisms over the \texttt{FEDTAX} column in the Census database.}
    \label{fig:Experiment1-NoisyAverage-Laplace-Census-FEDTAX}
\end{figure}

\begin{figure}[H]
    \includegraphics[width=0.25\textwidth]{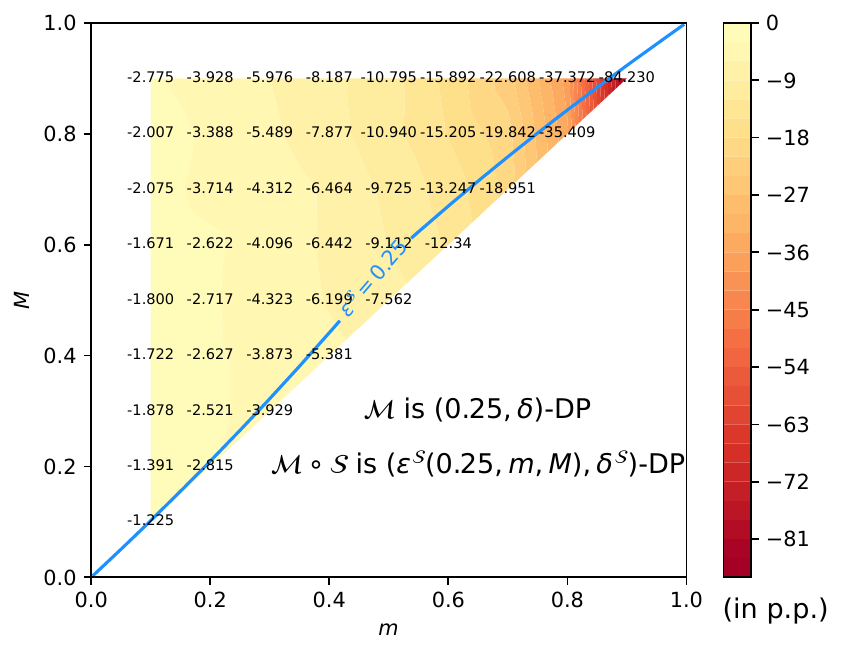}%
    \includegraphics[width=0.25\textwidth]{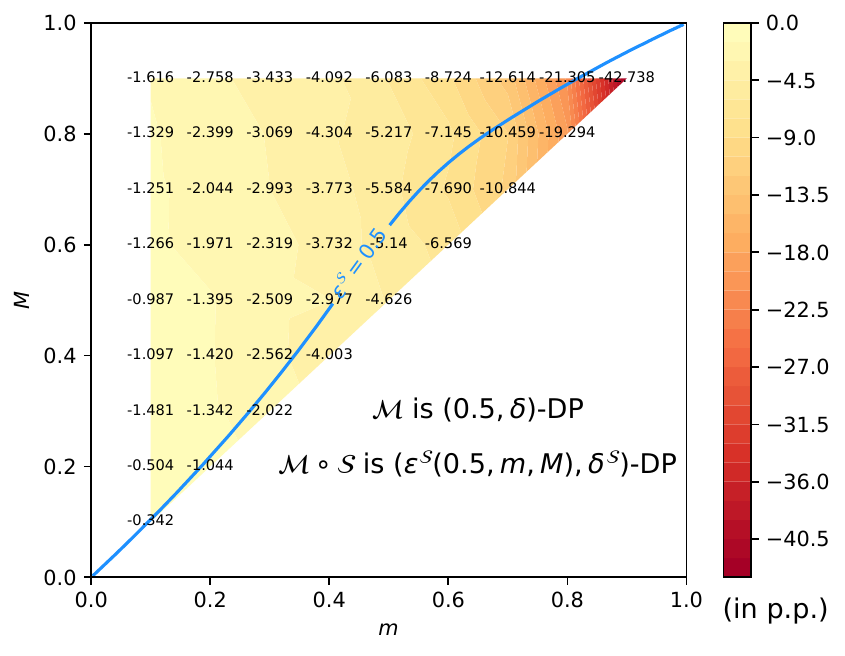}%
    \includegraphics[width=0.25\textwidth]{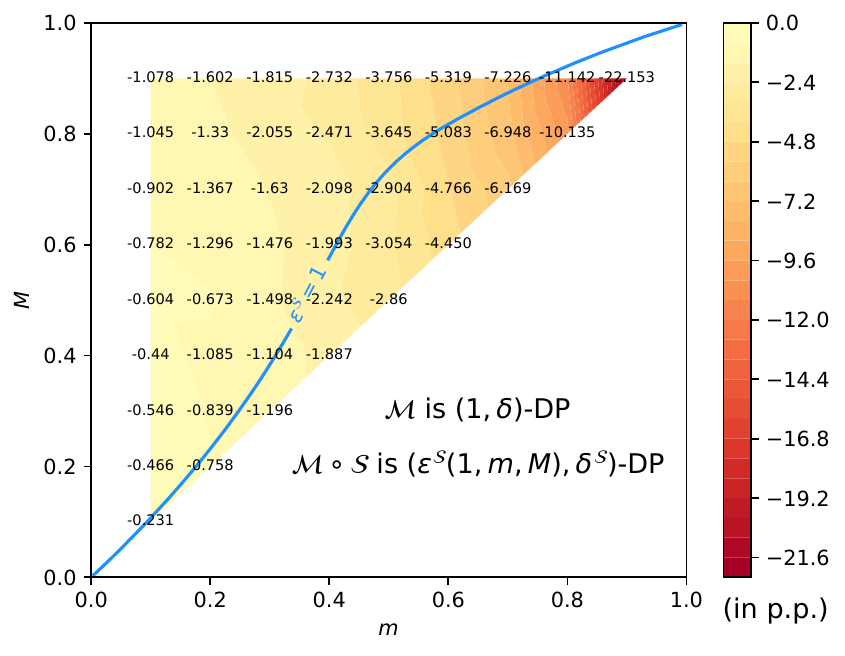}%
    \includegraphics[width=0.25\textwidth]{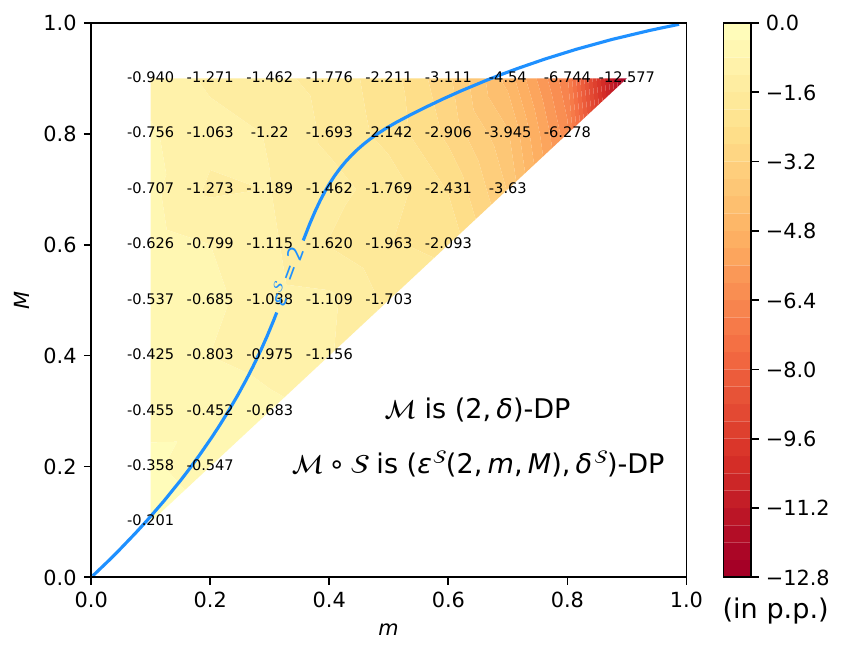}%
    \caption{The mean percent error (MPE) of $\M$ minus that of $\M\circ\S$ without the noise reduction. Results shown for the NoisyAverage with Gaussian mechanisms over the \texttt{FEDTAX} column in the Census database.}
    \label{fig:Experiment1-NoisyAverage-Gaussian-Census-FEDTAX}
\end{figure}

\begin{figure}[H]
    \includegraphics[width=0.25\textwidth]{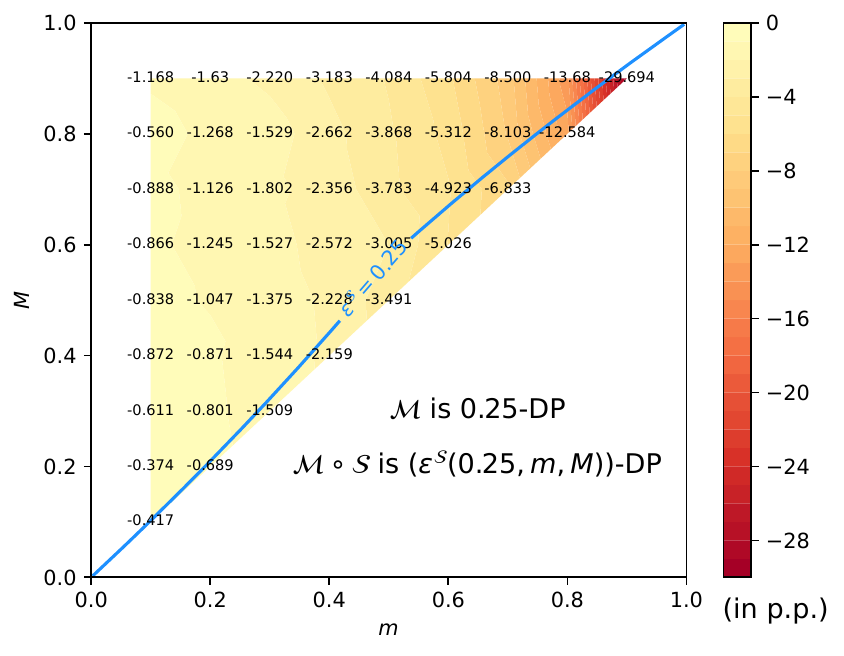}%
    \includegraphics[width=0.25\textwidth]{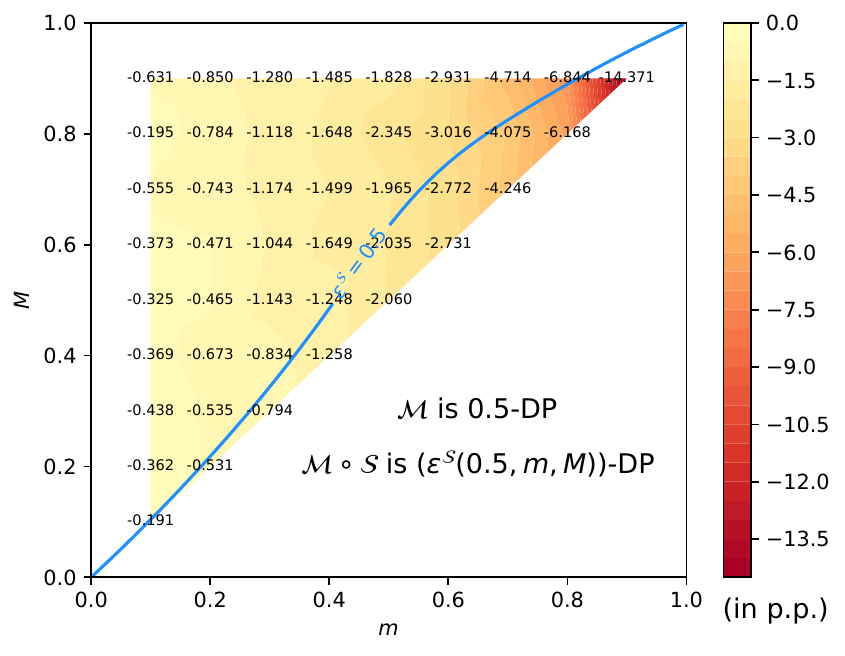}%
    \includegraphics[width=0.25\textwidth]{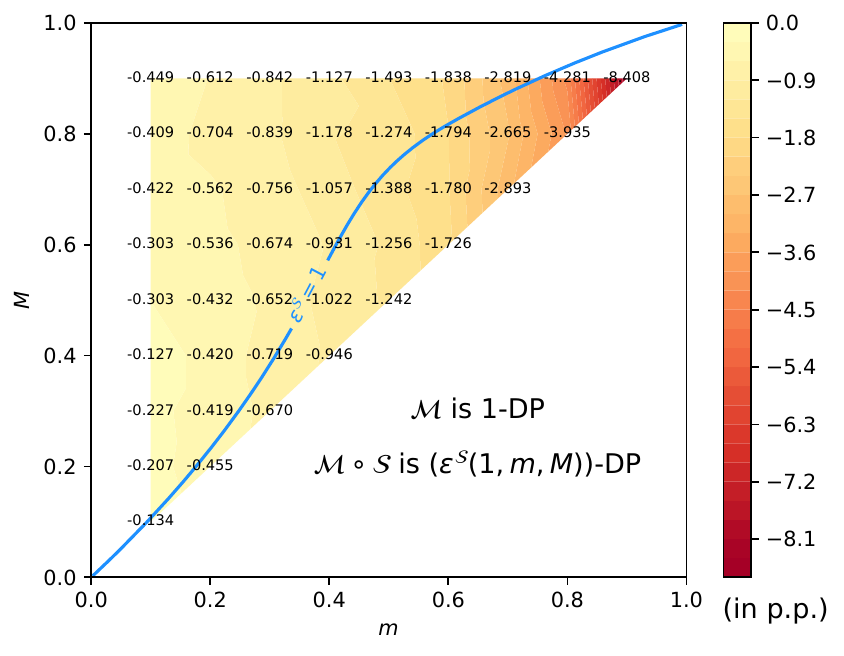}%
    \includegraphics[width=0.25\textwidth]{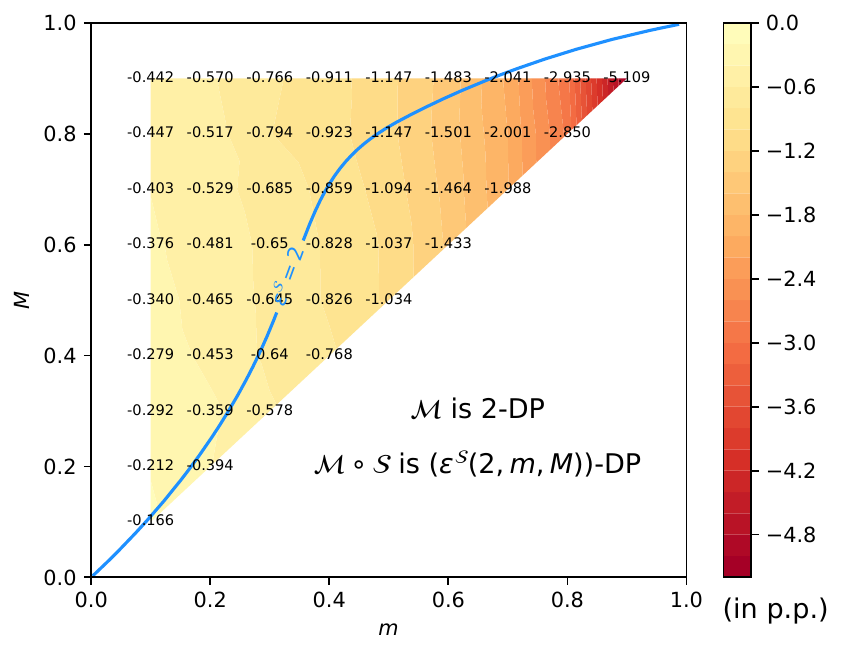}%
    \caption{The mean percent error (MPE) of $\M$ minus that of $\M\circ\S$ without the noise reduction. Results shown for the NoisyAverage with Laplace mechanisms over the \texttt{FICA} column in the Census database.}
    \label{fig:Experiment1-NoisyAverage-Laplace-Census-FICA}
\end{figure}

\begin{figure}[H]
    \includegraphics[width=0.25\textwidth]{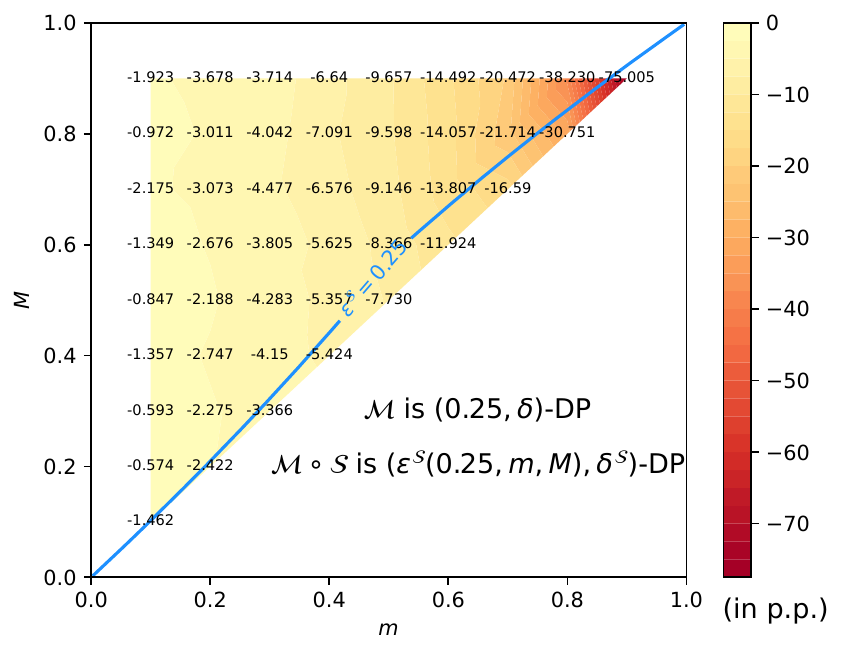}%
    \includegraphics[width=0.25\textwidth]{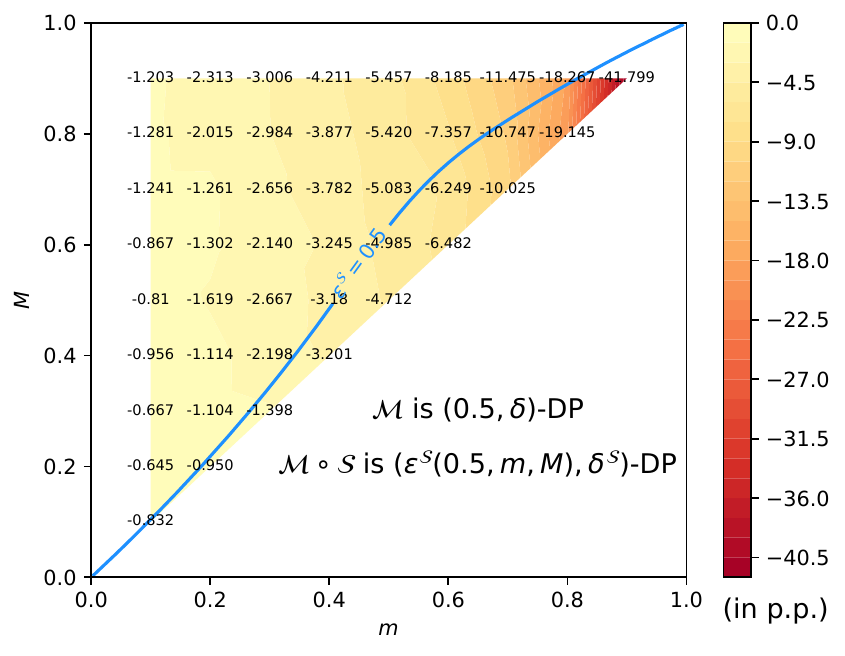}%
    \includegraphics[width=0.25\textwidth]{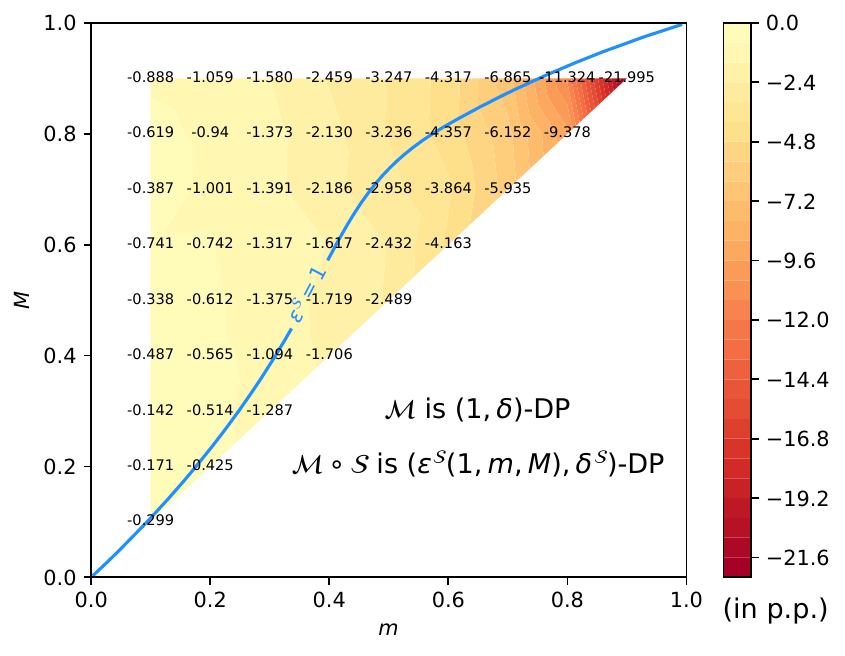}%
    \includegraphics[width=0.25\textwidth]{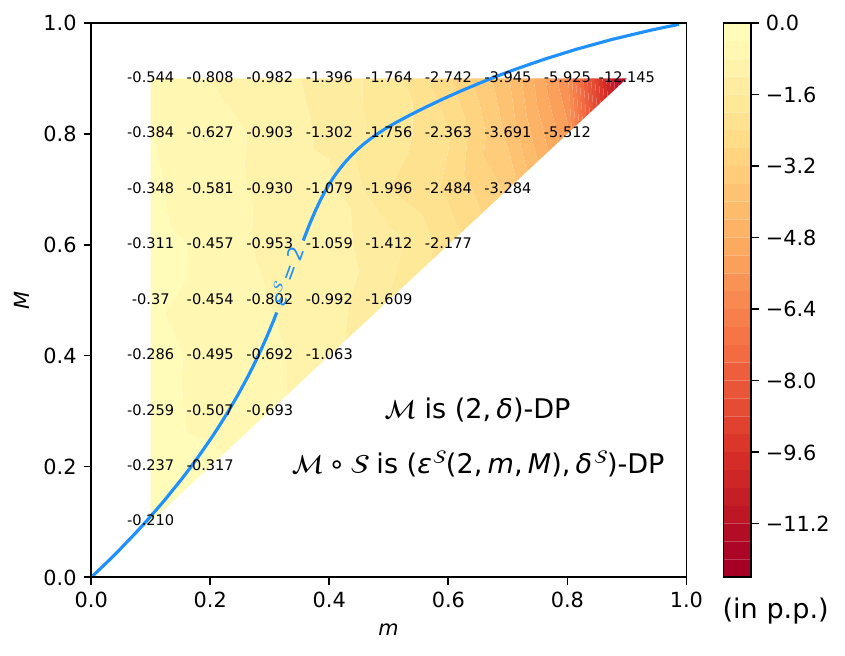}%
    \caption{The mean percent error (MPE) of $\M$ minus that of $\M\circ\S$ without the noise reduction. Results shown for the NoisyAverage with Gaussian mechanisms over the \texttt{FICA} column in the Census database.}
    \label{fig:Experiment1-NoisyAverage-Gaussian-Census-FICA}
\end{figure}

\begin{figure}[H]
    \includegraphics[width=0.25\textwidth]{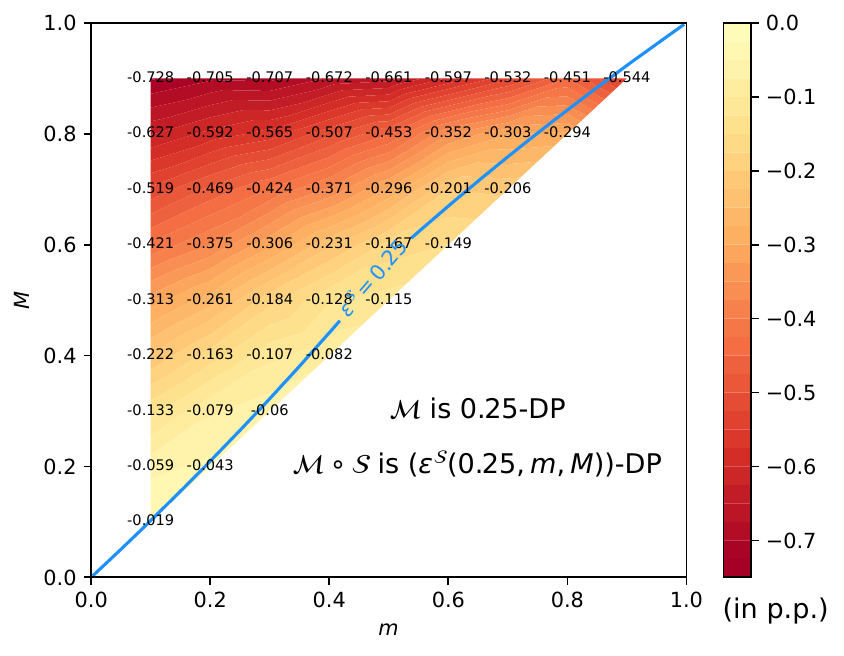}%
    \includegraphics[width=0.25\textwidth]{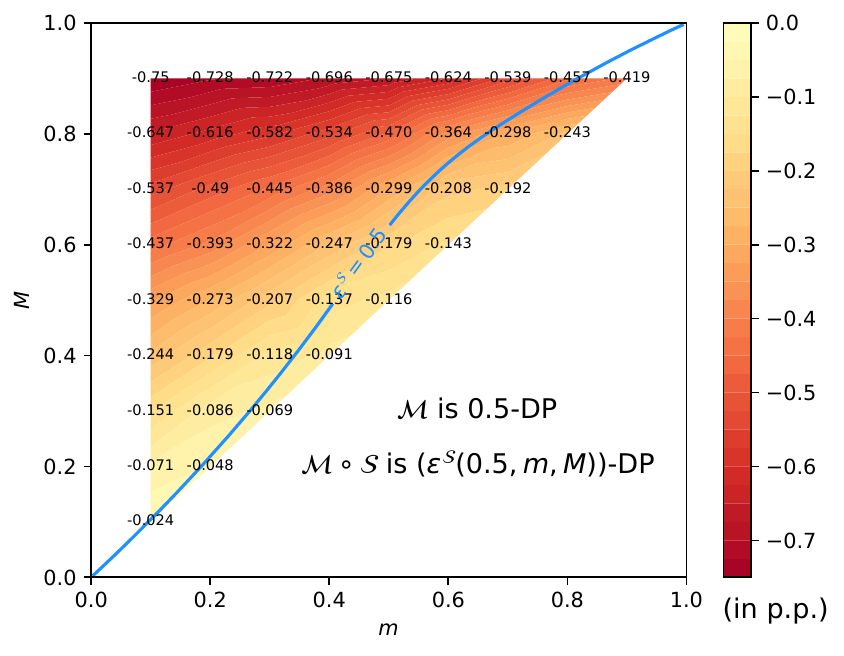}%
    \includegraphics[width=0.25\textwidth]{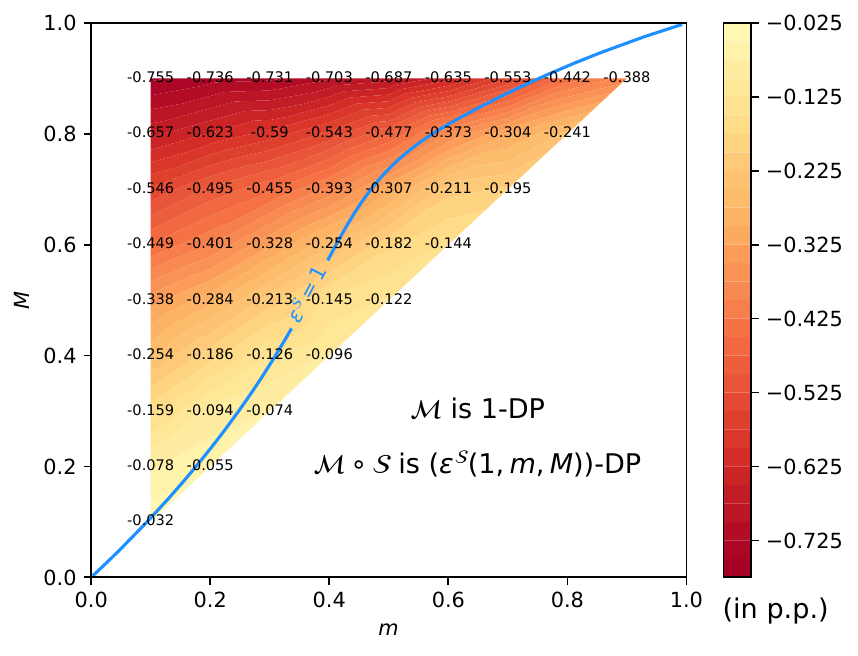}%
    \includegraphics[width=0.25\textwidth]{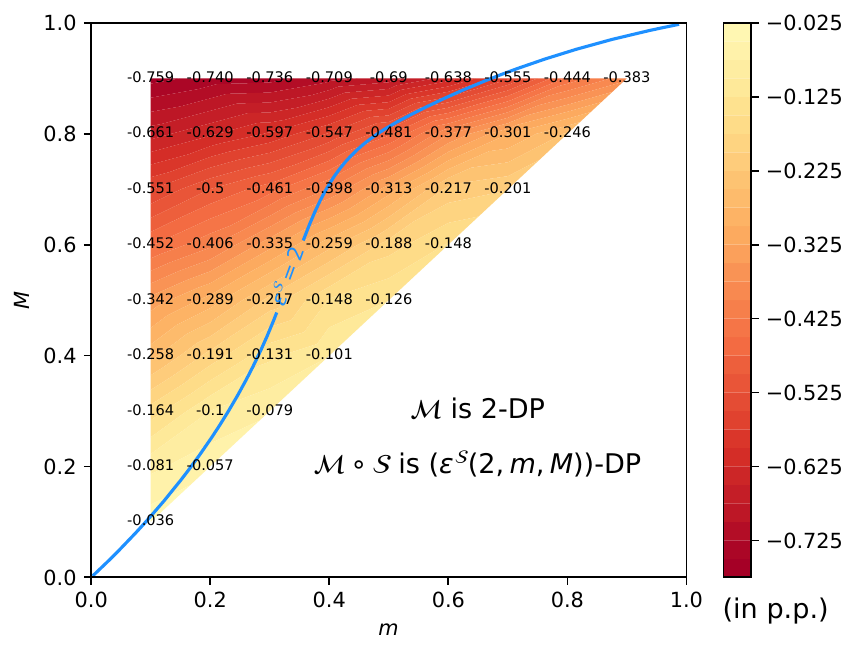}%
    \caption{The mean percent error (MPE) of $\M$ minus that of $\M\circ\S$ without the noise reduction. Results shown for the NoisyAverage with Laplace mechanisms over the \texttt{Age} column in the Irish database.}
    \label{fig:Experiment1-NoisyAverage-Laplace-Irishn-Age}
\end{figure}

\begin{figure}[H]
    \includegraphics[width=0.25\textwidth]{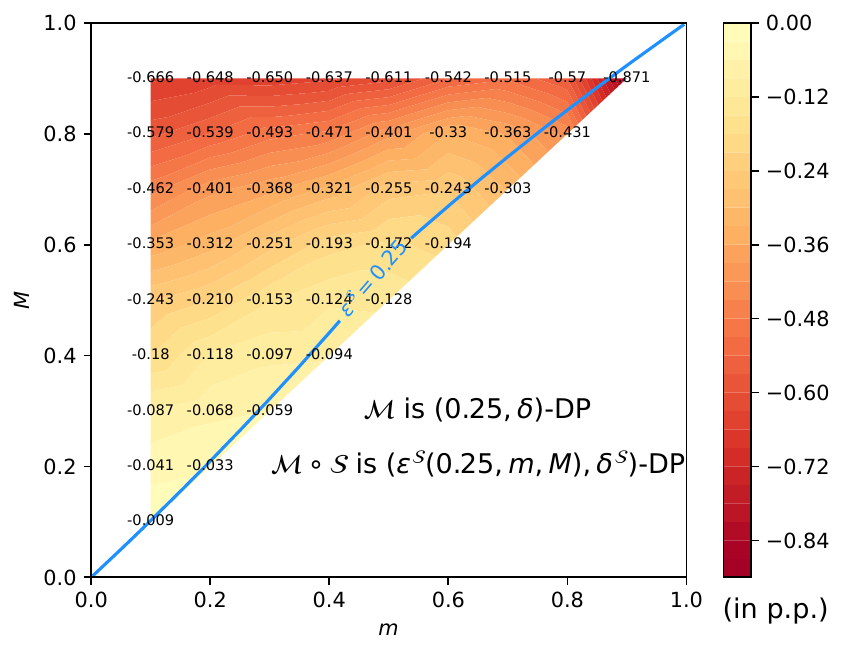}%
    \includegraphics[width=0.25\textwidth]{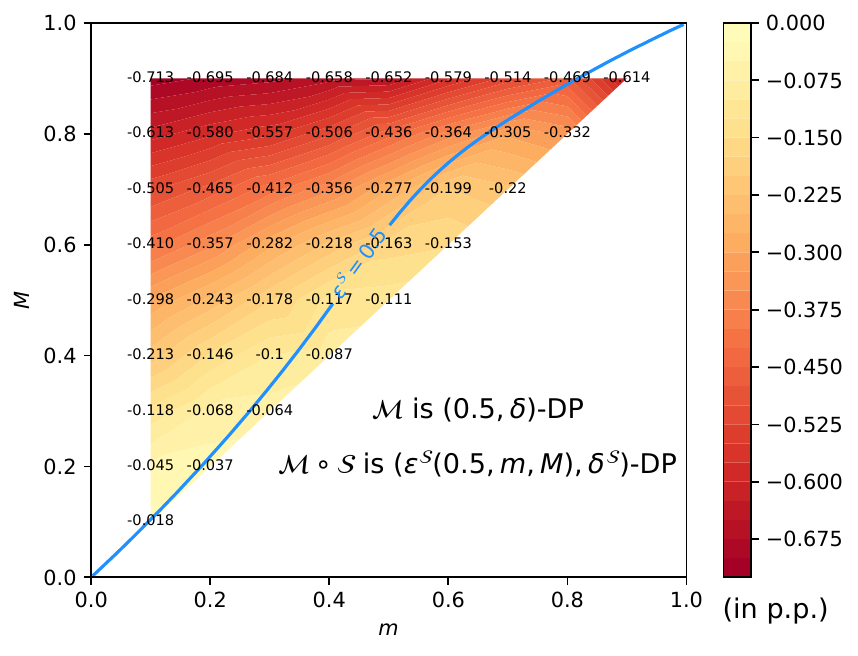}%
    \includegraphics[width=0.25\textwidth]{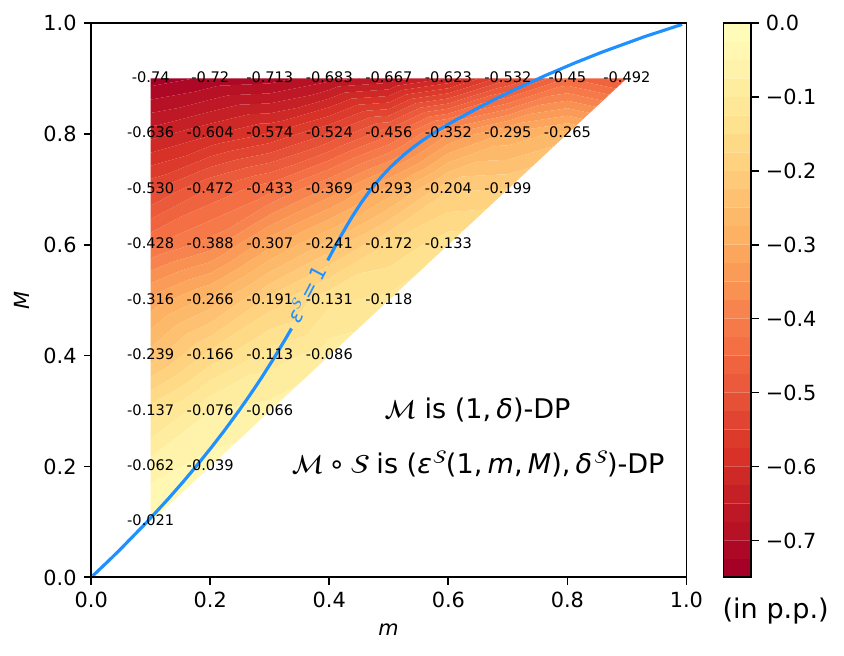}%
    \includegraphics[width=0.25\textwidth]{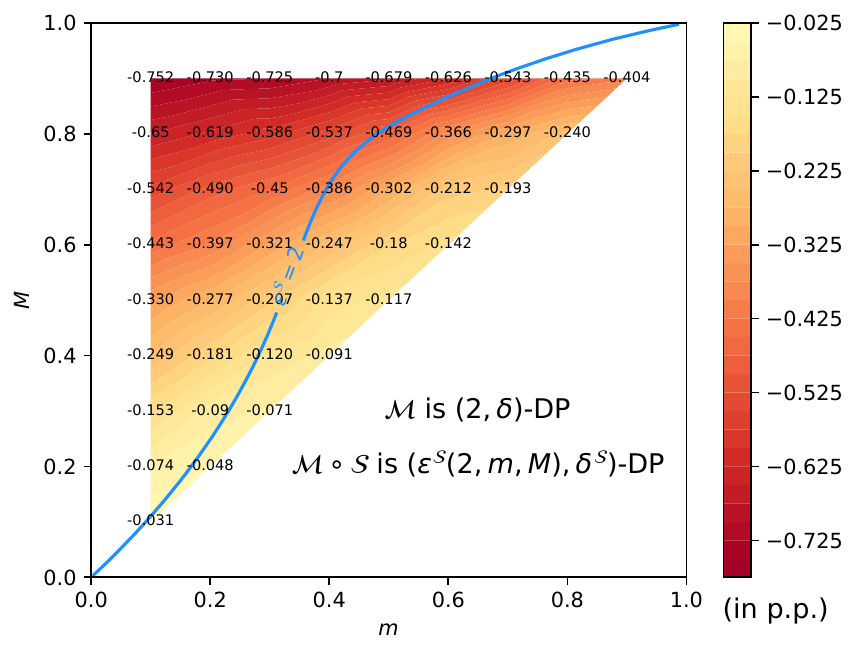}%
    \caption{The mean percent error (MPE) of $\M$ minus that of $\M\circ\S$ without the noise reduction. Results shown for the NoisyAverage with Gaussian mechanisms over the \texttt{Age} column in the Irish database.}
    \label{fig:Experiment1-NoisyAverage-Gaussian-Irishn-Age}
\end{figure}

\begin{figure}[H]
    \includegraphics[width=0.25\textwidth]{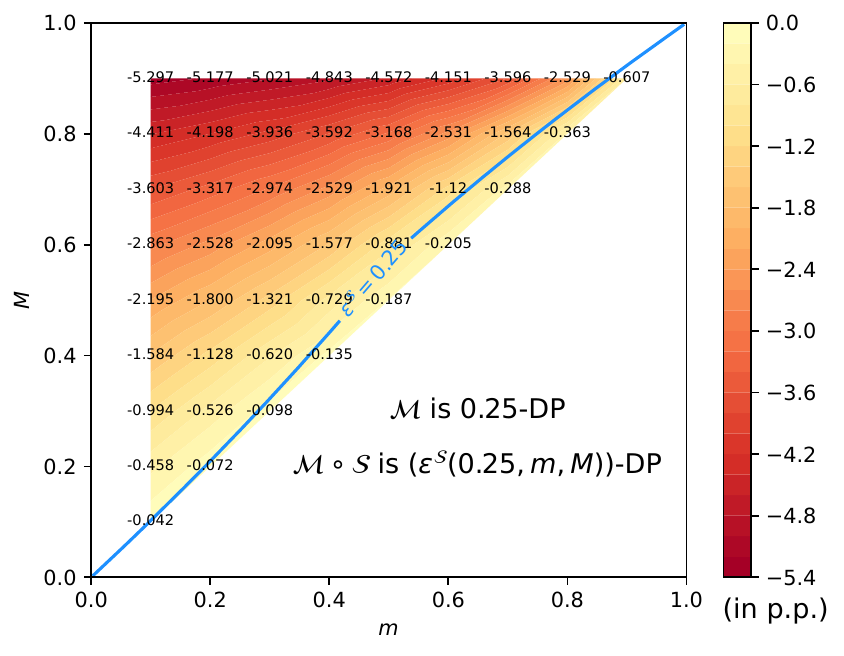}%
    \includegraphics[width=0.25\textwidth]{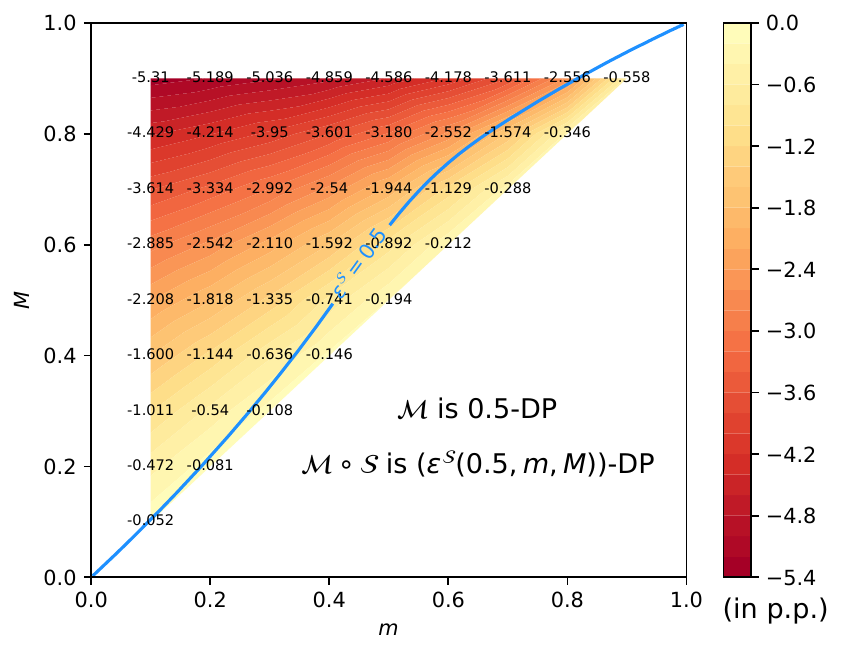}%
    \includegraphics[width=0.25\textwidth]{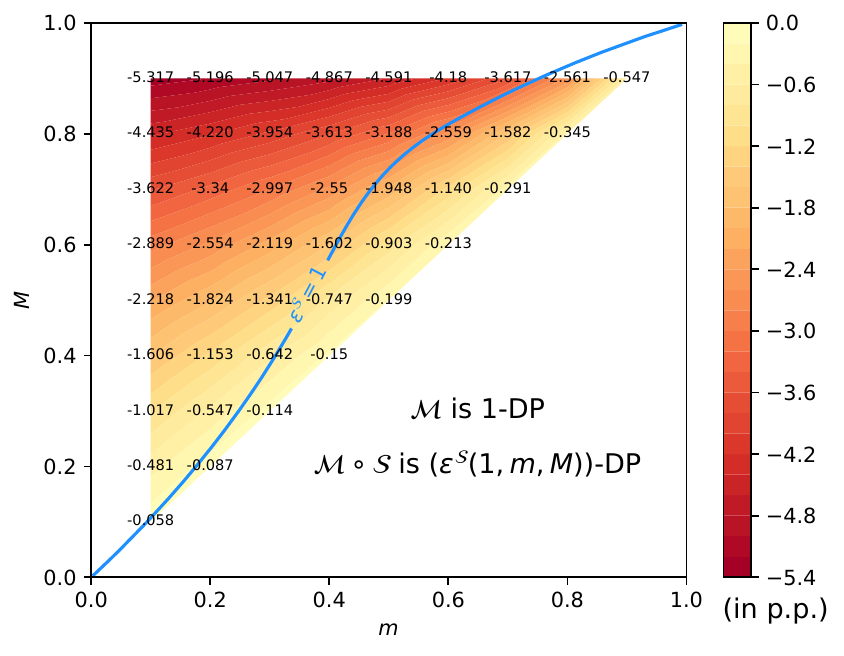}%
    \includegraphics[width=0.25\textwidth]{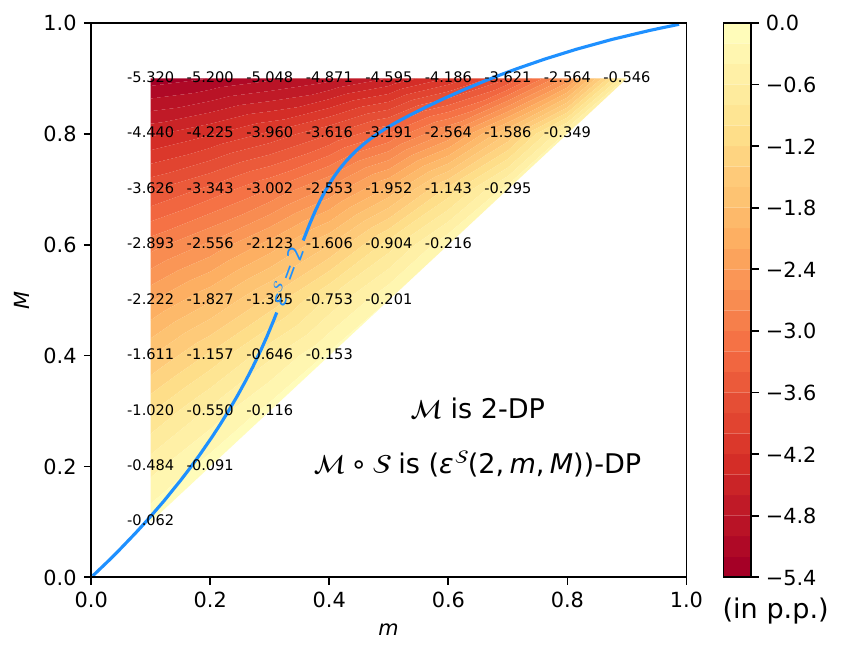}%
    \caption{The mean percent error (MPE) of $\M$ minus that of $\M\circ\S$ without the noise reduction. Results shown for the NoisyAverage with Laplace mechanisms over the \texttt{HighestEducationCompleted} column in the Irish database.}
    \label{fig:Experiment1-NoisyAverage-Laplace-Irishn-HighestEducationCompleted}
\end{figure}

\begin{figure}[H]
    \includegraphics[width=0.25\textwidth]{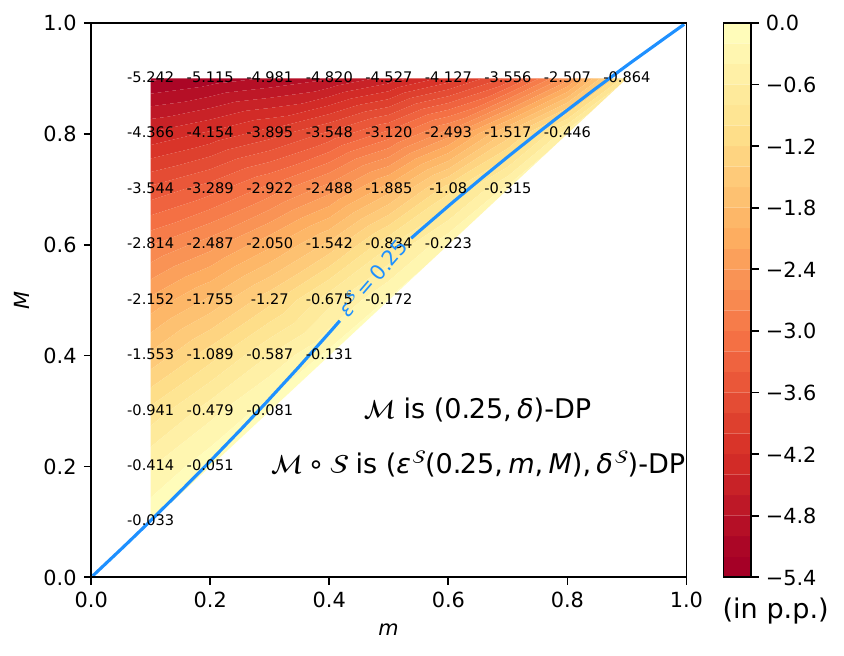}%
    \includegraphics[width=0.25\textwidth]{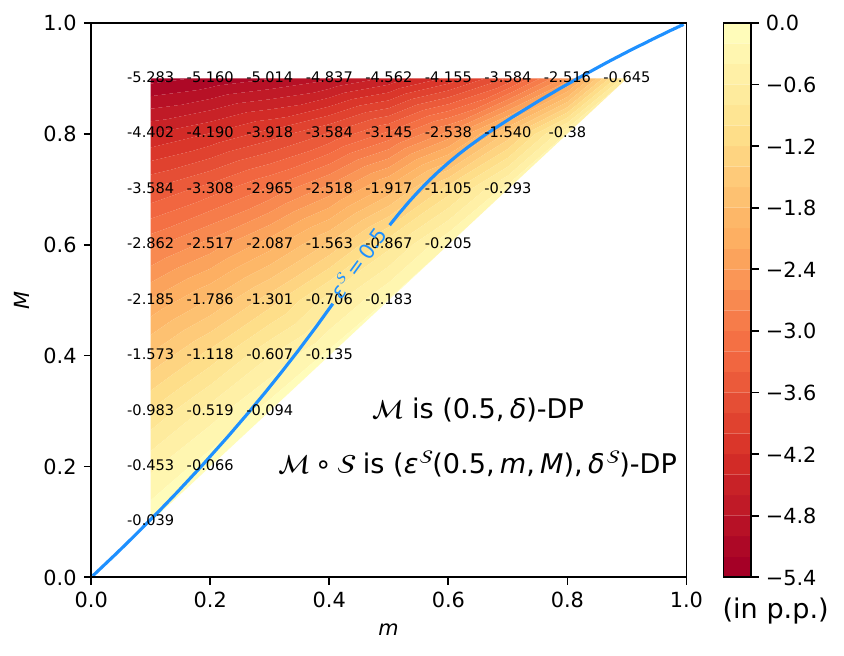}%
    \includegraphics[width=0.25\textwidth]{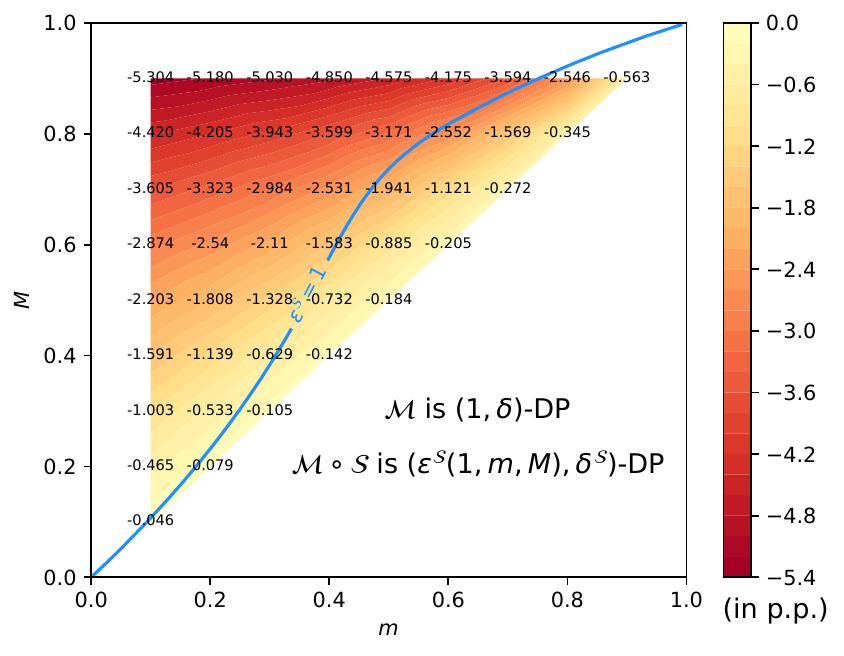}%
    \includegraphics[width=0.25\textwidth]{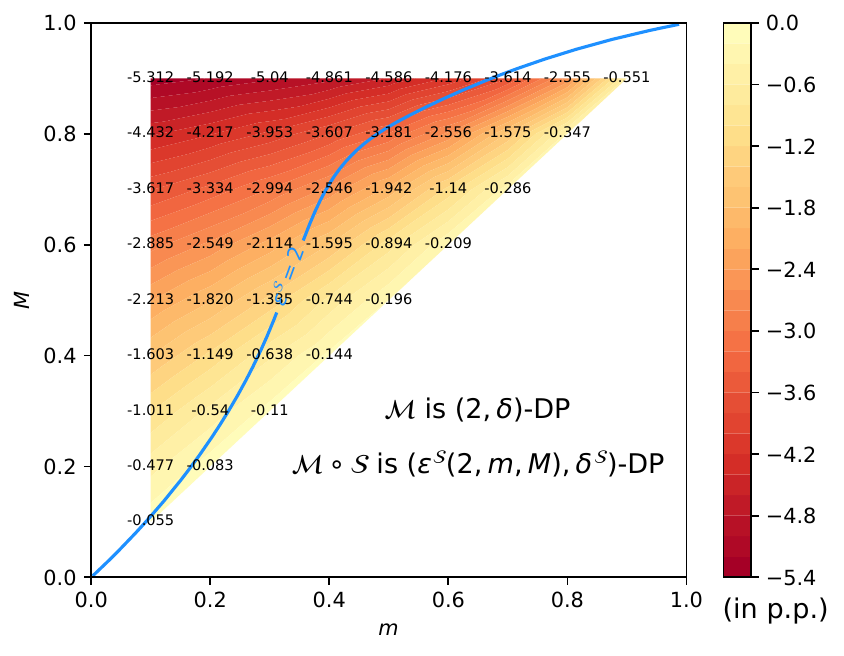}%
    \caption{The mean percent error (MPE) of $\M$ minus that of $\M\circ\S$ without the noise reduction. Results shown for the NoisyAverage with Gaussian mechanisms over the \texttt{HighestEducationCompleted} column in the Irish database.}
    \label{fig:Experiment1-NoisyAverage-Gaussian-Irishn-HighestEducationCompleted}
\end{figure}

\subsection{Plots of the Utility Difference between the Mechanisms \textit{without} the Noise Reduction for the Mode Computation}\label{sec:plots:SuppressionwithoutEpsDeltaChange2}

\begin{figure}[H]
    \includegraphics[width=0.25\textwidth]{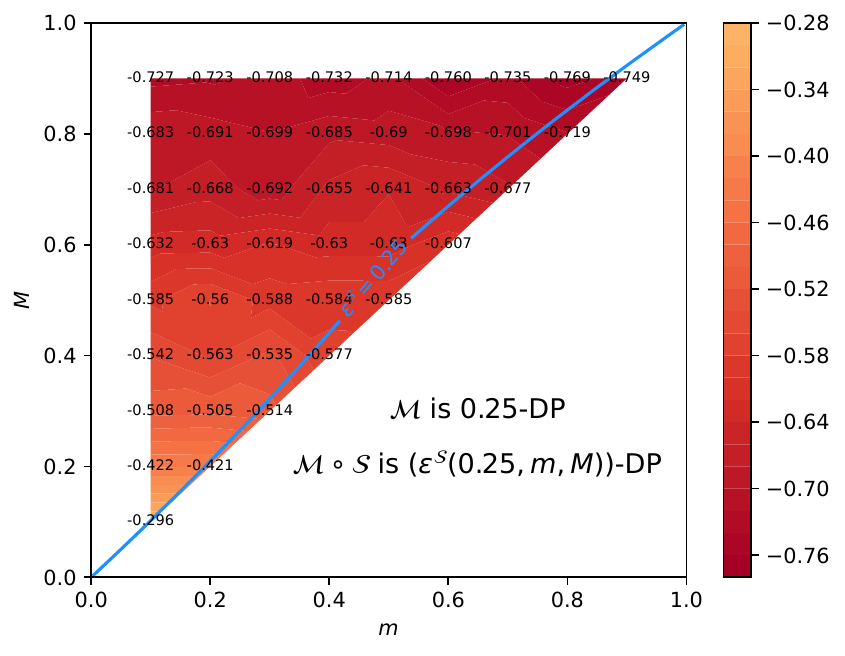}%
    \includegraphics[width=0.25\textwidth]{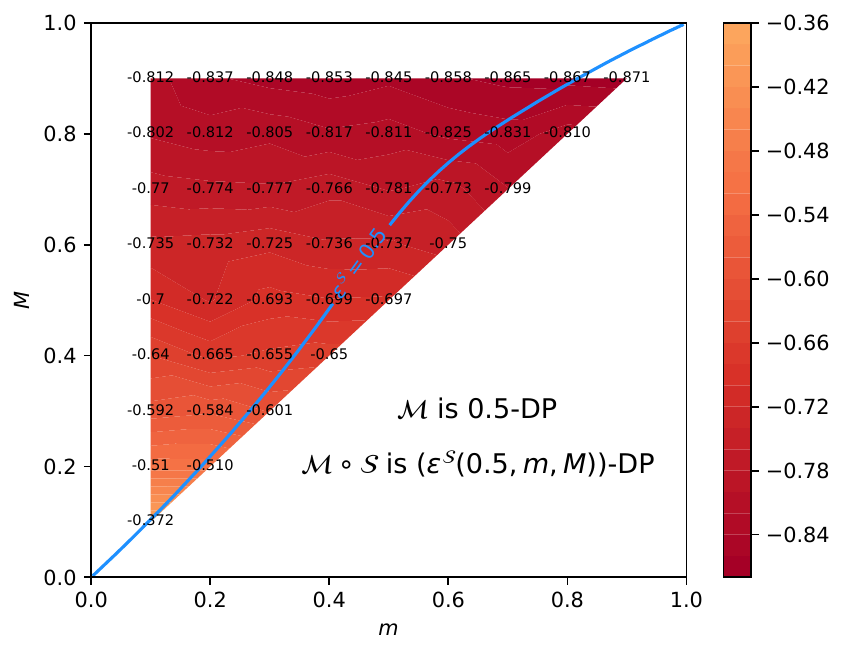}%
    \includegraphics[width=0.25\textwidth]{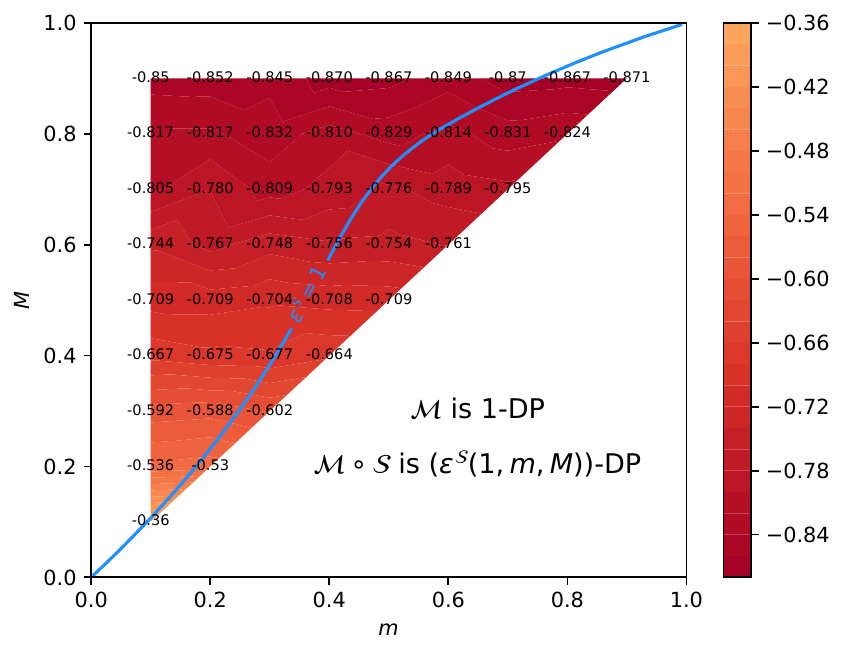}%
    \includegraphics[width=0.25\textwidth]{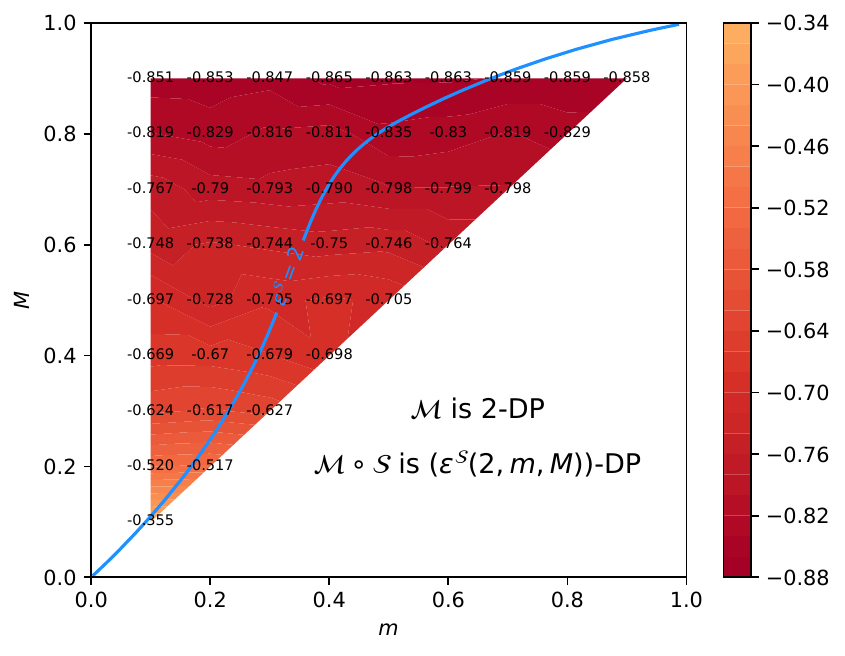}%
    \caption{The probability of outputting an incorrect mode of $\M$ minus that of $\M\circ\S$ without the noise reduction. Results shown for the RNM with Laplace noise over the \texttt{age} column in the Adult database.}
    \label{fig:Experiment1-RNM-Laplace-Adult-age}
\end{figure}

\begin{figure}[H]
    \includegraphics[width=0.25\textwidth]{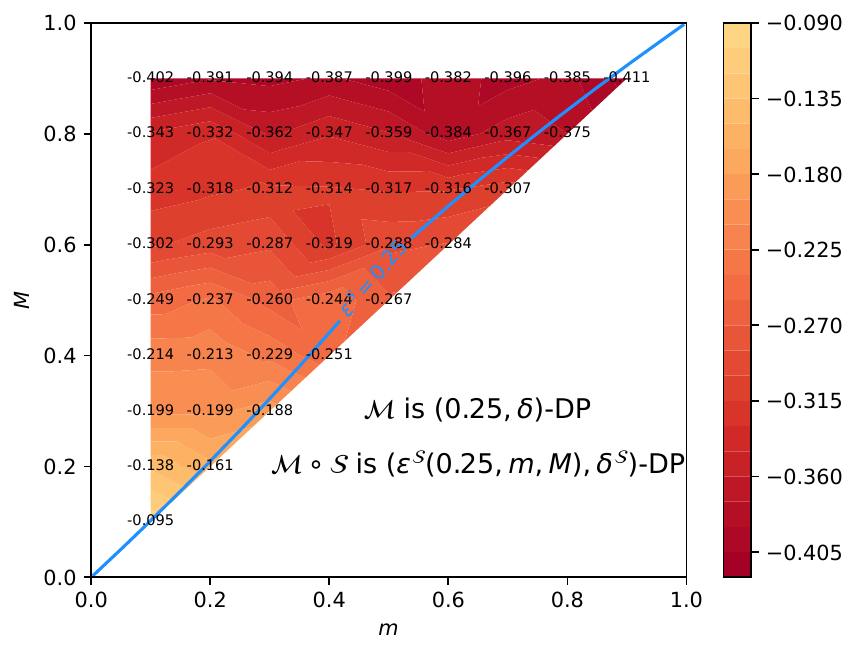}%
    \includegraphics[width=0.25\textwidth]{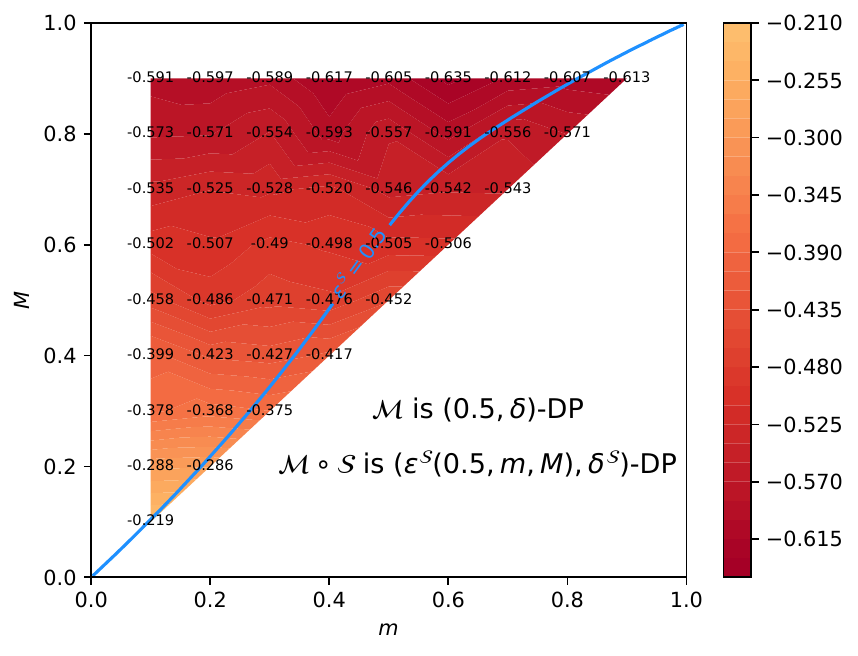}%
    \includegraphics[width=0.25\textwidth]{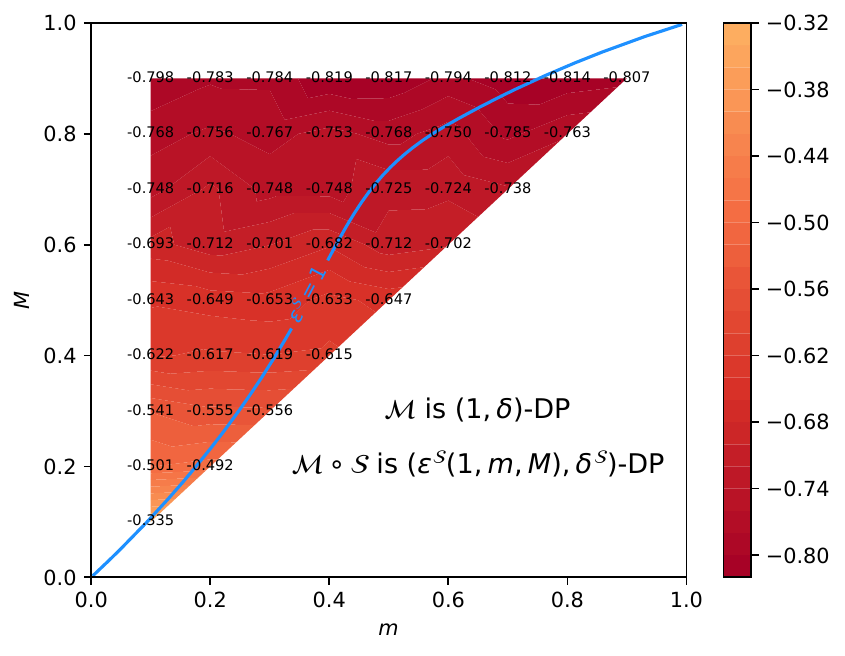}%
    \includegraphics[width=0.25\textwidth]{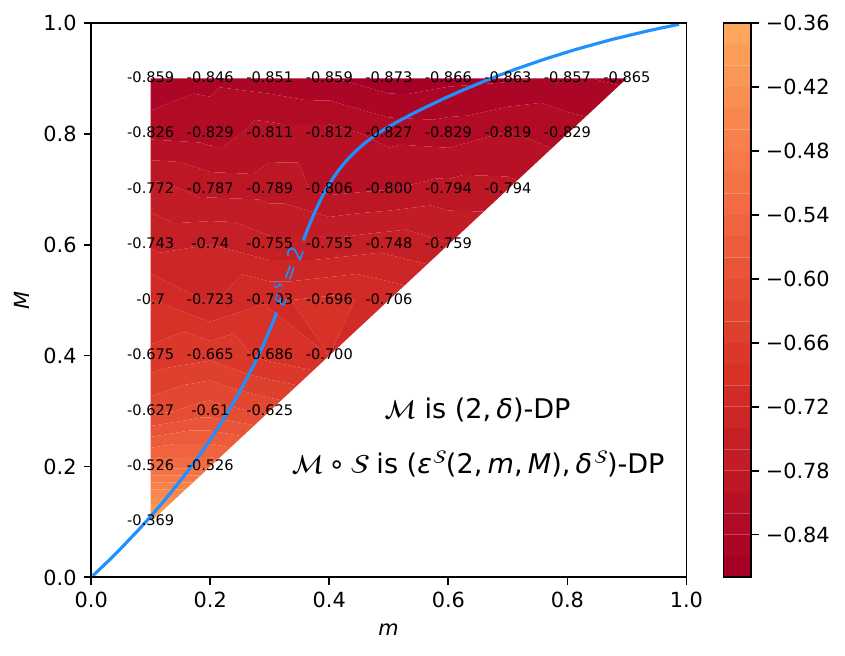}%
    \caption{The probability of outputting an incorrect mode of $\M$ minus that of $\M\circ\S$ without the noise reduction. Results shown for the RNM-like variant with Gaussian noise over the \texttt{age} column in the Adult database.}
    \label{fig:Experiment1-RNM-Gaussian-Adult-age}
\end{figure}

\begin{figure}[H]
    \includegraphics[width=0.25\textwidth]{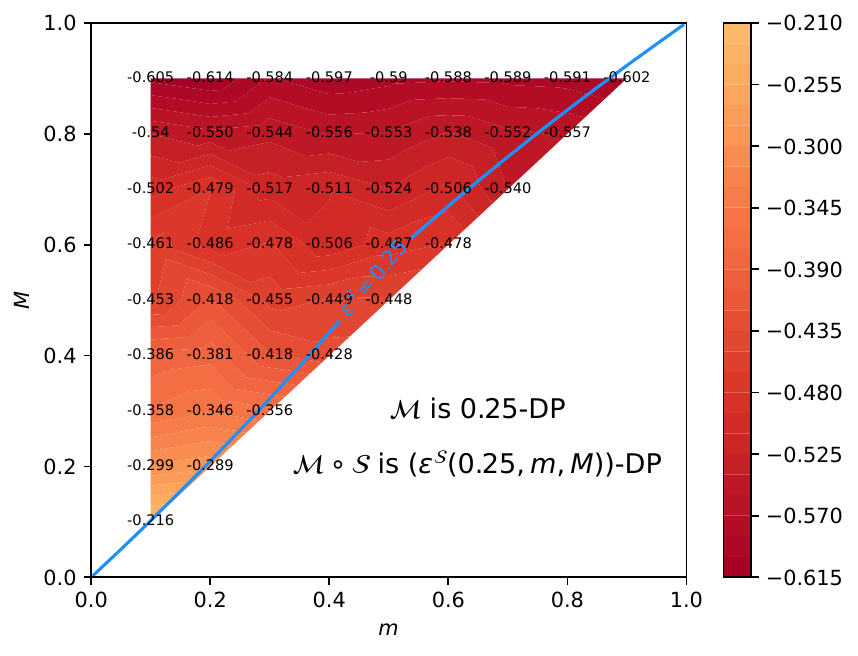}%
    \includegraphics[width=0.25\textwidth]{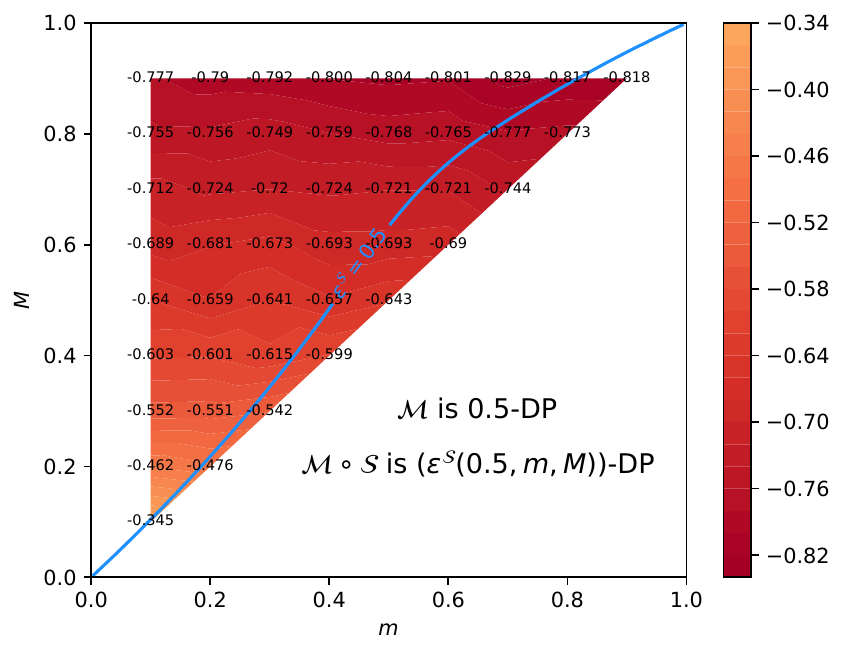}%
    \includegraphics[width=0.25\textwidth]{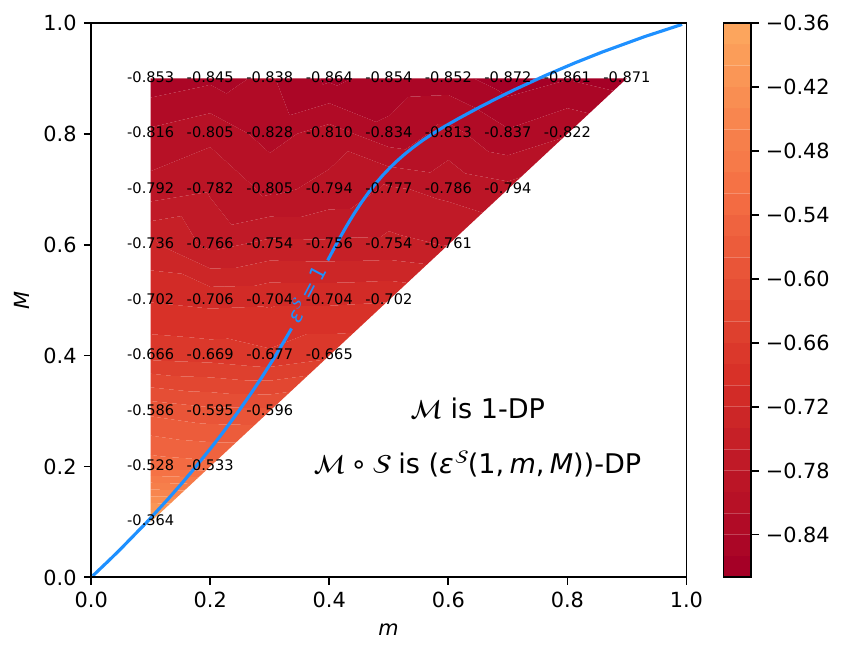}%
    \includegraphics[width=0.25\textwidth]{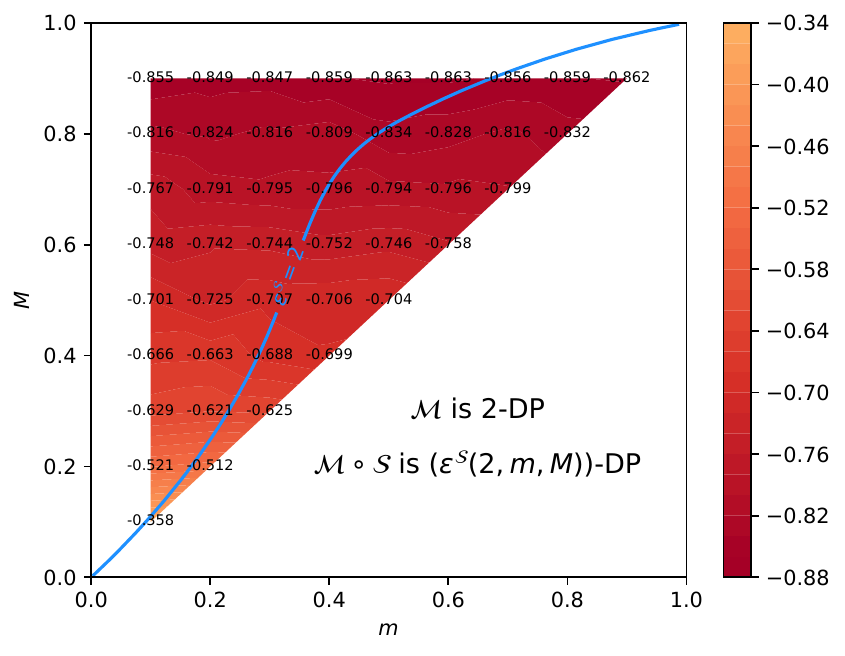}%
    \caption{The probability of outputting an incorrect mode of $\M$ minus that of $\M\circ\S$ without the noise reduction. Results shown for the RNM with exponential noise over the \texttt{age} column in the Adult database.}
    \label{fig:Experiment1-RNM-Exponential-Adult-age}
\end{figure}

\begin{figure}[H]
    \includegraphics[width=0.25\textwidth]{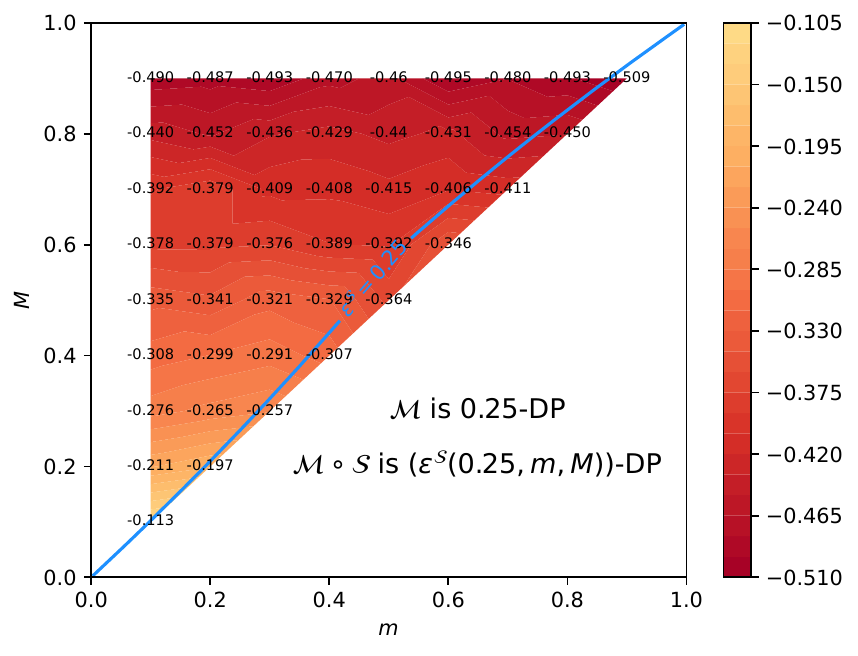}%
    \includegraphics[width=0.25\textwidth]{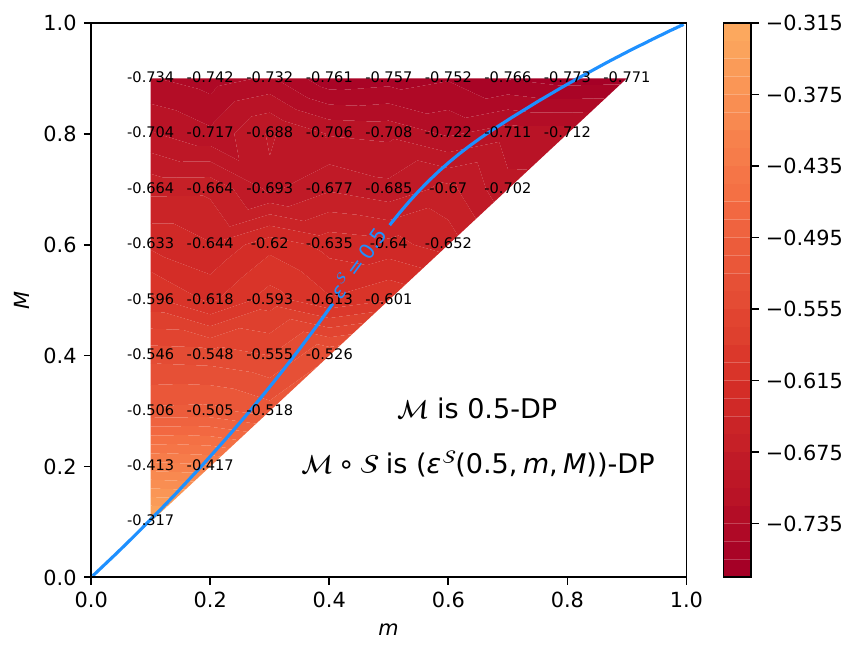}%
    \includegraphics[width=0.25\textwidth]{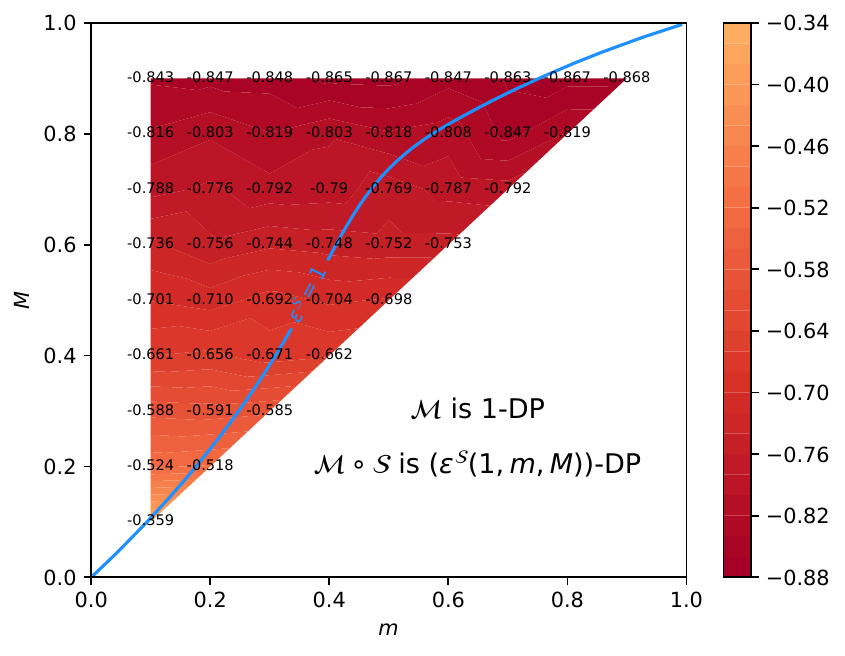}%
    \includegraphics[width=0.25\textwidth]{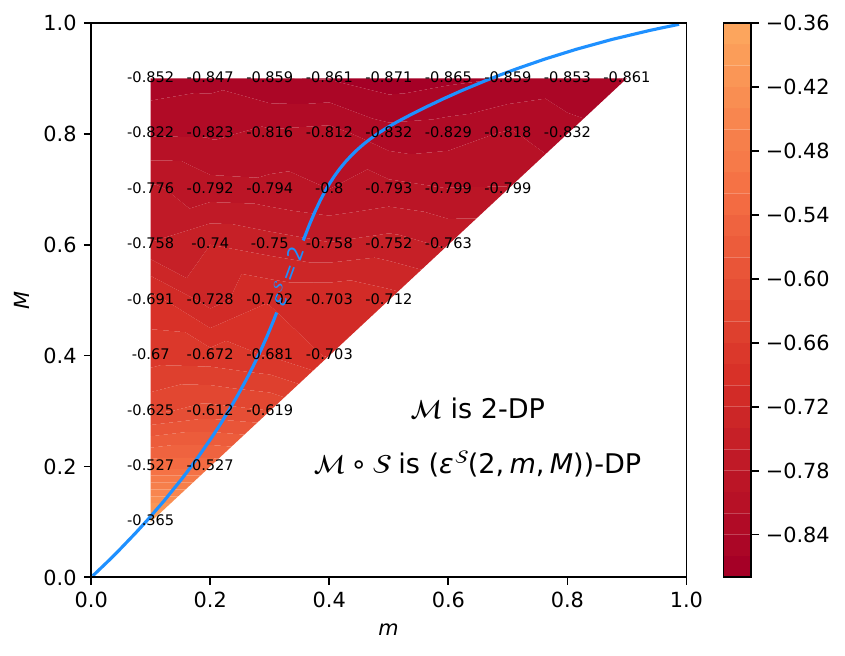}%
    \caption{The probability of outputting an incorrect mode of $\M$ minus that of $\M\circ\S$ without the noise reduction. Results shown for the exponential mechanism over the \texttt{age} column in the Adult database.}
    \label{fig:Experiment1-RNM-ExponentialMechanism-Adult-age}
\end{figure}

\begin{figure}[H]
    \includegraphics[width=0.25\textwidth]{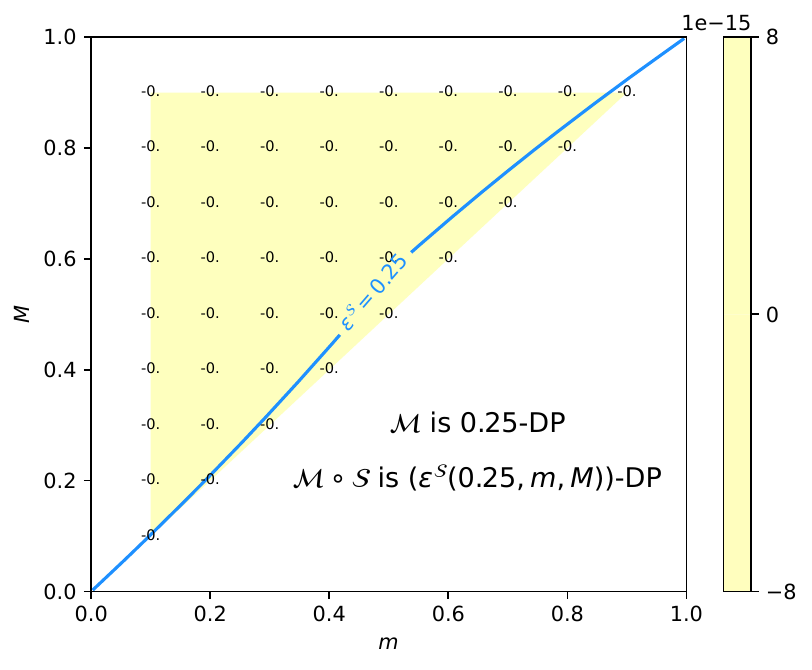}%
    \includegraphics[width=0.25\textwidth]{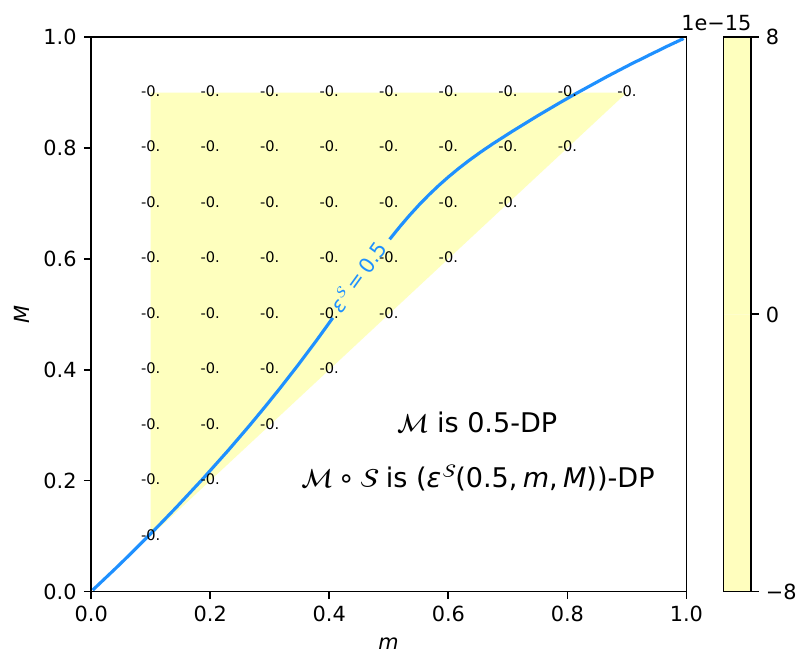}%
    \includegraphics[width=0.25\textwidth]{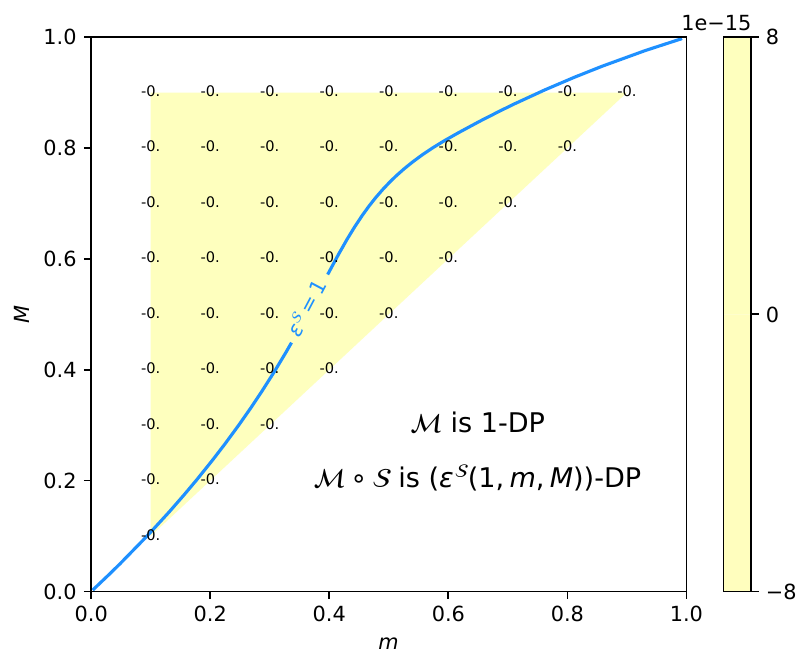}%
    \includegraphics[width=0.25\textwidth]{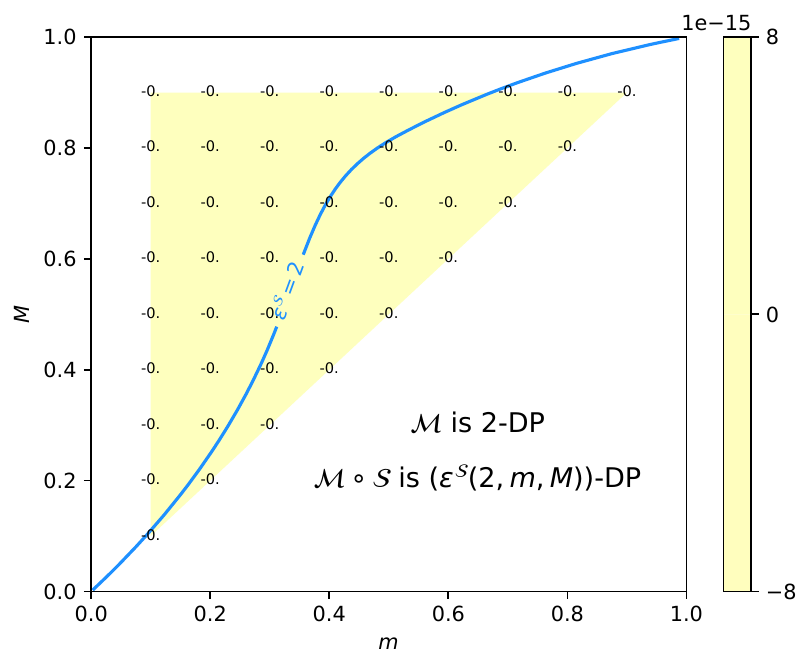}%
    \caption{The probability of outputting an incorrect mode of $\M$ minus that of $\M\circ\S$ without the noise reduction. Results shown for the RNM with Laplace noise over the \texttt{hours-per-week} column in the Adult database.}
    \label{fig:Experiment1-RNM-Laplace-Adult-hours-per-week}
\end{figure}

\begin{figure}[H]
    \includegraphics[width=0.25\textwidth]{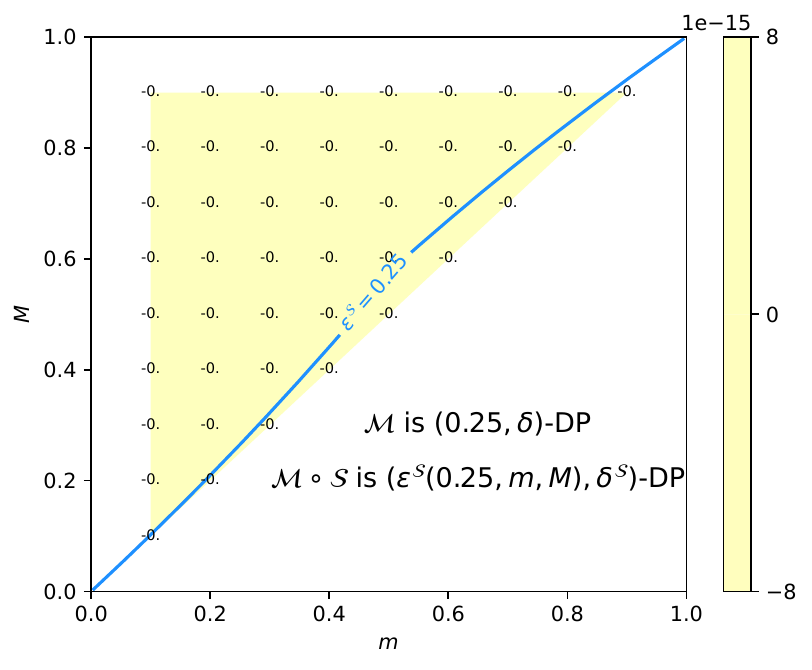}%
    \includegraphics[width=0.25\textwidth]{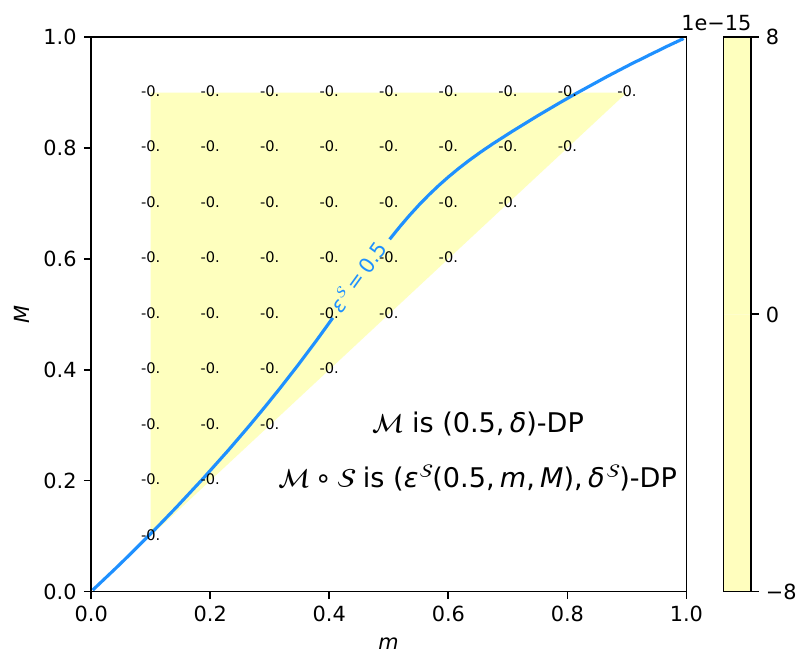}%
    \includegraphics[width=0.25\textwidth]{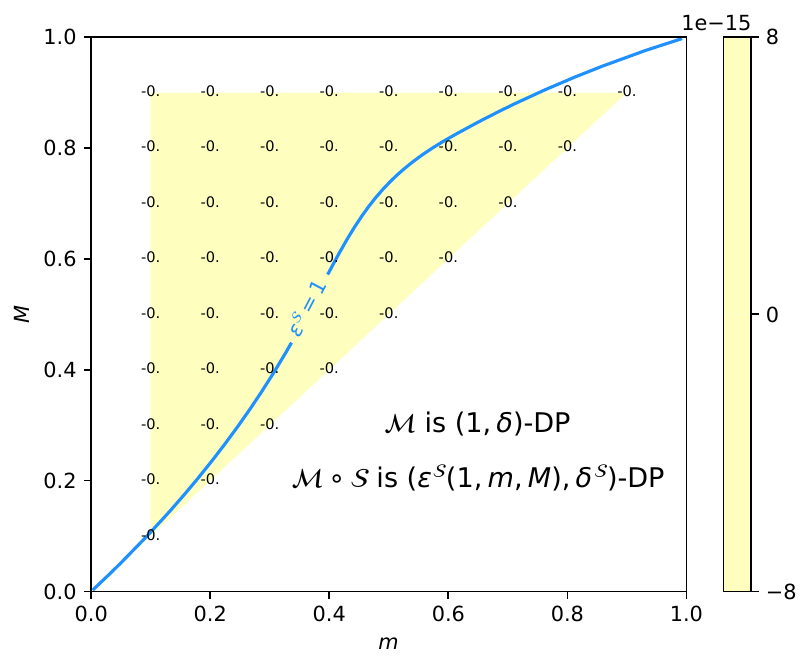}%
    \includegraphics[width=0.25\textwidth]{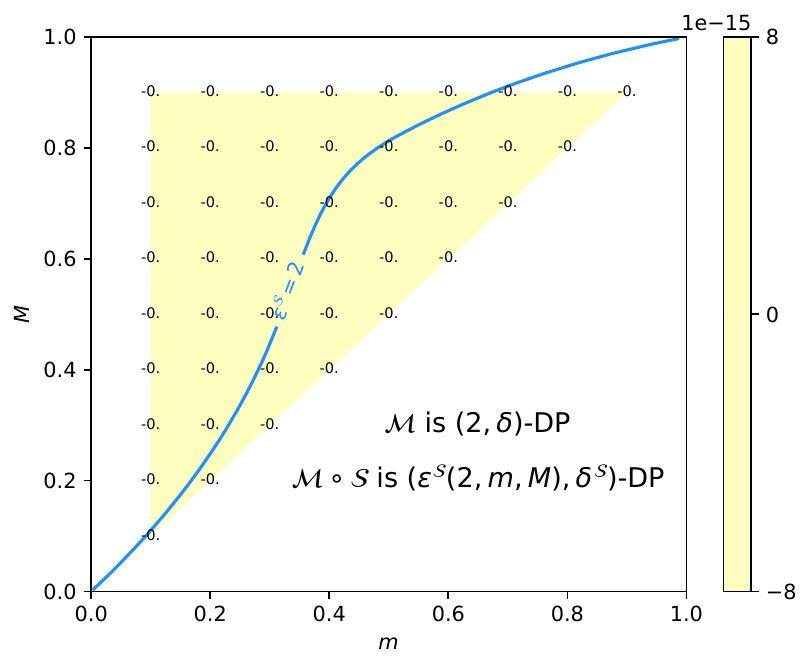}%
    \caption{The probability of outputting an incorrect mode of $\M$ minus that of $\M\circ\S$ without the noise reduction. Results shown for the RNM-like variant with Gaussian noise over the \texttt{hours-per-week} column in the Adult database.}
    \label{fig:Experiment1-RNM-Gaussian-Adult-hours-per-week}
\end{figure}

\begin{figure}[H]
    \includegraphics[width=0.25\textwidth]{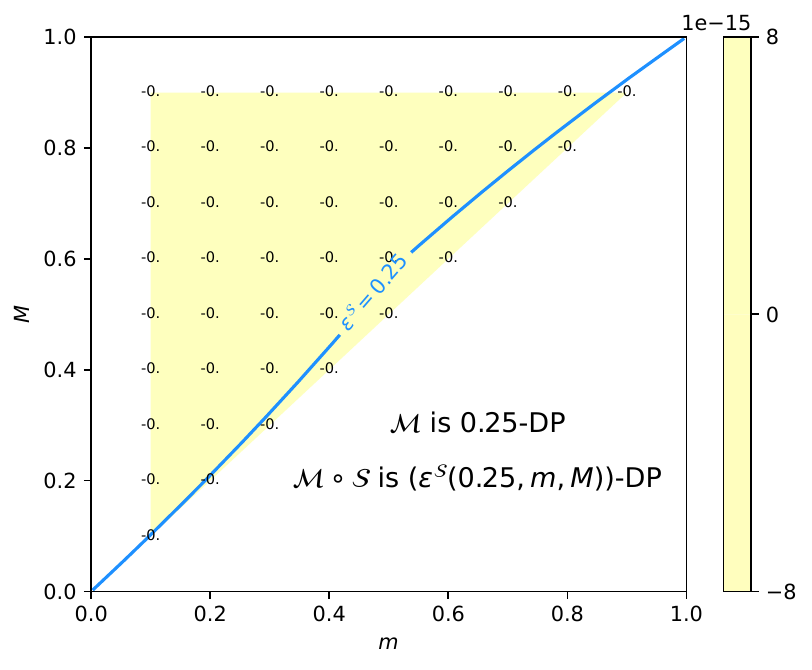}%
    \includegraphics[width=0.25\textwidth]{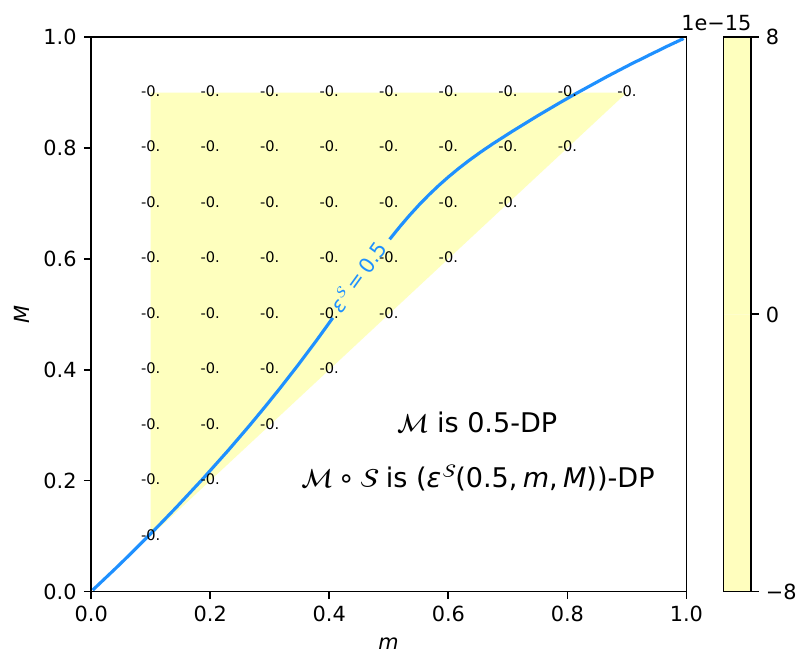}%
    \includegraphics[width=0.25\textwidth]{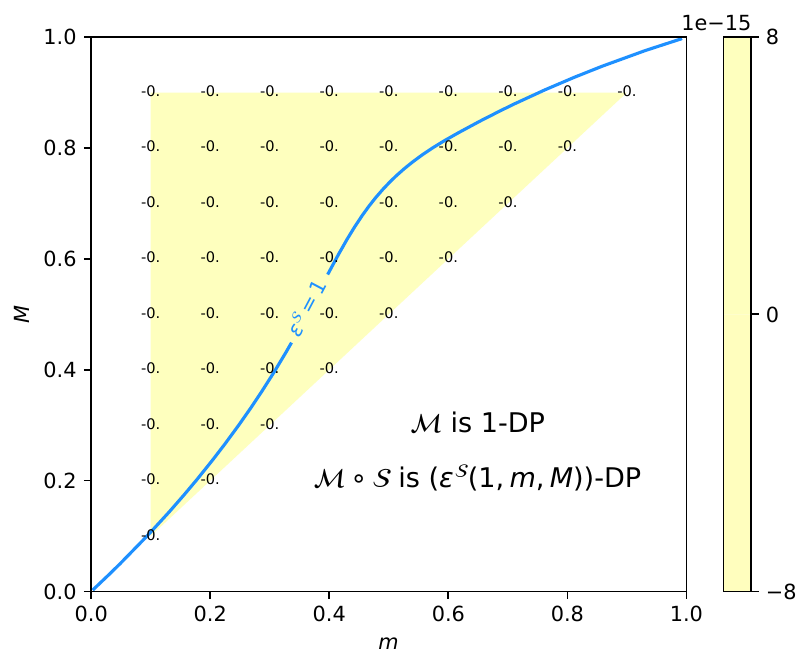}%
    \includegraphics[width=0.25\textwidth]{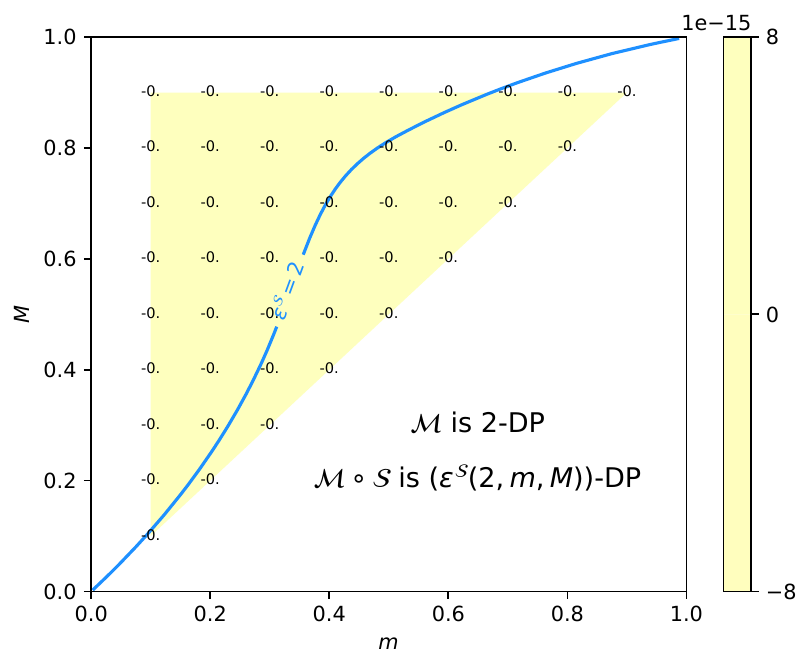}%
    \caption{The probability of outputting an incorrect mode of $\M$ minus that of $\M\circ\S$ without the noise reduction. Results shown for the RNM with exponential noise over the \texttt{hours-per-week} column in the Adult database.}
    \label{fig:Experiment1-RNM-Exponential-Adult-hours-per-week}
\end{figure}

\begin{figure}[H]
    \includegraphics[width=0.25\textwidth]{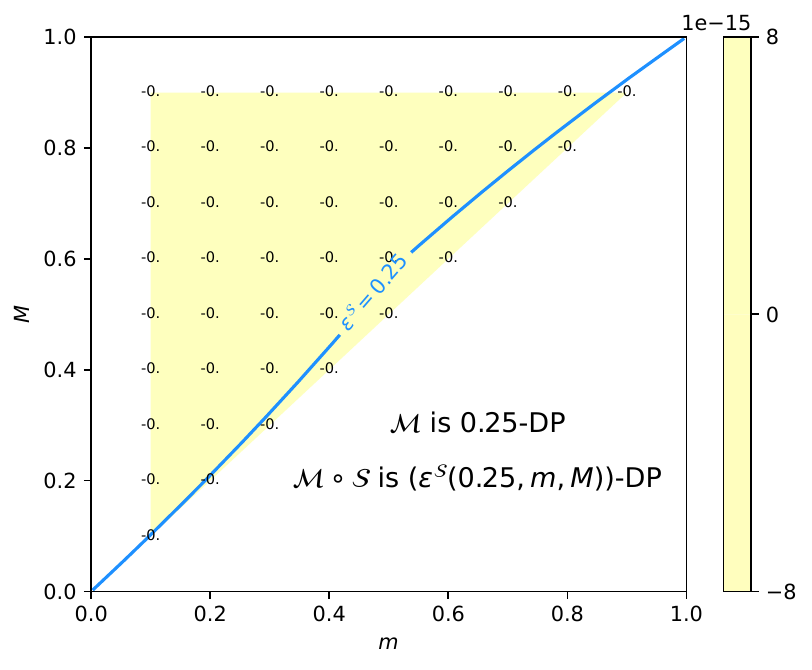}%
    \includegraphics[width=0.25\textwidth]{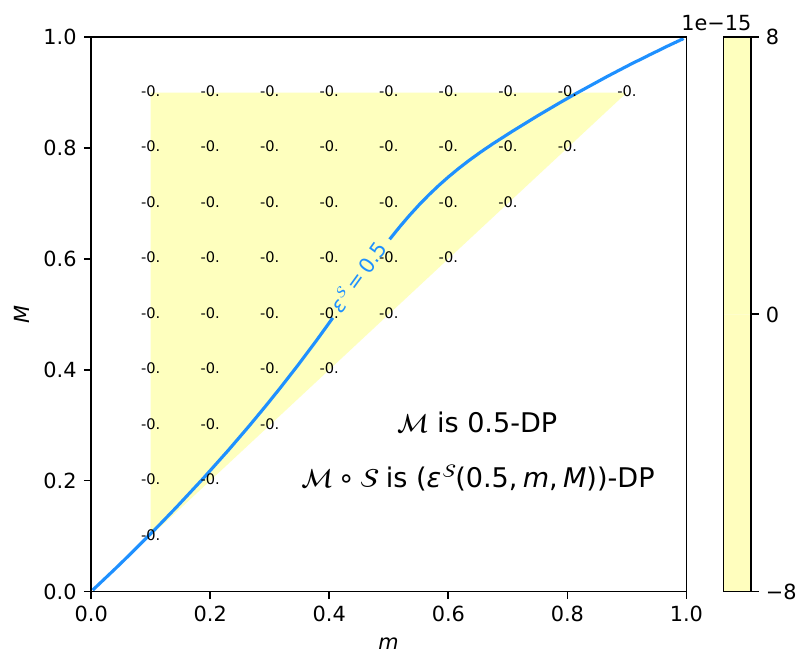}%
    \includegraphics[width=0.25\textwidth]{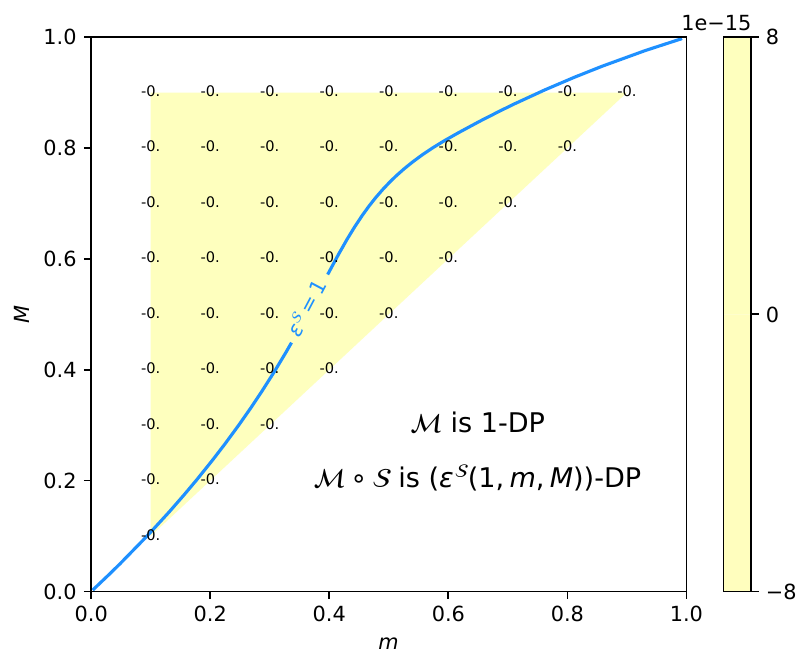}%
    \includegraphics[width=0.25\textwidth]{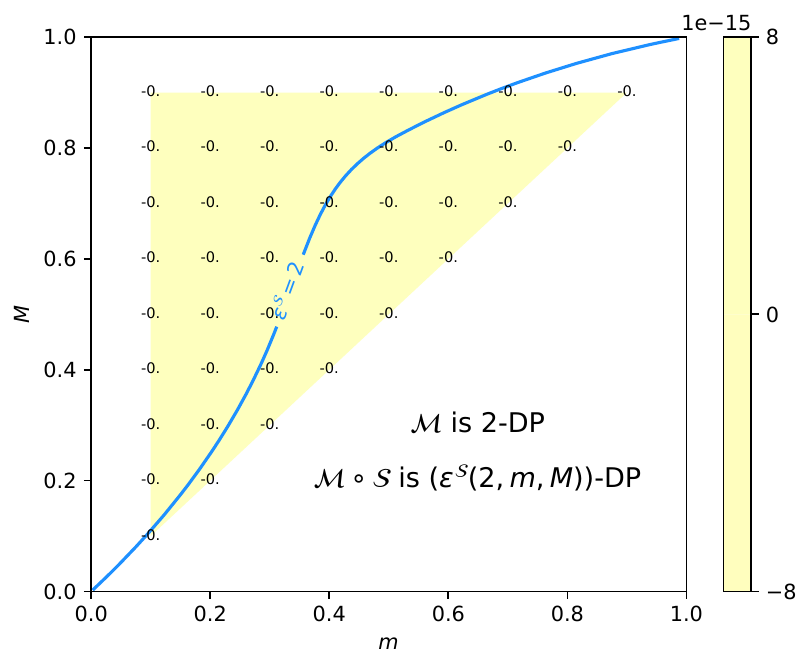}%
    \caption{The probability of outputting an incorrect mode of $\M$ minus that of $\M\circ\S$ without the noise reduction. Results shown for the exponential mechanism over the \texttt{hours-per-week} column in the Adult database.}
    \label{fig:Experiment1-RNM-ExponentialMechanism-Adult-hours-per-week}
\end{figure}

\begin{figure}[H]
    \includegraphics[width=0.25\textwidth]{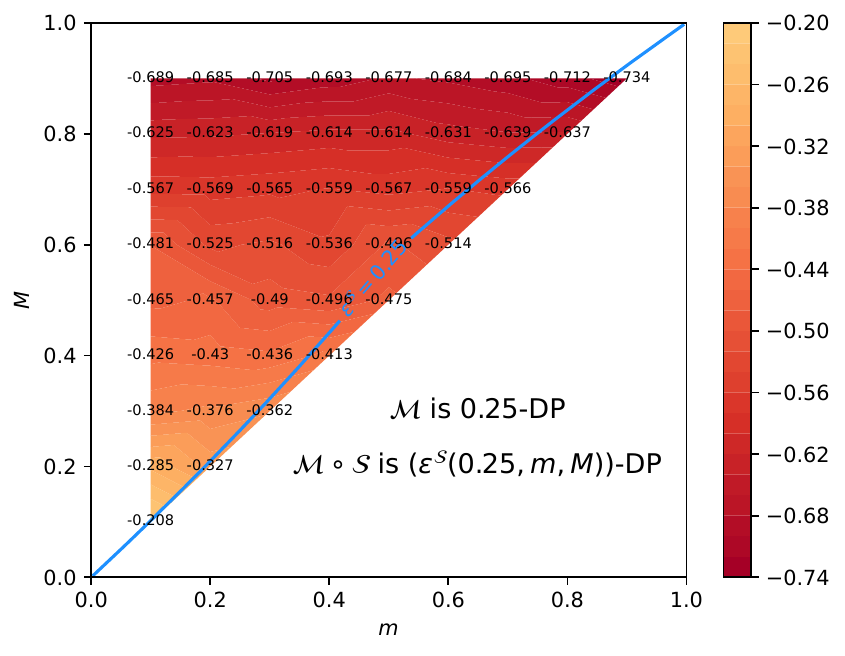}%
    \includegraphics[width=0.25\textwidth]{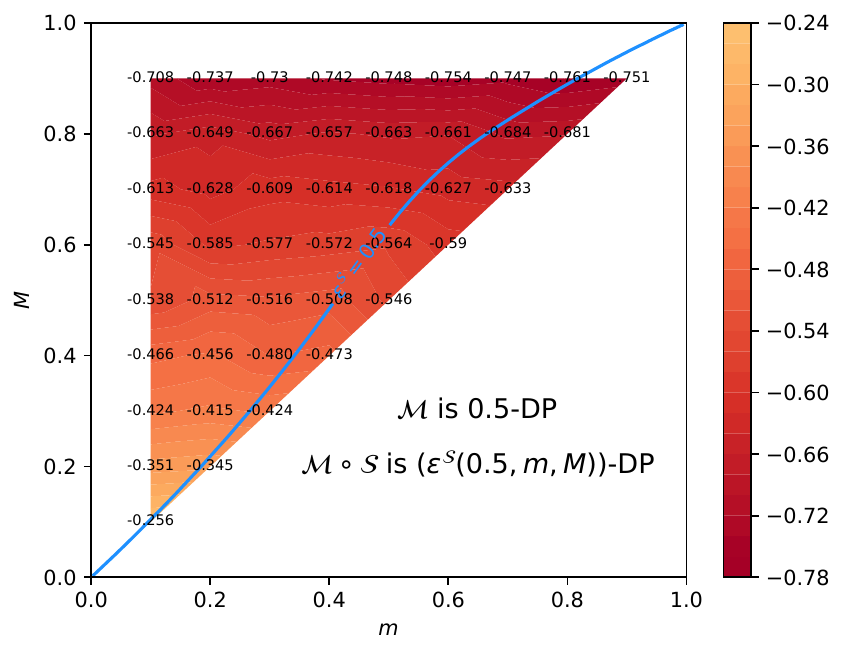}%
    \includegraphics[width=0.25\textwidth]{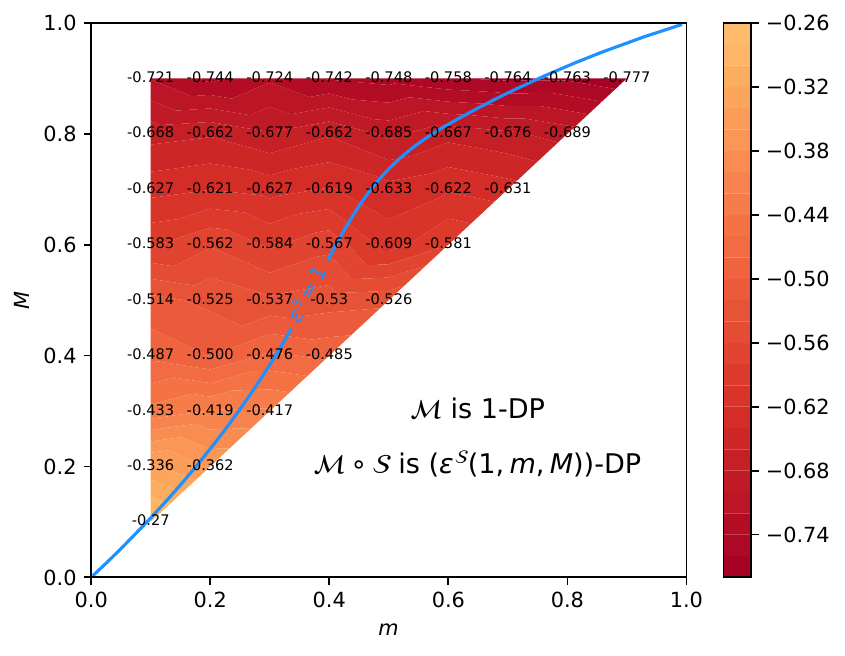}%
    \includegraphics[width=0.25\textwidth]{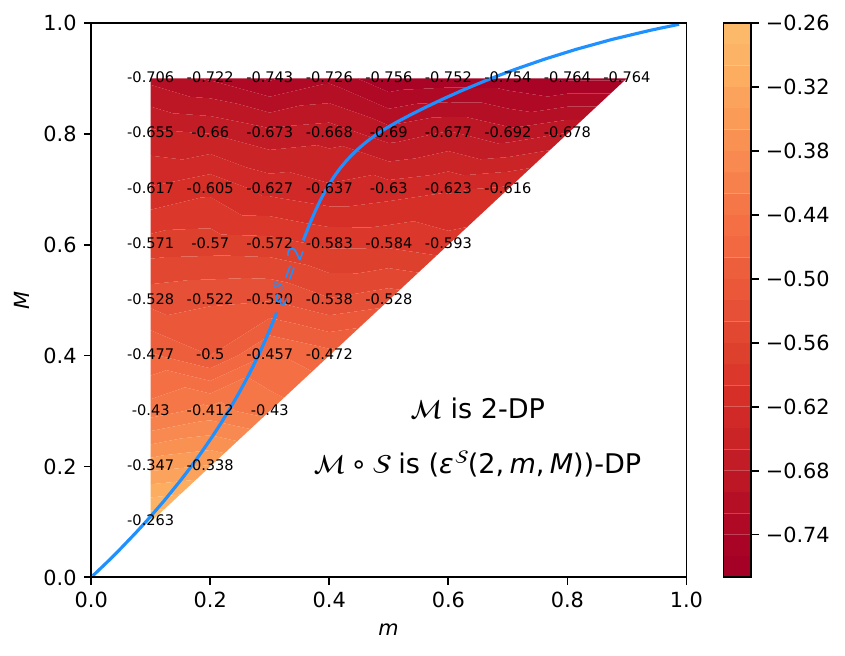}%
    \caption{The probability of outputting an incorrect mode of $\M$ minus that of $\M\circ\S$ without the noise reduction. Results shown for the RNM with Laplace noise over the \texttt{Age} column in the Irish database.}
    \label{fig:Experiment1-RNM-Laplace-Irishn-Age}
\end{figure}

\begin{figure}[H]
    \includegraphics[width=0.25\textwidth]{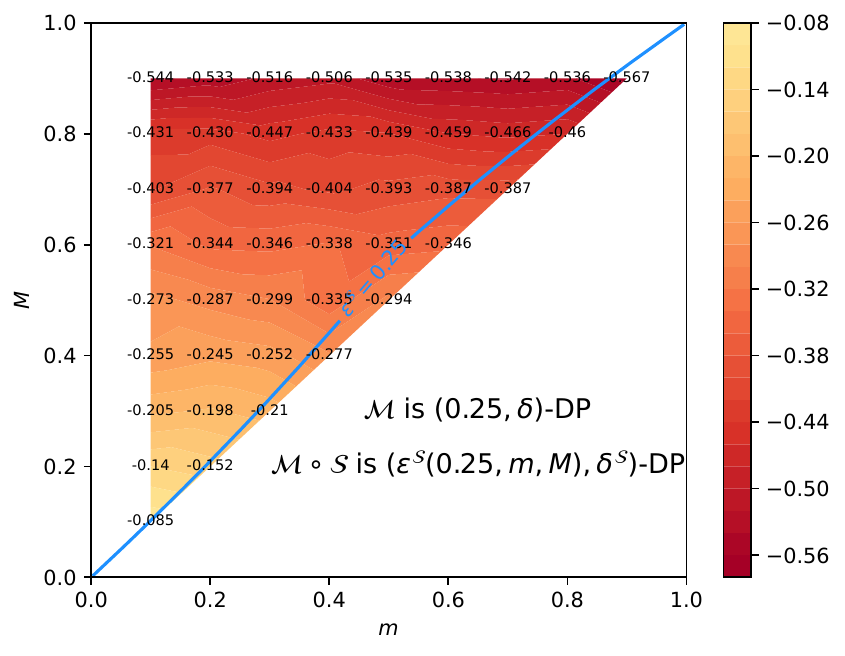}%
    \includegraphics[width=0.25\textwidth]{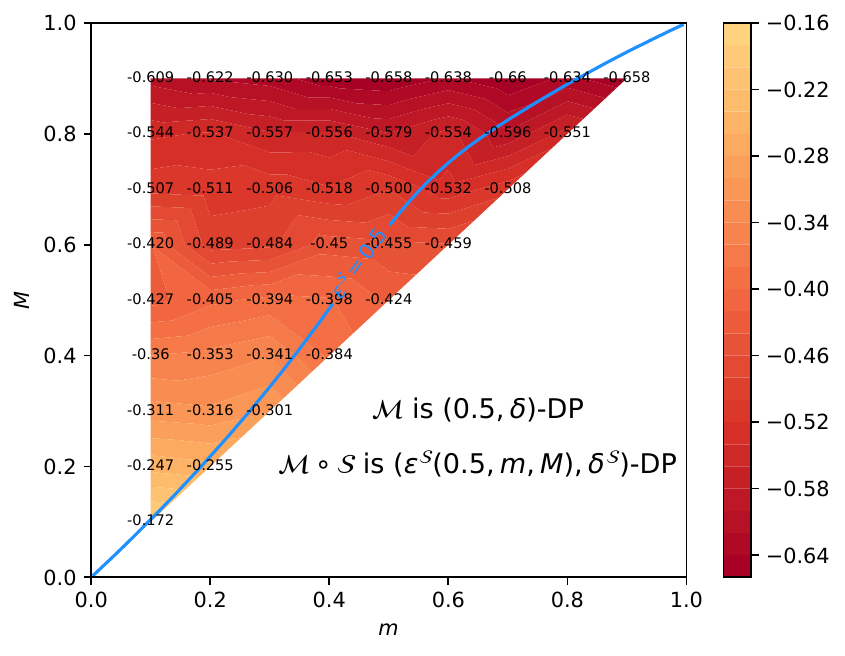}%
    \includegraphics[width=0.25\textwidth]{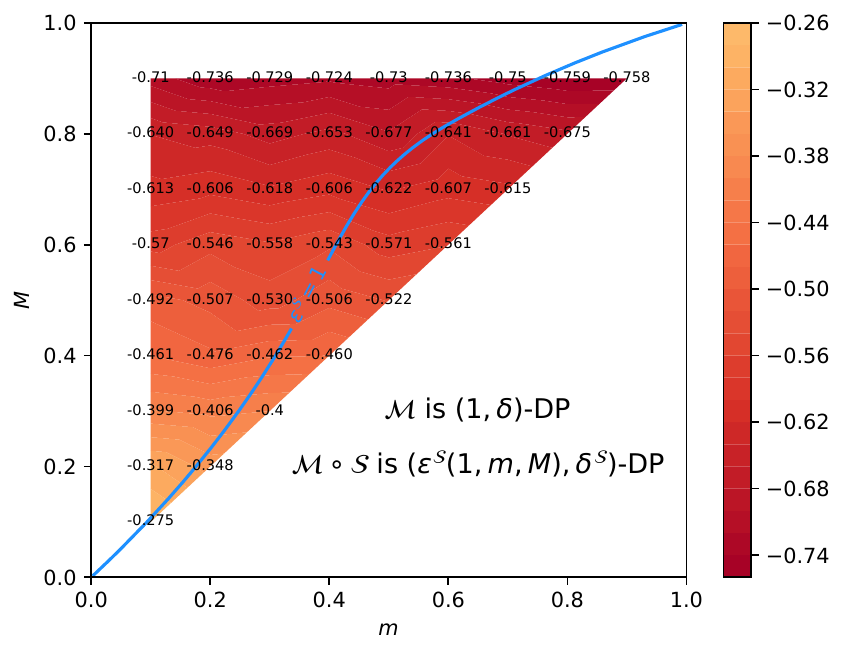}%
    \includegraphics[width=0.25\textwidth]{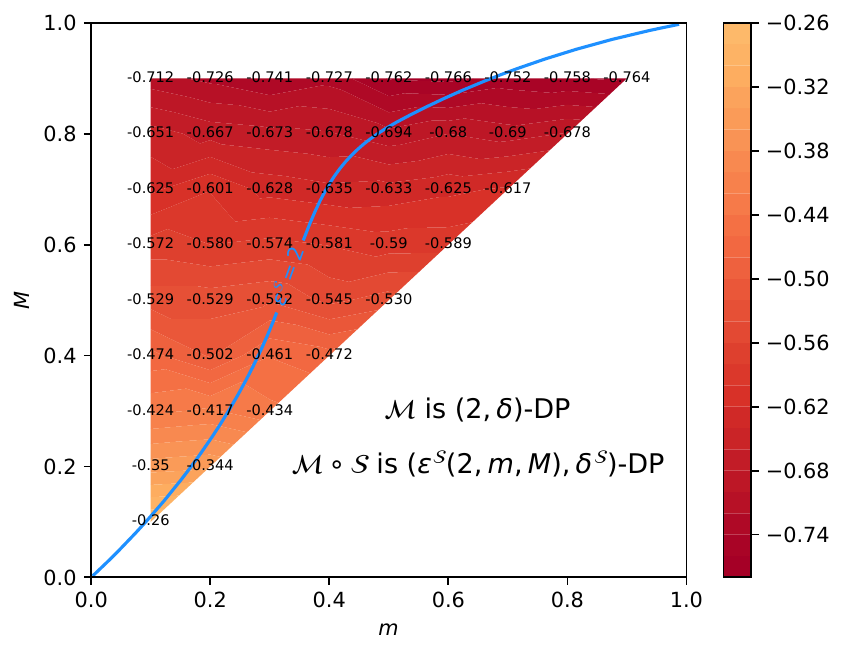}%
    \caption{The probability of outputting an incorrect mode of $\M$ minus that of $\M\circ\S$ without the noise reduction. Results shown for the RNM-like variant with Gaussian noise over the \texttt{Age} column in the Irish database.}
    \label{fig:Experiment1-RNM-Gaussian-Irishn-Age}
\end{figure}

\begin{figure}[H]
    \includegraphics[width=0.25\textwidth]{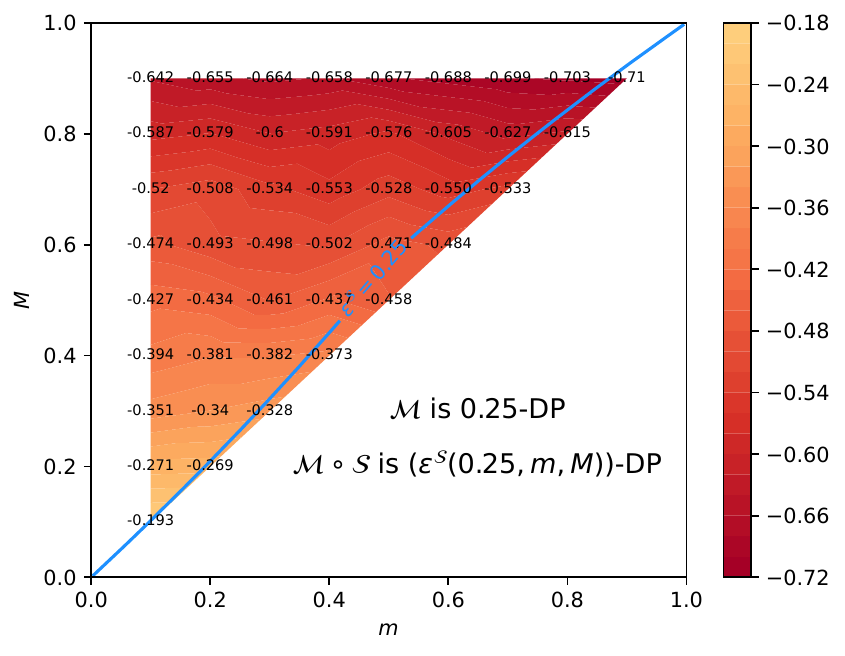}%
    \includegraphics[width=0.25\textwidth]{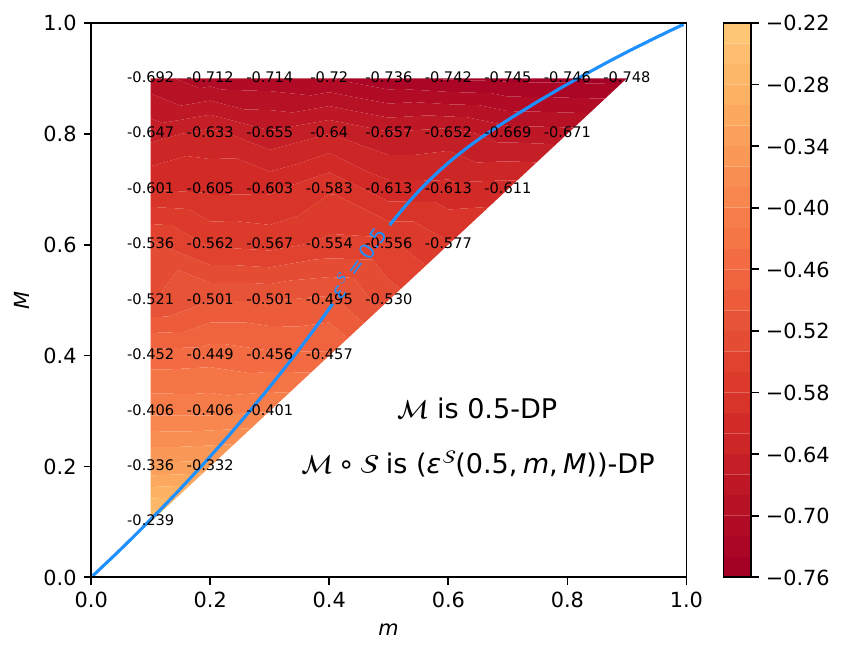}%
    \includegraphics[width=0.25\textwidth]{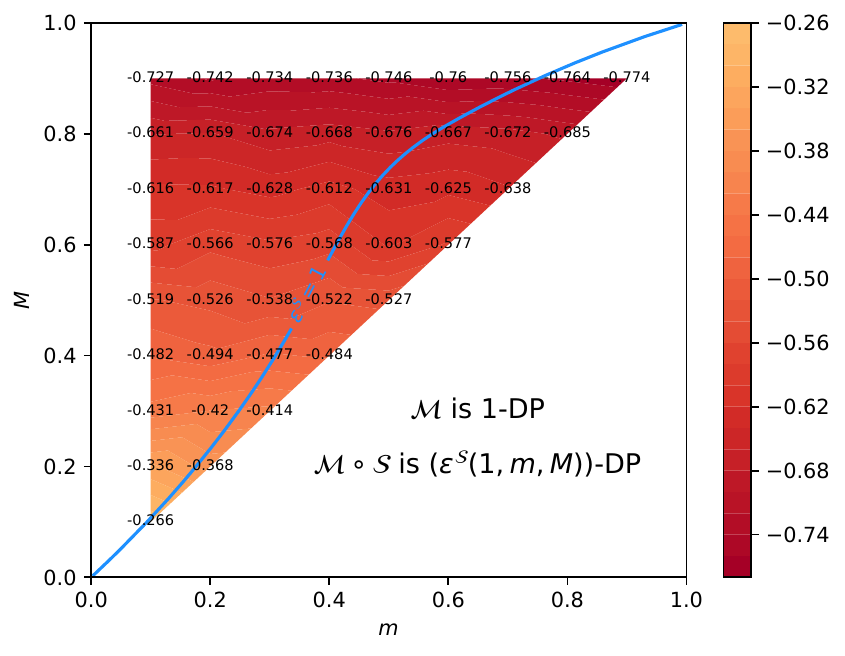}%
    \includegraphics[width=0.25\textwidth]{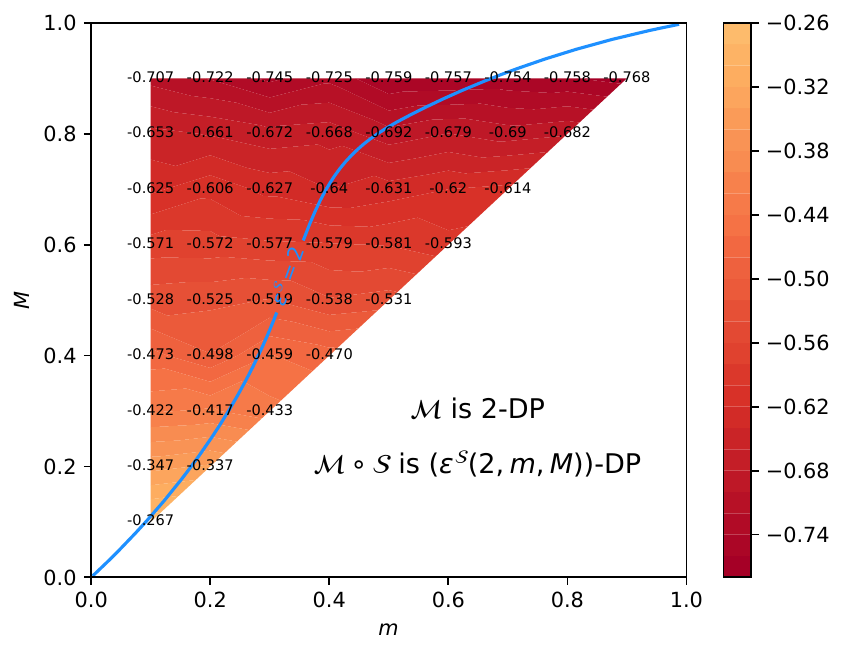}%
    \caption{The probability of outputting an incorrect mode of $\M$ minus that of $\M\circ\S$ without the noise reduction. Results shown for the RNM with exponential noise over the \texttt{Age} column in the Irish database.}
    \label{fig:Experiment1-RNM-Exponential-Irishn-Age}
\end{figure}

\begin{figure}[H]
    \includegraphics[width=0.25\textwidth]{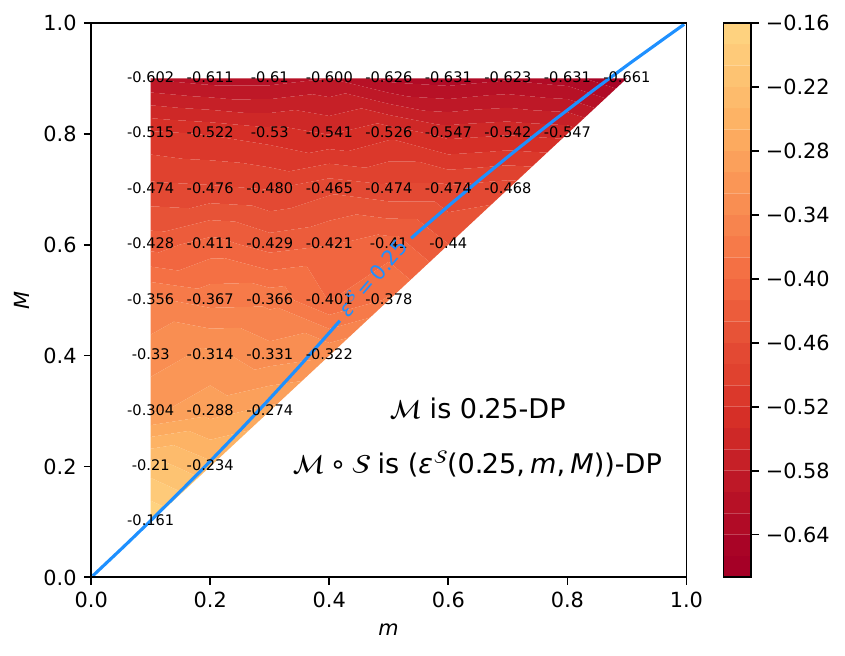}%
    \includegraphics[width=0.25\textwidth]{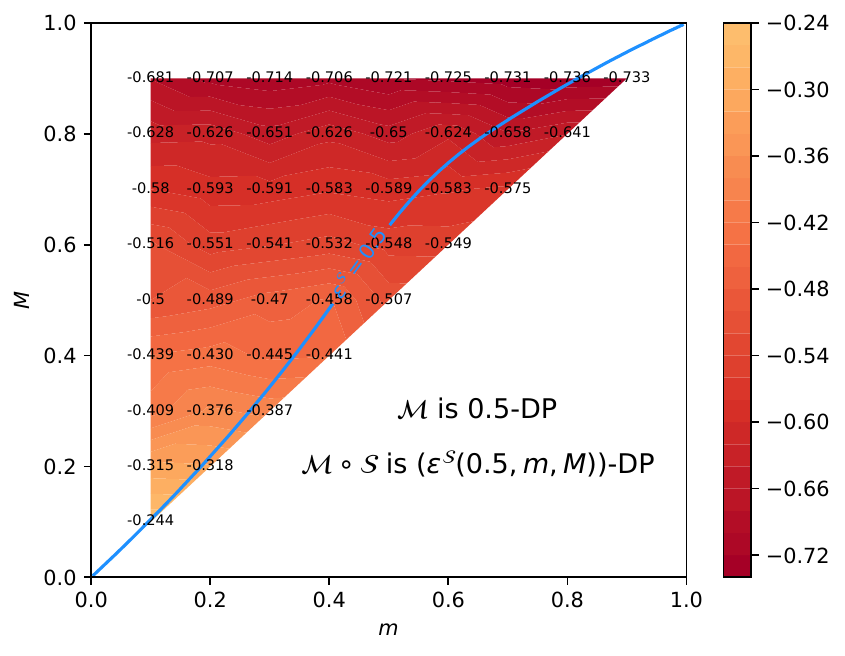}%
    \includegraphics[width=0.25\textwidth]{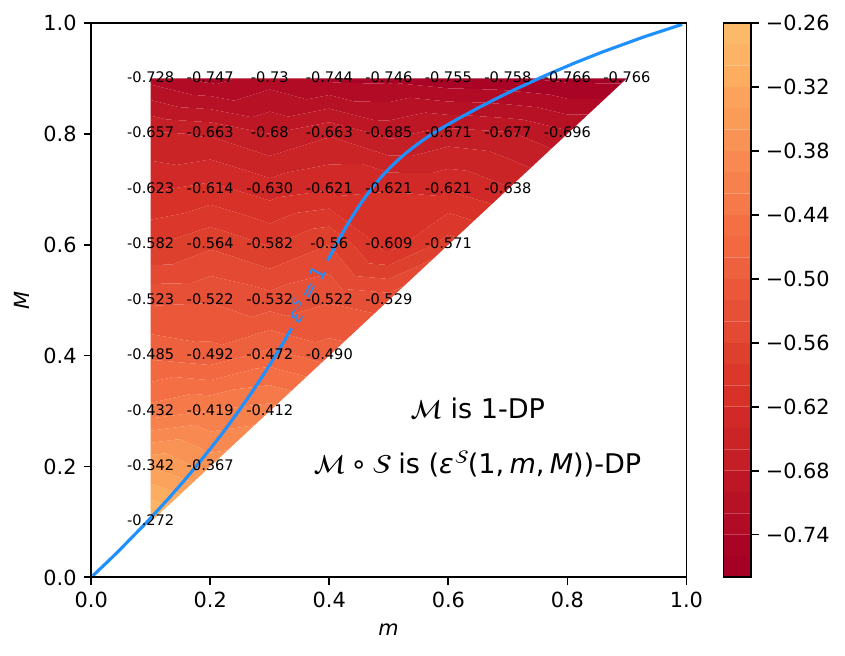}%
    \includegraphics[width=0.25\textwidth]{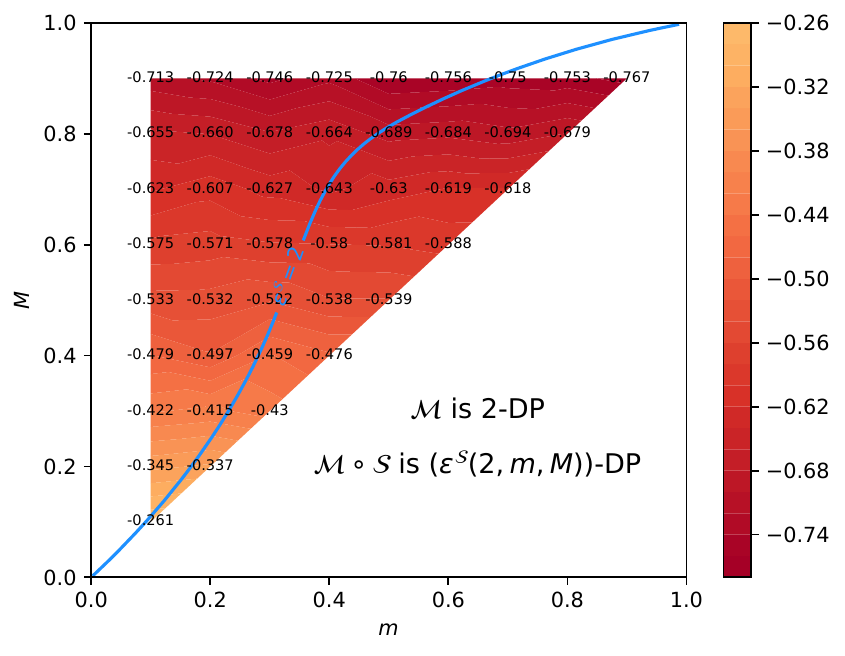}%
    \caption{The probability of outputting an incorrect mode of $\M$ minus that of $\M\circ\S$ without the noise reduction. Results shown for the exponential mechanism over the \texttt{Age} column in the Irish database.}
    \label{fig:Experiment1-RNM-ExponentialMechanism-Irishn-Age}
\end{figure}

\begin{figure}[H]
    \includegraphics[width=0.25\textwidth]{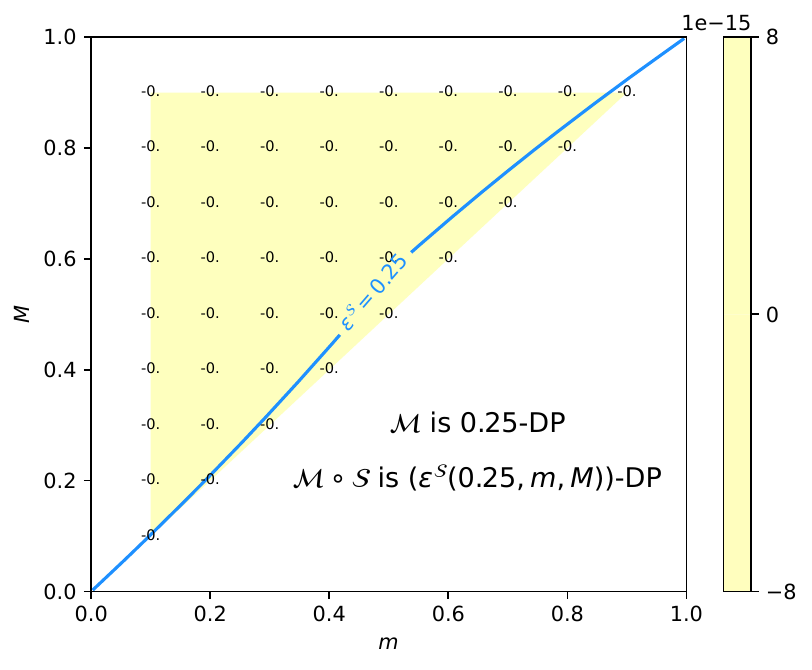}%
    \includegraphics[width=0.25\textwidth]{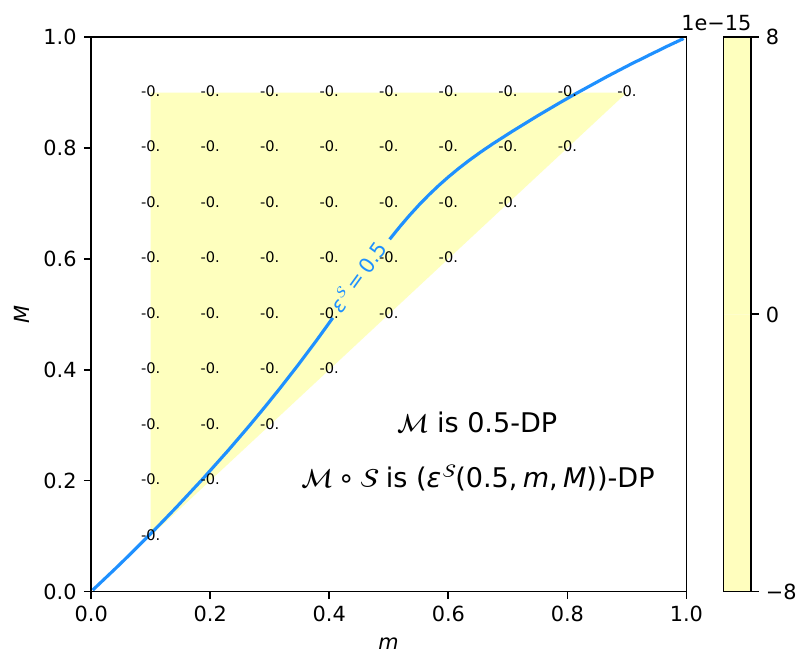}%
    \includegraphics[width=0.25\textwidth]{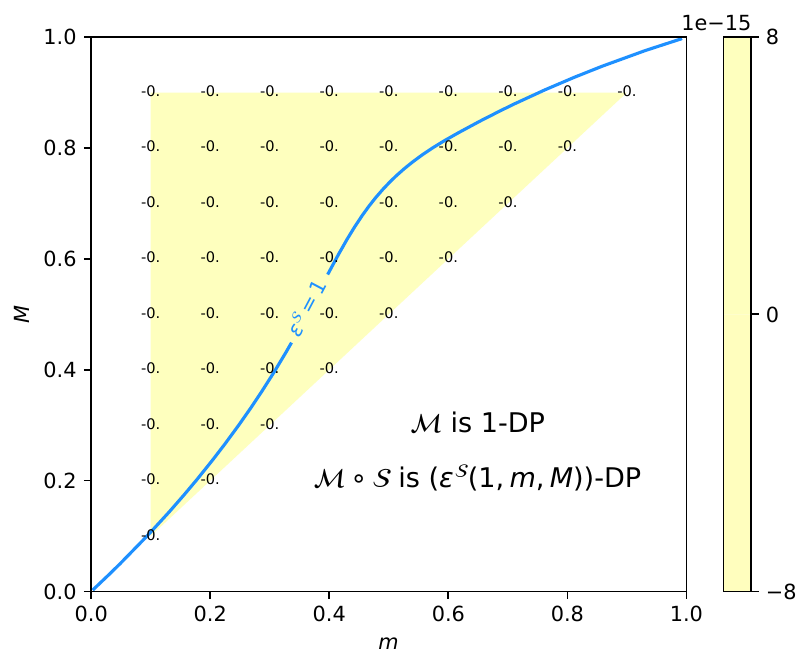}%
    \includegraphics[width=0.25\textwidth]{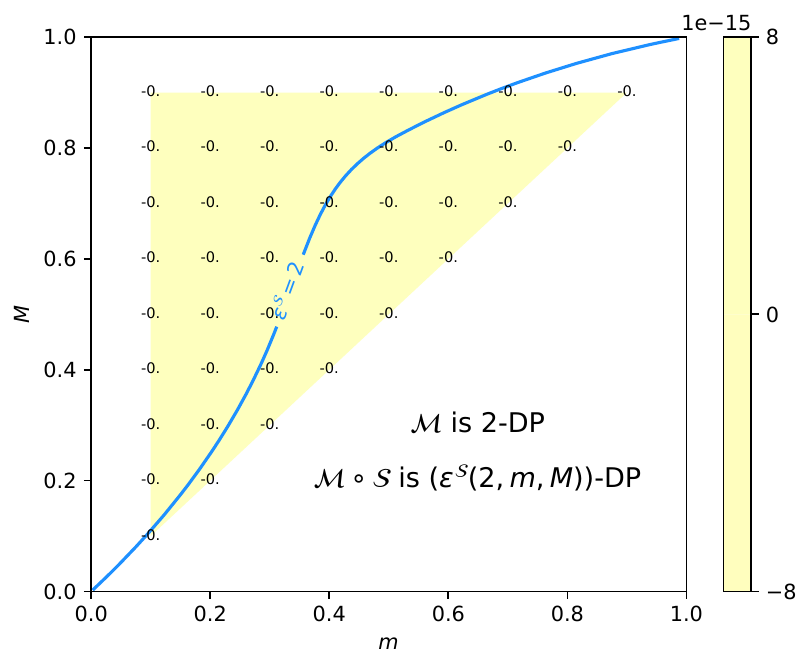}%
    \caption{The probability of outputting an incorrect mode of $\M$ minus that of $\M\circ\S$ without the noise reduction. Results shown for the RNM with Laplace noise over the \texttt{HighestEducationCompleted} column in the Irish database.}
    \label{fig:Experiment1-RNM-Laplace-Irishn-HighestEducationCompleted}
\end{figure}

\begin{figure}[H]
    \includegraphics[width=0.25\textwidth]{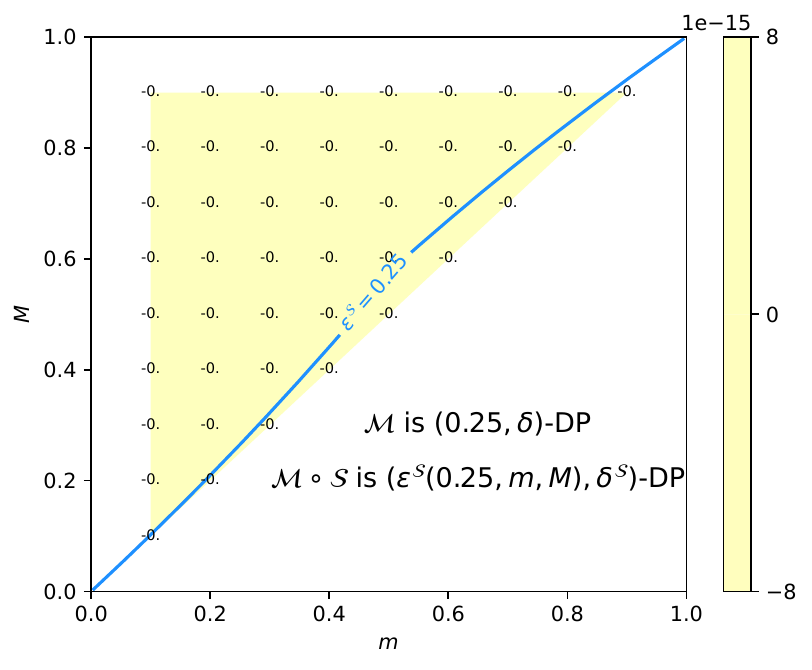}%
    \includegraphics[width=0.25\textwidth]{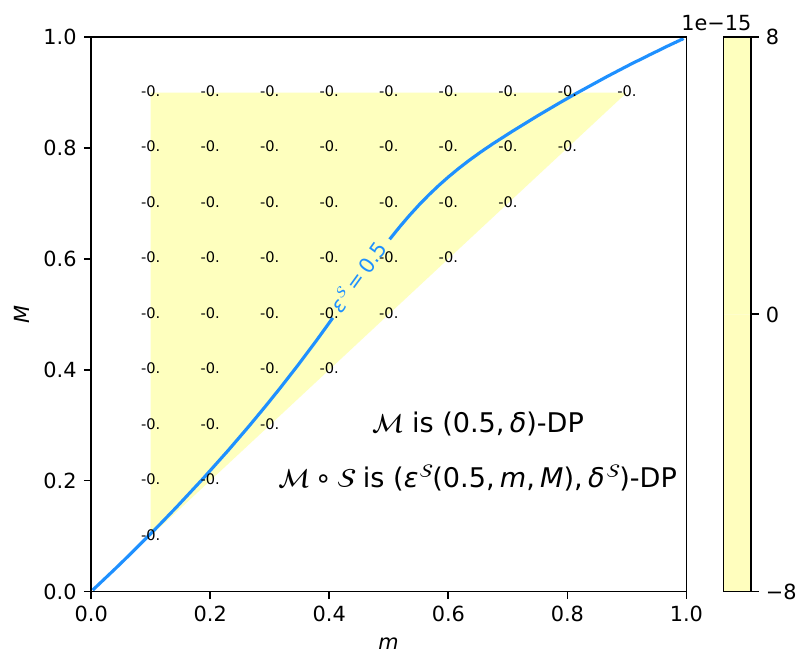}%
    \includegraphics[width=0.25\textwidth]{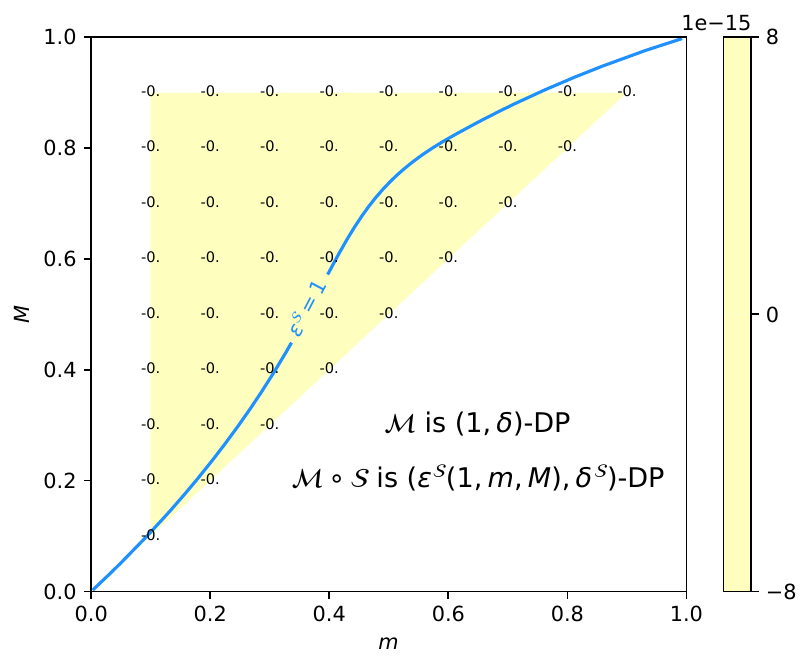}%
    \includegraphics[width=0.25\textwidth]{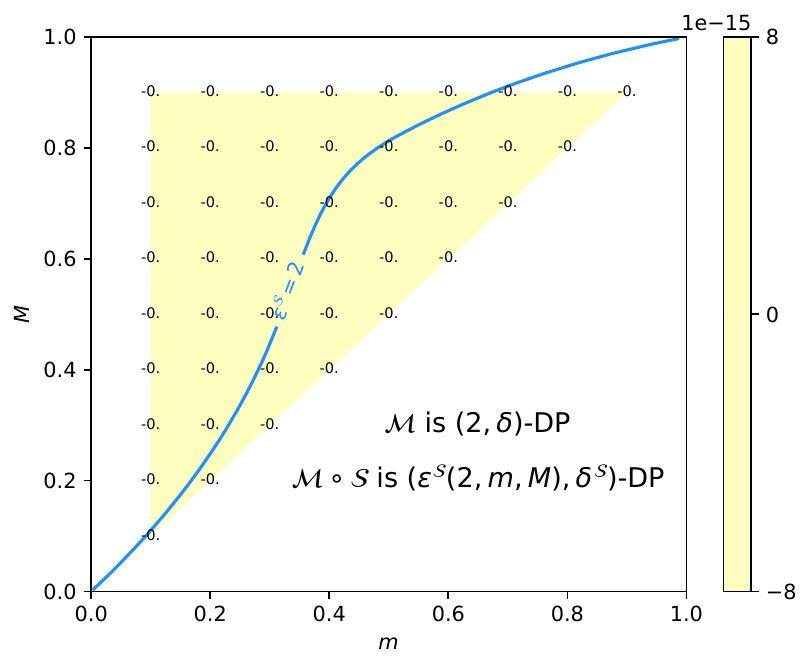}%
    \caption{The probability of outputting an incorrect mode of $\M$ minus that of $\M\circ\S$ without the noise reduction. Results shown for the RNM-like variant with Gaussian noise over the \texttt{HighestEducationCompleted} column in the Irish database.}
    \label{fig:Experiment1-RNM-Gaussian-Irishn-HighestEducationCompleted}
\end{figure}

\begin{figure}[H]
    \includegraphics[width=0.25\textwidth]{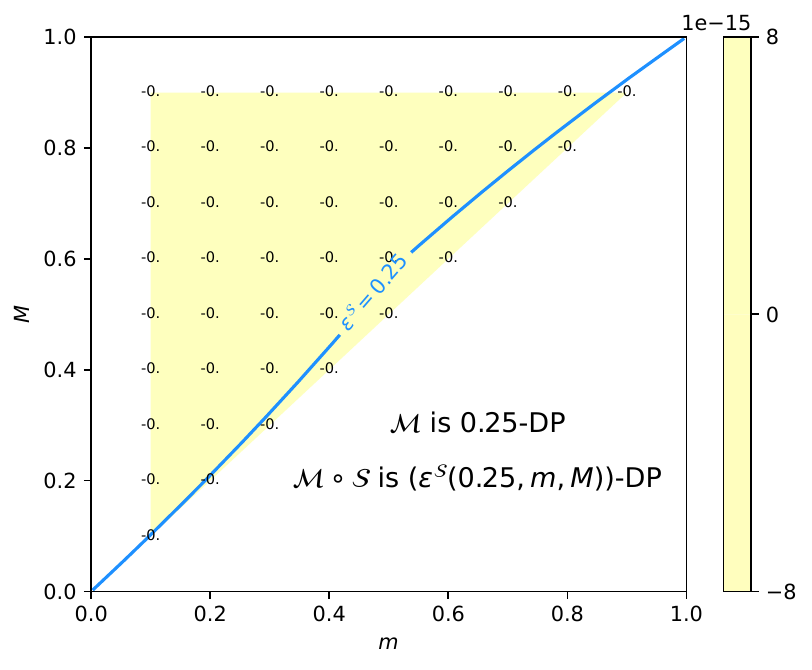}%
    \includegraphics[width=0.25\textwidth]{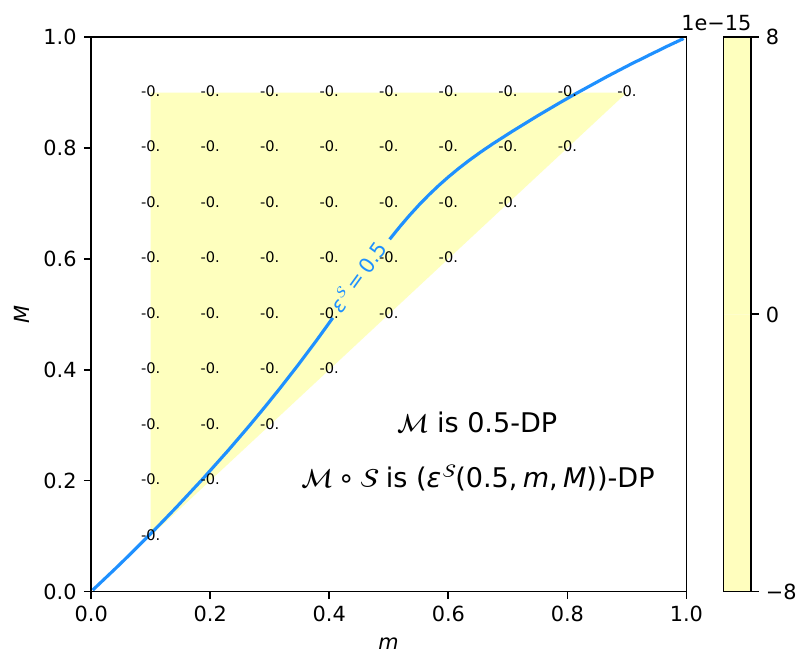}%
    \includegraphics[width=0.25\textwidth]{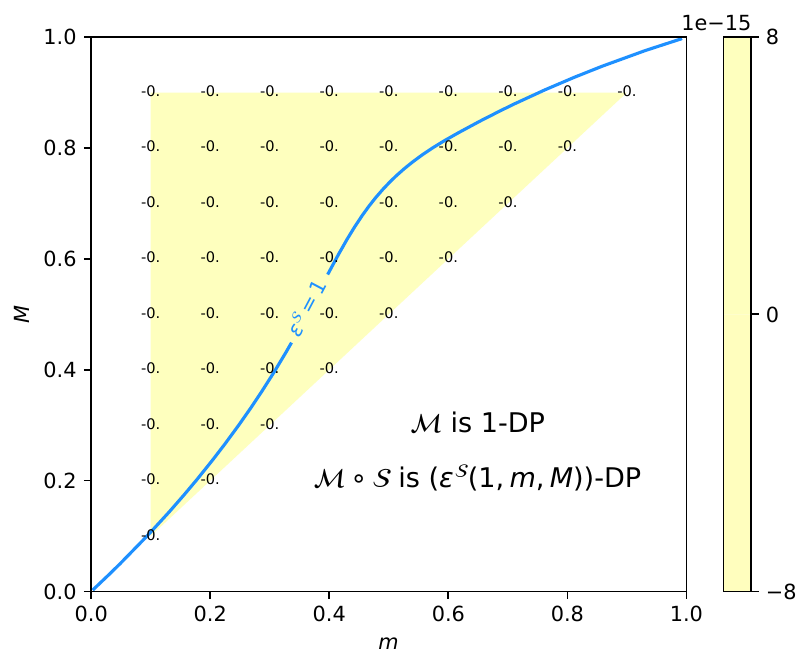}%
    \includegraphics[width=0.25\textwidth]{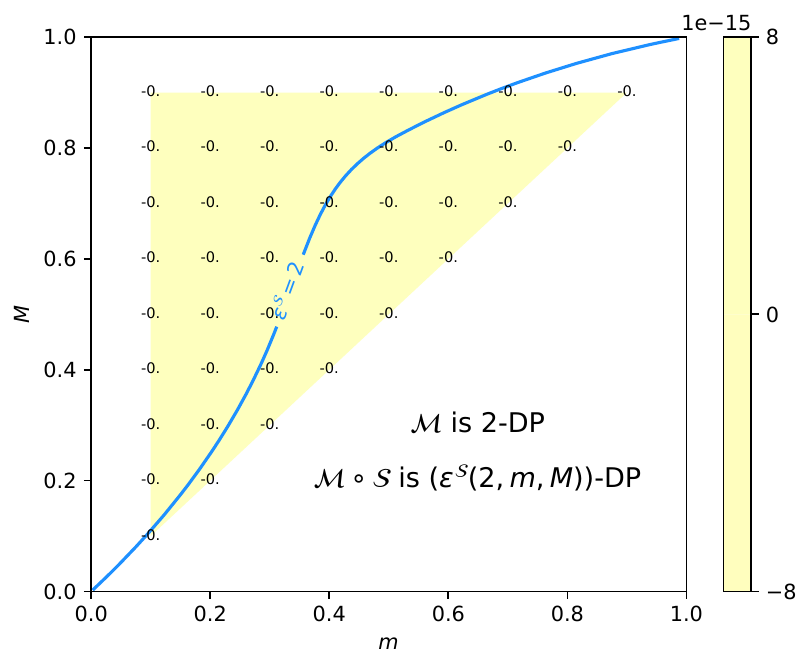}%
    \caption{The probability of outputting an incorrect mode of $\M$ minus that of $\M\circ\S$ without the noise reduction. Results shown for the RNM with exponential noise over the \texttt{HighestEducationCompleted} column in the Irish database.}
    \label{fig:Experiment1-RNM-Exponential-Irishn-HighestEducationCompleted}
\end{figure}

\begin{figure}[H]
    \includegraphics[width=0.25\textwidth]{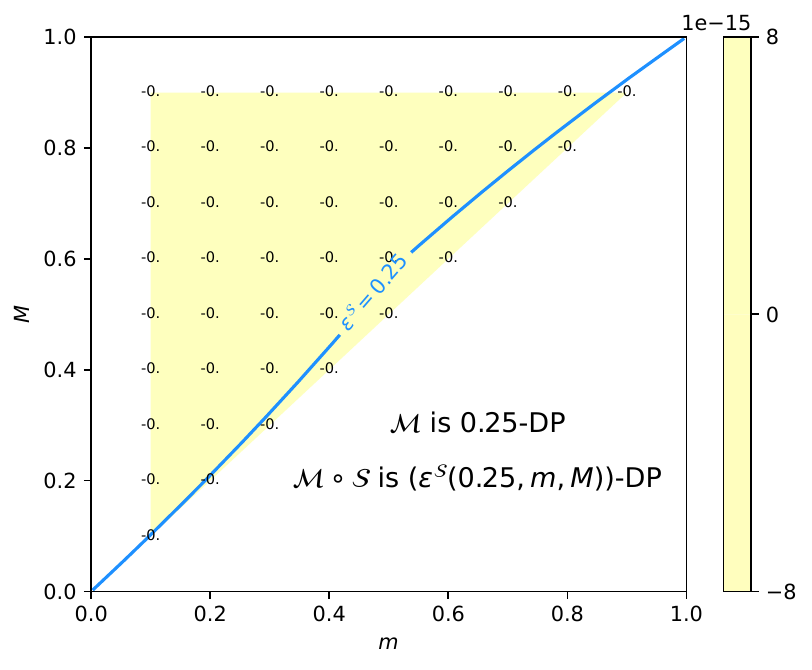}%
    \includegraphics[width=0.25\textwidth]{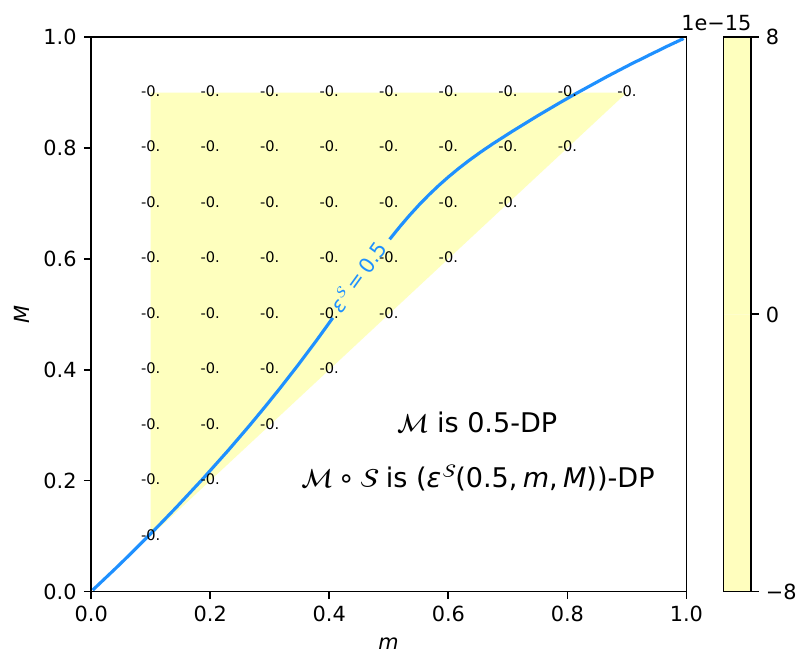}%
    \includegraphics[width=0.25\textwidth]{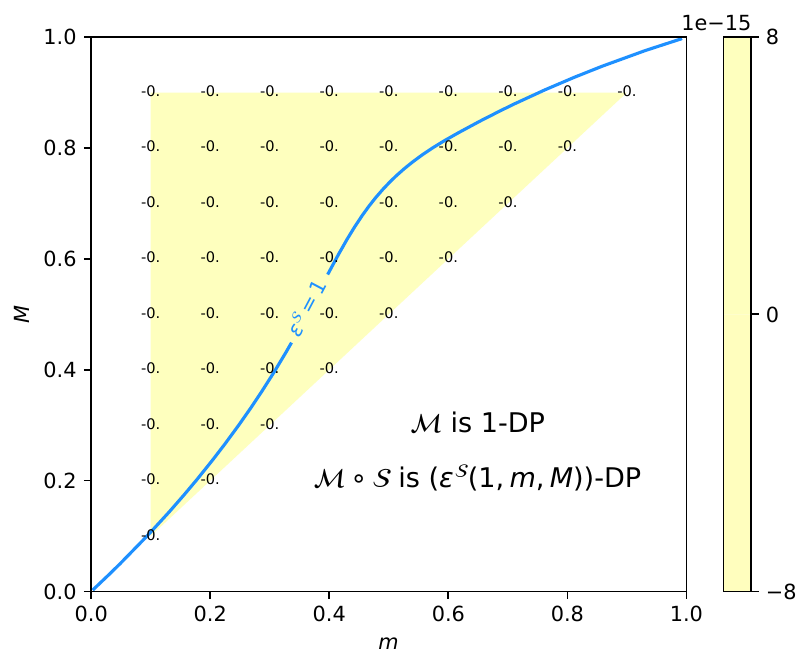}%
    \includegraphics[width=0.25\textwidth]{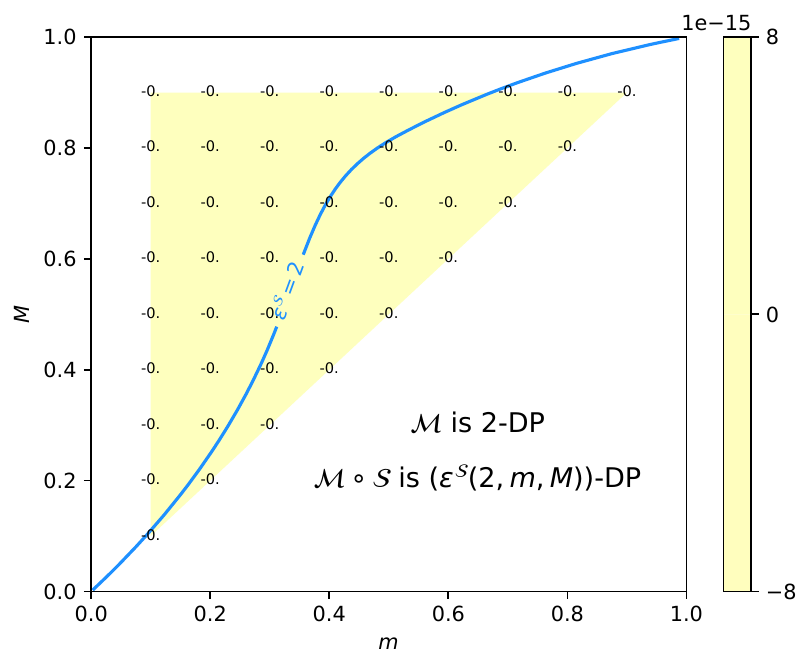}%
    \caption{The probability of outputting an incorrect mode of $\M$ minus that of $\M\circ\S$ without the noise reduction. Results shown for the exponential mechanism over the \texttt{HighestEducationCompleted} column in the Irish database.}
    \label{fig:Experiment1-RNM-ExponentialMechanism-Irishn-HighestEducationCompleted}
\end{figure}

\subsection{Plots of the Utility Difference between the Mechanisms \textit{without} the Noise Reduction for the Clustering Mechanisms}\label{sec:plots:SuppressionwithoutEpsDeltaChange3}

\begin{figure}[H]
    \includegraphics[width=0.25\textwidth]{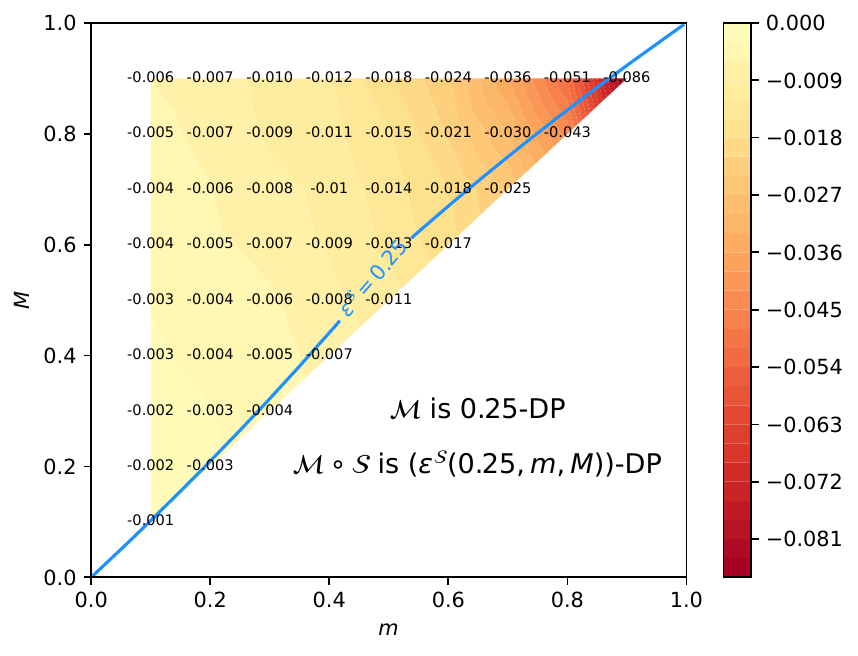}%
    \includegraphics[width=0.25\textwidth]{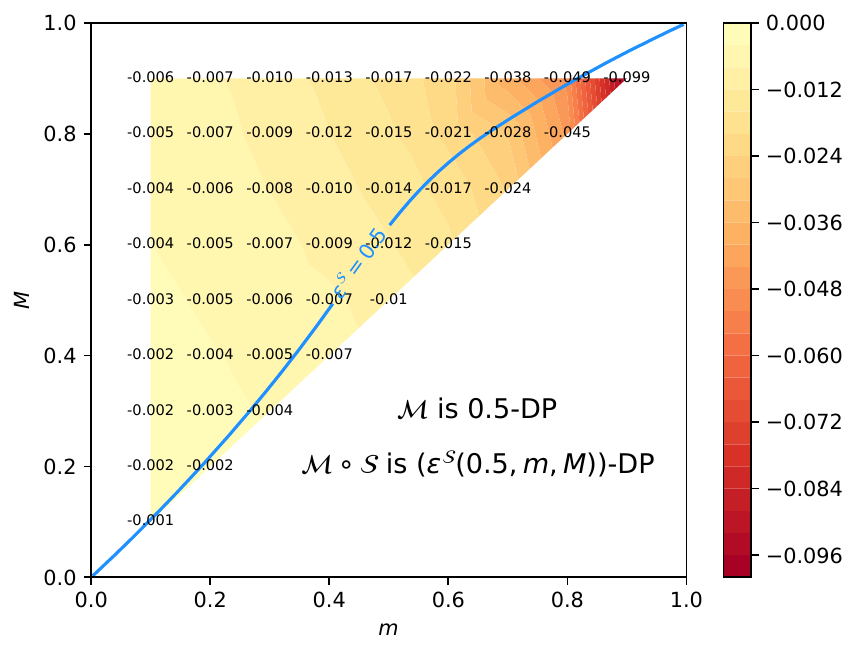}%
    \includegraphics[width=0.25\textwidth]{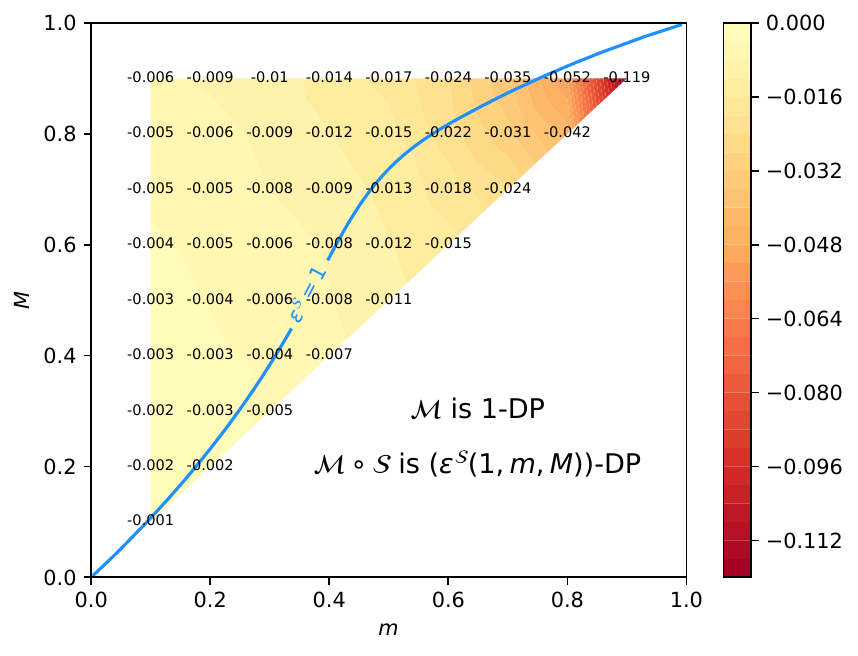}%
    \includegraphics[width=0.25\textwidth]{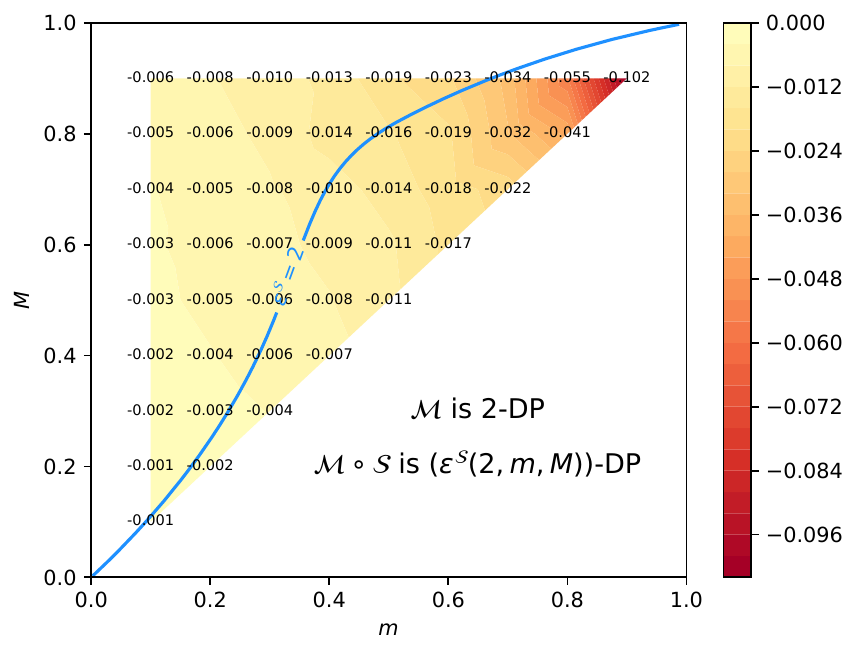}%
    \caption{The average cost of $\M$ minus that of $\M\circ\S$ without the noise reduction. Results shown for $k$-median over our synthetic database.}
    \label{fig:Experiment1-Clusteringkmedian}
\end{figure}

\begin{figure}[H]
    \includegraphics[width=0.25\textwidth]{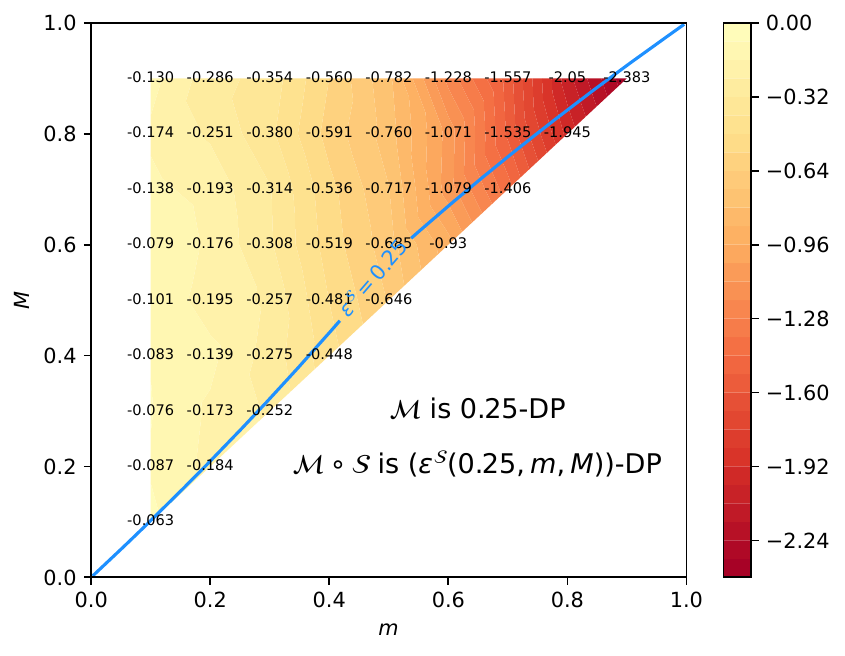}%
    \includegraphics[width=0.25\textwidth]{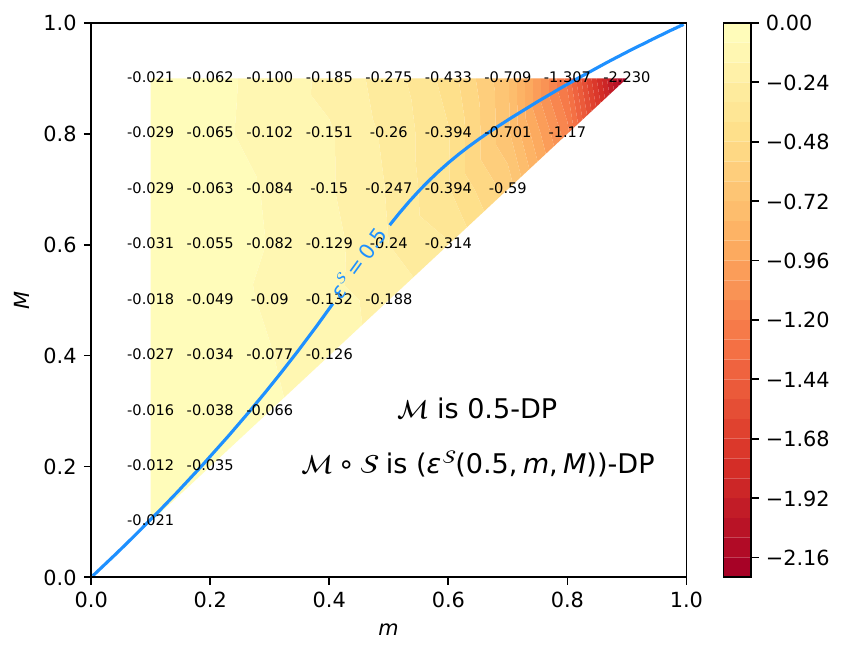}%
    \includegraphics[width=0.25\textwidth]{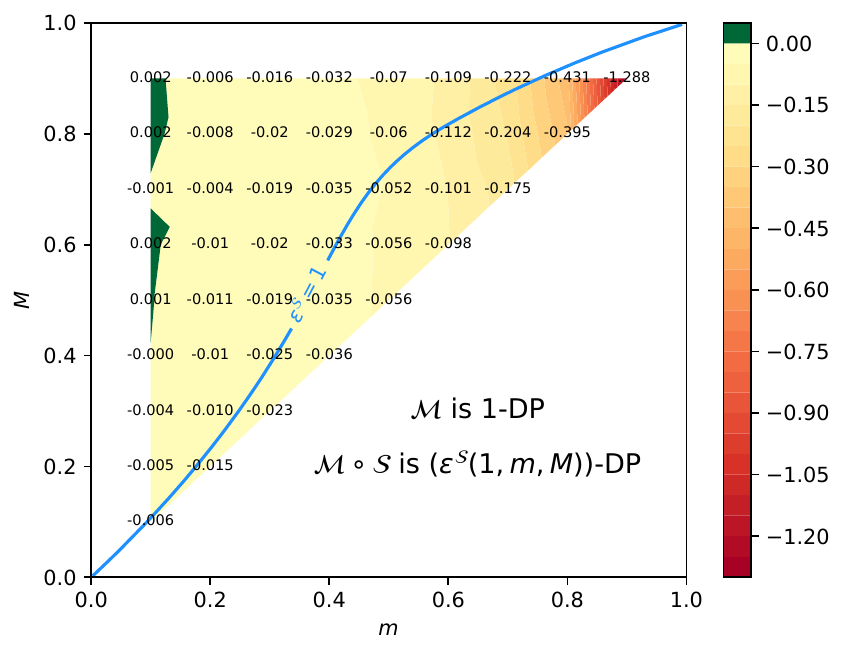}%
    \includegraphics[width=0.25\textwidth]{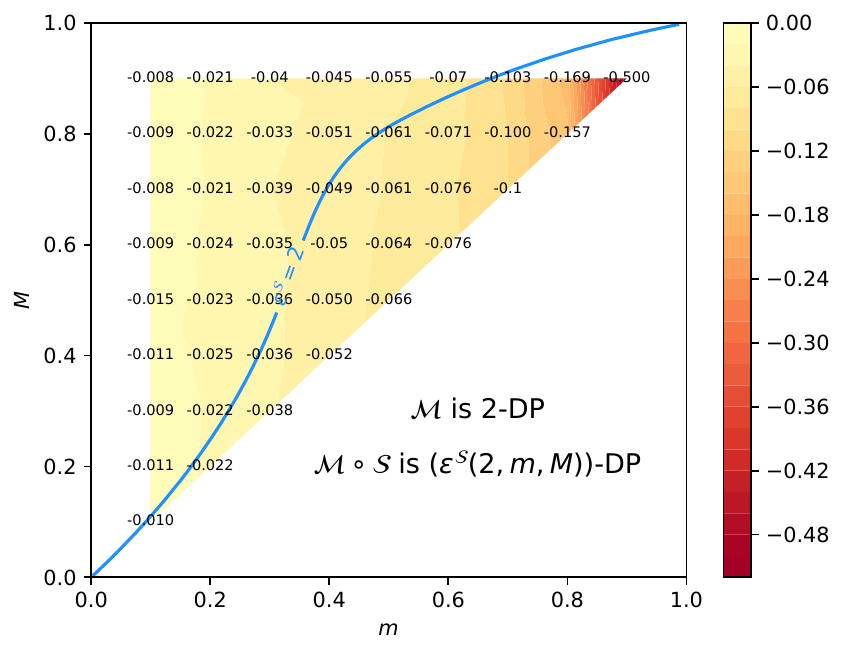}%
    \caption{The normalized intracluster variance of $\M$ minus that of $\M\circ\S$ without the noise reduction. Results shown for DPLloyd over the six numerical columns of the Adult database.}
    \label{fig:Experiment1-ClusteringDPLloyd}
\end{figure}